\documentclass[a4paper,twoside,openright,11pt]{report}
\usepackage[english]{babel}
\usepackage[ansinew]{inputenc}
\usepackage{fontenc}
\usepackage{amsfonts}
\usepackage{amsmath}
\usepackage{amsthm}
\usepackage{enumerate}
\usepackage{textcomp}
\usepackage{mathrsfs}
\usepackage{natbib}
\usepackage{dsfont} % For the indicator function definition
\usepackage{graphicx}
\usepackage{multirow}
\usepackage{pst-all} % pour le titre
\usepackage{lastpage}
\usepackage[nottoc]{tocbibind} % index of biblio

\usepackage[colorlinks=true,linkcolor=red,citecolor=blue]{hyperref}

\usepackage{geometry}
\numberwithin{equation}{section}

\newcommand{\n}{\noindent}
\newcommand{\F}{\textbf{F}}
\newcommand{\Q}{\textbf{Q}}

\newcommand{\ind}[1]{\ensuremath{\mathds{1}_{#1}}}
\DeclareMathOperator*{\argmin}{arg\,min\,}
\DeclareMathOperator*{\argmax}{arg\,max\,}
\DeclareMathOperator*{\argsup}{arg\,sup\,}
\DeclareMathOperator*{\arginf}{arg\,inf\,}

\newtheorem{theorem}{Theorem}[section]
\newtheorem{definition}{Definition}[section]
\newtheorem{corollary}{Corollary}[section]
\newtheorem{lemma}{Lemma}[section]
\newtheorem{proposition}{Proposition}[section]
\newtheorem{conclusion}{Conclusion}

\theoremstyle{definition}
\newtheorem{remark}{Remark}[section]
\newtheorem{example}{Example}[section]
\usepackage{fancyhdr}
\begin{document}
\begin{titlepage}
\psset{unit=1in,linewidth=4pt} %parametrage des unites pour pstricks
%\rput(1,1){\includegraphics[scale=0.8]{logo-limsi.png}}
%\rightline{\includegraphics[scale = 0.9]{logoweb-ups.png}}
\begin{tabular}{l p{8cm} r}
\includegraphics[scale=0.8]{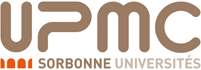} &  & \includegraphics[scale = 0.9]{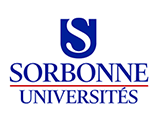}
\end{tabular}

\vspace{1cm}
%\rule{8.4in}{0mm} \\
%\rule{8.4in}{1mm}\\
\begin{center}
\Large{\textbf{\'Ecole Doctorale de Science Math\'ematiques de Paris Centre}}\\
\vspace{1cm}
\Huge{Th\`ese de Doctorat}\\

\vspace{1cm}

\LARGE{Discipline: Math\'ematiques}\\
\large{Sp\'ecialit\'e: Statistiques}\\

\vspace{2cm} 

\large{pr\'esent\'ee par} \\
\Large{\textbf{Diaa AL MOHAMAD}} \\

\vspace{1cm}
\rule{\textwidth}{1pt}
\LARGE \textbf{Esimation d'un Mod\`ele de M\'elange param\'etrique et Semiparam\'etrique par des $\varphi-$divergences}\\
\rule{\textwidth}{1pt}
\vspace{1cm}
{\Large dirig\'ee par Pr. Michel BRONIATOWSKI} \\

\vspace{1cm}

\large Soutenance le 17 novembre devant le jury compos\'e de \\  
\vspace{0.5cm}
\resizebox{\textwidth}{!}{ 
\begin{tabular}{lll}
Michel BRONIATOWSK &  LSTA - Universit\'e Paris VI & Directeur de th\`ese \\
St\'ephane CHRETIEN & National Physical Laboratory - UK & Rapporteur \\
Pierre VANDEKERKHOVE & LAMA - Universit\'e Marne-la-Vall\'ee & Rapporteur\\
Laurent BORDES & LMAP - Universit\'e de Pau et des Pays de l'Adour & Examinateur \\
Catherine MATIAS & LPMA - Universit\'e Paris VI & Examinateur\\
St\'ephane ROBIN & AgroParisTech / INRA & Examinateur
\end{tabular}}
\end{center}
%\rule{15.5cm}{1mm}
\vspace{\fill}
\end{titlepage}  %fin du titre
%\title{{\small Université de Paris-Sud 11}\\
%			{\small Spécialité Statistique Appliquée} \\
%			\vspace{8cm}
%			Rapport de stage de fin d'étude\\
%			Modèles d'Ordonnancement pour Apprendre dans les Systèmes de Traduction Automatique\vspace{3cm}}
%\author{Diaa AL MOHAMAD}
%\date{\today}
%\maketitle
%\tableofcontents

\clearpage

%===========================================================================================================================
%
%***************************************************************************************************************************
%***************************************************************************************************************************
%
%===========================================================================================================================
\chapter*{Remerciements}
\addcontentsline{toc}{chapter}{Remerciement}
Arrivant \`a la fin de cette th\`ese, je suis reconnaissant pour certaines personnes qui avaient un r\^ole important sur les travaux que j'ai pu mener dans mon travail de th\`ese. Il y a certaines personnes sans qui je n'aurai pas eu la chance de revenir ici en France afin de reprendre mon travail de th\`ese apr\`es deux ans et demi d'interruption.\\
Je vaudrais remercier tout d'abord ma m\`ere Alia qui m'a support\'e pendant toute la dur\'ee de mon s\'ejour en France et qui m'accorde sa confiance et son support \`a tout moment. Je vaudrais remercier mon p\`ere Hayel qui n'est plus avec nous dans ce monde depuis longtemps. Il m'encourageait toujours pour poursuivre mes \'etudes et il a consacr\'e sa vie pour me faire agrandir. Je remercie mes parents qui m'ont appris comment appr\'ecier le temps. Je vaudrais remercier mon \'epouse Tamader qui m'a accompagn\'e pendant toute la dur\'ee de ma th\`ese et qui \'etait patiente et a compris la pression que je subissais \`a cause de la courte dur\'ee de ma bourse d''etude (2 ans). Je lui remercie surtout parce qu'elle a support\'e ma d\'ecision de quitter mon pays la Syrie pour revenir en France, et puis d'avoir accept\'e d'\^etre avec moi ici en France sans qu'elle ait d'amie ni de famille. Il faudrait que je remercie ma fille Leen qui a r\'ejouit ma vie avec ses beaux sourires. Je suis tr\`es reconnaissant pour mon beau-p\`ere Moafk et ma belle-m\`ere Anam qui m'ont h\'eberg\'e avec ma petite famille pendant la p\'eriode (deux ans environs) que je pr\'eparais mes papiers pour revenir en France. Je leur remercie pour leur support qui m'accompagne tous les jours. Je vaudrais remercie \'egalement mes fr\`eres, Nouaf, Yasser et Fayez qui m'ont accord\'e leur confiance et m'ont aid\'e pour revenir en France. Je vaudrais remercier mes beaux-fr\`eres, Maaen, Mauthana, Omar et Mohamad que je consid\`ere comme mes fr\`eres et amis.\\

Pour l'avancement dans cette th\`ese, je remercie Michel Broniatowski qui a tout d'abord accept\'e de me prendre \`a nouveau apr\`es deux ans et demi d'interruption \`a cause de la stituation terrible dans mon pays. Ses id\'ees, ses conseils ainsi que sa vision de mon travail ont \'et\'e des facteurs importants pour achever et finir cette recherche. Je lui remercie pour plusieurs caf\'es auxquels il m'a invit\'es. Je lui remercie d'avoir compris ma situation familiale particuli\`ere. Je lui remercie de la confiance qu'il m'a accord\'ee pendant toute la dur\'ee de la th\`ese. Je remercie Alexis pour quelques courtes discussions sur les L-moments. Je remercie Assia qui m'a parl\'e des mod\`eles de m\'elanges semiparam\'etriques et des questions qui y sont li\'ees. Je remercie Matthias Kohl de l'Universit\'e de Furwagen en Allemagne pour des discussions sur l'int\'egration num\'erique. Je remercie Gilles Celeux, Jean-Patrick Baudrey et Olivier Schwander pour de vives discussions sur le mod\`ele de m\'elange semiparam\'etrique durant un groupe de travail. Je remercie mes camarades du laboratoire, Dimby, Moukhtar, Matthiau et les autres pour plusieurs discussions scientifiques ou non. Je remercie les rapporteurs de ma th\`ese, Mr. Vandekerkhove et Mr. Chr\'etien, pour leurs commentaires et leurs avis tr\`es positif et encourageant.\\

Il y a plusieurs personnes \`a qui je n'aurais peut-\^etre pas l'occasion de dire merci. Je remercie Khaled Halawa mon ancien prof de math\'ematiques et qui est devenu un cher coll\`egue en Syrie apr\`es. Je remercie Said Dsouki et Abd Allah Aboshahein qui ont convaincu l'administration de mon institut en Syrie de m'raccorder la bourse de la th\`ese apr\`es sa supension, Je remercie finalement Bassam Alzneika qui m'a aid\'e, pendant une p\'eriode tr\`es difficile dans la r\'egion o\`u je vivais en Syrie, \`a finir les papiers administratifs n\'ecessaires pour l'obtention de la bourse.\\

MERCI A TOUS ! \\

Diaa. \today.

\chapter*{Summary}
\addcontentsline{toc}{chapter}{Summary}
\section*{R\'esum\'e}
L'\'etude des mod\`eles de m\'elanges est un champ d'\'etude tr\`es vaste. D'autre part, les $\varphi-$divergences sont des outils statistiques qui de plus en plus attirent l'attention des statisticiens et les praticiens. Dans cette th\`ese, nous pr\'esentons et \'etudions quelques aspects et propri\'et\'es des $\varphi-$divergences et les estimateurs qui sont construits \`a base d'une $\varphi-$divergence. Nous employons ces estimateurs \`a l'estimation des mod\`eles de m\'elanges. Dans une seconde partie de cette th\`ese, nous construisons et d\'eveloppons une nouvelle structure pour les mod\`eles de m\'elanges semiparam\'etriques. L'estimation dans ce nouveau mod\`ele est bas\'ee sur les $\varphi-$divergences qui offrent de bons outils pour le traitement de notre nouvelle approche.\\

Nous pr\'esentons dans la premi\`ere partie de cette th\`ese les $\varphi-$divergences, et nous en faisons un rappel des principales propri\'et\'es. Nous pr\'esentons les m\'ethodes existantes dont l'objectif est de produire des estimateurs pour des mod\`eles param\'etriques bas\'es sur des $\varphi-$divergences. Nous nous int\'eressons \`a l'\'etude de la m\'ethode de Beran, de l'approche de Basu-Lindsay et de la forme dual des $\varphi-$divergences. Nous nous int\'eressons en particulier \`a la derni\`ere approche. Nous montrons que les estimateurs bas\'es sur la forme duale des $\varphi-$divergences dans un contexte param\'etrique ne sont pas robustes. Ceci est explor\'e th\'eoriquement et exp\'erimentalement avec des simulations num\'eriques. Nous proposons ensuite une modification qui rend ces estimateurs robustes. Le nouvel estimateur est alors compar\'e avec les autres m\'ethodes existantes qui produisent des estimateurs robustes \`a base des $\varphi-$divergences. La comparaison est \'egalement men\'ee par rapport \`a un estimateur consid\'er\'e comme $\ll$tr\`es performant$\gg$ pour l'estimation dans des mod\`eles param\'etriques; le minimum density power divergence. Notre nouvel estimateur montre de bonnes propri\'et\'es par rapport aux autres estimateurs en comp\'etition.\\
Dans un second temps, nous pr\'esentons un algorithme d'optimisation dont l'objectif est de calculer les estimateurs \`a base de divergences. Notre algorithme est un algorithme proximal qui perturbe la fonction objective \`a chaque it\'eration par une autre fonction convenablement choisie. La convergence des s\'equences g\'en\'er\'ees par l'algorithme est \'etudi\'ee. Nous montrons que les points limites des s\'equences g\'en\'er\'ees par l'algorithme sont des points stationnaires de la fonction objective. D'autres propri\'et\'es de convergence globale de la s\'equence vers un optimum local de la fonction objective sont explor\'ees mais en imposant des hypoth\`eses plus restrictives. Nous \'etudions la convergence de la s\'equence g\'en\'er\'ee par l'algorithme proximal dans plusieurs exemples. La convergence de l'algorithme EM est \'etudi\'ee \`a nouveau sur quelques exemples mais dans l'esprit de notre approche.\\

Dans la deuxi\`eme partie de cette th\`ese, nous construisons une nouvelle structure pour les mod\`eles de m\'elanges semiparam\'etriques \`a deux composantes dont l'une est inconnue. La nouvelle approche permet l'incorporation d'une information a priori lin\'eaire sur la composante inconnue; par exemple des contraintes de moments ou de L-moments. Nous d\'eveloppons deux approches pour l'estimation de ce mod\`ele; une approche pour les contraintes de type moments et une approche pour les contraintes de L-moments. Les propri\'et\'es asymptotiques des estimateurs r\'esultant sont \'etudi\'ees et prouv\'ees sous des hypoth\`eses standards. Nous illustrons par des simulations num\'eriques les avantages de la nouvelle approche et de l'incorporation d'une information a priori par rapport aux m\'ethodes existantes d'estimation d'un mod\`ele semiparam\'etrique sans aucune information pr\'eliminaire \`a part une hypoth\`ese de sym\'etrie.\\

\textbf{Mots cl\'es:} Mod\`ele de m\'elange, $\varphi-$divergence, estimateur \`a noyau sym\'etrique et asym\'etrique, algorithme proximal, dualit\'e de Fenche-Legendre, mod\`ele semiparam\'etrique, mod\`ele semiparam\'etrique de quantile, L-moments.

\clearpage
\section*{Abstract}
The study of mixture models constitutes a large domain of research. On the other hand, $\varphi-$divergences attract more and more the attention of statisticians and practitioners. In this work, we show some of the aspects and properties of $\varphi-$divergence-based estimators. We employ these estimators in mixture models. We build and develop in a second part a new structure for semiparametric mixture models and estimate the new model using $\varphi-$divergences efficiently.\\

In the first part of this work, we present $\varphi-$divergences and recall some of their basic properties. We present some of the existing methods in the literature which produce parametric estimation using $\varphi-$divergences; namely Beran's approach, the Basu-Lindsay approach and the dual formula of $\varphi-$divergences. We are interested in the later method. We show that the estimator based on the existing dual formula of $\varphi-$divergences in the parametric settings is not robust against outliers in several models. The problem is explored theoretically and experimentally. We propose a modification to this estimation method in order to robustify it. The new estimator is then compared to existing methods based on $\varphi-$divergences. The comparison is also done with respect to a powerful estimator in parametric estimation; the minimum density power divergence estimator. Our new estimator shows encouraging performances.\\
We present after that an optimization algorithm in order to calculate estimators based on a divergence criterion. The algorithm is a proximal-point algorithm which optimizes a modified version of the objective function by adding a suitable regularization term. Convergence properties of the presented algorithm are studied. We prove the convergence of the limiting points of the sequence generated by the algorithm to stationary points of the objective function. More properties are explored but with further assumptions in order to prove the convergence of the whole sequence towards a local optimum of the objective. Several examples are discussed, and another proof of convergence of the EM algorithm is given in several mixtures in the light of our approach.\\

In a second part of this work, we construct a new structure for semiparametric two-component mixture models where one component is unknown. The new structure permits to incorporate some prior linear information about the unknown component such as moments or L-moments constraints. Two different approaches are developed using $\varphi-$divergences in order to estimate the semiparametric mixture model; an approach for moment-type constraints and an approach for L-moments constraints. The asymptotic properties of the resulting estimators are studied and proved under standard conditions. The new structure is demonstrated by simulations to produce better estimates using the prior information than existing methods in the literature which do not consider in general any prior information except for a symmetric assumption.\\

\textbf{Keywords:} Mixture model, $\varphi-$divergence, symmetric and asymmetric kernel density estimator, proximal-point algorithm, Fenchel-Legendre duality, semiparametric model, semiparametric linear quantile model, L-moment.

\clearpage
\tableofcontents
\clearpage
\listoffigures
\clearpage
\listoftables
%===========================================================================================================================
%
%***************************************************************************************************************************
%***************************************************************************************************************************
%
%===========================================================================================================================
\chapter*{Introduction}
\addcontentsline{toc}{chapter}{Introduction}
Le but de cette th\`ese est l'\'etude et l'estimation d'un mod\`ele de m\'elanges de lois de la forme:
\[P(.|\phi) = \sum_{i=1}^j{\lambda_i P_i(.|\theta_i)}, \qquad \text{t.q. } \sum_{i=1}^{j}{\lambda_i}=1.\]
Dans la premi\`ere partie de la th\`ese (Chap.1 et 2), le mod\`ele de m\'elange est param\'etrique. C'est \`a dire que la distribution de chaque composante, index\'ee par $i=1,\cdots, j$, correspond a une loi connue param\'etr\'ee par $\theta_i$. Dans la deuxi\`eme partie de la th\`ese (Chap.3 et 4), nous aborderons le cas particulier o\`u $j=2$ en supposant que  la premi\`ere composante du m\'elange est param\'etrique alors que la deuxi\`eme est nonparam\'etrique. Nous proposons ensuite une structure o\`u la deuxi\`eme composante est semiparam\'etrique au sens o\`u elle appartient \`a une famille de lois d\'efinies par des contraintes lin\'eaires, par exemple l'ensemble des lois de probabilit\'e ayant une variance \'egale au carr\'e de l'esp\'erance. Les mod\`eles de m\'elanges param\'etriques ont des applications diverses en biologie, en machine learning etc., voir \cite{Titterington} ou \cite{FinMixModMclachlan} pour plus de d\'etails. Les mod\`eles de m\'elanges semiparam\'etriques ont \'et\'e employ\'es dans diff\'erents contextes en g\'en\'etique (\cite{JunMaDiscret}), biologie (\cite{Bordes06b}) en machine learning (\cite{Song}) pour des algorithmes de clustering, etc. Le mod\`ele semiparam\'etrique pourrait \^etre appliqu\'e dans d'autres situations et sur plus d'applications comme en traitement du signal.\\

L'estimation d'un mod\`ele de m\'elange param\'etrique se fait en g\'en\'eral avec l'algorithme EM de \cite{Dempster}. L'algorithme EM offre une proc\'edure facile \`a programmer et dont la complexit\'e est faible. En effet, l'algorithme EM maximise la log-vraisemblance du mod\`ele de m\'elange it\'erativement. A chaque it\'eration, nous maximisons la vraisemblance \`a l'int\'erieur de chaque classe (composante) en attribuant des poids $h_{i,k}$ \`a chaque observation (num\'ero $i$) mesurant son appartenance \`a la classe $k$. Cependant, l'algorithme EM produit des estimateurs non-robustes parce que nous calculons le maximum de vraisemblance en fin de compte. Le maximum de vraisemblance est un estimateur qui est connu d'\^etre sensible aux points aberrants \emph{(outliers)} et au fait que le mod\`ele ne contient pas la vraie distribution des donn\'ees \emph{(misspecification)}. L'objectif de la premi\`ere partie est d'appliquer un autre outil d'estimation qui produit des estimateurs robustes. Nous formulons \'egalement un algorithme qui ressemble \`a l'algorithme EM au sens o\`u l'optimisation n'est pas men\'ee sur tous les param\`etres en m\^eme temps, mais sur les proportions dans une \'etape et sur les param\`etres d\'ecrivant les classes dans une autre \'etape.\\
L'estimation d'un mod\`ele de m\'elange semiparam\'etrique \`a deux composantes est un sujet tr\`es r\'ecent. Plusieurs m\'ethodes existent pour estimer la proportion et/ou les param\`etres de la composante param\'etrique sans qu'il y ait de contraintes sur la composante inconnue. Une hypoth\`ese de symm\'etrie sur la composante inconnue a \'et\'e employ\'ee afin de mieux estimer le mod\`ele de m\'elange; voir \cite{Bordes10} et \cite{Maiboroda2012}. L'estimation d'un mod\`ele de m\'elange semiparam\'etrique sans qu'il y ait de contraintes sur la composante inconnue est difficile surtout lorsque nous avons \`a estimer des param\`etres inconnus de la composante param\'etrique. Nous proposons une m\'ethode pour estimer un mod\`ele de m\'elange semiparam\'etrique lorsque la composante inconnue est d\'efinie par des contraintes lin\'eaires de type moments ou L-moments. Nous \'etudions les propri\'et\'es asymptotiques des estimateurs obtenus. Plusieurs simulations num\'eriques sont pr\'esent\'ees afin d'illustrer l'avantage de la nouvelle approche.

%%%%%%%%%%%%%%%%%%%%%%%%%%%%%%%%%%%%%%%%%%%%%%%%%%%%%%%%%%%%\`u
%==============================
%%%%%%%%%%%%%%%%%%%%%%%%%%%%%%%%%%%%%%%%%%%%%%%%%%%%%%%%%%%%\`u

\section{Premi\`ere partie: Estimation robuste bas\'ee sur des \texorpdfstring{$\varphi-$}{phi-}divergences avec application aux mod\`eles de m\'elanges param\'etriques}
\subsection{Chapitre 1: Estimation bas\'ee sur des \texorpdfstring{$\varphi-$}{phi-}divergences}
Les $\varphi-$divergences sont des mesures de distance ou dissimilarit\'e entre des distributions de probabilit\'e ou plus g\'en\'eralement entre des measures $\sigma-$finies. Elles ont \'et\'e introduites ind\'ependamment par \cite{Csiszar1963} et \cite{AliSilvey}. Pour deux mesures $P$ et $Q$ $\sigma-$finies telles que $Q$ est absolument continue par rapport \`a $P$, nous d\'efinissons la $\varphi-$divergence entre $Q$ et $P$ par:
\[D_{\varphi}(Q,P) = \int_{\mathbb{R}^r}{ \varphi\left( \dfrac{dQ}{dP}(x) \right)dP(x)},\]
o\`u $\varphi$ est une fonction positive convexe telle que $\varphi(1)=0$. Si $\varphi$ est strictement convexe alors:
\[D_{\varphi}(Q,P) = 0 \;\; \textrm{ si et seulement si} \;\; P = Q.\]
L'estimation par une $\varphi-$divergence se fait en minimisant celle-ci entre une famille de lois et une mesure de probabilit\'e $P_T$ inconnue. La loi $P_T$ n'est connue en g\'en\'erale que par un \'echantillon $Y_1,\cdots,Y_n$ observ\'ee. La famille de lois est un mod\`ele $P_{\phi}$ avec $\phi\in\Phi\subset\mathbb{R}^d$ param\'etr\'e par le param\`etre $\phi$. Le but est de trouver le meilleur vecteur de param\`etres $\phi^T$ tel que $P_{\phi^T}$ soit le plus proche possible de $P_T$ d'un point de vue d'une $\varphi-$divergence. En particulier, si $P_T$ est un membre du mod\`ele $(P_{\phi})_{\phi}$, alors il existe $\phi^T$ tel que $P_T=P_{\phi^T}$, et
\[\phi^T = \argmin_{\phi\in\Phi} D_{\varphi}(P_{\phi},P_T).\]
Comme $P_T$ est inconnue, nous avons besoin de la remplacer par un estimateur afin d'estimer $\phi^T$. Dans le cas o\`u les mesures de probabilit\'e sont d\'efinies sur des espaces discrets, $P_T$ est remplac\'ee par la mesure empirique, voir \cite{LindsayRAF}. Dans le cas des mod\`eles continus, remplacer $P_T$ par sa version empirique n'est pas convenable, car le mod\`ele ne serait pas absolument continu par rapport \`a la mesure empirique pour n'importe quel $n$, et aucune proc\'edure d'estimation ne pourrait \^etre produite, voir \cite{BroniatowskiSeveralApplic}. Plusieurs approches ont \'et\'e propos\'ees afin d'approximer la $\varphi-$divergence lorsque le mod\`ele est continu. 
\begin{itemize}
\item L'approche de \cite{Beran}. Cette approche consiste \`a remplacer directement $P_T$ par un estimateur \`a noyau. L'approche de \cite{Beran} a \'et\'e propos\'ee dans le cas de la divergence de Hellinger. Cette approche a \'et\'e plus tard g\'en\'eralis\'ee \`a la classe des $\varphi-$divergences par \cite{ParkBasu} et \cite{KumarBasu}. \\
\item L'approche de \cite{BasuLindsay}. Cette approche consiste \`a remplacer $P_T$ par un estimateur \`a noyau et convoler le mod\`ele avec le m\^eme noyau afin de r\'eduire le r\^ole de la fen\^etre. Les auteurs d\'emontrent que sous une certaine condition (difficile) sur le noyau (noyau transparent), l'estimateur bas\'e sur leur approche est consistant sans avoir besoin que l'estimateur \`a noyau soit consistant. \\
\item La forme dual des $\varphi-$divergences. Cette approche a \'et\'e d\'evelopp\'ee ind\'ependamment par \cite{BroniaKeziou2006} et \cite{LieseVajdaDivergence}. On peut montrer que pour trois densit\'es de probabilit\'es $p_{\alpha},p_{\phi}$ et $p_T$, nous avons:
\[D_{\varphi}(p_{\phi},p_T) \geq \int{\varphi'\left(\frac{p_{\phi}}{p_{\alpha}}\right)(x)p_{\phi}(x)dx} - \int{\varphi^{\#}\left(\frac{p_{\phi}}{p_{\alpha}}\right)(y) p_T(y)dy},\]
o\`u $\varphi^{\#}(t)=t\varphi'(t)-\varphi(t)$, et que l'\'egalit\'e est atteinte lorsque $p_{\alpha}=p_T$. Alors:
\begin{equation}
D_{\varphi}(p_{\phi},p_T) = \sup_{\alpha\in\Phi}\left\{\int{\varphi'\left(\frac{p_{\phi}}{p_{\alpha}}\right)(x)p_{\phi}(x)dx} - \int{\varphi^{\#}\left(\frac{p_{\phi}}{p_{\alpha}}\right)(y) p_T(y)dy}\right\}.
\label{eqn:ParametricDualFormIntro}
\end{equation}
Il suffit maintenant de remplacer $p_T(y)dy$ par la mesure empirique $dP_n$ afin d'avoir un estimateur dit \emph{dual} de la $\varphi-$divergence. L'int\'er\^et de cette approche est que, contrairement aux autres approches, nous n'avons pas besoin d'un estimateur \`a noyau, et donc nous n'avons pas \`a chercher un \emph{bon} noyau et une \emph{bonne} fen\^etre.
\end{itemize}
Dans le premier chapitre, nous nous int\'eressons \`a la forme duale des $\varphi-$divergences. \cite{BroniatowskiKeziou2007} et \cite{LieseVajdaDivergence} proposent le minimum dual $\varphi-$divergence estimateur (MD$\varphi$DE) de $\phi^T$:
\[\hat{\phi} = \arginf_{\phi\in\Phi} \sup_{\alpha\in\Phi}\left\{\int{\varphi'\left(\frac{p_{\phi}}{p_{\alpha}}\right)(x)p_{\phi}(x)dx} - \frac{1}{n}\sum_{i=1}^n{\varphi^{\#}\left(\frac{p_{\phi}}{p_{\alpha}}\right)(y_i)}\right\}.\]
Lorsque $P_T$ appartient au mod\`ele $P_{\phi}$, la forme duale estime bien la $\varphi-$divergence parce que le supremum sur $\alpha$ sera atteint \`a $\alpha=\phi^T$. Cependant, si $P_T$ n'est pas dans le mod\`ele, ce n'est plus le cas. Par exemple, le cas des donn\'ees contenant des outliers. En effet, si $P_T$ n'est pas dans le mod\`ele, alors nous aurons:
\[D_{\varphi}(p_{\phi},p_T) \geq \sup_{\alpha\in\Phi}\left\{\int{\varphi'\left(\frac{p_{\phi}}{p_{\alpha}}\right)(x)p_{\phi}(x)dx} - \int{\varphi^{\#}\left(\frac{p_{\phi}}{p_{\alpha}}\right)(y) p_T(y)dy}\right\}.\]
Nous illustrons \`a la figure (\ref{fig:UnderEstimationIntro}) un exemple simple qui montre l'impact de ce probl\`eme. Nous illustrons \'egalement la solution que nous allons proposer ensuite.
\begin{figure}[h]
\centering
\includegraphics[scale=0.48]{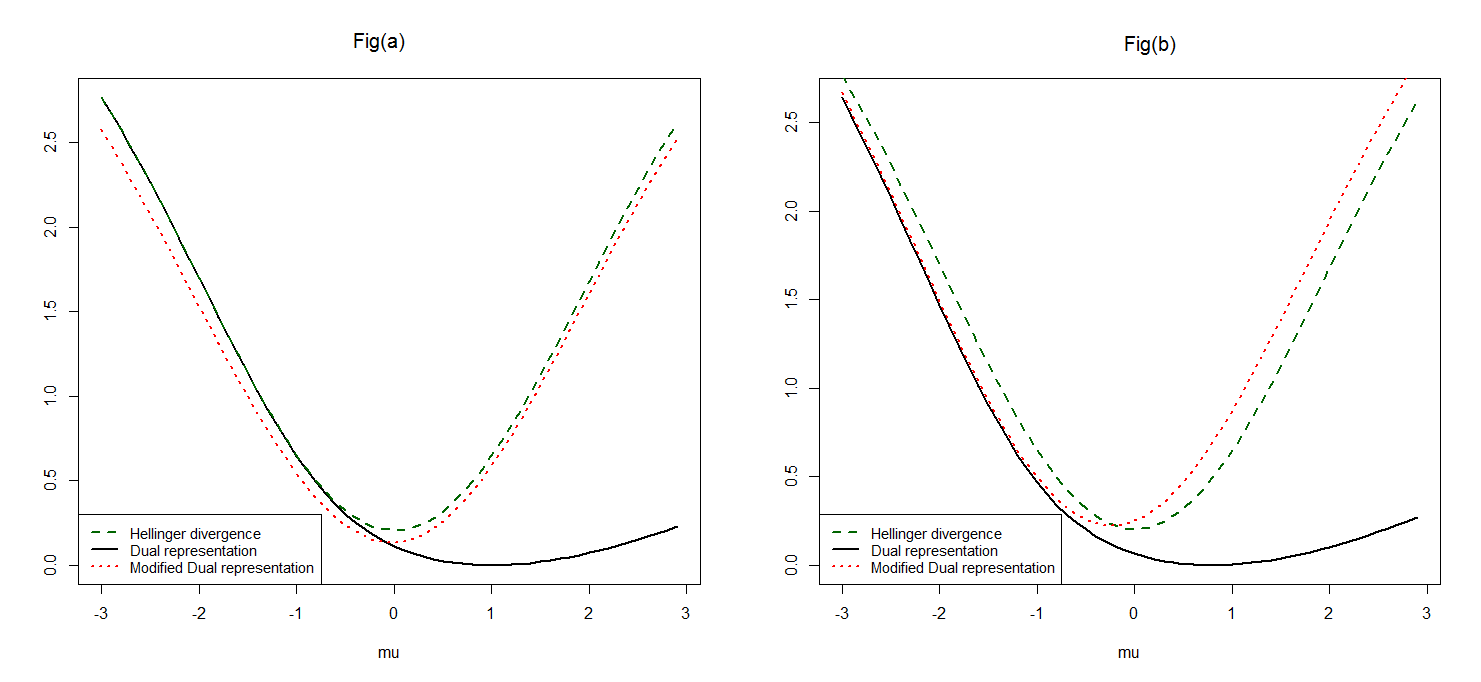}
\caption{Sous-estimation caus\'ee par la forme duale classique en comparaison avec notre alternative. La vraie distribution $P_T$ est $0.9\mathcal{N}(\mu=0,\sigma=1)+0.1\mathcal{N}(\mu=10,\sigma=2)$. Figure (a) montre la forme duale classique (\ref{eqn:ParametricDualFormIntro}) en comparaison avec la nouvelle formulation duale d\'efinie par (\ref{eqn:NewExactDualFormIntro}). Figure (b) montre les approximations correspondantes apr\`es avoir remplac\'e la vraie distribution par sa version empirique.}
\label{fig:UnderEstimationIntro}
\end{figure}
Notre solution consiste \`a utiliser un estimateur \`a noyau au lieu de $p_{\alpha}$. Ceci permet de s'adapter \`a $p_T$ que ce soit sous le mod\`ele ou non, et en m\^eme temps permet de se d\'ebarasser de la forme supr\'emale. Notre nouvel estimateur est d\'efini par:
\begin{equation}
D_{\varphi}(p_{\phi},p_T) = \sup_{w>0}\left\{\int{\varphi'\left(\frac{p_{\phi}}{K_{n,w}}\right)(x)p_{\phi}(x)dx} - \int{\varphi^{\#}\left(\frac{p_{\phi}}{K_{n,w}}\right)(y) p_T(y)dy}\right\}.
\label{eqn:NewExactDualFormIntro}
\end{equation}
En effet, un "bon" choix de la fen\^etre permet d'introduire le nouvel estimateur:
\begin{equation*}
\hat{\phi}_{n} = \arginf_{\phi\in\Phi} \int{\varphi'\left(\frac{p_{\phi}}{K_{n,w_{\text{opt}}}}\right)(x)p_{\phi}(x)dx} - \frac{1}{n}\sum_{i=1}^n{\varphi^{\#}\left(\frac{p_{\phi}}{K_{n,w_{\text{opt}}}}\right)(y_i)}.
\end{equation*}
Nous d\'emontrons que ce nouvel estimateur est consistant et asymptotiquement Gaussian sous des hypoth\`eses standards. Les simulations num\'eriques montrent que le nouvel estimateur performe mieux que les autres estimateurs pr\'esent\'es dans ce chapitre. Nous comparons la performance de cet estimateur avec le minimum density power divergence (MDPD) de \cite{BasuMPD} qui est un estimateur de Bregman. Notre estimateur performe aussi bon que le MDPD dans plusieurs simulations et performe mieux dans un mod\`ele \`a queue lourde qui est le GPD.

%%%%%%%%%%%%%%%%%%%%%%%%%%%%%%%%%%%%%%%%%%%%%%%%%%%%%%%%%%%%%%%%%%%%%%%%%%%

\subsection{Chapitre 2:}
Les proc\'edures d'estimation pr\'esent\'ees dans le chapitre pr\'ec\'edent sont en g\'en\'eral nonconvexes. Les m\'ethodes d'optimisation convexe ne garantissent que la convergence vers un optimum local de la fonction objective si celle-ci n'est pas convexe. Nous introduisons dans ce chapitre un algorithme d'optimisation proximale. Un algorithme proximal est un algorithme it\'eratif qui \`a chaque it\'eration optimise une version r\'egularis\'ee de la fonction objective. Les algorithmes proximaux ont \'et\'e prouv\'es de produire de meilleurs r\'esultats que les algorithmes d'optimisation classiques, voir \cite{Goldstein}. \\
Soient $(X_1,Y_1),\cdots,(X_n,Y_n)$ un \'echantillon de couples de variables al\'eatoires i.i.d. distribu\'ees selon la densit\'e $f(x,y|\phi^T)$ pour un $\phi^T\in\Phi$. Soient $(x_1,y_1),\cdots,(x_n,y_n)$ des r\'ealisations de ces couples. Les donn\'ees $y_1,\cdots,y_n$ sont les donn\'ees observ\'ees et les donn\'ees $x_1,\cdots,x_n$ sont les donn\'ees inobserv\'ees ou les \'etiquettes. Par exemple, les donn\'ees $x_1,\cdots,x_n$ sont les classes correspondant aux points $y_1,\cdots,y_n$. \\
L'algorithme EM est une proc\'edure it\'erative qui estime le vecteur de param\`etres $\phi^T$ en maximisant l'esp\'erance de la log-vraisemblance compl\'et\'ee sachant les donn\'ees observ\'ees, c.\`a.d.
\begin{eqnarray*}
\phi^{k+1} & = & \argmax_{\Phi} Q(\phi,\phi^k) \\
					 & = & \argmax_{\Phi} \mathbb{E}\left[\log(f(\textbf{X},\textbf{Y}|\phi)) \left| \textbf{Y}=\textbf{y},\phi^k\right.\right],
\end{eqnarray*}
o\`u $\textbf{X} = (X_1,\cdots,X_n)$, $\textbf{Y} = (Y_1,\cdots,Y_n)$ et $\textbf{y}=(y_1,\cdots,y_n)$. On peut d\'emontrer que ces it\'erations s'\'ecrivent de la fa\c{c}on suivante:
\begin{equation}
\phi^{k+1} = \argmax_{\Phi} \sum_{i=1}^n{\log\left(p_{\phi}(y_i)\right)} + \sum_{i=1}^n\int_{\mathcal{X}}{\log\left(\frac{h_i(x|\phi)}{h_i(x|\phi^k)}\right) h_i(x|\phi^k) dx}.
\label{eqn:EMProximalIntro}
\end{equation}
o\`u $h_i(x|\phi^k) = \frac{f(x,y_i|\phi^k)}{p_{\phi^k}(y_i)}$ est la densit\'e conditionnelle des \'etiquettes sachant une donn\'ee $y_i$, et $p_{\phi}$ est la loi marginale des donn\'ees observ\'ees. L'algorithme EM (\ref{eqn:EMProximalIntro}) s'\'ecrit comme un algorithme proximal, car nous sommes en train de maximiser la log-vraisemblance en la perturbant \`a chaque it\'eration de l'algorithme par une fonction positive qui ressemble \`a une distance de Kullback-Lebiler mais entre les densit\'es conditionnelles des \'etiquettes.\\
\cite{Tseng} propose de g\'en\'eraliser (\ref{eqn:EMProximalIntro}) en permettant au terme proximal \`a prendre d'autres formes plus g\'en\'erales guid\'ees par une fonction g\'en\'eratrice $\psi$. Tseng propose l'algorithme suivant:
\begin{equation}
\phi^{k+1} = \argsup_{\phi} J(\phi) - D_{\psi}(\phi,\phi^k),
\label{eqn:TsengAlgoIntro}
\end{equation}
o\`u $J(\phi)$ est la log-vraisemblance et
\begin{equation}
D_{\psi}(\phi,\phi^k) = \sum_{i=1}^n\int_{\mathcal{X}}{\psi\left(\frac{h_i(x|\phi)}{h_i(x|\phi^k)}\right)h_i(x|\phi^k)dx}.
\label{eqn:DivergenceClassesNtNormIntro}
\end{equation}
La fonction $\psi$ est prise comme une fonction convexe positive qui v\'erifie $\psi(1)=\psi'(1)=0$. Pour l'algorithme (\ref{eqn:EMProximalIntro}), $\psi(t)=-\log(t)+t-1$. \\
L'algorithme EM ainsi que la g\'en\'eralisation faite par Tseng ont une objective de maximiser la log-vraisemblance. Par cons\'equent, les estimateurs issus de ces algorithmes ne sont pas robustes contre les outliers ou une perturbation du mod\`ele autour de la vraie distribution des donn\'ees. Pour cela, nous proposons de g\'en\'eraliser l'algorithme de Tseng en utilisant le lien entre la maximisation de la log-vraisemblance et la minimisation de la distance de Kullback-Leibler entre la mesure empirique et le mod\`ele dans les mod\`eles discrets. Bien \'evidemment, pour les mod\`eles continus, ce lien est atteint de mani\`ere diff\'erente, car la distance entre le mod\`ele et la mesure empirique n'est pas bien d\'efinie\footnote{Dans le monde des $\varphi-$divergences, cette distance est consid\'er\'ee infinie.}. Nous proposons de remplacer la log-vraisemblance par un estimateur d'une $\varphi-$divergence.
\begin{equation}
\phi^{k+1} = \arginf_{\phi} \hat{D}_{\varphi}(p_{\phi},p_T) + \frac{1}{n}D_{\psi}(\phi,\phi^k).
\label{eqn:DivergenceAlgoPreVersionIntro}
\end{equation}
En prenant $\hat{D}_{\varphi}(p_{\phi},p_T)$ l'estimateur induit par la forme duale (\ref{eqn:ParametricDualFormIntro}) apr\`es avoir replac\'e $p_T(y)dy$ par $dP_n$, et pour $\varphi(t)=-\log(t)+t-1$, nous avons:
\begin{eqnarray*}
\phi^{k+1} & = & \arginf_{\phi}\left\{\sup_{\alpha} \frac{1}{n}\sum_{i=1}^n{\log(p_{\alpha}(y_i))} - \frac{1}{n}\sum_{i=1}^n{\log(p_{\phi}(y_i))} + \frac{1}{n}
D_{\psi}(\phi,\phi^k)\right\} \\
           & = & \arginf_{\phi}\left\{-\frac{1}{n}\sum_{i=1}^n{\log(p_{\phi}(y_i))} +\frac{1}{n}D_{\psi}(\phi,\phi^k)\right\} \\
           & = & \argsup_{\phi}\left\{\frac{1}{n}\sum_{i=1}^n{\log(p_{\phi}(y_i))} - \frac{1}{n}D_{\psi}(\phi,\phi^k)\right\} \\
           & = & \argsup_{\phi} J(\phi) - \frac{1}{n}D_{\psi}(\phi,\phi^k).
\end{eqnarray*}
Donc, notre algorithme contient la g\'en\'eralisation de Tseng. De plus, pour $\psi(t)=-\log(t)+t-1$, on se trouve avec l'algorithme EM.
Nous proposons \'egalement dans le cas d'un mod\`ele de m\'elange
\[p_{\phi}(y) = \sum_{i=1}^s{\lambda_i p_i(y|\theta_i)}\]
un algorithme proximal \`a deux niveaux; une sous-\'etape qui optimise sur les proportions $\lambda_i$ et une sous-\'etape qui optimise sur les param\`etres d\'ecrivant les classes $\theta_i$. En d'autres termes:
\begin{eqnarray}
\lambda^{k+1} & = & \arginf_{\lambda\in[0,1]^s, s.t. (\lambda,\theta^k)\in\Phi} \hat{D}_{\varphi}(p_{\lambda,\theta^k},p_{\phi^*}) + D_{\psi}((\lambda,\theta^k),\phi^k); \label{eqn:DivergenceAlgoSimp1Intro} \\
\theta^{k+1} & = & \arginf_{\theta\in\Theta, s.t. (\lambda^{k+1},\theta)\in\Phi} \hat{D}_{\varphi}(p_{\lambda^{k+1},\theta},p_{\phi^*}) + D_{\psi}((\lambda^{k+1},\theta),\phi^k).
\label{eqn:DivergenceAlgoSimp2Intro}
\end{eqnarray}
Nous d\'emontrons sous des hypoth\`eses standards que les s\'equences $(\phi^k)_k$ g\'en\'er\'ees par l'un des algorithmes (\ref{eqn:DivergenceAlgoPreVersionIntro}) ou (\ref{eqn:DivergenceAlgoSimp1Intro}, \ref{eqn:DivergenceAlgoSimp2Intro}) convergent vers un point stationnaire de l'estimateur de la divergence $\phi\mapsto\hat{D}_{\varphi}(p_{\phi},p_T)$. \\
D\'efinissons l'ensemble $\Phi^0$ par:
\begin{equation}
\Phi^0 = \{\phi\in\Phi: \hat{D}_{\varphi}(p_{\phi},p_T)\leq \hat{D}_{\varphi}(p_{\phi^0},p_T)\}.
\label{eqn:SetPhis0Intro}
\end{equation} 
Nous citons ici les deux principaux r\'esultats th\'eoriques concernant la convergence de la s\'equence $\Phi^k$ g\'en\'er\'ee par les algorithmes (\ref{eqn:DivergenceAlgoPreVersionIntro}) ou (\ref{eqn:DivergenceAlgoSimp1Intro}, \ref{eqn:DivergenceAlgoSimp2Intro}).
\begin{proposition}
Supposons que les s\'equences (\ref{eqn:DivergenceAlgoPreVersionIntro}) et (\ref{eqn:DivergenceAlgoSimp1Intro}, \ref{eqn:DivergenceAlgoSimp2Intro}) sont bien d\'efinies dans $\Phi$. Pour les deux algorithmes, la s\'equence $(\phi^{k})_k$ v\'erifie les propri\'et\'es suivantes:
\begin{itemize}
\item[(a)] $D_{\varphi}(p_{\phi^{k+1}}|p_T)\leq D_{\varphi}(p_{\phi^k}|p_T)$;
\item[(b)] $\forall k, \phi^k \in \Phi^0$;
\item[(c)] Supposons que les fonctions $\phi\mapsto\hat{D}_{\varphi}(p_{\phi}|p_T), D_{\psi}$ sont semicontinues inf\'erieurement et que l'ensemble $\Phi^0$ est compact, alors la s\'equence $(\phi^k)_k$ est bien d\'efinie et born\'ee. De plus, la s\'equence $\left(\hat{D}_{\varphi}(p_{\phi^k}|p_T)\right)_k$ converge.
\end{itemize}
\end{proposition}
Cette proposition annonce une propri\'et\'e essentielle de l'algorithme. Sous des conditions simples, nous avons que la fonction objective (l'estimateur de la divergence) converge le long de la s\'equence $\phi^k$, voir figure (\ref{fig:DecreaseDivGaussChi2Chi2Intro}). Cette propri\'et\'e peut \^etre utilis\'ee comme un crit\`ere d'arr\^et de l'algorithme au cas o\`u la s\'equence de vecteurs $\phi^k$ ne converge pas. 
\begin{figure}[ht]
\centering
\includegraphics[scale = 0.3]{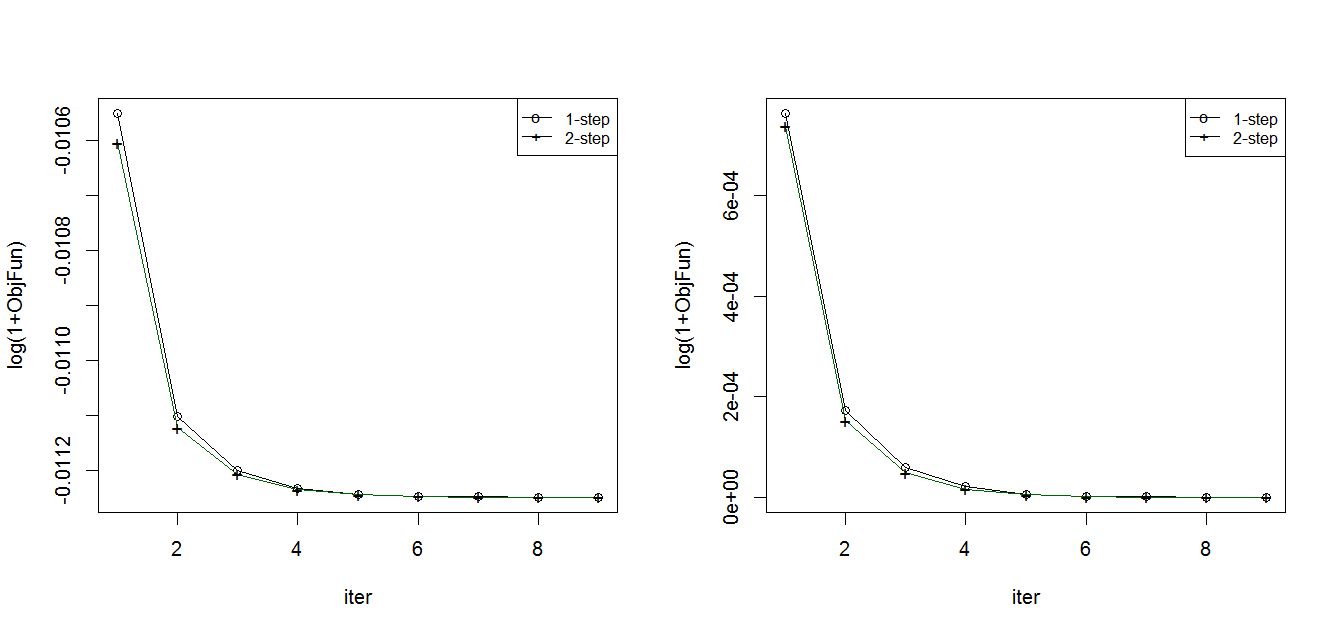}
\caption{D\'ecroissance de la (estimateur de la) distance de Hellinger entre la vraie distribution des donn\'ees et le mod\`ele estim\'e \`a chaque it\'eration de l'algorithme proximal dans le cas d'un mod\`ele de m\'elange \`a 2 composantes Gaussiennes. La figure de gauche illustre la courbe des valeurs pour le nouvel estimateur dual (\ref{eqn:NewExactDualFormIntro}) estim\'e. La figure de droite illustre la courbe des valeurs de l'estimateur dual classique (\ref{eqn:ParametricDualFormIntro}) estim\'e. Les valeurs sont repr\'esent\'ees sur une \'echelle logarithmique $\log(1+x)$. Le 1-step repr\'esente l'algorithme (\ref{eqn:DivergenceAlgoPreVersionIntro}), et le 2-step repr\'esente l'algorithme (\ref{eqn:DivergenceAlgoSimp1Intro}, \ref{eqn:DivergenceAlgoSimp2Intro})} 
\label{fig:DecreaseDivGaussChi2Chi2Intro}
\end{figure}
Un deuxi\`eme r\'esultat principal dans ce travail prouve la convergence des sous-suites vers un point stationnaire de la divergence estim\'ee. Ce r\'esultat est nouveau, car les r\'esultats existants supposent que le terme proximal $D_{\psi}$ v\'erifie une hypoth\`ese d'identifiabilit\'e $D_{\psi}(\phi,\phi')=0$ ssi $\phi=\phi'$. Dans ce r\'esultat, nous ne demandons pas cette hypoth\`ese.
\begin{proposition}
Supposons que 
\begin{enumerate}
\item les fonctions $\phi\mapsto\hat{D}_{\varphi}(p_{\phi}|p_T), D_{\psi}$ et $\nabla_1 D_{\psi}$ sont d\'efinies et continues sur, respectivement, $\Phi, \Phi\times\Phi$ et $\Phi\times\Phi$;
\item $\nabla \hat{D}_{\varphi}(p_{\phi}|p_T)$ est d\'efinie et continue sur $\Phi$;
\item $\Phi^0$ est compact,
\end{enumerate}
alors pour l'algorithme d\'efini par (\ref{eqn:DivergenceAlgoPreVersionIntro}), toute sous-suite convergente converge vers un point stationnaire de la fonction objective $\phi\rightarrow\hat{D}(p_{\phi},p_T)$. De plus, si l'hypoth\`ese 2 n'est pas v\'erifi\'ee, alors 0 appartient au sous-gradient de $\phi\mapsto\hat{D}(p_{\phi},p_T)$ calcul\'e au point limite. En d'autres termes, les sous-suites convergentes convergent vers des points stationnaires g\'en\'eralis\'es.
\end{proposition}
\begin{example}
Dans un m\'elange de deux composantes Gaussiennes:
\[p_{\lambda,\mu}(x) = \frac{\lambda}{\sqrt{2\pi}} e^{-\frac{1}{2}(x-\mu_1)^2} + \frac{1-\lambda}{\sqrt{2\pi}} e^{-\frac{1}{2}(x-\mu_2)^2},\]
on peut d\'emontrer avec notre approche que les points limites g\'en\'er\'es par l'algorithme EM sont des points stationnaires de la log-vraisemblance d\`es que le vecteur initial $\phi^{0}$ v\'erifie la condition suivante:
\[J(\phi^0)>\max\left[J\left(0,\infty,\frac{1}{n}\sum_{i=1}^n{y_i}\right),\; J\left(1,\frac{1}{n}\sum_{i=1}^n{y_i},\infty\right)\right],\]
\end{example}
\noindent o\`u $J$ est la log-vraisemblance. D'autres exemples et discussions sont pr\'esent\'es dans le chapitre ci-dessous. Les simulations num\'eriques men\'ees sur des m\'elanges de Gaussiens et des m\'elanges de Weibull montrent que l'algorithme proximal marche, et nous arrivons \`a calculer les estimateurs bas\'es sur les formes duales pr\'esent\'ees dans le chapitre pr\'ec\'edent. Une application de notre algorithme proximal sur le density power divergence de \cite{BasuMPD} est \'egalement pr\'esent\'ee.

%%%%%%%%%%%%%%%%%%%%%%%%%%%%%%%%%%%%%%%%%%%%%%%%%%%%%%%%%%%%%%%%%%%%%%%%%%%%%%%%%
% =====================================================
%%%%%%%%%%%%%%%%%%%%%%%%%%%%%%%%%%%%%%%%%%%%%%%%%%%%%%%%%%%%%%%%%%%%%%%%%%%%%%%%%
\section{Deuxi\`eme partie: Mod\`eles de m\'elanges semiparam\'etriques \`a deux composantes dont l'une est inconnue}
Un mod\`ele de m\'elange semiparam\'etrique \`a deux composantes est d\'efinie par:
\begin{equation}
f(x) = \lambda f_1(x|\theta) + (1-\lambda) f_0(x), \qquad \text{for } x\in\mathbb{R}^r,
\label{eqn:GeneralSemiParaMixIntro}
\end{equation}
o\`u $\lambda$ et $\theta$ sont deux param\`etres \`a estimer. $f_0$ est consid\'er\'ee inconnue durant l'estimation.
\subsection{Chapitre 3: Mod\`eles de m\'elanges semiparam\'etriques \`a deux composantes dont l'une est d\'efinie par des contraintes lin\'eaires sur sa distribution}
Le mod\`ele de m\'elange semiparam\'etrique (\ref{eqn:GeneralSemiParaMixIntro}) a \'et\'e \'etudi\'e et employ\'e dans certaines applications r\'ecemment par plusieurs auteurs, \cite{Bordes06b}, \cite{Robin}, \cite{Song}, \cite{JunMaDiscret} et \cite{Xiang}. Dans ces papiers, le mod\`ele n'a pas \'et\'e trait\'e avec la d\'efinition g\'en\'erale donn\'ee ci-dessus. Certains auteurs ont consid\'er\'e que $\theta$ est connu de mani\`ere \`a ce que la composante param\'etrique soit enti\`erement connue. D'autres auteurs ont consid\'er\'e une distribution pr\'ecise comme le Gaussien pour $f_1$ sans aucune g\'en\'eralisation \`a d'autres familles de distributions. 	\\
Plusieurs m\'ethodes d'estimation du mod\`ele de m\'elanges semiparam\'etrique (\ref{eqn:GeneralSemiParaMixIntro}) ont \'et\'e introduites. \cite{Bordes06b} proposent d'\'etudier le mod\`ele (\ref{eqn:GeneralSemiParaMixIntro}) lorsque $r=1$, $f_0$ est sym\'etrique par rapport \`a une valeur inconnue $\mu_0$ et $f_1$ est enti\`erement connue, c.\`a.d. $\theta$ est connu. \cite{Song} proposent d'\'etudier le mod\`ele (\ref{eqn:GeneralSemiParaMixIntro}) lorsque $f_1$ est un Gaussien centr\'e en 0 avec une variance inconnue. Ils supposent en plus que $f_0(0)=0$. D'autres auteurs ont propos\'e des m\'ethodes de type EM, voir \cite{Robin}, \cite{Song} et \cite{JunMaDiscret}. Une m\'ethode bas\'ee sur la distance de Hellinger a \'et\'e propos\'ee par \cite{Xiang}, mais nous n'en parlons pas, car l'algorithme pr\'esent\'e dans leur article n'est pas claire et contient des calculs int\'egrales qui ne se font pas avec des m\'ethodes num\'eriques. \\
Ces m\'ethodes d'estimation ne sont pas bas\'ees sur une th\'eorie solide except\'ee la m\'ethode de \cite{Bordes06b}. De plus, le comportement asymptotique des algorithmes propos\'es n'a pas \'et\'e \'etudi\'e. La convergence des m\'ethodes it\'eratives de type EM n'a pas \'et\'e \'etablie non plus. La m\'ethode de \cite{Bordes06b} exploite la structure de sym\'etrie impos\'ee sur $f_0$ pour construire un estimateur consistant et asymptotiquement Gaussien, voir \cite{Bordes10}. \\
L'article de \cite{Xiang} illustre une comparaison entre plusieurs m\'ethodes d'estimation pour le mod\`ele semiparam\'etrique (\ref{eqn:GeneralSemiParaMixIntro}) lorsque $\theta$ est connu. Les m\'ethodes donnent de bonnes performances sans qu'il y ait une m\'ethode gagnante. Les donn\'ees ont \'et\'e g\'en\'er\'ees par des m\'elanges de deux Gaussiens. Nous avons refait des simulations similaires mais en consid\'erant $\theta$ inconnu. Les r\'esultats n'ont pas \'et\'e satisfaisants. La m\'ethode de \cite{Bordes10} a tendance \`a donner de bonnes r\'esultats pour une proportion basse de la partie param\'etrique, mais non pas tr\`es proche de z\'ero. Les autres m\'ethodes ont tendance \`a donner de bonnes performances lorsque la proportion de la partie param\'etrique est \'elev\'ee. Nous croyons que le probl\`eme vient du degr\'e de difficult\'e du mod\`ele de m\'elange semiparam\'etrique. L'ajout d'une information a priori comme la sym\'etrie de $f_0$ a permis d'am\'eliorer l'estimation et de mieux \'etudier la th\'eorie li\'ee \`a la m\'ethode.\\
Suivant l'id\'ee de \cite{Bordes06b}, nous proposons d'ajouter une information a priori relativement g\'en\'erale de mani\`ere \`a ce que nous puissions bien estimer le mod\`ele de m\'elange semiparam\'etrique. Nous proposons d'ajouter une information lin\'eaire comme les contraintes de moments. Les contraintes lin\'eaires peuvent \^etre trait\'ees avec des outils d'analyse convexe. Nous d\'efinissons ainsi notre mod\`ele de m\'elange semiparam\'etrique sous des contraintes lin\'eaires sur la composante inconnue par:
\begin{eqnarray}
P(.| \phi) & = &  \lambda P_1(.|\theta) + (1-\lambda) P_0 \quad \text{s.t. } \nonumber\\
P_0\in\mathcal{M}_{\alpha} & = & \left\{Q \in M \text{ s.t. } \int_{\mathbb{R}^r}{dQ(x)}=1, \int_{\mathbb{R}^r}{g(x)dQ(x)}=m(\alpha) \right\}
\label{eqn:SetMalphaIntro}
\end{eqnarray}
o\`u $g(x)=(g_1(x),\cdots,g_{\ell}(x))$ et $m(\alpha)=(m_1(\alpha),\cdots,m_{\ell}(\alpha))$.\\
L'estimation d'un mod\`ele semiparam\'etrique d\'efini par des contraintes lin\'eaires a \'et\'e \'etudi\'ee par \cite{BroniaKeziou12}. Les auteurs proposent d'estimer un mod\`ele semiparam\'etrique par des $\varphi-$divergences. \cite{AlexisGSI13} ont travaill\'e avec des mod\`eles semiparam\'etriques sous des contraintes de L-moments. Afin d'exploiter leurs m\'ethodologies, nous devons travailler plut\^ot sur un mod\`ele d\'efini par $P_0$. En effet, supposons avoir un \'echantillon $X_1,\cdots,X_n$ distribu\'e selon $P_T$ un m\'elange de deux lois $P_1(.|\theta^*)$ et $P_0^*$. Nous avons:
\begin{equation}
P_0^* = \frac{1}{1-\lambda^*} P_T - \frac{\lambda^*}{1-\lambda^*} P_1(.|\theta^*).
\label{eqn:TrueP0modelIntro}
\end{equation}
D\'efinissons l'ensemble:
\begin{equation}
\mathcal{N} = \left\{Q=\frac{1}{1-\lambda} P_T - \frac{\lambda}{1-\lambda} P_1(.|\theta), \quad \lambda\in(0,1),\theta\in\Theta\right\}.
\label{eqn:SetNnewModelIntro}
\end{equation}
Notons que $P_0^*$ appartient \`a $\mathcal{N}$ pour $\lambda=\lambda^*$ et $\theta=\theta^*$. Par ailleurs, comme $P_0^*$ v\'erifie l'ensemble de contraintes de $\mathcal{M}_{\alpha^*}$ pour un $\alpha^*$ inconnu, nous pouvons \'ecrire:
\[P_0^*\in \mathcal{N} \bigcap \cup_{\alpha\in\mathcal{A}}\mathcal{M}_{\alpha}.\]
Alors, il est raisonnable de d\'efinir une proc\'edure d'estimation qui minimise la distance entre $\mathcal{N}$ et $\cup_{\alpha\in\mathcal{A}}\mathcal{M}_{\alpha}$. Cette distance est atteinte en $P_0^*$. Nous avons alors:
\begin{equation}
(\lambda^*,\theta^*,\alpha^*) \in \arginf_{\lambda,\theta,\alpha}\inf_{P_0\in\mathcal{M}_{\alpha}}D_{\varphi}\left(P_0,\frac{1}{1-\lambda} P_T - \frac{\lambda}{1-\lambda} P_1(.|\theta)\right).
\label{eqn:MomentTrueEstimProcIntro}
\end{equation}
Ceci est un probl\`eme d'optimisation sur un espace de dimension infinie. Pour le r\'esoudre, nous utilisons un r\'esultat de dualit\'e de Fenchel-Legendre, voir Proposition 1.4 de \cite{AlexisThesis} (voir \'egalement Proposition 4.2 de \cite{BroniaKeziou12}).
\begin{eqnarray}
(\lambda^*,\theta^*,\alpha^*)  & = &  \arginf_{\phi}\inf_{Q\in\mathcal{M}_{\alpha}} D_{\varphi}\left(Q,\frac{1}{\lambda-1} P_T - \frac{\lambda}{1-\lambda} P_1(.|\theta)\right) \nonumber\\
& = & \arginf_{\phi}\sup_{\xi\in\mathbb{R}^{l+1}} \xi^t m(\alpha)  - \frac{1}{1-\lambda}\int{\psi\left(\xi^t g(x)\right) dP_T(x)}\nonumber \\
	&  & + \frac{\lambda}{1-\lambda} \int{\psi\left(\xi^t g(x)\right) dP_1(x|\theta)}. \label{eqn:MomentEstimProcIntroTrue}
\end{eqnarray}
o\`u $\psi(t)=\sup_x tx-\varphi(x)$. Dans cette formule, nous avons $m(\alpha)=(1,m_1(\alpha),\cdots,m_{\ell}(\alpha))$. Il est possible maintenant d'estimer le triplet $(\lambda^*,\theta^*,\alpha^*)$ \`a la base d'un \'echantillon $X_1,\cdots,X_n$ par:
\begin{eqnarray}
(\hat{\lambda}, \hat{\theta},\hat{\alpha}) & = & \arginf_{\lambda, \theta,\alpha} \sup_{\xi\in\mathbb{R}^{l+1}} \xi^t m(\alpha)  - \frac{1}{1-\lambda}\frac{1}{n}\sum_{i=1}^n{\psi\left(\xi^t g(X_i)\right)} \nonumber\\
 &  & \qquad \qquad \qquad + \frac{\lambda}{1-\lambda} \int{\psi\left(\xi^t g(x)\right) dP_1(x|\theta)}.
\label{eqn:MomentEstimProcIntro}
\end{eqnarray}
Nous prouvons que cet estimateur est consistant et asymptotiquement Gaussien sous des hypoth\`eses standards.
\begin{example} \label{ex:Chi2MomentIntro}
Prenons l'exemple d'un mod\`ele de m\'elange \`a deux composantes dont l'une est d\'efinie par trois contraintes de moments (les trois premiers moments). L'ensemble $\mathcal{M}_{\alpha}$ est d\'efini par:
\[\mathcal{M}_{\alpha} = \left\{Q: \int{dQ(x)}=1,\; \int{xdQ(x)}=m_1(\alpha),\; \int{x^2dQ(x)}=m_2(\alpha),\; \int{x^3dQ(x)}=m_3(\alpha)\right\}.\]
Si $\varphi(t)=(t-1)^2/2$, alors $\psi(t)=\frac{1}{2}t^2+t$ et l'optimum sur $\xi$ est donn\'e par:
\[\xi(\phi) = \Omega^{-1}\left(m(\alpha) - \int{g(x)\left(\frac{1}{1-\lambda}dP(x)-\frac{\lambda}{1-\lambda}dP_1(x|\theta)\right)}\right), \text{ pour } \phi\in\Phi^+.\]
o\`u 
\[\Omega = \int{g(x)g(x)^t\left(\frac{1}{1-\lambda}dP(x)-\frac{\lambda}{1-\lambda}dP_1(x|\theta)\right)}. \]
$\Phi^+$ est l'ensemble de param\`etres pour lesquels la fonction objective est concave par rapport \`a $\xi$. En d'autres termes, $\Phi^+=\{\phi : \Omega \text{ est sym\'etrique d\'efinie positive}\}$. Soit $M_i$ le moment d'ordre $i$ de $P_T$. D\'enotons \'egalement $M_i^{(1)}(\theta)$ le moment d'ordre $i$ de la composante param\'etrique $P_1(.|\theta)$.
\[M_i = \mathbb{E}_{P_T}[X^i],\qquad M_i^{(1)}(\theta)=\mathbb{E}_{P_1(.|\theta)}[X^i].\]
Un calcul simple montre que:
\begin{eqnarray*}
\Omega & = & \int{g(x)g(x)^t\left(\frac{1}{1-\lambda}dP(x)-\frac{\lambda}{1-\lambda}dP_1(x|\theta)\right)} \\
  & = & \left[\frac{1}{1-\lambda}M_{i+j-2} - \frac{\lambda}{1-\lambda}M_{i+j-2}^{(1)}(\theta)\right]_{i,j\in\{1,\cdots,4\}}.
\end{eqnarray*}
et la fonction objective dans (\ref{eqn:MomentEstimProcIntroTrue}) est donn\'ee par:
\begin{multline*}
H(\phi,\xi) = \xi^tm(\alpha) - \left[\frac{1}{2}\xi_1^2+\xi_1+(\xi_1\xi_2+\xi_2)\left(\frac{1}{1-\lambda}M_1-\frac{\lambda}{1-\lambda}M_1^{(1)}(\theta)\right)\right. \\ +(\xi_2^2/2+\xi_1\xi_2+\xi_3)\left(\frac{1}{1-\lambda}M_2-\frac{\lambda}{1-\lambda}M_2^{(1)}(\theta)\right) + (\xi_1\xi_4+\xi_2\xi_3+\xi_4)\left(\frac{1}{1-\lambda}M_3-\frac{\lambda}{1-\lambda}M_3^{(1)}(\theta)\right) \\
+ (\xi_3^2/2+\xi_2\xi_4)\left(\frac{1}{1-\lambda}M_4-\frac{\lambda}{1-\lambda}M_4^{(1)}(\theta)\right) + \xi_3\xi_4\left(\frac{1}{1-\lambda}M_5-\frac{\lambda}{1-\lambda}M_5^{(1)}(\theta)\right) \\
\left. + \xi_4^2/2 \left(\frac{1}{1-\lambda}M_6-\frac{\lambda}{1-\lambda}M_6^{(1)}(\theta)\right)\right].
\end{multline*}
Cet exemple montre que notre estimateur peut \^etre calcul\'e de mani\`ere efficace et avec une complexit\'e lin\'eaire sans que la dimension des donn\'ees intervienne.
\end{example}
Des simulations num\'eriques ont \'et\'e men\'ees afin de tester la validit\'e de notre approche et de comparer sa performance aux m\'ethdes existantes.
%%%%%%%%%%%%%%%%%%%%%%%%%%%%%%%%%%%%%%%%%%%%%%%%%%%%%%%%%%%%%%%%%%%%%
\subsection{Chapitre 4: Mod\`eles de m\'elanges semiparam\'etrique \`a deux composantes dont l'une est d\'efinie par des contraintes  de L-moments}
Le sujet de ce chapitre est consid\'er\'e comme la suite du chapitre pr\'ec\'edent. Nous avons propos\'e une structure pour un mod\`ele de m\'elange semiparam\'etrique \`a deux composantes dont l'une est d\'efinie par des contraintes lin\'eaires. En prenant des contraintes de moments, on est aper\c{c}u que les chiffres de calculs explosent facilement. Exemple \ref{ex:Chi2MomentIntro} montre le cas des trois premiers moments impos\'es sur la composante inconnue $P_0$. Le calcul de la matrice $\Omega$ ainsi que la fonction objective ensuite $H(\phi,\xi)$ est une arithm\'etique entre les moments de la composante param\'etrique et les moments du m\'elange. Avec les trois premiers moments, nous avons d\'ej\`a \`a calculer les moments jusqu'\`a l'ordre 6. Si l'on travaille avec des distributions \`a queues lourdes, le moment d'ordre 6 prendra des valeurs d'ordre $10^{10}$ voire plus. Les moments du m\'elange eux-m\^emes seront remplac\'es durant l'estimation par des moments empiriques. Ceux-ci explosent rapidement pour un \'echantillon donn\'e m\^eme pour des distributions \`a queues l\'eg\`eres. En effet, le calcul de l'inverse de la matrice $\Omega$ devient d\'elicat. Par exemple, durant nos simulations num\'eriques, il n'a pas \'et\'e possible d'utiliser une m\'ethode num\'erique pour inverser la matrice $\Omega$ car elle a eu une sensibilit\'e \'elev\'ee dans certains voisinages des param\`etres $\phi$, et par cons\'equent, nous avons calcul\'e l'inverse avec des m\'ethodes d'inversion par bloc directes.\\ 
R\'ecemment, les L-moments ont \'et\'e propos\'es par \cite{Hoskings} et ils sont de plus en plus utilis\'es comme des alternatives des moments standards. Les L-moments sont repr\'esentatifs et caract\'erisent la loi de probabilit\'e d\`es que son esp\'erance existe, voir Th\'eor\`eme 1 de \cite{Hoskings}. De plus, les quatre premiers L-moments sont des indicateurs de la moyenne, l'\'echelle, le skewness et le kurtosis. Ces premi\`eres propri\'et\'es sont d\'ej\`a int\'eressantes pour consid\'erer les L-moments. Nous citons aussi le fait que les calculs num\'eriques des L-moments ne s'explosent pas facilement et dans nos simulations, les valeurs num\'eriques ont \'et\'e toujours proches de 1 et aucun probl\`eme de sensibilit\'e de matrices n'a \'et\'e rencontr\'e.\\
Avant de proc\'eder \`a l'introduction du nouveau mod\`ele, nous rappelons la d\'efinition des L-moments.
Soient $X_{1:n} < \ldots < X_{n:n} $ les statistiques d'ordre associ\'ees \`a un \'echantillon $X_1,\cdots,X_n$ i.i.d. ayant une fonction de r\'epartition $\mathbb{F}_T$.
\begin{definition}
Le L-moment d'ordre $r$, not\'e $\lambda_r$, $r=1,2,\ldots$ est d\'efini par:
\begin{equation*}
  \lambda_r = \frac{1}{r} \sum_{k=0}^{r-1}(-1)^k \binom{r-1}{k}\mathbb{E}\left( X_{r-k:r} \right).
\end{equation*}
\end{definition}
Une propri\'et\'e tr\`es int\'eressante et essentielle des L-moments est qu'ils sont lin\'eaires en les mesures de quantile. Une mesure de quantile est d\'efinie par le moyen de la fonction quantile de la mani\`ere suivante. Pour tout bor\'elien de $\mathcal{B}([0,1])$
\[{\bf{F}}^{-1}(B)=\int_0^1{\ind{x\in B}d\mathbb{F}^{-1}(x)} \in\mathbb{R}\cup\{-\infty,+\infty\}.\]
Les L-moments peuvent \^etre r\'e\'ecrits par
\begin{equation}
\lambda_r = -\int_{\mathbb{R}}{K_r(t)d\F^{-1}(t)}, \qquad r\geq 2
\label{eqn:LmomRepIntShiftLegIntro}
\end{equation}
o\`u
\begin{equation}
K_r(t) = \int_0^t{L_{r-1}(u)du} = \sum_{k=0}^{r-1}{\frac{(-1)^{r-k}}{k+1}\binom{r}{k}\binom{r+k}{k}t^{k+1}}
\label{eqn:IntShiftLegPolyIntro}
\end{equation}
sont les polyn\^omes de Legendre translat\'es et int\'egr\'es. La lin\'earit\'e des L-moments par rapport aux quantiles est la cl\'e essentielle pour construire notre m\'ethode d'estimation dans un m\'elange semiparam\'etrique bas\'ee sur le r\'esultat de dualit\'e de Fenchel-Legendre. Nous d\'efinissons notre mod\`ele de m\'elange semiparam\'etrique \`a deux composantes dont l'une est d\'efinie par des contraintes de L-moments par:
\begin{eqnarray}
P(.| \phi) & = &  \lambda P_1(.|\theta) + (1-\lambda) P_0 \quad \text{s.t. } \nonumber\\
\F_0^{-1}\in\mathcal{M}_{\alpha} & = & \left\{\Q^{-1} \in M^{-1}, \Q^{-1}\ll\F_0^{-1} \text{ s.t. } \int_{0}^1{K(u)d\Q^{-1}}=m(\alpha)\right\}
\label{eqn:SetMalphaLmomIntro}
\end{eqnarray}
o\`u $m(\alpha) = (m_2(\alpha),\cdots,m_{\ell-1}(\alpha))$. Nous d\'efinissons de mani\`ere similaire au chapitre pr\'ec\'edent un mod\`ele par le moyen de $P_0$. D\'efinissons les ensembles:
\begin{eqnarray*}
\Phi^+ & = & \left\{(\lambda,\theta)\in(0,1)\times\Theta : \frac{1}{1-\lambda}\mathbb{F}_T - \frac{\lambda}{1-\lambda}\mathbb{F}_1(.|\theta) \text{ est une fonction de r\'epartition}\right\}, \\
\mathcal{N}^{-1} & = &  \left\{\Q^{-1} \in M^{-1}\; : \; \exists (\lambda,\theta)\in\Phi^+ \text{ s.t. } \mathbb{Q}^{-1} = \left(\frac{1}{1-\lambda}\mathbb{F}_T - \frac{\lambda}{1-\lambda}\mathbb{F}_1(.|\theta)\right)^{-1}\right\}.
\end{eqnarray*}
L'ensemble $\Phi^+$ repr\'esente l'ensemble effectif des param\`etres concern\'es dans le nouveau mod\`ele \'ecrit avec $P_0$. En effet, un couple $(\lambda,\theta)$ ne d\'efinit pas en g\'en\'eral une fonction de r\'epartition ayant la forme $\frac{1}{1-\lambda}\mathbb{F}_T - \frac{\lambda}{1-\lambda}\mathbb{F}_1(.|\theta)$. Pour cela, il est important pour le moment de ne garder que les param\`etres qui rendent cette fonction une fonction de r\'epartition. Nous aurons la possibilit\'e plus tard d'ignorer ce probl\`eme et de travailler sur tout l'ensemble $\Phi$.\\
Nous avons:
\[{\mathbb{F}_0^*}^{-1}\in \mathcal{N}^{-1} \bigcap \cup_{\alpha}\mathcal{M}_{\alpha}.\]
Donc, il est raisonnable de d\'efinir une proc\'edure d'estimation qui minimise une distance entre les deux ensembles $\mathcal{N}^{-1}$ et $\cup_{\alpha}\mathcal{M}_{\alpha}$.
\begin{equation}
(\lambda^*,\theta^*,\alpha^*) \in \arginf_{(\lambda,\theta,\alpha)\in\Phi^+}\inf_{\F_0^{-1}\in\mathcal{M}_{\alpha}} D_{\varphi}\left(\F_0^{-1},\left(\frac{1}{1-\lambda}\F_T - \frac{\lambda}{1-\lambda}\F_1(.|\theta)\right)^{-1}\right).
\label{eqn:EstimProcQuantileFunsIntro}
\end{equation}
Afin de r\'esoudre ce probl\`eme d'optimisation qui est men\'e sur un espace de dimension infinie, nous utilisons \`a nouveau la Proposition 1.4 de \cite{AlexisThesis} (voir \'egalement Proposition 4.2 de \cite{BroniaKeziou12}) pour \'ecrire:
\[(\lambda^*,\theta^*,\alpha^*) \in \arginf_{(\lambda,\theta,\alpha)\in\Phi^+} \sup_{\xi\in\mathbb{R}^{\ell-1}} \xi^t m(\alpha) - \int_0^1{\psi\left(\xi^tK(u)\right)d\left(\frac{1}{1-\lambda}\F_T - \frac{\lambda}{1-\lambda}\F_1(.|\theta)\right)^{-1}(u)}.\]
En utilisant le Lemme 1.2 de \cite{AlexisThesis}, nous pouvons \'ecrire:
\begin{equation}
(\lambda^*,\theta^*,\alpha^*) \in \arginf_{(\lambda,\theta,\alpha)\in\Phi^+}\sup_{\xi\in\mathbb{R}^{\ell-1}} \xi^tm(\alpha) - \int_{\mathbb{R}}{\psi\left(\xi^tK\left(\frac{1}{1-\lambda}\mathbb{F}_T(x) - \frac{\lambda}{1-\lambda}\mathbb{F}_1(x|\theta)\right)\right)dx}.
\label{eqn:EstimProcPhiPLusIntro}
\end{equation}
Ceci est une proc\'edure d'estimation o\`u la fonction de r\'epartition g\'en\'erant les donn\'ees $\mathbb{F}_T$ peut \^etre approxim\'ee par sa version empirique afin d'estimer le triple $(\lambda^*,\theta^*,\alpha^*)$ \`a la base d'un \'echantillon donn\'e. Cependant, la caract\'erisation de l'ensemble $\Phi^+$ ici n'est pas \'evidente et tr\`es co\^uteuse num\'eriquement. De plus, l'ensemble $\Phi^+$ pourrait prendre des formes qui ne sont pas ad\'equates pour les algorithmes d'optimisation num\'eriques, voir figure (\ref{fig:DiffFormPhiPlusIntro}).
\begin{figure}[ht]
\centering
\includegraphics[scale=0.45]{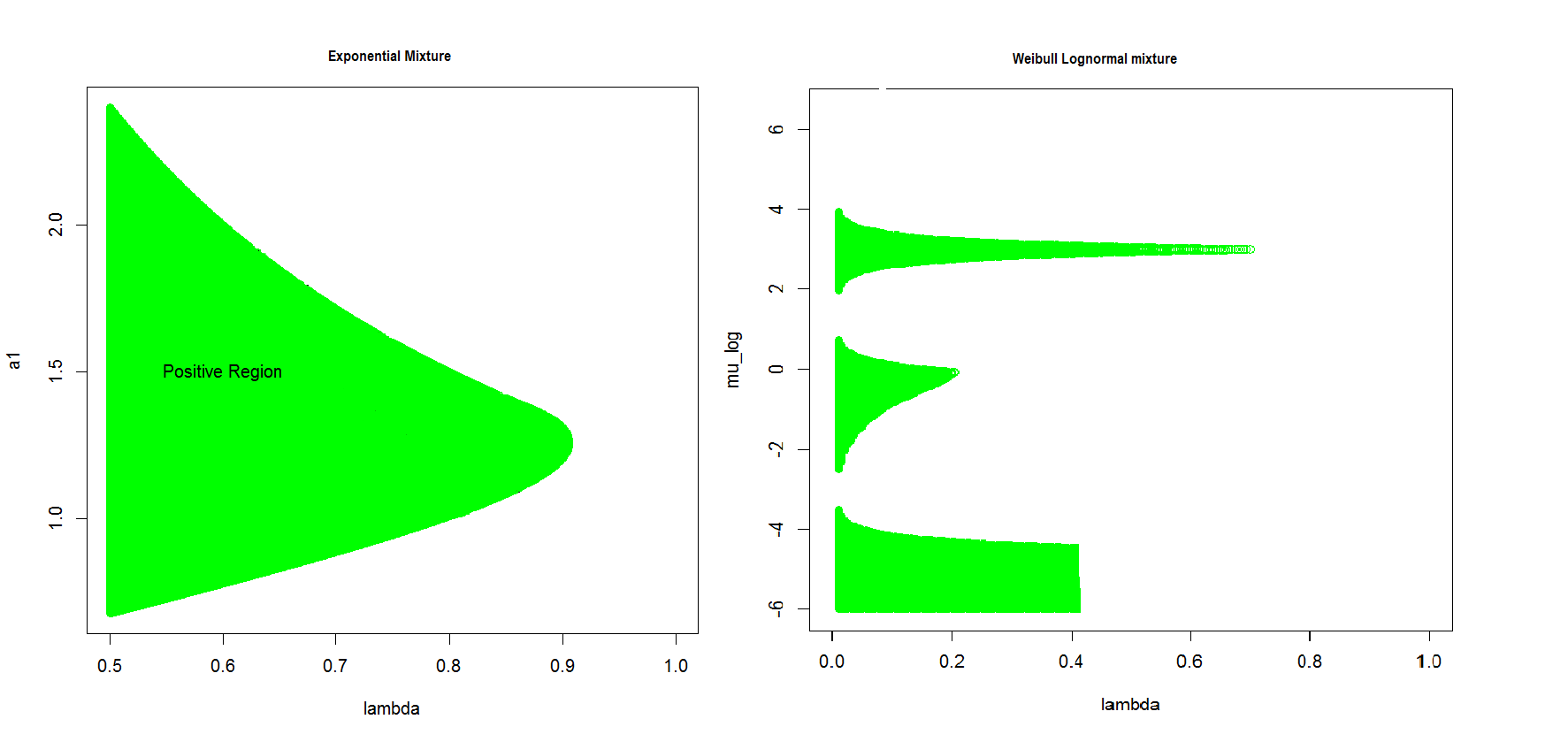}
\caption{Differentes formes de l'ensemble $\Phi^+$. Pour le m\'elange Weibull-Lognormal, c'est le Weibull qui est la composante semiparametrique.}
\label{fig:DiffFormPhiPlusIntro}
\end{figure}
Afin de r\'esoudre ce probl\`eme, nous montrons que $\phi^*=(\lambda^*,\theta^*,\alpha^*)$ est un infimum global de la fonction $H(\phi,\xi(\phi))$ o\`u:
\[H(\phi,\xi)  =  \xi^t m(\alpha) - \int{\psi\left[\xi^tK\left(\frac{1}{1-\lambda} \mathbb{F}_T(y) - \frac{\lambda}{1-\lambda} \mathbb{F}_1(y|\theta)\right)\right]dy},\]
et 
\[\xi(\phi) = \argsup_{\xi\in\mathbb{R}^{\ell-1}} H(\phi,\xi).\]
Donc, si la fonction $H(\phi,\xi(\phi))$ n'a qu'un seul infimum global, ceci ne sera autre que $\phi^*$. Cela justifie notre proc\'edure d'estimation:
\begin{equation}
\phi^* = \arginf_{(\alpha,\theta,\lambda)\in\Phi}\sup_{\xi\in\mathbb{R}^{\ell-1}} \xi^t m(\alpha) - \int{\psi\left[\xi^tK\left(\frac{1}{1-\lambda} \mathbb{F}_T(x) - \frac{\lambda}{1-\lambda} \mathbb{F}_1(x|\theta)\right)\right]dx}.
\label{eqn:EstimProcLmomDualCDFVersionTrueIntro}
\end{equation}
Avec un \'echantillon $X_1,\cdots,X_n$ distribu\'e selon $\mathbb{F}_T$, nous estimons $\phi^*$ par:
\begin{equation}
\hat{\phi} = \arginf_{(\alpha,\theta,\lambda)\in\Phi}\sup_{\xi\in\mathbb{R}^{\ell-1}} \xi^t m(\alpha) - \int{\psi\left[\xi^tK\left(\frac{1}{1-\lambda} \mathbb{F}_n(x) - \frac{\lambda}{1-\lambda} \mathbb{F}_1(x|\theta)\right)\right]dx}.
\label{eqn:EstimProcLmomDualCDFVersionIntro}
\end{equation}
Nous montrons que cet estimateur est consistant et asymptotiquement Gaussien sous des hypoth\`eses standards. Les simulations num\'eriques montrent le gain de l'utilisation des contraintes de L-moments par rapport aux moments standard surtout lorsque la proportion de la composante param\'etrique est tr\`es basse; 0.05 voire 0.01.
\begin{example}
Prenons le cas d'une divergence de $\chi^2$ pour $\varphi(t)=(t-1)^2/2$ et $\psi(t)=t^2/2+t$. La fonction objective $H(\phi,\xi)$ est donn\'ee par:
\[H(\phi,\xi) = \xi^tm(\alpha) - \int{\frac{1}{2}\left(\xi^tK\left(\mathbb{F}_0(y|\phi)\right)\right)^2 + \xi^tK\left(\mathbb{F}_0(y|\phi)\right) dy},\]
o\`u 
\[\mathbb{F}_0(y|\phi) = \frac{1}{1-\lambda} \mathbb{F}_T(y) - \frac{\lambda}{1-\lambda} \mathbb{F}_1(y|\theta).\]
La fonction $H$ est un polyn\^ome de degr\'ee 2 en $\xi$. Pour tout $\phi\in\Phi$, le supremum par rapport \`a $\xi$ est donn\'e par:
\[\xi(\phi) = \Omega^{-1}\left(m(\alpha) - \int{K(\mathbb{F}_0(y|\phi))}dy\right),\]
o\`u
\[\Omega = \int{K\left(\mathbb{F}_0(y|\phi)\right)K\left(\mathbb{F}_0(y|\phi)\right)^tdy}.\]
La matrice Hessienne de $\xi\mapsto H(\phi,\xi)$ est \'egale \`a $-\Omega$, et donc c'est une matrice sym\'etrique d\'efinie n\'egative pour tout $\phi$ dans $\Phi$. Par cons\'equent, $\xi(\phi)$ est un supremum global de $\xi\mapsto H(\phi,\xi)$.
\end{example}

\part{Robust Estimation Using \texorpdfstring{$\varphi-$}{phi-}Divergences with Application to Parametric Mixture Models}
\chapter{Estimation using a phi-divergence}
The maximum likelihood estimation is a simple and efficient method to estimate unknown parameters of a given model. The most common drawback of this method is its sensibility to contamination and misspecification. From the first years of the twentieth century, many researchers such as Pearson, Hellinger, Kullback, Neymann and others started developing different approaches using distance-like functions between probability density functions called divergences. Several divergence-based techniques permit to construct robust estimators, such as $\varphi-$divergences (\cite{Csiszar}, \cite{AliSilvey}), $S-$divergences (\cite{GoshSDivergence}), R\'enyi pseudodistances (see for example \cite{TomaAubin}), Bregman divergences and many others. In this work, we are particularly interested in $\varphi-$divergences on the one hand, and on the other hand, in comparing the resulting estimators and existing approaches with the maximum likelihood estimator (MLE) and some well-known divergences.\\
Estimation using $\varphi-$divergences is based on the idea of minimizing a distance between the true distribution and a given model. In practice, the true distribution is replaced by its empirical version calculated on the basis of an $n-$sample. When working with discrete models, everything goes well and a simple plug-in of the empirical distribution results in a plausible and good estimation procedure, see \cite{LindsayRAF}. The true challenge appears when we work with continuous models and smoothing techniques are apparently necessary tools. Several techniques were proposed in the literature. We give a brief summary of some of these approaches and present a new method which has very encouraging performances and properties. A comparison of several methods based on $\varphi-$divergences will be presented at the end of this chapter with an extensive simulation study on several distributions. The comparison is held with respect to the maximum likelihood estimator (MLE) and a powerful estimator called the minimum density power divergence (MDPD) introduced by \cite{BasuMPD}.
% ---------------------------------------------------------------------------------
%
%====================================================
%%%%%%%%%%%%%%%%%%%%%%%%%%%%%%%%%%%%%%%%%%%%%%%%%%%%%%%%%%%%%%%%%%%%%%%%%%%%%%%%%%%
\section{A brief introduction about \texorpdfstring{$\varphi$}{phi}-divergences}
\subsection{Definition, useful properties and standard examples}\label{subsec:DefPhiDiv}
$\varphi$-divergences were introduced independently by \cite{Csiszar1963} (as "$f$-divergences") and \cite{AliSilvey}. Let $P$ and $Q$ be two $\sigma-$finite measures defined on $(\mathbb{R}^{r}, \mathscr{B}(\mathbb{R}^{r}))$ such that $Q$ is absolutely continuous (a.c.) with respect to (w.r.t.) $P$. Let $\varphi : \mathbb{R} \mapsto [0, +\infty]$ be a proper convex function with $\varphi(1) = 0$ and such that its domain $\textrm{dom}\varphi = \left\lbrace   x \in \mathbb{R} \;\; \textrm{such that} \;\; \varphi(x) < \infty \right\rbrace := (a_{\varphi},b_{\varphi})$ with $a_{\varphi} < 1 < b_{\varphi}$. The $\varphi$-divergence between $Q$ and $P$ is defined by:
\begin{equation}
D_{\varphi}(Q,P) = \int_{\mathbb{R}^r}{ \varphi\left( \dfrac{dQ}{dP}(x) \right)dP(x)},
\label{eqn:PhiDivergence}
\end{equation}
\n where $\dfrac{dQ}{dP}$ is the Radon-Nikodym derivative. When $Q$ is not a.c.w.r.t. $P$, we set $D_{\varphi}(Q,P) = + \infty$. When, $P = Q$ then $D_{\varphi}(Q,P) = 0$. Furthermore, if the function $x \mapsto \varphi(x)$ is strictly convex on a neighborhood of $x=1$, then 
\begin{equation}
\label{fondamental property of divergence}
D_{\varphi}(Q,P) = 0 \;\; \textrm{ if and only if} \;\; P = Q.
\end{equation}

\n In the definition of $\varphi-$divergences, we have considered the general case of $\sigma-$finite measures. In the whole Part I (Chapters 1 and 2) of this work, we will only be interested in $\varphi-$divergences between probability measures. In Chapter 3, we will be working with finite signed measures, and finally in Chapter 4, we will be working in the frame frame of the general case of $\sigma-$finite measures.\\ 
Several standard statistical divergences can be expressed as $\varphi-$divergences; the Hellinger, the Pearson's and the Neymann's $\chi^2$, and the (modified) Kullback-Leibler. They all belong to the class of Cressie-Read (see \cite{CressieRead1984}), also known as "power divergences", defined through the generator function $\varphi_{\gamma}$ given by:
\begin{equation}
\varphi_{\gamma}(x) := \dfrac{x^{\gamma}-\gamma x + \gamma -1}{\gamma(\gamma -1)},
\label{eqn:CressieReadPhi}
\end{equation}
for $\gamma=\frac{1}{2},2,-2,0,1$ respectively\footnote{For $\gamma\in\{0,1\}$, the limit is calculated since it is not well-defined. We denote $\varphi_0(x)=-\log x + x -1$ for the case of the modified Kullback-Leibler and $\varphi_1(x) = x\log x -x + 1$ for the Kullback-Leibler.}.
More details and properties can be found in \cite{LieseVajda} or \cite{Pardo}. \\
Estimators based on $\varphi-$divergences were developed in the parametric (see \cite{Beran},\cite{LindsayRAF},\cite{ParkBasu},\cite{BroniaKeziou09}) and the semiparametric setups (see \cite{BroniaKeziou12} and \cite{AlexisGSI13}). In completly nonparametric setup, we may mention the work of \cite{KarunamuniWu} on two component mixture models when both components are unknown.\\

%%%%%%%%%%%%%%%%%%%%%%%%%%%%%%%%%%%%%%%%%%%%%%%%%%%%%%%%%%%%%%%%%%%%%%%%%%%%%%%%%%
%
%%%%%%%%%%%%%%%%%%%%%%%%%%%%%%%%%%%%%%%%%%%%%%%%%%%%%%%%%%%%%%%%%%%%%%%%%%%%%%%%%%

\subsection{General estimation based on \texorpdfstring{$\varphi$}{phi}-divergences}
Estimation based on $\varphi-$divergences consists in finding the projection of the true distribution $P_T$ on the set $\{P_{\phi},\phi\in\Phi\}$, i.e. the model. Minimum discrepancy or minimum divergence estimators are defined by:
\begin{equation}
\phi^T = \argmin_{\phi\in\Phi} D_{\varphi}(P_{\phi},P_T).
\label{eqn:DivergenceBasedEstim}
\end{equation}
This procedure was proved to be robust in the sens that a perturbation of the model in a small neighborhood of $P_T$ would result in a small perturbation in the resulting estimates, see \cite{Donoho}. \cite{Beran} has proved that $\argmin_{\phi\in\Phi} D_{\varphi}(P_{\phi},P)$ for the case of the Hellinger divergence is continuous as a function of $P$ in a Hellinger neighborhood of $P_T$. This is also translated into an automatic robustness of the Hellinger divergence for small perturbations around the true distribution in the Hellinger topology.\\

In practice, the estimation procedure (\ref{eqn:DivergenceBasedEstim}) needs to be approximated on the basis of a dataset $X_1,\cdots,X_n$ since the true distribution is unknown. When working with discrete models, $\varphi-$divergences are approximated using a direct plug-in of the empirical distribution $P_n$. This is possible because both the model and the empirical distribution are absolutely continuous with respect to each others for large $n$. Efficient and robust estimators were derived and extensively studied; see for example \cite{Simpson} and \cite{LindsayRAF}. \\

For continuous models, the empirical distribution is no longer suitable to replace directly the true distribution since the model has a continuous support. Thus, the model is not absolutely continuous with respect to $P_n$ for any $n$ and no estimation procedure can  be produced, see \cite{BroniatowskiSeveralApplic} for a discussion about this point. Authors such as \cite{Beran}, \cite{ParkBasu} and \cite{KumarBasu} proposed to simply smooth the empirical distribution using kernels, see paragraph \ref{subsec:BeranApproach}. \cite{BasuLindsay} proposed to smooth both the model and the empirical distribution; see paragraph \ref{subsec:BLapproach}. Although smoothing the model may result in a loss of information, Basu and Lindsay show, in simple models, that this loss is rather small. They also notice that there is still a difficulty in the choice of the window and the kernel for the smoothing.\\
Recently, an approach based on some convexity arguments has been proposed independently by \cite{LieseVajdaDivergence} and \cite{BroniaKeziou2006}, see paragraph \ref{subsec:ClassicalDualFormula}. In both articles, the authors provide similar "supremal" representations of $\varphi-$divergences where a simple plug-in of the empirical distribution is possible without any smoothing techniques. The resulting estimators were called as minimum dual $\varphi-$divergence estimators (MD$\varphi$DE). Another estimator based on the dual formula called the dual $\varphi-$divergence estimator (D$\varphi$DE) was proposed. This estimator is proved to be consistent by \cite{BroniatowskiKeziou2007}, see paragraph \ref{sec:DphiDE} for more details. Since the introduction of the MD$\varphi$DE, no complete study about its robustness was proposed except for the calculus of the influence function in \cite{TomaBronia} and \cite{BroniatowskiSeveralApplic}. There were no simulation studies either, except for the paper of \cite{Frydlova}. However, in the later, the authors have considered only the case of Gaussian model where the MD$\varphi$DE coincides with the maximum likelihood estimator. \cite{Broniatowski2014} has proved that the MD$\varphi$DE coincides with the MLE on any regular exponential family, hence on a Gaussian model. Hence, the simulation results of \cite{Frydlova} shows only that known fact that the MLE is not robust.\\
The dual representation proposed by both \cite{LieseVajdaDivergence} and \cite{BroniaKeziou2006} yields estimators which perform well under the model and have efficiency comparable to the MLE\footnote{The MD$\varphi$DE is even as efficient as the MLE in regular exponential families.}. Weak and strong consistency is reached under classical conditions (see \cite{BroniatowskiKeziou2007}). Limit laws of the MD$\varphi$DE and the estimated divergence are simple and were exploited to build statistical tests. However, when we are not under the model, this approach suffers from lack of robustness. Under contamination or under misspecification, this approach does not approximate well the $\varphi-$divergence between the true distribution and the model. It even remarkably underestimates its value. We propose in the sequel a brief explanation of this problem and provide a general solution, see paragraph \ref{subsec:DefautsClassicalDualForm}. A new robust estimator called kernel-based MD$\varphi$DE is introduced. Our estimator avoids the supremal form of the MD$\varphi$DE, see paragraph (\ref{subsec:KernelSolution}). We study asymptotic properties of this estimator in Section \ref{sec:AsymptotProper}.\\
%%%%%%%%%%%%%%%%%%%%%%%%%%%%%%%%%%%%%%%%%%%%%%%%%%%%%%%%%%%%%%%%%%%%%%%%%%%%%%%%%%%%%%%%%
%
% ========================================================================
%
% ========================================================================
%
%%%%%%%%%%%%%%%%%%%%%%%%%%%%%%%%%%%%%%%%%%%%%%%%%%%%%%%%%%%%%%%%%%%%%%%%%%%%%%%%%%%%%%%%%

\section{Estimation based on \texorpdfstring{$\varphi-$}{phi-}divergences in continuous models}
In what follows, we suppose to have an i.i.d. sample $Y_1,\cdots,Y_n$ drawn from the probability distribution $P_T$. The function $K$ will denote a kernel function defined on $\mathbb{R}^r$ not necessarily symmetric. In this section, we present two general approaches to approximate a $\varphi-$divergence on the basis of a given sample.
%%%%%%%%%%%%%%%%%%%%%%%%%%%%%%%%%%%%%%%%%%%%%%%%%%%%%%%%
\subsection{Beran's approach: Smoothing of the empirical distribution}\label{subsec:BeranApproach}
A simple and natural approach to approximate the $\varphi-$divergence between the true distribution of the data $P_T$ and the model is to replace $P_T$ by a smoothed version of the empirical distribution $P_n$, say $K_{n,w}$ with $w$ a smoothing parameter. An estimator of $\phi^T$ is then given by:
\begin{equation}
\hat{\phi} = \argmin_{\phi\in\Phi} \int{\varphi\left(\frac{p_{\phi}}{K_{n,w}}\right)(y)K_{n,w}(y)dy}.
\label{eqn:BeranEstimator}
\end{equation}
This method was first introduced in the context of $\varphi-$divergences by \cite{Beran} who studied the Hellinger divergence in a univariate context and proved that it is robust and asymptotically efficient in the same time. It was then generalized to the class of $\varphi-$divergences in the univariate context, see \cite{ParkBasu} and \cite{KumarBasu}. The Hellinger divergence has very favorable properties. Indeed, \cite{Jimenez} proved that no minimum power-divergence estimator performs better than the minimum Hellinger in terms of both second order efficiency and robustness. The idea of Beran was also employed in the estimation of the proportion of a nonparametric mixture model, see \cite{KarunamuniWu}. In their approach, however, we suppose to have three i.i.d. samples; a sample drawn according to each component and a sample drawn from the whole mixture. See also \cite{TangRegression} for an application on finite mixture regression models and the references therein. \\
In the multivariate context, \cite{TamuraBoos} have studied the asymptotic properties of the Hellinger divergence. Surprisingly, the estimator needs a correction term in order to converge to a multivariate Gaussian distribution at a $\sqrt{n}$ speed. In the univariate case, this correction term does not exist since it converges to zero in probability when multiplied by $\sqrt{n}$.\\

Asymptotic properties of the resulting estimators were only studied in the previous references when $K_{n,w}$ is the Parzen-Rosenblatt kernel density estimator, i.e. a symmetric kernel density estimator. In the context of nonnegative supported distributions, the use of symmetric kernels is not advised especially if there is a considerable mass near zero which is the case for example of the exponential distribution. Several techniques for bias correction were proposed, see \cite{BiasCorrSurvey} for a survey. We mention also asymmetric kernels, see for example \cite{Libengue} for a general approach. We give in the next paragraph more details especially about asymmetric kernels and provide some examples. These two solutions provided considerable improvement in the estimation of densities defined on the half real line. To the best of our knowledge, the use of asymmetric kernels in parametric estimation has not been considered in the literature. This is may be because asymmetric kernels is still a recent topic and the first paper goes back to \cite{Chen}. Moreover, the theory is not sufficiently developed yet. Indeed, consistency of asymmetric kernel density estimators is only proved on every compact subset of the domain of definition of the true density and not on the whole domain. Consistency becomes more difficult to prove when the density explodes for example near zero. Furthermore, the rules for the choice of the window are not very efficient as we will see in the simulations in Section \ref{sec:Simulations}.\\
\begin{remark}
Unlike symmetric kernels, generalization of asymmetric kernels to the multivariate case is not simple. So far, and to the best of our knowledge, there is only two recent papers which treat the multivariate case \cite{BouezmarniMultivariate} and \cite{FunkeMultivariate}. The two papers suppose that the data is bounded. Both methods can be applied in our kernel-based MD$\varphi$DE and in Beran's method (and its generalization), but more investigations are needed in order to employ them in the Basu-Lindsay approach.
\end{remark}
%%%%%%%%%%%%%%%%%%%%%%%%%%%%%%%%%%%%%%%%%%%%%%%%%%%%%%%%%%%%%%%%%%%%%%%
%%%%%%%%%%%%%%%%%%%%%%%%%%%%%%%%%%%%%%%%%%%%%%%%%%%%%%%%%%%%%%%%%%%%%%%
\subsection{The Basu-Lindsay approach: Smoothing the model}\label{subsec:BLapproach}
The idea of smoothing the empirical distribution was applied to avoid the problem of absolute continuity of the model with respect to $P_n$ when we use the later to replace the true distribution in (\ref{eqn:PhiDivergence}). \cite{BasuLindsay} argue that the use of this method requires consistency and rates of convergence for the kernel estimator. Thus, they propose to smooth not only the empirical distribution, but also the model. Indeed, smoothing equally the model $p_{\phi}$ and the empirical measure $P_n$, as in (\ref{eqn:BasuLindsayDiv}) here below, may reduce the influence of the choice of the \emph{window} on the resulting estimator. For example, if the smoothing is by convolution with a symmetric kernel $K$ such as the Gaussian kernel, the Basu-Lindsay approach is summarized in the following two lines:
\begin{eqnarray}
p_{\phi}^*(x) & = & \frac{1}{w}\int_{\mathbb{R}}{p_{\phi}(y) K\left(\frac{x-y}{w}\right)dy}; \nonumber\\
\hat{\phi} & = & \arginf_{\phi\in\Phi} \int_{\mathbb{R}}{\varphi\left(\frac{p_{\phi}^*(x)}{K_{n,w}(x)}\right)K_{n,w}(x)dx},
\label{eqn:BasuLindsayDiv}
\end{eqnarray}
where $K_{n,w}(x) = \frac{1}{nw}\sum{K\left(\frac{x-y_i}{w}\right)}$ is the Parzen-Rosenblatt symmetric-kernel estimator. For example, in the Gaussian model $\mathcal{N}(\mu,\sigma^2)$, the smoothed model is merely a Gaussian density with variance equal to $\sigma^2+h^2$. Thus, the Basu-Lindsay approach appears as if we are calculating a divergence between a \emph{weighted} version of the model and the kernel estimator.\\
The authors prove the robustness of (\ref{eqn:BasuLindsayDiv}) using the residual adjustment function (RAF), see \cite{LindsayRAF}, since the corresponding influence function is generally unbounded, keeping first order efficiency in hand. The basic problem from a theoretical point of view is that in order to eliminate the role of the smoothing window, one needs to find what the authors call a \emph{transparent} kernel\footnote{The transparency assumption here means that the smoothed score function (derivative of the log-likelihood) is proportional to the non smoothed one. The proportion rate can only be a function of the parameters.}. This is a very hard task in general as has already been mentioned in \cite{KumarBasu} for example. Basu and Lindsay have only provided three simple examples (Gaussian, Poisson and gamma) where one can provide a transparent kernel but have not shown any leads for a general method. They have also shown in simple examples that when we use non transparent kernels, loss of information is not large. Besides, consistency and rates of convergence for the kernel estimator become necessary in order to obtain the consistency of the resulting estimator.\\
We will show in the following paragraph that if we are working with non classical situations such as densities defined on $[0,\infty)$, we may encounter further difficulties in the Basu-Lindsay approach.
%%%%%%%%%%%%%%%%%%%%%%%%%%%%%%%%%%%%%%%%%%%%%%%%%
\subsubsection{Smoothing-the-model's effect: symmetric versus asymmetric kernels}\label{subsec:SmoothingModelEffect}
The Basu-Lindsay approach seems to be more sensitive to the choice of the \emph{kernel} than standard methods. For example, let's take the case of densities defined on $(0,\infty)$ (with zero possibly included). Simple examples of such distributions are Weibull distributions and generalized Pareto distributions (GPDs). It is well-known that estimation based on symmetric kernels is biased near zero. Thus, smoothing the model with such kernels will result in similar bias near zero. Figure \ref{fig:SmoothingeffectGamma} shows the influence of a Gaussian kernel on a GPD model. The smoothed model has a peak near zero and decreases then towards zero, and hence largely underestimates the values of the "not smoothed" model near zero. Thus, the divergence calculates a distance between a biased estimator of the true distribution and a biased model, and there is no intuitive guarantee of what should give the minimization of such function. Standard methods which do not smooth the model would suffer less from this sort of problems since the bias is only in the kernel estimator.\\
Simulation results in Section \ref{sec:Simulations} show that among the three methods which use a kernel estimator (Beran's approach, the Basu-Lindsay approach and our kernel-based MD$\varphi$DE which will be introduced later on) the Basu-Lindsay approach is the most sensitive one. Under the model, all three methods do not give satisfactory results in comparison to the MLE (or the classical MD$\varphi$DE which will be presented later on) when we use symmetric kernels. When outliers are present, they still give a better result than the MLE.\\
\begin{figure}[h]
\centering
\includegraphics[scale=0.5]{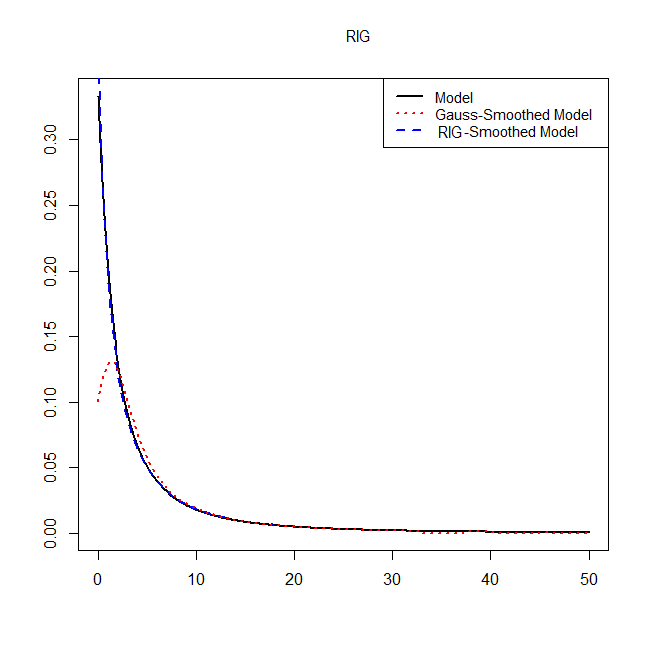}
\caption{Smoothing the model with a Gaussian kernel results in a great loss in information. The use of an asymmetric kernel such as the the reciprocal inverse Gaussian (RIG) seems to be a good alternative}
\label{fig:SmoothingeffectGamma}
\end{figure}

\noindent The solution for the previous problem is of course to either use a bias-correction method, see \cite{BiasCorrSurvey}, or to use asymmetric kernels which do not suffer from the boundary bias, see \cite{Libengue}. A more intriguing example is a Weibull distribution with shape parameter in $(0,1)$. The density function explodes to infinity as we approach from zero\footnote{Of course, if we are defining the Weibull distribution with a location parameter, the pdf explodes to infinity near the value of the location parameter.}. Cases such as GPD models can be treated efficiently using bias-correction methods since the support needs to be semi-closed. Models which have singularities such as the Weibull model can be treated using asymmetric kernels such as gamma kernels or reciprocal inverse Gaussian kernels\footnote{Asymmetric kernels have an attractive property that they can treat both bounded and unbounded densities.}. Such kernels could be employed to recover a good performance in the Basu-Lindsay approach.\\ 

Let's see how this kind of solution can be applied on the Basu-Lindsay approach. We discuss only the case of asymmetric kernels since similar arguments hold for bias-correction methods. Let $\hat{f}$ be the asymmetric-kernel estimator defined by:
\[\hat{f}(x) = \frac{1}{n c(y_1,\cdots,y_n)}\sum_{i=1}^n{K_{x,w}(y_i)},\]
where $K_{x,w}$ is the asymmetric kernel calculated at observation $y_i$, and $c(y_1,\cdots,y_n)$ is a constant which ensures integrability to 1. For example, $K$ is the gamma kernel:
\[K_{x,w}(y) = \frac{y^{x/w}}{\Gamma(1+x/w) h^{1+x/w}} e^{-y/w}, \qquad \text{for } y\in[0,\infty),\]
where $\Gamma$ is the classical gamma function. Estimator $\hat{f}$ can no longer be defined as the convolution between the asymmetric kernel and the empirical distribution in the same way as symmetric ones. Thus, the smoothed model in the Basu-Lindsay approach can no longer be obtained by simple convolution. It is given by:
\[p_{\phi}^*(x) = \int_{0}^{\infty}{\frac{1}{c(y)}K_{x,w}(y)p_{\phi}(y)dy},\]
where $c(y)$ is a function which normalizes the kernel for each value of $y$ in order to be a density. It is given by:
\[c(y) = \int_{0}^{\infty}{K_{z,w}(y)dz}.\]
Unfortunately, this normalization function cannot be calculated but numerically. Taking into account the number of integrations needed to perform such a task and the calculus of the $\varphi-$divergence afterwards which also needs numerical integration, we get a high complexity and execution time. In comparison to the classical approach of smoothing only the empirical distribution (\cite{KumarBasu}), the calculus of the smoothed model imposes two extra embedded integrals making the calculus of the $\varphi-$divergence very difficult on two levels. The first one is the execution time, and the second one is the subtlety of the whole calculus since all these integrals are carried out over slow decreasing functions on the half real line\footnote{The calculus of bounded integrals is far more simple than infinite integrals. Besides, a slow decreasing function (at the border of the its domain), even if it is smooth, is harder to be handled by numerical integration methods than fast decreasing ones.}.
\begin{remark}
We were unable to use asymmetric kernels in the Basu-Lindsay approach, because integration calculus (three embedded ones) failed even when restricting the calculus of the normalizing function $c(y)$ on a finite interval. The execution time using the statistical tool \cite{Rtool} on an i7 laptop with 8G RAM took 12 minutes for a simple calculus of the smoothed model. One can imagine now the execution time of the $\varphi-$divergence and finally the optimization over $\phi$. The method should work if one can handle efficiently the problem of numerical integrations and give close results to the case when we do not smooth the model.
\end{remark}
\begin{remark}
The use of the normalization function is necessary to get a very small loss of information. If it is not used, there will be a similar underestimation near zero to the case of symmetric kernels when applied on models defined on a semi-closed intervals.\\
\end{remark}
%%%%%%%%%%%%%%%%%%%%%%%%%%%%%%%%%%
\subsubsection{The Varying Kernel Density Estimator}
Very recently, \cite{vKDE} have proposed a kernel-type estimator which does not contain a normalization function. Their approach is based on the so called Mellin transform to approximate the distribution function and then derive an estimate of the density function. Let $\mathbb{F}$ be a cdf and define the following operator:
\[\left(\mathcal{M}\mathbb{F}\right)(j) = \int_{0}^{\infty}{t^{-j}d\mathbb{F}(t)}=\mu_j, \qquad \text{for }j=0,1,\cdots\]
Introduce the sequence of operators $\mathcal{M}_{\alpha}^{-1}$:
\begin{equation}
\left(\mathcal{M}_{\alpha}^{-1}\mu\right)(x) = 1 - \sum_{k=0}^{\alpha}\frac{(\alpha x)^k}{k!}\sum_{j=k}^{\infty}\frac{(-\alpha x)^{j-k}}{(j-k)!}\mu_j,\qquad x\in\mathbb{R}_+.
\label{eqn:MellinInverse}
\end{equation}
Here $\mu=\{\mu_j,j=0,1,\cdots\}$ and $\alpha\rightarrow\infty$ at a specific rate. The transform $\mathcal{M}\mathbb{F}(1-z)$ where $z$ is a complex variable, is known as the Mellin transform. It is possible to recover a function from its Mellin Transform, see for example \cite{Tagliani}. Under some conditions, we may write:
\[\mathbb{F}_{\alpha} = \mathcal{M}_{\alpha}^{-1}\mathcal{M}\mathbb{F} \xrightarrow[\text{weakly}]{\alpha\rightarrow\infty} \mathbb{F}.\]
Inserting the empirical moments in (\ref{eqn:MellinInverse}), we may construct an estimator of $\mathbb{F}$ as follows:
\[\tilde{\mathbb{F}}(x) = 1 - \frac{1}{n}\sum_{i=1}^n\sum_{k=0}^{\alpha}{\frac{1}{k!}\left(\frac{\alpha}{X_i}x\right)^k\exp\left(-\frac{\alpha}{X_i}x\right)}, \qquad x\in\mathbb{R}_+.\]
Notice that the inner sum in the previous display tends to $\ind{X_i>x}$ as $\alpha\rightarrow\infty$ which means that $\tilde{\mathbb{F}}$ approximates indeed $\mathbb{F}$. Note also that the estimator $\tilde{\mathbb{F}}$ is derivable so that an estimator of the density can be deduced directly by derivation as follows:
\begin{equation}
\hat{f}_{\alpha}(x) = \frac{1}{n}\sum_{i=1}^n{\frac{1}{y_i}\frac{1}{\Gamma(\alpha)}\left(\frac{\alpha x}{y_i}\right)^{\alpha}\exp\left(-\frac{\alpha x}{y_i}\right)},
\label{eqn:MTKDE}
\end{equation}
for a "bandwidth" $\alpha\in\mathbb{N}^*$. This is called by \cite{vKDE} the varying kernel density estimator (vKDE). This estimator is different from the estimators defined based on symmetric or asymmetric kernels as explained by the authors. They provide a bias-corrected version of this estimator to reduce the bias at the boundary in practice. Nevertheless, we prefer to use (\ref{eqn:MTKDE}) because it integrates to 1 and the Basu-Lindsay approach can be performed more efficiently and reasonably in comparison to the use of asymmetric kernels when working with distributions defined on $\mathbb{R}_+$. The parameter $\alpha$ is a natural number, and (\ref{eqn:MTKDE}) is $L1$--consistent as $\alpha$ goes to infinity under suitable conditions. It even achieves the optimal rate of convergence for MSE and MISE. \\
It is important to notice that $\hat{f}_{\alpha}(0)=0$ for $\alpha\geq 1$. Thus, it is preferable in the context of density estimation to be used for densities which have value equal to 0 at 0 or for densities which are defined on $(0,\infty)$. However, in kernel-based estimation procedures, the value at zero is not important because it disappears in integration calculus. Besides, no observation will have exactly the value zero. Thus, (\ref{eqn:MTKDE}) can still be used in a parameter estimation procedure even if we are working with densities not well defined on zero or have a positive value at zero.

%%%%%%%%%%%%%%%%%%%%%%%%%%%%%%%%%%%%%%%%%%%%%%%%%%%%%%%%%%%%%%%%%%%%%%%%%%%%%%%%%%%%%%%%%
%
% ========================================================================
%
% ========================================================================
%
%%%%%%%%%%%%%%%%%%%%%%%%%%%%%%%%%%%%%%%%%%%%%%%%%%%%%%%%%%%%%%%%%%%%%%%%%%%%%%%%%%%%%%%%%

\section{A plug-in estimate: the dual formula of \texorpdfstring{$\varphi-$}{phi-}divergences}\label{sec:introductorysec}
\subsection{The minimum dual \texorpdfstring{$\phi$}{phi}-divergence estimator}\label{subsec:ClassicalDualFormula}
\cite{LieseVajdaDivergence} propose the following "supremal" representation of $\varphi-$divergences. Let $\mathcal{P}$ be a class of mutually absolutely continuous distributions such that for any triplet $P,P_T$ and $Q$, $\varphi'(dP/dQ)$ is $P_T$-integrable. Theorem 17 in \cite{LieseVajdaDivergence} states that:
\begin{equation}
D_{\varphi}(P_T,P) = \sup_{Q\in\mathcal{P}}\int{\varphi'\left(\frac{dQ}{dP}\right)dP_T} + \int{\varphi\left(\frac{dQ}{dP}\right)dP} - \int{\varphi'\left(\frac{dQ}{dP}\right)dQ}
\label{eqn:LieseVajdaRep}
\end{equation}
and the supremum is attained when $Q=P_T$.\\
\cite{BroniaKeziou2006} have also developed a similar and a more general representation of $D_{\varphi}(P,P_T)$. Let $\mathcal{F}$ be any class of $\mathcal{B}-$measurable real valued functions. Let $\mathcal{M}_{\mathcal{F}}$ be the subspace of the space of probability measures $\mathcal{M}$ defined by $\mathcal{M}_{F} = \{P\in\mathcal{M} \text{ s.t. } \int{|f|dP<\infty, \forall f\in\mathcal{F}}\}$. Assume that $\varphi$ is differentiable and strictly convex. Then, for all $P\in\mathcal{M}_{\mathcal{F}}$ such that $D_{\varphi}(P,P_T)$ is finite and $\varphi'(dP/dP_T)$ belongs to $\mathcal{F}$, the $\varphi-$divergence admits the dual representation (see Theorem 4.4 in \cite{BroniaKeziou2006}):
\begin{equation}
D_{\varphi}(P,P_T) = \sup_{f\in\mathcal{F}} \int{f dP} - \int{\varphi^*(f)dP_T},
\label{eqn:GeneralDualRep}
\end{equation}
where $\varphi^*(x)=\sup_{t\in\mathbb{R}} tx-\varphi(t)$ is the Fenchel-Legendre convex conjugate of $\varphi$. Moreover, the supremum is attained at $f=\varphi'(dP/dP_T)$.\\
When substituting $\mathcal{F}$ by the class of functions $\{\varphi'(dP/dQ)\}$, and using the property $\varphi^*(\varphi'(t)) = t\varphi'(t)-\varphi(t)$, we obtain the same representation given above in (\ref{eqn:LieseVajdaRep}). Both formulations (\ref{eqn:LieseVajdaRep}) and (\ref{eqn:GeneralDualRep}) are interesting in their own and in their proofs. The second formula gives us the opportunity to reproduce many supremal forms for the $\varphi-$divergence.\\
In a parametric setup where $dP_{\phi} = p_{\phi}dx$ for $\phi\in\Phi\subset\mathbb{R}^d$ and the true distribution generating the data is a member of the model, i.e. $P_T=P_{\phi^T}$ for some $\phi^T\in\Phi$, \cite{BroniatowskiKeziou2007} propose to use the class of functions $\mathcal{F}_{\phi} = \{\varphi'(p_{\phi}/p_{\alpha}),\alpha\in\Phi\}$. Assume that $\varphi$ and its convex dual are strictly convex. Suppose also the integrability condition:
\begin{equation}
\int{\left|\varphi'\left(\frac{p_{\phi}}{p_{\alpha}}\right)(x)\right|p_{\phi}(x)dx}<\infty,\quad \forall \alpha,\phi\in\Phi.
\label{eqn:IntegCondDualForm}
\end{equation}
Then, the dual representation of $D_{\varphi}$ in the parametric setting is now written as:
\begin{equation}
D_{\varphi}(p_{\phi},p_{\phi_T}) = \sup_{\alpha\in\Phi}\left\{\int{\varphi'\left(\frac{p_{\phi}}{p_{\alpha}}\right)(x)p_{\phi}(x)dx} - \int{\varphi^{\#}\left(\frac{p_{\phi}}{p_{\alpha}}\right)(y) p_{\phi^T}(y)dy}\right\}.
\label{eqn:ParametricDualForm}
\end{equation}
where $\varphi^{\#}(t) = t\varphi'(t)-\varphi(t)$. The idea behind this choice is that the supremum is attained when $\alpha = \phi^T$. Since $p_{\phi^T}$ is unknown, one thinks about replacing $p_{\phi^T}dy$ by the empirical distribution. This seems very natural and does not cause any problem of absolute continuity as in formula (\ref{eqn:PhiDivergence}). Moreover, no smoothing is needed. We now get the following approximation:
\begin{equation}
\hat{D}_{\varphi}(p_{\phi},p_{\phi_T}) = \sup_{\alpha\in\Phi}\left\{\int{\varphi'\left(\frac{p_{\phi}}{p_{\alpha}}\right)(x)p_{\phi}(x)dx} - \frac{1}{n}\sum_{i=1}^n{\varphi^{\#}\left(\frac{p_{\phi}}{p_{\alpha}}\right)(y_i)}\right\}.
\label{eqn:DivergenceDef}
\end{equation}
Both \cite{BroniatowskiKeziou2007} and \cite{LieseVajdaDivergence} propose to estimate the set of parameters $\phi^T$ by:
\begin{equation}
\hat{\phi}_n = \arginf_{\phi\in\Phi} \sup_{\alpha\in\Phi} \hat{D}_{\varphi}(p_{\phi},p_{\phi_T}).
\label{eqn:MDphiDEClassique}
\end{equation}
This was called by \cite{BroniatowskiKeziou2007} the minimum dual $\varphi-$divergence estimator (MD$\varphi$DE). The authors have studied the asymptotic properties and provided sufficient conditions for the consistency of this estimator. They have also built some statistical tests based on it. \cite{TomaBronia} and \cite{BroniatowskiSeveralApplic} have studied the robustness of such an estimator from an influence function (IF) point of view. The IF is unfortunately unbounded in general and does not even depend on $\varphi$ for the class of Cressie-Read functions $\varphi_{\gamma}$ presented in the introduction. This fact is still not sufficient to conclude the non robustness of the MD$\varphi$DE. It was pointed out by many authors in the context of $\varphi-$divergences that one may have an  unbounded influence function, still the resulting estimators enjoy good robustness against outliers, see \cite{Beran} for the Hellinger divergence in continuous models and \cite{LindsayRAF} for a general class of $\varphi-$divergences in discrete models. In the former paper, Beran has studied the robustness by considering the \emph{Hellinger} continuity of the approximate distribution for the estimator when the model varies in a small Hellinger neighborhood of the true distribution. In the later paper, Lindsay has studied the robustness through Pearson's residuals by introducing a new criterion called as the residual adjustment function (RAF). Robustness properties were studied through the RAF and by simulations. In the context of S-divergences, \cite{GoshSDivergence} has shown that the robustness of the resulting estimator depends on two parameters although the IF only depends on one of them.\\
So far, and to the best of our knowledge, there is not even a simulation study of the robustness of the MD$\varphi$DE although it is an estimator which, similarly to the power density estimator of \cite{BasuMPD}, does not require any smoothing or escort parameters. Besides, the asymptotic properties are proved with merely classical conditions on the model. The only simulation study is done by \cite{Frydlova} and focuses only on the Gaussian model. In their results, the MD$\varphi$DE yields similar results to the maximum likelihood estimator when no contamination is present, while they get some cases where the MD$\varphi$DE is robust under contamination, although they \emph{should not} as we will see later in paragraph \ref{subsec:DefautsClassicalDualForm}.

%
% ========================================================================
%
% ========================================================================
%
%%%%%%%%%%%%%%%%%%%%%%%%%%%%%%%%%%%%%%%%%%%%%%%%%%%%%%%%%%%%%%%%%%%%%%%%%%%%%%%%%%%%%%%%%

\subsection{The Dual \texorpdfstring{$\varphi-$}{phi}divergence estimator}\label{sec:DphiDE}
\subsubsection{General facts and comments}
The dual $\varphi-$divergence estimator (D$\varphi$DE) was defined in \cite{BroniatowskiKeziou2007} (see also \cite{KeziouThesis}) as the argument of the supremum in (\ref{eqn:DivergenceDef}). It is defined by:
\begin{equation}
\hat{\alpha}_n(\phi) = \argsup_{\alpha\in\Phi}\left\{\int{\varphi'\left(\frac{p_{\phi}}{p_{\alpha}}\right)(x)p_{\phi}(x)dx} - \frac{1}{n}\sum_{i=1}^n{\left[\frac{p_{\phi}}{p_{\alpha}} \varphi'\left(\frac{p_{\phi}}{p_{\alpha}}\right) - \varphi'\left(\frac{p_{\phi}}{p_{\alpha}}\right)\right](y_i)}\right\}
\label{eqn:DphiDE}
\end{equation}
for a given "escort" parameter $\phi$. This M-estimator is far more simple than the classical MD$\varphi$DE defined by (\ref{eqn:MDphiDEClassique}) since it needs only one optimization over $\alpha$ for a given choice of the escort parameter $\phi$. Besides, \cite{TomaBronia} proved that this estimator is robust in some scale and location models from an IF point of view, \emph{provided a suitable choice of the escort parameter}. \cite{TomaBronia}, \cite{TomaAubinTests} and \cite{KeziouThesis} built robust tests using this estimator. \\
%%%%%%%%%%%%%%%%%%%%%%%%%%%%%%%%%%%%%%%%%%%%%%%%%%%%%%%%%%%%%%%%%%%
\subsubsection{Relation with the density power divergences}
The minimum density power divergence (MDPD) was first introduced by \cite{BasuMPD}. It is defined by:
\begin{eqnarray}
\hat{\phi}_n & = & \arginf_{\phi\in\Phi} \int{p_{\phi}^{1+a}}(z) dz - \frac{a+1}{a}\frac{1}{n}\sum_{i}^n{p_{\phi}^{a}(y_i)} \nonumber \\
 & = & \arginf_{\phi\in\Phi} \mathbb{E}_{P_{\phi}}\left[p_{\phi}^a\right] - \frac{a+1}{a}\mathbb{E}_{P_n}\left[p_{\phi}^{a}\right].
\label{eqn:MDPDdef}
\end{eqnarray}
Let's look at the D$\varphi$DE for power divergences with $\gamma=-a<0$. It is given by:
\begin{eqnarray}
\hat{\alpha}_n & = & \argsup_{\alpha\in\Phi} \frac{1}{\gamma-1}\int{\frac{p_{\theta}^{\gamma}}{p_{\alpha}^{\gamma-1}}(x)dx} - \frac{1}{\gamma}\frac{1}{n}\sum_{i=1}^n{\left[\frac{p_{\theta}}{p_{\alpha}}\right]^{\gamma}(y_i)}\nonumber\\
 & = & \argsup_{\alpha\in\Phi} -\frac{1}{a+1}\int{\frac{p_{\alpha}^{a+1}}{p_{\theta}^{a}}(x)dx} + \frac{1}{a}\frac{1}{n}\sum_{i=1}^n{\left[\frac{p_{\alpha}}{p_{\theta}}\right]^{a}(y_i)}\nonumber\\
& = & \arginf_{\alpha\in\Phi} \int{\frac{p_{\alpha}^{a+1}}{p_{\theta}^{a}}(x)dx} - \frac{1+a}{a}\frac{1}{n}\sum_{i=1}^n{\left[\frac{p_{\alpha}}{p_{\theta}}\right]^{a}(y_i)}\nonumber \\
& = & \arginf_{\alpha\in\Phi}\mathbb{E}_{P_{\alpha}}\left[\left(\frac{p_{\alpha}}{p_{\theta}}\right)^a\right] - \frac{a+1}{a}\mathbb{E}_{P_n}\left[\left(\frac{p_{\alpha}}{p_{\theta}}\right)^a\right].
\label{DphiDENeg}
\end{eqnarray}
By comparing (\ref{eqn:MDPDdef}) and (\ref{DphiDENeg}), we can deduce that the D$\varphi$DE seems to be a penalized form of the MDPD. This penalization by a density $p_{\theta}$ creates a big trouble from a robustness point of view. The robustness of the D$\varphi$DE is now not only controlled by the divergence power $a=-\gamma$ but also through $p_{\theta}$. We have seen in the previous paragraph that the robustness of the D$\varphi$DE in a two-component Gaussian mixture varies according to the position of $\theta$ with respect to $\theta^T$ the true vector of parameters. The difficulty of the choice of this escort parameter constitutes the only drawback of the D$\varphi$DE in comparison to the MDPD. In \cite{BroniatowskiSeveralApplic}, the authors establish an interesting link between $\varphi-$divergences and the density power divergence (\ref{eqn:MDPDdef}), see their Theorem 4.1.2.

%%%%%%%%%%%%%%%%%%%%%%%%%%%%%%%%%%%%%%%%%%%%%%%%%%%%%%%%%%%%%%%%%%%%%%%%%%%%%%%%%%%%%%%%%
%
% ========================================================================
%
% ========================================================================
%
%%%%%%%%%%%%%%%%%%%%%%%%%%%%%%%%%%%%%%%%%%%%%%%%%%%%%%%%%%%%%%%%%%%%%%%%%%%%%%%%%%%%%%%%%

\section{Limitations of the MD\texorpdfstring{$\varphi$}{phi}DE and D\texorpdfstring{$\varphi$}{phi}DE}
\subsection{The influence of the escort parameter on the robustness of the D\texorpdfstring{$\varphi$}{phi}DE}
The IF of the D$\varphi$DE is given by (see \cite{TomaBronia}):
\[\text{IF}(y|\phi) = \left[\int{J_f(x)p_{\phi^T}(x)dx}\right]^{-1}\left[\int{\left(\frac{p_{\phi}}{p_{\phi^T}}\right)^{\gamma}(x)\nabla_{\phi} p_{\phi^T}(x) dx} - \left(\frac{p_{\phi}}{p_{\phi^T}}\right)^{\gamma}(y)\frac{\nabla_{\phi} p_{\phi^T}(y)}{p_{\phi^T}(y)}\right],\]
where:
\[f(\alpha,\phi,y) = \int{\frac{p_{\phi}^{\gamma}}{p_{\alpha}^{\gamma-1}}p_{\phi}dx} - \left[\frac{p_{\phi}}{p_{\alpha}}(y)\right]^{\gamma}.\]
Previous papers which discussed the choice of the escort parameter have either let the choice arbitrary in the region where the IF is bounded (\cite{TomaBronia}), or proposed to use robust estimates for the escort parameters (\cite{Cherfi} and \cite{Frydlova}). The first idea is very complicated since we have no idea about the true value of the parameters and a bad choice of the escort parameter even inside the region where the IF is bounded does not ensure a good result. In \cite{Frydlova} and \cite{Cherfi}, experimental results show that the D$\varphi$DE in a Gaussian model is very close to the escort parameter and coincides with the escort parameter when the later is equal to the MLE. The last fact can be easily verified following the proof of Theorem 6 in \cite{Broniatowski2014}. Indeed, one may show that the MLE is a zero of the estimating equation of the D$\varphi$DE and has a definite negative Hessian matrix of the corresponding objective function. On the other hand, the use of a robust escort parameter is not always a good idea as we will show in the following two examples.
\begin{example}
Consider a two-component Gaussian mixture model. We will give some conditions on the escort parameter in order to make the IF bounded. The first term in the influence function is a matrix which is independent of $y$ and is constant. Supposing that it is invertible, we investigate both the existence of the integral, which is also a constant, and the boundedness of the remaining term which depends on $y$. The integral exists since the the fraction is of order $e^{ax}$ whereas the derivative is of order $e^{-x^2}$. Boundedness of the IF is therefore equivalent to the boundedness of the remaining term. One can show by simple limit calculus that the escort parameter needs to verify either of the following conditions according to the value of $\gamma$:
\begin{eqnarray}
\mu_1>\mu_1^T, \quad \mu_2<\mu_2^T \qquad & \text{if } & \gamma>0; \label{eqn:MixGaussRobustCond1}\\
\mu_1<\mu_1^T, \quad \mu_2>\mu_2^T \qquad & \text{if } & \gamma<0, \label{eqn:MixGaussRobustCond2}
\end{eqnarray}
in order for the IF to be bounded. Simulation results show that the use of a robust escort parameter verifying the set of conditions (\ref{eqn:MixGaussRobustCond1}, \ref{eqn:MixGaussRobustCond2}) leads to a more robust parameter than the escort. However, the use of a \emph{robust} escort parameter which does \emph{not} fulfill the set of conditions (\ref{eqn:MixGaussRobustCond1}, \ref{eqn:MixGaussRobustCond2}) has a negative impact on the resulting estimator. In our simulations in Section \ref{sec:Simulations}, we have analyzed the mixture whose true set of parameters is $(\lambda^T=0.35,\mu_1^T=-2,\mu_2^T=1.5)$ where the dataset was contaminated by $10\%$ of outliers, see paragraph \ref{subsec:GaussMix} for more details. We used our new MD$\varphi$DE, defined in Section \ref{sec:NewMDphiDE}, as an escort parameter $\hat{\phi}_1$ which is robust. The divergence criterion is the Hellinger divergence which corresponds to $\gamma=0.5$. Thus, we are in the context of condition (\ref{eqn:MixGaussRobustCond1}). The new MD$\varphi$DE verifies this condition and the resulting D$\varphi$DE has a better error, see table \ref{tab:DphiDEGaussMixEx} here below. In the same table, we give another escort parameter $\hat{\phi}_2$ which is as good as the previous one based on the total variation distance (see Section \ref{sec:Simulations} for the definition), and even slightly better. If we calculate the D$\varphi$DE using the escort parameter $\hat{\phi}_2$ which clearly does not verify condition (\ref{eqn:MixGaussRobustCond1}), the error is nearly doubled. \\
\begin{table}[h]
\label{tab:DphiDEGaussMixEx}
\centering
\begin{tabular}{|c|c|}
\hline
Estimator &  Total variation \\
\hline
$\hat{\phi}_1=(\hat{\lambda}=0.349,\hat{\mu}_1=-1.767,\hat{\mu}_2=1.377)$ & 0.087 \\
$\hat{\phi}_2 = (\hat{\lambda}=0.36,\hat{\mu}_1=-2.2, \hat{\mu}_2=1.7)$  & 0.079\\
\hline
D$\varphi$DE($\hat{\phi}_1$)  & 0.076 \\
D$\varphi$DE($\hat{\phi}_2$) & 0.115 \\
\hline
\end{tabular}
\caption{The influence of a robust escort parameter on the D$\varphi$DE in a mixture of two Gaussian components. The error is calculated between the true distribution and the estimated one, see Sec. \ref{sec:Simulations}}
\end{table}
\end{example}
%%%%%%%%%%%%%%%
\begin{example}
Let $p_{\phi}$ be a generalized Pareto distribution:
\[p_{\nu,\sigma}(y) = \frac{1}{\sigma}\left(1+\nu\frac{y}{\sigma}\right)^{-1-\frac{1}{\nu}},\quad \text{for } y\geq 0.\]
The shape and the scale are supposed to be unknown and equal to $\nu^T=0.7, \sigma^T=3$. It is necessary for the IF of the D$\varphi$DE to be bounded\footnote{The IF contains an inverse of a $2\times 2$ matrix which cannot be simply calculated. Since it is a mere constant, we only discussed the other terms in the IF.} following the value of $\gamma$ to locate the shape of the escort parameter with respect to the true value of the shape parameter. If $\gamma\in(0,1)$, it is necessary for the IF to be bounded that $\nu<\nu^T$. If $\gamma<0$, then the IF can be bounded whenever $\nu>\nu^T$. Our simulation results in paragraph \ref{subsec:SimulationGPD} show that for $\gamma=0.5$, the D$\varphi$DE calculated using a robust escort parameter (our kernel-based MD$\varphi$DE) has deteriorated the performance significantly. The total variation distance corresponding to the escort parameter is 0.05 whereas the total variation distance corresponding to the D$\varphi$DE is $0.12$. 
\end{example}
The past two examples\footnote{See the remaining of the simulations for more examples.} form an opposed result to the conjecture of both articles \cite{Frydlova} and \cite{Cherfi} about the use of robust escort parameter. The use of a robust escort is a gamble and does not guarantee a better estimator than the escort itself. Thus, we are taking a great risk by using the D$\varphi$DE. Notice, finally, that the D$\varphi$DE is still more robust than the MLE and the classical MD$\varphi$DE even if the IF is not bounded.\\
There is still a remedy, but we did not consider in our simulations yet. We may use several escort parameters and calculate for each of them the corresponding D$\varphi$DE. Then, use a procedure to combine the results of such estimators. \cite{LavancierRochet} provide a way to combine several estimators in order to obtain a better one by searching for the "best linear" combination between initial estimates. 
%%%%%%%%%%%%%%%%%%%%%%%%%%%%%%%%%%%%%%%%%%%%%%%%%%%%%%%%%%%%%%%%%%%%%%%%%%%%%%%%%%%%%%%%%
%%%%%%%%%%%%%%%%%%%%%%%%%%%%%%%%%%%%%%%%%%%%%%%%%%%%%%%%%%%%%%%%%%%%%%%%%%%%%%%%%%%%%%%%%
\subsection{Lack of robustness of the MD\texorpdfstring{$\varphi$}{phi}DE}\label{subsec:DefautsClassicalDualForm}
\paragraph{Unboundedness of the IF}
The influence function of the MD$\varphi$DE is given by (see \cite{TomaBronia} or \cite{BroniatowskiSeveralApplic}):
\[\text{IF}(y) = \left[\int{\frac{\nabla_{\phi}p_{\phi^T}(x).\left(\nabla_{\phi}p_{\phi^T}(x)\right)^t}{p_{\phi^T}(x)}}dx\right]^{-1}\frac{\nabla_{\phi}p_{\phi^T}(y)}{p_{\phi^T}(y)}.\]
The matrix is constant, hence if we suppose that it is invertible, boundedness properties of the IF is determined by the fraction $\frac{\nabla_{\phi}p_{\phi^T}(y)}{p_{\phi^T}(y)}$. We will calculate this fraction in two examples; a mixture of Gaussian distributions and a mixture of Weibull distributions. The fraction is unbounded in both examples. Besides, it is immediate to see that the same conclusion holds in an exponential family model.
\begin{example}
Consider the mixture of two Gaussian components
\[p_{(\lambda,\mu_1,\mu_2)}(y) = \lambda\frac{1}{\sqrt{2\pi}} e^{-\frac{1}{2}(y-\mu_1)^2} + (1-\lambda)\frac{1}{\sqrt{2\pi}}e^{-\frac{1}{2}(y-\mu_1)^2}.\]
We have
\[\frac{\nabla_{\phi}p_{\phi^T}(y)}{p_{\phi^T}(y)} = \left[\begin{array}{c} 
\frac{e^{-\frac{1}{2}(y-\mu_1^T)^2} - e^{-\frac{1}{2}(y-\mu_2^T)^2}}{\lambda^T e^{-\frac{1}{2}(y-\mu_1^T)^2} + (1-\lambda^T)e^{-\frac{1}{2}(y-\mu_1^T)^2}} \\
\frac{\lambda^T(y-\mu_1)e^{-\frac{1}{2}(y-\mu_1^T)^2}}{\lambda^T e^{-\frac{1}{2}(y-\mu_1^T)^2} + (1-\lambda^T)e^{-\frac{1}{2}(y-\mu_1^T)^2}}\\
\frac{(1-\lambda^T)(y-\mu_2)e^{-\frac{1}{2}(y-\mu_2^T)^2}}{\lambda^T e^{-\frac{1}{2}(y-\mu_1^T)^2} + (1-\lambda^T)e^{-\frac{1}{2}(y-\mu_1^T)^2}}
 \end{array}\right] = \left[\begin{array}{c} \frac{1-e^{(\mu_2^T-\mu_1^T)y +\frac{1}{2}(\mu_1^T)^2-\frac{1}{2}(\mu_2^T)^2}}{\lambda^T + (1-\lambda^T)e^{(\mu_2^T-\mu_1^T)y +\frac{1}{2}(\mu_1^T)^2-\frac{1}{2}(\mu_2^T)^2}} \\
\frac{\lambda^T(y-\mu_1)}{\lambda^T + (1-\lambda^T)e^{(\mu_2^T-\mu_1^T)y +\frac{1}{2}(\mu_1^T)^2-\frac{1}{2}(\mu_2^T)^2}} \\
\frac{(1-\lambda)^T(y-\mu_2)}{\lambda^Te^{(\mu_1^T-\mu_2^T)y +\frac{1}{2}(\mu_2^T)^2-\frac{1}{2}(\mu_1^T)^2+1-\lambda^T}}
\end{array}\right].\]
Let's suppose that $\mu_1^T<\mu_2^T$. The first component of the previous vector is bounded at both plus and minus infinity. The second component is bounded at $+\infty$, whereas it has a $-\infty$ limit at $-\infty$. The third component is bounded at $-\infty$ whereas it has a $+\infty$ limit at $+\infty$. This shows that the IF of the MD$\varphi$DE is unbounded.
\end{example}
\begin{example}
Consider the mixture of two Weibull components:
\[p_{(\lambda,\nu_1,\nu_2)}(x) = 2\lambda\nu_1 (2x)^{\nu_1-1} e^{-(2x)^{\nu_1}}+(1-\lambda)\frac{\nu_2}{2}\left(\frac{x}{2}\right)^{\nu_2-1} e^{-\left(\frac{x}{2}\right)^{\nu_2}}.\]
We calculate the fraction $\frac{\nabla_{\nu}p_{\nu^T}(y)}{p_{\nu^T}(y)}$.
\[\frac{\nabla_{\nu}p_{\nu^T}(x)}{p_{\nu^T}(x)} = \left[\begin{array}{c} 
\frac{2\nu_1^T (2x)^{\nu_1^T-1} e^{-(2x)^{\nu_1^T}}-\frac{\nu_2^T}{2}\left(\frac{x}{2}\right)^{\nu_2^T-1} e^{-\left(\frac{x}{2}\right)^{\nu_2^T}}}{2\lambda^T\nu_1^T (2x)^{\nu_1^T-1} e^{-(2x)^{\nu_1^T}}+(1-\lambda^T)\frac{\nu_2^T}{2}\left(\frac{x}{2}\right)^{\nu_2^T-1} e^{-\left(\frac{x}{2}\right)^{\nu_2^T}}} \\ \\
\frac{2\lambda^T\left(1+\nu_1^T\log(2x)-\nu_1^T\log(2x)(2x)^{\nu_1^T}\right)(2x)^{\nu_1^T-1}e^{-(2x)^{\nu_1^T}}}{2\lambda^T\nu_1^T (2x)^{\nu_1^T-1} e^{-(2x)^{\nu_1^T}}+(1-\lambda^T)\frac{\nu_2^T}{2}\left(\frac{x}{2}\right)^{\nu_2^T-1} e^{-\left(\frac{x}{2}\right)^{\nu_2^T}}}\\ \\
\frac{\frac{1-\lambda^T}{2}\left(1+\nu_2^T\log(2x)-\nu_2^T\log(2x)(2x)^{\nu_2^T}\right)(2x)^{\nu_2^T-1}e^{-(2x)^{\nu_2^T}}}{2\lambda^T\nu_1^T (2x)^{\nu_1^T-1} e^{-(2x)^{\nu_1^T}}+(1-\lambda^T)\frac{\nu_2^T}{2}\left(\frac{x}{2}\right)^{\nu_2^T-1} e^{-\left(\frac{x}{2}\right)^{\nu_2^T}}}
 \end{array}\right].\]
The second component is clearly unbounded neither near zero (it is of order $\log(2x)$) nor at infinity (it is of order $e^{(2x)^{\nu_2^T}-(2x)^{\nu_1^T}}$). Hence the IF of the MD$\varphi$DE is unbounded for the mixture of Weibull distribution.\\
\end{example} 

\paragraph{Equality with MLE in exponential families.} An important aspect about the classical MD$\varphi$DE is that it coincides with the maximum likelihood estimator in full exponential models whenever the corresponding \emph{true} divergence $D_{\varphi}$ is finite, see \cite{Broniatowski2014}. This covers the standard Gaussian model for which \cite{Frydlova} provided clear robust properties of the MD$\varphi$DE when outliers are generated by the standard Cauchy distribution. This contradicts with the theoretical result presented in \cite{Broniatowski2014} which is an exact one and depends only on analytic arguments. We have done similar simulations and found out that numerical problems may play a role here. Generally, such problems come from numerical approximations such as numerical integration. In a Gaussian model, all integrals in (\ref{eqn:DivergenceDef}) have close formulas and easy to calculate, see \cite{Frydlova} or \cite{BroniatowskiSeveralApplic}. However, when using the standard Cauchy distribution to generate outliers, we get points with very large values superior to 100. These points participate only in the sum term in the MD$\varphi$DE (\ref{eqn:DivergenceDef}). A Gaussian density with parameters not very far from the standard ones ($\mu=0,\sigma=1$) will produce a value equal to 0 in numerical computer programs. Thus, numerical problems of the form $0/0$ would appear when calculating the sum term in (\ref{eqn:DivergenceDef}) since the summand is of the form $g(p_{\theta}/p_{\alpha})(y_i)$. If one uses simple practical solutions to avoid this, such as adding a very small value (e.g. $10^{-100}$) to the denominator or the nominator, a thresholding effect is produced and the \emph{true} fraction is badly calculated. As a result, such outliers would have practically no effect in the procedure as if they were not added, and one would obtain "forged robust estimates". The same thresholding effect does not happen in the MLE since the likelihood function does not contain any fractions. On the other hand, if one calculates the fraction using the properties of the exponential function, i.e. $p_{\theta}(y_i)/p_{\alpha}(y_i) = \exp[(y_i-\alpha)^2/2 - (y_i-\phi)^2/2]$, the MD$\varphi$DE defined by (\ref{eqn:MDphiDEClassique}) gives the same result as the maximum likelihood estimator and never better even with Cauchy contamination.\\ 
We have performed further simulations on several models which do not belong to the exponential family and found out that the MD$\varphi$DE have a very similar behavior to the MLE, see Section \ref{sec:Simulations} below. Papers such as \cite{BarronSheu} discussed how one can estimate a probability density using exponential families and proved interesting convergence rates. Such paper can explain partially our claim passing by the result in \cite{Broniatowski2014}.\\

\paragraph{Non robustness of the MD$\varphi$DE under outliers or, more generally, under misspecification can be explained.} When $P_T$ is a member of the model, the approximated dual formula converges to the $\varphi-$divergence, and the argument of the infimum to the corresponding one, as the number of observations increases, see Proposition 3.1 in \cite{BroniatowskiKeziou2007}. This result, however, does not hold when $P_T$ is not a member of the model, i.e. under contamination or misspecification. Indeed, consistency is a consequence of the following limit:
\[\hat{D}_n(P_{\phi},P_T)\rightarrow \sup_{\alpha\in\Phi}\left\{\int{\varphi'\left(\frac{p_{\phi}}{p_{\alpha}}\right)(x)p_{\phi}(x)dx} - \int{\varphi^{\#}\left(\frac{p_{\phi}}{p_{\alpha}}\right)(y) dP_T(y)}\right\},\]
together with the fact that the arginf of the left hand side converges to the arginf of the right hand side. However, the limiting quantity is the dual representation of the $\varphi-$divergence, and since the equality in (\ref{eqn:ParametricDualForm}) holds uniquely when $p_{\alpha}=dP_T/dy$ (otherwise there is inequality) then it is never attained as long as $P_T$ is not a member of the model. Moreover, the limiting quantity is a lower bound of the divergence and minimizing the former does not guarantee the minimization of the later. Figure \ref{fig:UnderEstimation} represents this idea on a standard Gaussian model where the mean is unknown and the standard deviation is fixed at 1 and is known. The true distribution is then contaminated by a Gaussian distribution $\mathcal{N}(\mu=10,\sigma=2)$. Thus $P_T$ has the density $0.9\mathcal{N}(\mu=0,\sigma=1)+0.1\mathcal{N}(\mu=10,\sigma=2)$. The model $p_{\phi}$ is a Gaussian model $\mathcal{N}(\mu,1)$. Taking the Hellinger divergence, $\varphi(t)=(\sqrt{t}-1)^2/2$, we plot the dual $\varphi-$divergence formula (\ref{eqn:ParametricDualForm}) in Fig(a) and its empirical version (\ref{eqn:DivergenceDef}) in Fig(b) using a 100-sample drawn from $P_T$. We also plot the true values of the Hellinger divergence calculated using formula (\ref{eqn:PhiDivergence}). The minimum of the dual representation is attained at approximately $\mu=1$ whereas it is attained at approximately 0 for the true divergence. The curve of the dual representation is almost all the time below the curve of the true divergence. We also included in the figures the alternative dual formula introduced in the following paragraph which overcomes this problem.
\begin{figure}[h]
\centering
\includegraphics[scale=0.48]{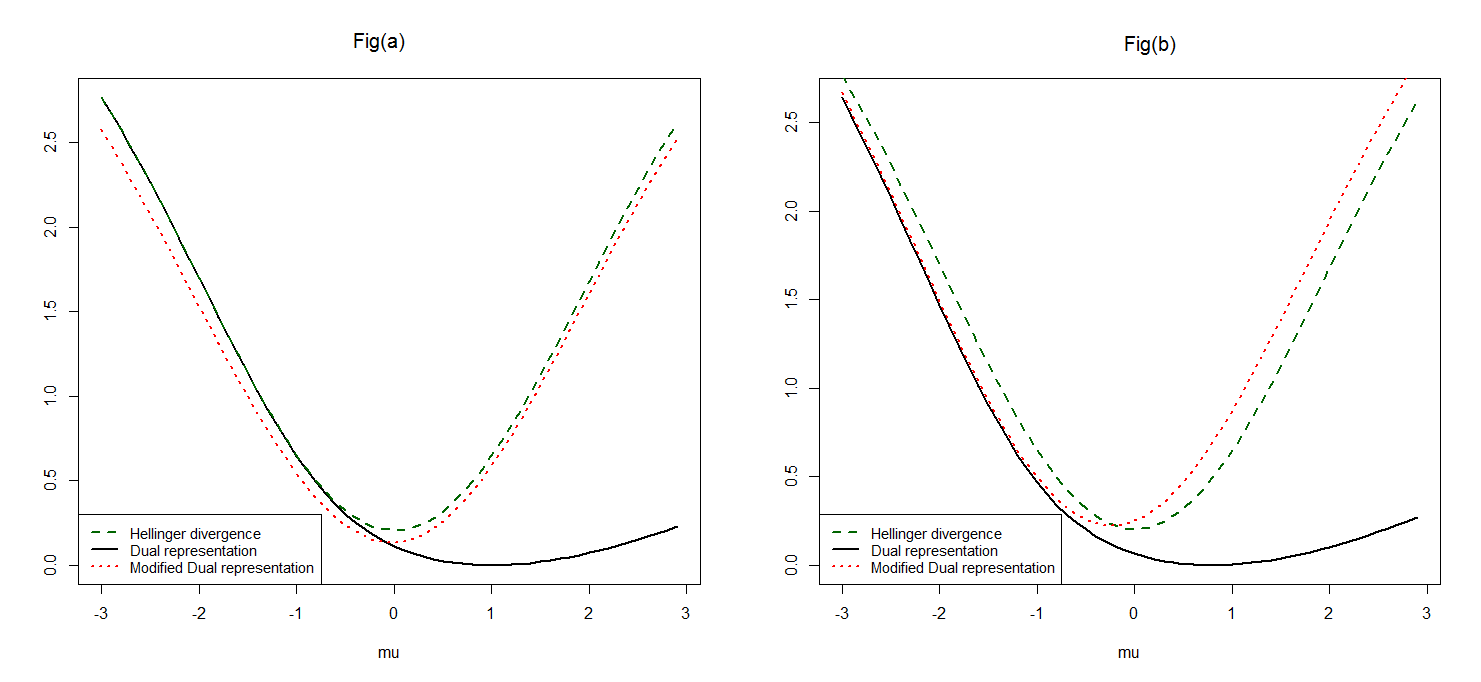}
\caption{Underestimation caused by the classical dual representation compared to the new one. The true distribution is taken to be $0.9\mathcal{N}(\mu=0,\sigma=1)+0.1\mathcal{N}(\mu=10,\sigma=2)$. Figure (a) shows the dual representation defined by (\ref{eqn:ParametricDualForm}) in comparison with the new reformulation defined by (\ref{eqn:NewExactDualForm}). Figure (b) shows the corresponding approximations when we replace the true distribution by its empirical version}
\label{fig:UnderEstimation}
\end{figure}

%%%%%%%%%%%%%%%%%%%%%%%%%%%%%%%%%%%%%%%%%%%%%%%%%%%%%%%%%%%%%%%%%%%%%%%%%%%%%%%%%%%%%%%%%
%
% ========================================================================
%
% ========================================================================
%
%%%%%%%%%%%%%%%%%%%%%%%%%%%%%%%%%%%%%%%%%%%%%%%%%%%%%%%%%%%%%%%%%%%%%%%%%%%%%%%%%%%%%%%%%

\section{A new robust estimator: kernel-based dual formula}\label{sec:NewMDphiDE}
\subsection{New reformulation of the dual representation}\label{subsec:KernelSolution}
As stated previously, if the model $(p_{\alpha})_{\alpha}$ does not contain the true distribution $p_T$, the supremum in the dual formula is no longer attained and formula (\ref{eqn:ParametricDualForm}) is no longer an identity because the right hand side underestimates the divergence $D_{\varphi}(p_{\phi},p_T)$.\\
An intuitive solution is to replace $p_{\alpha}$ by some adaptive (nonparametric) estimator of $p_T$ which does not take into account the restriction of being in the model. In the resulting dual representation the supremum is attained whether we are under the model or not and we have equality between the dual representation and the $\varphi-$divergence. This way, the resulting criterion should inherit robustness properties against possible contamination as it approximates a $\varphi-$divergence.\\
One should be able to propose many solutions which correspond to this idea in order to reach a supremal attainment in the dual representation which may vary depending on the situation. For example, if we face a proportion of large-values outliers, one may add an extra component to $p_{\alpha}$, i.e. replace $p_{\alpha}$ with the mixture $\lambda p_{\alpha} + (1-\lambda) q_{\theta}$. The extra component covers the outliers part in a smooth way. This suggestion is still very specific and treats only the case of contaminated data. Any nonparametric estimator of $p_T$ can be used whose parameters may be determined automatically in the supremum calculus\footnote{These parameters can be a window for a kernel estimator or parameters of a limited development in a suitable basis of functions, see \cite{BarronSheu} for some examples of such approaches.}. We propose here to use a kernel density estimator. In what follows $K_{n,w}$ denotes a kernel estimator of $p_T$ defined using a symmetric or asymmetric kernel with or without bias-correction treatment.\\ 
For the definition of the new estimator, let $y_1,\cdots,y_n$ be an i.i.d. sample drawn from the probability law $P_T$. The number of observations $n$ is fixed here. More formally, define the following class of functions $\mathcal{F}_{\phi,n}=\{\varphi'(p_{\phi}/K_{n,w}), w>0\}$. The dual representation is now given by:
\begin{equation}
D_{\varphi}^{\mathcal{Y}_n}(p_{\phi},p_T) = \sup_{w>0}\left\{\int{\varphi'\left(\frac{p_{\phi}}{K_{n,w}}\right)(x)p_{\phi}(x)dx} - \int{\varphi^{\#}\left(\frac{p_{\phi}}{K_{n,w}}\right)(y) p_T(y)dy}\right\}
\label{eqn:NewExactDualForm}
\end{equation}
under a similar condition to (\ref{eqn:IntegCondDualForm}) which is given by,
\begin{equation}
\int{\left|\varphi'\left(\frac{p_{\phi}}{K_{n,w}}\right)(x)\right|p_{\phi}(x)dx} < \infty, \quad \forall w>0, \forall \phi\in\Phi.
\label{eqn:IntegCondNewDualForm}
\end{equation}
We avoided to write $D_{\varphi}(p_{\phi},p_T)$ in formula ({eqn:NewExactDualForm}), because the new dual formula may not ensure equality with the $\varphi-$divergence, but only a good approximation using a sample $\mathcal{Y}_n=\{Y_1,\cdots,Y_n\}$. A "good" choice of the window $w_{\text{opt}}$ would yield\footnote{Recall that the supremum in (\ref{eqn:NewExactDualForm}) is attained at a window for which $K_{n,w}$ is as close as possible to $p_T$. Thus, a good choice of the window should result in a kernel estimator close to $p_{T}$ and could be a good guess to the argument of the supremum in equation (\ref{eqn:NewExactDualForm}).}:
\[D_{\varphi}(p_{\phi},p_T) \approx \int{\varphi'\left(\frac{p_{\phi}}{K_{n,w_{\text{opt}}}}\right)(x)p_{\phi}(x)dx} - \int{\varphi^{\#}\left(\frac{p_{\phi}}{K_{n,w_{\text{opt}}}}\right)(y) p_T(y)dy}.\]
Replace $p_T$ by its empirical version. Our final approximation is given by:
\begin{equation}
\hat{D}_{\varphi}(p_{\phi},p_T) = \int{\varphi'\left(\frac{p_{\phi}}{K_{n,w_{\text{opt}}}}\right)(x)p_{\phi}(x)dx} - \frac{1}{n}\sum_{i=1}^n{\varphi^{\#}\left(\frac{p_{\phi}}{K_{n,w_{\text{opt}}}}\right)(y_i)}.
\label{eqn:EmpiricalNewDualForm}
\end{equation}
We avoid any indexation with respect to the sample or to $n$ for the sake of clarity. Define now the new minimum dual $\varphi-$divergence estimator by:
\begin{equation}
\hat{\phi}_{n} = \arginf_{\phi\in\Phi} \int{\varphi'\left(\frac{p_{\phi}}{K_{n,w_{\text{opt}}}}\right)(x)p_{\phi}(x)dx} - \frac{1}{n}\sum_{i=1}^n{\varphi^{\#}\left(\frac{p_{\phi}}{K_{n,w_{\text{opt}}}}\right)(y_i)}.
\label{eqn:NewMDphiDE}
\end{equation}
In comparison to the MD$\varphi$DE defined by (\ref{eqn:MDphiDEClassique}), we have removed the internal optimization procedure leaving only one simple optimization which keeps our procedure at the same level of complexity as other estimation procedures such as Beran's approach (\cite{Beran}) and its generalization, and \cite{BasuMPD}.\\
An important question which arises now is: what should be the value of $w_{\text{opt}}$ since its calculus demands knowing the true distribution? In the literature on kernel estimation, there exists many rules (automatic or not) to determine sub-optimum windows such as Silverman's (or Scott's) rule-of-thumb, cross-validation methods, etc; see for example \cite{VenablesRipley} Chap 5. Figure \ref{fig:UnderEstimation} shows in a Gaussian example contaminated by a Gaussian component $\mathcal{N}(10,2)$ the use of Silverman's rule with a Gaussian kernel. The classical dual representation clearly underestimates the true divergence whereas the new reformulation stays close to it.\\
It is important to point out that the optimization problem in (\ref{eqn:NewMDphiDE}) is in general not convex. This is the case of the general class of divergence-(or disparity-)based estimators. Thus, we need to use a numerical optimization algorithm in order to calculate our kernel-based MD$\varphi$DE, see Section \ref{sec:Simulations} for more details.\\
\begin{remark}
The new MD$\varphi$DE keeps the MLE as a member of its class for the choice of $\varphi(t)=-\log(t)+t-1$. Indeed, $\varphi'(t)=-1/t+1$ and $t\varphi'(t)-\varphi(t)=\log(t)$. Thus,
\begin{eqnarray*}
\int{\varphi'\left(\frac{p_{\phi}}{K_{n,w_{\text{opt}}}}\right)(x)p_{\phi}(x)dx} & = &  1,\\
\frac{1}{n}\sum_{i=1}^n{\left[\frac{p_{\phi}}{K_{n,w_{\text{opt}}}} \varphi'\left(\frac{p_{\phi}}{K_{n,w_{\text{opt}}}}\right) - \varphi\left(\frac{p_{\phi}}{K_{n,w_{\text{opt}}}}\right)\right](y_i)} & = & \frac{1}{n}\sum_{i=1}^n{\log\left(p_{\phi}\right)- \log\left(K_{n,w_{\text{opt}}}\right)(y_i)}.\\
 %&  & \qquad \qquad \log\left(K_{n,w_{\text{opt}}}(y_i)\right).
\end{eqnarray*}
This entails that :
\begin{eqnarray*}
\hat{\phi}_{n} & = & \arginf_{\phi\in\Phi} 1 - \frac{1}{n}\sum_{i=1}^n{\log\left(p_{\phi}(y_i)\right)}+ \frac{1}{n}\sum_{i=1}^n{\log\left(K_{n,w_{\text{opt}}}(y_i)\right)}\\
 & = & \argsup_{\phi\in\Phi} \frac{1}{n}\sum_{i=1}^n{\log\left(p_{\phi}(y_i)\right)}\\
 & = & \text{MLE}.
\end{eqnarray*}
\end{remark}
\begin{remark}
In the spirit of our approach, one can write a dual formula for Beran's approach in the case of the Hellinger divergence or more generally for any $\varphi-$divergence, see paragraph \ref{subsec:BeranApproach}. Consider the class of functions $\mathcal{F}_{\phi,n}=\{\varphi'(p_{\phi}/K_{n,w}), w>0\}$, then by (\ref{eqn:NewExactDualForm}) we can write:
\begin{eqnarray*}
D_{\varphi}\left(p_{\phi},K_{n,w_0	}\right) & = & \sup_{w>0}\left\{\int{\varphi'\left(\frac{p_{\phi}}{K_{n,w}}\right)(x)p_{\phi}(x)dx} - \int{\varphi^{\#}\left(\frac{p_{\phi}}{K_{n,w}}\right)(y) K_{n,w_0}(y)dy}\right\} \\
 & = &  \int{\varphi'\left(\frac{p_{\phi}}{K_{n,w_0}}\right)(x)p_{\phi}(x)dx} - \int{\varphi^{\#}\left(\frac{p_{\phi}}{K_{n,w_0}}\right)(y) K_{n,w_0}(y)dy},
\end{eqnarray*}
where $w_0$ is a window calculated using an automatic rule as mentioned here above for $w_{\text{opt}}$. The only difference with the kernel-based dual formula (\ref{eqn:EmpiricalNewDualForm}) is that we are integrating function $\varphi^{\#}\left(\frac{p_{\phi}}{K_{n,w_0}}\right)$ with respect to the empirical distribution instead of a smoothed version of it.
\end{remark}

% -----------------------------------------------------------------------
%%%%%%%%%%%%%%%%%%%%%%%%%%%%%%%%%%%%%%%%%%%%%%%%%%%%%%%%%%%
%========================================================
%%%%%%%%%%%%%%%%%%%%%%%%%%%%%%%%%%%%%%%%%%%%%%%%%%%%%%%%%%%
% -----------------------------------------------------------------------

\section{Asymptotic properties and robustness of the new kernel-based MD\texorpdfstring{$\varphi$}{phi}DE}\label{sec:AsymptotProper}
We present in this section some of the asymptotic properties of the new MD$\varphi$DE defined by (\ref{eqn:NewMDphiDE}). We use Theorem 5.7 from \cite{Vaart} which we restate here. Consistency of the kernel-based MD$\varphi$DE means that $\hat{\phi}_n$ defined by (\ref{eqn:NewMDphiDE}) converges in probability to $\phi^T$ the true vector of parameters when we are under the model, i.e. $P_T = P_{\phi^T}$. If we are not under the model, consistency holds towards the projection of $P_T$ on the model in the sens of the divergence. In other terms, the projection $P_{\phi^T}$ is the member of the model $P_{\phi}$ whose parameters are defined by $\phi^T = \arginf_{\phi\in\Phi} D_{\varphi}(P_{\phi},P_T)$.\\ 
Similarly to \cite{BasuLindsay}, there are some cases (which are rare) such as the location Gaussian model in which consistency of the kernel-based MD$\varphi$DE does not require any condition on the kernel window. Thus, one may find simpler versions of the results we give below. We will be however interested in the general situation where the window needs to converge towards zero at a certain rate.\\
In a second part of this section, we calculate the limiting law of the new estimator under strong but standard assumptions. We, finally, calculate the influence function of the kernel-based MD$\varphi$DE for a fixed window, and show how the use of a kernel estimate in place of the model $p_{\alpha}$ in the dual formula (\ref{eqn:ParametricDualForm}) interferes to make the IF bounded. \\
We use the same notations as in \cite{Vaart} to denote integration. Thus, if $f$ is a $P-$integrable function, we denote $Pf$ the integral $\int fdP$. Moreover, $K_w*P$ denotes the operation of smoothing $dP$ by the kernel $K_w$ with bandwidth equal to $w$. This smoothing can be done by simple convolution as in the case of Rosenblatt-Parzen kernel estimator. Other kinds of smoothing are presented in Section \ref{subsec:BLapproach}. In this section only, the smoothing is supposed to be an additive operator in the sense that $K_w*(P\pm Q) = K_w*P \pm K_w*Q$.
%%%%%%%%%%%%%%%%%%%%%%%%%%%%%%%%%%%%%%%%%%%%%%%%%%%
\subsection{Consistency}
Theorem 5.7 from \cite{Vaart} permits to treat the consistency of a general class of M-estimates. It is stated as follows:
\begin{theorem}
\label{theo:VanderVaart}
Let $M_n$ be random functions and let $M$ be a fixed function of $\phi$ such that for every $\varepsilon>0$
\begin{eqnarray}
\sup_{\phi\in\Phi} |M_n(\phi) - M(\phi)| \xrightarrow[]{\mathbb{P}} 0,\label{eqn:ConsistP1}\\
\inf_{\phi:\|\phi-\phi^T\|\geq\varepsilon} M(\phi) > M(\phi^T).\label{eqn:ConsistP2}
\end{eqnarray}
Then any sequence of estimators $\hat{\phi}_n$ with $M_n(\hat{\phi}_n)\leq M_n(\phi^T) - o_P(1)$ converges in probability to $\phi^T$.
\end{theorem}
In our approach, function $M_n$ corresponds to the criterion function $P_n H(P_n,\phi)$, where $H(P_n,\phi,y)$ is defined by:
\[
H(P_n,\phi,y) = \int{\varphi'\left(\frac{p_{\phi}}{K_{w}*P_n}\right)(x)p_{\phi}(x)dx} - \varphi^{\#}\left(\frac{p_{\phi}(y)}{K_{w}*P_n(y)}\right).
\]
Function $M$ is simply defined by the \emph{expected}\footnote{In the literal sense and not mathematically.} limit in probability of $M_n$, since the Law of Large Numbers cannot be used because the average term is not a sum of i.i.d. random variables. It is given by $P_T h(P_T,\phi)$ where $h(P_T,\phi,y)$ is defined as:
\[h(P_T,\phi,y) = \int{\varphi'\left(\frac{p_{\phi}}{p_T}\right)(x)p_{\phi}(x)dx} - \varphi^{\#}\left(\frac{p_{\phi}}{p_T}\right)(y). \]
In order to prove (\ref{eqn:ConsistP1}), write:
\begin{multline}
\sup_{\phi\in\Phi} |P_nH(P_n,\phi) - P_Th(P_T,\phi)| \leq \sup_{\phi\in\Phi} |P_nH(P_n,\phi) - P_nh(P_T,\phi)| + \\ \sup_{\phi\in\Phi} |P_Th(P_T,\phi) - P_nh(P_T,\phi)|.
\label{eqn:DecompProofIneq}
\end{multline}
Now, the second supremum tends to 0 in probability by the Glivenko-Cantelli theorem as soon as $\{h(P_T,\phi), \phi\in\Phi\}$ is a Glivenko-Cantelli class of functions, see \cite{Vaart} Chap. 19 Section 2 and the examples therein. The problem then resides in finding conditions under which the first term tends to 0 in probability. The remaining of the paragraph will be concerned with the search for such conditions. In the whole section concerning the consistency of our new estimator, the window parameter $w$ is supposed to depend on $n$ in order to be able to use Theorem \ref{theo:VanderVaart} without any modification. Besides, the construction of the estimator from (\ref{eqn:NewExactDualForm}) shows the explicit link of the window with $n$.\\
We next provide a set of sufficient conditions in order for the new estimator to be consistent. We treat the general class of $\varphi-$divergences in a first theorem. The result imposes strong but standard assumptions on the model. After that, We present a result for a subclass of $\varphi-$divergences with simpler conditions on the model. Some exceptional cases may be studied separately in order to deduce simpler conditions.
\begin{remark}
Large values of $\gamma$ in absolute value are not interesting in general and may lead to practical complications. Values of $\gamma$ greater than 1 leads to integrability problems in condition (\ref{eqn:IntegCondNewDualForm}) for the new MD$\varphi$DE and in (\ref{eqn:IntegCondDualForm}) for the classical one in standard examples such as the scale Gaussian model, and the supremum in (\ref{eqn:ParametricDualForm}) is not well defined. The special case of $\gamma=2$ which corresponds to the Pearson's $\chi^2$ is included by this remark. This not very surprising. The Pearson's $\chi^2$ is a very sensitive criterion and measures the relative error committed. Thus small errors committed at values where the distribution has small values will have the same influence as the values where distribution attributes a greater density.
\end{remark}
An essential assumption will be the consistency of the kernel estimator. We refer to \cite{WiedWeibbach}, \cite{ZambomDias} or \cite{Libengue} Chap. 1 for a brief survey on symmetric kernels. When using asymmetric kernels, unfortunately consistency is proved only on every compact subset of the support of the distribution function, see \cite{Bouezmarni} or \cite{Libengue} Chap. 3 for a more general approach. Thus, our proof does not cover these kernels. \\
Assumption (\ref{eqn:ConsistP2}) in Theorem \ref{theo:VanderVaart} means that function $\phi\mapsto P_Th(P_T,\phi)$ has a unique and well separated minimum. Uniqueness is achieved when we are under the model ($P_T=P_{\phi^T}$) since function $\phi\mapsto P_Th(P_T,\phi)$ is non other than the dual representation (with the supremum calculated) of the $\varphi-$divergence $D_{\varphi}(P_{\phi},P_{\phi^T})$. Using the property that $D_{\varphi}(p_{\phi},p_{\phi^T})=0$ iff $p_{\phi}=p_{\phi^T}$, uniqueness is immediately verified as soon as the model is identifiable. If we are not under the model (misspecification), the projection of $P_T$ on the model $P_{\phi}$ may not be unique and assumption (\ref{eqn:ConsistP2}) is still needed.
%%%%%%%%%%%%%%%%%%%%%%%%%%%%%%%%%%%%%%%%
%%%%%%%%%%%%%%%%%%%%%%%%%%%%%%%%%%%%%%%%%%
\subsubsection{General Result}
We will derive in this paragraph a result which concerns the general class of $\varphi-$divergences. It was difficult to build such result without imposing strong assumptions on the model. Hereafter, simpler conditions will be proved for the particular class of Cressie-Read functions $\varphi_{\gamma}$ for $\gamma\in(-1,0)$. As mentioned here above, using inequality (\ref{eqn:DecompProofIneq}), the difficult term is first one. It is given by:
\begin{multline*}
P_nH(P_n,\phi)-P_nh(P,\phi) = \int{\left[\varphi'\left(\frac{p_{\phi}}{K_{w}*P_n}\right)-\varphi'\left(\frac{p_{\phi}}{p_T}\right)\right](x)p_{\phi}(x)dx} \\ - \frac{1}{n}\sum_{i=1}^n{\varphi^{\#}\left(\frac{p_{\phi}}{K_{w}*P_n}\right)(y_i) - \varphi^{\#}\left(\frac{p_{\phi}}{p_T}\right)(y_i)}.
\end{multline*}
The key idea is to treat each term (the integral and the sum) separately and prove its uniform convergence in probability towards 0. Another important step is to apply the mean value theorem in order to transfer the difference from functions $\varphi'$ and $\varphi^{\#}$ into a difference between the kernel estimator and the true distribution where consistency of the former is exploited. The proof of the following theorem is differed to Appendix \ref{appendix:proofHellinger}.
\begin{theorem}
\label{theo:HellingerConsist}
Assume that:
\begin{enumerate}
\item function $t\mapsto\varphi(t)$ is twice differentiable;
\item the kernel is defined on a compact and the kernel density estimator is consistent, i.e. $\sup_x\left|K_w*P_n(x) - p_T(x)\right|\rightarrow 0$ in probability;
\item the model $p_{\phi}$ is defined on a compact set and bounded independently of $\phi$, and $p_T$ is also defined on the same compact and bounded;
\item for any $\varepsilon>0$, $\inf_{\phi:\|\phi-\phi^T\|\geq\varepsilon} P_T h(P_T,\phi) > P_T h(P_T,\phi^T)$,
\end{enumerate}
then the minimum dual $\varphi-$divergence estimator defined by (\ref{eqn:NewMDphiDE}) is consistent whenever it exists.
\end{theorem}
Notice that assumption 3 is strong but stays standard. It was already used in the literature, see for example \cite{Beran}. The need to impose such restrictive assumptions stems from the nature of the dual formula which contains the quotient $p_{\phi}/K_{n,w}$ on the one hand. On the other hand, our approach which consists in using a data-based estimator results in sums of strongly dependent terms which cannot be treated in a simple and a general way. In the next paragraph, we treat a subclass of $\varphi-$divergences where we can control the terms of inequality (\ref{eqn:DecompProofIneq}) without the need to assumption 3.

%%%%%%%%%%%%%%%%%%%%%%%%%%%%%%%%%%%%%%%
%%%%%%%%%%%%%%%%%%%%%%%%%%%%%%%%%%%%%%%
\subsubsection{Case of power divergences with \texorpdfstring{$\gamma\in(-1,0)$}{gamma in (-1,0)}}
Here we have:
\begin{equation}
P_nH(P_n,\phi)-P_nh(P,\phi) = \frac{1}{\gamma-1}\int{\frac{\left(K_{w}*P_n\right)^{1-\gamma}-p_T^{1-\gamma}}{p_{\phi}^{-\gamma}}(x)dx} - \frac{1}{n\gamma}\sum_{i=1}^n{\frac{\left(K_{w}*P_n\right)^{-\gamma}-p_T^{-\gamma}}{p_{\phi}^{-\gamma}}(y_i)}.
\label{eqn:PHdifferenceNeymChi2}
\end{equation}
The key idea for proving the consistency is to use the uniform continuity of functions $t\mapsto t^{-\gamma}$ and $t\mapsto t^{(-\gamma+1)/2}$. The proof of the following theorem is differed to Appendix \ref{appendix:proofNeymChi2}.
\begin{theorem}
\label{theo:NeymChi2Consist}
For the class of power divergences defined through the class of Cressie-Read functions $\varphi_{\gamma}$ with $\gamma\in(-1,0)$, assume that:
\begin{enumerate}
\item the kernel estimator is consistent, i.e. $\sup_x|K_w*P_n(x) - p_T(x)|\rightarrow 0$ in probability;
\item $\left\{\left(\frac{p_{\phi}}{p_T}\right)^{\gamma}, \phi\in\Phi\right\}$ is a Glivenko-Cantelli class of functions;
\item there exists $n_0$ such that $\forall n\geq n_0$, the probability that the quantity 
\[\mathcal{A}_n = \sup_{\phi}\int{\frac{\left(K_w*P_n\right)^{\frac{-\gamma+1}{2}}(x) + p_T^{\frac{-\gamma+1}{2}}(x)}{p_{\phi}^{-\gamma}(x)}dx}\] 
is upper bounded independently of $n$ is greater than $1-\eta_n$ for some $\eta_n\rightarrow 0$;
\item there exists $n_0$ such that $\forall n\geq n_0$, the probability that the quantity 
\[\mathcal{B}_n = \sup_{\phi}\frac{1}{n}\sum_{i=1}^n{p_{\phi}^{\gamma}(y_i)}\] 
is upper bounded independently of $n$ is greater than $1-\eta_n$ for some $\eta_n\rightarrow 0$;
\item for any $\varepsilon>0$, $\inf_{\phi:\|\phi-\phi^T\|\geq\varepsilon} P_T h(P_T,\phi) > P_T h(P_T,\phi^T)$,
\end{enumerate}
then the minimum dual $\varphi-$divergence estimator defined by (\ref{eqn:NewMDphiDE}) is consistent whenever it exists.
\end{theorem}
This result is clearly more general than the result of Theorem \ref{theo:HellingerConsist}. We treat models which may be defined on the whole set $\mathbb{R}^d$. The assumptions are still accessible as we will demonstrate in the following example. 

%%%%%%%%%%%%%%%%%%%%%%%%%%%%%%%%%%%%%%%%%
\begin{example}
\label{example:GaussConsist}
We take a simple example of a Gaussian model with unknown mean $\phi=\mu$ which is supposed to be in a close interval $[\mu_{\min},\mu_{\max}]$. We consider power divergences for which $\gamma\in(-1,0)$. We use Theorem \ref{theo:NeymChi2Consist} to prove consistency. The Gaussian kernel is used. Assumption 1 is easily checked by considering the list of conditions in Theorem A in \cite{Silverman}. Assumption 2 holds since
\[\frac{p_{\phi}^{\gamma}}{p_{\phi^T}^{\gamma}}p_{\phi^T}(x) = e^{-\frac{1}{2}x^2 - \mu y + \frac{1}{2}\mu^2}.\]
The verification of assumption 3 is technical, so that we let it to the end. For assumption 4, in order to study $\mathcal{B}_n$, it suffices to consider the quantity $\sup_{\phi}\int{p_{\phi}^{\gamma}p_{\phi^T}}$. Indeed, the Glivenko-Cantelli theorem states that both quantities $\sup_{\phi}\frac{1}{n}\sum_{i=1}^n{p_{\phi}^{\gamma}(y_i)}$ and $\sup_{\phi}\int{p_{\phi}^{\gamma}p_{\phi^T}}$ are uniformly close for sufficiently large $n$ independently of $\phi$, hence boundedness of either of them implies boundedness of the other. We have:
\[\int{p_{\phi}^{\gamma}p_{\phi^T}} = c_2e^{-\frac{-\gamma}{1+\gamma}\frac{\mu^2}{2}}.\]
for some constant $c_2$. Here again, since $\mu$ is supposed to be in a closed interval, the previous quantity is bounded. This entails that $\mathcal{B}_n$ is bounded and assumption 4 is fulfilled. \\
We move now to assumption 5. By the dual representation of the divergence, we have $P_Th(P_T,\phi) = D_{\varphi}(p_{\phi},p_{\phi^T})$. This implies that :
\[P_Th(P_T,\phi) = \frac{1}{\gamma(\gamma-1)}e^{\frac{\gamma^2-\gamma}{2}\mu^2} - \frac{1}{\gamma(\gamma-1)}.\]
This function clearly verifies assumption 5 since it has a minimum at $\mu=0$ and this minimum is well separated. \\
We go back to assumption 3. The second term is given by:
\[\int{\frac{p_{\phi^T}^{\frac{-\gamma+1}{2}}(x)}{p_{\phi}^{-\gamma}(x)}} = c_1e^{\frac{\gamma^2-\gamma}{2(1+\gamma)}\mu^2}\]
for some constant $c_1$. It is thus bounded since $\mu$ is supposed to be in a closed interval. For the first term, let $\eta>0$. By Jensen's inequality, we may write:
\begin{eqnarray*}
\int{\frac{(K_w*P_n)^{\frac{1-\gamma}{2}}(y)}{p_{\phi}^{-\gamma}(y)}dy} & = & \frac{(2\pi)^{\frac{1-\eta}{2}}}{\eta^{3/2}} \int{\left(\frac{K_w*P_n}{p_{\phi}^{2\frac{\eta-\gamma}{1-\gamma}}}\right)^{\frac{1-\gamma}{2}}(y) \frac{1}{\sqrt{2\pi/\eta}}e^{-\eta\frac{(y-\mu)^2}{2}}dy} \\
 & \leq & \frac{(2\pi)^{\frac{1-\eta}{2}}}{\eta^{3/2}}\left(\int{\frac{K_w*P_n}{p_{\phi}^{2\frac{\eta-\gamma}{1-\gamma}}}(y) \frac{1}{\sqrt{2\pi/\eta}}e^{-\eta\frac{(y-\mu)^2}{2}}dy}\right)^{\frac{1-\gamma}{2}}\\
 & \leq & (2\pi)^{\eta/2-\gamma-1/2}e^{\left(\frac{\eta-\gamma}{1-\gamma}-\frac{\eta}{2}\right)\mu^2} \times \left(\frac{1}{nw}\right.\\
\sum_{i=1}^n& e^{-\frac{y_i^2}{2w^2}} & \left. \int{e^{-\frac{1-\gamma-2(\eta-\gamma)w^2+\eta(1-\gamma)w^2}{2w^2(1-\gamma)}y^2 + \frac{(1-\gamma)y_i-2w^2(\eta-\gamma)\mu+w^2(1-\gamma)\eta\mu}{w^2(1-\gamma)}y}dy}\right)^{\frac{1-\gamma}{2}}.
\end{eqnarray*}
We calculate each integral separately:
\begin{multline*}
\int{e^{-\frac{1-\gamma-2(\eta-\gamma)w^2+\eta(1-\gamma)w^2}{2w^2(1-\gamma)}y^2 + \frac{(1-\gamma)y_i-2w^2(\eta-\gamma)\mu+w^2(1-\gamma)\eta\mu}{w^2(1-\gamma)}y}dy} = \\  \sqrt{2\pi}w\sqrt{\frac{1-\gamma}{1-\gamma+(-\eta+2\gamma-\eta\gamma)w^2}}\exp\left[\frac{\left((1-\gamma)y_i-2w^2(\eta-\gamma)\mu+w^2(1-\gamma)\eta\mu\right)^2}{2w^2(1-\gamma)\left(1-\gamma-2(\eta-\gamma)w^2+\eta(1-\gamma)w^2\right)}\right].
\end{multline*}
We now proceed to estimate the sum over $i$. First, the only important term in the precedent integral is the one with factor $y_i^2$. Therefore, we denote in the precedent integral $c_2,c_1,c_0$ respectively the coefficients of terms $y_i^2,y_i$ and the constant term. We denote also $c$ the constant before the exponential (without the $w$). We only give the form of $c_2$:
\[c_2 = \frac{1-\gamma}{2w^2\left(1-\gamma-2(\eta-\gamma)w^2+\eta(1-\gamma)w^2\right)}.\]
We now have:
\begin{multline*}
\frac{1}{nw}\sum_{i=1}^n{e^{-\frac{y_i^2}{2w^2}} \int{e^{-\frac{1-\gamma-2(\eta-\gamma)w^2+\eta(1-\gamma)w^2}{2w^2(1-\gamma)}y^2 + \frac{(1-\gamma)y_i-2w^2(\eta-\gamma)\mu+w^2(1-\gamma)\eta\mu}{w^2(1-\gamma)}y}dy}} \\
 = c\frac{1}{n}\sum_{i=1}^n{\exp\left[\left(c_2-\frac{1}{2w^2}\right)y_i^2+c_1y_i+c_0\right]} \\
= c\frac{1}{n}\sum_{i=1}^n{\exp\left[\frac{\eta-2\gamma+\eta\gamma}{2\left(1-\gamma-2(\eta-\gamma)w^2+\eta(1-\gamma)w^2\right)}y_i^2+c_1y_i+c_0\right]}.
\end{multline*}
The final step is to use a version of the law of large numbers for independent random variables such as the two series theorem of Kolomogrov (see \citep{Feller} Chap VII, Theorem 3 page 238) since the terms of the sum do not have the same probability law, but guided by the standard Gaussian law. The general term of the sum is given by:
\[
Z_i  =  \exp\left[\frac{\eta-2\gamma+\eta\gamma}{2\left(1-\gamma-2(\eta-\gamma)w^2+\eta(1-\gamma)w^2\right)}y_i^2+c_1y_i+c_0\right]. \]
One can verify that the expectation of $Z_i$ exists as soon as the following condition is fulfilled:
\[ 0\leq \eta  <  1 \qquad \text{and} \qquad \gamma>-1.\]
Indeed, 
\begin{eqnarray*}
\mathbb{E}[Z_i] = \frac{1}{\sqrt{2\pi}}\int{\exp\left[\frac{(\eta-1)(\gamma+1) + (\eta-2\gamma+\eta\gamma)w^2}{2\left(1-\gamma-2(\eta-\gamma)w^2+\eta(1-\gamma)w^2\right)}y^2+(c_1+\mu)y+c_0-\mu^2/2\right]}.
\end{eqnarray*}
The dominating term in the integral is the one with $y^2$. It suffices then that the coefficient of $y^2$ to be negative so that the integral exists. We have
\[\frac{(\eta-1)(\gamma+1) + (\eta-2\gamma+\eta\gamma)w^2}{2\left(1-\gamma-2(\eta-\gamma)w^2+\eta(1-\gamma)w^2\right)}<0\]
if the denominator is positive and the nominator is negative. The denominator is equal to $1-\gamma + (-\eta+(2-\eta)\gamma)w^2$. Suppose that $\eta\in(0,1)$, then the denominator is positive as soon as :
\begin{equation}
w^2 < \frac{1-\gamma}{\eta-(2-\eta)\gamma}.
\label{eqn:CondWindow1}
\end{equation}
On the other hand, the nominator is equal to $(\eta-1)(\gamma+1) + (\eta-2\gamma+\eta\gamma)w^2$. If $\eta\in(0,1)$, then the nominator is negative as soon as:
\begin{equation}
w^2<\frac{(1-\eta)(1+\gamma)}{\eta-(2-\eta)\gamma}.
\label{eqn:CondWindow2}
\end{equation}
Combining this with (\ref{eqn:CondWindow1}) and since $(1-\eta)(1+\gamma)<1-\gamma$, then if $\eta\in(0,1)$, the coefficient of $y^2$ is negative as soon as (\ref{eqn:CondWindow2}) is fulfilled. Recall that, both $\eta$ and $\gamma$ are fixed values which do not depend on $n$ (the sample size). Moreover, the window $w$ need to go to zero as $n$ goes to infinity to ensure the consistency of the kernel density estimator. Thus condition (\ref{eqn:CondWindow2}) is fulfilled for any $\gamma\in(-1,0)$ as soon as $\eta\in(0,1)$.\\
The variance can be calculated similarly and proved to be finite under some condition to be identified. It suffices to calculate the second order moment. We have:
\begin{eqnarray*}
\mathbb{E}[Z_i^2] = \frac{1}{\sqrt{2\pi}}\int{\exp\left[\frac{2\eta-3\gamma+2\eta\gamma-1+(\eta-2\gamma+\eta\gamma)w^2}{2\left(1-\gamma-2(\eta-\gamma)w^2+\eta(1-\gamma)w^2\right)}y^2+(c_1+\mu)y+c_0-\mu^2/2\right]}.
\end{eqnarray*}
The coefficient of the dominating term is negative as soon as the denominator is positive (when $\eta\in(0,1)$) and the nominator is negative. This is translated into the following condition:
\[w^2<\frac{-2\eta-2\eta\gamma+3\gamma +1}{\eta-2\gamma+\eta\gamma}.\]
The right hand side is positive only if:
\[\gamma>\frac{2\eta-1}{3-2\eta}.\]
Since $\eta\in(0,1)$ and the function $\eta\mapsto \frac{2\eta-1}{3-2\eta}$ is increasing, then possible values for $\gamma$ where the variance is finite is $(-1/3,0)$. It results that for $\gamma\in[-\frac{1}{3},0)$, the Kolomogrov's two series theorem applies and the average $\frac{1}{n}\sum{Z_i}$ now converges in probability. Besides, the remaining factor $c$ also converges as $n$ goes to infinity (and $w$ goes to zero) to a constant (equal to 1).  Thus, boundedness of $\int{\frac{(K_w*P_n)^{\frac{1-\gamma}{2}}(y)}{p_{\phi}^{-\gamma}(y)}dy}$ is ensured. The argument becomes uniform on $\mu$ since it is supposed to be inside a closed interval $[\mu_{\min},\mu_{\max}]$.
All assumptions of Theorem \ref{theo:NeymChi2Consist} are now verified, and the kernel-based MD$\varphi$DE defined by (\ref{eqn:NewMDphiDE}) is consistent in the Gaussian model.
\end{example}
%%%%%%%%%%%%%%%%%%%%%%%%%%%%%%%%%%%%%%%%%%%%%%%%%%%%%
%%%%%%%%%%%%%%%%%%%%%%%%%%%%%%%%%%%%%%%%%%%%%%%%%%%%%
\subsection{Asymptotic normality}
In the literature on M-estimators, we study the asymptotic normality starting from the estimating equation, see \cite{Vaart} Chap. 5 Section 5.3. The idea, then, consists in using a Taylor expansion. Keeping the same notation as in the previous paragraph, the estimating equation has the form:
\[\nabla P_n H(P_n,\hat{\phi}) = 0.\]
We apply a Taylor expansion on $\nabla P_n H(P_n,\phi)$ between $\hat{\phi}$ and $\phi^T$.
\[\nabla P_n H(P_n,\hat{\phi}) = \nabla P_n H(P_n,\phi^T) + J_{P_n H(P_n,\phi^T)} (\hat{\phi}-\phi^T) + o_P\left(n^{-1/2}\right).\]
The left hand side is zero by definition of $\hat{\phi}$. The $o_P\left(n^{-1/2}\right)$ comes from a suitable control on the third derivatives of the objective function. If the matrix of second order derivatives $J_{P_n H(P_n,\phi^T)}$ converges in probability to an invertible matrix $J$, then, we may write:
\begin{equation}
\sqrt{n} (\hat{\phi}-\phi^T) = J^{-1} \sqrt{n} \nabla P_n H(P_n,\phi^T) + o_P(1).
\label{eqn:AsymptotMeanVal}
\end{equation}
The main problem resides in showing that $\sqrt{n} \nabla P_n H(P_n,\phi^T)$ has a multivariate Gaussian limit law. In the case of our kernel-based MD$\varphi$DE, the vector $\sqrt{n} \nabla P_n H(P_n,\phi^T)$ is not a sum of i.i.d. terms (the case of M-estimates). It contains the difficulty of the case of divergences approximated by replacing the true distribution by a kernel density estimator (\cite{Beran} or \cite{ParkBasu}). Besides, there is a sum of strongly dependent terms making the treatment of this vector very complicated in a general setup. Strong, but standard, conditions are needed in order to study the asymptotic distribution. We only study the case of univarite densities. The general case of multivariate densities is more complicated, because there should be a correction term similarly to \cite{TamuraBoos}. We leave this part to a future work.\\
The following result covers all power divergences, i.e. $\varphi-$divergences with $\varphi=\varphi_{\gamma}$ with $\gamma\in\mathbb{R}\setminus\{0,1\}$. The proof is differed to Appendix \ref{Append:TheoNormalAsyptot}.
\begin{theorem}
\label{theo:NormalAsyptotNewMD}
For the class of power divergences with $\varphi=\varphi_{\gamma}$ for $\gamma\neq \{0,1\}$, assume that:
\begin{enumerate}
\item the kernel-based MD$\varphi$DE $\hat{\phi}$ is consistent;
\item the kernel $K$ is symmetric and has a compact support where it is of class $\mathcal{C}^1$. Moreover $z^2K(z)$ is integrable;
\item the density $p_{\phi^T}$ is defined on a compact, is positive, bounded and twice derivable such that $p_{\phi^T}''$ is bounded. Moreover, there exists a neighborhood of $\phi^T$ such that the partial derivatives up to third order with respect to $\phi$ are bounded;
\item $\frac{n^{1/2}w}{-\log w}\rightarrow \infty$ and $n^{1/2}w^2\rightarrow 0$,
\end{enumerate}
then,
\begin{equation}
\sqrt{n}(\hat{\phi}-\phi^T) \xrightarrow[\mathcal{L}]{} \mathcal{N}\left(0,(2\gamma^2+1)J^{-1}S\left(J^{-1}\right)^t\right),
\label{eqn:AsymptotNormalResNewMDphiDE}
\end{equation}
where $S=\int{\nabla p_{\phi^T}\nabla p_{\phi^T}^t}$.
\end{theorem}
Assumptions on the model in this theorem are restrictive, but stay standard. Similar conditions were considered in the study of the rates of convergence of the kernel density estimation, see \cite{WiedWeibbach}. They were also considered in the study of asymptotic properties of $\varphi-$divergences, see \cite{Beran} Theorem 4. Furthermore, conditions on the bandwidth are verified for $w=n^{-1/4-\delta}$ for $\delta\in(0,1/8)$, see \cite{Bordes10} Remark 3.1. Notice that even with such strong assumptions, the proof is not simple and demands several techniques and results from the theory on kernel density estimation.\\
\begin{remark}
Under the assumptions of Theorem \ref{theo:NormalAsyptotNewMD}, consistency of the kernel-based MD$\varphi$DE is ensured by Theorem \ref{theo:HellingerConsist} provided the differentiability of $\varphi$ up to second order (assumption 1) and the uniqueness and well separability of the minimum (assumption 4). Notice that the consistency of the kernel estimator (assumption 2) is also fulfilled, see \cite{WiedWeibbach} Theorem 2.
\end{remark}
\begin{remark}
The condition on the existence and boundedness of $p_{\phi^T}''$ can be relaxed to being Lipschitz function. This demands however more assumptions on the bandwidth, see lemma 3.1 in \cite{Bordes10}. 
\end{remark}
%%%%%%%%%%%%%%%%%%%%%%%%%%%%%%%%%%%%%%%%%%%%%%%%%%%%%
%%%%%%%%%%%%%%%%%%%%%%%%%%%%%%%%%%%%%%%%%%%%%%%%%%%%%
\subsection{Influence Function for a given window}
In practice, the choice of the window is based on methods such as cross-validation, Gaussian approximations or even based on personal experience. Thus, it is interesting to study the robustness properties supposing that the window is generated by an external tool. Although in practice, we estimate the window on the basis of an observed dataset, this creates a complication in the definition and the calculation of the influence function. For this reason, we suppose that the window is fixed and independent of $n$.\\
The influence function (IF), although being limited to the existence of a noise-component, is easy to calculate in general\footnote{This is regardless of the theoretical justifications of its existence.} and gives an aspect of the robustness of an estimator whenever the IF is bounded. We derive in this paragraph the influence function of the new MD$\varphi$DE for the class of power divergences. The general case of function $\varphi$ seems to give an incomprehensible formula, and is not as interesting as the case of power divergences.\\

We recall the definition of the IF. Let $C$ be a functional which gives for a probability distribution $P$ the estimator corresponding to the argument of the infimum of $PH(P,\phi)$ defined earlier, i.e.
\[C(P) = \arginf_{\phi\in\Phi}\int{\varphi'\left(\frac{p_{\phi}}{K_{w}*P}\right)(x)p_{\phi}(x)dx} - \int{\varphi^{\#}\left(\frac{p_{\phi}(y)}{K_{w}*P(y)}\right) dP(x)}.\]
Hence, $C(P_n)$ is non other than the estimator given by (\ref{eqn:NewMDphiDE}) for a given $w$. Fisher consistency is translated by $C(P_{\phi^T})=\phi^T$. This is unfortunately not verified in general when the window is supposed to be calculated by an external tool, because the dual formula is a priori a lower bound of $D_{\varphi}(P_{\phi},P_{\phi^T})$, and we cannot be sure that it would verify the same identifiability property, i.e. $D(Q,P)=0$ iff $P=Q$ whenever $\varphi$ is strictly convex. Example \ref{example:GaussConsist} shows, however, a case where Fisher consistency is attained for any value of the window $w$.\\
The influence function measures the impact of a small perturbation in the distribution $P$ on the resulting estimator. It is hence defined by:
\[\text{IF}(P,Q) = \lim_{\varepsilon\rightarrow 0} \frac{C\left((1-\varepsilon)P+\varepsilon Q\right) - C(P)}{\varepsilon}.\]
We generally detect the influence of an outlier $x_0$ by observing what happens when we replace $P$ by $(1-\varepsilon)P + \varepsilon \delta_{x_0}$.\\ 
In the literature on M-estimates, one may derive the IF from the estimating equation. For power divergences, the estimating equation corresponding to $P$ is given by:
\begin{equation}
\frac{\gamma}{\gamma-1}\int{\frac{p_{C(P)}^{\gamma-1}\nabla p_{C(P)}}{(K_w*P)^{\gamma-1}}(x)dx} = \int{\frac{p_{C(P)}^{\gamma-1}\nabla p_{C(P)}}{(K_w*P)^{\gamma}}(x)dP(x)}, 
\label{eqn:EstimEq}
\end{equation}
where the gradient is calculated with respect to $\phi$. The influence function is obtained by "derivation"\footnote{The arginf function is a troublesome function when it comes to continuity and derivatives.} of the two sides with respect to $\varepsilon$ after having replaced $P$ by $(1-\varepsilon)P+\varepsilon Q$. Denote $J_{p_{\phi}}$ the matrix of second derivatives of $p_{\phi}$ with respect to $\phi$. The following result gives the formula of the IF for power divergences when the noise is generated by an arbitrary distribution $Q$ or when an outlier is present. The proof is differed to Appendix \ref{appendix:proofIF}.
\begin{theorem}
\label{theo:IF}
The influence function of the kernel-based MD$\varphi$DE defined by (\ref{eqn:NewMDphiDE}) for a given window is given by:
\begin{multline}
\text{IF}(P_T,Q) = \gamma A^{-1}\int{\frac{p_{C(P_T)}^{\gamma-1}\left[K_w*Q\right]\nabla p_{C(P_T)}}{(K_w*P_T)^{\gamma}}\left(1-\frac{p_T}{K*P_T}\right)(x) dx} \\ + A^{-1}\int{\frac{p_{C(P_T)}^{\gamma-1}\nabla p_{C(P_T)}}{(K_w*P_T)^{\gamma}}(x)dQ(x)}.
\label{eqn:IFGeneralDef}
\end{multline}
If $C$ is Fisher consistent, i.e. $C(P_T) = \phi^T$, then the influence function is given by:
\begin{multline}
\text{IF}(P_T,Q) = \gamma A^{-1}\int{\frac{p_{\phi^T}^{\gamma-1}\left[K_w*Q\right]\nabla p_{\phi^T}}{(K_w*P_T)^{\gamma}}\left(1-\frac{p_T}{K*P_T}\right)(x) dx} \\ + A^{-1}\int{\frac{p_{\phi^T}^{\gamma-1}\nabla p_{\phi^T}}{(K_w*P_T)^{\gamma}}(x)dQ(x)}.
\label{eqn:IFFisherConsist}
\end{multline}
Finally, if $Q = \delta_{x_0}$, then the IF is given by:
\begin{multline}
\text{IF}(P_T,x_0) = \gamma A^{-1}\int{\frac{p_{C(P_T)}^{\gamma-1}\left[K_w*\delta_{x_0}\right]\nabla p_{C(P_T)}}{(K_w*P_T)^{\gamma}}\left(1-\frac{p_T}{K_w*P_T}\right)(x) dx} \\ + A^{-1}\frac{p_{C(P_T)}^{\gamma-1}\nabla p_{C(P_T)}}{(K_w*P_T)^{\gamma}}(x_0),
\label{eqn:IFFisherConsistOutliers}
\end{multline}
where
\begin{equation}
A = \int{\left(\frac{\gamma}{\gamma-1} - \frac{p_T(x)}{K_w*P_T}\right)\frac{\left[(\gamma-1)\nabla p_{C(P_T)}\left(\nabla p_{C(P_T)}\right)^t+p_{C(P_T)}J_{p_{C(P_T)}}\right]p_{C(P_T)}^{\gamma-2}}{(K_w*P_T)^{\gamma-1}}}.
\label{eqn:IFmatrixA}
\end{equation}
\end{theorem}
\begin{remark}
The form of the IF is somewhat similar to the IF of the classical MD$\varphi$DE defined by (\ref{eqn:MDphiDEClassique}). \cite{TomaBronia} show that the IF of the classical MD$\varphi$DE is given by:
\begin{equation}
\text{IF}(P_T,x) = J^{-1}\frac{\nabla p_{\phi^T}}{p_{\phi^T}},
\label{eqn:IFMDphiDEClassique}
\end{equation}
where $J$ is the information matrix given by $\int{\frac{\nabla p_{\phi^T} (\nabla p_{\phi^T})^t}{p_{\phi^T}}}$. Going back to the IF of the new MD$\varphi$DE given by (\ref{eqn:IFFisherConsistOutliers}) and making $w$ goes to zero would replace $K_w*P_T$ by $p_{\phi^T}$. The first term thus disappears, and the IF would give $A^{-1}\frac{\nabla p_{\phi^T}}{p_{\phi^T}}$, where $A = J + \frac{1}{\gamma-1}J_{p_{\phi^T}}$ and $J_{p_{\phi^T}}$ is the matrix of second derivatives of $p_{\phi}$ with respect to $\phi$.\\
On the other hand, a general comparison for $w>0$ is also interesting. Indeed, the second term in (\ref{eqn:IFFisherConsistOutliers}) is a function of $x$ and seems to be the term guiding the boundedness of the IF. We can rewrite it as follows:
\[A^{-1}\frac{p_{C(P_T)}^{\gamma-1}\nabla p_{C(P_T)}}{(K_w*P_T)^{\gamma}}(x_0) = A^{-1}\frac{p_{C(P_T)}^{\gamma}}{(K_w*P_T)^{\gamma}}\frac{\nabla p_{C(P_T)}}{p_{C(P_T)}}(x_0).\]
A direct comparison with (\ref{eqn:IFMDphiDEClassique}) shows that our approach has resulted in the term $\frac{p_{C(P_T)}^{\gamma}}{(K_w*P_T)^{\gamma}}$ which could oblige the IF to be bounded in some cases. This is the ratio between the true density and the smoothed one. When $\gamma>0$, it is surprising that the IF becomes \emph{more} bounded as the ratio between the true distribution and the smoothed one decreases, which means that the smoothing is producing over estimation at the tail of the distribution.  
\end{remark}
\begin{example}
\label{Example:GaussIF}
We resume the univariate Gaussian example. It can be proved that in this model, the kernel-based MD$\varphi$DE is Fisher consistent. Indeed, function $P_TH(P_T,\mu)$ is given by:
\begin{multline*}
P_TH(P_T,\mu)=\frac{1}{\gamma-1}\sqrt{\frac{1+w^2}{1+\gamma w^2}} e^{-\frac{\gamma(1-\gamma)}{2(1+\gamma w^2)}\mu^2} - \frac{1}{\gamma}\sqrt{\frac{1+w^2}{(\gamma+1)w^2+1}} e^{-\frac{\gamma(w^2+1-\gamma)}{2(1+(\gamma+1)w^2)}\mu^2}\\ - \frac{1}{\gamma(\gamma-1)}.
\end{multline*}
This formula holds for $\gamma\in[-1,\infty)\setminus\{0,1\}$ whatever the value of $w$. It also holds for $\gamma<-1$ whenever $w^2\leq -\frac{1}{\gamma}$. This function has a minimum at $\mu=0$ whatever the value of $w$ when $\gamma>0$ (but different from 1), and has a local\footnote{The minimum of $P_TH(P_T,\mu)$ is $-\infty$ and is attained at $\pm\infty$, but when we restrain the set of possible parameters to a compact subset around zero, the estimation procedure becomes possible. Recall also that consistency was only proved when $\mu\in[\mu_{\min},\mu_{\max}]$, see example \ref{example:GaussConsist} in this section.} minimum at zero when $\gamma<0$ whatever the value of $w$, see figure (\ref{fig:ObjFun}).\\
\begin{figure}[ht]
\centering
\includegraphics[scale=0.26]{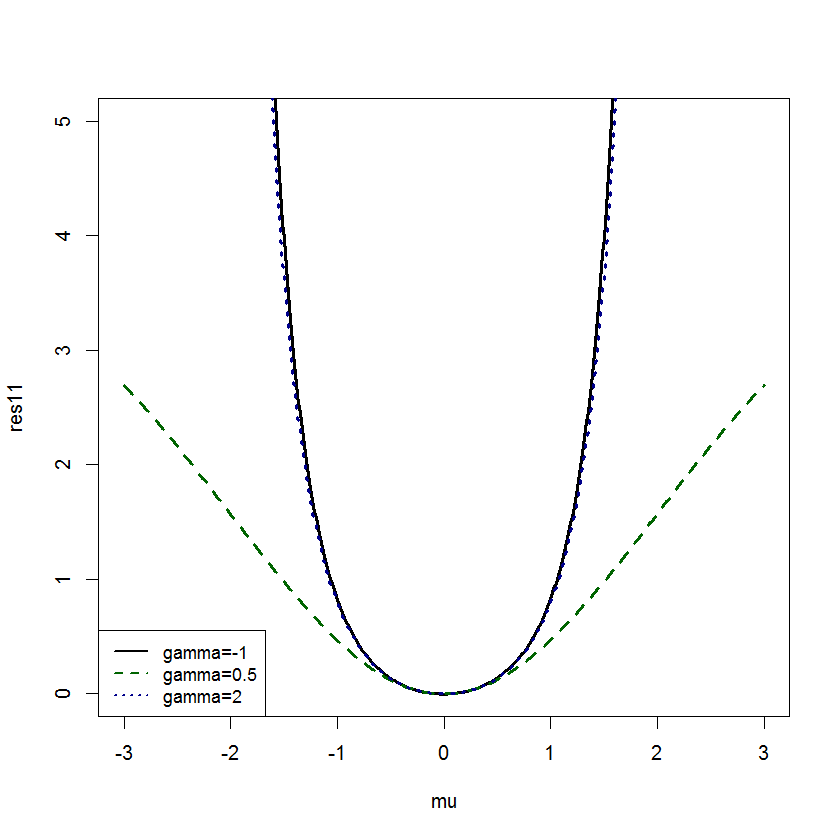}
\caption{Function $P_TH(P_T,\mu)$ for different windows and divergences. They all have an infimum at zero}
\label{fig:ObjFun}
\end{figure}

\noindent Let's calculate the IF given by (\ref{eqn:IFFisherConsistOutliers}). We leave the calculus of the matrix $A$ to the end. The second term is given by:
\begin{equation*}
\frac{p_{\phi^T}^{\gamma-1}\nabla p_{\phi^T}}{(K_w*P)^{\gamma}}(x_0) = (1+w^2)^{\gamma/2} x_0e^{-\frac{\gamma w^2}{2(1+w^2)}x_0^2}.
\end{equation*}
Hence, this quantity is bounded as soon as $\gamma>0$. The second quantity is an integral which needs to exist and be finite. We have:
\begin{multline*}
\frac{p_{\phi^T}^{\gamma-1}(x)K((x-x_0)/w)\nabla p_{\phi^T}(x)}{(K_w*P)^{\gamma}(x)}\left(1-\frac{p}{K_w*P}\right)(x) = \\ \frac{(1+w^2)^{\frac{\gamma+1}{2}}}{w} \exp\left[-\frac{\gamma w^4 + 1}{2w^2(1+w^2)}x^2 + \frac{xx_0}{w^2} -\frac{x_0^2}{2w^2}\right] \times \left(\frac{1}{\sqrt{1+w^2}}e^{-\frac{x^2}{2(1+w^2)}} - e^{-\frac{x^2}{2}}\right).
\end{multline*}
It is clear now that if $\gamma>0$, the integral exists. We should not forget that the integral term also depends on $x_0$. The dominating term is $e^{-x_0^2}$, so that the integral term is bounded as a function of $x_0$ as soon as the integral exists.\\
It remains to show that the term $A$ exists and is invertible. Since $\nabla p_{\phi^T}(x) = xe^{-x^2/2}$, and $J_{p_{\phi^T}} = (1+x^2)e^{-x^2/2}$, then:
\begin{eqnarray*}
A  & = & \sqrt{\frac{1+w^2}{2\pi}}\frac{\gamma}{\gamma-1} \left(\sqrt{\frac{2\pi}{a}} + \gamma\sqrt{\frac{2\pi}{a^3}}\right) - \frac{1+w^2}{\sqrt{2\pi}} \left(\sqrt{\frac{2\pi}{b}} + \gamma\sqrt{\frac{2\pi}{b^3}}\right),
\end{eqnarray*}
where $a = \frac{\gamma w^2 + 1}{1+w^2}$ and $b = \frac{\gamma w^2 + w^2 + 1}{1+w^2}$. It is clear that for $\gamma\in(0,1)$, the two terms constituting $A$ have the same sign, hence $A$ cannot be zero since it is the sum of two negative terms. However, if $\gamma>1$, $A$ may by zero for some cases. Indeed, $A$ is 0 whenever $\gamma^2 (1+\gamma+2\gamma w^2)^2 (1+(\gamma+1)w^2)^3 - (\gamma-1)(1+w^2)(1+\gamma+(\gamma+2)w^2)^2 = 0$. Notice that function $w\mapsto \gamma^2 (1+\gamma+2\gamma w^2)^2 (1+(\gamma+1)w^2)^3 - (\gamma-1)(1+w^2)(1+\gamma+(\gamma+2)w^2)^2$ is equal to $2\gamma-1>0$ when $w=0$, whereas it has a $-\infty$ limit at $+\infty$. Thus, it passes by zero for some $w>0$ since it is a continuous function. \\
Previous arguments permit us to conclude for sure that for $\gamma\in(0,1)$, the influence function of the estimator defined by (\ref{eqn:NewMDphiDE}) is bounded in the Gaussian model independently of the bandwidth of the Gaussian kernel. Moreover, it is unbounded for $\gamma<0$. Hence, one can hope to get a robust estimation when $\gamma\in(0,1)$, whereas further investigations are needed for the case of $\gamma<0$.
\end{example}

%\chapter{Estimation of a parametric mixture model}
%\section{Mixture models: definition and notations}

%\section{Estimation using EM algorithm}
%\subsection{Method}
%\subsection{Convergence}
%\section{Estimation using exponential series expansions}
%\subsection{Method}
%\subsection{Theoretical results}
\section{Simulation study: comparison}\label{sec:Simulations}
\paragraph{Summary of the estimation methods and error criterion.} We summarize the results of 100 experiments by giving the average of the estimates and the error committed, and the corresponding standard deviation. We consider two error criteria; the total variation distance (TVD) and the Chi square divergence between the true distribution and the estimated one. These criteria are defined as follows:
\begin{eqnarray}
\sqrt{\chi^2(p_{\phi},p_{\phi^T})} & = & \sqrt{\int{\frac{\left(p_{\phi}(y)-p_{\phi^T}(y)\right)^2}{p_{\phi^T}(y)}dy}}; \label{eqn:Chi2Error}\\ 
\text{TVD}(p_{\phi},p_{\phi^T}) & = & \sup_{A\in\mathcal{B}_n(\mathbb{R})}\left|dP_{\phi}(A) - dP_{\phi^T}(A)\right|. \label{eqn:TVDError}
\end{eqnarray}
We prefer to use the Chi square divergence, because it measures the relative error between two probability laws. Hence, the error committed on sets where the true distribution attributes small values is penalized in a similar way to sets where the true distribution attributes large values. We use also the TVD because it has the property of measuring the largest error committed when measuring a set $A$ using the estimated distribution instead of the true one. The TVD can be directly calculated using the $L1$ distance. Indeed, the Scheff\'e lemma (see \cite{Meister} page 129.) states that:
\[\sup_{A\in\mathcal{B}_n(\mathbb{R})}\left|dP_{\phi}(A) - dP_{\phi^T}(A)\right| = \frac{1}{2}\int_{\mathbb{R}}{\left|p_{\phi}(y) - p_{\phi^T}(y)\right|dy}.\]
We consider the Hellinger divergence for estimators based on $\varphi-$divergences. Our preference of the Hellinger divergence is that we hope to obtain robust estimators without loss of efficiency, see \cite{Jimenz}. The parameter vector is estimated using six methods:
\begin{enumerate}
\item Maximum likelihood (MLE) which is calculated using EM for mixture models;
\item The classical MD$\varphi$DE defined by (\ref{eqn:MDphiDEClassique});
\item Our kernel-based MD$\varphi$DE defined by (\ref{eqn:NewMDphiDE}) with different choices for the kernel and its bandwidth;
\item The Basu-Lindsay approach with different choices for the kernel and its bandwidth;
\item The dual $\varphi$--divergence estimator (D$\varphi$DE) defined by (\ref{eqn:DphiDE}) with escort parameter the result of our kernel-based MD$\varphi$DE with the best choice of the kernel and window among presented possibilities; 
\item The minimum power density estimator (MDPD) of \cite{BasuMPD} defined by (\ref{eqn:MDPDdef}) for $a\in\{0.1,0.25,0.5,0.75,1\}$.
\end{enumerate}
We give for each experiment a summary of the results with comments, and precise the used kernels and the corresponding windows choices. We finally give an overall conclusion with some practical remarks.\\
\paragraph{Practical issues:}
Optimization was done using the Nelder-Mead algorithm. Integrations calculus were done using function \texttt{distrExIntegrate} of package \texttt{distrEx} which is a slight modification of the standard function \texttt{integrate}. It performs a Gauss-Legendre quadrature when function \texttt{integrate} returns an error. We have noticed that functions such as \texttt{integral} of package \texttt{pracma}\footnote{Function \texttt{integral} includes a variety of adaptive numerical integration methods such as Kronrod-Gauss quadrature, Romberg's method, Gauss-Richardson quadrature, Clenshaw-Curtis (not adaptive) and (adaptive) Simpson's method.	}, although has a good performance, is slow. Besides, function \texttt{int} of package \texttt{rmutil}, which uses either the Romberg method or algorithm 614 of the collected algorithms from ACM, seems to underestimate the value of the integral in slightly difficult circumstances such as heavy tailed distributions. For example, when we used it to calculate the classical MD$\varphi$DE in the GPD case, it gave robust results because it underestimated the infinity part of the integral (forged thresholding effect). Finally, during some experiences on GPD observations and Weibull distributions based on the Basu-Lindsay approach, function \texttt{distrExIntegrate} failed to converge and function \texttt{integral} was used to attain a result.  \\
\paragraph{Summary of the models and presentation of the results.}
Our simulation study covers the following models:
\begin{enumerate}
\item Gaussian model with unknown mean and variance;
\item Gaussian mixture with two components where the proportion and the two means are unknown;
\item Generalized Pareto distribution with unknown shape and scale;
\item Three Weibull mixtures with two components where the proportion and the two shapes are unknown.
\end{enumerate}
Outliers were added in the original data in many ways which will be specified according to each case. We have either added noise outside the support of the dataset or by dispersing the noise over the whole dataset. We have also used different distributions to produce the noise.\\
For the first two models, we only used a Gaussian kernel with window chosen using either Silverman's rule (nrd0 in the statistical tool R) or Sheather and Jones' rule (SJ). For the heavy tailed models which are defined on the half real line, we needed to use non classical kernels such as asymmetric kernels (RIG: reciprocal inverse Gaussian and GA: gamma kernels) and the varying KDE of \cite{vKDE} denoted here as MT (Mellin transform) defined here above by (\ref{eqn:MTKDE}), followed by the value of the bandwidth $\alpha\in\{5,10,15,20\}$. In the GPD model and the first Weibull mixture, we present a simple comparison between symmetric kernels and other non classical methods and show the advantage of the later in such context. We therefore avoided using symmetric kernels for other Weibull mixtures. For the Basu-Lindsay approach, we did not implement asymmetric kernels, see discussion in paragraph \ref{subsec:SmoothingModelEffect}. We only used the varying KDE.\\
Concerning the rule for deciding the window for the non classical kernels, we have tried out the cross-validation method (CV), but it resulted always in large (small for the varying KDE) and inconvenient windows especially when outliers are inserted. We were, therefore, obliged to use fixed windows in order to obtain good results. For each kernel and method, the window value or the rule used to calculate it is written next to it. More details can be found at each paragraph.
%%%%%%%%%%%%%%%%%%%%%%%%%%%%%%%%%%%%%%%%%%%%%%%%%%%%%%%%%%%%%%%%%%%%%%
\subsection{Univariate Gaussian model}
We consider the Gaussian distribution $\mathcal{N}(\mu,\sigma^2)$ when both parameters $\mu$ and $\sigma$ are unknown. We generate at each run a 100-sample of the standard Gaussian distribution $\mathcal{N}(0,1)$. Outliers are added simply by replacing the 10 largest values in the sample by the value 10.\\
The maximum likelihood estimator of the parameters are simply the empirical mean and variance $\hat{\mu}=\frac{1}{100}\sum{y_i}, \hat{\sigma}^2 = \frac{1}{99}\sum{(y_i-\hat{\mu})^2}$. For methods which need kernels, we used a Gaussian kernel with two rules for the window; Silverman's rule and Sheather and Jones' one. We calculate the minimum density power divergence estimator (MDPD) for values of the tradeoff parameter $a\in\{0.1,0.25,0.5,0.75,1\}$. The D$\varphi$DE was calculated using the kernel-based MD$\varphi$DE as an escort with Silverman's rule. Estimation results are summarized in table \ref{tab:EstimGauss}. Estimation error is calculated in table \ref{tab:ErrGauss}.\\
When we are under the model, all compared methods give the same result with very slight differences. As we add $10\%$ outliers, the classical MD$\varphi$DE and the MLE give the same result which is positively deviated from the true mean with a large variance. This is already expected by virtue of the result of \citep{Broniatowski2014}. Other methods, ours included, give robust results except for MDPD with $a=0.1$. Our estimator (for both windows choices) is at the same level of efficiency as the MLE under the model. Besides, the window choice seems irrelevant for methods based on kernels but for Beran's method where Silverman's rule is slightly better. The MDPD seems to give the best tradeoff between efficiency and robustness for $a=0.5$ conquering other methods. The kernel-based MD$\varphi$DE and the Basu-Lindsay approaches give slightly better efficiency which is traded with slightly lower robustness in comparison to the result of MDPD with $a=0.5$.\\

\begin{table}[hp]
\centering
\begin{tabular}{|l|c|c|c|c||c|c|c|c|}
\hline
\multirow{2}{2.5cm}{Estimation method} & \multicolumn{4}{|c||}{No Outliers} & \multicolumn{4}{|c|}{$10\%$ Outliers}\\
\cline{2-9}
  & $\mu$ & sd$(\mu)$ & $\sigma$ & sd$(\sigma)$ & $\mu$ & sd$(\mu)$ & $\sigma$ & sd$(\sigma)$\\
 \hline
 \hline
\multicolumn{9}{|c|}{Hellinger} \\
\hline
\hline
Classical MD$\varphi$DE & 0.005 & 0.111 & 0.983 & 0.082 & 0.833 & 0.103 & 3.157 & 0.039\\
\hline
New MD$\varphi$DE - Silverman & 0.005 & 0.113 & 0.967 & 0.081 & -0.187 & 0.114 & 0.810 & 0.069\\
New MD$\varphi$DE - SJ & 0.005 & 0.113 & 0.973 & 0.082 & -0.191 & 0.114 & 0.800 & 0.068\\
\hline
Basu-Lindsay - Silverman & 0.005 & 0.114 & 0.968 & 0.081 & -0.191 & 0.114 & 0.805 & 0.068\\
Basu-Lindsay - SJ & 0.005 & 0.113 & 0.970 & 0.081 & -0.193 & 0.114 & 0.799 & 0.067\\
\hline
Beran - Silverman & 0.005 & 0.113 & 1.024 & 0.087 & -0.191 & 0.114 & 0.878 & 0.075\\
Beran - SJ & 0.005 & 0.112 & 1.048 & 0.089 & -0.192 & 0.114 & 0.853 & 0.073\\
\hline
\hline
MDPD 0.1 & 0.005 & 0.112 & 0.983 & 0.082 & 0.319 & 0.111 & 2.451 & 0.079\\
MDPD 0.25 & 0.006 & 0.112 & 0.983 & 0.083 & -0.145 & 0.114 & 0.854 & 0.074\\
MDPD 0.5 & 0.008 & 0.117 & 0.979 & 0.087 & -0.115 & 0.116 & 0.875 & 0.081\\
MDPD 0.75 & 0.010 & 0.123 & 0.975 & 0.093 & -0.093 & 0.120 & 0.894 & 0.089\\
MDPD 1 & 0.012 & 0.129 & 0.971 & 0.098 & -0.077 & 0.124 & 0.910 & 0.094\\
\hline
\hline
D$\varphi$DE & 0.005 & 0.112 & 0.982 & 0.082 & -0.164 & 0.114 & 0.873 & 0.080\\
\hline
\hline
MLE & 0.005 & 0.111 & 0.988 & 0.082 & 0.833 & 0.103 & 3.172 & 0.039 \\
\hline
\end{tabular}
\caption{The mean value and the standard deviation of the estimates in a 100-run experiment in the standard Gaussian model. The divergence criterion is the Hellinger divergence. The escort parameter of the D$\varphi$DE is taken as the new MD$\varphi$DE with Silverman's rule.}
\label{tab:EstimGauss}
\end{table}

\begin{table}[hp]
\centering
\resizebox{\textwidth}{!}{ 
\begin{tabular}{|l|c|c|c|c||c|c|c|c|}
\hline
\multirow{2}{2.5cm}{Estimation method} & \multicolumn{4}{|c||}{No Outliers} & \multicolumn{4}{|c|}{$10\%$ Outliers}\\
\cline{2-9}
  & $\chi^2$ & sd($\chi^2$) & TVD & sd(TVD) & $\chi^2$ & sd($\chi^2$) & TVD & sd(TVD)\\
 \hline
 \hline
\multicolumn{9}{|c|}{Hellinger} \\
\hline
\hline
Classical MD$\varphi$DE & 0.104 & 0.052 & 0.054 & 0.026 & 8.503 & 0.113 & 0.516 & 0.002\\
\hline
New MD$\varphi$DE - Silverman & 0.106 & 0.052 & 0.056 & 0.028 & 0.230 & 0.063 & 0.136 & 0.041\\
New MD$\varphi$DE - SJ & 0.105 & 0.052 & 0.055 & 0.027 & 0.239 & 0.062 & 0.141 & 0.041\\
\hline
Basu-Lindsay - Silverman & 0.105 & 0.052 & 0.055 & 0.028 & 0.235 & 0.062 & 0.139 & 0.040\\
Basu-Lindsay - SJ & 0.105 & 0.052 & 0.055 & 0.027 & 0.240 & 0.062 & 0.142 & 0.040\\
\hline
Beran - Silverman & 0.114 & 0.063 & 0.054 & 0.025 & 0.191 & 0.067 & 0.110 & 0.042\\
Beran - SJ & 0.125 & 0.076 & 0.057 & 0.026 & 0.205 & 0.066 & 0.119 & 0.042\\
\hline
D$\varphi$DE & 0.104 & 0.052 & 0.054 & 0.026 & 0.183 & 0.068 & 0.105 & 0.042\\
\hline
\hline
MDPD 0.1 & 0.104 & 0.051 & 0.053 & 0.026 & 5.772 & 0.356 & 0.411 & 0.013\\
MDPD 0.25 & 0.105 & 0.052 & 0.054 & 0.026 & 0.185 & 0.066 & 0.107 & 0.042\\
MDPD 0.5 & 0.110 & 0.054 & 0.057 & 0.028 & 0.165 & 0.068 & 0.094 & 0.042\\
MDPD 0.75 & 0.116 & 0.060 & 0.060 & 0.032 & 0.152 & 0.070 & 0.086 & 0.043\\
MDPD 1 & 0.121 & 0.066 & 0.063 & 0.036 & 0.144 & 0.070 & 0.080 & 0.043\\
\hline
\hline
MLE & 0.104 & 0.052 & 0.053 & 0.025 & 8.522 & 0.111 & 0.518 & 0.002\\
\hline
\end{tabular}}
\caption{The mean value of errors committed in a 100-run experiment with the standard deviation. The divergence criterion is the Hellinger divergence. The escort parameter of the D$\varphi$DE is taken as the new MD$\varphi$DE with Silverman's rule.}
\label{tab:ErrGauss}
\end{table}
\clearpage
%%%%%%%%%%%%%%%%%%%%%%%%%%%%%%%%%%%%%%%%%%%%%%%%%%%%%%%%%%
% ===============================================================
%%%%%%%%%%%%%%%%%%%%%%%%%%%%%%%%%%%%%%%%%%%%%%%%%%%%%%%%%%
\subsection{Mixture of two Gaussian components}\label{subsec:GaussMix}
We show in this paragraph several simulations from a two-component Gaussian mixture where the data is contaminated or not by a $10\%$ of outliers. The true values of the mixture parameters are $\lambda = 0.35, \mu_1 = -2, \mu_2 = 1.5$. The variance of both components is supposed to be known and fixed at 1. Contamination was done for the first mixture by adding in the original sample to the 5 lowest values random observations from the uniform distribution $\mathcal{U}[-5,-2]$. We also added to the 5 largest values random observations from the uniform distribution $\mathcal{U}[2,5]$. Estimation results are summarized in table \ref{tab:EstimGaussMix}. Estimation error is calculated in table \ref{tab:ErrGaussMix}. Maximum likelihood estimates are calculated using the EM algorithm. Table \ref{tab:ErrGaussMix1000Obs} contains a simulation with 1000 observations in each sample to illustrate that the comparison holds with higher number of observations. \\
Under the model, all compared methods give the same performance. When outliers are added, both classical MD$\varphi$DE and MLE are not robust and give the same result. Other methods provide robust results. Error values are close for $\varphi-$divergence-based estimators and very close to results obtained by the MDPD which gives slightly better performances.\\
\begin{table}[hp]
\caption{The mean value and the standard deviation of the estimates in a 100-run experiment in the two-component Gaussian mixture}
\label{tab:EstimGaussMix}
\centering
\resizebox{\textwidth}{!}{ 
\begin{tabular}{|l|c|c|c|c|c|c||c|c|c|c|c|c|}
\hline
\multirow{2}{2.5cm}{Estimation method} & \multicolumn{6}{|c||}{No Outliers} & \multicolumn{6}{|c|}{$10\%$ Outliers}\\
\cline{2-13}
  & $\lambda$ & sd($\lambda$) & $\mu_1$ & sd$(\mu_1)$ & $\mu_2$ & sd$(\mu_2)$ & $\lambda$ & sd($\lambda$) & $\mu_1$ & sd$(\mu_1)$ & $\mu_2$ & sd$(\mu_2)$\\
% \hline
% \hline
%\multicolumn{13}{|c|}{Hellinger} \\
\hline
\hline
Classical MD$\varphi$DE &0.360 & 0.054 & -1.989 & 0.204 & 1.493 & 0.136 & 0.342 & 0.064 & -2.617 &0.288 & 1.713 & 0.172\\
\hline
New MD$\varphi$DE - Gauss:Silverman & 0.360 & 0.054 & -1.993 & 0.208 & 1.499 & 0.133 & 0.349 & 0.058 & -1.767 &0.226 & 1.377 & 0.135\\
%New MD$\varphi$DE - SJ & 0.359 & 0.054 & -1.981 & 0.206 & 1.490 & 0.134 & 0.346 & 0.059 & -1.706 &0.218 & 1.333 & 0.136\\
New MD$\varphi$DE - Gauss:1.2 & 0.359 & 0.054 & -2.024 & 0.210 & 1.523 & 0.132 & 0.348 & 0.058 & -1.811 &0.218 & 1.411 & 0.132\\
\hline
Basu-Lindsay - Gauss:Silverman & 0.361 & 0.055 & -1.979 & 0.207 & 1.490 & 0.139 & 0.339 & 0.062 & -1.927 &0.305 & 1.377 & 0.158\\
%Basu-Lindsay - SJ & 0.360 & 0.054 & -1.977 & 0.203 & 1.486 & 0.135 & 0.346 & 0.059 & -1.751 &0.227 & 1.339 & 0.140\\
Basu-Lindsay - Gauss:0.9 & 0.361 & 0.055 & -1.976 & 0.215 & 1.489 & 0.143 & 0.334 & 0.066 & -1.987 &0.288 & 1.378 & 0.162\\
\hline
Beran - Gauss:Silverman & 0.371 & 0.050 & -1.985 & 0.203 & 1.546 & 0.132 & 0.369 & 0.053 & -1.788 & 0.218 & 1.477 & 0.134\\
%Beran - SJ & 0.366 & 0.052 & -1.983 & 0.204 & 1.522 & 0.134 & 0.355 & 0.056 & -1.743 & 0.217 & 1.384 & 0.136\\
Beran - Gauss:0.9 & 0.381 & 0.046 & -1.968 & 0.202 & 1.594 & 0.127 & 0.375 & 0.048 & -1.785 & 0.218 & 1.502 & 0.130\\
\hline
D$\varphi$DE & 0.361 & 0.054 & -1.988 & 0.203 & 1.492 & 0.136 & 0.355 & 0.056 & -2.132 &0.224 & 1.605 & 0.137\\
\hline
\hline
MDPD 0.1 & 0.360 & 0.054 & -1.991 & 0.207 & 1.493 & 0.134 & 0.346 & 0.059 & -2.052 & 0.243 & 1.452 & 0.144\\
MDPD 0.25 & 0.360 & 0.053 & -1.994 & 0.213 & 1.492 & 0.133 & 0.351 & 0.057 & -1.832 & 0.223 & 1.394 & 0.134\\
MDPD 0.5 & 0.360 & 0.053 & -1.997 & 0.226 & 1.489 & 0.136 & 0.353 & 0.056 & -1.819 & 0.218 & 1.404 & 0.132\\
\hline
\hline
MLE (EM) & 0.360 & 0.054 & -1.989 & 0.204 & 1.493 & 0.136 & 0.342 & 0.064 & -2.617 &0.288 & 1.713 & 0.172\\
\hline
\end{tabular}}
\end{table}

\begin{table}[hp]
\caption{The mean value of errors committed in a 100-run experiment with the standard deviation in the two-component Gaussian mixture}
\label{tab:ErrGaussMix}
\centering
\resizebox{\textwidth}{!}{
\begin{tabular}{|l|c|c|c|c||c|c|c|c|}
\hline
\multirow{2}{2.5cm}{Estimation method} & \multicolumn{4}{|c||}{No Outliers} & \multicolumn{4}{|c|}{$10\%$ Outliers}\\
\cline{2-9}
  & $\sqrt{\chi^2}$ & sd($\sqrt{\chi^2}$) & TVD & sd(TVD) & $\sqrt{\chi^2}$ & sd($\sqrt{\chi^2}$) & TVD & sd(TVD)\\
 %\hline
 %\hline
%\multicolumn{9}{|c|}{Hellinger} \\
\hline
\hline
Classical MD$\varphi$DE & 0.113 & 0.044 & 0.064 & 0.025 & 0.335 & 0.102 & 0.150 & 0.034\\
\hline
New MD$\varphi$DE - Gauss:Silverman & 0.113 & 0.045 & 0.064 & 0.025 & 0.155 & 0.059 & 0.087 & 0.033\\
%New MD$\varphi$DE - SJ & 0.113 & 0.045 & 0.064 & 0.025 & 0.179 & 0.061 & 0.101 & 0.035\\
New MD$\varphi$DE - Gauss:1.2 & 0.114 & 0.047 & 0.064 & 0.025 & 0.139 & 0.053 & 0.078 & 0.030\\
\hline
Basu-Lindsay - Gauss:Silverman & 0.115 & 0.043 & 0.065 & 0.024 & 0.155 & 0.073 & 0.085 & 0.033\\
%Basu-Lindsay - SJ & 0.113 & 0.043 & 0.064 & 0.024 & 0.170 & 0.062 & 0.096 & 0.035\\
Basu-Lindsay - Gauss:0.9 & 0.118 & 0.043 & 0.067 & 0.024 & 0.147 & 0.059 & 0.083 & 0.034\\
\hline
Beran - Gauss:Silverman & 0.113 & 0.046 & 0.064 & 0.025 & 0.132 & 0.050 & 0.073 & 0.027\\
%Beran - SJ & 0.112 & 0.045 & 0.063 & 0.025 & 0.157 & 0.057 & 0.087 & 0.032\\
Beran - Gauss:0.9 & 0.117 & 0.050 & 0.066 & 0.028 & 0.127 & 0.049 & 0.070 & 0.026\\
\hline
D$\varphi$DE & 0.112 & 0.044 & 0.064 & 0.025 & 0.142 & 0.061 & 0.076 & 0.031\\
\hline
\hline
MDPD 0.1 & 0.113 & 0.044 & 0.064 & 0.025 & 0.124 & 0.052 & 0.069 & 0.029\\
MDPD 0.25 & 0.114 & 0.045 & 0.064 & 0.025 & 0.140 & 0.054 & 0.079 & 0.030\\
MDPD 0.5 & 0.117 & 0.047 & 0.065 & 0.025 & 0.138 & 0.053 & 0.078 & 0.030\\
\hline
\hline
MLE & 0.113 & 0.044 & 0.064 & 0.025 & 0.335 & 0.102 & 0.150 & 0.034\\
\hline
\end{tabular}}
\end{table}

\begin{table}[hp]
\label{tab:ErrGaussMix1000Obs}
\centering
\resizebox{\textwidth}{!}{
\begin{tabular}{|l|c|c|c|c||c|c|c|c|}
\hline
\multirow{2}{2.5cm}{Estimation method} & \multicolumn{4}{|c||}{No Outliers} & \multicolumn{4}{|c|}{$10\%$ Outliers}\\
\cline{2-9}
  & $\sqrt{\chi^2}$ & sd($\sqrt{\chi^2}$) & TVD & sd(TVD) & $\sqrt{\chi^2}$ & sd($\sqrt{\chi^2}$) & TVD & sd(TVD)\\
 %\hline
 %\hline
%\multicolumn{9}{|c|}{Hellinger} \\
\hline
\hline
Classical MD$\varphi$DE & 0.036 & 0.016 & 0.020 & 0.009 & 0.308 & 0.031 & 0.142 & 0.013\\
\hline
New MD$\varphi$DE - Gauss:Silverman & 0.036 & 0.015 & 0.020 & 0.009 & 0.146 & 0.024 & 0.082 & 0.014\\
%New MD$\varphi$DE - SJ & 0.113 & 0.045 & 0.064 & 0.025 & 0.179 & 0.061 & 0.101 & 0.035\\
New MD$\varphi$DE - Gauss:1.2 & 0.042 & 0.017 & 0.023 & 0.009 & 0.095 & 0.022 & 0.051 & 0.012\\
\hline
Basu-Lindsay - Gauss:Silverman & 0.037 & 0.016 & 0.021 & 0.009 & 0.132 & 0.026 & 0.074 & 0.014\\
%Basu-Lindsay - SJ & 0.113 & 0.043 & 0.064 & 0.024 & 0.170 & 0.062 & 0.096 & 0.035\\
Basu-Lindsay - Gauss:1.2 & 0.040 & 0.017 & 0.022 & 0.009 & 0.129 & 0.041 & 0.071 & 0.022\\
Basu-Lindsay - Gauss:1 & 0.039 & 0.017 & 0.022 & 0.009 & 0.076 & 0.022 & 0.045 & 0.014\\
\hline
Beran - Gauss:Silverman & 0.038 & 0.016 & 0.021 & 0.009 & 0.116 & 0.024 & 0.062 & 0.013\\
%Beran - SJ & 0.112 & 0.045 & 0.063 & 0.025 & 0.157 & 0.057 & 0.087 & 0.032\\
Beran - Gauss:1.2 & 0.114 & 0.018 & 0.066 & 0.010 & 0.114 & 0.022 & 0.065 & 0.012\\
Beran - Gauss:1 & 0.078 & 0.018 & 0.045 & 0.010 & 0.087 & 0.022 & 0.044 & 0.011\\
\hline
D$\varphi$DE & 0.045 & 0.016 & 0.020 & 0.009 & 0.079 & 0.023 & 0.044 & 0.013\\
\hline
\hline
MDPD 0.1 & 0.036 & 0.015 & 0.020 & 0.009 & 0.046 & 0.016 & 0.027 & 0.010\\
MDPD 0.25 & 0.036 & 0.015 & 0.020 & 0.009 & 0.095 & 0.023 & 0.053 & 0.013\\
MDPD 0.5 & 0.037 & 0.015 & 0.021 & 0.009 & 0.092 & 0.022 & 0.050 & 0.012\\
\hline
\hline
MLE & 0.036 & 0.016 & 0.020 & 0.009 & 0.308 & 0.031 & 0.142 & 0.013\\
\hline
\end{tabular}}
\caption{The mean value of errors committed in a 100-run experiment with the standard deviation in the two-component Gaussian mixture. Number of observations is 1000}
\end{table}

\clearpage
%%%%%%%%%%%%%%%%%%%%%%%%%%%%%%%%%%%%%%%%%%%%%%%%%%%%%%%%%%%%%%%%%%%%%%%%%%%%%
%%%%%%%%%%%%%%%%%%%%%%%%%%%%%%%%%%%%%%%%%%%%%%%%%%%%%%%%%%%%%%%%%%%%%%%%%%%%%
\subsection{Generalized Pareto distribution}\label{subsec:SimulationGPD}
We show in this paragraph several simulations from the generalized Pareto distribution (GPD) where the data is contaminated or not by a $10\%$ of outliers. A GPD with a fixed location at zero, a scale parameter $\sigma>0$ and a shape parameter $\nu>0$ is defined by:
\[p_{\nu,\sigma}(y) = \frac{1}{\sigma}\left(1+\nu\frac{y}{\sigma}\right)^{-1-\frac{1}{\nu}},\quad \text{for } y\geq 0.\]
The true set of parameters is $\nu = 0.7, \sigma=3$. Outliers are added by replacing 10 observations (chosen randomly) from each sample by observations from the distribution GPD($\nu=1,\sigma=10,\mu=500$) where $\mu$ is the location parameter. Estimation results are summarized in table \ref{tab:EstimGPD}. Estimation error is calculated in table \ref{tab:ErrGPD}. The maximum likelihood estimator was calculated using the \texttt{gpd.fit} function of package \texttt{ismev}.\\

\underline{Under the model}, all presented methods except for the Basu-Lindsay approach have close performance to the MLE and sometimes even better for given choices of the kernel or the tradeoff parameter. Our kernel-based MD$\varphi$DE attained a similar performance to the MLE for \emph{all} non classical kernels and the corresponding choices of the window, and attained an even better efficiency than the MLE. Beran's method attained this performance only with the varying KDE. The MDPD attained it only for small values of $a$ (=0.1). The use of a symmetric kernel (here the Gaussian) did not give good results in kernel-based methods except for our kernel-based MD$\varphi$DE with a Silverman's rule for the window\footnote{The Sheather and Jones' rule did not give satisfactory results.}. This may be some indication of low sensitivity to the kernel used. \\
\underline{When outliers are added}, the performance of kernel-based methods was slightly deteriorated whereas other methods (the MDPD included for all values of $a$) were greatly influenced, and the error is at least doubled; MDPD for all cases included. The use of asymmetric kernels seems to be the most convenient for a GPD model. Our kernel-based MD$\varphi$DE seems to give the best result (in $\chi^2$ and TVD) for all kernels and corresponding windows keeping a great gap in its favor in comparison with other methods. \\
\begin{remark}
The nature of the heavy tail of the GPD (slow decrease at infinity) made integration calculus difficult, and some integration functions failed to give fairly correct results. We, therefore, and in order to avoid integration on an infinite interval $[0,\infty)$, propose to use a quantile trick which is translated by the change of variable:
\[\int_{0}^{\infty}{\varphi'\left(\frac{p_{\phi}}{p_{\alpha}}\right)(x) p_{\phi}(x) dx} = \int_{0}^1{\varphi'\left(\frac{p_{\phi}}{p_{\alpha}}\right)p_{\phi}(\mathbb{F}_{\phi}^{-1}(y))dy},\]
where $\mathbb{F}_{\phi}^{-1}(y) = \frac{\sigma}{\nu}((1-y)^{-\nu}-1)$ is the quantile of the GPD probability law $P_{\phi}$. This idea may appear ineffective since it does not change anything in the integral (the quantile function takes back values from $[0,1)$ into $[0,\infty)$). In fact, integration methods perform in general better when integrating on a finite interval than when integrating on an infinite one.
\end{remark}

\begin{table}[hp]
\centering
\resizebox{\textwidth}{!}{ 
\begin{tabular}{|l|c|c|c|c||c|c|c|c|}
\hline
\multirow{2}{2.5cm}{Estimation method} & \multicolumn{4}{|c||}{No Outliers} & \multicolumn{4}{|c|}{$10\%$ Outliers}\\
\cline{2-9}
  & $\nu$ & sd$(\nu)$ & $\sigma$ & sd$(\sigma)$ & $\nu$ & sd$(\nu)$ & $\sigma$ & sd$(\sigma)$\\
 \hline
 \hline
\multicolumn{9}{|c|}{Hellinger} \\
\hline
\hline
Classical MD$\varphi$DE & 0.721  & 0.174 & 3.029 & 0.575 & 1.655 & 0.113 & 2.694 & 0.491\\
\hline
New MD$\varphi$DE - Gauss Silverman & 0.463 & 0.142 & 2.719 & 0.586 & 0.571 & 0.197 & 2.427 & 0.599\\
New MD$\varphi$DE - Gauss SJ & 0.343 & 0.108 & 2.858 & 0.597 &  0.368 & 0.141 & 2.798 & 0.569\\
New MD$\varphi$DE - RIG CV & 0.528 & 0.140 & 3.125 & 0.611 &  0.775 & 0.202 & 2.844 & 0.571\\
New MD$\varphi$DE - RIG Nrd0 & 0.562 & 0.139 & 3.133 & 0.605 & 0.817 & 0.219 & 2.815 & 0.545\\
New MD$\varphi$DE - RIG SJ & 0.522 & 0.129 & 3.138 & 0.616 & 0.688 & 0.191 & 2.903 & 0.574\\
New MD$\varphi$DE - GA CV & 0.530 & 0.139 & 3.117 & 0.610 & 0.766  & 0.204 & 2.833 & 0.577\\
New MD$\varphi$DE - GA Nrd0 & 0.564 & 0.139 & 3.112 & 0.601 & 0.814  & 0.211 & 2.787 & 0.544\\
New MD$\varphi$DE - GA SJ & 0.520 & 0.126 & 3.135 & 0.607 & 0.691  & 0.185 & 2.895 & 0.576\\
New MD$\varphi$DE - MT 5 & 0.641 & 0.156 & 3.217 & 0.615 & 1.202 & 0.161 & 2.806 & 0.510\\
New MD$\varphi$DE - MT 10 & 0.607 & 0.153 & 3.272 & 0.628 & 1.090 & 0.195 & 2.876 & 0.552\\
New MD$\varphi$DE - MT 15 & 0.588 & 0.150 & 3.307 & 0.636 & 1.026 & 0.206 & 2.920 & 0.565\\
New MD$\varphi$DE - MT 20 & 0.573 & 0.148 & 3.331 & 0.643 & 0.979 & 0.212 & 2.956 & 0.577\\
\hline
Basu-Lindsay - Gauss Silverman & 0.128  & 0.125 & 6.022 & 1.522 & 0.122  & 0.109 & 7.151 & 2.025\\
Basu-Lindsay - Gauss SJ & 0.078  & 0.066 & 4.603 & 1.057 &  0.097 & 0.087 & 4.843 & 1.316\\
Basu-Lindsay - MT 5 & 0.833  & 0.156 & 2.232 & 0.651 & 0.765 & 0.189 & 2.937 & 0.666\\
Basu-Lindsay - MT 10 & 0.853  & 0.197 & 2.297 & 0.659 & 0.777 & 0.193 & 2.880 & 0.704\\
Basu-Lindsay - MT 15 &  0.881 & 0.176 & 2.293 & 0.517 & 1.164 & 0.169 & 2.893 & 0.530\\
Basu-Lindsay - MT 20 &  0.907 & 0.180 & 2.337 & 0.603 & 0.936 & 0.206 & 2.694 & 0.580\\
\hline
Beran - Gauss Nrd0 & 0.216  & 0.108 & 5.165 & 1.218 & 0.197 & 0.125 & 6.084 & 1.546\\
Beran - Gauss SJ & 0.231  & 0.108 & 3.988 & 0.919 & 0.229 & 0.134 & 4.135 & 0.939\\
Beran - RIG CV & 0.516  & 0.134 & 3.890 & 0.832 & 0.833 & 0.218 & 3.944 & 0.745\\
Beran - RIG Nrd0 & 0.515  & 0.138 & 4.441 & 1.026 & 0.878 & 0.233 & 4.229 & 0.954\\
Beran - RIG SJ & 0.507  & 0.136 & 3.813 & 0.787 & 0.732 & 0.200 & 3.641 & 1.113\\
Beran - GA CV & 0.486  & 0.134 & 3.936 & 0.847 & 0.745  & 0.207 & 4.097 & 0.822\\
Beran - GA Nrd0 & 0.475  & 0.139 & 4.510 & 0.998 & 0.778 & 0.220 & 4.547 & 1.032\\
Beran - GA SJ & 0.503  & 0.133 & 3.780 & 0.773 & 0.703 & 0.186 & 3.589 & 0.781\\
Beran - MT 5 & 0.711  & 0.150 & 3.384 & 0.640 & 1.339 & 0.140 & 2.979 & 0.551\\
Beran - MT 10 & 0.665  & 0.150 & 3.315 & 0.620 & 1.231 & 0.155 & 2.900 & 0.530\\
Beran - MT 15 & 0.637  & 0.154 & 3.310 & 0.640 & 1.164 & 0.169 & 2.893 & 0.530\\
Beran - MT 20 & 0.627  & 0.156 & 3.302 & 0.637 & 0.936 & 0.206 & 2.694 & 0.580\\
\hline
D$\varphi$DE & 0.720  & 0.179 & 3.026 & 0.580 & 1.45 & 0.290 & 2.749 & 0.524\\
\hline
\hline
MDPD 1 & 0.729 & 0.402 & 3.023 & 0.660 & 1.039 & 0.483  & 3.273 & 0.681 \\
MDPD 0.75 & 0.716 & 0.331 & 3.025 & 0.631 & 1.021 & 0.416  & 3.242 & 0.645 \\
MDPD 0.5 & 0.715 & 0.263 & 3.023 & 0.603 & 1.028 & 0.361  & 3.171 & 0.605 \\
MDPD 0.25 & 0.722 & 0.200 & 3.019 & 0.581 & 1.292 & 0.240  & 2.955 & 0.532 \\
MDPD 0.1 & 0.723 & 0.175 & 3.019 & 0.568 & 1.564 & 0.154  & 2.779 & 0.500 \\
\hline
\hline
MLE &   0.719 & 0.174 & 3.031 & 0.58 & 1.654 & 0.113  & 2.695 & 0.492 \\
\hline
\end{tabular}}
\caption{The mean value and the standard deviation of the estimates in a 100-run experiment in the GPG model. The escort parameter of the D$\varphi$DE is taken as the new MD$\varphi$DE with Silverman's rule.}
\label{tab:EstimGPD}
\end{table}

\begin{table}[hp]
\centering
\resizebox{\textwidth}{!}{ 
\begin{tabular}{|l|c|c|c|c||c|c|c|c|}
\hline
\multirow{2}{2.5cm}{Estimation method} & \multicolumn{4}{|c||}{No Outliers} & \multicolumn{4}{|c|}{$10\%$ Outliers}\\
\cline{2-9}
  & $\chi^2$ & sd($\chi^2$) & TVD & sd(TVD) & $\chi^2$ & sd($\chi^2$) & TVD & sd(TVD)\\
 \hline
 \hline
\multicolumn{9}{|c|}{Hellinger} \\
\hline
\hline
Classical MD$\varphi$DE &  0.099 & 0.077 & 0.044 & 0.026 & 1.027 & 0.195 & 0.142 & 0.014\\
\hline
New MD$\varphi$DE - Silverman & 0.159 & 0.056 & 0.087 & 0.034 & 0.171 & 0.070 & 0.097 & 0.044\\
New MD$\varphi$DE - SJ & 0.189  & 0.052 & 0.100 & 0.035 & 0.183  & 0.066 & 0.098 & 0.042\\
New MD$\varphi$DE - RIG CV & 0.109 & 0.045 & 0.058 & 0.027 & 0.114 & 0.065 & 0.053 & 0.029\\
New MD$\varphi$DE - RIG Nrd0 & 0.100 & 0.044 & 0.054 & 0.027 & 0.142 & 0.130 & 0.056 & 0.029\\
New MD$\varphi$DE - RIG SJ & 0.110 & 0.044 & 0.059 & 0.027 & 0.104 & 0.056 & 0.054 & 0.030\\
New MD$\varphi$DE - GA CV & 0.108 & 0.045 & 0.058 & 0.027 & 0.114 & 0.063 & 0.054 & 0.029\\
New MD$\varphi$DE - GA Nrd0 & 0.100 & 0.044 & 0.054 & 0.027 & 0.132 & 0.092 & 0.056 & 0.028\\
New MD$\varphi$DE - GA SJ & 0.109 & 0.044 & 0.058 & 0.027 & 0.104 & 0.056 & 0.054 & 0.030\\
New MD$\varphi$DE - MT 5 & 0.093 & 0.053 & 0.049 & 0.028 & 0.472 & 0.307 & 0.089 & 0.024\\
New MD$\varphi$DE - MT 10 & 0.095 & 0.050 & 0.051 & 0.028 & 0.336 & 0.243 & 0.078 & 0.026\\
New MD$\varphi$DE - MT 15 & 0.097 & 0.048 & 0.053 & 0.028 & 0.268 & 0.193 & 0.072 & 0.027\\
New MD$\varphi$DE - MT 20 & 0.099 & 0.047 & 0.054 & 0.029 & 0.226 & 0.154 & 0.068 & 0.028\\
\hline
Basu-Lindsay - Silverman & 0.301  & 0.08 & 0.179 & 0.048 & 0.361  & 0.110 & 0.214 & 0.061\\
Basu-Lindsay - SJ & 0.256  & 0.046 & 0.145 & 0.033 & 0.264  & 0.055 & 0.151 & 0.039\\
Basu-Lindsay - MT 5 & 0.155  & 0.082 & 0.090 & 0.047 & 0.100 & 0.077 & 0.051 & 0.036\\
Basu-Lindsay - MT 10 & 0.155  & 0.080 & 0.085 & 0.043 & 0.102 & 0.078 & 0.053 & 0.038\\
Basu-Lindsay - MT 15 & 0.140  & 0.107 & 0.071 & 0.050 & 0.421 & 0.278 & 0.086 & 0.025\\
Basu-Lindsay - MT 20 & 0.157  & 0.085 & 0.078 & 0.044 & 0.160 & 0.083 & 0.059 & 0.031\\
\hline
Beran - Gauss Nrd0 & 0.241  & 0.072 & 0.142 & 0.045 & 0.297 & 0.090 & 0.177 & 0.053\\
Beran - Gauss SJ & 0.199  & 0.049 & 0.109 & 0.034 & 0.207 & 0.044 & 0.114 & 0.032\\
Beran - RIG CV & 0.133  & 0.060 & 0.076 & 0.038 & 0.226 & 0.128 & 0.094 & 0.041\\
Beran - RIG Nrd0 & 0.164  & 0.085 & 0.097 & 0.051 & 0.306 & 0.235 & 0.114 & 0.054\\
Beran - RIG SJ & 0.123  & 0.060 & 0.069 & 0.039 & 0.146 & 0.097 & 0.070 & 0.048\\
Beran - GA CV & 0.136  & 0.060 & 0.078 & 0.038 &  0.195 & 0.100 & 0.094 & 0.044\\
Beran - GA Nrd0 & 0.169  & 0.078 & 0.101 & 0.048 & 0.267 & 0.186 & 0.121 & 0.057\\
Beran - GA SJ &  0.120 & 0.058 & 0.068 & 0.037 & 0.130 & 0.078 & 0.065 & 0.040\\
Beran - MT 5 & 0.103 & 0.067 & 0.052 & 0.030 & 0.915 & 0.729 & 0.111 & 0.022\\
Beran - MT 10 & 0.093 & 0.057 & 0.049 & 0.029 & 0.581 & 0.615 & 0.095 & 0.023\\
Beran - MT 15 & 0.094 & 0.054 & 0.050 & 0.029 & 0.421 & 0.278 & 0.086 & 0.025\\
Beran - MT 20 & 0.095 & 0.055 & 0.051 & 0.029 & 0.371 & 0.298 & 0.081 & 0.026\\
\hline
D$\varphi$DE &  0.099 & 0.077 & 0.048 & 0.028 & 0.843 & 0.407 & 0.120 & 0.030\\
\hline
\hline
MDPD 1 & 0.211 & 0.310 & 0.068 & 0.038 & 0.477 & 0.665 & 0.089 & 0.047\\
MDPD 0.75 & 0.204 & 0.389 & 0.062 & 0.034 & 0.424 & 0.545 & 0.085 & 0.043\\
MDPD 0.5 & 0.141 & 0.160 & 0.056 & 0.030 & 0.419 & 0.515 & 0.082 & 0.039\\
MDPD 0.25 & 0.106 & 0.082 & 0.049 & 0.028 & 0.669 & 0.441 & 0.104 & 0.030\\
MDPD 0.1 & 0.099 & 0.083 & 0.047 & 0.027 & 0.955 & 0.326 & 0.133 & 0.019\\
\hline
\hline
MLE & 0.099 & 0.077 & 0.048 & 0.026 & 1.025  & 0.195 & 0.142 & 0.014\\
\hline
\end{tabular}}
\caption{The mean value of errors committed in a 100-run experiment with the standard deviation for the GPD model. The escort parameter of the D$\varphi$DE is taken as the new MD$\varphi$DE with the gamma kernel.}
\label{tab:ErrGPD}
\end{table}

\clearpage
%%%%%%%%%%%%%%%%%%%%%%%%%%%%%%%%%%%%%%%%%%%%%%%%%%%%%%%%%%%%%%%%%%%%%%%%%%

\subsection{Mixtures of Two Weibull Components}\label{subsec:TwoWeibullSimuPara}
We present the results of estimating three different two-component Weibull mixtures. The model has the following density:
\[p_{\phi}(x) = 2\lambda\nu_1 (2x)^{\nu_1-1} e^{-(2x)^{\nu_1}}+(1-\lambda)\frac{\nu_2}{2}\left(\frac{x}{2}\right)^{\nu_2-1} e^{-\left(\frac{x}{2}\right)^{\nu_2}}.\]
Scale parameters are supposed to be known and equal to $0.5$ for the first component and $2$ for the second component. The proportion is unknown and fixed at $0.35$. Shape parameters are supposed unknown. Our examples cover a variety of cases of a Weibull mixture where the density function has either a finite limit at zero or goes to infinity for one of the components:
\begin{enumerate}
\item a mixture with close modes $\nu_1 = 1.2, \nu_2 = 2$;
\item a mixture with one mode and with limit equal to infinity at zero $\nu_1 = 0.5, \nu_2 = 3$;
\item a mixture with no modes and with limit equal to infinity at zero $\nu_1 = 0.5, \nu_2 = 1$.
\end{enumerate}
We plot these mixtures in figure \ref{fig:ThreeWeibullMixtures}. Outliers were added in different ways to illustrate several scenarios. For the first mixture, outliers were added by replacing 10 observations of each sample chosen randomly by 10 observations drawn independently from a Weibull distribution with shape $\nu = 0.9$ and scale $\sigma = 3$. See tables (\ref{tab:EstimWeibullMixTwoModes}) and (\ref{tab:ErrWeibullMixTwoModes}). For the second mixture, we added to the 10 largest observations of each sample a random observation drawn from the uniform distribution $\mathcal{U}[2,10]$. See tables \ref{tab:EstimWeibullMixOneMode} and \ref{tab:ErrWeibullMixOneMode}. For the third one, outliers were added by replacing 10 observations, chosen randomly, of each sample by observations from the uniform distribution $\mathcal{U}[\max y_i, 75]$ after having verified that no observation in the overall data has exceeded the value 50.  See tables \ref{tab:EstimWeibullMixNoMode} and \ref{tab:ErrWeibullMixNoMode}.\\
\begin{figure}[h]
\centering
\includegraphics[scale=0.4]{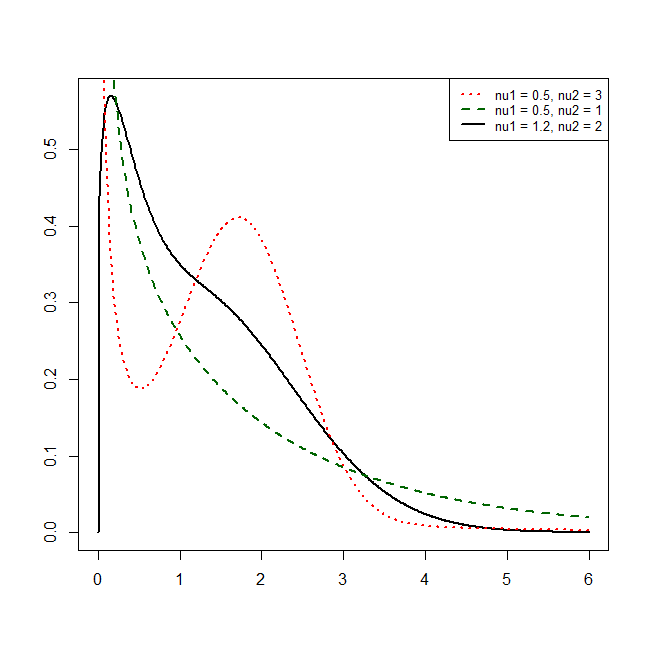}
\caption{The three Weibull mixtures used in our experience.}
\label{fig:ThreeWeibullMixtures}
\end{figure}

\noindent The caclulus of the $\chi^2$ divergence between the estimated model and the true distribution gave often infinity on all mixtures for all estimation methods even under the model. This is because a small bias in the estimation of the shape parameter results in a great relative error in both the tail behavior and near zero. We therefore, only provide the TVD as an error criterion.\\  
The first Weibull mixture was the least complicated case. We were able to get satisfactory results for our kernel-based MD$\varphi$DE using a Gaussian kernel. The two other mixtures were more challenging, and we needed to use asymmetric kernels to solve the problem of the bias near zero. It is worth noting that the Basu-Lindsay approach provided very bad estimates in the three mixtures which keeps it out of the competition. Note also that the use of a Gaussian kernel gave very pleasant results for the first mixture in spite of the boundary bias. We excluded it from mixtures which have infinity limit at zero because it did not work well because of the large bias at zero.\\

\textbf{For the first mixture,} under the model all presented methods provide close results (and sometimes better) to the MLE except for the Basu-Lindsay approach with all available choices and Beran's method with the varying KDE (MT) for windows 5 and 10 which fail. Under contamination, our method gives better results than all other methods and have very close (even slightly better) performance to the MDPD for tradeoff parameter higher than 0.25.\\
\textbf{For the second mixture,} the Basu-Lindsay approach failed again. Beran's method gave good result under the model only in one case; the RIG with window 0.01. The MDPD worked very well only for a tradeoff parameter lower than 0.5 and gave a good compromise between robustness and efficiency. It gave the best compromise in the presented methods. Our kernel-based MD$\varphi$DE has close results to the MDPD with difference of $0.01$ in the TVD. It is worth noting that our kernel-based MD$\varphi$DE gave faire results for the two proposed kernels; the asymmetric kernel RIG for window 0.01 as before and the varying KDE MT for windows 10, 15 and 20. A fact which was not verified for other kernel-based methods showing again a less sensibility towards the kernel.\\
\textbf{For the third mixture,} the Basu-Lindsay approach did not give good results especially under the model. The only satisfactory results (which gave a good tradeoff between robustness and efficiency) were obtained by our kernel-based MD$\varphi$DE for RIG kernel with window 0.01, Beran's method with the same kernel and window and the MDPD for $a=0.5$. Our method and Beran's gave the same result with difference of 0.015 in favor of the MDPD. Better efficiency were obtained by other choices but on the cost of the robustness of the resulting estimator under contamination.
%%%%%%%%%%%%%%%%%%%
\begin{table}[hp]
\centering
\resizebox{\textwidth}{!}{
\begin{tabular}{|l|c|c|c|c|c|c||c|c|c|c|c|c|}
\hline
\multirow{2}{2.5cm}{Estimation method} & \multicolumn{6}{|c||}{No Outliers} & \multicolumn{6}{|c|}{$10\%$ Outliers}\\
\cline{2-13}
  & $\lambda$ & sd($\lambda$) & $\nu_1$ & sd$(\nu_1)$ & $\nu_2$ & sd$(\nu_2)$ & $\lambda$ & sd($\lambda$) & $\nu_1$ & sd$(\nu_1)$ & $\nu_2$ & sd$(\nu_2)$\\
 \hline
 \hline
\multicolumn{13}{|c|}{Hellinger} \\
\hline
\hline
Classical MD$\varphi$DE & 0.355 & 0.066 & 1.245 & 0.228 & 2.054 & 0.237 & 0.410 & 0.257 & 1.045 & 0.255 & 1.718 & 0.849\\
\hline
New MD$\varphi$DE - Gauss Silverman & 0.384 & 0.067 & 1.221 & 0.244 & 2.138 & 0.291 & 0.348 & 0.076 & 1.121 & 0.265 & 1.822 & 0.319\\
New MD$\varphi$DE - Gauss SJ & 0.387 & 0.067 & 1.227 & 0.240 & 2.188 & 0.308 & 0.356 & 0.076 & 1.133 & 0.261 & 1.905 & 0.319\\
New MD$\varphi$DE - RIG 0.01 & 0.371 & 0.066 & 1.297 & 0.231 & 2.215 & 0.321 & 0.355 & 0.100 & 1.213 & 0.229 & 1.955 & 0.344\\
New MD$\varphi$DE - RIG 0.1 & 0.358 & 0.065 & 1.233 & 0.210 & 2.065 & 0.267 & 0.330 & 0.117 & 1.127 & 0.226 & 1.741 & 0.304\\
New MD$\varphi$DE - RIG SJ & 0.351 & 0.066 & 1.217 & 0.207 & 2.001 & 0.245 & 0.324 & 0.132 & 1.107 & 0.226 & 1.670 & 0.297\\
New MD$\varphi$DE - MT 5 & 0.328 & 0.112 & 1.301 & 0.235 & 1.809 & 0.192 & 0.363 & 0.229 & 1.195 & 0.213 & 1.592 & 0.356\\
New MD$\varphi$DE - MT 10 & 0.330 & 0.091 & 1.355 & 0.235 & 1.923 & 0.220 & 0.351 & 0.204 & 1.247 & 0.230 & 1.645 & 0.285\\
New MD$\varphi$DE - MT 15 & 0.327 & 0.076 & 1.383 & 0.234 & 1.973 & 0.237 & 0.348 & 0.199 & 1.275 & 0.233 & 1.680 & 0.294\\
New MD$\varphi$DE - MT 20 & 0.328 & 0.076 & 1.403 & 0.233 & 2.002 & 0.249 & 0.348 & 0.198 & 1.295 & 0.235 & 1.702 & 0.297\\
\hline
Basu-Lindsay - Gauss Silverman & 0.752 & 0.064 & 2.199 & 0.248 & 38.66 & 8.66 & 0.822  & 0.083 & 1.927 & 0.276 & 32.37 & 13.52 \\
Basu-Lindsay - Gauss SJ & 0.723 & 0.059 & 2.205 & 0.257 & 16.18 & 10.75 & 0.759  & 0.065 & 1.958 & 0.263 & 19.52 & 10.56 \\
Basu-Lindsay - MT 5 & 0.403 & 0.072 & 1.339 & 0.224 & 3.241 & 0.547 & 0.346 & 0.076 & 1.260 & 0.210 & 2.874 & 0.338 \\
Basu-Lindsay - MT 10 & 0.390 & 0.069 & 1.409 & 0.234 & 3.281 & 0.465 & 0.337 & 0.067 & 1.319 & 0.217 & 2.813 & 0.233 \\
Basu-Lindsay - MT 15 & 0.393 & 0.067 & 1.458 & 0.248 & 3.297 & 0.476 & 0.333 & 0.062 & 1.340 & 0.232 & 2.823 & 0.257 \\
Basu-Lindsay - MT 20 & 0.399 & 0.066 & 1.472 & 0.221 & 3.282 & 0.458 & 0.335 & 0.068 & 1.362 & 0.225 & 2.819 & 0.300 \\
\hline
Beran - Gauss Silverman & 0.254 & 0.058 & 1.313 & 0.087 & 2.010 & 0.200 & 0.182  & 0.074 & 1.174 & 0.162 & 1.703 & 0.253\\
Beran - Gauss SJ & 0.295 & 0.067 & 1.371 & 0.104 & 2.085 & 0.225 & 0.240 & 0.079 & 1.284 &0.127 & 1.794 & 0.266\\
Beran - RIG 0.01 & 0.368 & 0.064 & 1.240 & 0.198 & 2.147 & 0.277 & 0.339 & 0.094 & 1.151 &0.200 & 1.858 & 0.332\\
Beran - RIG 0.1 & 0.345 & 0.061 & 1.117 & 0.103 & 1.897 & 0.172 & 0.289 & 0.095 & 1.033 &0.125 & 1.570 & 0.247\\
Beran - RIG SJ & 0.320 & 0.060 & 1.069 & 0.074 & 1.725 & 0.138 & 0.260 & 0.123 & 0.997 &0.088 & 1.416 & 0.203\\
Beran - MT 5 & 0.453 & 0.307 & 1.146 & 0.178 & 1.386 & 0.180 & 0.626 & 0.349 & 1.055 & 0.172 & 1.461 & 0.531\\
Beran - MT 10 & 0.354 & 0.201 & 1.238 & 0.201 & 1.553 & 0.133 & 0.419 & 0.304 & 1.134 & 0.202 & 1.450 & 0.425\\
Beran - MT 15 & 0.334 & 0.153 & 1.286 & 0.211 & 1.664 & 0.143 & 0.404 & 0.277 & 1.178 & 0.188 & 1.500 & 0.370\\
Beran - MT 20 & 0.334 & 0.136 & 1.317 & 0.218 & 1.738 & 0.156 & 0.383 & 0.256 & 1.207 & 0.198 & 1.542 & 0.348\\
\hline
D$\varphi$DE & 0.356 & 0.066 & 1.248 & 0.232 & 2.069 & 0.278 & 0.332 & 0.142 & 1.113 & 0.248 & 1.700 & 0.289\\
\hline
\hline
MDPD 1 & 0.358 & 0.087 & 1.238 & 0.252 & 2.127 & 0.521 & 0.343 & 0.113 & 1.167 & 0.239 & 2.005 & 0.517\\
MDPD 0.75 & 0.353 & 0.073 & 1.236 & 0.237 & 2.088 & 0.397 & 0.341 & 0.108 & 1.164 & 0.235 & 1.951 & 0.432\\
MDPD 0.5 & 0.354 & 0.068 & 1.238 & 0.230 & 2.071 & 0.345 & 0.336 & 0.105 & 1.159 & 0.237 & 1.860 & 0.344\\
MDPD 0.25 & 0.354 & 0.066 & 1.239 & 0.226 & 2.053 & 0.272 & 0.324 & 0.131 & 1.132 & 0.235 & 1.699 & 0.321\\
MDPD 0.1 & 0.355 & 0.066 & 1.242 & 0.227 & 2.048 & 0.238 & 0.394 & 0.241 & 1.091 & 0.215 & 1.780 & 0.792\\
\hline
\hline
MLE (EM) & 0.355 & 0.066 & 1.245 & 0.228 & 2.054 & 0.237 & 0.321 & 0.187 & 0.913 & 0.313 & 1.575 & 0.325\\
\hline
\end{tabular}}
\caption{The mean value and the standard deviation of the estimates in a 100-run experiment on a two-component Weibull mixture ($\lambda=0.35,\nu_1=1.2,\nu_2=2$). The escort parameter of the D$\varphi$DE is taken as the new MD$\varphi$DE with the SJ bandwidth choice.}
\label{tab:EstimWeibullMixTwoModes}
\end{table}

\begin{table}[hp]
\centering
\begin{tabular}{|l|c|c|c||c|c|c|}
\hline
\multirow{2}{2.5cm}{Estimation method} & \multicolumn{3}{|c||}{No Outliers} & \multicolumn{3}{|c|}{$10\%$ Outliers}\\
\cline{2-7}
  & mean & median & sd & mean & median & sd\\
 \hline
 \hline
\multicolumn{7}{|c|}{Hellinger} \\
\hline
\hline
Classical MD$\varphi$DE & 0.052 & 0.048 & 0.025 & 0.108 & 0.094 & 0.099\\
\hline
New MD$\varphi$DE - Gauss Silverman & 0.058 & 0.054 & 0.029 & 0.068 & 0.065 & 0.034\\
New MD$\varphi$DE - Gauss SJ & 0.058 & 0.053 & 0.029 & 0.064 & 0.061 & 0.031\\
New MD$\varphi$DE - RIG 0.01 & 0.058 & 0.052 & 0.030 & 0.059 & 0.057 & 0.030 \\
New MD$\varphi$DE - RIG 0.1 & 0.051 & 0.049 & 0.026 & 0.066 & 0.062 & 0.032 \\
New MD$\varphi$DE - RIG SJ & 0.050 & 0.050 & 0.026 & 0.071 & 0.066 & 0.032 \\
New MD$\varphi$DE - MT 5 & 0.057 & 0.055 & 0.025 & 0.081 & 0.074 & 0.032 \\
New MD$\varphi$DE - MT 10 & 0.054 & 0.053 & 0.026 & 0.075 & 0.071 & 0.032 \\
New MD$\varphi$DE - MT 15 & 0.054 & 0.054 & 0.026 & 0.073 & 0.069 & 0.032 \\
New MD$\varphi$DE - MT 20 & 0.055 & 0.054 & 0.027 & 0.073 & 0.069 & 0.031 \\
\hline
Basu Lindsay - Gauss Silverman & 0.298 & 0.289 & 0.042 & 0.247 & 0.253 & 0.050\\
Basu Lindsay - Gauss SJ & 0.252 & 0.256 & 0.051 & 0.242 & 0.246 & 0.044\\
Basu Lindsay - MT 5 & 0.127 & 0.141 & 0.046 & 0.121 & 0.111 & 0.042\\
Basu Lindsay - MT 10 & 0.133 & 0.136 & 0.039 & 0.117 & 0.111 & 0.036\\
Basu Lindsay - MT 15 & 0.134 & 0.141 & 0.039 & 0.118 & 0.110 & 0.038\\
Basu Lindsay - MT 20 & 0.132 & 0.138 & 0.039 & 0.117 & 0.109 & 0.039\\
\hline
Beran - Gauss Silverman & 0.068 & 0.062 & 0.028 & 0.082 & 0.081 & 0.031\\
Beran - Gauss SJ & 0.060 & 0.054 & 0.028 & 0.067 & 0.065 & 0.029 \\
Beran - RIG 0.01 & 0.052 & 0.048 & 0.026 & 0.060 & 0.058 & 0.029 \\
Beran - RIG 0.1 & 0.042 & 0.039 & 0.020 & 0.067 & 0.061 & 0.030 \\
Beran - RIG SJ & 0.045 & 0.044 & 0.017 & 0.079 & 0.076 & 0.030 \\
Beran - MT 5 & 0.099 & 0.097 & 0.016 & 0.125 & 0.125 & 0.022 \\
Beran - MT 10 & 0.073 & 0.070 & 0.021 & 0.102 & 0.100 & 0.028 \\
Beran - MT 15 & 0.064 & 0.060 & 0.022 & 0.092 & 0.089 & 0.030 \\
Beran - MT 20 & 0.059 & 0.055 & 0.023 & 0.086 & 0.084 & 0.030 \\
\hline
D$\varphi$DE & 0.053 & 0.049 & 0.027 & 0.068 & 0.065 & 0.031\\
\hline
\hline
MDPD 1 & 0.065 & 0.061 & 0.034 & 0.068 & 0.064 & 0.030 \\
MDPD 0.75 & 0.059 & 0.056 & 0.029 & 0.063 & 0.060 & 0.029 \\
MDPD 0.5 & 0.056 & 0.052 & 0.029 & 0.061 & 0.056 & 0.029 \\
MDPD 0.25 & 0.052 & 0.048 & 0.027 & 0.068 & 0.067 & 0.031 \\
MDPD 0.1 & 0.051 & 0.048 & 0.026 & 0.088 & 0.083 & 0.039 \\
\hline
\hline
MLE & 0.052 & 0.048 & 0.025 & 0.095 & 0.098 & 0.035 \\
\hline
\end{tabular}
\caption{The mean value with the standard deviation of the TVA committed in a 100-run experiment on a two-component Weibull mixture ($\lambda=0.35,\nu_1=1.2,\nu_2=2$). The escort parameter of the D$\varphi$DE is taken as the new MD$\varphi$DE with the SJ bandwidth choice.}
\label{tab:ErrWeibullMixTwoModes}
\end{table}

%%%%%%%%%%%%%%%%%%%
%%%%%%%%%%%%%%%%%%%
\begin{table}[hp]
\centering
\resizebox{\textwidth}{!}{
\begin{tabular}{|l|c|c|c|c|c|c||c|c|c|c|c|c|}
\hline
\multirow{2}{2.5cm}{Estimation method} & \multicolumn{6}{|c||}{No Outliers} & \multicolumn{6}{|c|}{$10\%$ Outliers}\\
\cline{2-13}
  & $\lambda$ & sd($\lambda$) & $\nu_1$ & sd$(\nu_1)$ & $\nu_2$ & sd$(\nu_2)$ & $\lambda$ & sd($\lambda$) & $\nu_1$ & sd$(\nu_1)$ & $\nu_2$ & sd$(\nu_2)$\\
 \hline
 \hline
\multicolumn{13}{|c|}{Hellinger} \\
\hline
\hline
Classical MD$\varphi$DE & 0.344 & 0.059 & 0.497 & 0.079 & 3.063 & 0.476  & 0.376 & 0.053 & 0.339 & 0.030 & 2.892 & 0.484\\
\hline
New MD$\varphi$DE RIG - 0.01 & 0.330 & 0.061 & 0.540 & 0.140 & 3.170 & 0.503 & 0.338 & 0.061 & 0.432 & 0.105 & 3.055 & 0.583\\
New MD$\varphi$DE RIG - 0.1 & 0.371 & 0.063 & 0.468 & 0.138 & 3.045 & 0.452 & 0.392  & 0.072 & 0.372 & 0.085 & 2.927 & 0.464\\
New MD$\varphi$DE RIG - SJ & 0.395 & 0.072 & 0.442 & 0.134 & 3.013 & 0.443 & 0.424  & 0.086 & 0.354 & 0.082 & 2.916 & 0.459\\
New MD$\varphi$DE MT - 5 & 0.311 & 0.062 & 0.520 & 0.065 & 2.875 & 0.451 & 0.316 & 0.063 & 0.376 & 0.036 & 2.699 & 0.471\\
New MD$\varphi$DE MT - 10 & 0.302 & 0.062 & 0.548 & 0.077 & 2.903 & 0.433 &  0.306 & 0.062 & 0.384 & 0.039 & 2.727 & 0.448\\
New MD$\varphi$DE MT - 15 & 0.295 & 0.063 & 0.564 & 0.084 & 2.927 & 0.434 & 0.298 & 0.063 & 0.388 & 0.042 & 2.745 & 0.450\\
New MD$\varphi$DE MT - 20 & 0.289 & 0.063 & 0.575 & 0.091 & 2.943 & 0.437 & 0.291 & 0.063 & 0.392 & 0.044 & 2.758 & 0.454\\
\hline
Basu-Lindsay MT - 5 & 0.250 & 0.070 & 0.834 & 0.168 & 2.849 & 0.733 & 0.185 & 0.074 & 0.715 & 0.208 & 2.189& 0.155\\
Basu-Lindsay MT - 10 & 0.240 & 0.065 & 0.797 & 0.157 & 2.789 & 0.550 & 0.197 & 0.087 & 0.707 & 0.201 & 2.324& 0.132\\
Basu-Lindsay MT - 15 & 0.254 & 0.073 & 0.745 & 0.140 & 2.915 & 0.584 & 0.204 & 0.078 & 0.674 & 0.181 & 2.352& 0.092\\
\hline
Beran RIG - 0.01 & 0.298 & 0.058 & 0.647 & 0.082 & 3.017 & 0.437 & 0.295 & 0.057 & 0.486 & 0.081 & 2.842 & 0.460\\
Beran RIG - 0.1 & 0.234 & 0.054 & 0.652 & 0.105 & 2.374 & 0.245 & 0.216 & 0.053 & 0.408 & 0.056 & 2.149 & 0.291\\
Beran RIG - SJ & 0.194 & 0.056 & 0.653 & 0.134 & 1.936 & 0.246 & 0.142 & 0.065 & 0.402 & 0.144 & 1.601 & 0.325\\
Beran MT - 5 & 0.250 & 0.070 & 0.463 & 0.058 & 1.603 & 0.140 & 0.245 & 0.083 & 0.340 & 0.062 & 1.494 & 0.208\\
Beran MT - 10 & 0.278 & 0.066 & 0.501 & 0.069 & 2.005 & 0.181 & 0.275 & 0.079 & 0.354 & 0.033 & 1.868 & 0.260\\
Beran MT - 15 & 0.286 & 0.065 & 0.524 & 0.075 & 2.224 & 0.218 & 0.284 & 0.071 & 0.365 & 0.033 & 2.068 & 0.280\\
\hline
D$\varphi$DE & 0.343 & 0.059 & 0.5004 & 0.084 & 3.047 & 0.474 & 0.372 & 0.056 & 0.357 & 0.056 & 2.897 & 0.502\\
\hline
\hline
MDE 0.75 & 0.444 & 0.126 & 0.595 & 0.080 & 3.466 & 0.643 & 0.417 & 0.127 & 0.602 & 0.087 & 3.233 & 0.606\\
MDE 0.5 & 0.376 & 0.067 & 0.551 & 0.093 & 3.159 & 0.488 & 0.357 & 0.067 & 0.555 & 0.097 & 2.980 & 0.484\\
MDE 0.25 & 0.347 & 0.061 & 0.512 & 0.096 & 3.057 & 0.472 & 0.331 & 0.062 & 0.471 & 0.068 & 2.879 & 0.491\\
MDE 0.1 & 0.344 & 0.059 & 0.496 & 0.084 & 3.050 & 0.470 & 0.343 & 0.058 & 0.384 & 0.037 & 2.859 & 0.484\\
\hline
\hline
MLE (EM) & 0.344 & 0.059 & 0.498 & 0.079 & 3.063 & 0.476 & 0.376 & 0.053 & 0.339 & 0.303 & 2.892 & 0.482\\
\hline
\end{tabular}}
\caption{The mean value and the standard deviation of the estimates in a 100-run experiment in a two-component Weibull mixture ($\lambda=0.35,\nu_1=0.5,\nu_2=3$). The escort parameter of the D$\varphi$DE is taken as the new MD$\varphi$DE with Silverman's rule.}
\label{tab:EstimWeibullMixOneMode}
\end{table}

\begin{table}[hp]
\centering
\begin{tabular}{|l|c|c|c||c|c|c|}
\hline
\multirow{2}{2.5cm}{Estimation method} & \multicolumn{3}{|c||}{No Outliers} & \multicolumn{3}{|c|}{$10\%$ Outliers}\\
\cline{2-7}
  & mean & median & sd & mean & median & sd\\
 \hline
 \hline
\multicolumn{7}{|c|}{Hellinger} \\
\hline
\hline
Classical MD$\varphi$DE & 0.060 & 0.055 & 0.024 & 0.096 & 0.094 & 0.025\\
\hline
New MD$\varphi$DE RIG - 0.01 & 0.074 & 0.070 & 0.034 & 0.076 & 0.073 & 0.039\\
New MD$\varphi$DE RIG - 0.1 & 0.079 & 0.064 & 0.053 & 0.099 & 0.086 & 0.062\\
New MD$\varphi$DE RIG - SJ & 0.091 & 0.075 & 0.068 & 0.120 & 0.099 & 0.078\\
New MD$\varphi$DE MT - 5 & 0.062 & 0.061 & 0.027 & 0.081 & 0.073 & 0.031\\
New MD$\varphi$DE MT - 10 & 0.066 & 0.064 & 0.028 & 0.076 & 0.070 & 0.030\\
New MD$\varphi$DE MT - 15 & 0.069 & 0.068 & 0.028 & 0.076 & 0.071 & 0.030\\
New MD$\varphi$DE MT - 20 & 0.072 & 0.073 & 0.029 & 0.076 & 0.071 & 0.030\\
\hline
Basu-Lindsay MT - 5 & 0.119 & 0.114 & 0.039 & 0.131 & 0.121 & 0.029  \\
Basu-Lindsay MT - 10 & 0.109 & 0.106 & 0.033 & 0.119 & 0.100 & 0.038  \\
Basu-Lindsay MT - 15 & 0.107 & 0.103 & 0.030 & 0.112 & 0.097 & 0.033  \\
\hline
Beran RIG - 0.01 & 0.077 & 0.080 & 0.026 & 0.066 & 0.063 & 0.029\\
Beran RIG - 0.1 & 0.105 & 0.104 & 0.025 & 0.112 & 0.108 & 0.038\\
Beran RIG - SJ & 0.157 & 0.032 & 0.032 & 0.193 & 0.180 & 0.053 \\
Beran MT - 5 & 0.182 & 0.183 & 0.025 & 0.207 & 0.202 & 0.032 \\
Beran MT - 10 & 0.127 & 0.127 & 0.028 & 0.153 & 0.146 & 0.037 \\
Beran MT - 15 & 0.102 & 0.104 & 0.029 & 0.126 & 0.121 & 0.036 \\
\hline
D$\varphi$DE & 0.060 & 0.057 & 0.024 & 0.091 & 0.088 & 0.027\\
\hline
\hline
MDP 0.75 & 0.103 & 0.083 & 0.067 & 0.097 & 0.083 & 0.065 \\
MDP 0.5 & 0.068 & 0.067 & 0.029 & 0.069 & 0.067 & 0.028 \\
MDP 0.25 & 0.062 & 0.058 & 0.026 & 0.064 & 0.062 & 0.029 \\
MDP 0.1 & 0.061 & 0.059 & 0.024 & 0.076 & 0.072 & 0.027 \\
\hline
\hline
MLE & 0.060 & 0.056 & 0.024 & 0.096 & 0.094 & 0.024 \\
\hline
\end{tabular}
\caption{The mean value with the standard deviation of the TVA committed in a 100-run experiment on a two-component Weibull mixture ($\lambda=0.35,\nu_1=0.5,\nu_2=3$). The escort parameter of the D$\varphi$DE is taken as the new MD$\varphi$DE with the SJ bandwidth choice.}
\label{tab:ErrWeibullMixOneMode}
\end{table}

%%%%%%%%%%%%%%%%%%%
%%%%%%%%%%%%%%%%%%%
\begin{table}[hp]
\centering
\resizebox{\textwidth}{!}{
\begin{tabular}{|l|c|c|c|c|c|c||c|c|c|c|c|c|}
\hline
\multirow{2}{2.5cm}{Estimation method} & \multicolumn{6}{|c||}{No Outliers} & \multicolumn{6}{|c|}{$10\%$ Outliers}\\
\cline{2-13}
  & $\lambda$ & sd($\lambda$) & $\nu_1$ & sd$(\nu_1)$ & $\nu_2$ & sd$(\nu_2)$ & $\lambda$ & sd($\lambda$) & $\nu_1$ & sd$(\nu_1)$ & $\nu_2$ & sd$(\nu_2)$\\
 \hline
 \hline
\multicolumn{13}{|c|}{Hellinger} \\
\hline
\hline
Classical MD$\varphi$DE & 0.367 & 0.102 & 0.550 & 0.104 & 1.054 & 0.194 & 0.352  & 0.158 & 0.273 & 0.050 & 1.051 & 0.407\\
\hline
New MD$\varphi$DE - 0.01 & 0.445 & 0.103 & 0.562 & 0.135 & 1.212 & 0.284 & 0.409 & 0.133 & 0.464 & 0.156 & 1.148 & 0.293\\
New MD$\varphi$DE - 0.1 & 0.432 & 0.101 & 0.502 & 0.141 & 1.139 & 0.241 & 0.460  & 0.210 & 0.378 & 0.125 & 1.114 & 0.302\\
New MD$\varphi$DE - SJ & 0.431 & 0.101 & 0.485 & 0.141 & 1.127 & 0.244 & 0.487 & 0.216 & 0.356 & 0.108 & 1.110 & 0.309\\
New MD$\varphi$DE MT - 5 & 0.350 & 0.158 & 0.619 & 0.134 & 1.006 & 0.211 & 0.436 & 0.313 & 0.375 & 0.121 & 1.245 & 1.177\\
New MD$\varphi$DE MT - 10 & 0.338 & 0.148 & 0.643 & 0.135 & 1.019 & 0.167 & 0.474 & 0.322 & 0.409 & 0.140 & 1.150 & 0.516\\
New MD$\varphi$DE MT - 15 & 0.335 & 0.148 & 0.658 & 0.135 & 1.029 & 0.161 & 0.456 & 0.321 & 0.411 & 0.146 & 1.292 & 1.689\\
\hline
Basu-Lindsay MT - 5 & 0.392 & 0.178 & 0.734 & 0.122 & 1.042 & 0.022 & 0.351 & 0.225 & 0.757 & 0.177 & 1.048 & 0.026\\
Basu-Lindsay MT - 10 & 0.340 & 0.149 & 0.742 & 0.103 & 1.037 & 0.024 & 0.260 & 0.175 & 0.712 & 0.147 & 1.039 & 0.024\\
Basu-Lindsay MT - 15 & 0.340 & 0.149 & 0.742 & 0.103 & 1.037 & 0.024 & 0.222 & 0.126 & 0.696 & 0.125 & 1.043 & 0.016\\
\hline
Beran - 0.01 & 0.370 & 0.098 & 0.685 & 0.091 & 1.125 & 0.188 & 0.381 & 0.211 & 0.572 & 0.183 & 1.058 & 0.215\\
Beran - 0.1 & 0.234 & 0.093 & 0.747 & 0.113 & 1.028 & 0.118 & 0.419 & 0.372 & 0.479 & 0.211 & 1.181 & 0.553\\
Beran RIG - SJ & 0.211 & 0.185 & 0.745 & 0.130 & 1.034 & 0.230 & 0.259 & 0.331 & 0.367 & 0.181 & 1.105 & 0.542\\
Beran MT - 5 & 0.302 & 0.205 & 0.584 & 0.129 & 0.867 & 0.120 & 0.471 & 0.388 & 0.376 & 0.128 & 1.097 & 0.738\\
Beran MT - 10 & 0.327 & 0.175 & 0.610 & 0.132 & 0.929 & 0.121 & 0.490 & 0.347 & 0.394 & 0.131 & 1.155 & 0.803\\
Beran MT - 15 & 0.331 & 0.165 & 0.623 & 0.128 & 0.962 & 0.128 & 0.470 & 0.340 & 0.400 & 0.132 & 1.174 & 0.893\\
\hline
D$\varphi$DE & 0.371 & 0.111 & 0.544 & 0.100 & 1.064 & 0.240 & 0.473 & 0.293 & 0.382 & 0.175 & 1.431 & 1.818\\
\hline
\hline
MDPD 0.75 & 0.494 & 0.181 & 0.619 & 0.089 & 1.341 & 0.689 & 0.505 & 0.243 & 0.625 & 0.087 & 1.313 & 0.641\\
MDPD 0.5 & 0.413 & 0.134 & 0.577 & 0.101 & 1.143 & 0.349 & 0.412 & 0.255 & 0.582 & 0.101 & 1.059 & 0.358\\
MDPD 0.25 & 0.366 & 0.108 & 0.542 & 0.110 & 1.064 & 0.349 & 0.554 & 0.348 & 0.503 & 0.117 & 1.205 & 0.995\\
MDPD 0.1 & 0.368 & 0.109 & 0.539 & 0.106 & 1.059 & 0.237 & 0.451 & 0.322 & 0.370 & 0.111 & 1.280 & 1.407\\
\hline
\hline
MLE (EM) & 0.372 & 0.108 & 0.549 & 0.100 & 1.055 & 0.192 & 0.417 & 0.194 & 0.291 & 0.073 & 1.114 & 0.468\\
\hline
\end{tabular}}
\caption{The mean value and the standard deviation of the estimates in a 100-run experiment in a two-component Weibull mixture ($\lambda=0.35,\nu_1=0.5,\nu_2=1$). The escort parameter of the D$\varphi$DE is taken as the new MD$\varphi$DE with Silverman's rule.}
\label{tab:EstimWeibullMixNoMode}
\end{table}

\begin{table}[hp]
\centering
\begin{tabular}{|l|c|c|c||c|c|c|}
\hline
\multirow{2}{2.5cm}{Estimation method} & \multicolumn{3}{|c||}{No Outliers} & \multicolumn{3}{|c|}{$10\%$ Outliers}\\
\cline{2-7}
  & mean & median & sd & mean & median & sd\\
 \hline
 \hline
\multicolumn{7}{|c|}{Hellinger} \\
\hline
\hline
Classical MD$\varphi$DE & 0.056 & 0.055 & 0.026 & 0.124 & 0.114 & 0.035\\
\hline
New MD$\varphi$DE RIG - 0.01 & 0.079 & 0.073 & 0.039 & 0.090 & 0.082 & 0.044\\
New MD$\varphi$DE RIG - 0.1 &  0.079 & 0.065 & 0.059 & 0.112 & 0.101 & 0.050\\
New MD$\varphi$DE RIG - SJ & 0.076 & 0.065 & 0.041 & 0.129 & 0.117 & 0.065\\
New MD$\varphi$DE MT - 5 & 0.063 & 0.058 & 0.029 & 0.114 & 0.095 & 0.041\\
New MD$\varphi$DE MT - 10 & 0.067 & 0.063 & 0.028 & 0.112 & 0.102 & 0.038\\
New MD$\varphi$DE MT - 15 & 0.069 & 0.067 & 0.028 & 0.111 & 0.105 & 0.036\\
\hline
Basu-Lindsay MT - 5 & 0.095 & 0.067 & 0.078 & 0.118 & 0.087 & 0.088 \\
Basu-Lindsay MT - 10 & 0.094 & 0.074 & 0.073 & 0.112 & 0.088 & 0.080 \\
Basu-Lindsay MT - 15 & 0.093 & 0.072 & 0.067 & 0.103 & 0.088 & 0.063 \\
\hline
Beran RIG 0.01 & 0.079 & 0.081 & 0.028 & 0.089 & 0.087 & 0.033\\
Beran RIG 0.1 & 0.087 & 0.085 & 0.023 & 0.103 & 0.102 & 0.025\\
Beran RIG - SJ & 0.094 & 0.092 & 0.023 & 0.100 & 0.097 & 0.021 \\
Beran MT - 5 & 0.061 & 0.060 & 0.022 & 0.127 & 0.134 & 0.044 \\
Beran MT - 10 & 0.059 & 0.055 & 0.025 & 0.115 & 0.096 & 0.041 \\
Beran MT - 15 & 0.060 & 0.056 & 0.025 & 0.112 & 0.097 & 0.039 \\
\hline
D$\varphi$DE & 0.057 & 0.055 & 0.028 & 0.117 & 0.113 & 0.034\\
\hline
\hline
MDPD 0.75 & 0.102 & 0.091 & 0.050 & 0.093 & 0.088 & 0.039 \\
MDPD 0.5 & 0.072 & 0.067 & 0.032 & 0.075 & 0.074 & 0.033 \\
MDPD 0.25 & 0.061 & 0.056 & 0.028 & 0.092 & 0.090 & 0.039 \\
MDPD 0.1 & 0.058 & 0.055 & 0.027 & 0.108 & 0.087 & 0.039 \\
\hline
\hline
MLE & 0.056 & 0.055 & 0.026 & 0.122 & 0.117 & 0.029 \\
\hline
\end{tabular}
\caption{The mean value with the standard deviation of errors committed in a 100-run experiment on a two-component Weibull mixture ($\lambda=0.35,\nu_1=0.5,\nu_2=1$). The escort parameter of the D$\varphi$DE is taken as the new MD$\varphi$DE with the SJ bandwidth choice.}
\label{tab:ErrWeibullMixNoMode}
\end{table}
\clearpage
%%%%%%%%%%%%%%%%%%%%%%%%%%%%%%%%%%%%%%%%%%%%%%%%%%%%%%%%%%%%%
%%%%%%%%%%%%%%%%%%%%%%%%%%%%%%%%%%%%%%%%%%%%%%%
\subsection{Concluding remarks and comments}
We summarize the most important remarks based on our simulations presented above.
\begin{itemize}
\item[$\bullet$] Our kernel-based MD$\varphi$DE gave very good results in all situations. It has the best performance especially in difficult situations, and the lowest sensitivity to the choice of the kernel and the window in the class of $\varphi-$divergence-based estimators. In comparison to the MDPD, results were close in the Gaussian and the Weibull mixture, and the MDPD had slightly better results. In the GPD model, our kernel-based MD$\varphi$DE had clearly better results than the MDPD making it a good competitor.\\
\item[$\bullet$] The execution time of the compared methods varies. Both the classical MD$\varphi$DE and the Basu-Lindsay approach were the most time consuming. The MLE and the MDPD were the best in execution time, whereas both our new kernel-based MD$\varphi$DE and Beran's approach were in the middle with close execution time.\\
\item[$\bullet$] Both the MLE and the classical MD$\varphi$DE have the best performance under the model even in \emph{difficult} models with heavy tails where kernel-based approaches could not give a satisfactory result. In regular situations such as the Gaussian mixture model, all methods were equivalent under the model.\\
\item[$\bullet$] When contamination is present, the compared estimators gave results as expected. Both the MLE and the classical MD$\varphi$DE are not robust against contamination. The D$\varphi$DE guided by our kernel-based MD$\varphi$DE gave very good results under the model. However, when contamination is present, there was no improvement and sometimes a deterioration in the performance in comparison to the escort parameter. This is the case of the Weibull mixtures and the GPD model. The obtained results are still better than MLE and the classical MD$\varphi$DE. \\
\item[$\bullet$] The Basu-Lindsay approach worked very well in regular situations and even showed a slight improvement in efficiency in comparison to the Beran's method which is concordant to the result of \cite{BasuSarkar}. It gave surprisingly good results in the GPD model under contamination when we used the varying KDE in comparison to the situation under the model. Unfortunately, it did not give satisfactory results in the Weibull mixtures. This method seems very sensitive to the kernel under difficult situations since the model is already influenced by the kernel creating a loss of information.\\
\item[$\bullet$] The minimum density power divergence gave very good results in all situations except for the GPD. The best tradeoff parameter from our set of candidates was $a=0.5$. \\
\item[$\bullet$] The Beran's method gave very good tradeoff (and many times the best) between robustness and performance under the model in most of the situations, but not very well in the GPD model. The best choice of the kernel for the GPD and the Weibull mixture was the RIG with window 0.01. It was sensitive to the choice of the kernel and its window in many situations.\\
\item[$\bullet$] The applicability of our kernel-based MD$\varphi$DE to multivariate situations is bound by the use of integration in higher dimensions which is the case of other $\varphi-$divergence-based estimators and the case of the MDPD when applied to mixture models except for the L2 distance (a=1) which still has its limitations. A general solution is to use Monte-Carlo approximation for the integral.\\
\item[$\bullet$] The results obtained using a fixed window for symmetric or asymmetric kernels give rise to an interesting question about the choice of the window. This question will be discussed in future work.
\end{itemize}
%%%%%%%%%%%%%%%%%%%%%%%%%%%%%%%%%%%%%%%%%%%

%%%%%%%%%%%%%%%%%%%%%%%%%%%%%%%%%%%%%%%%%%%

\section{Appendix: Proofs}


\subsection{Proof of Theorem \ref{theo:HellingerConsist}} \label{appendix:proofHellinger}

\begin{proof}
Let $\varepsilon>0$. We want to prove that $\lim_{n\rightarrow\infty} \mathbb{P}\left(\sup_{\phi\in\Phi} |P_nH(P_n,\phi) - P_nh(P_T,\phi)|<\varepsilon\right) = 1$. Since $\varphi$ is twice differentiable (which also implies the differentiability of $\varphi^{\#}$), then by the mean value theorem, there exist two functions $\lambda_1,\lambda_2:\mathbb{R}\rightarrow (0,1)$ such that:
\begin{eqnarray*}
\varphi'\left(\frac{p_{\phi}}{K_{w,n}}\right)(x)-\varphi'\left(\frac{p_{\phi}}{p_T}\right)(x) & = & \varphi''\left(\lambda_1(x)\frac{p_{\phi}}{K_{w,n}}(x)+(1-\lambda_1(x))\frac{p_{\phi}}{p_T}(x)\right) \left[\frac{p_{\phi}}{K_{w,n}}(x) - \frac{p_{\phi}}{p_T}(x)\right], \\
 & = & \varphi''\left(\lambda_1(x)\frac{p_{\phi}}{K_{w,n}}(x)+(1-\lambda_1(x))\frac{p_{\phi}}{p_T}(x)\right)\frac{p_{\phi}}{K_{w,n} p_T}(x) \left[p_T - K_{w,n}\right](x) \\
 & = & \mathcal{A}_n (x,\phi) \left[p_T - K_{w,n}\right](x) \\
\varphi^{\#}\left(\frac{p_{\phi}}{K_{w,n}}\right)(y_i)-\varphi^{\#}\left(\frac{p_{\phi}}{p_T}\right)(y_i) & = & \left(\varphi^{\#}\right)'\left(\lambda_2(y_i)\frac{p_{\phi}}{K_{w,n}}(y_i)+(1-\lambda_2(y_i))\frac{p_{\phi}}{p_T}(y_i)\right) \left[\frac{p_{\phi}}{K_{w,n}} - \frac{p_{\phi}}{p_T}\right](y_i) \\
& = & \left(\varphi^{\#}\right)'\left(\lambda_2(y_i)\frac{p_{\phi}}{K_{w,n}}(y_i)+(1-\lambda_2(y_i))\frac{p_{\phi}}{p_T}(y_i)\right) \frac{p_{\phi}}{K_{w,n}p_T}(y_i)\\
& & \hspace{6cm}\times\left[p_T - K_{w,n}\right](y_i) \\
 & = & \mathcal{B}_n(y_i,\phi) \left[p_T - K_{w,n}\right](y_i).
\end{eqnarray*}
We denoted:
\begin{eqnarray*}
\mathcal{A}_n (x,\phi) & = & \varphi''\left(\lambda_1(x)\frac{p_{\phi}}{K_{w,n}}(x)+(1-\lambda_1(x))\frac{p_{\phi}}{p_T}(x)\right)\frac{p_{\phi}}{K_{w,n} p_T}(x) \\
\mathcal{B}_n(y_i,\phi) & = & \left(\varphi^{\#}\right)'\left(\lambda_2(y_i)\frac{p_{\phi}}{K_{w,n}}(y_i)+(1-\lambda_2(y_i))\frac{p_{\phi}}{p_T}(y_i)\right) \frac{p_{\phi}}{K_{w,n}p_T}(y_i).
\end{eqnarray*}
Let $n$ be sufficiently large such that:
\[\sup_{x}\left|K_{w}*P_n(x) - p_T(x)\right| \leq \min\left(\varepsilon,\frac{\varepsilon}{\mathcal{A}_n},\frac{\varepsilon}{\mathcal{B}_n}\right) \]
where $\mathcal{A}_n = \sup_{\phi}\int{\mathcal{A}_n(x)dx}$ and $\mathcal{B}_n = \sup_{\phi}\frac{1}{n}\sum{\mathcal{B}_n(y_i)}$ which exist by virtue of assumption 3 of the present theorem on the one hand and on the other hand the fact that functions $x\mapsto\lambda_1(x)$ and $x\mapsto\lambda_2(x)$ are bounded uniformly inside $(0,1)$. This event occurs with probability $1-\eta_n$ with $\eta_n\rightarrow 0$ by the strong consistency assumption (point 2). This implies that both events:
\begin{eqnarray*}
\left|\int{\left[\varphi'\left(\frac{p_{\phi}}{K_{w}*P_n}\right)-\varphi'\left(\frac{p_{\phi}}{p_T}\right)\right]p_{\phi}}\right| &\leq & \varepsilon, \\
\left|\frac{1}{n}\sum_{i=1}^n{\varphi^{\#}\left(\frac{p_{\phi}}{K_{w}*P_n}\right)(y_i) - \varphi^{\#}\left(\frac{p_{\phi}}{p_T}\right)(y_i)}\right| & \leq & \varepsilon
\end{eqnarray*}
happen with probability greater than $1-\eta_n$ independently of $\phi$. Finally, we conclude that
\[\mathbb{P}\left(\sup_{\phi\in\Phi} |P_nH(P_n,\phi) - P_nh(P_T,\phi)|<2\varepsilon\right) \geq 1-\eta_n,\]
and hence $\sup_{\phi\in\Phi} |P_nH(P_n,\phi) - P_nh(P_T,\phi)|\rightarrow 0$ in probability. To end the proof, we use assumption 3 together with the result in Example 19.9 in \cite{Vaart} Chap. 19 which imply that $\{\varphi^{\#}\left(\frac{p_{\phi}}{p_T}\right),\phi\in\Phi\}$ is a Glivenko-Cantelli class of functions. Hence, $\sup_{\phi\in\Phi} |P_Th(P_T,\phi) - P_nh(P_T,\phi)|\rightarrow 0$ in probability. Using inequality (\ref{eqn:DecompProofIneq}), we conclude that $\sup_{\phi\in\Phi} |P_nH(P_n,\phi) - P_Th(P_T,\phi)| \rightarrow 0$ in probability. Finally, the previous arguments prove the first point (\ref{eqn:ConsistP1}) in Theorem \ref{theo:VanderVaart}. The second point in Theorem \ref{theo:VanderVaart} is the same as assumption 4 of the present theorem. By definition of the kernel-based MD$\varphi$DE as a minimum of the criterion function $\phi\mapsto P_nH(P_n,\phi)$, Theorem 1 entails the consistency of our new estimator. $\quad \Box$
\end{proof}
\subsection{Proof of Theorem \ref{theo:NeymChi2Consist}}\label{appendix:proofNeymChi2}
We follow the same idea of the proof of Theorem \ref{theo:HellingerConsist}. In order to treat the second term in the right hand side of equation (\ref{eqn:PHdifferenceNeymChi2}), we use the uniform continuity of function $t\mapsto t^{-\gamma}$. Indeed, if $|K_w*P_n(x) - p_T(x)|<\delta_2$, then:
\[|(K_w*P_n)^{-\gamma}(y_i) - p_T^{-\gamma}(y_i)|<\frac{\varepsilon}{\sup_{\phi}\frac{1}{n}\sum{p_{\phi}^{\gamma}(y_i)}}.\]
By the consistency of the kernel estimator, the previous inequality happens with probability $1-\eta_n$ with $\eta_n\rightarrow 0$. Thus, $\mathcal{B}_n$ of Theorem \ref{theo:HellingerConsist} is now replaced by the simpler quantity 
\begin{equation}
\mathcal{B}_n = \sup_{\phi}\frac{1}{n}\sum_{i=1}^n{p_{\phi}^{\gamma}(y_i)}.
\label{eqn:BnIneq}
\end{equation}
On the other hand, in order to treat the first term in the right hand side of equation (\ref{eqn:PHdifferenceNeymChi2}), we rewrite the integral as follows:
\begin{multline*}
\int{\frac{\left(K_w*P_n\right)^{-\gamma+1}(x) - p_T^{-\gamma+1}(x)}{p_{\phi}^{-\gamma}(x)}dx} = \\
\int{\frac{\left[\left(K_w*P_n\right)^{\frac{-\gamma+1}{2}}(x) - p_T^{\frac{-\gamma+1}{2}}(x)\right]\left[\left(K_w*P_n\right)^{\frac{-\gamma+1}{2}}(x) + p_T^{\frac{-\gamma+1}{2}}(x)\right]}{p_{\phi}^{-\gamma}(x)}dx}.
\end{multline*}
Now, using the uniform continuity of function\footnote{notice that $\frac{-\gamma+1}{2}\in(0,1)$ since $\gamma\in(-1,0)$.} $t\mapsto t^{\frac{-\gamma+1}{2}}$, we may deduce that if $|K_w*P_n(x) - p_T(x)|<\delta_1$, then:
\begin{equation}
|(K_w*P_n)^{\frac{-\gamma+1}{2}}(x) - p_T^{\frac{-\gamma+1}{2}}(x)|<\frac{\varepsilon}{\sup_{\phi}\int{\frac{\left(K_w*P_n\right)^{\frac{-\gamma+1}{2}}(x) + p_T^{\frac{-\gamma+1}{2}}(x)}{p_{\phi}^{-\gamma}(x)}dx}}. 
\label{eqn:AnIneq}
\end{equation}
Again, by the consistency of the kernel estimator, the previous inequality happens with probability $1-\eta_n$ with $\eta_n\rightarrow 0$. Thus $\mathcal{A}_n$ of Theorem \ref{theo:HellingerConsist} is now replaced by the quantity
\[\mathcal{A}_n = \sup_{\phi}\int{\frac{\left(K_w*P_n\right)^{\frac{-\gamma+1}{2}}(x) + p_T^{\frac{-\gamma+1}{2}}(x)}{p_{\phi}^{-\gamma}(x)}dx}.\]
Existness and finitness of both $\mathcal{A}_n$ and $\mathcal{B}_n$ in probability are ensured by assumptions 3 and 4. Now, using inequalities (\ref{eqn:AnIneq}) and (\ref{eqn:BnIneq}), both events 
\begin{eqnarray*}
\left|\int{\frac{\left(K_{w}*P_n\right)^{1-\gamma}-p_T^{1-\gamma}}{p_{\phi}^{-\gamma}}(x)dx}\right|& < & \varepsilon; \\
\left|\frac{1}{n}\sum_{i=1}^n{\frac{\left(K_{w}*P_n\right)^{-\gamma}-p_T^{-\gamma}}{p_{\phi}^{-\gamma}}(y_i)}\right|& < & \varepsilon ,
\end{eqnarray*}
happen with probability greater than $1-\eta_n$ independently of $\phi$. Finally, we conclude that
\[\mathbb{P}\left(\sup_{\phi\in\Phi} |P_nH(P_n,\phi) - P_nh(P_T,\phi)|<2\varepsilon\right) \geq 1-\eta_n,\]
and hence $\sup_{\phi\in\Phi} |P_nH(P_n,\phi) - P_nh(P_T,\phi)|\rightarrow 0$ in probability. To end the proof, we use assumption 2 together with the Glivenko-Cantelli theorem to deduce that $\sup_{\phi\in\Phi} |P_Th(P_T,\phi) - P_nh(P_T,\phi)|\rightarrow 0$ in probability. Using inequality (\ref{eqn:DecompProofIneq}), we conclude that $\sup_{\phi\in\Phi} |P_nH(P_n,\phi) - P_Th(P_T,\phi)| \rightarrow 0$ in probability. Finally, the previous arguments prove the first point (\ref{eqn:ConsistP1}) in Theorem \ref{theo:VanderVaart}. The second point in Theorem \ref{theo:VanderVaart} is the same as assumption 4 of the present theorem. By definition of the kernel-based MD$\varphi$DE as a minimum of the criterion function $\phi\mapsto P_nH(P_n,\phi)$, Theorem 1 entails the consistency of our new estimator. $\quad \Box$

%%%%%%%%%%%%%%%%%%%%%%%%%%%%%%%%%
%\subsection{Proof of Theorem \ref{theo:KLConsist}}\label{appendix:proofKL}
%For the first term in the right hand side of equation (\ref{eqn:PHDifferenceKLConsist}), we use the mean value theorem in order to show explicitly the difference $|K_{n,w}(x) - p_T(x)|$. For the second term, a simple development of the fractions shows this difference.\\
%By the mean value theorem, there exists $\lambda(x)\in(0,1)$ such that:
%\begin{eqnarray*}
%\left|\log\left(\frac{p_T(x)}{K_{n,w}(x)}\right)\right| & = & \frac{|K_{n,w}(x) - p_T(x)|}{\lambda(x)p_T + (1-\lambda(x))K_{n,w}(x)} \\
 %& \leq & |K_{n,w}(x) - p_T(x)| \left[\frac{1}{p_T(x)} + \frac{1}{K_{n,w}(x)}\right].
%\end{eqnarray*}
%The second line comes from the fact that function $t\mapsto\frac{1}{t}$ is convex and the fact that $\lambda(x)\in(0,1)$. The remainging of the proof follows the same arguments presented in the proofs of Theorems \ref{theo:HellingerConsist} and \ref{theo:NeymChi2Consist}.$\quad \Box$

%%%%%%%%%%%%%%%%%%%%%%%%%%%%%%%%%%%%%%%%%%%
\subsection{Proof of Theorem \ref{theo:NormalAsyptotNewMD}}\label{Append:TheoNormalAsyptot}
In the whole proof, the index $T$ will be omitted from $\phi^T$ for the sake of clarity. We start with calculating the gradient $\nabla P_nH(P_n,\phi)$.
\begin{equation}
\nabla P_nH(P_n,\phi) = \frac{\gamma}{\gamma-1}\int{\nabla p_{\phi} \frac{p_{\phi}^{\gamma-1}}{K_{n,w}^{\gamma-1}}dx} - \frac{1}{n}\sum_{i=1}^n{\nabla p_{\phi}\frac{p_{\phi}^{\gamma-1}}{K_{n,w}^{\gamma}}(y_i)}.
\label{eqn:GradHnAsymptotNorm}
\end{equation}
We treat each term separately. The first term can be rewritten as:
\[\int{\nabla p_{\phi} \frac{p_{\phi}^{\gamma-1}}{K_{n,w}^{\gamma-1}}dx} = \int{\nabla p_{\phi} p_{\phi}^{\gamma-1}\left[K_{n,w}^{1-\gamma} - p_{\phi}^{1-\gamma}\right]dx} + \int{\nabla p_{\phi}dx}.\]
The second term in the right hand side is zero because $p_{\phi}$ is a density, provided changeability between integration and differentiation. For the first term, we write a second order Taylor expansion of function $t\mapsto t^{1-\gamma}$:
\[K_{n,w}^{1-\gamma} - p_{\phi}^{1-\gamma} = (1-\gamma)(K_{n,w}-p_{\phi})p_{\phi}^{-\gamma} + \frac{-\gamma}{2}\left(K_n-p_{\phi}\right)^2M_n(x)^{-\gamma-1},\]
where $M_n(x)$ is a point in between $K_{n,w}(x)$ and $p_{\phi}(x)$. We now have:
\begin{multline}
\int{\nabla p_{\phi} p_{\phi}^{\gamma-1}\left[K_{n,w}^{1-\gamma} - p_{\phi}^{1-\gamma}\right]dx} = (1-\gamma)\int{\nabla p_{\phi} p_{\phi}^{-1}\left[K_{n,w} - p_{\phi}\right]dx} + \\ \frac{-\gamma}{2}\int{\nabla p_{\phi} p_{\phi}^{-1}M_n(x)^{-\gamma-1}\left[K_{n,w} - p_{\phi}\right]^2dx}.
\label{eqn:AsymptotDiffIntegrals}
\end{multline}
Using equations (3.11-3.13) from \cite{Beran}, we may write:
\begin{equation}
\sqrt{n}\int{\nabla p_{\phi} p_{\phi}^{-1}\left[K_{n,w} - p_{\phi}\right]dx} \xrightarrow[\mathcal{L}]{} \mathcal{N}(0,S),
\label{eqn:FirstPart1stterm}
\end{equation}
where $S=\int{\nabla p_{\phi}\nabla p_{\phi}^t dx}$.\\
The second term will be handled in a similar way to equations (3.11-3.13) from \cite{Beran}. Let $K_n(x) = K(x/w_n)/w_n$. Write
\begin{multline}
\int{\nabla p_{\phi} p_{\phi}^{-1}M_n(x)^{-\gamma-1}\left[K_{n,w} - p_{\phi}\right]^2dx} = \int{\nabla p_{\phi} p_{\phi}^{-1}M_n(x)^{-\gamma-1}\left[K_{n,w} - K_n*P_{\phi}\right]^2dx} + \\ \resizebox{\textwidth}{!}{$2\int{\nabla p_{\phi} p_{\phi}^{-1}M_n(x)^{-\gamma-1}\left[K_{n,w} - K_n*P_{\phi}\right]\left[K_n*P_{\phi} - p_{\phi}\right]dx} + \int{\nabla p_{\phi} p_{\phi}^{-1}M_n(x)^{-\gamma-1}\left[K_n*P_{\phi} - p_{\phi}\right]^2dx}$},
\label{eqn:AsymptotSquaredTerms}
\end{multline}
and prove that each term has a limit equal to zero when multiplied by $\sqrt{n}$. There are two essential arguments. The first one uses equation (3.11) from \cite{Beran} to write:
\begin{equation}
\sup_{x}\left|K_n*P_{\phi} - p_{\phi}\right| \leq \frac{w^2}{2}\sup_{x}\left|p_{\phi}''(x)\right|\int{x^2K(x)dx}.
\label{eqn:Beran311}
\end{equation}
The second one is a result of Corollary 5 from \cite{WiedWeibbach}:
\begin{equation}
\lim_{n\rightarrow\infty} \sqrt{\frac{nw}{-2\log{w}}} \sup_{x}\frac{\left|K_{n,w}-K_n*P_{\phi}\right|}{\sqrt{p_{\phi}}} = \left(\int{K^2(y)dy}\right).
\label{eqn:Cor5Wied}
\end{equation}
We treat the first term in equation (\ref{eqn:AsymptotSquaredTerms}) using equation (\ref{eqn:Cor5Wied}).
\begin{eqnarray*}
\sqrt{n}\int{\left|\nabla p_{\phi}\right| p_{\phi}^{-1}M_n(x)^{-\gamma-1}\left[K_{n,w} - K_n*P_{\phi}\right]^2dx} & \leq & \left[\sqrt{\frac{nw}{-2\log{w}}} \sup_{x}\frac{\left|K_{n,w}-K_n*P_{\phi}\right|}{\sqrt{p_{\phi}}}\right]^2 \\
&  & \times \frac{-2\log(w)}{n^{1/2}w} \int{\left|\nabla p_{\phi}\right|M_n(x)^{-\gamma-1}dx} \\
& = & \mathcal{O}\left(\frac{-2\log(w)}{n^{1/2}w}\right).
\end{eqnarray*}
We treat the second term in equation (\ref{eqn:AsymptotSquaredTerms}) using equations (\ref{eqn:Cor5Wied}) and (\ref{eqn:Beran311}).
\begin{eqnarray*}
\sqrt{n}\int{\frac{|\nabla p_{\phi}|}{p_{\phi}M_n(x)^{\gamma+1}}\left[K_{n,w} - K_n*P_{\phi}\right]\left[K_n*P_{\phi} - p_{\phi}\right]dx} & \leq & \sqrt{\frac{nw}{-2\log{w}}} \sup_{x}\frac{\left|K_{n,w}-K_n*P_{\phi}\right|}{\sqrt{p_{\phi}}} \\
 \times \sup_{x}\left|p_{\phi}''(x)\right| \int{x^2K(x)dx} \sqrt{-2\log(w)} & \frac{w^{3/2}}{2} &  \int{|\nabla p_{\phi}| p_{\phi}^{-1/2}M_n(x)^{-\gamma-1}dx} \\
& = & \mathcal{O}\left(w^{3/2}\sqrt{-2\log(w)}\right).
\end{eqnarray*}
We treat the third term in equation (\ref{eqn:AsymptotSquaredTerms}) using equation (\ref{eqn:Beran311}).
\begin{eqnarray*}
\sqrt{n}\int{|\nabla p_{\phi}| p_{\phi}^{-1}M_n(x)^{-\gamma-1}\left[K_n*P_{\phi} - p_{\phi}\right]^2dx} & \leq & \sqrt{n}\frac{h^4}{2}\sup_x\left|p_{\phi}''(x)\right|\left[\int{x^2K(x)dx}\right]^2 \\
 & & \times \int{\frac{|\nabla p_{\phi}|}{p_{\phi}M_n(x)^{\gamma+1}}dx} \\
& = & \mathcal{O}\left(n^{1/2}h^4\right).
\end{eqnarray*}
We conclude using assumption 3 that :
\[\sqrt{n}\int{\nabla p_{\phi} p_{\phi}^{-1}M_n(x)^{-\gamma-1}\left[K_{n,w} - p_{\phi}\right]^2dx} \xrightarrow[\mathbb{P}]{} 0.\]
This entails together with (\ref{eqn:FirstPart1stterm}) that the first term in $P_nH(P_n,\phi)$ multiplied by $\sqrt{n}$ is a centered multivariate Gaussian with covariance matrix $S$.\\
The sum term in $P_nH(P_n,\phi)$ can be treated similarly. Firstly, write:
\[\frac{1}{\sqrt{n}}\sum_{i=1}^n{\nabla p_{\phi}\frac{p_{\phi}^{\gamma-1}}{K_{n,w}^{\gamma}}(y_i)} = \frac{1}{\sqrt{n}}\sum_{i=1}^n{\nabla p_{\phi}p_{\phi}^{\gamma-1}\left[K_{n,w}^{-\gamma} - p_{\phi}^{-\gamma}\right](y_i)} + \frac{1}{\sqrt{n}}\sum_{i=1}^n{\frac{\nabla p_{\phi}}{p_{\phi}}(y_i)}.\]
Now the second term in the right hand side, is asymptotically Gaussian with mean zero and covariance matrix equal to $S$. For the first term, we apply the mean value theorem on function $z^{-\gamma}$. There exists a bounded function $M_n(y_i)$ in between $K_{n,w}(y_i)$ and $p_{\phi}(y_i)$ such that:
\begin{multline*}
\frac{1}{\sqrt{n}}\sum_{i=1}^n{\nabla p_{\phi}p_{\phi}^{\gamma-1}\left[K_{n,w}^{-\gamma} - p_{\phi}^{-\gamma}\right](y_i)} = \frac{-\gamma}{\sqrt{n}}\sum_{i=1}^n{\frac{\nabla p_{\phi}}{p_{\phi}^2}\left[K_{n,w} - p_{\phi}\right](y_i)} \\ + \frac{\gamma(\gamma+1)}{\sqrt{n}}\sum_{i=1}^n{\frac{\nabla p_{\phi}}{p_{\phi}^3}\left[K_{n,w} - p_{\phi}\right]^2(y_i)}
\end{multline*}
The treatment of the second term in the right hand side can be done similarly to the second term in equation (\ref{eqn:AsymptotDiffIntegrals}) and thus converges to zero when multiplied by $\sqrt{n}$. The first term will be proved to have the same asymptotic behavior to the first term in equation (\ref{eqn:AsymptotDiffIntegrals}). Write the difference between these terms. Let $\psi(x) = \frac{\nabla p_{\phi}}{p_{\phi}^2}$.
\begin{eqnarray*}
\sqrt{n}\left|\frac{-1}{\sqrt{n}}\sum_{i=1}^n{\psi(x)\left[K_{n,w} - p_{\phi}\right](y_i)} - \int{\psi(x)\left[K_{n,w} - p_{\phi}\right]p_{\phi}(x)dx}\right| & \leq &\\
\sup_x|K_{n,w}(x) - p_{\phi}(x)|\
\sqrt{n}\left|\int{\psi(x)\left[K_{n,w} - p_{\phi}\right](dP_n - dP_{\phi})(x)}\right| & = & \\
\sqrt{n}\left|\int{(\mathbb{F}_n - \mathbb{F}_{\phi})(x) d\left(\psi(x)\left[K_{n,w} - p_{\phi}\right]\right)(x)}\right| & \leq & \\
\sqrt{n}\sup|\mathbb{F}_n-\mathbb{F}_{\phi}|\left[\sup|K_{n,w}' - p_{\phi}'|\int{\psi(x)dx}+\sup|K_{n,w} - p_{\phi}|\int{\psi'(x)dx} \right] & . &
\end{eqnarray*}
Now, using rates of convergence of the empirical distribution function (see for example \cite{Vaart} p. 268), the kernel density estimator (see for example \cite{Bordes10} Lemma 3.1) and the derivative of the kernel density estimator (see \cite{Schuster} Theorem 2.5), we prove easily that the right hand side of the inequality in the previous display tends to zero in probability. This proves our claim. Now it remains to use the asymptotic normality limit in equation (\ref{eqn:FirstPart1stterm}) to deduce that:
\[\sqrt{n}\left(\frac{-1}{\sqrt{n}}\sum_{i=1}^n{\frac{\nabla p_{\phi}}{p_{\phi}^2}\left[K_{n,w} - p_{\phi}\right](y_i)}\right) \xrightarrow[\mathcal{L}]{} \mathcal{N}(0,S).\]
Collecting the three pieces which generate the asymptotic normality in the whole calculus, we may conclude that:
\[\sqrt{n}\nabla P_nH(P_n,\phi) \xrightarrow[\mathcal{L}]{} \mathcal{N}\left(0,(2\gamma^2+1)\int{\nabla p_{\phi}\nabla p_{\phi}^t}\right).\]
The matrix of second order partial derivatives $J_{P_nH(P_n,.)}$ can be treated in an easier way than the vector $\nabla P_nH(P_n,\phi)$. It can be shown using similar techniques to those used here above that $J_{P_nH(P_n,.)}$ converges in probability at rate $o_P(n^{-1/2})$. We may conclude now that the asymptotic normality result (\ref{eqn:AsymptotNormalResNewMDphiDE}) holds.
%%%%%%%%%%%%%%%%%%%%%%%%%%%%%%%%%%%%%%%%%%%
\subsection{Proof of Theorem \ref{theo:IF}}\label{appendix:proofIF}
For a clearer writing, we omit the index $T$ from $P_T$ in this proof. Deriving the left hand side of the estimating equation (\ref{eqn:EstimEq}) gives:
\begin{multline*}
\frac{\gamma}{\gamma-1} \int{\frac{\left[(\gamma-1)\nabla p_{C(P)}\left(\nabla p_{C(P)}\right)^t+p_{C(P)}J_{p_{C(P)}}\right]p_{C(P)}^{\gamma-2}}{(K_w*P)^{\gamma-1}}} \text{IF}(P,Q) \\ - \gamma\int{\frac{p_{C(P)}^{\gamma-1}\left[K_w*(Q-P)\right]\nabla p_{C(P)}}{(K_w*P)^{\gamma}}(x) dx}.
\end{multline*}
Deriving the right hand side of the estimating equation (\ref{eqn:EstimEq}) gives:
\begin{multline*}
\int{\frac{\left[(\gamma-1)\nabla p_{C(P)}\left(\nabla p_{C(P)}\right)^tp_{C(P)}^{\gamma-2}+p_{C(P)}^{\gamma-1}J_{p_{C(P)}}\right]}{(K_w*P)^{\gamma}}(x)dP(x)}\text{IF}(P,Q)\\ - \gamma\int{\frac{p_{C(P)}^{\gamma-1}\left[K_w*(Q-P)\right]\nabla p_{C(P)}}{(K_w*P)^{\gamma+1}}(x) dP(x)}
+ \int{\frac{p_{C(P)}^{\gamma-1}\nabla p_{C(P)}}{(K_w*P)^{\gamma}}(x)(dQ-dP)(x)}.
\end{multline*}
We have now:
\begin{multline*}
A\; \text{IF}(P,Q) = \gamma\int{\frac{p_{C(P)}^{\gamma-1}\left[K_w*(Q-P)\right]\nabla p_{C(P)}}{(K_w*P)^{\gamma}}(x) dx} + \int{\frac{p_{C(P)}^{\gamma-1}\nabla p_{C(P)}}{(K_w*P)^{\gamma}}(x)(dQ-dP)(x)} \\
-\gamma\int{\frac{p_{C(P)}^{\gamma-1}\left[K_w*(Q-P)\right]\nabla p_{C(P)}}{(K_w*P)^{\gamma+1}}(x) dP(x)},
\end{multline*}
where $A$ is defined by formula (\ref{eqn:IFmatrixA}). Assuming that $A$ is invertible and using the estimating equation (\ref{eqn:EstimEq}), we can write:
\[
\text{IF}(P,Q) = \gamma A^{-1}\int{\frac{p_{C(P)}^{\gamma-1}\left[K_w*Q\right]\nabla p_{C(P)}}{(K_w*P)^{\gamma}}\left(1-\frac{p}{K*P}\right)(x) dx} + A^{-1}\int{\frac{p_{C(P)}^{\gamma-1}\nabla p_{C(P)}}{(K_w*P)^{\gamma}}(x)dQ(x)}.
\]
The remaining of the proof is a simple substitution of $C(P)$ by $\phi^T$ when $P=P_{\phi^T}$, and replacing $Q$ by the dirac measure on a point $x_0$. $\quad \Box$

\chapter{Iterative Proximal-Point Algorithm for the Calculus of Divergence-Based Estimators with Application to Mixture Models}
In the previous chapter, we have presented and introduced several estimators; an estimator based on Beran's approach (\ref{eqn:BeranEstimator}), an estimator based on the Basu-Lindsay approach (\ref{eqn:BasuLindsayDiv}), the MD$\varphi$DE (\ref{eqn:MDphiDEClassique}), the D$\varphi$DE (\ref{eqn:DphiDE}), our new kernel-based MD$\varphi$DE (\ref{eqn:NewMDphiDE}) and the MDPD (\ref{eqn:MDPDdef}). All these estimators, the MLE included, are in general non convex (or non concave for the D$\varphi$DE) optimization problems. The calculus of these estimators in general is then not guaranteed to give a good result for a finite-sample setup when we use any standard optimization algorithm. There exist several optimization algorithms such as Gradient descent algorithms (first and second order gradient descent and gradient-conjugate algorithms), the BFGS algorithm, the Nelder-Mead's algorithm, Brent's algorithm among others, see \cite{OptimKenneth}. These algorithms guarantee the convergence of the iterative procedure to a global optimum whose objective function it is a strictly convex (or concave) function. If it is not the case, the algorithm converges to a local optimum. Each optimization method has its own advantages and drawbacks. There are also some algorithms which treat functions which can be written as the difference of two convex functions called convex-concave optimization algorithms, see \cite{CCCP}. These algorithms, for example, give in general better results than convex optimization algorithms for this kind of functions. \\

There is on the other hand, another type of optimization algorithms which attack a modified version of the objective function, say $D(\phi) + g(\phi,\phi^k)$, where $D$ is the objective function and $g$ is a perturbation function which depends on the current iteration $k$. A perturbation of the objective function has a goal of giving it a "better form". The iterative procedure then proceeds to optimize the modified function iteratively as the perturbation becomes less and less important as the number of the iteration increases. This kind of algorithms is called proximal-point algorithms. It was first proposed by \cite{Martinet} who used a perturbation of the form $g(\phi,\phi^k)=\|\phi-\phi^{k}\|$. Generally, the proximal term has a regularization effect in the sense that a proximal point algorithm is more stable and frequently outperforms classical optimization algorithms, see \cite{Goldstein}. Furthermore, and as mentioned in \citep{ChretienHeroProxGener}, proximal point algorithms permit to avoid saddle points.\\
The EM algorithm is a very interesting example of proximal point algorithms, see paragraph \ref{subsec:EMProximalForm} for a detailed calculus or the papers of \cite{ChretienHero} and \cite{Tseng}. Indeed, one may rewrite the conditional expectation of the complete log-likelihood as a sum of the log-likelihood function and a distance-like function over the conditional densities of the labels provided an observation. Thus, the EM algorithm has the log-likelihood as an objective function which is being perturbed by a distance-like function. Chr\'etien and Hero \cite{ChretienHeroAccel} proved superlinear convergence of a proximal point algorithm derived by the EM algorithm. Notice that EM-type algorithms usually enjoy no more than linear convergence.\\

Taking into consideration the need for robust estimators, and the fact that the MLE is the least robust estimator among the class of divergence-type estimators, we generalize the EM algorithm (and the version in \cite{Tseng}) by replacing the log-likelihood function by an estimator of a $\varphi-$divergence between the true distribution of the data and the model. We, thus, propose to calculate divergence-based estimators mentioned here above using a proximal-point algorithm based on the work of \cite{Tseng} on the log-likelihood function. This proximal-point algorithm extends the EM algorithm. Our convergence proof of the iterative procedure requires some regularity of the estimated divergence with respect to the parameter vector which can be easily checked using Lebesgue theorems except for the dual formula (\ref{eqn:DivergenceDef}). Indeed, the supremal form of the estimated divergence in the dual formula complicates the situation. Recent results in \cite{Rockafellar} provide sufficient conditions to solve this problem. It may at time be very difficult to prove that the objective function is differentiable with respect to $\phi$, therefore, our results cover the case when the objective function is not differentiable. \\
We also propose a two-step iterative algorithm to calculate divergence-based estimators for mixture models motivated by the EM algorithm; a step to calculate the proportion and a step to calculate the parameters of the components. Proofs for this simplified version become more technical. The goal of this simplification is to reduce the dimension over which we optimize since in lower dimensions, optimization procedures are more efficient\footnote{This does not cover all optimization methods. For example, the Nelder-Mead algorithm is considered as "unreliable" in univariate optimization. The Brent method can be used as an alternative. Note that these two algorithms are suitable for non differentiable functions since they only use function values to reach an optimum.}. \\
Another contribution of this work concerns the assumptions ensuring the convergence of the algorithm. In the previous works on such type of proximal algorithms such as the papers of \cite{Tseng} and \cite{ChretienHero}, the proximal term is supposed to verify an identifiability property. In other words $g(\phi,\phi')=0$ if and only if $\phi=\phi'$. We show that such property is difficult to verify and it is often not fulfilled in mixture models. We provide a way to relax such condition without imposing further assumptions.

%%%%%%%%%%%%%%%%%%%%%%%%%%%%%%%%%%%%%%%%%%%%%%%%%%%%%%%%%%%%%%%%%%%%%%%%%%%%
%
%--------------------------------------------------------------------------------------
% =======================================================================================
%--------------------------------------------------------------------------------------
%
%%%%%%%%%%%%%%%%%%%%%%%%%%%%%%%%%%%%%%%%%%%%%%%%%%%%%%%%%%%%%%%%%%%%%%%%%%%%

\section{Development of the proximal-point algorithm from the EM algorithm}\label{sec:IntrodPart}
\subsection{General context and notations}
Let $(X,Y)$ be a couple of random variables with joint probability density function $f(x,y|\phi)$ parametrized by a vector of parameters $\phi\in\Phi\subset\mathbb{R}^d$. Let $(X_1,Y_1),\cdots,(X_n,Y_n)$ be n copies of $(X,Y)$ independently and identically distributed. Finally, let $(x_1,y_1),\cdots,(x_n,y_n)$ be n realizations of the n copies of $(X,Y)$. The $x_i$'s are the unobserved data (labels) and the $y_i$'s are the observations. The vector of parameters $\phi$ is unknown and need to be estimated. \\
The observed data $y_i$ are supposed to be real vectors and the labels $x_i$ belong to a space $\mathcal{X}$ not necessarily finite unless mentioned otherwise. Denote $dx$ the measure on the label space $\mathcal{X}$ (for example the counting measure if $\mathcal{X}$ is discrete). The marginal density of the observed data is given by $p_{\phi}(y)=\int{f(x,y|\phi)}dx$.\\
For a parametrized function $f$ with a parameter $a$, we write $f(x|a)$. We use the notation $\phi^k$ for sequences with the index above. Derivatives of a real valued function $\psi$ defined on $\mathbb{R}$ are written as $\psi',\psi'',$ etc. We use $\nabla f$ for the gradient of real function $f$ defined on $\mathbb{R}^d$, $\partial f$ to its subgradient and $J_f$ to the matrix of second order partial derivatives. For a generic function $H$ of two variables $(\phi,\theta)$, $\nabla_1 H(\phi,\theta)$ denotes the gradient with respect to the first (vectorial) variable $\phi$. 
%--------------------------------------------------------------------------------------
% =======================================================================================
%--------------------------------------------------------------------------------------
\subsection{EM algorithm and Tseng's generalization}\label{subsec:EMProximalForm}
The EM algorithm is a well-known method for calculating the maximum likelihood estimator of a model where incomplete data is considered. For example, when working with mixture models in the context of clustering, the labels or classes of observations are unknown during the training phase. Several variants of the EM algorithm were proposed, see \cite{McLachlanEM}. The EM algorithm estimates the unknown parameter vector by generating the sequence (see \citep{Dempster}):
\begin{eqnarray*}
\phi^{k+1} & = & \argmax_{\Phi} Q(\phi,\phi^k) \\
					 & = & \argmax_{\Phi} \mathbb{E}\left[\log(f(\textbf{X},\textbf{Y}|\phi)) \left| \textbf{Y}=\textbf{y},\phi^k\right.\right],
\end{eqnarray*}
where $\textbf{X} = (X_1,\cdots,X_n)$, $\textbf{Y} = (Y_1,\cdots,Y_n)$ and $\textbf{y}=(y_1,\cdots,y_n)$. By independence between the couples $(X_i,Y_i)$'s, the previous iteration may be rewritten as:
\begin{eqnarray}
\phi^{k+1} & = & \argmax_{\Phi} \sum_{i=1}^n{\mathbb{E}\left[\log(f(X_i,Y_i|\phi)) \left| Y_i=y_i,\phi^k\right.\right]} \nonumber\\
					 & = & \argmax_{\Phi} \sum_{i=1}^n\int_{\mathcal{X}}{\log(f(x,y_i|\phi)) h_i(x|\phi^k) dx},
\label{eqn:EMAlgo}
\end{eqnarray}
where $h_i(x|\phi^k)$ is the conditional density of the labels (at step $k$) provided $y_i$. It is given by:
\begin{equation}
h_i(x|\phi^k) = \frac{f(x,y_i|\phi^k)}{p_{\phi^k}(y_i)}.
\label{eqn:ConditionalDensLabel}
\end{equation}
This justifies the recurrence equation given by \citep{Tseng}. It is slightly different from the EM recurrence defined in \citep{Dempster}. The conditional expectation of the logarithm of the complete likelihood provided the data and the parameter vector of the previous iteration is calculated, here, on the vector of observed data. The expectation is replaced by an integral against the corresponding conditional density of the labels. \\
It is well-known that the EM iterations can be rewritten as a difference between the log-likelihood and a \emph{Kullback-Liebler} distance-like function. Indeed, using (\ref{eqn:ConditionalDensLabel}) in (\ref{eqn:EMAlgo}), one can write:
\begin{eqnarray*}
\phi^{k+1} & = & \argmax_{\Phi} \sum_{i=1}^n\int_{\mathcal{X}}{\log\left(h_i(x|\phi)\times p_{\phi}(y_i)\right) h_i(x|\phi^k) dx} \\
 					 & = & \argmax_{\Phi} \sum_{i=1}^n\int_{\mathcal{X}}{\log\left(p_{\phi}(y_i)\right) h_i(x|\phi^k) dx} + \sum_{i=1}^n\int_{\mathcal{X}}{\log\left(h_i(x|\phi)\right) h_i(x|\phi^k) dx} \\
           & = & \argmax_{\Phi} \sum_{i=1}^n{\log\left(p_{\phi}(y_i)\right)} + \sum_{i=1}^n\int_{\mathcal{X}}{\log\left(\frac{h_i(x|\phi)}{h_i(x|\phi^k)}\right) h_i(x|\phi^k) dx}\\
					& & \qquad \qquad \qquad \qquad + \sum_{i=1}^n\int_{\mathcal{X}}{\log\left(h_i(x|\phi^k)\right) h_i(x|\phi^k) dx}.
\end{eqnarray*}
The final line is justified by the fact that $h_i(x|\phi)$ is a density, therefore it integrates to 1. The additional term does not depend on $\phi$ and, hence, can be omitted. We now have the following iterative procedure:
\begin{equation}
\phi^{k+1} = \argmax_{\Phi} \sum_{i=1}^n{\log\left(p_{\phi}(y_i)\right)} + \sum_{i=1}^n\int_{\mathcal{X}}{\log\left(\frac{h_i(x|\phi)}{h_i(x|\phi^k)}\right) h_i(x|\phi^k) dx}.
\label{eqn:EMProximal}
\end{equation}
As stated in \citep{Tseng}, the previous iteration has the form of a proximal point maximization of the log-likelihood, i.e. a perturbation of the log-likelihood by a (modified) Kullback distance-like function defined on the conditional densities of the labels. Tseng proposed to generalize the Kullback distance-like term into other types of divergences. Tseng's recurrence is now defined by:
\begin{equation}
\phi^{k+1} = \argsup_{\phi} J(\phi) - D_{\psi}(\phi,\phi^k),
\label{eqn:TsengAlgo}
\end{equation}
where $J$ is the log-likelihood function and $D_{\psi}$ is a distance-like function defined on the conditional probabilities of the classes provided the observations and is given by:
\begin{equation}
D_{\psi}(\phi,\phi^k) = \sum_{i=1}^n\int_{\mathcal{X}}{\psi\left(\frac{h_i(x|\phi)}{h_i(x|\phi^k)}\right)h_i(x|\phi^k)dx},
\label{eqn:DivergenceClassesNtNorm}
\end{equation}
for a real positive convex function $\psi$ such that $\psi(1)=\psi'(1)=0$. $D_{\psi}(\phi_1,\phi_2)$ is positive and equals zero if $\phi_1=\phi_2$. Moreover, $D_{\psi}(\phi_1,\phi_2)=0$ if and only if $\forall i, h_i(x|\phi_1) = h_i(x|\phi_2)$ $dx-$almost everywhere. Clearly, (\ref{eqn:TsengAlgo}) and (\ref{eqn:EMProximal}) are equivalent for $\psi(t)=-\log(t)+t-1$.\\

%--------------------------------------------------------------------------------------
% =======================================================================================
%--------------------------------------------------------------------------------------
\subsection{Generalization of Tseng's algorithm}\label{subsec:OurAlgo}
We use the relation between maximizing the log-likelihood and minimizing the Kullback-Liebler divergence to generalize the previous algorithm. We therefore replace the log-likelihood function by a $\varphi-$divergence $D_{\varphi}$ (in the sense of \citep{Csiszar}) between the true density of the data $p_{\phi_T}$ and the model $p_{\phi}$. Since the value of the divergence depends on the true density which is unknown, an estimator of the divergence needs to be considered. We may use any estimator among (\ref{eqn:BeranEstimator}), (\ref{eqn:BasuLindsayDiv}), (\ref{eqn:DivergenceDef}) or (\ref{eqn:EmpiricalNewDualForm}). Our new algorithm is defined by the following recurrence:
\begin{equation}
\phi^{k+1} = \arginf_{\phi} \hat{D}_{\varphi}(p_{\phi},p_{\phi_T}) + \frac{1}{n}D_{\psi}(\phi,\phi^k)
\label{eqn:DivergenceAlgoPreVersion}
\end{equation}
where $D_{\psi}(\phi,\phi^k)$ is defined by (\ref{eqn:DivergenceClassesNtNorm}). When $\varphi(t) = -\log(t)+t-1$, it is easy to see that we get recurrence (\ref{eqn:TsengAlgo}). Take for example the case of the approximation (\ref{eqn:DivergenceDef}). Since $\varphi'(t) = \frac{-1}{t} + 1$, we have $\int{\varphi'\left(\frac{p_{\phi}}{p_{\alpha}}\right)p_{\phi}dx} = 0$. Hence,
\[\hat{D}_{\varphi}(p_{\phi},p_{\phi_T}) = \sup_{\alpha} \frac{1}{n}\sum_{i=1}^n{\log(p_{\alpha}(y_i))} - \frac{1}{n}\sum_{i=1}^n{\log(p_{\phi}(y_i))}.\]
Using the fact that the first term in $\hat{D}_{\varphi}(p_{\phi},p_{\phi_T})$ does not depend on $\phi$, so it does not count in the $\arginf$ defining $\phi^{k+1}$, we may rewrite (\ref{eqn:DivergenceAlgoPreVersion}) as:
\begin{eqnarray*}
\phi^{k+1} & = & \arginf_{\phi}\left\{\sup_{\alpha} \frac{1}{n}\sum_{i=1}^n{\log(p_{\alpha}(y_i))} - \frac{1}{n}\sum_{i=1}^n{\log(p_{\phi}(y_i))} + \frac{1}{n}
D_{\psi}(\phi,\phi^k)\right\} \\
           & = & \arginf_{\phi}\left\{-\frac{1}{n}\sum_{i=1}^n{\log(p_{\phi}(y_i))} +\frac{1}{n}D_{\psi}(\phi,\phi^k)\right\} \\
           & = & \argsup_{\phi}\left\{\frac{1}{n}\sum_{i=1}^n{\log(p_{\phi}(y_i))} - \frac{1}{n}D_{\psi}(\phi,\phi^k)\right\} \\
           & = & \argsup_{\phi} J(\phi) - D_{\psi}(\phi,\phi^k).
\end{eqnarray*}
For notational simplicity, from now on, we redefine $D_{\psi}$ with a normalization by $n$, i.e. 
\begin{equation}
D_{\psi}(\phi,\phi^k) = \frac{1}{n} \sum_{i=1}^n\int_{\mathcal{X}}{\psi\left(\frac{h_i(x|\phi)}{h_i(x|\phi^k)}\right)h_i(x|\phi^k)dx}.
\label{eqn:DivergenceClasses}
\end{equation}
Hence, our set of algorithms is redefined by:
\begin{equation}
\phi^{k+1} = \arginf_{\phi} \hat{D}_{\varphi}(p_{\phi},p_{\phi_T}) + D_{\psi}(\phi,\phi^k).
\label{eqn:DivergenceAlgo}
\end{equation}
We will see later that this iteration forces the estimated divergence to decrease and that under suitable conditions, it converges to a (local) minimum of $\hat{D}_{\varphi}(p_{\phi},p_{\phi_T})$. It results that, algorithm (\ref{eqn:DivergenceAlgo}) is a way to calculate the minimum $\varphi-$divergence estimator defined by (\ref{eqn:BeranEstimator}), (\ref{eqn:BasuLindsayDiv}), (\ref{eqn:MDphiDEClassique}) or (\ref{eqn:NewMDphiDE}).\\
Before proceeding to study the convergence properties of such algorithm, we will propose another algorithm for the case of mixture models. In the EM algorithm, the estimation of the parameters of a mixture model is done mainly by two steps, see paragraph \ref{subsec:EMMixtures}. The first step estimates the proportions of the classes whereas the second step estimates the parameters defining the classes. Our idea is based on a directional optimization of the objective function in (\ref{eqn:DivergenceAlgo}). Convergence properties of the two-step algorithm will also be studied, but the proofs are more technical.
%%%%%%%%%%%%%%%%%%%%%%%%%%%%%%%%%%%%%%%%%%%%%%%%%%%%%%%%%%%%%%%%%%%%%%%%%%%%
%
%--------------------------------------------------------------------------------------
% =======================================================================================
%--------------------------------------------------------------------------------------
%
%%%%%%%%%%%%%%%%%%%%%%%%%%%%%%%%%%%%%%%%%%%%%%%%%%%%%%%%%%%%%%%%%%%%%%%%%%%%
\section{Two-step Algorithm for mixtures}\label{sec:TwoStepAlgo}
Let $p_{\phi}$ be a mixture model with $s$ components:
\begin{equation}
p_{\phi}(y) = \sum_{i=1}^s{\lambda_i f_i(y|\theta_i)}.
\label{eqn:MixModelDef}
\end{equation}
Here, $\phi = (\lambda,\theta)$ with $\lambda = (\lambda_1,\cdots,\lambda_s)\in [0,1]^s$ such that $\sum_j{\lambda_j}=1$, and $\theta = (\theta_1,\cdots,\theta_s)\in\Theta\subset\mathbb{R}^{d-s}$ such that $\Phi \subset [0,1]^s\times\Theta$. In the EM algorithm, the corresponding optimization to (\ref{eqn:DivergenceAlgo}) can be solved by calculating an estimate of the $\lambda$'s as the proportions of classes, and then proceed to optimize on the $\theta$'s (see for example \citep{Titterington}). This simplifies the optimization in terms of complexity (optimization in lower spaces) and clarity (separate proportions from classes parameters). We want to build an algorithm with the same property and divide the optimization problem into two parts. One which estimates the proportions $\lambda$ and another which estimates the parameters defining the form of each component $\theta$. We propose the following algorithm:
\begin{eqnarray}
\lambda^{k+1} & = & \arginf_{\lambda\in[0,1]^s, s.t. (\lambda,\theta^k)\in\Phi} \hat{D}_{\varphi}(p_{\lambda,\theta^k},p_{\phi_T}) + D_{\psi}((\lambda,\theta^k),\phi^k); \label{eqn:DivergenceAlgoSimp1} \\
\theta^{k+1} & = & \arginf_{\theta\in\Theta, s.t. (\lambda^{k+1},\theta)\in\Phi} \hat{D}_{\varphi}(p_{\lambda^{k+1},\theta},p_{\phi_T}) + D_{\psi}((\lambda^{k+1},\theta),\phi^k).
\label{eqn:DivergenceAlgoSimp2}
\end{eqnarray}
This algorithm corresponds to a directional optimization for recurrence (\ref{eqn:DivergenceAlgo}) by considering simply the unit vectors as directions. We can therefore prove analogously that the estimated divergence between the model and the true density decreases as we proceed with the recurrence.\\

We end the first part of this chapter by three remarks:
\begin{itemize}
\item[$\bullet$] Function $\psi$ defining the distance-like proximal term $D_{\psi}$ needs not to be convex as in \cite{Tseng}. As we will see in the convergence proofs, the only properties needed are: $\psi$ is a non negative function defined on $\mathbb{R}_+$ verifying $\psi(t)=0$ iff $t=1$, and $\psi'(t)=0$ iff $t=1$.
\item[$\bullet$] The simplified version is not restricted to mixture models. Indeed, any parametric model, whose vector of parameters can be separated into two independent parts, can be estimated using the simplified version.
\item[$\bullet$] As we will see in the proofs, results on the simplified version (\ref{eqn:DivergenceAlgoSimp1}, \ref{eqn:DivergenceAlgoSimp2}) can be extended to a further simplified one. In other words, one may even consider an algorithm which attack a lower level of optimization. We may optimize on each class of the mixture model instead of the whole set of parameters. Since the analytic separation is not evident, one should expect some loss of quality as a cost of a less optimization time.
\end{itemize}
The remaining of the chapter is devoted entirely to the study of the convergence of the sequences generated by either of the two sets of algorithms (\ref{eqn:DivergenceAlgo}) and (\ref{eqn:DivergenceAlgoSimp1}, \ref{eqn:DivergenceAlgoSimp2}) presented above. A key feature which will be needed in the proofs is the regularity of the objective function $\hat{D}_{\varphi}(p_{\phi},p_{\phi_T})$. Regularity of all divergence estimators mentioned at the begining of this chapter can be checked using Lebesgue theorems except for the dual formula (\ref{eqn:DivergenceDef}). Indeed, continuity and differentiability are not simple since the dual formula is defined through a supremum. The following section is devoted to the study of the regularity of a function written as the supremum of a bivariable function.

%%%%%%%%%%%%%%%%%%%%%%%%%%%%%%%%%%%%%%%%%%%%%%%%%%%%%%%%%%%%%%%%%%%%%%%%%%%%
%
%--------------------------------------------------------------------------------------
% =======================================================================================
%--------------------------------------------------------------------------------------
%
%%%%%%%%%%%%%%%%%%%%%%%%%%%%%%%%%%%%%%%%%%%%%%%%%%%%%%%%%%%%%%%%%%%%%%%%%%%%

\section{Analytical properties of the dual formula of \texorpdfstring{$\varphi-$}{phi-}divergences}\label{sec:AnalyticalDiscuss}

The dual formula defining the estimator of the divergence between the true density and the model defined by (\ref{eqn:DivergenceDef}) seems quite complicated. This is basically because of a functional integral and a supremum over it. Continuity and differentiation of the integral is resolved by Lebesgue theorems. We only need that the integrand as well as its partial derivatives to be uniformly bounded with respect to the parameter. However, continuity or differentiability of the supremum is more subtle. Indeed, even if the optimized function is $\mathcal{C}^{\infty}$, it does not imply the continuity of its supremum. Take for example function $f(x,u)=-e^{xu}$. We have:
\[\sup_{x} f(x,u) = \left\{\begin{array}{ccc} -1 & \text{if} & u=0; \\
                                              0 & \text{if} & u\neq 0.
\end{array}\right.\]
On the basis of the theory presented in \citep{Rockafellar} about parametric optimization, we present two ways for studying continuity and differentiability of $\hat{D}_{\varphi}(p_{\phi},p_{\phi_T})$ defined through (\ref{eqn:DivergenceDef}). The first one is the most important because it is easier and demands less mathematical notations. In the first approach, we provide sufficient conditions in order to prove continuity and differentiability almost everywhere of the dual estimator of the divergence. This approach will be used in the study of the convergence of our proximal-point algorithm, see Section \ref{sec:Examples}. The second approach is presented for the sake of completness of the study. We give sufficient conditions which permit to prove the differentiability \emph{everywhere}.\\ 
We recall first the definition of a subgradient of a real valued function $f$.
\begin{definition}[Definition 8.3 in \cite{Rockafellar}]
Consider a function $f:\mathbb{R}^d\rightarrow\bar{\mathbb{R}}$ and a point $\phi^*$ with $f(\phi^*)$ finite. For a vector $v$ in $\mathbb{R}^d$, one says that:
\begin{itemize}
\item[(a)] $v$ is a regular subgradient of $f$ at $\phi^*$, written $v\in\hat{\partial} f(\phi^*)$, if:
\[f(\alpha) \geq f(\phi^*) + <v,\alpha-\phi^*> + o\left(|\alpha-\phi^*|\right);\]
\item[(b)] $v$ is a (general) subgradient of $f$ at $\phi^*$, written $v\in\partial f(\phi^*)$, if there are sequences $\alpha^{n}\rightarrow\phi^*$ with $f(\alpha^n)\rightarrow f(\phi^*)$, and $v^n\in\hat{\partial}f(\alpha^n)$ with $v^n\rightarrow v$.
\end{itemize}
\end{definition}
%%%%%%%%%%%%%%%%%%%%%%%%%%%%%%%%%%%%%%%%%%%%%%%
\subsection{A result of differentiability almost everywhere : Lower-\texorpdfstring{$\mathcal{C}^1$}{TEXT} functions}\label{para:LowerC1}
\begin{definition} [\citep{Rockafellar} Chap 10.] A function $D:\Phi\rightarrow\mathbb{R}$, where $\Phi$ is an open set in $\mathbb{R}^d$, is said to be lower-$\mathcal{C}^1$ on $\Phi$, if on some neighborhood $V$ of each $\phi$ there is a representation
\[D(\phi) = \sup_{\alpha\in T} f(\alpha,\phi)\]
in which the functions $\alpha\mapsto f(\alpha,\phi)$ are of class $\mathcal{C}^1$ on $V$ and the set $T$ is a compact set such that $f(\alpha,\phi)$ and $\nabla_{\phi}f(\alpha,\phi)$ depend continuously not just on $\phi\in \Phi$ but jointly on $(\alpha,\phi)\in T\times V$.
\end{definition}
\noindent In our case, the supremum form is globally defined. Moreover, $T = \Phi$. In case $\Phi$ is bounded, it suffices then to take $T=cl(\Phi)$ the closure of $\Phi$ since $\alpha\mapsto f(\alpha,\phi)$ is continuous. The condition on $T$ to be compact is essential here, and can not be compromised, so that it is necessary to reduce in a way or in another the optimization on $\alpha$ into a compact or at least a bounded set. For example, one may prove that the values of $\alpha\mapsto f(\alpha,\phi)$ near infinity are lower than some value inside $\Phi$ independently of $\phi$.\\
\begin{theorem}[Theorem 10.31 in \citep{Rockafellar}] \label{theo:LowerC1}
Any lower-$\mathcal{C}^1$ function $D$ on an open set $\Phi\subset\mathbb{R}^d$ is both (strictly\footnote{A strictly continuous function $f$ is a local Lipschitz continuous function, i.e. for each $x_0\in\text{int}{\Phi}$, the following limit exists and is finite \[\limsup_{x,x'\rightarrow x_0} \frac{|f(x')-f(x)|}{x'-x}\]}) continuous and continuously differentiable where it is differentiable. Moreover, if $\Delta$ consists of the points where $D$ is differentiable, then $\Phi\setminus\Delta$ is negligible\footnote{A set is called negligible if for every $\varepsilon>0$, there is a family of boxes $\{B_k\}_k$ with $d-$dimensional volumes $\varepsilon_k$ such that $A\subset\cup_{k}{B_k}$ and $\sum_{k}{\varepsilon_k}<\varepsilon$.}.
\end{theorem}
The stated result can be ensured by simple hypotheses on the model $p_{\phi}$ and the function $\varphi$. Unfortunately, since the estimated divergence $\hat{D}_{\varphi}(p_{\phi},p_{\phi_T})$ will not be everywhere differentiable, we can no longer talk about the stationarity of $\hat{D}_{\varphi}(p_{\phi},p_{\phi_T})$ at a limit point of the sequence $\phi^k$ generated for example by (\ref{eqn:DivergenceAlgo}). We therefore, use the notion of subgradients. Indeed, when a function $g$ is not differentiable, a necessary condition for $x_0$ to be a local minimum of $g$ is that $0\in\partial g(x_0)$ and it becomes sufficient whenever $g$ is proper convex\footnote{See \citep{Rockafellar} theorem 10.1.}. Moreover, as $g$ becomes differentiable at $x_0$, then $\nabla g(x_0)\in\partial g(x_0)$ with equality if and only if $g$ is $\mathcal{C}^1$. In other words, proving that $0\in\partial \hat{D}_{\varphi}(p_{\hat{\phi}},p_{\phi_T})$ means that $\hat{\phi}$ is a sort of a \emph{generalized stationary point} of $\phi\mapsto\hat{D}_{\varphi}(p_{\phi},p_{\phi_T})$.\\
We will be studying later on in paragraphs (\ref{Example:CauchyZeroLoc}) and (\ref{Example:DivergenceMixture}) examples where we verify with more details the previous conditions and see the resulting consequences on the sequence $(\phi^k)_k$.
%%%%%%%%%%%%%%%%%%%%%%%%%%%%%%%%%%%%%%%%%%%%%%%%%%%%
\subsection{A result of everywhere differentiability: Level-bounded functions}\label{para:LevelBoundFun}
\begin{definition}[\citep{Rockafellar} Chap 1.]
A function $f:\mathbb{R}^d\times\mathbb{R}^d\rightarrow\bar{\mathbb{R}}$ with values $f(\alpha,\phi)$ is (upper) level-bounded in $\alpha$ locally uniformly in $\phi$ if for each $\phi_0$ and $a\in\mathbb{R}$ there is a neighborhood $V$ for $\phi_0$ such that the set $\{(\alpha,\phi)|\phi\in V, f(\alpha,\phi)\geq a\}$ is bounded in $\mathbb{R}^d\times\mathbb{R}^d$ for every $a\in\mathbb{R}$.
\end{definition}
For a fixed $\phi$, the level-boundedness property corresponds to having $f(\alpha,\phi)\rightarrow -\infty$ as $\|\alpha\|\rightarrow\infty$. In order to state the main result for this case, let $\phi_0$ be a point at which we need to study continuity and differentiability of $\phi\mapsto\sup_{\alpha}{f(\alpha,\phi)}$. A first result gives sufficient conditions under which the supremum function is continuous. We state it as follows:
\begin{theorem}[\citep{Rockafellar} Theorem 1.17] \label{theo:LevelBoundedCont}
Let $f:\mathbb{R}^n\times\mathbb{R}^m\rightarrow\bar{\mathbb{R}}$ be an upper semicontinuous function. Suppose that $f(\alpha,\phi)$ is level-bounded in $\alpha$ locally uniformly in $\phi$. For function $\phi\mapsto \sup_{\alpha}f(\alpha,\phi)$ to be continuous at $\phi_0$, a sufficient condition is the existence of $\alpha_0\in\argmax_{\alpha}f(\alpha,\phi_0)$ such that $\phi\mapsto f(\alpha_0,\phi)$ is continuous at $\phi_0$.
\end{theorem}
Since in general, we do not know exactly where the supremum will be, one proves the continuity of $\phi\mapsto f(\alpha,\phi)$ for every $\alpha$.\\
A Further result about continuity and differentiability of the supremum function can also be stated. Define, at first, the sets $Y(\phi_0)$ and $Y_{\infty}(\phi_0)$ as follows:
\begin{eqnarray*}
Y(\phi_0) & = & \bigcup_{\alpha\in\argsup_{\beta} f(\beta,\phi_0)} M(\alpha,\phi_0),\quad \text{for } M(\alpha,\phi_0) = \{a|(0,a)\in\partial f(\alpha,\phi_0)\} \\
Y_{\infty}(\phi_0) & = & \bigcup_{\alpha\in\argsup_{\beta} f(\beta,\phi_0)} M_{\infty}(\alpha,\phi_0),\quad \text{for } M_{\infty}(\alpha,\phi_0) = \{a|(0,a)\in\partial^{\infty} f(\alpha,\phi_0)\}
\end{eqnarray*}
where $\partial^{\infty} f$ is the horizon subgradient, see Definition 8.3 (c) in \cite{Rockafellar}. We avoided to mention the definition here in order to keep the text clearer. Furthermore, in the whole chapter, the horizon subgradient will always be equal to the set $\{0\}$.\\
%%%%%%%%%%%%%%theorem
\begin{theorem}[Corollary 10.14 in \citep{Rockafellar}] 
\label{theo:levelbounded}
For a proper upper semicontinuous function $f:\mathbb{R}^d\times\mathbb{R}^d\rightarrow \bar{\mathbb{R}}$ such that $f(\alpha,\phi)$ is level-bounded in $\alpha$ locally uniformly in $\phi$, and for $\phi_0\in$ dom $\sup_{\alpha}f(\alpha,\phi)$:
\begin{itemize}
\item[(a)] If $Y_{\infty}(\phi_0) = \{0\}$, then $\phi\mapsto\sup_{\alpha}f(\alpha,\phi)$ is strictly continuous at $\phi_0$;
\item[(b)] if $Y(\phi_0) = \{a\}$ too, then\footnote{In the statement of the corollary in \citep{Rockafellar}, the supremum function becomes strictly differentiable, but to avoid extra vocabularies, we replaced it with an equivalent property.} $\phi\mapsto\sup_{\alpha}f(\alpha,\phi)$ is $\mathcal{C}^1$ at $\phi_0$ with $\nabla \sup_{\alpha}f(\alpha,\phi) = a$.
\end{itemize}
\end{theorem}
In our examples, $f$ will be a continuous function and even $\mathcal{C}^1(\Phi\times\Phi)$. This implies that $\partial^{\infty} f(\alpha, \phi) = \{0\}$ and $\partial f(\alpha,\phi) = \{\nabla f(\alpha,\phi)\}$, see Exercise 8.8 in \cite{Rockafellar}. Hence, $Y_{\infty}(\phi_0) = \{0\}$ whatever $\phi_0$ in $\Phi$. Moreover $M(\alpha,\phi_0) = \{\nabla_{\phi} f(\alpha,\phi_0)\}$ so that $Y(\phi_0)=\bigcup\{\nabla_{\phi} f(\alpha,\phi_0)\}$ and the union is on the set of suprema of $\alpha\mapsto f(\alpha,\phi_0)$. If $f(\alpha,\phi)$ is level-bounded in $\alpha$ locally uniformly in $\phi$, then the supremum function becomes strictly continuous. Moreover, if the function $f$ has the same gradient with respect to $\phi$ for all the suprema of $\alpha\mapsto f(\alpha,\phi)$, then $\sup_{\alpha} f(\alpha,\phi)$ becomes continuously differentiable. This is for example the case when function $\alpha\mapsto f(\alpha,\phi)$ has a unique global supremum for a fixed $\phi$, which is for example the case of a strictly concave function (with respect to $\alpha$ for a fixed $\phi$).\\

\begin{example}
Let $(p_{\phi})_{\phi}$ be an exponential model defined by:
\[p_{\phi}(x) = \exp\left[T(x).\phi - C(\phi)\right].\]
Let $\varphi(t)=t\log(t)-t+1$. The dual representation of the divergence (formula (\ref{eqn:DivergenceDef})) is then given by:
\[\hat{D}_{\varphi}(p_{\phi},p_{\phi*}) = \sup_{\alpha} \left\{\mathbb{E}_{p_{\phi}}[T(X)].(\phi-\alpha)+C(\alpha)- C(\phi) - \frac{1}{n}\sum_{i=1}^n{e^{T(y_i).(\phi-\alpha)+C(\alpha)-C(\phi)}}\right\}+1.\]
In order to prove that the optimized function is level-bounded in $\alpha$ locally uniformly in $\phi$, we take a bounded open neighborhood around $\phi$, and we prove that the optimized function tends to $-\infty$ as $\|\alpha\|$ tends to infinity. For example, for the Gaussian case with the mean $\mu$ as the parameter of interest, we have:
\[\hat{D}_{\varphi}(p_{\mu},p_{\mu*}) = \sup_{\beta\in\mathbb{R}} \frac{1}{2}\mu^2-\mu\beta + \frac{1}{2}\beta^2 - \frac{1}{n}\sum_{i=1}^n{e^{y_i(\mu-\beta)+\frac{1}{2}\beta^2-\frac{1}{2}\mu^2}}+1.\]
It is clear that $e^{\beta^2}$ is the dominant term at $\infty$, and by putting $\mu$ in a bounded interval, the limit of the optimized function when $\beta$ tends to infinity is easily calculated and equals $-\infty$.\\
For the exponential case $p_a(x)=ae^{-ax}$ with $a>0$ the parameter of interest, we have:
\[\hat{D}_{\varphi}(p_{a},p_{a^*}) = \sup_{b>0} \frac{b}{a}-\log(b)+\log(a) - \frac{1}{n}\sum_{i=1}^n{e^{-y_i(a-b)-\log(b)+\log(a)}}.\]
Here again, the dominant term is $e^{y_i b}$ at infinity. Since an observation of an exponential law is positive, the limit when $b$ tends to infinity is hence easily calculated and equals $-\infty$.	\\
For part (a) of Theorem \ref{theo:levelbounded} to be verified, we still need to prove that $Y_{\infty}(\phi_0)=\{0\}$. However, this is verified because the optimized function, here, is continuously differentiable, so that it is strictly continuous. This implies that $Y_{\infty}(\phi_0)=\{0\}$. Hence, $\hat{D}_{\varphi}(p_{\phi},p_{\phi^T})$ is strictly continuous.\\
To prove that it is also $\mathcal{C}^1$, we need to prove that $Y(\phi_0)$ contains but one element. First of all, since the optimized function is differentiable, $Y(\phi_0) = \bigcup\{\nabla_{\phi}f(\alpha,\phi_0)\}$. The union is over the set $\{\argmax_{\alpha} f(\alpha,\phi_0)\}$. Let's calculate the jacobian matrix with respect to $\alpha$ and see when it might be definite negative, and hence function $\alpha\mapsto f(\alpha,\phi)$ would be strictly concave and would only have one maximum whenever it exists\footnote{Notice that for both the Gaussian and exponential example, the derivative passes by zero as will be explained later on. Therefore, strict concavity would imply the existence of one maximum.}. If for any $\phi$, it is definite negative, this should be sufficient to prove the claim.\\
\begin{eqnarray*}
\nabla_{\alpha} f(\alpha,\phi) & = & -\mathbb{E}_{p_{\phi}}[T(X)] + \nabla C(\alpha) - \frac{1}{n}\sum_{i=1}^n{(\nabla C(\alpha) - T(y_i))e^{T(y_i).(\phi-\alpha)+C(\alpha)-C(\phi)}} \\
J_f(\alpha,\phi) & = & J_C(\alpha) - \frac{1}{n}\sum_{i=1}^n{\left(J_C(\alpha) + (\nabla C(\alpha) - T(y_i)).(\nabla C(\alpha) - T(y_i))^t\right)e^{T(y_i).(\phi-\alpha)+C(\alpha)-C(\phi)}}
\end{eqnarray*}
For the Gaussian example, we have:
\begin{eqnarray*}
\frac{\partial f}{\partial \beta}(\beta,\mu) & =  & -\mu +\beta - \frac{1}{n}\sum_{i=1}^n{(\beta-y_i)e^{y_i(\mu-\beta)+\frac{1}{2}\beta^2-\frac{1}{2}\mu^2}} \\
\frac{\partial^2 f}{\partial \beta^2}(\beta,\mu) & = & 1-\frac{1}{n}\sum_{i=1}^n{\left(1+(\beta-y_i)^2\right)e^{y_i(\mu-\beta)+\frac{1}{2}\beta^2-\frac{1}{2}\mu^2}}
\end{eqnarray*}
The gradient has at least one zero since it is continuous and has $+\infty$ limit at $-\infty$ and $-\infty$ limit at $+\infty$. The second derivative with respect to $\beta$ is unfortunately not necessarily negative so that function $\beta\mapsto f(\beta,\mu)$ is not concave. An analytical study of function $f$ seems very difficult. Let's simulate a 10-sample of the standard Gaussian probability law, and let a mathematical tool such as Mathematica do the painting. Table (\ref{tab:GaussDerivEx}) shows the dataset used.\\
\begin{table}[h]
\centering
\begin{tabular}{|c|c|c|c|c|c|c|c|c|c|c|}
$y_i$ & 0.644 & -3.144 & -1.029 & -0.367 & 0.353 & -0.704 & 1.148 & 0.674 & 0.148 & -0.721\\
\hline
\end{tabular}
\caption{A 10-sample Gaussian dataset.}
\label{tab:GaussDerivEx}
\end{table}

\noindent We make a 3D plot for $f(\beta,\mu)$ in two parts. The first part for $\mu>0$ and the second is for $\mu<0$ to get a clear view about what happens when $\mu$ changes, see figure (\ref{fig:GaussExample}). Although the second derivative with respect to $\beta$ is not necessarily negative, function $\beta\mapsto f(\beta,\mu)$ has only one maximum point. We conclude that function $\sup_{\beta} f(\beta,\mu)$ is continuously derivable \emph{for the dataset provided in table (\ref{tab:GaussDerivEx}).}\\
\begin{figure}[ht]
\centering
\includegraphics[scale = 0.5]{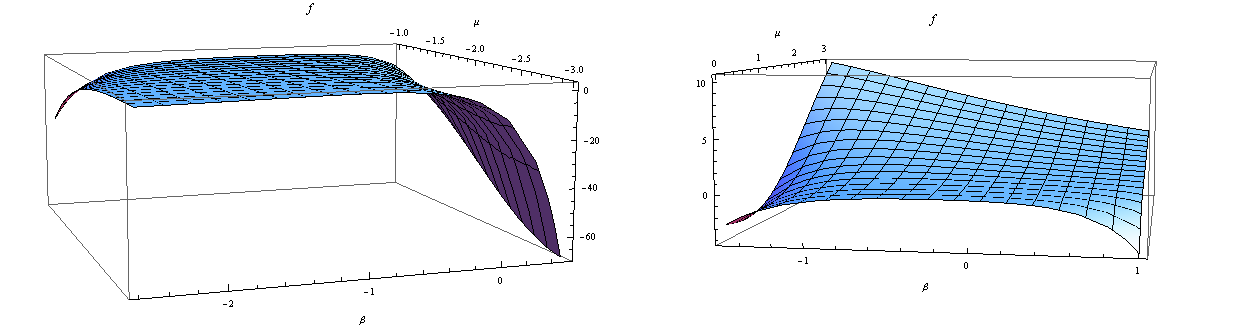}
\caption{A 3D plot of function $f$ in the Gaussian example shows that there is only one maximum for each value of $\mu$.}
\label{fig:GaussExample}
\end{figure}

\noindent For the exponential case, we have:
\begin{eqnarray*}
\frac{\partial f}{\partial b} & = & \frac{1}{a}-\frac{1}{b} - \frac{1}{n}\sum_{i=1}^n{\left(by_i-\frac{1}{b}\right)e^{-y_i(a-b)-\log(b)+\log(a)}}\\
\frac{\partial^2 f}{\partial b^2} & = & \frac{1}{b^2} - \frac{1}{n}\sum_{i=1}^n{\left(y_i+\frac{1}{b^2}+\left(by_i-\frac{1}{b}\right)^2\right)e^{-y_i(a-b)-\log(b)+\log(a)}}
\end{eqnarray*}
We similarly have the same previous problem. The second derivative is not necessarily negative, however, when simulating a dataset and plotting function $f$, we may conclude that it has only one maximum whenever $a$ is fixed. Hence function $\sup_b f(a,b)$ becomes continuously derivable.\\
These two examples, although very simple, shows the difficulty in proving differentiability. Our proof, earlier, depends heavily on \emph{graphical tools}, which may still appear not totally convincing. Another undesirable aspect is that even if we admit the \emph{graphical} indications about the existence of a unique maximum, the final conclusion stays related to the dataset we are working on. An analytical proof for previous examples remain an open problem for further work.
\end{example}
%%%%%%%%%%%%%%%%%%%%%%%%%%%%%%%%%%%%%%%%%%%%%%%%%%%%
\begin{remark}[An implicit function point of view]
One may try in case $f(\alpha,\phi)$ is concave to calculate the gradient with respect to $\alpha$. At the supremum, whenever it exists, we have $\nabla_{\alpha}f(\alpha,\phi)=0$. The solution in $\alpha$ is \emph{a priori} a function of $\phi$, say $\alpha(\phi)$. The implicit function theorem provides a way to prove the existence of such function and gives sufficient conditions for continuity and differentiability. Notice that a \emph{global} version of the implicit function theorem is needed here in order to define $\alpha(\phi)$ on the whole $\Phi$ and not locally. As soon as we have such a function, we may write $f(\alpha(\phi),\phi) = \hat{D}_{\varphi}(p_{\phi},p_{\phi_T})$, and the divergence becomes differentiable using a simple chain rule. The problem with this solution is that the conditions to ensure a global function $\alpha(\phi)$ are not simple, see for example \citep{Cristea} and the references therein.
\end{remark}

%%%%%%%%%%%%%%%%%%%%%%%%%%%%%%%%%%%%%%%%%%%%%%%%%%%%%%%%%%%%%%%%%%%%%%%%%%%%
%
%--------------------------------------------------------------------------------------
% =======================================================================================
%--------------------------------------------------------------------------------------
%
%%%%%%%%%%%%%%%%%%%%%%%%%%%%%%%%%%%%%%%%%%%%%%%%%%%%%%%%%%%%%%%%%%%%%%%%%%%%

\section{Convergence properties}\label{sec:Proofs}
We adapt the ideas given in \citep{Tseng} to develop a suitable proof for our proximal algorithm. We present some propositions which show how according to some possible situations one may prove convergence of the algorithms defined by recurrences (\ref{eqn:DivergenceAlgo}) and (\ref{eqn:DivergenceAlgoSimp1}, \ref{eqn:DivergenceAlgoSimp2}). Let $\phi^0=(\lambda^0,\theta^0)$ be a given initialization for the parameters, and define the following set
\begin{equation}
\Phi^0 = \{\phi\in\Phi: \hat{D}_{\varphi}(p_{\phi},p_{\phi_T})\leq \hat{D}_{\varphi}(\phi^0,\phi_T)\}
\label{eqn:SetPhis0}
\end{equation} 
where $\hat{D}_{\varphi}(\phi,\phi_T)$ is any estimator of the $\varphi-$divergence among (\ref{eqn:BeranEstimator}), (\ref{eqn:BasuLindsayDiv}), (\ref{eqn:DivergenceDef}) or (\ref{eqn:EmpiricalNewDualForm}). We suppose that $\Phi^0$ is a subset of $int(\Phi)$. The idea of defining such a set in this context is inherited from the paper of \citep{Wu} which provided the first \emph{correct proof} of convergence for the EM algorithm. Before going any further, we recall the following definition of a (generalized) stationary point.\\ 
\begin{definition}
Let $f:\mathbb{R}^d\rightarrow\mathbb{R}$ be a real valued function. If $f$ is differentiable at a point $\phi^*$ such that $\nabla f(\phi^*)=0$, we then say that $\phi^*$ is a stationary point of $f$. If $f$ is not differentiable at $\phi^*$ but the subgradient of $f$at $\phi^*$, say $\partial f(\phi^*)$, exists such that $0\in\partial f(\phi^*)$, then $\phi^*$ is called a generalized stationary point of $f$.
\end{definition}
Using continuity and differentiability assumptions on both $\hat{D}_{\varphi}$ and $D_{\psi}$, we will prove the following results:
\begin{itemize}
\item[$\bullet$] For both algorithms (\ref{eqn:DivergenceAlgo}) and (\ref{eqn:DivergenceAlgoSimp1}, \ref{eqn:DivergenceAlgoSimp2}), if $\Phi^0$ is closed and $\{\phi^{k+1}-\phi^k\}\rightarrow 0$, then any limit point of $(\phi^k)_k$ is a stationary point of the objective function $\hat{D}_{\varphi}(\phi,\phi_T)$;
\item For algorithm (\ref{eqn:DivergenceAlgo}), if we only have $\Phi^0$ is compact, then any limit point is a stationary point of the objective function;
\item For algorithm (\ref{eqn:DivergenceAlgoSimp1}, \ref{eqn:DivergenceAlgoSimp2}), if $\Phi^0$ is compact and $\|\lambda^{k+1}-\lambda^k\|\rightarrow 0$, then any limit point is a stationary point of the objective function;
\item[$\bullet$] For both algorithms (\ref{eqn:DivergenceAlgo}) and (\ref{eqn:DivergenceAlgoSimp1}, \ref{eqn:DivergenceAlgoSimp2}), if $\Phi^0$ is compact and $D_{\psi}(\phi,\phi')>0$ iff $\phi\neq\phi'$, then $\{\phi^{k+1}-\phi^k\}\rightarrow 0$ and any limit point is a stationary point of the objective function.
\item In case the objective function $\phi\mapsto\hat{D}_{\varphi}(p_{\phi}|p_{\phi_T})$ is not continuously differentiable, we prove previous points for algorithm (\ref{eqn:DivergenceAlgo}) with generalized stationary point instead of stationary point.
\end{itemize}
We will be using the following assumptions which will be checked in several examples later on.
\begin{itemize}
\item[A0.] Functions $\phi\mapsto\hat{D}_{\varphi}(p_{\phi}|p_{\phi_T}), D_{\psi}$ are lower semicontinuous;
\item[A1.] Functions $\phi\mapsto\hat{D}_{\varphi}(p_{\phi}|p_{\phi_T}), D_{\psi}$ and $\nabla_1 D_{\psi}$ are defined and continuous on, respectively, $\Phi, \Phi\times\Phi$ and $\Phi\times\Phi$;
\item[AC.] $\nabla \hat{D}_{\varphi}(p_{\phi}|p_{\phi_T})$ is defined and continuous on $\Phi$;
\item[A2.] $\Phi^0$ is a compact subset of int($\Phi$);
\item[A3.] $D_{\psi}(\phi,\bar{\phi})>0$ for all $\bar{\phi}\neq \phi \in \Phi$.
\end{itemize}
Recall also the assumptions on functions $h_i$ defining $D_{\psi}$. We suppose that $h_i(x|\phi)>0, dx-a.e.$, and $\psi(t)=0$ iff $t=1$. Besides $\psi'(t)=0$ iff $t = 1$.\\
Concerning assumptions A1 and AC, we have previously discussed the analytical properties of $\hat{D}_{\varphi}(p_{\phi}|p_{\phi_T})$ after Section \ref{sec:TwoStepAlgo} and in Section \ref{sec:AnalyticalDiscuss}. In what concerns $D_{\psi}$, continuity and differentiability can be obtained merely by fulfilling Lebesgue theorems conditions. For example, if $h_i(x,\phi)$ is continuous and bounded uniformly away from 0 independently of $\phi$, then continuity is guaranteed as soon as $\psi$ is continuous. If we also suppose that $\nabla_{\phi} h_i(x,\phi)$ exists, is continuous and is uniformly bounded independently of $\phi$, then as soon as $\psi$ is continuously differentiable, $D_{\psi}$ becomes continuously differentiable. For assumption A2, there is no universal method. Still, in all the examples that will be discussed later, we use the fact that the inverse image of a closed set by a continuous function is closed. Boundedness is usually ensured using a \emph{suitable} choice of $\phi^0$. Finally, assumption A3 is checked using Lemma 2 proved in \citep{Tseng} which we restate here.
\begin{lemma}[\cite{Tseng} Lemma 2] \label{Lem:Tseng}
Suppose $\psi$ to be a continuous non negative function such that $\psi(t)=0$ iff $t=1$. For any $\phi$ and $\phi'$ in $\Phi$, if $h_i(x|\phi)\neq h_i(x|\phi')$ for some $i\in\{1,\cdots,n\}$ and some $x\in int(X)$ at which both $h_i(.|\phi)$ and $h_i(.|\phi')$ are continuous, then $D_{\psi}(\phi,\phi')>0$.
\end{lemma}
\noindent In section (\ref{sec:Examples}), we present three different examples; a two-component Gaussian mixture, a two-component Weibull mixture and a Cauchy model. We will see that the Cauchy example verifies assumption A3. However, the Gaussian mixture does not seem to verify it. Indeed, the same fact stays true for any mixture of the exponential family.\\
We start by providing some general facts about the sequence $(\phi^k)_k$ and its existence. We also prove convergence of the sequence $(\hat{D}_{\varphi}(p_{\phi^k}|p_{\phi_T}))_k$.
\begin{remark}
All results concerning algorithm (\ref{eqn:DivergenceAlgo}) are proved even when assumption AC is not fulfilled. We give proofs using the subgradient of the estimated $\varphi-$divergence. In the case of the two-step algorithm (\ref{eqn:DivergenceAlgoSimp1}, \ref{eqn:DivergenceAlgoSimp2}), it was not possible and thus remains an open problem. The difficulty resides in manipulating the \emph{partial} subgradients with respect to $\lambda$ and $\theta$ which cannot be handled in a similar way to the partial derivatives.
\end{remark}
\begin{remark}
Convergence properties are proved without using the special form of the estimated $\varphi-$divergence. Thus, our theoretical approach applies to any optimization problem whose objective is to minimize a function $\phi\mapsto D(\phi)$. For example, our approach can be applied on density power divergences (\cite{BasuMPD}), Bregman divergences, S-divergences (\cite{GoshSDivergence}), etc.
\end{remark}
%%%%%%%%%%%%%%%%%%%%%%%%%%%%%%%%%%%%%%%%%%%%%%%%%%%%%%%%
\begin{proposition}\label{prop:DecreaseDphi}
We assume that recurrences (\ref{eqn:DivergenceAlgo}) and (\ref{eqn:DivergenceAlgoSimp1}, \ref{eqn:DivergenceAlgoSimp2}) are well defined in $\Phi$. For both algorithms, the sequence $(\phi^{k})_k$ verifies the following properties:
\begin{itemize}
\item[(a)] $\hat{D}_{\varphi}(p_{\phi^{k+1}}|p_{\phi_T})\leq \hat{D}_{\varphi}(p_{\phi^k}|p_{\phi_T})$;
\item[(b)] $\forall k, \phi^k \in \Phi^0$;
\item[(c)] Suppose that assumptions A0 and A2 are fulfilled, then the sequence $(\phi^k)_k$ is defined and bounded. Moreover, the sequence $\left(\hat{D}_{\varphi}(p_{\phi^k}|p_{\phi_T})\right)_k$ converges.
\end{itemize}
\end{proposition}
The proof of this proposition is differed to Appendix \ref{AppendProximal:Prop1}.
%%%%%%%%%%%%%%%%%%%%%%%%%%%%%%%%%%%%%%%%%%%%%%%%%%%%%%%%%%%%%%%%%%%%%
%============================================
%%%%%%%%%%%%%%%%%%%%%%%%%%%%%%%%%%%%%%%%%%%%%%%%%%%%%%%%%%%%%%%%%%%%%
\noindent The interest of Proposition \ref{prop:DecreaseDphi} is that the objective function is ensured, under mild assumptions, to decrease alongside the sequence $(\phi^k)_k$. This permits to build a stop criterion for the algorithm since in general there is no guarantee that the whole sequence $(\phi^k)_k$ converges. It may also continue to fluctuate in a neighborhood of an optimum. The following result provides a first characterization about the properties of the limit of the sequence $(\phi^k)_k$ as (generalized) a stationary point of the estimated $\varphi-$divergence. The proof is differed to Appendix \ref{AppendProximal:Prop2}.
\begin{proposition}
\label{prop:StationaryPhiDiff}
Suppose that A1 is verified, and assume that $\Phi^0$ is closed and $\{\phi^{k+1}-\phi^k\}\rightarrow 0$.
\begin{itemize}
\item[(a)] For both algorithms (\ref{eqn:DivergenceAlgo}) and (\ref{eqn:DivergenceAlgoSimp1},\ref{eqn:DivergenceAlgoSimp2}), if AC is verified, then the limit of every convergent subsequence is a stationary point of $\hat{D}_{\varphi}(.|p_{\phi_T})$;
\item[(b)] For the first algorithm (\ref{eqn:DivergenceAlgo}), if $\hat{D}_{\varphi}(.|p_{\phi_T})$ is not differentiable,  then the limit of every convergent subsequence is a "generalized" stationary point of $\hat{D}_{\varphi}(.|p_{\phi_T})$, i.e. zero belongs to the subgradient of $\hat{D}_{\varphi}(.|p_{\phi_T})$ calculated at the limit point;
\end{itemize}
\end{proposition}
\noindent Assumption $\{\phi^{k+1}-\phi^k\}\rightarrow 0$ used in Proposition \ref{prop:StationaryPhiDiff} is not easy to be checked unless one has a close formula of $\phi^k$. This is the case of the EM algorithm applied on a Gaussian mixture, see \cite{Tseng} Section 5. In general, we prove that $\{\phi^{k+1}-\phi^k\}\rightarrow 0$ by imposing an identifiability assumption over the proximal term, see \cite{ChretienHero} Lemma 5 or \cite{Tseng} Lemma 1. The following proposition is a mere adaptation of such results to the context of $\varphi-$divergences and the two-step algorithm. The proof is differed to Appendix \ref{AppendProximal:Prop3}. We will present later a result which does not need such assumption.
%%%%%%%%%%%%%%%%%%%%%%%%%%%%%%%%%%%%%%%%%%%%%%%%%%%%%%%%%%%%%%%%%%%%%%%%%%%%%%%
%%%%%%%%%%%%%%
%%%%%%%%%%%%%%%%%%%%%%%%%%%%%%%%%%%%%%%%%%%%%%%%%%%%%%%%%%%%%%%%%%%%%%%%%%%%%%%

\begin{proposition}
\label{prop:PhiDiffConverge}
For both algorithms defined by (\ref{eqn:DivergenceAlgo}) and (\ref{eqn:DivergenceAlgoSimp1},\ref{eqn:DivergenceAlgoSimp2}), assume A1, A2 and A3 verified, then $\{\phi^{k+1}-\phi^k\}\rightarrow 0$. Thus, by proposition 2 (according to whether AC is verified or not) implies that any limit point of the sequence $\phi^k$ is a (generalized)\footnote{The case where AC is not verified is only proved for the first algorithm (\ref{eqn:DivergenceAlgo})} stationary point of $\hat{D}_{\varphi}(.|p_{\phi_T})$.
\end{proposition}
We can go further in exploring the properties of the sequence $(\phi^k)_k$, but we need to impose more assumptions. The following corollary provides a convergence result of the \emph{whole} sequence and not only some subsequence. The convergence is also towards a local minimum as soon as the estimated divergence is locally strictly convex. The proof of the following result is differed to Appendix \ref{AppendProximal:Cor}.
%%%%%%%%%%%%%%%%%%%%%%%%%%%%%%%%%%%%%%%%%%%%%%%%%%%%%%%%%%%%%%%%
\begin{corollary} 
\label{Cor:TotalConverg}
Under the assumptions of Proposition 3, the set of accumulation points of $(\phi^k)_k$ is a connected compact set. Moreover, if $\hat{D}(p_{\phi},p_{\phi_T})$ is strictly convex in a neighborhood of a limit point\footnote{This assumption can be replaced by local strict convexity since \emph{a priori}, we have no idea where might find a limit point of the sequence $(\phi^k)_k$.} of the sequence $(\phi^k)_k$, then the whole sequence $(\phi^k)_k$ converges to a local minimum of $\hat{D}(p_{\phi},p_{\phi_T})$.
\end{corollary}
%%%%%%%%%%%%%%%%%%%%%%%%%%%%%%%%%%%%%%%%%%%%%%%%%%%%%%%%%%%%%%%%
\noindent Proposition \ref{prop:PhiDiffConverge} although provides a general solution to prove that $\{\phi^{k+1}-\phi^k\}\rightarrow 0$, the identifiability assumption over the proximal term is hard to be fulfilled. It is not verified in the most simple mixtures such as a two component Gaussian mixture, see Section (\ref{Example:GaussMix}).\\
This was the reason behind our next result. We prove that we do not need to assume identifiability of the proximal term in order to prove that any convergent subsequence of $(\phi^k)_k$ is a (generalized) stationary point of the estimated $\varphi-$divergence.\\
A similar idea was employed in \citep{ChretienHeroProxGener} who studied a proximal algorithm for the log-likelihood function with a relaxation parameter\footnote{A sequence of decreasing positive numbers multiplied by the proximal term.}. Their work however requires that the log-likelihood has $-\infty$ limit as $\|\phi\|\rightarrow\infty$ which is not verified on several mixture models (e.g. the Gaussian mixture model). Our result treat the problem from another approach based on the introduction of the set $\Phi^0$.\\
%%%%%%%%%%%%%%%%%%%%%%%%%%%%%%%%%%%%%%%%%%%%%%%%%%%%%%%%%%%%
\begin{proposition}\label{prop:NewRes}
Assume A1, AC and A2 verified. For the algorithm defined by (\ref{eqn:DivergenceAlgo}), any convergent subsequence converges to a stationary point of the objective function $\phi\rightarrow\hat{D}(p_{\phi},p_{\phi_T})$. If AC is dropped, then 0  belongs to the subgradient of $\phi\mapsto\hat{D}(p_{\phi},p_{\phi_T})$ at the limit point.
\end{proposition}
\noindent The proof is differed to Appendix \ref{AppendProximal:Prop4} We could not perform the same idea on the two-step algorithm (\ref{eqn:DivergenceAlgoSimp1},\ref{eqn:DivergenceAlgoSimp2}) without assuming that the difference between two consecutive terms of either the sequence of weights $(\lambda^k)_k$ or the sequence of form parameters $(\theta^k)_k$ converges to zero. The proof is differed to Appendix \ref{AppendProximal:Prop5}.
%%%%%%%%%%%%%%%%%%%%%%%%%%%%%%%%%%%%%%%%%%%%%%%%%%%%%%%%%%%%%%%%%%%%%%%%%%%%%%%%%%%%
\begin{proposition} \label{prop:NewResTwoStep}
Assume A1 and A2 verified. For the algorithm defined by (\ref{eqn:DivergenceAlgoSimp1},\ref{eqn:DivergenceAlgoSimp2}). If $\|\theta^{k+1} - \theta^k\|\rightarrow 0$, then any convergent subsequence $(\phi^{N(k)})_k$ converges to a stationary point of the objective function $\phi\rightarrow\hat{D}(p_{\phi},p_{\phi_T})$.
\end{proposition}
\begin{remark}
The previous proposition demands a condition on the distance between two consecutive members of the sequence $(\theta^k)_k$ which is \emph{a priori} weaker than the same condition on the whole sequence $\phi^k=(\lambda^k,\theta^k)$. Still, as the regularization term $D_{\psi}$ does not verify the identifiability condition A3, it stays an open problem for a further work. It is interesting to notice that condition $\|\theta^{k+1}-\theta^k\|\rightarrow 0$ can be replaced by $\|\lambda^{k+1}-\lambda^k\|\rightarrow 0$, but we then need to change the order of steps (\ref{eqn:DivergenceAlgoSimp1}) and (\ref{eqn:DivergenceAlgoSimp2}). A condition over the proportions seems to be \emph{simpler}.
\end{remark}
\begin{remark}
We can define an algorithm which converges to a global infimum of the estimated $\varphi-$divergence (see paragraph \ref{subsec:GlobaInf}). The idea is very simple. We need to multiply the proximal term by a sequence of positive numbers which decreases to zero, say $1/k$. The justification of such variant can be deduced from Theorem 3.2.4 in \citep{ChretienHeroProxGener}. The problem with this approach is that it depends heavily on the fact that the supremum on each step of the algorithm is calculated exactly. This does not happen in general unless function $\hat{D}_{\varphi}(p_{\phi},p_{\phi^T}) + \beta_k D_{\psi}(\phi,\phi^k)$ is strictly convex. Although in our approach, we use similar assumption to prove the consecutive decreasing of $\hat{D}_{\varphi}(p_{\phi},p_{\phi^T})$, we can replace the infimum calculus in (\ref{eqn:DivergenceAlgo}) by two things. We require at each step that we find a local infimum of $\hat{D}_{\varphi}(p_{\phi},p_{\phi^T}) + D_{\psi}(\phi,\phi^k)$ whose evaluation with $\phi\mapsto\hat{D}_{\varphi}(p_{\phi},p_{\phi^T})$ is less than the previous term of the sequence $\phi^k$. If we can no longer find any local maxima verifying the claim, the procedure stops with $\phi^{k+1}=\phi^k$. This ensures the availability of all proofs presented in this paper with no further changes.
\end{remark}

%%%%%%%%%%%%%%%%%%%%%%%%%%%%%%%%%%%%%%%%%%%%%%%%%%%%%%%%%%%%%%%%%%%%%%%%%%%%
%
%--------------------------------------------------------------------------------------
% =======================================================================================
%--------------------------------------------------------------------------------------
%
%%%%%%%%%%%%%%%%%%%%%%%%%%%%%%%%%%%%%%%%%%%%%%%%%%%%%%%%%%%%%%%%%%%%%%%%%%%%

\section{Case Studies and Variants of the algorithm}\label{sec:CaseStudies}
\subsection{An algorithm with theoretically global infimum attainment}\label{subsec:GlobaInf}
We present a variant of algorithm (\ref{eqn:DivergenceAlgo}) which ensures \emph{theoretically} convergence to a global infimum of the objective function $\hat{D}_{\varphi}(p_{\phi},p_{\phi^T})$ as long as there exists a convergent subsequence. The idea is the same as Theorem 3.2.4 in \citep{ChretienHeroProxGener}. Define $\phi^{k+1}$ by:
\[\phi^{k+1} = \arginf_{\phi} \hat{D}_{\varphi}(p_{\phi},p_{\phi^T}) + \beta_k D_{\psi}(\phi,\phi^k).\]
The proof of convergence is very simple and does not depend on the differentiability of any of the two functions $\hat{D}_{\varphi}$ or $D_{\psi}$. We only assume A1 and A2 to be verified. Let $(\phi^{N(k)})_k$ be a convergent subsequence. Let $\phi^{\infty}$ be its limit. This is guaranteed by the compactness of $\Phi^0$ and the fact that the whole sequence $(\phi^k)_k$ resides in $\Phi^0$ (see Proposition \ref{prop:DecreaseDphi}-b). Suppose also that the sequence $(\beta_k)_k$ converges to 0 as $k$ goes to infinity.\\ 
Let $\phi$ by a vector of $\Phi$ which has a value of $\hat{D}_{\varphi}(p_{\phi},p_{\phi^T})$ strictly inferior to the value of the same function at $\phi^{\infty}$, i.e.
\begin{equation}
\hat{D}_{\varphi}(p_{\phi},p_{\phi^T}) < \hat{D}_{\varphi}(p_{\phi^{\infty}},p_{\phi^T}).
\label{eqn:GlobalInfAlgo1}
\end{equation}
By definition of $\phi^{N(k)}$, we have:
\[\hat{D}_{\varphi}(p_{\phi^{N(k)}},p_{\phi^T}) + \beta_{N(k)-1} D_{\psi}(\phi^{N(k)},\phi^{N(k)-1}) \leq \hat{D}_{\varphi}(p_{\phi},p_{\phi^T}) + \beta_{N(k)-1} D_{\psi}(\phi,\phi^{N(k)}),\]
which holds for every $\phi$ in the whole set $\Phi$. Using the non negativity of the term $\beta_{N(k)-1}D_{\psi}(\phi^{N(k)},\phi^{N(k)-1})$, one can write:
\begin{equation}
\hat{D}_{\varphi}(p_{\phi^{N(k)}},p_{\phi^T}) \leq \hat{D}_{\varphi}(p_{\phi},p_{\phi^T}) + \beta_{N(k)} D_{\psi}(\phi,\phi^{N(k)}).
\label{eqn:GlobalInfAlgo2}
\end{equation}
As we pass to the limit on $k$, recall firstly that $(\beta_k)_k$ converges to 0, so that any subsequence $(\beta_{N(k)})_k$ also converges to 0. Secondly, the continuity assumption on $D_{\psi}$ implies that, since $\phi^{N(k)}\rightarrow\phi^{\infty}$, $D_{\psi}(\phi,\phi^{N(k)})$ converges to $D_{\psi}(\phi,\phi^{\infty})$. By compactness of $\Phi^0$ and Proposition \ref{prop:DecreaseDphi}-b, we have $\phi^{\infty}\in\Phi^0$. The continuity again of $D_{\psi}$ will imply that the quantity $D_{\psi}(\phi,\phi^{\infty})$ is finite. Finally, inequality (\ref{eqn:GlobalInfAlgo2}) now implies that:
\[\hat{D}_{\varphi}(p_{\phi^{\infty}},p_{\phi^T}) \leq \hat{D}_{\varphi}(p_{\phi},p_{\phi^T})\]
which contradicts with the choice of $\phi$ verifying (\ref{eqn:GlobalInfAlgo1}). Hence, $\phi^{\infty}$ is a global infimum on $\Phi$.\\
The problem with this approach is that it depends heavily on the fact that the supremum on each step of the algorithm is calculated exactly. This does not happen in general unless function $\hat{D}_{\varphi}(p_{\phi},p_{\phi^T}) + \beta_k D_{\psi}(\phi,\phi^k)$ is convex or that we dispose of an algorithm which can solve perfectly non convex optimization problems\footnote{In this case, there is no meaning in applying an iterative proximal algorithm. We would have used the optimization algorithm directly on the objective function $\hat{D}_{\varphi}(p_{\phi},p_{\phi^T})$}. Although in our approach, we use similar assumption to prove the consecutive decreasing of $\hat{D}_{\varphi}(p_{\phi},p_{\phi^T})$, we can replace the infimum calculus in (\ref{eqn:DivergenceAlgo}) by two things. We demand that at each step that we find a local infimum of $\hat{D}_{\varphi}(p_{\phi},p_{\phi^T}) + D_{\psi}(\phi,\phi^k)$ whose evaluation with $\phi\mapsto\hat{D}_{\varphi}(p_{\phi},p_{\phi^T})$ is less than the previous term of the sequence $\phi^k$. If we can no longer find any local infima verifying the claim, the procedure stops with $\phi^{k+1}=\phi^k$. This ensures the availability of all the proofs presented in this paper with no change.
%%%%%%%%%%%%%%%%%%%%%%%%%%%%%%%%%%%%%%%%%%%%%%%%%%%%%%%%%ù
\subsection{The EM algorithm in the context of mixture models}\label{subsec:EMMixtures}
In the case of mixture models (\ref{eqn:MixModelDef}), the EM recurrence can be rewritten in two parts; a part where the maximization is on the proportions, and a part on the parameters describing the form of classes. We code the $s$ classes directly by their indices $\{1,\cdots,s\}$. Starting from (\ref{eqn:EMAlgo}), one may insert directly the constraint on the $\lambda$'s into the optimized function as follows:
\begin{multline*}(\lambda_1^{k+1},\cdots,\lambda_{s-1}^{k+1},\theta_1^{k+1},\cdots,\theta_s^{k+1}) =  \argsup_{\lambda_1\geq 0,\cdots,\lambda_{s-1}\geq 0,(\theta_1,\cdots\theta_s)\in\Theta}  \sum_{i=1}^n\sum_{j=1}^{s-1}{\log\left(\lambda_j p(y_i|\theta_j)\right)h_i(j|\phi^k)} \\ + \sum_{i=1}^n{\log\left(\left(1-\sum_{j=1}^{s-1}{\lambda_j}\right) p(y_i|\theta_j)\right)h_i(j|\phi^k)}
\end{multline*}
where
\[ h_i(x_j|\phi^k)=\frac{\lambda_j^k f_j(y_i|\theta_j^k)}{\sum_j{\lambda_j^k f_j(y_i|\theta_j^k)}}.\]
Now, the property of the logarithmic function in transforming the product into a sum is the cornerstone in the simplification of the optimization.
\begin{multline*}
(\lambda_1^{k+1},\cdots,\lambda_{s-1}^{k+1},\theta_1^{k+1},\cdots,\theta_s^{k+1}) =  \argsup_{\lambda_1\geq 0,\cdots,\lambda_{s-1}\geq 0,(\theta_1,\cdots\theta_s)\in\Theta}  \sum_{i=1}^n\sum_{j=1}^{s-1}{\log(\lambda_j)h_i(j|\phi^k)}  +\\ \sum_{i=1}^n{\log\left(1-\sum_{j=1}^{s-1}{\lambda_j}\right)h_i(s|\phi^k)} + \sum_{i=1}^n\sum_{j=1}^{s-1}{\log\left(\lambda_j p(y_i|\theta_j)\right)h_i(j|\phi^k)} + \sum_{i=1}^n{\log\left(p(y_i|\theta_j)\right)h_i(s|\phi^k)}.
\end{multline*}
The optimized function is, thus, written as the sum of two independent functions in the sense that the first one contains only proportions parameters whereas the second contains other parameters. Since the parameters (proportions and the others) are independent from each others\footnote{There is no common constraint between them.}, one can rewrite the previous optimization problem as the sum of two optimization problems:
\begin{eqnarray*}
(\lambda_1^{k+1},\cdots,\lambda_{s-1}^{k+1}) & = & \argsup_{\lambda_1\geq 0,\cdots,\lambda_{s-1}\geq 0} 
\sum_{i=1}^n\sum_{j=1}^{s-1}{\log(\lambda_j)h_i(j|\phi^k)}  + \sum_{i=1}^n{\log\left(1-\sum_{j=1}^{s-1}{\lambda_j}\right)h_i(s|\phi^k)};  \\
(\theta_1^{k+1},\cdots,\theta_s^{k+1}) & = & \argsup_{(\theta_1,\cdots\theta_s)\in\Theta} \sum_{i=1}^n\sum_{j=1}^{s-1}{\log\left(\lambda_j p(y_i|\theta_j)\right)h_i(j|\phi^k)} + \sum_{i=1}^n{\log\left(p(y_i|\theta_j)\right)h_i(s|\phi^k)}.
\end{eqnarray*}
The first step can be explicitly calculated. Solving the gradient equation gives:
\begin{eqnarray}
\frac{1}{\lambda_j} \sum_{i=1}^n{h_i(j|\phi^k)} - \frac{1}{\sum_{l=1}^{s-1}{\lambda_l}}\sum_{i=1}^n{h_i(s|\phi^k)} & = & 0 \quad  \;\; \forall j\in\{1,\cdots,s-1\};\nonumber\\
\frac{\sum_{i=1}^n{h_i(s|\phi^k)}}{\sum_{i=1}^n{h_i(j|\phi^k)}}\lambda_j & = & 1 - \sum_{l=1}^{s-1} \quad \forall j\in\{1,\cdots,s-1\}; \label{eqn:EMProp1} \\
\lambda_1+\cdots+\left(1+\frac{\sum_{i=1}^n{h_i(s|\phi^k)}}{\sum_{i=1}^n{h_i(j|\phi^k)}}\right)\lambda_j+\lambda_{s-1} & = & 1 \quad  \forall j\in\{1,\cdots,s-1\}; \label{eqn:EMProp2}
\end{eqnarray}
Equation (\ref{eqn:EMProp1}) implies that:
\begin{equation}
\forall j\in\{1,\cdots,s-1\},\quad \frac{1}{\sum_{i=1}^n{h_i(j|\phi^k)}}\lambda_j = \frac{1}{\sum_{i=1}^n{h_i(j|\phi^k)}}\lambda_1.
\label{eqn:EMProp3}
\end{equation}
Rewriting equation (\ref{eqn:EMProp2}) for $j=1$ and using previous identities gives:
\begin{eqnarray*}
1 & = & \left(1+\frac{\sum_{i=1}^n{h_i(s|\phi^k)}}{\sum_{i=1}^n{h_i(1|\phi^k)}}\right)\lambda_1 + \sum_{l=2}^{s-1}{\frac{\sum_{i=1}^n{h_i(l|\phi^k)}}{\sum_{i=1}^n{h_i(1|\phi^k)}}\lambda_1}; \\
1 & = & \lambda_1 \left[1+\frac{n-\sum_{i=1}^n{h_i(1,\phi^k)}}{\sum_{i=1}^n{h_i(1,\phi^k)}}\right]; \\
\lambda_1 & = & \frac{1}{n}\sum_{i=1}^n{h_i(1,\phi^k)}.
\end{eqnarray*}
In the second line, we used the fact that $h_i(s|\phi^k) = 1-\sum_{l=1}^{s-1}{h_i(l|\phi^k)}$. Finally, we use (\ref{eqn:EMProp3}) to deduce that:
\[\lambda_1  = \frac{1}{n}\sum_{i=1}^n{h_i(1,\phi^k)} \forall j\in\{1,\cdots,s-1\}.\]
Now, the EM recurrence is given by:
\begin{eqnarray*}
\lambda_j^{k+1} & = & \frac{1}{n}\sum_{i=1}^n{h_i(x_j|\phi^k)} \quad j\in\{1,\cdots,s-1\};\\
\theta^{k+1} & = & \argsup_{\theta}\sum_{i=1}^n\sum_{j=1}^s{\log\left(f_j(y_i|\theta_j)\right)h_i(x_j|\phi^k)}.
\end{eqnarray*}
This was the idea behind our algorithm defined by (\ref{eqn:DivergenceAlgoSimp1},\ref{eqn:DivergenceAlgoSimp2}). Furthermore, the second part of the optimization can be simplified more than that. We may write an optimization corresponding to each class since the optimized function is a sum of terms each of which depends only of the parameter vector $\theta_j$ defining the corresponding class. The EM algorithm can be rewritten as follows:
\begin{eqnarray*}
\lambda_j^{k+1} & = & \frac{1}{n}\sum_{i=1}^n{h_i(x_j|\phi^k)} j\in\{1,\cdots,s-1\}\\
\theta_j^{k+1} & = & \argsup_{\theta_j}\sum_{i=1}^n{\log\left(f_j(y_i|\theta_j)\right)h_i(x_j|\phi^k)} \quad j\in\{1,\cdots,s\}
\end{eqnarray*}
This suggests to go further in algorithm (\ref{eqn:DivergenceAlgoSimp1},\ref{eqn:DivergenceAlgoSimp2}) and use the same idea of directional optimization on the second part (\ref{eqn:DivergenceAlgoSimp2}). The convergence results can be extended to this variant without any additional assumptions.

%%%%%%%%%%%%%%%%%%%%%%%%%%%%%%%%%%%%%%%%%%%%%%%%%%%%%%%%%%%%%%%%%%%%%%%%%%%%
%
%--------------------------------------------------------------------------------------
% =======================================================================================
%--------------------------------------------------------------------------------------
%
%%%%%%%%%%%%%%%%%%%%%%%%%%%%%%%%%%%%%%%%%%%%%%%%%%%%%%%%%%%%%%%%%%%%%%%%%%%%

\section{Theoretical study of convergence on some mixtures with application to the EM algorithm}\label{sec:Examples}
In this section, we present three examples where we check assumptions A0-A3 and AC and study the convergence properties of the sequence $\phi^k$. We only consider for the estimated divergence the two dual formula presented in paragraphs \ref{subsec:ClassicalDualFormula} and \ref{subsec:KernelSolution}. Other $\varphi-$divergence-based estimators; Beran's and Basu-Linsday's approaches, can be treated in a similar way to the kernel-based MD$\varphi$DE.
%%%%%%%%%%%%%%%%%%%%%%%%%%%%%%%%%%%%%%%%%%%%%%%%%%%%%%%%%%%%%%%%%%
\subsection{two-component Gaussian mixture}\label{Example:GaussMix}
We suppose that the model $(p_{\phi})_{\phi\in\Phi}$ is a mixture of two Gaussian densities, and suppose that we are only interested in estimating the means $\mu=(\mu_1,\mu_2)\in\mathbb{R}^2$ and the proportions $\lambda = (\lambda_1,\lambda_2)\in[\eta,1-\eta]^2$. The use of $\eta$ is to avoid cancellation of any of the two components and to keep the hypothesis about the conditional densities $h_i$ true, i.e. $h_i(x|\phi)>0$ for $x=1,2$. We also suppose to simplify the calculus that the components variances are reduced ($\sigma_i = 1$). The model takes the form:
\begin{equation}
p_{\lambda,\mu}(x) = \frac{\lambda}{\sqrt{2\pi}} e^{-\frac{1}{2}(x-\mu_1)^2} + \frac{1-\lambda}{\sqrt{2\pi}} e^{-\frac{1}{2}(x-\mu_2)^2},
\label{eqn:GaussMixModel}
\end{equation}
where $\Phi = [\eta,1-\eta]^s\times\mathbb{R}^s$. Here $\phi=(\lambda,\mu_1,\mu_2)$. The distance-like function $D_{\psi}$ is defined by:
\[D_{\psi}(\phi,\phi^k) = \sum_{i=1}^n{\psi\left(\frac{h_i(1|\phi)}{h_i(1|\phi^k)}\right)h_i(1|\phi^k)} + \sum_{i=1}^n{\psi\left(\frac{h_i(2|\phi)}{h_i(2|\phi^k)}\right)h_i(2|\phi^k)},\]
where:
\[h_i(1|\phi) = \frac{\lambda e^{-\frac{1}{2}(y_i-\mu_1)^2}}{\lambda e^{-\frac{1}{2}(y_i-\mu_1)^2} + (1-\lambda) e^{-\frac{1}{2}(y_i-\mu_2)^2}}, \quad h_i(2|\phi) = 1-h_i(1|\phi).\]
It is clear that functions $h_i$ are of class $\mathcal{C}^1$ on (int($\Phi$)), and as a consequence, $D_{\psi}$ is also of class $\mathcal{C}^1$ on (int($\Phi$)).\\
\textbf{If we are using the dual estimator of the $\varphi-$divergence given by (\ref{eqn:DivergenceDef})}, then assumption A0 can be verified using the maximum theorem of \cite{Berge}. There is still a great difficulty in studying the properties (closedness or compactness) of the set $\Phi^0$. Moreover, all convergence properties of the sequence $\phi^k$ require the continuity of the estimated $\varphi-$divergence $\hat{D}_{\varphi}(p_{\phi},p_{\phi^T})$ with respect to $\phi$. In the context of paragraph \ref{sec:AnalyticalDiscuss}, $\hat{D}_{\varphi}(p_{\phi},p_{\phi^T})=\sup_{\alpha\in\Phi}f(\alpha,\phi)$ for the Gaussian mixture cannot be treated directly using any of the two presented approaches. We propose to assume that $\Phi$ is compact, i.e. assume that the means are included in an interval of the form $[\mu_{\min},\mu_{\max}]$. Now $\hat{D}_{\varphi}(p_{\phi},p_{\phi^T})=\sup_{\alpha\in\Phi}f(\alpha,\phi)$ is a lower$-\mathcal{C}^1$ function since $f(\alpha,\phi)$ is of class $\mathcal{C}^1(\Phi)$ using Lebesgue theorems. Thus, using Theorem \ref{theo:LowerC1}, $\hat{D}_{\varphi}(p_{\phi},p_{\phi^T})$ is continuous and differentiable almost everywhere with respect to $\phi$. \\
The compactness assumption of $\Phi$ implies directly the compactness of $\Phi^0$. Indeed
\begin{eqnarray*}
\Phi^0 & = & \left\{\phi\in\Phi, \hat{D}_{\varphi}(p_{\phi},p_{\phi^T})\leq \hat{D}_{\varphi}(p_{\phi^0},p_{\phi^T})\right\} \\
			 & = & \hat{D}_{\varphi}(p_{\phi},p_{\phi^T})^{-1}\left((-\infty, \hat{D}_{\varphi}(p_{\phi^0},p_{\phi^T})]\right).
\end{eqnarray*}
$\Phi^0$ is then the inverse image by a continuous function of a closed set, so it is closed in $\Phi$. Hence, it is compact.
\begin{conclusion}
\label{conc:conclusion1}
Using Propositions \ref{prop:NewRes} and \ref{prop:DecreaseDphi}, if $\Phi=[\eta,1-\eta]\times [\mu_{\min},\mu_{\max}]^2$, the sequence $(\hat{D}_{\varphi}(p_{\phi^k},p_{\phi^T}))_k$ defined through formula (\ref{eqn:DivergenceDef}) converges and there exists a subsequence $(\phi^{N(k)})$ which converges to a stationary point of the estimated divergence. Moreover, every limit point of the sequence $(\phi^k)_k$ is a stationary point of the estimated divergence. 
\end{conclusion}
\textbf{If we are using the kernel-based dual estimator given by (\ref{eqn:EmpiricalNewDualForm})} with a Gaussian kernel density estimator, then function $\phi\mapsto \hat{D}_{\varphi}(p_{\phi},p_{\phi^T})$ is continuously differentiable over $\Phi$ even if the means $\mu_1$ and $\mu_2$ are not bounded. For example, take $\varphi = \varphi_{\gamma}$ defined by (\ref{eqn:CressieReadPhi}). There is one condition which relates the window of the kernel, say $w$, with the value of $\gamma$; $\gamma(w^2-1)>-1$. For $\gamma=2$ (the Pearson's $\chi^2$), we need that $w^2>1/2$. For $\gamma=1/2$ (the Hellinger), there is no condition on $w$.\\
Closedness of $\Phi^0$ is proved similarly to the previous case. Boundedness is however must be treated differently since $\Phi$ is not necessarily compact and is supposed to be $\Phi = [\eta,1-\eta]^s\times\mathbb{R}^s$. For simplicity take $\varphi=\varphi_{\gamma}$. The idea is to choose $\phi^0$ an initialization for the proximal algorithm in a way that $\Phi^0$ does not include unbounded values of the means. Continuity of $\phi\mapsto\hat{D}_{\varphi}(p_{\phi},p_{\phi^T})$ permits to calculate the limits when either (or both) of the means tends to infinity. If both means goes to infinity, then $p_{\phi}(x)\rightarrow 0,\forall x$. Thus, for $\gamma\in(0,\infty)\setminus \{1\}$, we have $\hat{D}_{\varphi}(p_{\phi},p_{\phi^T})\rightarrow \frac{1}{\gamma(\gamma-1)}$. For $\gamma<0$, the limit is infinity. If only one of the means tends to $\infty$, then the corresponding component vanishes from the mixture. Thus, if we choose $\phi^0$ such that:
\begin{eqnarray}
\hat{D}_{\varphi}(p_{\phi^0},p_{\phi^T}) & < & \min\left(\frac{1}{\gamma(\gamma-1)}, \inf_{\lambda,\mu}\hat{D}_{\varphi}(p_{(\lambda,\infty,\mu)},p_{\phi^T})\right) \text{ if } \gamma \in(0,\infty)\setminus \{1\};\label{eqn:CondGaussMixNewDual1} \\
\hat{D}_{\varphi}(p_{\phi^0},p_{\phi^T}) & < & \inf_{\lambda,\mu}\hat{D}_{\varphi}(p_{(\lambda,\infty,\mu)},p_{\phi^T})\qquad \text{ if } \gamma <0,
\label{eqn:CondGaussMixNewDual2}
\end{eqnarray}
then the algorithm starts at a point of $\Phi$ whose function value is inferior to the limits of $\hat{D}_{\varphi}(p_{\phi},p_{\phi^T})$ at infinity. By Proposition \ref{prop:DecreaseDphi}, the algorithm will continue to decrease the value of $\hat{D}_{\varphi}(p_{\phi},p_{\phi^T})$ and never goes back to the limits at infinity. Besides, the definition of $\Phi^0$ permits to conclude that if $\phi^0$ is chosen according to condition (\ref{eqn:CondGaussMixNewDual1},\ref{eqn:CondGaussMixNewDual2}), then $\Phi^0$ is bounded. Thus, $\Phi^0$ becomes compact. Unfortunately the value of $\inf_{\lambda,\mu}\hat{D}_{\varphi}(p_{(\lambda,\infty,\mu)},p_{\phi^T})$ can be calculated but numerically. We will see next that in the case of Likelihood function, a similar condition will be imposed for the compactness of $\Phi^0$, and there will be no need for any numerical calculus.
\begin{conclusion}
\label{conc:conclusion2}
Using Propositions \ref{prop:NewRes} and \ref{prop:DecreaseDphi}, under condition (\ref{eqn:CondGaussMixNewDual1}, \ref{eqn:CondGaussMixNewDual2}) the sequence $(\hat{D}_{\varphi}(p_{\phi^k},p_{\phi^T}))_k$ defined through formula (\ref{eqn:EmpiricalNewDualForm}) converges and there exists a subsequence $(\phi^{N(k)})$ which converges to a stationary point of the estimated divergence. Moreover, every limit point of the sequence $(\phi^k)_k$ is a stationary point of the estimated divergence. 
\end{conclusion}
\textbf{In the case of the likelihood $\varphi(t)=-\log(t)+t-1$}, the set $\Phi^0$ can be written as:
\begin{eqnarray*}
\Phi^0 & = & \left\{\phi\in\Phi, J(\phi)\geq J(\phi^0)\right\} \\
			 & = & J^{-1}\left([J(\phi^0),+\infty)\right),
\end{eqnarray*}
where $J$ is the log-likelihood function. Function $J$ is clearly of class $\mathcal{C}^1$ on (int($\Phi$)). We prove that $\Phi^0$ is closed and bounded which is sufficient to conclude its compactness, since the space $[\eta,1-\eta]^s\times\mathbb{R}^s$ provided with the euclidean distance is complete.\\ 
\textbf{Closedness.} The set $\Phi^0$ is the inverse image by a continuous function (the log-likelihood). Therefore it is closed in $[\eta,1-\eta]^s\times\mathbb{R}^s$.\\
\textbf{Boundedness.} By contradiction, suppose that $\Phi^0$ is unbounded, then there exists a sequence $(\phi^l)_l$ which tends to infinity. Since $\lambda^l\in[\eta,1-\eta]$, then either of $\mu_1^l$ or $\mu_2^l$ tends to infinity. Suppose that both $\mu_1^l$ and $\mu_2^l$ tend to infinity, we then have $J(\phi^l) \rightarrow -\infty$. Any finite initialization $\phi^0$ will imply that $J(\phi^0)>-\infty$ so that $\forall \phi\in\Phi^0, J(\phi)\geq J(\phi^0)>-\infty$. Thus, it is impossible for both $\mu_1^l$ and $\mu_2^l$ to go to infinity.\\
Suppose that $\mu_1^l \rightarrow \infty$, and that $\mu_2^l$ converges\footnote{Normally, $\mu_2^l$ is bounded; still, we can extract a subsequence which converges.} to $\mu2$. The limit of the likelihood has the form:
\[L(\lambda, \infty, \phi_2) = \prod_{i=1}^n{\frac{(1-\lambda)}{\sqrt{2\pi}}e^{-\frac{1}{2}(y_i-\mu_2)^2}},\]
which is bounded by its value for $\lambda = 0$ and $\mu_2 = \frac{1}{n}\sum_{i=1}^n{y_i}$. Indeed, since $1-\lambda\leq 1$, we have:
\[L(\lambda, \infty, \phi_2) \leq \prod_{i=1}^n{\frac{1}{\sqrt{2\pi}}e^{-\frac{1}{2}(y_i-\mu_2)^2}}.\]
The right hand side of this inequality is the likelihood of a Gaussian model $\mathcal{N}(\mu_2,0)$, so that it is maximized when $\mu_2=\frac{1}{n}\sum_{i=1}^n{y_i}$. Thus, if $\phi^0$ is chosen in a way that $J(\phi^0)>J\left(0,\infty,\frac{1}{n}\sum_{i=1}^n{y_i}\right)$, the case when $\mu_1$ tends to infinity and $\mu_2$ is bounded would never be allowed. For the other case where $\mu_2\rightarrow\infty$ and $\mu_1$ is bounded, we choose $\phi^0$ in a way that $J(\phi^0)>J\left(1,\frac{1}{n}\sum_{i=1}^n{y_i},\infty\right)$. In conclusion, with a choice of $\phi^0$ such that:
\begin{equation}
J(\phi^0)>\max\left[J\left(0,\infty,\frac{1}{n}\sum_{i=1}^n{y_i}\right),\; J\left(1,\frac{1}{n}\sum_{i=1}^n{y_i},\infty\right)\right]
\label{eqn:TwoGaussMixCond}
\end{equation}
the set $\Phi^0$ is bounded.\\
This condition on $\phi^0$ is very natural and means that we need to begin at a point at least better than the extreme cases where we only have one component in the mixture. This can be easily verified by choosing a random vector $\phi^0$, and calculate the corresponding log-likelihood value. If $J(\phi^0)$ does not verify the previous condition, we draw again another random vector until satisfaction.
\begin{conclusion}
\label{conc:conclusion3}
Using Propositions \ref{prop:NewRes} and \ref{prop:DecreaseDphi}, under condition (\ref{eqn:TwoGaussMixCond}) the sequence $(J(\phi^k))_k$ converges and there exists a subsequence $(\phi^{N(k)})$ which converges to a stationary point of the likelihood function. Moreover, every limit point of the sequence $(\phi^k)_k$ is a stationary point of the likelihood. 
\end{conclusion}
\textbf{Assumption A3 is not fulfilled} (this part applies for all aforementioned situations). As mentioned in the paper of \citep{Tseng}, for the two Gaussian mixture example, by changing $\mu_1$ and $\mu_2$ by the same amount and suitably adjusting $\lambda$, the value of $h_i(x|\phi)$ would be unchanged. We explore this more thoroughly by writing the corresponding equations. Let's suppose, by absurd, that for distinct $\phi$ and $\phi'$ we have $D_{\psi}(\phi|\phi')=0$. By definition of $D_{\psi}$, it is given by a sum of non negative terms, which implies that all terms need to be equal to zero. The following lines are equivalent $\forall i \in \{1,\cdots,n\}$:
\begin{eqnarray*}
h_i(0|\lambda,\mu_1,\mu_2) & = & h_i(0|\lambda',\mu'_1,\mu'_2) \\
\frac{\lambda e^{-\frac{1}{2}(y_i-\mu_1)^2}}{\lambda e^{-\frac{1}{2}(y_i-\mu_1)^2} + (1-\lambda) e^{-\frac{1}{2}(y_i-\mu_2)^2}} & = & \frac{\lambda' e^{-\frac{1}{2}(y_i-\mu'_1)^2}}{\lambda' e^{-\frac{1}{2}(y_i-\mu'_1)^2} + (1-\lambda') e^{-\frac{1}{2}(y_i-\mu'_2)^2}} \\
\log\left(\frac{1-\lambda}{\lambda}\right) - \frac{1}{2}(y_i-\mu_2)^2 + \frac{1}{2}(y_i-\mu_1)^2 & = & \log\left(\frac{1-\lambda'}{\lambda'}\right) - \frac{1}{2}(y_i-\mu'_2)^2 + \frac{1}{2}(y_i-\mu'_1)^2
\end{eqnarray*}
Looking at this set of $n$ equations as an equality of two polynomials on $y$ of degree 1 at $n$ points, we deduce that as we dispose of two distinct observations, say, $y_1$ and $y_2$, the two polynomials need to have the same coefficients. Thus the set of $n$ equations is equivalent to the following two equations:
\begin{equation}
\left\{\begin{array}{ccc}\mu_1-\mu_2 & = & \mu'_1-\mu'_2 \\
		\log\left(\frac{1-\lambda}{\lambda}\right) + \frac{1}{2}\mu_1^2 - \frac{1}{2}\mu_2^2 & = & \log\left(\frac{1-\lambda'}{\lambda'}\right) + \frac{1}{2}{\mu'_1}^2 - \frac{1}{2}{\mu'_2}^2	
		\end{array}\right.
\label{eqn:EqSysGaussMix}
\end{equation}
These two equations with three variables have an infinite number of solutions. Take for example $\mu_1 = 0,\mu_2=1,\lambda=\frac{2}{3},\mu'_1=\frac{1}{2}, \mu'_2=\frac{3}{2},\lambda'=\frac{1}{2}$. \\
\begin{remark} 
The previous conclusion can be extended to any two-component mixture of exponential families having the form:
\[p_{\phi}(y) = \lambda e^{\sum_{i=1}^{m_1}{\theta_{1,i}y^{i}} - F(\theta_1)} + (1-\lambda)e^{\sum_{i=1}^{m_2}{\theta_{2,i}y^{i}} - F(\theta_2)}.\]
One may write the corresponding $n$ equations. The polynomial of $y_i$ has a degree of at most $\max(m_1,m_2)$. Thus, if one disposes of $\max(m_1,m_2)+1$ distinct observations, the two polynomials will have the same set of coefficients. Finally, if $(\theta_1,\theta_2)\in\mathbb{R}^{d-1}$ with $d>\max(m_1,m_2)$, then assumption A3 does not hold.
\end{remark}
This conclusion holds for both algorithms (\ref{eqn:DivergenceAlgo}) or (\ref{eqn:DivergenceAlgoSimp1},\ref{eqn:DivergenceAlgoSimp2}). Unfortunately, we have no an information about the difference between consecutive terms $\|\phi^{k+1}-\phi^k\|$ except for the case of $\psi(t) = \varphi(t)=-\log(t)+t-1$ which corresponds to the classical EM recurrence:
\[\lambda^{k+1} = \frac{1}{n}\sum_{i=1}^n{h_i(0|\phi^k)},\quad \mu_1^{k+1} = \frac{\sum_{i=1}^n{y_ih_i(0|\phi^k)}}{\sum_{i=1}^n{h_i(0|\phi^k)}}\quad \mu_1^{k+1} = \frac{\sum_{i=1}^n{y_ih_i(1|\phi^k)}}{\sum_{i=1}^n{h_i(1|\phi^k)}}.\]
\cite{Tseng} has shown that we can prove directly that $\phi^{k+1}-\phi^k$ converges to 0.

% ----------------------------------------------------------------------------------------------------
%%%%%%%%%%%%%%%%%%%%%%%%%%%%%%%%%%%%%%%%%%%%%%%%%%%%%%%
% ----------------------------------------------------------------------------------------------------
\subsection{Two-component Weibull mixture}\label{subsec:WeibullMixEx}
Let $p_{\phi}$ be a two-component Weibull mixture:
\begin{equation}
p_{\phi}(x) = 2\lambda\alpha_1 (2x)^{\alpha_1-1} e^{-(2x)^{\alpha_1}}+(1-\lambda)\frac{\alpha_2}{2}\left(\frac{x}{2}\right)^{\alpha_2-1} e^{-\left(\frac{x}{2}\right)^{\alpha_2}}
\label{eqn:WeibullMixture}
\end{equation}
We have $\Phi = (0,1)\times\mathbb{R}_+^*\times\mathbb{R}_+^*$. Similarly to the Gaussian example, we will study convergence properties in light of our theoretical approach. We will only be interested in divergences with the class of Cressie-Read functions $\varphi=\varphi_{\gamma}$ given by (\ref{eqn:CressieReadPhi}).\\
The weight functions $h_i$ are given by:
\[h_i(1|\phi) = \frac{2\lambda\alpha_1 (2x)^{\alpha_1-1} e^{-(2x)^{\alpha_1}}}{2\lambda\alpha_1 (2x)^{\alpha_1-1} e^{-(2x)^{\alpha_1}}+(1-\lambda)\frac{\alpha_2}{2}\left(\frac{x}{2}\right)^{\alpha_2-1} e^{-\left(\frac{x}{2}\right)^{\alpha_2}}},\quad h_i(2|\phi)=1-h_i(1|\phi).\]
It is clear the functions $h_i$ are of class $\mathcal{C}^1(\text{int}(\Phi))$ and so does $\phi\mapsto D_{\psi}(\phi,\phi')$ for any $\phi'\in\Phi$.\\
\textbf{If we are using the dual estimator defined by (\ref{eqn:DivergenceDef})}, then continuity can be treated similarly to the case of the Gaussian example. Here, however, the continuity and differentiability of the optimized function $f(\alpha,\phi)$, where $\hat{D}_{\varphi}(p_{\phi},p_{\phi^T}) = \sup_{\alpha}f(\alpha,\phi)$, are more technical. We list the following three results without any proof, because it suffices to study the integral term in the formula. Suppose, without loss of generality, that $\phi_1<\phi_2$ and $\alpha_1<\alpha_2$.
\begin{enumerate}
\item For $\gamma>1$, which includes the Pearson's $\chi^2$ case, the dual representation is \emph{not} well defined since $\sup_{\alpha}f(\alpha,\phi)=\infty$;
\item For $\gamma\in(0,1)$, function $f(\alpha,\phi)$ is continuous. 
\item For $\gamma<0$, function $f(\alpha,\phi)$ is continuous and well defined for $\phi_1<\frac{\gamma-1}{\gamma}\alpha_1$ and $\alpha_2\geq \phi_2$. Otherwise $f(\alpha,\phi)=-\infty$, but the supremum $\sup_{\alpha}f(\alpha,\phi)$ is still well defined.
\end{enumerate}
In both cases 2 and 3, differentiability of function $f(\alpha,\phi)$ holds only on a subset of $\Phi\times\Phi$ which cannot be written as $A\times B$, and thus the theoretical approaches presented in Section \ref{sec:AnalyticalDiscuss} are not suitable. In order to end this part, we emphasize the fact that, similarly to the Gaussian example, even continuity of the estimated divergence $\hat{D}_{\varphi}(p_{\phi},p_{\phi^T})$ with respect to $\phi$ cannot be treated by our theoretical approaches unless we suppose that $\Phi$ is compact. If $\Phi$ is compact, function $f(\alpha,\phi)$ becomes level-bounded and Theorem \ref{theo:LevelBoundedCont} applies and we can deduce that the estimated divergence is continuous. Differentiability is far more subtle if we use Theorem \ref{theo:LowerC1}.\\
Similar conclusion as Conclusion \ref{conc:conclusion1} can be stated here with no changes except for the fact that assumption AC is not fulfilled. This entails that our conclusion will be about the subgradient of the estimated divergence.\\
\begin{remark}[Strict continuity of $\hat{D}_{\varphi}$ under boundedness assumption of the shape parameters]
When $\gamma<0$, we can prove strict continuity of $\phi\mapsto\hat{D}_{\varphi}(p_{\phi},p_{\phi^T})$ using Theorem \ref{theo:levelbounded}. We need to calculate $Y_{\infty}$ defined by:
\[Y_{\infty}(\phi_0) = \bigcup_{\alpha\in\argsup_{\beta} f(\beta,\phi_0)} M_{\infty}(\alpha,\phi_0),\quad \text{for } M_{\infty}(\alpha,\phi_0) = \{a|(0,a)\in\partial^{\infty} f(\alpha,\phi_0)\}.\]
Let $\alpha_0\in\argsup_{\beta} f(\beta,\phi_0)$. Since the value of $f(\beta,\phi)$ is $-\infty$ on the set $\{(\alpha,\phi)\in\Phi|\alpha_2<\phi_2\}$, then its supremum over $\beta$ is attained outside of it. Consequently, $(\alpha_0,\phi_0)$ belongs to the set $\{(\alpha,\phi)\in\Phi|\alpha_2\geq\phi_2\}$ where the integral in function $f$ is of class $\mathcal{C}^1$ which implies that $f$ is also of class $\mathcal{C}^1$. This entails that $f(\alpha,\phi)$ is strictly continuous at $(\alpha_0,\phi_0)$ which is, by Theorem 9.13 in \citep{Rockafellar}, equivalent to $\partial^{\infty}f(\alpha_0,\phi_0)=\{0\}$. Now, we may conclude that $M_{\infty}(\alpha,\phi_0)=\{0\}$, and hence $Y_{\infty}=\{0\}$. All ingredients of Theorem \ref{theo:LowerC1} are ready, and we conclude that the dual representation of the divergence $\phi\mapsto\hat{D}_{\varphi}(p_{\phi},p_{\phi^T})$ is strictly continuous. \\
If we could prove that the set $Y(\phi)$ contains only one element given a vector $\phi$, then the differentiability of the estimated divergence would be obtained using point ($b$) of Theorem \ref{theo:levelbounded}. This demands however a great effort since the characterization of the set $Y(\phi)$ demands an investigation of the form of the estimated divergence and the model used.
\end{remark}
\textbf{If we are using the kernel-based dual estimator given by (\ref{eqn:EmpiricalNewDualForm})} with a Gaussian kernel density estimator, then things are a lot simplified. We need only to treat the integral term. From an analytic point of view, the study of continuity depends on the kernel used; more specifically its tail behavior. If we take a Gaussian kernel, then we have:
\begin{itemize}
\item[$\bullet$] For $\gamma>1$, it is necessary that $\min(\phi_1,\phi_2)>2$, otherwise the estimated divergence is infinity. Thus, it is necessary for either of the true values of the shapes to be inferior to 2 in order for the estimation to be valid;
\item[$\bullet$] For $\gamma\in(0,1)$, then the estimated divergence is $\mathcal{C}^1(\text{int}(\Phi))$;
\item[$\bullet$] For $\gamma<0$, it is necessary that $\min(\phi_1,\phi_2)<1-\frac{1}{\gamma}$ and $\max(\phi_1,\phi_2)<2$. If these conditions do not hold, then the estimated divergence is minimized at $-\infty$ at any vector of parameters which does not verify the previous condition.
\end{itemize}
In the first case, if we use a heavier-tailed kernel such as the Cauchy Kernel, the estimated divergence becomes $\mathcal{C}^1(\text{int}(\Phi))$. In the third case, if we use a compact-supported kernel such as the Epanechnikov's kernel, the condition is reduced to only $\min(\phi_1,\phi_2)<1-\frac{1}{\gamma}$.\\
Similar conditions to (\ref{eqn:CondGaussMixNewDual1},\ref{eqn:CondGaussMixNewDual2}) can be obtained and we have the same conclusion as Conclusion \ref{conc:conclusion2}. \\
\textbf{In the case of the Likelihood} $\varphi(t)=-\log(t)+t-1$, we illustrate the convergence of the EM algorithm through our theoretical approach. Assumptions A1 and AC are clearly verified since both the log-likelihood and the proximal term are sums of continuously differentiable functions, and integrals do not intervene here. The set $\Phi^0$ is given by: 
\begin{eqnarray*}
\Phi^0 & = & \left\{\phi\in\Phi, J(\phi)\geq J(\phi^0)\right\} \\
       & = & J^{-1}\left([J(\phi^0),\infty)\right) \\
			 & = & \left\{\phi\in\Phi, L(\phi)\geq L(\phi^0)\right\}
\end{eqnarray*}
where $L(\phi)$ is the likelihood of the model, and $J(\phi) = \log(L(\phi))$ is the log-likelihood function. We will show that under similar conditions to the Gaussian mixture, the set $\Phi^0$ is compact. \\
\textbf{Closedness of $\Phi^0$}. Since the shape parameter is supposed to be positive, continuity of the log-likelihood would imply only that $\Phi^0$ is closed in $[0,1]\times\mathbb{R_+^*}\times\mathbb{R_+^*}$, a space which is not closed and hence is not complete. We therefore, propose to extend the definition of shape parameter on 0. From a statistical point of view, this extension is not reasonable since the density function of Weibull distribution with a shape parameter equal to 0 is the zero function which is not a probability density. Besides, identifiability problems would appear for a mixture model. Nevertheless, our need is only for analytical purpose. We will add suitable conditions on $\phi^0$ in order to avoid such subtlety keeping in hand the closedness property.\\
We suppose now that the shape parameter can have values in $\mathbb{R_+}$. The set $\Phi^0$ is now the inverse image of $[L(\phi^0),\infty)$ by the likelihood function\footnote{We do not use this time the log-likelihood function since it is not defined when both shape parameters are zero.} which is continuous on $[0,1]\times\mathbb{R_+}\times\mathbb{R_+}$. Hence, it is closed in the space $[0,1]\times\mathbb{R_+}\times\mathbb{R_+}$ embedded with the euclidean norm which is complete. It suffices then to prove that $\Phi^0$ is bounded.\\
\textbf{Boundedness of $\Phi^0$.} We will make similar arguments to the case of the Gaussian mixture example. We need to calculate the limit at infinity when the shape parameter of either of the two components tends to infinity. If both $\alpha_1$ and $\alpha_2$ goes to infinity, the log-likelihood tends to $-\infty$. Hence any choice of a finite $\phi^0$ can avoid this case. Suppose now that $\alpha_1$ goes to infinity whereas $\alpha_2$ stays bounded. The corresponding limit of the log-likelihood functions is given by:
\[J(\lambda,\infty,\alpha_2) = \sum_{i=1}^n{\log\left((1-\lambda)\frac{\alpha_2}{2}\left(\frac{y_i}{2}\right)^{\alpha_2-1} e^{-\left(\frac{y_i}{2}\right)^{\alpha_2}}\right)}\]
if there is no observation $y_i$ equal to $\frac{1}{2}$. In fact, if there is $y_i=\frac{1}{2}$, the limit is $+\infty$ and the set $\Phi^0$ cannot be bounded. However, it is improbable to get such an observation since the probability of getting an observation equal to $\frac{1}{2}$ is zero. The case when $\alpha_2$ goes to infinity whereas $\alpha_1$ stays bounded is treated similarly. \\
To avoid the two previous scenarios, one should choose the initial point of the algorithm $\phi^0$ in a way that it verifies:
\begin{equation}
J(\phi^0)> \max\left(\sup_{\lambda,\alpha_2} J(\lambda,\infty,\alpha_2), \sup_{\lambda,\alpha_1} J(\lambda,\alpha_1,\infty)\right).
\label{eqn:ConditionWeibullMix}
\end{equation}
Since all vectors of $\Phi^0$ have a log-likelihood value greater than $J(\phi^0)$, the previous choice permits the set $\Phi^0$ to avoid non finite values of $\phi$. Thus it becomes bounded whenever $\phi_0$ is chosen according to condition (\ref{eqn:ConditionWeibullMix}). Finally, the calculus of both terms $\sup_{\lambda,\alpha_1} J(\lambda,\alpha_1,\infty)$ and $\sup_{\lambda,\alpha_2} J(\lambda,\alpha_2,\infty)$ is not feasible but numerically. Those, however, can be simplified a little. One can notice by writing these terms without the logarithm (as a product), the term which has $\lambda$ is maximized when it is equal to 1. The remaining of the calculus is a maximization of the likelihood function of a Weibull model\footnote{In a Weibull model, the calculus of the MLE cannot be done but numerically when the parameter of interest is the shape parameter.}.\\
We conclude that the set $\Phi^0$ is compact under condition (\ref{eqn:ConditionWeibullMix}). Finally, it is important to notice that condition (\ref{eqn:ConditionWeibullMix}) permits also to avoid the border values which corresponds to $\alpha_1=0$ or $\alpha_2=0$. Indeed, when either of the shape parameters is zero, the corresponding component vanishes and the corresponding log-likelihood value is less than the upper bound in condition (\ref{eqn:ConditionWeibullMix}). The same conclusion as Conclusion \ref{conc:conclusion3} can be stated here for the Weibull mixture model.\\
Notice that the verification of assumption A3 is a hard task here because it results in a set of $n$ non-linear equations.
% ----------------------------------------------------------------------------------------------------
%%%%%%%%%%%%%%%%%%%%%%%%%%%%%%%%%%%%%%%%%%%%%%%%%%%%%%%
% ----------------------------------------------------------------------------------------------------
\subsection{Pearson's \texorpdfstring{$\chi^2$}{chi square} algorithm for a Cauchy model}\label{Example:Cauchy}
Let $\{(x_i,y_i),i=0,\cdots,n\}$ be an n-sample drawn from the joint probability law defined by the density function:
\[f(x,y|a,x_0) = \frac{a(y-x_0)^2e^x}{\pi\left(a^2+(y-x_0)^2e^x\right)^2}, \quad x\in[0,\infty), y\in\mathbb{R}\]
where $a\in[\varepsilon,\infty)$, with $\varepsilon>0$, denotes a scale parameter and $x_0\in\mathbb{R}$ denotes a location parameter. We define an exponential probability law with parameter $\frac{1}{2}$ on the labels. It is given by the density function:
\[q(x)=\frac{1}{2}e^{-x/2}.\]
Now, the model defined on the observed data becomes a Cauchy model with two parameters:
\[p_{(a,x_0)}(y) = \int_{0}^{\infty}{f(x,y|a,x_0)dx} = \frac{a}{\pi(a^2+(y-x_0)^2)},\quad a\geq\varepsilon>0, x_0\in\mathbb{R}.\]
The goal of this example is to show how we prove assumptions A1-3 and AC in order to explore the convergence properties of the sequence $\phi^k$ generated by either of the algorithms (\ref{eqn:DivergenceAlgo}) and (\ref{eqn:DivergenceAlgoSimp1},\ref{eqn:DivergenceAlgoSimp2}). We also discuss the analytical properties of the dual representation of the divergence.\\
In this example, we only focus on the dual representation of the divergence given by (\ref{eqn:DivergenceDef}) because the resulting MD$\varphi$DE is robust against outliers (so does the MLE). Thus there is no need to use a robust estimator such as the kernel-based MD$\varphi$DE which needs a choice of a suitable kernel and window.
%%%%%%%%%%%%%%%%%%%%%%%%%%%%%%%%%%%%%%%%%%%%%%%
\subsubsection{Cauchy model with zero location}\label{Example:CauchyZeroLoc}
 We suppose here that $x_0=0$, and we are only interested in estimating the scale parameter $a$. The Pearson's $\chi^2$ divergence is given by:
\[D(p_{a},p_{a^*}) = \frac{1}{2}\int{\left[\frac{p_{a}(y)}{p_{a^*}}-1\right]^2p_{a^*}(y)dy}.\]
Let's rewrite the dual representation of the Chi square divergence:
\[\hat{D}(p_{a},p_{a^*}) = \sup_{b\geq\varepsilon}\left\{\int_{\mathbb{R}}{\frac{p_b^2(x)}{p_a(x)}dx} - \frac{1}{2n}\sum_{i=1}^n{\frac{p_b^2(y_i)}{p_a^2(y_i)}}\right\} - \frac{1}{2}.\]
A simple calculus shows:
\[\int_{\mathbb{R}}{\frac{p_b^2(x)}{p_a(x)}dx} = \frac{(a^2+b^2)\pi}{2ab}.\]
This implies a simpler form for the dual representation of the divergence:
\begin{equation}
\hat{D}(p_{a},p_{a^*}) = \sup_{b\geq\varepsilon}\left\{\frac{(a^2+b^2)}{2ab} - \frac{1}{2n}\sum_{i=1}^n{\frac{a^2(b^2+y_i^2)^2}{b^2(a^2+y_i^2)^2}}\right\} - \frac{1}{2}.
\label{eqn:DivergenceCauchy}
\end{equation}
Let $f(a,b)$ denote the optimized function in the above formula. We calculate the first derivative with respect to $b$:
\[\frac{\partial f}{\partial b}(a,b) = -\frac{\pi a}{2b^2} + \frac{\pi}{2a} - \frac{1}{2n}\sum_{i=1}^n{\frac{a^2}{(a^2+y_i^2)^2}\left(2b-\frac{2y_i^4}{b^3}\right)}.\]
Notice that as $a\geq\varepsilon$ the term $\frac{\pi}{2a}$ stays bounded away from infinity uniformly. Therefore, it suffices then that $b$ exceeds a finite value $b_0$ in order that the derivative becomes negative. Hence, there exists $b_0$ such that $b\mapsto f(a,b)$ becomes decreasing independently of $a$. On the other hand $\forall a>0, \lim_{b\rightarrow\infty} f(a,b) = -\infty$. It results that all values of the function $b\mapsto f(a,b)$ for $b>b_0$ does not have any use in the calculus of the supremum in (\ref{eqn:DivergenceCauchy}), since, by the decreasing property if $b\mapsto f(a,b)$, they all should have values less than the value at $b_0$. We may now rewrite the dual representation of the Chi square divergence as :
\begin{equation}
\hat{D}(p_{a},p_{a^*}) = \sup_{b\in[\varepsilon,b_0]}\left\{\frac{(a^2+b^2)}{2ab} - \frac{1}{2n}\sum_{i=1}^n{\frac{a^2(b^2+y_i^2)^2}{b^2(a^2+y_i^2)}}\right\} - \frac{1}{2}.
\label{eqn:DualRepCauchyLocation}
\end{equation}
We have now two pieces of information about $f(a,b)$. The first is that it is level-bounded locally in $b$ uniformly in $a$ (see paragraph (\ref{para:LevelBoundFun})). The second is that we are exactly in the context of lower$-\mathcal{C}^1$ functions (\ref{para:LowerC1}). First of all, function $f$ is $\mathcal{C}^1([\varepsilon,\infty)\times[\varepsilon,\infty))$ function, so that part (a) of Theorem \ref{theo:levelbounded} is verified and the function $a\mapsto\hat{D}(p_{a},p_{a^*})$ is strictly continuous. To prove it is continuously differentiable, we need to prove that the set 
\[Y(a)=\bigcup_{b\in\argmax_{b'}f(a,b)}\left\{\frac{\partial f}{\partial a}(a,b)\right\}\]
contains but one element. From a theoretic point of view, two possible methods are available: Prove that either there is a unique maximum for a fixed $a$, or that the derivative with respect to $a$ at all maxima \emph{does not depend} on $a$ (they have the same value). In our example, function $b\mapsto f(a,b)$ is not concave. We may also plot it using any mathematical tool provided that we already have the data set. We tried out a simple example and generated a 10-sample of the standard Cauchy distribution ($a=1$), see table (\ref{tab:CauchyDerivEx}). We used Mathematica to draw a 3D figure of function $f$, see figure (\ref{fig:CauchyExample3D}).
\begin{table}[h]
\centering
\begin{tabular}{|c|c|c|c|c|c|c|c|c|c|c|}
$y_i$ & 0.534 & -18.197 & 0.726 & -0.439 & -1.945 & 0.0119 & 12.376 & -0.953 & 0.698 & 0.818\\
\hline
\end{tabular}
\caption{A 10-sample Cauchy dataset.}
\label{tab:CauchyDerivEx}
\end{table}

\begin{figure}[ht]
\centering
\includegraphics[scale = 0.35]{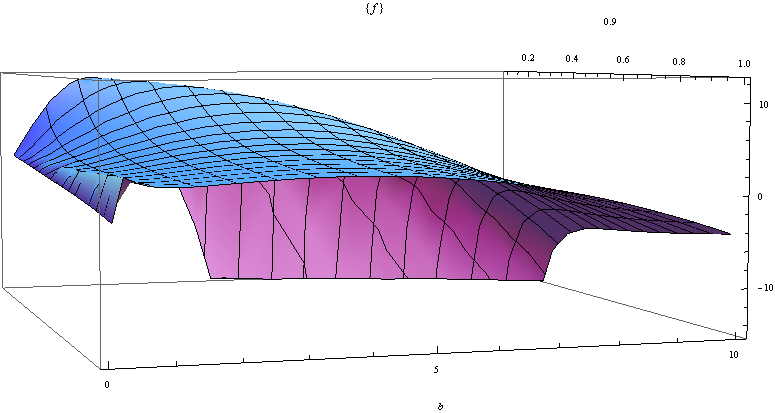}
\caption{A 3D plot of function $f(a,b)$ for a 10-sample of the standard Cauchy distribution.}
\label{fig:CauchyExample3D}
\end{figure}
It is clear that for a fixed $a$, the function $b\mapsto f(a,b)$ has two maxima which may both be global maxima. For example for $a=0.9$, one gets figure (\ref{fig:CauchyExample2D}). It is clearer now that conditions of Theorem \ref{theo:levelbounded} are not fulfilled, and we cannot prove that function $\hat{D}(p_{a},p_{a^*})$ is continuously differentiable every where.\\
\begin{figure}[ht]
\centering
\includegraphics[scale = 0.2]{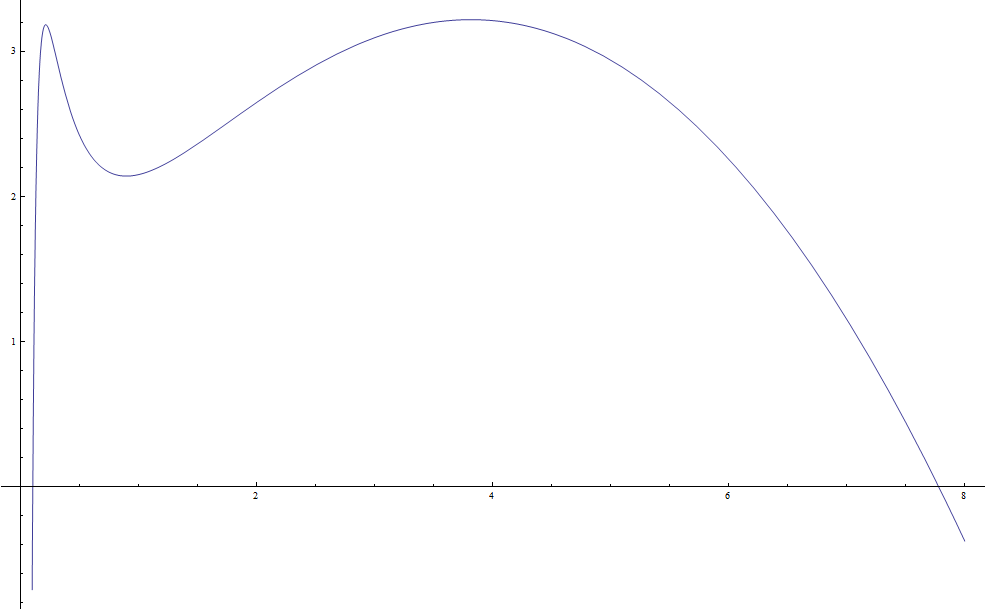}
\caption{A 2D plot of function $f(0.9,b)$ for a 10-sample of the standard Cauchy distribution.}
\label{fig:CauchyExample2D}
\end{figure}
\noindent It is however not the end of the road. We still have the results presented in paragraph (\ref{para:LowerC1}). Function $\hat{D}(p_{a},p_{a^*})$ is lower-$\mathcal{C}^1$. Therefore, it is strictly continuous and almost everywhere continuously differentiable. Hence, we may hope that the limit points of the sequence $(\phi^k)_k$ for algorithm (\ref{eqn:DivergenceAlgo}) are in the set of points where the dual representation of the Chi square divergence is $\mathcal{C}^1$, or be more reasonable and state any further result on the sequence in terms of the subgradient of $\hat{D}(p_{a},p_{a^*})$.\\
\paragraph{Compactness of $\Phi^0$.} We check when the set $\Phi^0=\{a|\hat{D}(p_{a},p_{a^*})\leq\hat{D}(p_{a_0},p_{a^*})\}$ is closed and bounded in $[\varepsilon,\infty)$ for an initial point $a_0$. \textbf{Closedness} is proved using continuity of $\hat{D}(p_{a},p_{a^*})$. Indeed,
\[\Phi^0 = \hat{D}^{-1}(p_{a},p_{a^*})\left((-\infty,\hat{D}(p_{a_0},p_{a^*})]\right).\]
\textbf{Boundedness} is proved by contradiction. Suppose that $\Phi^0$ is unbounded, then there exists a sequence $(a^l)_l$ of points of $\Phi^0$ which goes to infinity. Formula (\ref{eqn:DualRepCauchyLocation}) shows that $b$ stays in a bounded set during the calculus of the supremum. Hence the continuity of $\hat{D}(p_{a},p_{a^*})$ implies: 
\[\lim_{a\rightarrow\infty} \hat{D}(p_{a},p_{a^*}) = +\infty.\]
This shows that by choosing any finite $a_0$, the set $\Phi^0$ becomes bounded. Indeed, the relation defining $\Phi^0$ implies that $\forall l, \hat{D}(p_{a^l},p_{a^*})\leq\hat{D}(p_{a_0},p_{a^*})<\infty$, and a contradiction is reached by taking the limit of each part of this inequality. Hence $\Phi^0$ is closed and bounded in the space $[\varepsilon,\infty)$ which is complete provided with the euclidean distance. We conclude that $\Phi^0$ is compact\footnote{If we are to use a result which concerns the differentiability of $\hat{D}(p_{a},p_{a^*})$, one should consider the case when $\Phi^0$ shares a boundary with $\Phi$. A possible solution to avoid this is to consider an initial point $a^0$ such that $\hat{D}(p_{\varepsilon},p_{a^*})>\hat{D}(p_{a_0},p_{a^*})$. This expels the the boundary from the possible values of $\Phi^0$.}.\\
In this simple example, we only can use algorithm (\ref{eqn:DivergenceAlgo}) since there is only one parameter of interest. Proposition \ref{prop:NewRes} can be used to deduce convergence of any convergent subsequence to a generalized stationary point of $\hat{D}(p_{a},p_{a^*})$.\\
To deduce more results about the sequence $(a^k)_k$, we may try and verify assumption A3 using Lemma \ref{Lem:Tseng}. Let's write functions $h_i$.
\[h_i(x|a) = \frac{f(x,y_i|a)}{p_a(y_i)} = \frac{y_i^2e^x(a^2+y_i^2)}{(a^2+e^xy_i^2)^2}.\]
Clearly, for any $i\in\{1,\cdots,n\}$ and $a\geq\varepsilon$, function $x\mapsto h_i(x|a)$ is continuous. Let $a,b\geq\varepsilon$ such that $a\neq b$. Suppose that:
\[\forall i, \quad h_i(x|a)=h_i(x|b) \qquad \forall x\geq 0.\]
This entails that:
\[a^2b^4-a^4b2+(b^4-a^4)y_i^2+\left(a^2e^{2x}+2b^2e^x-b^2e^{2x}-2a^2e^x\right)y_i^4=0,\qquad i=1,\cdots,n.\]
This is a polynomial on $y_i$ of degree 4 which coincides with the zero polynomial on $n$ points. If there exists 5 distinct observations\footnote{If one uses the point $x=0$, the result follows directly without supposing the existence of distinct observations.}, then the two polynomials will have the same coefficients. Hence, we have $b^4-a^4=0$. This implies that $a=b$ since they are both positive real numbers. We conclude that $D_{\psi}(a,b)=0$ whenever $a=b$ which is equivalent to assumption A3. Proposition \ref{prop:PhiDiffConverge} can now be applied to deduce that sequence $(a^k)$ defined by (\ref{eqn:DivergenceAlgo}) (with $\phi^k$ replaced by $a^k$) is well defined and bounded. Furthermore, it verifies $a^{k+1}-a^k\rightarrow 0$, and the limit of any convergent subsequence is a generalized stationary point of $\hat{D}(p_{a},p_{a^*})$. The existence of such subsequence is guaranteed by the compactness of $\Phi^0$ and the fact that $\forall k, a^k\in\Phi^0$.
%%%%%%%%%%%%%%%%%%%%%%%%%%%%%%%%%%%%%%%%%%%%%%%%
\subsubsection{Cauchy model with both parameters}
The model is now defined by:
\[p_{(a,x_0)}(y) = \frac{a}{\pi(a^2+(y-x_0)^2)},\quad a\geq\varepsilon>0, x_0\in\mathbb{R}. \]
Formula (\ref{eqn:DivergenceCauchy}) of the dual representation of the Chi square divergence becomes:
\begin{equation}
\hat{D}(p_{a,x_0},p_{a^*,x_0}) = \sup_{b\geq\varepsilon,x_1\in\mathbb{R}}\left\{\frac{(a^2+b^2+(x_1-x_0)^2)}{2ab} - \frac{1}{2n}\sum_{i=1}^n{\frac{a^2(b^2+(y_i-x_1)^2)^2}{b^2(a^2+(y_i-x_0)^2)^2}}\right\} - \frac{1}{2}.
\label{eqn:DivergenceCauchyAllpara}
\end{equation}
Let $f(a,b,x_0,x_1)$ be the optimized function in the previous formula. This time, it does not seem easy to prove that the supremum can be calculated on a compact set. We, therefore, work on the second approach to study continuity of $\hat{D}(p_{a,x_0},p_{a^*,x_0})$, i.e. level-boundedness approach (see paragraph (\ref{para:LevelBoundFun})). For $a$, let $(a-\tilde{a},a+\tilde{a})\subset[\varepsilon,\infty)$ be an open neighborhood around $a$, and for $x_0$, let $(x_0-\tilde{x},x_0+\tilde{x})$ be an open neighborhood around $x_0$. It is clear that as either $b\rightarrow\infty$ or $x_1\rightarrow\pm\infty$, we have $f(a,b,x_0,x_1)\rightarrow -\infty$ since the first term in $f$ is of order $b$ (resp. $x_1^2$) whereas the second term in $f$ is of order $b^2$ (resp. $x_1^4$) as long as $a$ is bounded away from zero and $x_0$ is supposed to be bounded. Finally, when both $b$ and $x_1$ goes to infinity, the important terms in calculating the limit are of order $b-b^2$ and $\frac{x_1^2}{b} - \left(\frac{x_1^2}{b}\right)^2$. Hence the limit is \emph{a fortiori} $-\infty$. We conclude that:
\[f(a,b,x_0,x_1)\xrightarrow[\|(b,x_1)\|\rightarrow\infty]{}-\infty\]
Now that $f$ is level-bounded in $(b,x_1)$ locally uniformly in $(a,x_0)$, and since $f$ is readily continuous (so it is upper semicontinuous), all ingredients for Theorem \ref{theo:levelbounded} part $(a)$ are ready. Hence $\hat{D}(p_{a,x_0},p_{a^T,x_0^T})$ is strictly continuous\footnote{In order to get the same results we have on lower-$\mathcal{C}^1$ functions, we need to prove that $\hat{D}(p_{a,x_0},p_{a^T,x_0^T})$ is also regular in the sense of Clarke at all points of its domain; a result which we get by a theorem of Rademarcher, see \citep{Rockafellar} Chapter 9.}.\\
Now that $\hat{D}(p_{a,x_0},p_{a^*,x_0})$ is continuous, we may use analogous arguments to those given in the previous paragraph to prove closedness and boundedness of $\Phi^0$. Boundedness is treated a bit differently since the supremum is no longer calculated over a bounded set. By definition of the supremum, one can write:
\begin{eqnarray*}
\sup_{b\geq\varepsilon,x_1\in\mathbb{R}}\left\{\frac{((a^l)^2+b^2+(x_1-x^l_0)^2)}{2a^lb} - \frac{1}{2n}\sum_{i=1}^n{\frac{(a^l)^2(b^2+(y_i-x_1)^2)^2}{b^2((a^l)^2+(y_i-x^l_0)^2)^2}}\right\} - \frac{1}{2} \geq \\
\left\{\frac{((a^l)^2+b'^2+(x'_1-x_0^l)^2)}{2a^lb'} - \frac{1}{2n}\sum_{i=1}^n{\frac{(a^l)^2(b'^2+(y_i-x'_1)^2)^2}{b'^2((a^l)^2+(y_i-x_0^l)^2)^2}}\right\} - \frac{1}{2}
\end{eqnarray*}
for any $b'\geq \varepsilon$ and $x'_1\in\mathbb{R}$. As the sequence $(a^l,x_0^l)$ goes to infinity, the second hand of the previous inequality tends to infinity. Hence the limit of the left hand is also infinity. Thus, $\lim_{l\rightarrow\infty} \hat{D}(p_{a^l,x_0^l},p_{a^T,x_0^T}) = \infty$. We conclude that by choosing any finite initialization, the set $\Phi^0$ becomes bounded. As we could not give any argument about the differentiability of $\hat{D}(p_{a,x_0},p_{a^*,x_0})$, the only theoretical results we may state are about the subgradient of $\hat{D}(p_{a,x_0},p_{a^*,x_0})$.\\
Finally, we prove that $D_{\psi}((a,x_0),(b,x_1))>0$ as $(a,x_0)\neq (b,x_1)$. We use Lemma \ref{Lem:Tseng}. Let $x=0$; a point at which $h_i$ is (right)\footnote{In the proof of the lemma, we use continuity to deduce a certain result in a neighborhood of a point at which function $h_i$ is continuous. Here, right continuity still gives us a neighborhood of the form $[x,x+\epsilon)$ which suffices to complete the proof.} continuous. We need to prove that there exists $i$ such that if $h_i(0|a,x_0)= h_i(0|b,x_1)$ then $a=b$ and $x_0=x_1$. We have $h_i(0|a,x_0)= h_i(0|b,x_1)$ is equivalent to
\[\frac{(y_i-x_0)^2(a^2+(y_i-x_0)^2)}{(a^2+(y_i-x_0)^2)^2} = \frac{(y_i-x_1)^2(b^2+(y_i-x_1)^2)}{(b^2+(y_i-x_1)^2)^2}, \quad \forall i\in\{1,\cdots,n\},\]
or equivalently
\begin{multline*}
\left(b^2+x_0^2+x_1^2+4x_1x_0\right)y_i^2-2\left(b^2x_0+x_0x_1^2+x_1x_0^2\right)y_i+b^2x_0^2+x_1^2x_0^2 = \\
\left(a^2+x_1^2+x_0^2+4x_0x_1\right)y_i^2-2\left(a^2x_1+x_1x_0^2+x_0x_1^2\right)y_i+a^2x_1^2+x_0^2x_1^2.
\end{multline*}
Suppose that there are at least three distinct observations. The previous identities can be rewritten as an identity between two polynomial of degree 2 in $y_i$. Since they are equal at three distinct roots, they must be identical and all coefficients are equal with their corresponding ones. The coefficient of $y_i^2$ suffices to deduce that is $a=b$. Identify now the coefficients of $y_i$ to get that $b^2x_0=a^2x_1$. Since $a,b>0$, then $x_0=x_1$.\\
We finally conclude using Proposition \ref{prop:PhiDiffConverge} that if we use algorithm (\ref{eqn:DivergenceAlgo}) or algorithm (\ref{eqn:DivergenceAlgoSimp1},\ref{eqn:DivergenceAlgoSimp2}), the distance between two consecutive terms of the sequence $(a^k,x_0^k)_k$ tends to 0 and any limit point of the sequence is a generalized stationary point of $\hat{D}(p_{a,x_0},p_{a^T,x_0^T})$.
%%%%%%%%%%%%%%%%%%%%%%%%%%%%%%%%%%%%%%%%%%%%%%%%%%%%%%%%%%%%%%%%%%%%%%%%%%%%
%
%--------------------------------------------------------------------------------------
% =======================================================================================
%--------------------------------------------------------------------------------------
%
%%%%%%%%%%%%%%%%%%%%%%%%%%%%%%%%%%%%%%%%%%%%%%%%%%%%%%%%%%%%%%%%%%%%%%%%%%%%

\section{Simulation study}\label{sec:SimulationsProximal}
We summarize the results of 100 experiments on $100$-samples (with and without outliers) from two-component Gaussian and Weibull mixtures by giving the average of the error committed with the corresponding standard deviation. The error criterion is mainly the total variation distance (TVD) defined by (\ref{eqn:TVDError}). We also provide for the Gaussian mixture the values of the $\chi^2$ divergence between the estimated model and the true mixture, defined by (\ref{eqn:Chi2Error}), since it gave infinite values for the Weibull experiment.\\
We used different $\varphi-$divergences to estimate the parameters and compared the performances of the two dual formulas of estimating a $\varphi-$divergence (\ref{eqn:DivergenceDef}) and (\ref{eqn:EmpiricalNewDualForm}). We also included the MDPD of \cite{BasuMPD} defined by (\ref{eqn:MDPDdef}) for a tradeoff parameter $a=0.5$. This parameter resulted in very good results throughout all simulations carried in the previous chapter. Other estimators of the divergence could also be considered in a future work, see the simulations of the previous chapter for a detailed comparison between these estimators. For the Gaussian mixture, we used the Pearson's $\chi^2$ and the Hellinger divergences, whereas in the Weibull mixture, we used the Neymann's $\chi^2$ and the Hellinger divergences. For the proximal term, we used $\psi(t)=\frac{1}{2}(\sqrt{t}-1)^2$. We illustrate also the performance of the EM method in the light of our method, i.e. using initializations verifying conditions (\ref{eqn:TwoGaussMixCond}) for the Gaussian mixture and conditions (\ref{eqn:ConditionWeibullMix}) for the Weibull one. When outliers were added, these initializations did not give always good results and the convergence of the proportion was towards the border $\eta=0.1$ or $1-\eta=0.9$. In such situations, the EM algorithm was initialized using another starting point manually.\\
We used the Nelder-Mead algorithm (see \citep{NelderMead}) for all optimization calculus. The method proved to be more efficient in our context than other optimization algorithms although having a slow convergence. Such method is derivative-free and applies even if the the objective function is not differentiable which may be the case of the estimated divergence defined through (\ref{eqn:DivergenceDef}). The Nelder-Mead algorithm is known to give good results in problems with dimension at least 2 and does not perform well in dimension 1. We thus used Brent's method in such cases. It is also a derivative-free method which works only in a compact subset from $\mathbb{R}$. The calculus was done under the statistical tool \citep{Rtool}. In what concerns numerical integrations, see Section \ref{sec:Simulations}.
% ----------------------------------------------------------------------------------------------------
%%%%%%%%%%%%%%%%%%%%%%%%%%%%%%%%%%%%%%%%%%%%%%%%%%%%%%%
% ----------------------------------------------------------------------------------------------------
\subsection{The two-component Gaussian mixture revisited}\label{Example:DivergenceMixture}
We consider the Gaussian mixture (\ref{eqn:GaussMixModel}) presented earlier with true parameters $\lambda=0.35,\mu_1=2,\mu_2=1.5$ and fixed variances $\sigma_1=\sigma_2=1$. Since we are using an error function criterion, label-switching problems do not interfere. Figure (\ref{fig:DecreaseDivGaussChi2Chi2}) shows the values of the estimated divergence for both formulas (\ref{eqn:DivergenceDef}) and (\ref{eqn:EmpiricalNewDualForm}) on a logarithmic scale at each iteration of the algorithm. The 1-step algorithm refers to algorithm (\ref{eqn:DivergenceAlgo}), whereas 2-step refers to algorithm (\ref{eqn:DivergenceAlgoSimp1},\ref{eqn:DivergenceAlgoSimp2}). We omitted the initial point in order to produce a clear image of the decrease of the objective function. For the kernel-based dual formula, we used a Gaussian kernel. Results are given in table (\ref{tab:ErrGauss100RunEx}).\\
We used the same data simulated in paragraph \ref{subsec:GaussMix}, so that contamination was done by adding in the original sample to the 5 lowest values random observations from the uniform distribution $\mathcal{U}[-5,-2]$. We also added to the 5 largest values random observations from the uniform distribution $\mathcal{U}[2,5]$. Results are presented in table (\ref{tab:ErrGauss00RunExOutl}).\\
It is clear that the kernel-based MD$\varphi$DE is more robust than the EM algorithm and the classical MD$\varphi$DE for both the Pearson's $\chi^2$ and the Hellinger divergences. Differences between the two divergences are not significant for both estimation methods of the divergence. Besides, in comparison with the results obtained with a direct optimization in paragraph \ref{subsec:GaussMix}, we find no significant differences. The proximal point algorithm wored as well on the density power divergence. The MDPD produced robust estimates with minor differences with respect to the kernel-based MD$\varphi$DE in favor of the former.

 %%%%%%%%%%%%%%%%%%%%%%%%%%%%%%%%%%%%%%%%%%%%%%%%
%%%%%%%%%%%%%%%%%%%%%%%%%%%%%%%%%%%%%%%%%%%%%ùù
%\subsubsection{Simulations}
\begin{figure}[ht]
\centering
\includegraphics[scale = 0.3]{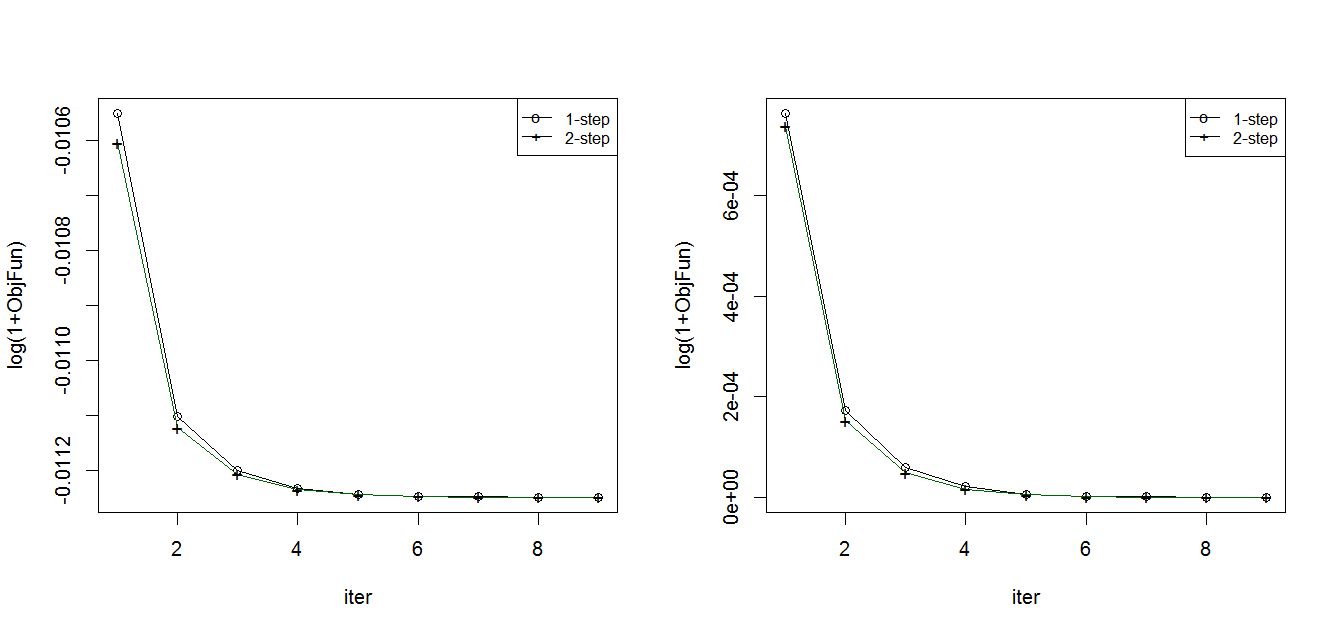}
\caption{Decrease of the (estimated) Hellinger divergence between the true density and the estimated model at each iteration in the Gaussian mixture. The figure to the left is the curve of the values of the kernel-based dual formula (\ref{eqn:EmpiricalNewDualForm}). The figure to the right is the curve of values of the classical dual formula (\ref{eqn:DivergenceDef}). Values are taken at a logarithmic scale $\log(1+x)$.}
\label{fig:DecreaseDivGaussChi2Chi2}
\end{figure}

\begin{table}[hp]
\centering
\begin{tabular}{|c|c|c|c|}
\hline
\multicolumn{2}{|c|}{\multirow{2}{2.5cm}{Estimation method}} & \multicolumn{2}{|c|}{Error criterion}\\
\cline{3-4}
 \multicolumn{2}{|c|}{} & $\sqrt{\chi^2}$ & TVD\\
 \hline
 \hline
\multicolumn{4}{|c|}{Chi square} \\
 \hline
 \hline
\multirow{2}{2.5cm}{Algorithm (\ref{eqn:DivergenceAlgo})}& MD$\varphi$DE & 0.108, sd = 0.052 & 0.061, sd = 0.029\\
& kernel-based MD$\varphi$DE & 0.118 , sd = 0.052 & 0.066 ,sd= 0.027\\
\hline
\multirow{2}{2.5cm}{Algorithm (\ref{eqn:DivergenceAlgoSimp1},\ref{eqn:DivergenceAlgoSimp2})}& MD$\varphi$DE & 0.108, sd = 0.052 & 0.061, sd = 0.029\\
& kernel-based MD$\varphi$DE & 0.118, sd = 0.051 & 0.066 ,sd= 0.027\\
%									& D$\varphi$DE & 0.003028, sd = 0.002693& 0.10820, sd = 0.05239 & 0.06097, sd = 0.0291\\
\hline
\hline
\multicolumn{4}{|c|}{Hellinger} \\
\hline
\hline
\multirow{2}{2.5cm}{Algorithm (\ref{eqn:DivergenceAlgo})}& MD$\varphi$DE & 0.108, sd = 0.052 & 0.050 , sd=0.025\\
& kernel-based MD$\varphi$DE & 0.113, sd = 0.044 & 0.064 ,sd=0.025\\
%           & D$\varphi$DE & 0.003027 , sd = 0.00269 & 0.1081, sd = 0.02606 & 0.06097, sd = 0.02908\\
\hline
%\hline
\multirow{2}{3cm}{Algorithm (\ref{eqn:DivergenceAlgoSimp1},\ref{eqn:DivergenceAlgoSimp2})}& MD$\varphi$DE & 0.108, sd = 0.052 & 0.061, sd = 0.029\\
& kernel-based MD$\varphi$DE & 0.113, sd = 0.045 & 0.064 ,sd=0.025\\
\hline
\hline
\multicolumn{2}{|c|}{MDPD $a=0.5$ - Algorithm (\ref{eqn:DivergenceAlgo})} & 0.117, sd = 0.049 & 0.065, sd = 0.025 \\
\multicolumn{2}{|c|}{MDPD $a=0.5$ - Algorithm (\ref{eqn:DivergenceAlgoSimp1},\ref{eqn:DivergenceAlgoSimp2})}  & 0.117, sd = 0.047 & 0.065, sd = 0.025 \\
\hline
\hline
\multicolumn{2}{|c|}{EM} 	& 0.113, sd = 0.044 & 0.064 , sd = 0.025\\
\hline
\end{tabular}
\caption{The mean value of errors committed in a 100-run experiment with the standard deviation. No outliers are considered here. The divergence criterion is the Chi square divergence or the Hellinger. The proximal term is calculated with $\psi(t) = \frac{1}{2} (\sqrt{t}-1)^2$.}
\label{tab:ErrGauss100RunEx}
\end{table}
%%%%%%%%%%%%%%%%%%%%%%%%%%%%%%%%%%%%%%%%%%%%%%%%%%%%%
\begin{table}[hp]
\centering
\begin{tabular}{|c|c|c|c|}
\hline
\multicolumn{2}{|c|}{\multirow{2}{2.5cm}{Estimation method}} & \multicolumn{2}{|c|}{Error criterion}\\
\cline{3-4}
 \multicolumn{2}{|c|}{} & $\chi^2$ & TVD\\  \hline
 \hline
\multicolumn{4}{|c|}{Chi square} \\
 \hline
 \hline
\multirow{2}{2.5cm}{Algorithm (\ref{eqn:DivergenceAlgo})}& MD$\varphi$DE & 0.334, sd = 0.097 & 0.146,sd=0.036\\
 & kernel-based MD$\varphi$DE & 0.149 , sd = 0.059 & 0.084 ,sd=0.033\\
\hline
%\hline
\multirow{2}{3cm}{Algorithm (\ref{eqn:DivergenceAlgoSimp1},\ref{eqn:DivergenceAlgoSimp2})}& MD$\varphi$DE & 0.333, sd = 0.097 & 0.149, sd = 0.033\\
& kernel-based MD$\varphi$DE & 0.149 , sd = 0.059 & 0.084, sd=0.033\\
%\cline{2-5}
%                               & D$\varphi$DE & 0.01333, sd = 0.0051 & 0.3332, sd = 0.0971 & 0.1487, sd = 0.03303\\ 
\hline
\hline
\multicolumn{4}{|c|}{Hellinger} \\
\hline
\hline
\multirow{2}{2.5cm}{Algorithm (\ref{eqn:DivergenceAlgo})}& MD$\varphi$DE & 0.321, sd = 0.096 & 0.146, sd=0.034\\
& kernel-based MD$\varphi$DE & 0.155 , sd = 0.059 & 0.087 ,sd=0.033\\
\hline
%\hline
\multirow{2}{3cm}{Algorithm (\ref{eqn:DivergenceAlgoSimp1},\ref{eqn:DivergenceAlgoSimp2})}& MD$\varphi$DE & 0.322, sd = 0.097 & 0.147, sd = 0.034\\
& kernel-based MD$\varphi$DE & 0.156 , sd = 0.059 & 0.087 ,sd=0.033\\
%\cline{2-5}
%                               & D$\varphi$DE & 0.01313, sd = 0.00529 & 0.3242 , sd = 0.0972 & 0.1476 , sd = 0.0339\\ 
\hline
\hline
\multicolumn{2}{|c|}{MDPD $a=0.5$ - Algorithm (\ref{eqn:DivergenceAlgo})}  & 0.129, sd = 0.049 & 0.065, sd = 0.025 \\
\multicolumn{2}{|c|}{MDPD $a=0.5$ - Algorithm (\ref{eqn:DivergenceAlgoSimp1},\ref{eqn:DivergenceAlgoSimp2})}  & 0.138, sd = 0.053 & 0.078, sd = 0.030 \\
\hline
\hline
\multicolumn{2}{|c|}{EM} & 0.335, sd = 0.102 & 0.150, sd = 0.034 \\
\hline
\end{tabular}
\caption{Error committed in estimating the parameters of a 2-component Gaussian mixture with $10\%$ outliers. The divergence criterion is the Chi square divergence or the Hellinger. The proximal term is calculated with $\psi(t) = \frac{1}{2} (\sqrt{t}-1)^2$.}
\label{tab:ErrGauss00RunExOutl}
\end{table}
%\clearpage
%%%%%%%%%%%%%%%%%%%%%%%%%%%%%%%%%%%%%%%%%%%%%%%%%%%%%%%%%%%%%%%%%%%%%%%%%%
%==========================================================================
%%%%%%%%%%%%%%%%%%%%%%%%%%%%%%%%%%%%%%%%%%%%%%%%%%%%%%%%%%%%%%%%%%%%%%%%%%
\subsection{The two-component Weibull mixture model revisited}
We consider the Weibull mixture (\ref{eqn:WeibullMixture}) with shapes $\phi_1 = 0.5, \phi_2 = 3$ and $\lambda = 0.35$ which are supposed to be unknown during the estimation procedure. Here, we denote $\phi=(\phi_1,\phi_2)$ ($\alpha=(\alpha_1,\alpha_2)$, respectively) the shapes of the Weibull mixture model $p_{(\lambda,\phi)}$ ($p_{(\lambda,\alpha)}$, respectively). We used the same data simulated in paragraph \ref{subsec:TwoWeibullSimuPara}, so that contamination was done by replacing 10 observations of each sample chosen randomly by 10 i.i.d. observations drawn from a Weibull distribution with shape $\nu = 0.9$ and scale $\sigma = 3$. Results are presented in tables (\ref{tab:ErrWeibull100RunEx}) and (\ref{tab:ErrWeibull100RunOutliersEx}).\\ 
Manipulating the optimization procedure for the Neymann's $\chi^2$ was difficult because of the numerical integration calculus and the fact that for a subset of $\Phi$ (or $\Phi\times\Phi$ according to whether we use the estimator (\ref{eqn:DivergenceDef}) or the estimator (\ref{eqn:EmpiricalNewDualForm})) the integral term produces infinity, see paragraph \ref{subsec:WeibullMixEx} for more details. We therefore needed to keep the optimization from approaching the border in order to avoid numerical problems. For the Hellinger divergence, there is no particular remark.\\
For the case of the estimated divergence (\ref{eqn:DivergenceDef}), if $\gamma=-1$, i.e. the Neymann $\chi^2$, we need that $\alpha_1 < \phi_1/2$, otherwise the integral term is equal to infinity. In order to avoid numerical complications, we optimized over $\alpha_1\leq 0.05+\phi_1/2$. The value $0.05$ ensures a small deviation from the border. \\
For the case of the estimated divergence (\ref{eqn:EmpiricalNewDualForm}), we used a Gaussian kernel for the Hellinger divergence. For the Neymann's $\chi^2$ divergence, we used the Epanechnikov's kernel to avoid problems at infinity. Besides, it permits to integrate only over $[0,\max(Y)+w]$, where $w$ is the window of the kernel, instead of $[0,\infty)$. In order to avoid problems near zero, it is necessary that $\min(\phi_1,\phi_2)<1-\frac{1}{\gamma}=2$.\\
\paragraph{Comments on the tables:} Experimental results show a clear robustness of the estimators calculated using the kernel-based MD$\varphi$DE in comparison to other estimators using the Hellinger divergence. When we are under the model, all estimation methods have the same performance. On the other hand, using the Neymann $\chi^2$ divergence, results are different in the presence of outliers. The classical MD$\varphi$DE calculated using formula (\ref{eqn:DivergenceDef}) shows better robustness than other estimators, but is still not as good as the robustness of the kernel-based MD$\varphi$DE using the Hellinger. Lack of robustness of the kernel-based MD$\varphi$DE is not very surprising since the influence function of the kernel-based MD$\varphi$DE is unbounded when we use the Neymann $\chi^2$ divergence in simple models such as the Gaussian model, see Example \ref{Example:GaussIF}.\\
In what concerns the proximal algorithm, there is no significant difference between the results obtained using the 1-step algorithm (\ref{eqn:DivergenceAlgo}) and the ones obtained using the 2-step algorithm (\ref{eqn:DivergenceAlgoSimp1},\ref{eqn:DivergenceAlgoSimp2}) using the Hellinger divergence. Differences appear when we used the Neymann $\chi^2$ divergence with the classical MD$\varphi$DE. This shows again the difficulty in handling the supermal form of the dual formal (\ref{eqn:DivergenceDef}). Finally, in comparison to the results obtained with a direct optimization in paragraph \ref{subsec:TwoWeibullSimuPara}, there is no significant differences.\\
The proximal-point algorithm worked as well using the density power divergence. The MDPD produced robust and efficient estimates which are slightly better than the results obtained by the kernel-based MD$\varphi$DE using the Hellinger divergence and clearly better than the results obtained using the Neymann $\chi^2$.

%===============================================================
%%%%%%%%%%%%%%%%%%%%%%%%%%%%%%%%%%%%%%%%%%%%%%%%%%%%%%%%%%%
\begin{table}[hp]
\centering
\begin{tabular}{|c|c|c|}
\hline
\multicolumn{2}{|c|}{\multirow{2}{2.5cm}{Estimation method}} & Error criterion\\
\cline{3-3}
 \multicolumn{2}{|c|}{} & TVD\\
 \hline
 \hline
\multicolumn{3}{|c|}{Neymann Chi square} \\
 \hline
 \hline
\multirow{2}{2.5cm}{Algorithm (\ref{eqn:DivergenceAlgo})}& MD$\varphi$DE & 0.114 , sd = 0.032\\
									& kernel-based MD$\varphi$DE & 0.057, sd = 0.028\\
\hline
\multirow{2}{2.5cm}{Algorithm (\ref{eqn:DivergenceAlgoSimp1},\ref{eqn:DivergenceAlgoSimp2})}& MD$\varphi$DE &  0.131, sd =  0.042\\
									& kernel-based MD$\varphi$DE & 0.056, sd = 0.026\\
\hline
\hline
\multicolumn{3}{|c|}{Hellinger} \\
\hline
\hline
\multirow{2}{2.5cm}{Algorithm (\ref{eqn:DivergenceAlgo})}& MD$\varphi$DE & 0.059, sd = 0.024\\
									& kernel-based MD$\varphi$DE & 0.057, sd = 0.029 \\
\hline
\multirow{2}{2.5cm}{Algorithm (\ref{eqn:DivergenceAlgoSimp1},\ref{eqn:DivergenceAlgoSimp2})}& MD$\varphi$DE & 0.061, sd = 0.026\\
									& kernel-based MD$\varphi$DE & 0.057, sd = 0.029\\
\hline
\hline
\multicolumn{2}{|c|}{MDPD $a=0.5$ - Algorithm (\ref{eqn:DivergenceAlgo})}  & 0.056, sd = 0.029 \\
\multicolumn{2}{|c|}{MDPD $a=0.5$ - Algorithm (\ref{eqn:DivergenceAlgoSimp1},\ref{eqn:DivergenceAlgoSimp2})}  & 0.056, sd = 0.029 \\
\hline
\hline
\multicolumn{2}{|c|}{EM} & 0.059, sd = 0.024\\
\hline
\end{tabular}
\caption{The mean value of errors committed in a 100-run experiment of a two-component Weibull mixture with the standard deviation. No outliers are considered. The divergence criterion is the Neymann's $\chi^2$ divergence or the Hellinger. The proximal term is taken with $\psi(t) = \frac{1}{2} (\sqrt{t}-1)^2$.}
\label{tab:ErrWeibull100RunEx}
\end{table}

\begin{table}[hp]
\centering
\begin{tabular}{|c|c|c|}
\hline
\multicolumn{2}{|c|}{\multirow{2}{2.5cm}{Estimation method}} & Error criterion\\
\cline{3-3}
 \multicolumn{2}{|c|}{} & TVD\\
 \hline
 \hline
\multicolumn{3}{|c|}{Neymann Chi square} \\
 \hline
 \hline
\multirow{2}{2.5cm}{Algorithm (\ref{eqn:DivergenceAlgo})}& MD$\varphi$DE & 0.085,  sd = 0.036\\
									& kernel-based MD$\varphi$DE & 0.138, sd = 0.066 \\
\hline
\multirow{2}{2.5cm}{Algorithm (\ref{eqn:DivergenceAlgoSimp1},\ref{eqn:DivergenceAlgoSimp2})}& MD$\varphi$DE & 0.096, sd = 0.057\\
									& kernel-based MD$\varphi$DE & 0.127, sd = 0.056\\
\hline
\hline
\multicolumn{3}{|c|}{Hellinger} \\
\hline
\hline
\multirow{2}{2.5cm}{Algorithm (\ref{eqn:DivergenceAlgo})}& MD$\varphi$DE & 0.120, sd = 0.034\\
									& kernel-based MD$\varphi$DE & 0.068, sd = 0.034 \\
\hline
\multirow{2}{2.5cm}{Algorithm (\ref{eqn:DivergenceAlgoSimp1},\ref{eqn:DivergenceAlgoSimp2})}& MD$\varphi$DE & 0.121, sd = 0.034\\
									& kernel-based MD$\varphi$DE & 0.068, sd = 0.034\\
\hline
\hline
\multicolumn{2}{|c|}{MDPD $a=0.5$ - Algorithm (\ref{eqn:DivergenceAlgo})}  & 0.060, sd = 0.029 \\
\multicolumn{2}{|c|}{MDPD $a=0.5$ - Algorithm (\ref{eqn:DivergenceAlgoSimp1},\ref{eqn:DivergenceAlgoSimp2})}  & 0.061, sd = 0.029 \\
\hline
\hline
\multicolumn{2}{|c|}{EM} & 0.129, sd = 0.046\\
\hline
\end{tabular}
\caption{The mean value of errors committed in a 100-run experiment of a two-component Weibull mixture with the standard deviation. $10\%$ outliers are considered. The divergence criterion is the Neymann's $\chi^2$ divergence or the Hellinger. The proximal term is taken with $\psi(t) = \frac{1}{2} (\sqrt{t}-1)^2$.}
\label{tab:ErrWeibull100RunOutliersEx}
\end{table}
\clearpage

\section{Conclusions}
We presented in this chapter a proximal-point algorithm whose objective was the minimization of (an estimate of) a $\varphi-$divergence. The set of algorithms proposed here contains by construction the EM algorithm. We provided in several examples a proof of convergence of the EM algorithm in the spirit of our approach. We also showed how we may prove convergence for the two estimates of the $\varphi-$divergence (\ref{eqn:DivergenceDef}) and (\ref{eqn:EmpiricalNewDualForm}). We reestablished similar results to the ones in \cite{Tseng} in the context of $\varphi-$divergences, and provided a new result by relaxing the identifiability condition on the proximal term. Although our simulations do not permit to confirm the practical gain in comparison to direct methods, they are sufficient to conclude that the proximal algorithm works. The two-step algorithm (\ref{eqn:DivergenceAlgoSimp1},\ref{eqn:DivergenceAlgoSimp2}) showed only slight deterioration in performance comparing to the original one (\ref{eqn:DivergenceAlgo}) which is very encouraging especially that the dimension of the optimization is reduced at each step. Simulations have shown again the robustness of $\varphi-$divergences against outliers in comparison to the MLE.

%%%%%%%%%%%%%%%%%%%%%%%%%%%%%%%%%%%%%%%%%%%%%%%%%%%%%%%%%%%%%%%%%%%%%%%%%%%%
%
%--------------------------------------------------------------------------------------
% =======================================================================================
%--------------------------------------------------------------------------------------
%
%%%%%%%%%%%%%%%%%%%%%%%%%%%%%%%%%%%%%%%%%%%%%%%%%%%%%%%%%%%%%%%%%%%%%%%%%%%%

\section{Appendix: Proofs}
\subsection{Proof of Proposition \ref{prop:DecreaseDphi}}\label{AppendProximal:Prop1}
\begin{proof}
\underline{We prove $(a)$}. \textbf{For the first algorithm} defined by (\ref{eqn:DivergenceAlgo}), we have by definition of the arginf:
\[\hat{D}_{\varphi}(p_{\phi^{k+1}},p_{\phi_T}) + D_{\psi}(\phi^{k+1},\phi^k) \leq \hat{D}_{\varphi}(p_{\phi^k},p_{\phi_T}) + D_{\psi}(\phi^k,\phi^k).\]
We use the fact that $D_{\psi}(\phi^k,\phi^k)=0$ for the right hand and that $D_{\psi}(\phi^{k+1},\phi^k)\geq 0$ for the left hand side of the previous inequality. Hence $\hat{D}_{\varphi}(p_{\phi^{k+1}},p_{\phi_T})\leq \hat{D}_{\varphi}(p_{\phi^k},p_{\phi_T})$.\\
\textbf{For the simplified algorithm} defined by (\ref{eqn:DivergenceAlgoSimp1}, \ref{eqn:DivergenceAlgoSimp2}), recurrence (\ref{eqn:DivergenceAlgoSimp1}) and the definition of the arginf give:
\begin{eqnarray}
\hat{D}_{\varphi}(p_{\lambda^{k+1},\theta^k},p_{\phi_T}) + D_{\psi}((\lambda^{k+1},\theta^k),\phi^k) & \leq & \hat{D}_{\varphi}(p_{\lambda^k,\theta^k},p_{\phi_T}) + D_{\psi}((\lambda^k,\theta^k),\phi^k) \nonumber\\
					& \leq & \hat{D}_{\varphi}(p_{\lambda^k,\theta^k},p_{\phi_T}).
\label{eqn:PartialDecrease1}					
\end{eqnarray}
The second inequality is obtained using the fact that $D_{\psi}(\phi,\phi)=0$. Using recurrence (\ref{eqn:DivergenceAlgoSimp2}), we get:
\begin{eqnarray}
\hat{D}_{\varphi}(p_{\lambda^{k+1},\theta^k},p_{\phi_T}) + D_{\psi}((\lambda^{k+1},\theta^k),\phi^k) & \geq & \hat{D}_{\varphi}(p_{\lambda^{k+1},\theta^{k+1}},p_{\phi_T}) + D_{\psi}((\lambda^{k+1},\theta^{k+1}),\phi^k) \\
  					& \geq & \hat{D}_{\varphi}(p_{\lambda^{k+1},\theta^{k+1}},p_{\phi_T}).
\label{eqn:PartialDecrease2}
\end{eqnarray}
The second inequality is obtained using the fact that $D(\phi|\phi')\geq 0$. The conclusion is reached by combining the two inequalities (\ref{eqn:PartialDecrease1}) and (\ref{eqn:PartialDecrease2}).\\
%%%%%%%%%%%%%%%%%%%%%%%%%%%%%%%%%%%%%%
\underline{We prove $(b)$}. Using the decreasing property previously proved in (a), we have by recurrence $\forall k, \hat{D}_{\varphi}(p_{\phi^{k+1}},p_{\phi_T})\leq \hat{D}_{\varphi}(p_{\phi^k},p_{\phi_T})\leq\cdots \leq \hat{D}_{\varphi}(p_{\phi^0},p_{\phi_T})$. The result follows for both algorithms directly by definition of $\Phi^0$.\\
%%%%%%%%%%%%%%%%%%%%%%%%%%%%%%%%%%%%%%%
\underline{We prove $(c)$}. By induction on $k$. For $k=0$, clearly $\phi^0 = (\lambda^0,\theta^0)$ is well defined (a choice we make\footnote{The choice of the initial point of the sequence may influence the convergence of the sequence. See the example of the Gaussian mixture in paragraph (\ref{Example:GaussMix}).}). Suppose for some $k\geq 0$ that $\phi^k = (\lambda^k,\theta^k)$ exists. \textbf{For the first algorithm} defined by (\ref{eqn:DivergenceAlgo}), we prove that the infimum is attained in $\Phi^0$. Let $\phi\in\Phi$ be any vector at which the value of the optimized function has a value less than its value at $\phi^k$, i.e. $\hat{D}_{\varphi}(p_{\phi},p_{\phi_T}) + D_{\psi}(\phi,\phi^k)\leq \hat{D}_{\varphi}(p_{\phi^k},p_{\phi_T}) + D_{\psi}(\phi^k,\phi^k)$. We have:
\begin{eqnarray*}
\hat{D}_{\varphi}(p_{\phi},p_{\phi_T}) & \leq & \hat{D}_{\varphi}(p_{\phi},p_{\phi_T}) + D_{\psi}(\phi,\phi^k) \\ 
& \leq & \hat{D}_{\varphi}(p_{\phi^k},p_{\phi_T}) + D_{\psi}(\phi^k,\phi^k) \\
  & \leq & \hat{D}_{\varphi}(p_{\phi^k},p_{\phi_T}) \\
  & \leq & \hat{D}_{\varphi}(p_{\phi^0},p_{\phi_T}).
\end{eqnarray*}
The first line follows from the non negativity of $D_{\psi}$. As $\hat{D}_{\varphi}(p_{\phi},p_{\phi_T})\leq \hat{D}_{\varphi}(p_{\phi^0},p_{\phi_T})$, then $\phi\in\Phi^0$. Thus, the infimum can be calculated for vectors in $\Phi^0$ instead of $\Phi$. Since $\Phi^0$ is compact and the optimized function is lower semicontinuous (the sum of two lower semicontinuous functions), then  the infimum exists and is attained in $\Phi^0$. We may now define $\phi^{k+1}$ to be a vector whose corresponding value is equal to the infimum.\\
\textbf{For the second algorithm} defined by (\ref{eqn:DivergenceAlgoSimp1},\ref{eqn:DivergenceAlgoSimp2}). Similarly, the infimum in (\ref{eqn:DivergenceAlgoSimp1}) can be calculated on $\lambda$'s such that $(\lambda,\theta^k)\in\Phi^0$. Indeed, suppose there exists a $\lambda$ at which the value of the optimized function is less than its value at $\lambda^k$, i.e. $\hat{D}_{\varphi}(p_{\lambda,\theta^k},p_{\phi_T}) + D_{\psi}((\lambda,\theta^k),\phi^k) \leq \hat{D}_{\varphi}(p_{\lambda^k,\theta^k},p_{\phi_T}) + D_{\psi}((\lambda^k,\theta^k),\phi^k)$. We have:
\begin{eqnarray*}
\hat{D}_{\varphi}(p_{\lambda,\theta^k},p_{\phi_T}) & \leq & \hat{D}_{\varphi}(p_{\lambda,\theta^k},p_{\phi_T}) + D_{\psi}((\lambda,\theta^k),\phi^k) \\
& \leq & \hat{D}_{\varphi}(p_{\lambda^k,\theta^k},p_{\phi_T}) + D_{\psi}((\lambda^k,\theta^k),\phi^k) \\
  & \leq & \hat{D}_{\varphi}(p_{\lambda^k,\theta^k},p_{\phi_T}) \\
  & \leq & \hat{D}_{\varphi}(p_{\phi^0},p_{\phi_T}).
\end{eqnarray*}
This means that $(\lambda,\theta^k)\in\Phi^0$ and that the infimum needs not to be calculated for all values of $\lambda$ in $\Phi$, and can be restrained onto values which verify $(\lambda,\theta^k)\in\Phi^0$.\\
Define now $\Lambda_k = \{\lambda\in [0,1]^s| (\lambda,\theta^k)\in\Phi^0\}$. First of all, $\lambda^k\in\Lambda_k$ since $(\lambda^k,\theta^k)\in\Phi^0$. Therefore, $\Lambda_k$ is not empty. Moreover, it is compact. Indeed, let $(\lambda^l)_l$ be a sequence of elements of $\Lambda_k$, then the sequence $((\lambda^l,\theta^k))_l$ is a sequence of elements of $\Phi^0$. By compactness of $\Phi^0$, there exists a subsequence which converges in $\Phi^0$ to an element of the form $(\lambda^{\infty},\theta^k)$ which clearly belongs to $\Lambda_k$. This proves that $\Lambda_k$ is compact. Finally, since by assumption A0, the optimized function is lower semicontinuous so that it attains its infimum on the compact set $\Lambda_k$. We may now define $\lambda^{k+1}$ as any vector verifying this infimum.\\
The second part of the proof treats the definition of $\theta^{k+1}$. Let $\theta$ be any vector such that $(\lambda^{k+1},\theta)\in\Phi$ and at which the value of the optimized function in (\ref{eqn:DivergenceAlgoSimp2}) is less than its value at $\phi^k$. We have
\begin{eqnarray*}
\hat{D}_{\varphi}(p_{\lambda^{k+1},\theta},p_{\phi_T}) & \leq & \hat{D}_{\varphi}(p_{\lambda^{k+1},\theta},p_{\phi_T}) + D_{\psi}((\lambda^{k+1},\theta),\phi^k) \\
& \leq & \hat{D}_{\varphi}(p_{\lambda^{k+1},\theta^k},p_{\phi_T}) + D_{\psi}((\lambda^{k+1},\theta^k),\phi^k) \\
& \leq & \hat{D}_{\varphi}(p_{\lambda^k,\theta^k},p_{\phi_T}) + D_{\psi}((\lambda^k,\theta^k),\phi^k) \\
& \leq & \hat{D}_{\varphi}(p_{\lambda^k,\theta^k},p_{\phi_T}) \\
& \leq & \hat{D}_{\varphi}(p_{\phi^0},p_{\phi_T})
\end{eqnarray*}
The third line comes from the previous definition of $\lambda^{k+1}$ as an infimum of (\ref{eqn:DivergenceAlgoSimp1}). This means that $(\lambda^{k+1},\theta)\in\Phi^0$, and that the infimum in (\ref{eqn:DivergenceAlgoSimp2}) can be calculated with respect to values $\theta$ which verifies $(\theta,\lambda^{k+1})\in\Phi^0$. Define now $\Theta_k = \{\theta\in\mathbb{R}^{d-s}| (\lambda^{k+1},\theta)\in\Phi^0\}$. One can prove analogously to $\Lambda_k$, that it is compact. The optimized function in (\ref{eqn:DivergenceAlgoSimp2}) is, by assumption A0, lower semicontinuous so that its infimum is attained on the compact $\Theta_k$. We may now define $\theta^{k+1}$ as any vector verifying this infimum.\\
Convergence of the sequence $(\hat{D}_{\varphi}(p_{\phi^k},p_{\phi_T}))_k$ in both algorithms comes from the fact that it is non increasing and bounded. It is non increasing by virtue of (a). Boundedness comes from the lower semicontinuity of $\phi\mapsto\hat{D}_{\varphi}(p_{\phi},p_{\phi_T})$. Indeed, $\forall k, \hat{D}_{\varphi}(p_{\phi^k},p_{\phi_T}) \geq \inf_{\phi\in\Phi^0}\hat{D}_{\varphi}(p_{\phi},p_{\phi_T})$. The infimum of a proper lower semicontinuous function on a compact set exists and is attained on this set. Hence, the quantity $\inf_{\phi\in\Phi^0}\hat{D}_{\varphi}(p_{\phi},p_{\phi_T})$ exists and is finite. This ends the proof.
\end{proof}

%%%%%%%%%%%%%%%%%%%%%%%%%%%%%%%%%%%%%%%%%%%%%%%%%%%%%%%%%%%%%%ù
%%%%%%%%%%%%%%%%%%%%%%%%%%%%%%%%%%%%%%%%%%%%%%%%%%%%%%%%%%%%%%ù
\subsection{Proof of Proposition \ref{prop:StationaryPhiDiff}}\label{AppendProximal:Prop2}
\begin{proof} 
\underline{We prove $(a)$}. Let $(\phi^{n_k})_k$ be a convergent subsequence of $(\phi^k)_k$ which converges to $\phi^{\infty}$. First, $\phi^{\infty} \in \Phi^0$, because $\Phi^0$ is closed and the subsequence $(\phi^{n_k})$ is a sequence of elements of $\Phi^0$ (proved in Proposition \ref{prop:DecreaseDphi}.b).\\
Let's show now that the subsequence $(\phi^{n_k+1})$ also converges to $\phi^{\infty}$. We simply have:
\begin{eqnarray*}
\|\phi^{n_k+1} - \phi^{\infty}\| & \leq & \|\phi^{n_k} - \phi^{\infty}\| + \|\phi^{n_k+1} - \phi^{n_k}\|
\end{eqnarray*}
Since $\phi^{k+1} - \phi^k \rightarrow 0$ and $\phi^{n_k} \rightarrow \phi^{\infty}$, we conclude that $\phi^{n_k+1} \rightarrow \phi^{\infty}$.\\
\textbf{Let's start with the first algorithm (\ref{eqn:DivergenceAlgo})}. By definition of $\phi^{n_k+1}$, it verifies the infimum in recurrence (\ref{eqn:DivergenceAlgo}), so that the gradient of the optimized function is zero:
\[\nabla \hat{D}_{\varphi}(p_{\phi^{n_k+1}},p_{\phi_T}) + \nabla D_{\psi}(\phi^{n_k+1},\phi^{n_k}) = 0\]
Using the continuity assumptions A1 and AC of the gradients, one can pass to the limit with no problem:
\[\nabla \hat{D}_{\varphi}(p_{\phi^{\infty}},p_{\phi_T}) + \nabla D_{\psi}(\phi^{\infty},\phi^{\infty}) = 0\]
However, the gradient $\nabla D_{\psi}(\phi^{\infty},\phi^{\infty})=0$ because (recall that $\psi'(1)=0$):
\[\nabla D_{\psi}(\phi,\phi) = \sum_{i=1}^n\int_{\mathcal{X}}{\frac{\nabla h_i(x|\phi)}{h_i(x|\phi)}\psi'\left(\frac{h_i(x|\phi)}{h_i(x|\phi)}\right)h_i(x|\phi)dx}=\sum_{i=1}^n\int_{\mathcal{X}}{\nabla h_i(x|\phi)\psi'(1)dx}\]
Hence the gradient $\nabla D_{\psi}(\phi,\phi)=0$. This implies that $\nabla \hat{D}_{\varphi}(p_{\phi^{\infty}},p_{\phi_T})=0$.\\
\textbf{For the second algorithm (\ref{eqn:DivergenceAlgoSimp1},\ref{eqn:DivergenceAlgoSimp2})}, by definition of $\lambda^{n_k+1}$ and $\theta^{n_k+1}$, they verify the infimum respectively in recurrences (\ref{eqn:DivergenceAlgoSimp1}) and (\ref{eqn:DivergenceAlgoSimp2}). Therefore, the gradient of the optimized function is zero for each step. In other words:
\begin{eqnarray*}
\nabla_{\lambda} \hat{D}_{\varphi}(p_{\lambda^{n_k+1},\theta^{n_k}},p_{\phi_T}) + \nabla_{\lambda} D_{\psi}((\lambda^{n_k+1},\theta^{n_k}),\phi^{n_k}) & = & 0 \\
\nabla_{\theta} \hat{D}_{\varphi}(p_{\lambda^{n_k+1},\theta^{n_k+1}},p_{\phi_T}) + \nabla_{\theta} D_{\psi}((\lambda^{n_k+1},\theta^{n_k+1}),\phi^{n_k}) & = & 0
\end{eqnarray*}
Since both $(\phi^{n_k+1})$ and $(\phi^{n_k})$ converge to the same limit $\phi^{\infty}$, then setting $\phi^{\infty} = (\lambda^{\infty},\theta^{\infty})$, we get $\lambda^{n_k+1}$ and $\lambda^{n_k}$ tends to $\lambda^{\infty}$. We also have $\theta^{n_k+1}$ and $\theta^{n_k}$ tends to $\theta^{\infty}$. The continuity of the two gradients (assumptions A1 and AC) implies that:
\begin{eqnarray*}
\nabla_{\lambda} \hat{D}_{\varphi}(p_{\lambda^{\infty},\theta^{\infty}},p_{\phi_T}) + \nabla_{\lambda} D_{\psi}((\lambda^{\infty},\theta^{\infty}),\phi^{\infty}) & = & 0 \\
\nabla_{\theta} \hat{D}_{\varphi}(p_{\lambda^{\infty},\theta^{\infty}},p_{\phi_T}) + \nabla_{\theta} D_{\psi}((\lambda^{\infty},\theta^{\infty}),\phi^{\infty}) & = & 0
\end{eqnarray*}
However, $\nabla D_{\psi}(\phi,\phi) = 0$, so that $\nabla_{\lambda} \hat{D}_{\varphi}(p_{\phi^{\infty}},p_{\phi_T})=0$ and $\nabla_{\theta} \hat{D}_{\varphi}(p_{\phi^{\infty}},p_{\phi_T})=0$. Hence $\nabla \hat{D}_{\varphi}(p_{\phi^{\infty}},p_{\phi_T})=0$.\\
%%%%%%%%%%%%%%%%%%%%%%%%%%%%%%%%%%%%%%%%%%%%%%%%%%%%%%%
\underline{We prove (b)}. For the first algorithm, we use again the definition of the arginf. As the optimized function is not necessarily differentiable at the points of the sequence $\phi^k$, a necessary condition for $\phi^{k+1}$ to be an infimum is that 0 belongs to the subgradient of the function on $\phi^{k+1}$. Since $D_{\psi}(\phi,\phi^k)$ is assumed to be differentiable, the optimality condition is translated into:
\[-\nabla D_{\psi}(\phi^{k+1},\phi^k) \in \partial \hat{D}_{\varphi}(p_{\phi^{k+1}},p_{\phi_T})\quad \forall k\]
Since $\hat{D}_{\varphi}(p_{\phi},p_{\phi_T})$ is continuous, then its subgradient is outer semicontinuous (see \citep{Rockafellar} Chap 8, proposition 7). We use the same arguments presented in (a) to conclude the existence of two subsequences $(\phi^{n_k})_k$ and $(\phi^{n_k+1})_k$ which converge to the same limit $\phi^{\infty}$. By definition of outer semicontinuity, and since $\phi^{n_k+1}\rightarrow\phi^{\infty}$, we have:
\begin{equation}
\limsup_{\phi^{n_k+1}\rightarrow\phi^{\infty}} \partial \hat{D}_{\varphi}(p_{\phi^{n_k+1}},p_{\phi_T})\subset \partial \hat{D}_{\varphi}(p_{\phi^{\infty}},p_{\phi_T})
\label{eqn:subgradInc}
\end{equation}
We want to prove that $0\in\limsup_{\phi^{n_k+1}\rightarrow\phi^{\infty}} \partial \hat{D}_{\varphi}(p_{\phi^{n_k+1}},p_{\phi_T})$. By definition of limsup\footnote{We use here the definition corresponding to the outer limit, see \citep{Rockafellar} Chap 4, definition 1 or Chap 5-B.}:
\[\limsup_{\phi\rightarrow\phi^{\infty}} \partial \hat{D}_{\varphi}(p_{\phi},p_{\phi_T}) = \left\{ u|\exists \phi^k\rightarrow\phi^{\infty},\exists u^k\rightarrow u \text{ with } u^k\in \partial \hat{D}_{\varphi}(p_{\phi^k},p_{\phi_T})\right\}\]
In our scenario, $\phi = \phi^{n_k+1}$, $\phi^k = \phi^{n_k+1}$, $u = 0$ and $u^k = \nabla_1 D_{\psi}(\phi^{n_k+1},\phi^{n_k})$. The continuity of $\nabla_1 D_{\psi}$ with respect to both arguments and the fact that the two subsequences $\phi^{n_k+1}$ and $\phi^{n_k}$ converge to the same limit, imply that $u^k\rightarrow\nabla_1 D_{\psi}(\phi^{\infty},\phi^{\infty}) = 0$. Hence $u = 0\in\limsup_{\phi^{n_k+1}\rightarrow\phi^{\infty}} \partial \hat{D}_{\varphi}(p_{\phi^{n_k+1}},p_{\phi_T})$. By inclusion (\ref{eqn:subgradInc}), we get our result:
\[0\in \partial \hat{D}_{\varphi}(p_{\phi^{\infty}},p_{\phi_T})\]
\end{proof}

%%%%%%%%%%%%%%%%%%%%%%%%%%%%%%%%%%%%%%%%%%%%%%%%%%%%%%%%%%%%%%%%%%
%%%%%%%%%%%%%%%%%%%%%%%%%%%%%%%%%%%%%%%%%%%%%%%%%%%%%%%%%%%%%%%%%%
\subsection{Proof of Proposition \ref{prop:PhiDiffConverge}}\label{AppendProximal:Prop3}
\begin{proof}
The arguments presented are the same for both algorithms (\ref{eqn:DivergenceAlgo}) and (\ref{eqn:DivergenceAlgoSimp1},\ref{eqn:DivergenceAlgoSimp2}). By contradiction, let's suppose that $\phi^{k+1}-\phi^k$ does not converge to 0. There exists a subsequence such that $\|\phi^{N_0(k)+1}-\phi^{N_0(k)}\| > \varepsilon,\; \forall k\geq k_0$. Since $(\phi^k)_k$ belongs to the compact set $\Phi^0$, there exists a convergent subsequence $(\phi^{N_1\circ N_0(k)})_k$ such that $\phi^{N_1\circ N_0(k)}\rightarrow \bar{\phi}$. The sequence $(\phi^{N_1\circ N_0(k)+1})_k$ belongs to the compact set $\Phi^0$, therefore we can extract a further subsequence $(\phi^{N_2\circ N_1\circ N_0(k)+1})_k$ such that $\phi^{N_2\circ N_1\circ N_0(k)+1}\rightarrow \tilde{\phi}$. Besides $\hat{\phi}\neq \tilde{\phi}$. Finally since the sequence $(\phi^{N_1\circ N_0(k)})_k$ is convergent, a further subsequence also converges to the same limit $\bar{\phi}$. We have proved the existence of a subsequence of $(\phi^k)_k$ such that $\phi^{N(k)+1}-\phi^{N(k)}$ does not converge to 0 and such that $\phi^{N(k)+1} \rightarrow \tilde{\phi}$, $\phi^{N(k)} \rightarrow \bar{\phi}$ with $\bar{\phi} \neq \tilde{\phi}$.\\
The real sequence $\hat{D}_{\varphi}(p_{\phi^k},p_{\phi_T})_k$ converges as proved in Proposition \ref{prop:DecreaseDphi}-c. As a result, both sequences $\hat{D}_{\varphi}(p_{\phi^{N(k)+1}},p_{\phi_T})$ and $\hat{D}_{\varphi}(p_{\phi^{N(k)}},p_{\phi_T})$ converge to the same limit being subsequences of the same convergent sequence. In the proof of Proposition \ref{prop:DecreaseDphi}, we can deduce the following inequality:
\begin{equation}
\hat{D}(p_{\lambda^{k+1},\theta^{k+1}},p_{\phi_T}) + D_{\psi}((\lambda^{k+1},\theta^{k+1}),\phi^k) \leq \hat{D}(p_{\lambda^k,\theta^k},p_{\phi_T})
\label{eqn:DivergenceDecreaseSeq}
\end{equation}
which is also verified to any substitution of $k$ by $N(k)$. By passing to the limit on k, we get $D_{\psi}(\tilde{\phi},\bar{\phi}) \leq 0$. However, the distance-like function $D_{\psi}$ is positive, so that it becomes zero. Using assumption A3, $D_{\psi}(\tilde{\phi},\bar{\phi}) = 0$ implies that $\tilde{\phi} = \bar{\phi}$. This contradicts the hypothesis that $\phi^{k+1}-\phi^k$ does not converge to 0.\\
The second part of the proposition is a direct result of Proposition \ref{prop:StationaryPhiDiff}.
\end{proof}
%%%%%%%%%%%%%%%%%%%%%%%%%%%%%%%%%%%%%%%%%%%%%%%%%%%%%%%%%%%%%%%%%%%%%%%%%%%%
%%%%%%%%%%%%%%%%%%%%%%%%%%%%%%%%%%%%%%%%%%%%%%%%%%%%%%%%%%%%%%%%%%%%%%%%%%%%
\subsection{Proof of Corollary \ref{Cor:TotalConverg}}\label{AppendProximal:Cor}
\begin{proof}
Since the sequence $(\phi)_k$ is bounded and verifies $\phi^{k+1}-\phi^k\rightarrow 0$, then Theorem 28.1 in \citep{Ostrowski} implies that the set of accumulation points of $(\phi^k)_k$ is a connected compact set. It is not empty since $\Phi^0$ is compact. Let $\phi^{\infty}$ be a limit point of $(\phi^k)_k$. The assumption about strict convexity of $\hat{D}(p_{\phi},p_{\phi_T})$ in a neighborhood of $\phi^{\infty}$ implies that it is isolated in the sense that if there are another limit point $\tilde{\phi}$, then there is $\varepsilon>0$ such that $\|\phi^{\infty} - \tilde{\phi}\|>\varepsilon$. Hence, the set of accumulation points can be written as the union of at least two disjoint open sets which contradicts the connectedness property. Thus, $\phi^{\infty}$ is the only limit point of the sequence $(\phi^k)$. To end the proof, we need to show that the whole sequence converge. By contradiction, if it does not converge, there exists then $\varepsilon>0$ and an infinity of terms which verifies $\|\phi^{N_0(k)} - \phi^{\infty}\|>\varepsilon$. By compactness of $\Phi^0$, one may extract a subsequence of $(\phi^{N_0(k)})_k$, say $(\phi^{N_1\circ N_0(k)})_k$, which converges to some $\hat{\phi}$. Moreover, by continuity of the euclidean norm, $\|\phi^{N_1\circ N_0(k)} - \phi^{\infty}\|\rightarrow \|\hat{\phi} - \phi^{\infty}\|$. Hence $\|\hat{\phi} - \phi^{\infty}\|\geq\varepsilon$. Contradiction is reached by uniqueness of the limit point of the sequence $(\phi^k)_k$. 
\end{proof}
%%%%%%%%%%%%%%%%%%%%%%%%%%%%%%%%%%%%%%%%%%%%%%%%%%%%%%%%%%%%%%%%%%%%%%%%%%%%
%%%%%%%%%%%%%%%%%%%%%%%%%%%%%%%%%%%%%%%%%%%%%%%%%%%%%%%%%%%%%%%%%%%%%%%%%%%%
\subsection{Proof of Proposition \ref{prop:NewRes} } \label{AppendProximal:Prop4}
\begin{proof} 
If $(\phi^k)_k$ converges to, say, $\phi^{\infty}$, the result falls simply from Proposition 2.\\
If $(\phi^k)_k$ does not converge. Since $\Phi^0$ is compact and $\forall k, \phi^k\in\Phi^0$ (proved in Proposition 1), there exists a subsequence $(\phi^{N_0(k)})_k$ such that $\phi^{N_0(k)}\rightarrow\tilde{\phi}$. Let's take the subsequence $(\phi^{N_0(k)-1})_k$. This subsequence does not necessarily converge; still it is contained in the compact $\Phi^0$, so that we can extract a further subsequence $(\phi^{N_1\circ N_0(k)-1})_k$ which converges to, say, $\bar{\phi}$. Now, the subsequence $(\phi^{N_1\circ N_0(k)})_k$ converges to $\tilde{\phi}$, because it is a subsequence of $(\phi^{N_0(k)})_k$. We have proved until now the existence of two convergent subsequences $\phi^{N(k)-1}$ and $\phi^{N(k)}$ with \emph{a priori} different limits. For simplicity and without any loss of generality, we will consider these subsequences to be $\phi^k$ and $\phi^{k+1}$ respectively.\\
Conserving previous notations, suppose that $\phi^{k+1}\rightarrow \tilde{\phi}$ and $\phi^{k}\rightarrow \bar{\phi}$. We use again inequality (\ref{eqn:DivergenceDecreaseSeq}):
\[\hat{D}(p_{\phi^{k+1}},p_{\phi_T}) + D_{\psi}(\phi^{k+1},\phi^k) \leq \hat{D}(p_{\lambda^k,\theta^k},p_{\phi_T})\]
By taking the limits of the two parts of the inequality as $k$ tends to infinity, and using the continuity of the two functions, we have 
\[\hat{D}(p_{\tilde{\phi}},p_{\phi_T}) + D_{\psi}(\tilde{\phi},\bar{\phi}) \leq \hat{D}(p_{\bar{\phi}},p_{\phi_T})\]
Recall that under A1-2, the sequence $\left(\hat{D}_{\varphi}(p_{\phi^k},p_{\phi_T})\right)_k$ converges, so that it has the same limit for any subsequence, i.e. $\hat{D}(p_{\tilde{\phi}},p_{\phi_T}) = \hat{D}(p_{\bar{\phi}},p_{\phi_T})$. We also use the fact that the distance-like function $D_{\psi}$ is non negative to deduce that $D_{\psi}(\tilde{\phi},\bar{\phi}) = 0$. Looking closely at the definition of this divergence (\ref{eqn:DivergenceClasses}), we get that if the sum is zero, then each term is also zero since all terms are non negative. This means that:
\[\forall i\in\{1,\cdots,n\}, \quad \int_{\mathcal{X}}{\psi\left(\frac{h_i(x|\tilde{\phi})}{h_i(x|\bar{\phi})}\right)h_i(x|\bar{\phi})dx} = 0\]
The integrands are non negative functions, so they vanish almost ever where with respect to the measure $dx$ defined on the space of labels.
\[\forall i\in\{1,\cdots,n\}, \quad \psi\left(\frac{h_i(x|\tilde{\phi})}{h_i(x|\bar{\phi})}\right)h_i(x|\bar{\phi}) = 0\quad dx-a.e.\]
The conditional densities $h_i$ are supposed to be positive\footnote{In the case of two Gaussian (or more generally exponential) components, this is justified by virtue of a suitable choice of the initial condition.}, i.e. $ h_i(x|\bar{\phi})>0, dx-a.e.$. Hence, $\psi\left(\frac{h_i(x|\tilde{\phi})}{h_i(x|\bar{\phi})}\right) = 0, dx-a.e.$. On the other hand, $\psi$ is chosen in a way that $\psi(z)=0$ iff $z=1$, therefore :
\begin{equation}
\forall i\in\{1,\cdots,n\},\quad h_i(x|\tilde{\phi}) = h_i(x|\bar{\phi}) \quad dx-a.e.
\label{eqn:ProportionsEquality}
\end{equation}
Since $\phi^{k+1}$ is, by definition, an infimum of $\phi\mapsto\hat{D}(p_{\phi},p_{\phi_T}) + D_{\psi}(\phi,\phi^k)$, then the gradient of this function is zero on $\phi^{k+1}$. It results that:
\[\nabla \hat{D}(p_{\phi^{k+1}},p_{\phi_T}) + \nabla D_{\psi}(\phi^{k+1},\phi^k) = 0,\quad \forall k\]
Taking the limit on $k$, and using the continuity of the derivatives, we get that:
\begin{equation}
\nabla \hat{D}(p_{\tilde{\phi}},p_{\phi_T}) + \nabla D_{\psi}(\tilde{\phi},\bar{\phi}) = 0
\label{eqn:GradientLimit}
\end{equation}
Let's write explicitly the gradient of the second divergence:
\[\nabla D_{\psi}(\tilde{\phi},\bar{\phi}) = \sum_{i=1}^n\int_{\mathcal{X}}{\frac{\nabla h_i(x|\tilde{\phi})}{h_i(x|\bar{\phi})}\psi'\left(\frac{h_i(x|\tilde{\phi})}{h_i(x|\bar{\phi})}\right)h_i(x|\bar{\phi})}\]
We use now the identities (\ref{eqn:ProportionsEquality}), and the fact that $\psi'(1)=0$, to deduce that:
\[\nabla D_{\psi}(\tilde{\phi},\bar{\phi}) = 0 \]
This entails using (\ref{eqn:GradientLimit}) that $\nabla \hat{D}(p_{\tilde{\phi}},p_{\phi_T}) = 0$.\\ 
Comparing the proved result with the notation considered at the beginning of the proof, we have proved that the limit of the subsequence $(\phi^{N_1\circ N_0(k)})_k$ is a stationary point of the objective function. Therefore, The final step is to deduce the same result on the original convergent subsequence $(\phi^{N_0(k)})_k$. This is simply due to the fact that $(\phi^{N_1\circ N_0(k)})_k$ is a subsequence of the convergent sequence $(\phi^{N_0(k)})_k$, hence they have the same limit. \\
\textbf{When assumption AC is dropped,} similar arguments to those used in the proof of Proposition 2-b. are employed. The optimality condition in (\ref{eqn:DivergenceAlgo}) implies :
\[-\nabla D_{\psi}(\phi^{k+1},\phi^k) \in \partial \hat{D}_{\varphi}(p_{\phi^{k+1}},p_{\phi_T})\quad \forall k\]
Function $\phi\mapsto\hat{D}_{\varphi}(p_{\phi},p_{\phi_T})$ is continuous, hence its subgradient is outer semicontinuous and:
\begin{equation}
\limsup_{\phi^{k+1}\rightarrow\phi^{\infty}} \partial \hat{D}_{\varphi}(p_{\phi^{k+1}},p_{\phi_T})\subset \partial \hat{D}_{\varphi}(p_{\tilde{\phi}},p_{\phi_T})
\label{eqn:OSCInclusion}
\end{equation}
By definition of limsup:
\[\limsup_{\phi\rightarrow\phi^{\infty}} \partial \hat{D}_{\varphi}(p_{\phi},p_{\phi_T}) = \left\{ u|\exists \phi^k\rightarrow\phi^{\infty},\exists u^k\rightarrow u \text{ with } u^k\in \partial \hat{D}_{\varphi}(p_{\phi^k},p_{\phi_T})\right\}\]
In our scenario, $\phi = \phi^{k+1}$, $\phi^k = \phi^{k+1}$, $u = 0$ and $u^k = \nabla_1 D_{\psi}(\phi^{k+1},\phi^{k})$. We have proved above in this proof that $\nabla_1 D_{\psi}(\tilde{\phi},\bar{\phi}) = 0$ using only convergence of $(\hat{D}_{\varphi}(p_{\phi^k},p_{\phi_T}))_k$, inequality (\ref{eqn:DivergenceDecreaseSeq}) and some properties of $D_{\psi}$. Assumption AC was not needed. Hence, $u^k\rightarrow 0$. This proves that, $u = 0\in\limsup_{\phi^{k+1}\rightarrow\phi^{\infty}} \partial \hat{D}_{\varphi}(p_{\phi^{n_k+1}},p_{\phi_T})$. Finally, using the inclusion (\ref{eqn:OSCInclusion}), we get our result:
\[0\in \partial \hat{D}_{\varphi}(p_{\tilde{\phi}},p_{\phi_T})\]
\end{proof}

%%%%%%%%%%%%%%%%%%%%%%%%%%%%%%%%%%%%%%%%%%%%%%%%%%%%%%%%%%%%%%%%%%%%%%%%%%%%
%%%%%%%%%%%%%%%%%%%%%%%%%%%%%%%%%%%%%%%%%%%%%%%%%%%%%%%%%%%%%%%%%%%%%%%%%%%%

\subsection{Proof of Proposition \ref{prop:NewResTwoStep}}\label{AppendProximal:Prop5}
\begin{proof}
We use the same lines from the previous proof to deduce the existence of two convergent subsequences $\phi^{N(k)-1}$ and $\phi^{N(k)}$ with \emph{a priori} different limits. For simplicity and without any loss of generality, we will consider these subsequences to be $\phi^k$ and $\phi^{k+1}$ respectively. Suppose that $\phi^k\rightarrow\bar{\phi} = (\bar{\lambda},\bar{\theta})$ and $\phi^{k+1}\rightarrow\tilde{\phi} = (\tilde{\lambda},\tilde{\theta})$.\\
We first use inequality (\ref{eqn:DivergenceDecreaseSeq}) as in the previous proposition, the convergence of the sequence $(\hat{D}_{\varphi}(p_{\lambda^k,\theta^k},p_{\phi_T}))_k$ and some basic properties of $D_{\psi}$ to deduce that:
\begin{equation}
\forall i\in\{1,\cdots,n\},\quad h_i(x|\tilde{\phi}) = h_i(x|\bar{\phi}) \quad dx-a.e.
\label{eqn:ProportionsEqualityAll}
\end{equation}
Let's calculate the gradient of the objective function with respect to $\lambda$ and $\theta$ separately at the limit of $(\phi^{k+1})_k$. By definition of $\theta^{k+1}$ as an arginf in (\ref{eqn:DivergenceAlgoSimp2}), we have:
\[\frac{\partial}{\partial \theta}\hat{D}_{\varphi}(p_{\lambda^{k+1},\theta^{k+1}},p_{\phi_T}) + \frac{\partial}{\partial \theta} D_{\psi}((\lambda^{k+1},\theta^{k+1}),\phi^k) = 0\quad \forall k\]
Using the continuity of the derivatives (Assumptions A1 and AC), we may pass to the limit inside the gradients:
\[\frac{\partial}{\partial \theta}\hat{D}_{\varphi}(p_{\tilde{\lambda},\tilde{\theta}},p_{\phi_T}) + \frac{\partial}{\partial \theta} D_{\psi}((\tilde{\lambda},\tilde{\theta}),\bar{\phi}) = 0\quad \forall k\]
As in the proof of Proposition 3, all terms in the gradient of $D_{\psi}$ depend on $\psi'\left(\frac{h_i(x|\tilde{\lambda}, \tilde{\theta})}{h_i(x|\bar{\phi})}\right)$ which is zero by virtue of (\ref{eqn:ProportionsEqualityAll}). Hence $\frac{\partial}{\partial \theta}\hat{D}_{\varphi}(p_{\tilde{\lambda},\tilde{\theta}},p_{\phi_T}) = 0$.\\
We prove now that $\frac{\partial}{\partial \lambda}\hat{D}_{\varphi}(p_{\tilde{\lambda},\tilde{\theta}},p_{\phi_T}) = 0$. This is basically ensured by recurrence (\ref{eqn:DivergenceAlgoSimp1}), identities (\ref{eqn:ProportionsEqualityAll}), assumptions A1-AC and the fact that $\psi'(1)=0$. Indeed, using recurrence (\ref{eqn:DivergenceAlgoSimp1}), $\lambda^{k+1}$ is an optimum so that the gradient of the objective function is zero:
\[\frac{\partial}{\partial\lambda} \hat{D}_{\varphi}(p_{\lambda^{k+1},\theta^k},p_{\phi_T}) + \frac{\partial}{\partial\lambda} D_{\psi}((\lambda^{k+1},\theta^k),\lambda^k,\theta^k) = 0, \quad \forall k\]
Since $\|\theta^{k+1}-\theta^k\|\rightarrow 0$, then $\bar{\theta} = \tilde{\theta}$. By passing to the limit in the previous identity and using the continuity of the derivatives, we have:
\[\frac{\partial}{\partial\lambda} \hat{D}_{\varphi}(p_{\tilde{\lambda},\bar{\theta}},p_{\phi_T}) +  \frac{\partial}{\partial\lambda} D_{\psi}((\tilde{\lambda},\tilde{\theta}),\bar{\lambda},\bar{\theta}) = 0\]
Since the derivative of $D_{\psi}$ is a sum of terms which depend all on $\psi'(\frac{h_i(x|\tilde{\lambda},\bar{\theta})}{h_i(|\bar{\lambda},\bar{\theta})})$, and using identities (\ref{eqn:ProportionsEqualityAll}), we conclude that $\psi'(\frac{h_i(|\tilde{\lambda},\bar{\theta})}{h_i(|\bar{\lambda},\bar{\theta})})=\psi'(1)=0$ and $\frac{\partial}{\partial\lambda} D_{\psi}((\tilde{\lambda},\bar{\theta}),\bar{\lambda},\bar{\theta}) = 0$. Finally, $\bar{\theta}=\tilde{\theta}$ implies that $\frac{\partial}{\partial\lambda} \hat{D}_{\varphi}(p_{\tilde{\lambda},\hat{\theta}},p_{\phi_T}) = 0$.\\
We have proved that $\frac{\partial}{\partial \lambda}\hat{D}_{\varphi}(p_{\tilde{\lambda},\tilde{\theta}},p_{\phi_T}) = 0$ and $\frac{\partial}{\partial \theta}\hat{D}_{\varphi}(p_{\tilde{\lambda},\tilde{\theta}},p_{\phi_T}) = 0$, so the gradient is zero and the stated result is proved.
\end{proof}

\part{Two-component Semiparametric Mixture Models When One Component is Unknown}
\chapter{Semiparametric two-component mixture models where one component is defined through linear constraints on its distribution function}
%\section*{Introduction}
A two-component mixture model with an unknown component is defined by:
\begin{equation}
f(x) = \lambda f_1(x|\theta) + (1-\lambda) f_0(x), \qquad \text{for } x\in\mathbb{R}^r
\label{eqn:GeneralSemiParaMix}
\end{equation}
for $\lambda\in(0,1)$ and $\theta\in\mathbb{R}^d$ to be estimated and the density $f_0$ is considered to be unknown. Such model appears in the study of gene expression data coming from microarray analysis. An application to two bovine gestation mode comparison is performed in \cite{Bordes06b}. The authors suppose that $\theta$ is known, $f_0$ is symmetric around an unknown $\mu$ and that $r=1$. \cite{Xiang} studied a more general setup by considering $\theta$ unknown and applied model (\ref{eqn:GeneralSemiParaMix}) on the Iris data by considering only the first principle component for each observed vector. Another application of model (\ref{eqn:GeneralSemiParaMix}) in genetics can be found in \cite{JunMaDiscret}. See also \cite{PatraRSSB} for applications arising in astronomy and from microarray experiment. \\ 
\cite{Robin} used the semiparametric model (supposing that $\theta$ is known) in multiple testing procedures in order to estimate the posterior population probabilities and the local false rate discovery. \cite{Song} studied a similar setup where $\theta$ is unknown without further assumptions on $f_0$. They applied the semiparametric model in sequential clustering algorithms as a second step. After a sequential clustering algorithm finds the center of a cluster, the next step is to identify the observations belonging to this cluster. If we assume that the center of the cluster is known and that the distribution of observations not belonging to the cluster is unknown, the problem of identifying observations in the cluster is similar to the problem of estimating the mixing proportion in a special two-component mixture model. The mixing proportion can be considered as the proportion of observations belonging to the cluster. Finally, model (\ref{eqn:GeneralSemiParaMix}) can also be regarded as a contamination model, see \cite{Titterington} or \cite{FinMixModMclachlan} for further applications of general mixture models. \\

Existing estimation methods for model (\ref{eqn:GeneralSemiParaMix}) were proved or illustrated to work but only in specific situations and the only simulated example was a dataset generated by a Gaussian mixture. The paper of \cite{Xiang} provides a comparison of several estimation methods for the semiparametric model. The compared methods give satisfactory results in most simulations, but no method performs uniformly good on all simulated mixtures. In all these simulations the authors consider $\theta$ to be given. We noticed that as we add $\theta$ to the set of unknown parameters, things become different. The performance of these methods depend on the proportion of the parametric component if it is high or low. They may have very poor performances sometimes even if the number of observations is high as we will demonstrate in the simulation section.\\
It is important for an estimation method to be applicable in contexts where the parametric component is not fully known. For example, the parametric component may be a signal whereas the unknown component is a noise. We need to extract the location of the signal or any other information concerning its shape and not only the proportion of the noise. We believe that the failure of the existing methods when the parametric component is not fully known comes from the degree of difficulty of the semiparametric model, i.e we do not possess sufficient information about the model in order to estimate it. \cite{Bordes10} considered a symmetric assumption on the unknown component (see also \cite{Bordes06b}). This gave the model a structure and permitted to improve the estimation and made the study of the asymptotic properties of the resulting estimators tractable. Moreover, they were able to give sufficient conditions under which the semiparametric model is identifiable. Nevertheless, such assumption is very restrictive and cannot be applied for example in the context of distributions defined on a subset of $\mathbb{R}$. It appears that the addition of prior information should be helpful and may lead to a more understandable theory and better estimation results. We propose to add some information about $f_0$ in a way that we stay in between a (restrictive) fully parametric settings and a (complex) fully semiparametric one.\\
In this chapter, we introduce a method which permits to add a, relatively general, prior information about the unknown component in order to decrease the degree of difficulty of the model and be able to better estimate it. Such information needs to apply linearly on the distribution function of the unknown component such as moment-type information. For example, we may have an information relating the first and the second moments of $f_0$ such as $\int{xf_0(x)}=\alpha$ and $\int{x^2f_0(x)dx} = m(\alpha)$, see \cite{BroniaKeziou12} and the references therein. Such information adds some structure to the model without precising the value of the moments. More examples will be discussed later on.\\
Unfortunately, the incorporation of linear constraints on the distribution function cannot be done directly in existing methods because the optimization\footnote{Not all existing methods, as we will see in the next paragraph, are defined through an optimization procedure. Hence, it becomes more difficult to introduce this kind of constraints inside the estimation procedure.} will be carried over a (possibly) infinite dimensional space, and we need a new approach. Convex analysis offers a way using the Fenchel-Legendre duality to transform an optimization problem over an infinite dimensional space. On the other hand, $\varphi-$divergences offer a way by their convexity properties to exploit this duality result. The paper of \cite{BroniaKeziou12} gives a complete study of this problem in the non mixture case, see also \cite{AlexisThesis} Chap. 1 and \cite{KeziouThesis} Chap. 3. We will exploit these results to build upon a new estimation procedure which takes into account linear information over the unknown component's distribution.\\
%%%%%%%%%%%%%%%%%%%%%%%%%%%%%%%%%%%%%%%%%%%%%%%%%%%%%%%%%%%%%%%%%%%%%%%%%%%%%%%%%
%
%==============================================================
%%%%%%%%%%%%%%%%%%%%%%%%%%%%%%%%%%%%%%%%%%%%%%%%%%%%%%%%%%%%%%%%%%%%%%%%%%%%
%==============================================================
%
%%%%%%%%%%%%%%%%%%%%%%%%%%%%%%%%%%%%%%%%%%%%%%%%%%%%%%%%%%%%%%%%%%%%%%%%%%%%%%%%%
\section{Semiparametric two-component mixture models in the literature}\label{sec:LiteratureSemiparaMix}
The literature on semiparametric mixture models contains several methods which permit to estimate efficiently the parameters with or without estimating the unknown component. All these methods were never tested on difficult situations except for datasets generated from a mixture of two Gaussian components with very close means (difference of means equal to 1.5) when the first component $f_1$ is \emph{fully known}. \\
We present in this section the principle estimation methods in the literature. We present the method of \cite{Bordes10} which is based on a symmetry constraint over the unknown component. We present also two EM-type algorithms introduced by \cite{Robin} and \cite{Song}, and an SEM-type algorithm developed by \cite{BordesStochEM}. There is also an interesting method developed in \cite{Song} based on the identifiability of a two-component mixture model when $f_1$ is Gaussian, called the $\pi-$maximizing algorithm. Finally, a method based on the Hellinger divergence was developed by \cite{Xiang}. However, the algorithm presented in their article is not clear and contains difficult integration calculus which cannot be calculated directly by a numerical method. The authors did not give any further explanations on how to do the calculus. We therefore prefer not to discuss it here. \\
We advise the reader to consult the simulation results in \cite{Xiang}. The article contains a comparison between some of these methods in a two-component Gaussian mixture. We will also be testing all these methods on further simulations and different models to explore their capacities in estimating the semiparametric mixture model.
%%%%%%%%%%%%%%%%%%%%%%%%%%%%%%%%%%%%%%%%%
\subsection{Semiparametric two-component mixture models under a symmetry assumption}
\cite{Bordes06b} proposed to study the semiparametric model (\ref{eqn:GeneralSemiParaMix}) when $r=1$, $\theta$ is given and the unknown component is supposed to be symmetric. Thus, the semiparametric model (\ref{eqn:GeneralSemiParaMix}) can be rewritten as:
\begin{equation}
f(x |\lambda,\theta)  =  \lambda f_1(x|\theta) + (1-\lambda) f_0(x-\mu_0), \qquad \forall x\in\mathbb{R}.
\label{eqn:SemiParaMixSymDef}
\end{equation}
\cite{Bordes06b} have studied the identifiability of this simplified model by imposing either a symmetry assumption over $f_1$, conditions over the characteristic function or conditions on the tail behavior of the components. Let us summarize their estimation procedure. Let $\mathbb{F}_0, \mathbb{F}_1$ and $\mathbb{F}$ be the cumulative distribution functions (cdf) of $f_0, f_1$ and $f$ respectively. We have:
\[\mathbb{F}_0(x) = \frac{1}{1-\lambda} \mathbb{F}(x+\mu_0|\lambda,\theta) - \frac{\lambda}{1-\lambda} \mathbb{F}_1(x+\mu_0|\theta).\]
Define the following functions:
\begin{eqnarray*}
H_1(x|\lambda,\theta,\mathbb{F}) & = & \frac{1}{1-\lambda} \mathbb{F}(x+\mu_0|\lambda,\theta) - \frac{\lambda}{1-\lambda} \mathbb{F}_1(x+\mu_0|\theta), \\
H_2(x|\lambda,\theta,\mathbb{F}) & = & 1-\frac{1}{1-\lambda} \mathbb{F}(\mu_0-x|\lambda,\theta) + \frac{\lambda}{1-\lambda} \mathbb{F}_1(\mu_0-x|\theta).
\end{eqnarray*}
By symmetry of $f_0$, we have $\mathbb{F}_0(x) = 1-\mathbb{F}_0(-x)$, and thus $H_1(x|\lambda,\theta,\mathbb{F}) = H_2(x|\lambda,\theta,\mathbb{F})$. This means that if $d$ is a distance mapping, then $d(H_1(.|\lambda,\theta,\mathbb{F}),H_2(.|\tilde{\lambda},\tilde{\theta},\mathbb{F}))=0$ if $\lambda=\tilde{\lambda}$ and $\theta=\tilde{\theta}$. Otherwise the distance will be positive if $d$ is chosen properly. In \cite{Bordes06b}, the authors propose to use an $L^q$ distance and define the estimation procedure by:
\[(\hat{\lambda},\hat{\theta}) = \arginf_{\lambda,\theta}\left(\int{\left|H_1(x|\lambda,\theta,\hat{\mathbb{F}}_n) - H_2(x|\lambda,\theta,\hat{\mathbb{F}}_n)\right|^qdx}\right)^{1/q}\]
where $\hat{\mathbb{F}}_n$ is an estimator of $\mathbb{F}$. The authors proved consistency of this estimation procedure, but were unable to prove its asymptotic normality. Besides, \cite{Bordes10} argue that the use of an $L^q$ distance leads to numerical instability. They also propose to use the alternative distance $L^2(d\mathbb{F})$. The new procedure is now defined by:
\begin{equation}
(\hat{\lambda},\hat{\theta}) = \arginf_{\lambda,\theta}\sum_{i=1}^n{\left[H_1(x_i|\lambda,\theta,\hat{\mathbb{F}}_n) - H_2(x_i|\lambda,\theta,\hat{\mathbb{F}}_n)\right]^2}.
\label{eqn:BoresSymmetryAlgo}
\end{equation}
\cite{Bordes10} prove that the above estimator is consistent and asymptotically Gaussian. This method produces in practice good estimates even in difficult situations such as two-component Gaussian mixture when the two components are close provided that we restrict the proportion parameter inside an interval of the form $(\eta,1-\eta)$ for a small $\eta$, say $0.1$. It is however unusable in the context of distributions which are defined on a subset of $\mathbb{R}$. Besides, a generalization to the multivariate case does not seem simple.\\
Notice that, in the original approach presented by \cite{Bordes06b} or \cite{Bordes10} we suppose that $\theta$ is given, so that the parametric component is fully known. We find however no problem in writing the same algorithm for the case when $\theta$ is considered unknown. In what concerns the theoretical results developed in these papers, we did not check their validity when $\theta$ is unknown. We mention the work of \cite{Maiboroda2012} who use the approach of \cite{Bordes10} and propose several methods to estimate the unknown component. They also study the theoretical properties of their estimators such as L-consistency and rates of convergence.
%%%%%%%%%%%%%%%%%%%%%%%%%%%%%%%%%%%%%%%%%
\subsection{EM-type algorithms}
This kind of algorithms is based on defining a vector of weights $w=(w_1,\cdots,w_n)$ for the observations and then estimate the unknown component using a weighted kernel estimator as follows:
\[\hat{f}_0(x|w) = \frac{1}{nh}\frac{1}{\sum{1-w_i}}\sum_{i=1}^n{(1-w_i) K\left(\frac{x-x_i}{h}\right)}.\]
The proportion is then, similarly to the EM algorithm, estimated by averaging these weights whereas the parameters of the known component are calculated by maximizing a weighted likelihood function. Such methods were proposed by several authors such as \cite{Song}, \cite{Robin} and \cite{JunMaDiscret}. They differ by how to calculate the weights. In t\cite{Robin} and \cite{Song}, we calculate the weights by an iterative procedure in the same way as the EM algorithm does, i.e. as the quotient of the probability of being in the first component to the probability of being in the mixture. In \cite{JunMaDiscret}, the authors propose a weighted histogram to estimate the unknown component where the bins and their number are chosen before the estimation procedure (prior guess). They calculate after that the weights by maximizing a likelihood-like function related to the discretized model. \\
The method of \cite{JunMaDiscret} has many drawbacks and practical issues. We need at first to precise a close interval for the values of the unknown component, then a good guess for the number of bins. Besides, if one thinks about increasing the bins in order to give a closer estimate of $f_0$, it will cost much on the optimization step as each bin has its own parameter. Besides, in multivariate situations, the number of bins explodes easily and the optimization over the unknown weights becomes very difficult.\\
Other EM-type algorithms are very simple to implement and have good execution time when the number of observations is small, however, they do not perform very well in situations when the two components are close enough.\\
We present briefly the algorithms of \cite{Song} and \cite{Robin}, and their analytical properties. 
\paragraph{\cite{Robin} EM-type algorithm.} The authors propose to estimate the weight vector $w$ by the following iterative algorithm. For some initial value of $\lambda$ and $\theta$, say $\lambda^0,\theta^0$, define at iteration $k+1$ for $k\geq 0$:
\begin{equation}
w_i^{(k+1)} = \frac{\lambda^{(k)} f_1(x_i|\theta^{(k)})}{\lambda^{(k)} f_1(x_i|\theta^{(k)})+(1-\lambda^{(k)}) \hat{f}_{0}(x_i|w^{(k)})}.
\label{eqn:RobinWeights}
\end{equation}
When the vector $\theta$ is known, the algorithm is proved to converge under mild conditions, see Theorem 1 in \cite{Robin}. We only need that the proportion to be inside the interval $(0,1)$ and that for all $j$ and $i$ the quantities $K((x_i-x_j)/h)$ are positive, which is \emph{theoretically} fulfilled as long as we are \emph{not} using a compacted-support kernel and there is no ties in the dataset\footnote{Which is simply ensured if the data is distributed from a continuous probability distribution.}. On the other hand, the algorithm is proved to converge towards a fixed point of function $\psi$ defined by:
\begin{eqnarray*}
\psi & = & (\psi_1,\cdots,\psi_n) \\
\psi_j(w_1,\cdots,w_n) & = & \frac{(1-\lambda)\sum_{i=1}^n{(1-w_i)K((x_i-x_j)/h)/h}}{(1-\lambda)\sum_{i=1}^n{(1-w_i)K((x_i-x_j)/h)/h} + \sum_{i=1}^n{w_i}f_1(x_j|\theta)},
\end{eqnarray*}
provided that $\theta$ is known. In what concerns our objective, $\theta$ will be unknown, and this theoretical result may not hold. The algorithm is still applicable. Besides, the theoretical result supposes that the proportion is known. The authors propose to estimate the proportion by $\sum{w_i}/n$ as a natural choice, but state that this can easily lead the algorithm to converge to a proportion equals to 0 or 1. They propose another estimator suitable to the application they considered (estimation of the FDR) by adding an assumption that the distribution of the unknown component is defined on a semi-closed interval $(-\infty,a)$. Besides, the cdf of the parametric component must verify $\mathbb{F}_1(a)<1$. This means that the semiparametric component distribution must have a lighter tail than the parametric component one. An R package, \texttt{kerfdr}, was written about this approach where the data is supposed to be bounded, see \cite{GuedjRobin}. In our examples and simulations, the distribution of both components is the same; either $(0,\infty)$ or the whole real line $\mathbb{R}$. Thus, we cannot adapt this methodology.\\
We adapt the following algorithm based on the ideas of \cite{Robin} as follows. Let $D$ be a kernel function and $h$ be a pre-chosen window, and denote $D_i(x)=D(\frac{x-x_i}{h})$. Initialize the algorithm with $\lambda^{(0)}$, $\theta^{(0)}$ and a vector of weights for the observations say $(w_1^{(0)},\cdots,w_n^{(0)})$. At iteration $k$ calculate
\begin{eqnarray*}
\hat{p}_0^{(k)}(y_i) & = & \frac{1}{\sum_j{w_j^{(k-1)}}}\sum_l{w_l^{(k-1)}D_i(y_l)};\\
\lambda^{(k)} & = & \frac{1}{n}\sum_l{w_l^{(k-1)}};\\
\theta^{(k)} & = & \argmax_\theta \sum_i w_i^{(k-1)} \log(p_1(y_i|\theta));\\
w_i^{(k)} & = & \frac{\lambda^{(k)} p_1(y_i|\theta^{(k)})}{\lambda^{(k)} p_1(y_i|\theta^{(k)})+(1-\lambda^{(k)}) \hat{p}_0^{(k)}(y_i)}.
\end{eqnarray*}
Repeat this iteration until convergence.\\
\paragraph{\cite{Song} EM-type algorithm.} The authors propose to estimate the weight vector without the need to estimate the unknown component and merely using an estimate of the whole mixture. The authors propose two methods with no proofs of convergence. Let $\hat{f}_n$ be an estimator of the mixture density on the basis of an $n-$sample. The first recurrence proposed by the authors is given by:
\begin{equation}
w_i^{(k+1)} = \min\left[1, \quad \frac{\lambda^{(k)} f_1(x_i|\theta^{(k)})}{\hat{f}_n} \right].
\label{eqn:SongEM1}
\end{equation}
The authors argue that this formulation does not stabilize well and propose the alternative iteration:
\begin{equation}
w_i^{(k+1)} = \min\left[1, \quad \frac{2\lambda^{(k)} f_1(x_i|\theta^{(k)})}{\lambda^{(k)} f_1(x_i|\theta^{(k)})+\hat{f}_n} \right],
\label{eqn:SongEM2}
\end{equation}
and state that this new formulation is better without any theoretical justification. It is worth noting that these algorithms were only proposed in the context of a mixture of a two Gaussian components when one component is unknown. We, however, see no problem in using them in a more general context and even in the multivariate case since no particular constraint on $f_1$ is needed in order to write the iterative procedure.
%%%%%%%%%%%%%%%%%%%%%%%%%%%%%%%%%%%%%%%%%
\subsection{Stochastic EM-type method}
\cite{BordesStochEM} have proposed a stochastic EM-type algorithm to estimate the parameters in model (\ref{eqn:SemiParaMixSymDef}), i.e. under the symetry assumption. There is however no problem in using their algorithm in the general context of model (\ref{eqn:GeneralSemiParaMix}). We give the general form for this algorithm which can also be used in a multivariate context as we will see in the simulation section.\\ 
The algorithm starts by giving an initial Bernoulli vector $Z^{(0)}$ which attributes a zero to coordinate $Z_i$ if the observation $x_i$ is drawn according to the unknown component $f_0$ and one if $x_i$ is drawn according to the parametric component $f_1(.|\theta)$. We then calculate the initial proportion $\lambda^{(0)} = \sum_{i=1}^n{Z_i^{(0)}}/n$, and give an initial value for the vector $\theta$, say $\theta^{(0)}$.\\
At iteration $k+1$, calculate the kernel density estimator of the unknown component as follows:
\[\hat{f}_0(x|Z^{(k)}) = \frac{1}{h(1-\sum{Z_i^{(k)}})}\sum_{i=1}^n{(1-Z_i^{(k)})K\left(\frac{x-x_i}{h}\right)}.\]
Calculate the weights by:
\[w_i^{(k)} = \frac{\lambda^{(k)} f_1(x_i|\theta^{(k)})}{\lambda^{(k)} f_1(x_i|\theta^{(k)}) + (1-\lambda^{(k)})\hat{f}_0(x_i|Z^{(k)})}.\]
Generate now a Bernoulli vector $Z^{(k+1)}$ with probabilities $(w_1^{(k)},\cdots,w_n^{(k)})$, and calculate now the new proportion:
\[\lambda^{(k+1)} = \frac{1}{n}\sum_{i=1}^n{Z_i^{(k+1)}}.\]
Finally, we calculate the new vector of parameters $\theta^{(k+1)}$ by maximum likelihood using the observations for which $Z^{(k+1)}_i$ is equal to 1. We repeat this procedure until the sequence of proportions $\lambda^{(k)}$ stabilizes. It is preferable in the context of the stochastic EM algorithm not to keep the final iteration, but to average the results of the $n_0$ final iterations instead. For our simulations, we performed 5000 iterations and averaged the 4000 final iterations, see figure (\ref{fig:SEMsemipara}) for a better understanding.\\
\begin{figure}[ht]
\centering
\includegraphics[scale=0.35]{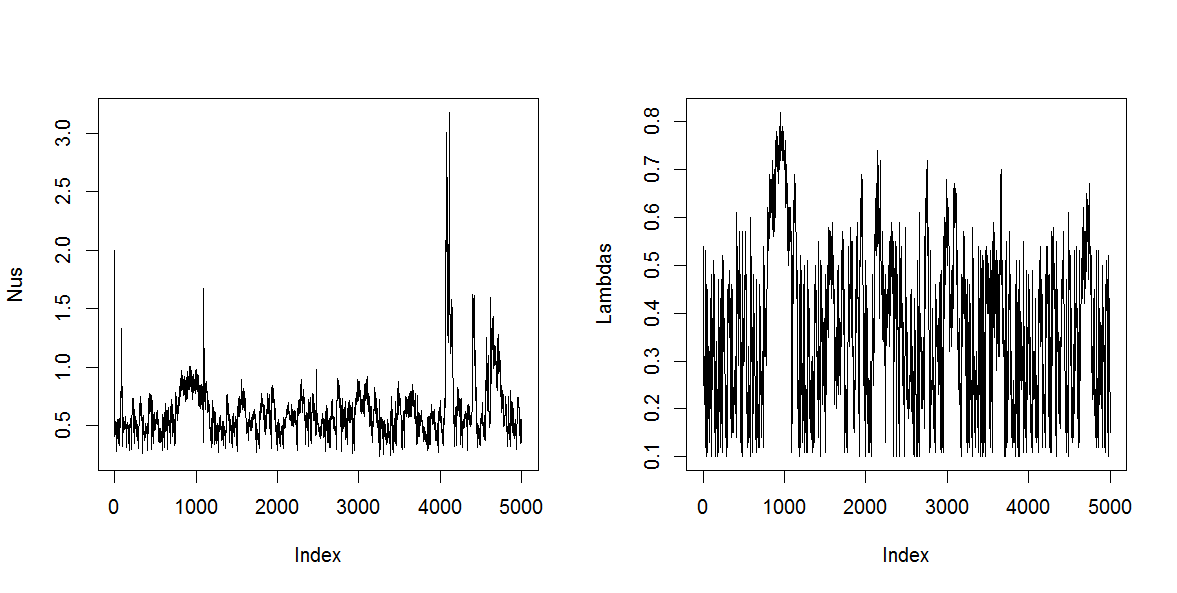}
\caption{Fluctuations in a trajectory of the semiparametric SEM algorithm in a Weibull mixture.}
\label{fig:SEMsemipara}
\end{figure}

\noindent This method has a good performance in regular situations when the two components do not overlap. Besides, the asymptotic behavior of the sequence of points which generates the algorithm remains an open problem.
%%%%%%%%%%%%%%%%%%%%%%%%%%%%%%%%%%%%%%%%%
\subsection{\texorpdfstring{$\pi-$}{pi-}maximizing method}
\cite{Song} propose another kind of algorithm which is based on the identifiability condition of a mixture of two Gaussian components with known means. They state that no estimating procedure can distinguish between two sets of parameters $(\lambda_1,\theta_1)$ and $(\lambda_2,\theta_2)$ as long as they both verify:
\[\lambda_i f_1(x|\theta_i)<f(x), \forall x,\quad \text{ for } i = 1,2.\]
This is because the unknown component whose form is not specified by any prior condition can take the form of a mixture to fill the gap between $\lambda f_1(x|\theta)$ and $f(x)$, see figure (\ref{fig:SongNonIndentifiable}).
\begin{figure}[h]
\centering
\includegraphics[scale=0.8]{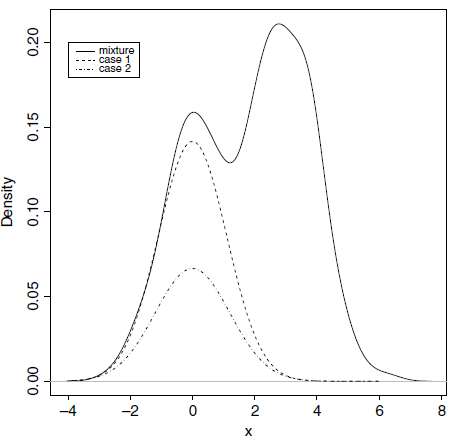}
\caption{The solid line is the density estimation of the whole mixture. The dotted lines are two examples of normal densities that fit under the mixture density. These two cases can not be distinguished by an estimating procedure. Figure copied from \cite{Song}.}
\label{fig:SongNonIndentifiable}
\end{figure}

\noindent Based on this idea, \cite{Song} propose the following estimation procedure:
\begin{eqnarray}
\hat{\lambda} & = & \sup_{\theta}\min_{x_i} \frac{\hat{f}(x_i)}{f_1(x_i|\theta)} \label{eqn:SongMaxi1}\\
\hat{\theta} & = & \argsup_{\theta}\min_{x_i} \frac{\hat{f}(x_i)}{f_1(x_i|\theta)}. \label{eqn:SongMaxi2}
\end{eqnarray}
The authors prove that if $f_0(0)=0$, $f_1$ is a centered Gaussian distribution with unknown variance $\sigma^2$ and that the tail of the mixture is not heavier than the tail of the Gaussian component, then:
\[\lambda^* = \sup_{\sigma>0}\inf_{x}\frac{f(x)}{f_1(x|\sigma)},\]
where the supremum is attained at $\sigma=\sigma^*$ the true value of the scale. This constitutes a simple theoretical justification to the estimation procedure (\ref{eqn:SongMaxi1}, \ref{eqn:SongMaxi2}). The authors do not provide, however, any \emph{real} proof of consistency or a generalization to other distributions other than the Gaussian. Notice that a proof of consistency does not seem to be simple mainly because we are dealing with a double optimization procedure.\\

%%%%%%%%%%%%%%%%%%%%%%%%%%%%%%%%%%%%%%%%%%%%%%%%%%%%%%%%%%%%%%%%%%%%%%%%%%%%%%%%%
%
%==============================================================
%%%%%%%%%%%%%%%%%%%%%%%%%%%%%%%%%%%%%%%%%%%%%%%%%%%%%%%%%%%%%%%%%%%%%%%%%%%%
%==============================================================
%
%%%%%%%%%%%%%%%%%%%%%%%%%%%%%%%%%%%%%%%%%%%%%%%%%%%%%%%%%%%%%%%%%%%%%%%%%%%%%%%%%

\section{Semiparametric models defined through linear constraints}\label{sec:GeneralModelDef}
Now that the idea of the semiparametric mixture model is presented, we proceed to propose our new model. We want to integrate linear information in the semiparametric model and propose an estimation procedure which permits to retrieve the true vector of parameters defining the model on the basis of a given i.i.d. sample $X_1,\cdots,X_n$ drawn from the mixture distribution $P_T$. \\
We prefer to proceed step by step. The previous paragraphs introduced semiparametric mixture models. Now we will present models which can be defined through a linear information. These models are not necessarily mixtures of distributions. Besides, the constraints or the linear information defining the model will apply over the whole model, i.e. if the model is a mixture then the constraints apply over the whole mixture and not only over one component. We give in this section a brief idea of what the literature offers us to study such model. In the next section we will proceed to aggregate the two ideas, i.e. mixture models and semiparametric models defined through linear constraints, in order to introduce our semiparametric mixture model where a component is parametric (but not fully known) and a component is defined through linear constraints.
\subsection{Definition and examples}
Denote by $M^{+}$ the set of all probability measures (p.m.) defined on the same measurable space as $P_T$, i.e. $(\mathbb{R}^{r},\mathscr{B}(\mathbb{R}^{r}))$.
\begin{definition}
Let $X_1,\cdots,X_n$ be random variables drawn independently from the probability distribution $P_T$. A semiparametric model is a collection of probability measures $\mathcal{M}_{\alpha}(P_T)$, for $\alpha\in\mathcal{A}\subset\mathbb{R}^s$, absolutely continuous with respect to $P_T$ which verifies a set of linear constraints, i.e.
\begin{equation}
\mathcal{M}_{\alpha}(P_T) = \left\{ Q\in M^{+} \; \; \textrm{such that} \;\; Q \ll P_T, \;\; \int g(x)dQ(x) = m(\alpha)\right\},
\label{eqn:SimpleConstraints}
\end{equation}
where $g : \mathbb{R}^{r} \rightarrow \mathbb{R}^{\ell}$ and $m:\mathcal{A}\rightarrow \mathbb{R}^{\ell}$ are specified vector-valued functions.
\end{definition}
\noindent This semiparametric model was studied by many authors, see \cite{BroniaKeziou12}, \cite{AlexisGSI13} (with $dP$ replaced by $dP^{-1}$ the quantile measure), \cite{Owen} (in the empirical likelihood context) and \cite{ChenQin} (for finite population problems). It is possible in the above definition to make $g$ depend on the parameter vector $\alpha$, but we stay for the sake of simplicity with the assumption that $g$ does not depend on $\alpha$. The theoretical approach we present in this chapter remains valid if $g$ depends on $\alpha$ with slight modification on the assumptions and more technicalities at the level of the proofs.
\begin{example}
A simple and standard example is a model defined through moment constraints. Let $P_T$ be the Weibull distribution with scale $a^*$ and shape $b^*$. We define $\mathcal{M}_{\alpha}$ with $\alpha = (a,b)\in(0,\infty)^2$ to be the set of all probability measures whose first three moments are given by:
\[\int{x^{i}dQ(x)} = a^{i}\Gamma(1+i/b), \qquad i=1,2,3.\]
The set $\mathcal{M}_{\alpha^*}$ is a "neighborhood" of probability measures of the Weibull distribution $P_T$. It contains all probability measures absolutely continuous with respect to the Weibull mixture $P_T$ and which share the first three moments with it. The union of the sets $\mathcal{M}_{\alpha}$ contains all probability measures whose first three moments share the same analytic form as a Weibull distribution.
\end{example}
If the true distribution $P_T$ verifies the set of $\ell$ constraints (\ref{eqn:SimpleConstraints}) for some $\alpha^*$, then the set 
\begin{equation}
\mathcal{M}_{\alpha^*}(P_T) = \left\lbrace Q\in M^{+} \; \; \textrm{such that} \;\; Q \ll P_T, \;\; \int g(x)dQ(x) = m(\alpha^*) \right\rbrace
\label{eqn:SetMProbaConstr}
\end{equation}
constitutes a "neighborhood" of probability measures of $P_T$. Generally, one would rather consider the larger "neighborhood" defined by
\begin{equation*}
\mathcal{M} = \bigcup_{\alpha \in \mathcal{A}} \mathcal{M}_{\alpha},
\end{equation*}
because the value of $\alpha^*$ is unknown and needs to be estimated. The estimation procedure aims at finding $\alpha^*$ the ("best") vector for which $P_T\in M_{\alpha^*}$. This is generally done by either solving the set of equations (\ref{eqn:SimpleConstraints}) defining the constraints for $Q$ replaced by (an estimate of) $P_T$ or by minimizing a suitable distance-like function between the set $\mathcal{M}$ and (an estimate of) $P_T$. In other words, we search for the "projection" of $P_T$ on $\mathcal{M}$. Solving the set of equations (\ref{eqn:SimpleConstraints}) is in general a difficult task since it is a set of nonlinear equations. In the literature, similar problems were solved using the Fenchel-Legendre duality. \cite{BroniaKeziou12} proposed to estimate the value of $\alpha^*$ using $\varphi-$divergences developing an efficient and a simple estimation method using the duality of Fenchel-Legendre. In the next chapter, we will see similar semiparametric models defined through linear constraints over the quantile measures. \cite{AlexisGSI13} proposed also to use $\varphi-$divergences and basing on the duality of Fenchel-Legendre to estimate his "semiparametric linear quantile models", see Chapter 4.\\

In the next paragraph, we present the duality technique which will be essential in the development of our estimation method.
%%%%%%%%%%%%%%%%%%%%%%%%%%%%%%%%%%%%%%%%%%%%%%%%%%%%%%%%%%%%%%
%%%%%%%%%%%%%%%%%%%%%%%%%%%%%%%%%%%%%%%%%%%%%%%%%%%%%%%%%%%%%
\subsection{Estimation using \texorpdfstring{$\varphi-$}{phi-}divergences and the duality technique}\label{subsec:DualTechRes}
As mentioned in the previous paragraph, $\varphi-$divergences offer an efficient tool to handle the projection of a probability measure on a set of probability measures. This remains also valid for finite signed measures. We will explain how we may use $\varphi-$divergences to find the "best" vector $\alpha^*$ such that $P_T\in\mathcal{M}_{\alpha^*}$. The optimality of the solution is absolute if the distribution $P_T$ verifies the constraints for some value $\alpha^*$. Otherwise, the vector $\alpha^*$ is sub-optimum in the sens that optimality is considered merely from a point of view of $\varphi-$projections (see definitions below). The following definitions concern the notion of $\varphi$-projection of finite signed measures over a set of finite signed measures and are essential in order to clearly present the estimation procedure and introduce our new estimation methodology. The context of semiparametric models presented earlier can be extended to finite signed measures, see the theory in \cite{BroniaKeziou12}. For our study, the use of finite \emph{signed} measures and not only probability measures is essential as will be demonstrated in the next section.
\begin{definition}
\label{def:phiDistance}
Let $\mathcal{M}$ be some subset of $M$, the space of finite signed measures. The $\varphi$-divergence between the set $\mathcal{M}$ and some finite signed measure $P$, noted as $D_{\varphi}(\mathcal{M},P)$, is given by
\begin{equation}
  D_{\varphi}\left(\mathcal{M},P\right)  := \inf_{Q \in \mathcal{M}}D_{\varphi}(Q,P).
	\label{eqn:PhiDistanceEleSet}
\end{equation}
Furthermore, we define the $\varphi-$divergence between two subsets of $M$, say $\mathcal{M}$ and $\mathcal{N}$ by:
\begin{equation*}
  D_{\varphi}\left(\mathcal{M},\mathcal{N}\right)  := \inf_{Q \in \mathcal{M}}\inf_{P\in\mathcal{N}}D_{\varphi}(Q,P).
\end{equation*}
\end{definition}

\begin{definition} \label{def:phiProj}
Assume that $D_{\varphi}(\mathcal{M},P)$ is finite. A measure $Q^* \in \mathcal{M}$ such that
\begin{equation*}
D_{\varphi}(Q^*,P) \leq D_{\varphi}(Q,P), \;\; \textrm{for all}\;\; Q \in \mathcal{M}
\end{equation*}
is called a $\varphi$-projection of $P$ onto $\mathcal{M}$. This projection may not exist, or may not be defined uniquely.
\end{definition}
Estimation of the semiparametric model using $\varphi-$divergences is summarized by the following optimization problem:
\begin{equation}
\alpha^* = \arginf_{\alpha\in\mathcal{A}}\inf_{Q\in\mathcal{M}_{\alpha}} D_{\varphi}(Q,P_T).
\label{eqn:SemiParaEstimGeneral}
\end{equation}
We are, then, searching for the projection of $P_T$ on the set $\mathcal{M}=\cup_{\alpha}\mathcal{M}_{\alpha}$. We are more formally interested in the vector $\alpha^*$ for which the projection of $P_T$ on $\mathcal{M}$ belongs to the set $\mathcal{M}_{\alpha^*}$. Notice that the sets $\mathcal{M}_{\alpha}$ need to be disjoint so that the projection cannot belong to several sets in the same time. \\
The magical property of $\varphi-$divergences stems from their characterization of the projection of a finite signed measure $P$ onto a set $\mathcal{M}$ of finite signed measures, see \cite{BroniaKeziou2006} Theorem 3.4. Such characterization permits to transform the search of a projection in an infinite dimensional space to the search of a vector $\xi$ in $\mathbb{R}^{\ell}$ through the duality of Fenchel-Legendre and thus simplify the optimization problem (\ref{eqn:SemiParaEstimGeneral}). Note that Theorem 3.4 from \cite{BroniaKeziou2006} provides a formal characterization of the projection, but we will only use it implicitly.\\ 
Let $\varphi$ be a strictly convex function which verifies the same properties mentioned in the definition of a $\varphi-$divergence, see paragraph \ref{subsec:DefPhiDiv}. The Fenchel-Legendre transform
of $\varphi$, say $\psi$ is defined by:
\begin{equation*}
\psi(t) = \sup_{x \in \mathbb{R}}\left\lbrace tx - \varphi(x) \right\rbrace, \qquad  \forall t	\in\mathbb{R}.
\end{equation*}
\noindent We are concerned with the convex optimization problem
\begin{equation}
\label{primal problem}
(\mathcal{P}) \qquad \inf_{Q \in \mathcal{M}_{\alpha}} D_{\varphi}(Q,P_T).
\end{equation}
\noindent We associate to $(\mathcal{P})$ the following dual problem
\begin{equation}
\label{dual problem}
(\mathcal{P}^*) \qquad \sup_{\xi\in\mathbb{R}^{\ell}} \xi^t m(\alpha)  - \int{\psi\left(\xi^t g(x)\right) dP_T(x)}.
\end{equation}
\noindent We require that $\varphi$ is differentiable. Assume furthermore that $\int |g_i(x)|dP_T(x) < \infty$ for all $i=1,\ldots, \ell$ and there exists some measure $Q_T$ a.c.w.r.t. $P_T$ such that $D_{\varphi}(Q_T,P_T) < \infty$. According to Proposition 1.4 in \cite{AlexisThesis} (see also Proposition 4.2 in \cite{BroniaKeziou12})  we have a strong duality attainment, i.e. $(\mathcal{P}) = (\mathcal{P}^*)$. In other words,
\begin{equation}
\inf_{Q\in\mathcal{M}_{\alpha}} D_{\varphi}\left(Q,P_T\right)=\sup_{\xi\in\mathbb{R}^{\ell}} \xi^t m(\alpha)  - \int{\psi\left(\xi^t g(x)\right) dP_T(x)}.
\label{strong duality}
\end{equation}
The estimation procedure of the semiparametric model (\ref{eqn:SemiParaEstimGeneral}) is now simplified into the following finite-dimensional optimization problem
\begin{eqnarray*}
\alpha^* & = & \arginf_{\alpha\in\mathcal{A}} D_{\varphi}(\mathcal{M}_{\alpha},P_T) \\
  & = & \arginf_{\alpha\in\mathcal{A}} \sup_{\xi\in\mathbb{R}^{\ell}} \xi^t m(\alpha)  - \int{\psi\left(\xi^t g(x)\right) dP_T(x)}.
\end{eqnarray*}
This is indeed a feasible procedure since we only need to optimize a real function over $\mathbb{R}^{\ell}$. Examples of such procedures can be found in \cite{BroniaKeziou12}, \cite{AlexisGSI13}, \cite{NeweySmith} and the references therein. Robustness of this procedure was studied theoretically by \cite{TomaRobustLinConstr} and was shown numerically\footnote{The results in \cite{AlexisGSI13} show that his estimator is not robust against outliers, but robust against misspecification.} in \cite{AlexisGSI13}. \\

Now that all notions and analytical tools are presented, we proceed to the main objective of this chapter; semiparametric mixtures models. The following section defines such models and presents a method to estimate them using $\varphi-$divergences. We study after that the asymptotic properties of the vector of estimates.

%%%%%%%%%%%%%%%%%%%%%%%%%%%%%%%%%%%%%%%%%%%%%%%%%%%%%%%%%%%%%%%%%%%%%%%%%%%%%%%%%
%
%==============================================================
%%%%%%%%%%%%%%%%%%%%%%%%%%%%%%%%%%%%%%%%%%%%%%%%%%%%%%%%%%%%%%%%%%%%%%%%%%%%
%==============================================================
%
%%%%%%%%%%%%%%%%%%%%%%%%%%%%%%%%%%%%%%%%%%%%%%%%%%%%%%%%%%%%%%%%%%%%%%%%%%%%%%%%%

\section{Semiparametric two-component mixture models when one component is defined through linear constraints}\label{sec:SemiparaMixturesConstr}
\subsection{Definition and identifiability}
\begin{definition} 
\label{def:SemiParaModel}
Let $X$ be a random variable taking values in $\mathbb{R}^{r}$ distributed from a probability measure $P$. We say that $P(.|\phi)$ with $\phi=(\lambda,\theta,\alpha)$ is a two-component semiparametric mixture model subject to linear constraints if it can be written as follows:
\begin{eqnarray}
P(.| \phi) & = &  \lambda P_1(.|\theta) + (1-\lambda) P_0 \quad \text{s.t. } \nonumber\\
P_0\in\mathcal{M}_{\alpha} & = & \left\{Q \in M \text{ s.t. } \int_{\mathbb{R}^r}{dQ(x)}=1, \int_{\mathbb{R}^r}{g(x)dQ(x)}=m(\alpha) \right\}
\label{eqn:SetMalpha}
\end{eqnarray}
for $\lambda\in(0,1)$ the proportion of the parametric component, $\theta\in\Theta\subset\mathbb{R}^{d}$ a set of parameters defining the parametric component, $\alpha\in\mathcal{A}\subset\mathbb{R}^{s}$ is the constraints parameter vector and finally $m(\alpha)=(m_1(\alpha),\cdots,m_{\ell}(\alpha))$ is a vector-valued function determining the value of the constraints.
\end{definition}

The identifiability of the model was not questioned in the context of Section \ref{sec:GeneralModelDef} because it suffices that the sets $\mathcal{M}_{\alpha}$ are disjoint (the function $m(\alpha)$ is one-to-one). However, in the context of this semiparametric mixture model, identifiability cannot be achieved only by supposing that the sets $\mathcal{M}_{\alpha}$ are disjoint. \\
\begin{definition} \label{def:identifiabilitySemipParaMom}
We say that the two-component semiparametric mixture model subject to linear constraints is identifiable if it verifies the following assertion. For two triplets $(\lambda,\theta,\alpha)$ and $(\tilde{\lambda},\tilde{\theta},\tilde{\alpha})$ in $\Phi=(0,1)\times\mathcal{\Theta}\times\mathcal{A}$, if
\begin{equation}
\lambda P_1(.|\theta) + (1-\lambda)P_0 = \tilde{\lambda} P_1(.|\tilde{\theta}) + (1-\tilde{\lambda}) \tilde{P}_0,\quad \text{with } P_0\in\mathcal{M}_{\alpha}, \tilde{P}_0\in\mathcal{M}_{\tilde{\alpha}}, 
\label{eqn:IdenitifiabilityDefEq}
\end{equation}
then $\lambda = \tilde{\lambda},\theta = \tilde{\theta}$ and $P_0=\tilde{P}_0$ (and hence $\alpha=\tilde{\alpha}$).
\end{definition}
This is the same identifiability concept considered in \cite{Bordes06b} where the authors exploited their symmetry assumption over $\mathbb{P}_0$ and built a system of moments equations. They proved that if $P_1$ is also symmetric, then equation (\ref{eqn:IdenitifiabilityDefEq}) has two solutions, otherwise it has three solutions. Their idea appears here in a natural way in order to prove the identifiability of our semiparametric mixture model (\ref{eqn:SetMalpha}).\\
\begin{proposition}
\label{prop:identifiabilityMixture}
For a given mixture distribution $P_T = P(.| \phi^*)$, suppose that the system of equations:
\[\frac{1}{1-\lambda} m^* - \frac{\lambda}{1-\lambda}m_1(\theta) = m_0(\alpha)\]
where $m^*=\int g(x)dP_T(x)$ and $m_1(\theta) = \int g(x) dP_1(x|\theta)$, has a unique solution $(\lambda^*,\theta^*,\alpha^*)$. Then, equation (\ref{eqn:IdenitifiabilityDefEq}) has a unique solution, i.e. $\lambda = \tilde{\lambda},\theta = \tilde{\theta}$ and $P_0=\tilde{P}_0$, and the semiparametric mixture model $P_T = P(.| \phi^*)$ is identifiable.
\end{proposition}
The proof is differed to Appendix \ref{AppendSemiPara:PropIdenitifiability}.
%Example......
\begin{example}[Semiparametric two-component Gaussian mixture]
Suppose that $P_1(.|\theta)$ is a Gaussian model $\mathcal{N}(\mu_1,1)$. Suppose also that the set of constraints is defined as follows:
\[\mathcal{M}_{\mu_0^*} = \left\{f_0 \text{ s.t. } \int{f_0(x)dx}=1,\quad\int_{\mathbb{R}}{xf_0(x)dx} = \mu_0^*,\quad \int_{\mathbb{R}}{x^2f_0(x)dx} = 1+{\mu_0^*}^2\right\}.\]
We would like to study the identifiability of the two-component semiparametric Gaussian mixture whose unknown component $P_0$ shares the first two moments with the Gaussian distribution $\mathcal{N}(\mu_0^*,1)$ for a known $\mu_0^*$. Using Proposition \ref{prop:identifiabilityMixture}, it suffices to study the system of equations 
\begin{eqnarray*}
\frac{1}{1-\lambda} \int{xf(x|\mu_1^*,\mu_0^*,\lambda^*)dx} - \frac{\lambda}{1-\lambda} \int{xf_1(x|\mu_1)dx} & = & \mu_0^* \\
\frac{1}{1-\lambda} \int{x^2f(x|\mu_1^*,\mu_0^*,\lambda^*)dx} - \frac{\lambda}{1-\lambda} \int{x^2f_1(x|\mu_1)dx} & = & 1+{\mu_0^*}^2.
\end{eqnarray*}
Recall that $\int{xf(x|\mu_1^*,\mu_0^*,\lambda^*)dx} = \lambda^*\mu_1^*+(1-\lambda^*)\mu_0^*$ and $\int{x^2f(x|\mu_1^*,\mu_0^*,\lambda^*)dx}=1+\lambda^*{\mu_1^*}^2+(1-\lambda^*){\mu_0^*}^2$. The first equation in the previous system entails that:
\begin{equation}
\lambda\mu_1 - \lambda\mu_0^* = \lambda^*\mu_1^* -\lambda^*\mu_0^*.
\label{eqn:GaussExampleCond1}
\end{equation}
The second equation gives:
\begin{equation}
\lambda^*(1+{\mu_1^*}^2)-\lambda^*(1+{\mu_0^*}^2) = \lambda\mu_1^2 - \lambda{\mu_0^*}^2
\label{eqn:GaussExampleCond2}
\end{equation}
The nonlinear system of equations (\ref{eqn:GaussExampleCond1}, \ref{eqn:GaussExampleCond2}) has a solution for $\mu_1=\mu_1^*, \lambda=\lambda^*$. Suppose by contradiction that $\mu_1\neq\mu_1^*$ and check if there are other solutions. The system (\ref{eqn:GaussExampleCond1}, \ref{eqn:GaussExampleCond2}) implies:
\[
\lambda = \frac{\lambda^*(1+{\mu_1^*}^2)-\lambda^*(1+{\mu_0^*}^2)}{(\mu_1-\mu_0^*)(\mu_1+\mu_0^*)} = \frac{\lambda^*\mu_1^* - \lambda^*\mu_0^*}{\mu_1-\mu_0^*}
\]
This entails that 
\[\mu_1 + \mu_0^* = \frac{\lambda^*\left[(1+{\mu_1^*}^2)-(1+{\mu_0^*}^2)\right]}{\lambda^*\left[\mu_1^* - \mu_0^*\right]}=\mu_1^*+\mu_0^*\] 
Hence, $\mu_1=\mu_1^*$ which contradicts what we have assumed. Thus $\mu_1=\mu_1^*, \lambda=\lambda^*$ is the only solution. We conclude that if $\mu_0^*$ is known and that we impose two moments constraints over $f_0$, then the semiparametric two-component Gaussian mixture model is identifiable.\\
Notice that imposing only one condition on the first moment is not sufficient since any value of $\lambda\in(0,1)$ would produce a corresponding solution for $\mu_1$ in equation (\ref{eqn:GaussExampleCond1}). We therefore are in need for the second constraint. Notice also that if $\lambda=\lambda^*$, then $\mu_1=\mu_1^*$. This means, by continuity of the equation over $(\lambda,\mu_1)$, if $\lambda$ is initialized in a close neighborhood of $\lambda^*$, then $\mu_1$ would be estimated near $\mu_1^*$. This may represent a remedy if we could not impose but one moment constraint.

\end{example}
%%%%%%%%%%%%%%%%%%%%%%%%%%%%%%%%%%%%%%%%%%%%%%%%%%%%%%%%%%%%%%%%%%%%%%%%%%%%%%%
%%%%%%%%%%%%%%%%%%%%%%%%%%%%%%%%%%%%%%%%%%%%%%%%%%%%%%%%%%%%%%%%%%%%%%%%%%%%%%
\subsection{An algorithm for the Estimation of the semiparametric mixture model}\label{subsec:procIntrod}
We have seen in paragraph \ref{subsec:DualTechRes} that it is possible to use $\varphi-$divergences to estimate a semiparametric model as long as the constraints apply over $P(.|\phi)$, i.e. the whole mixture. In our case, the constraints apply only on a component of the mixture. It is thus reasonable to consider a "model" expressed through $P_0$ instead of $P$. We have:
\[P_0 = \frac{1}{1-\lambda} P(.|\phi) - \frac{\lambda}{1-\lambda} P_1(.|\theta).\]
Denote $P_T=P(.|\phi^*)$ with $\phi^*=(\lambda^*,\theta^*,\alpha^*)$ the distribution which generates the observed data. Denote also $P_0^*$ to the true semiparametric component of the mixture $P_T$. The only information we hold about $P_0^*$ is that it belongs to a set $\mathcal{M}_{\alpha^*}$ for some (possibly unknown) $\alpha^*\in\mathcal{A}$. Besides, it verifies:
\begin{equation}
P_0^* = \frac{1}{1-\lambda^*} P_T - \frac{\lambda^*}{1-\lambda^*} P_1(.|\theta^*).
\label{eqn:TrueP0model}
\end{equation}
We would like to retrieve the value of the vector $\phi^*=(\lambda^*,\theta^*,\alpha^*)$ provided a sample $X_1,\cdots,X_n$ drawn from $P_T$ and that $P_0^*\in\cup_{\alpha}\mathcal{M}_{\alpha}$. Consider the set of signed measures:
\begin{equation}
\mathcal{N} = \left\{Q=\frac{1}{1-\lambda} P_T - \frac{\lambda}{1-\lambda} P_1(.|\theta), \quad \lambda\in(0,1),\theta\in\Theta\right\}.
\label{eqn:SetNnewModel}
\end{equation}
Notice that $P_0^*$ belongs to this set for $\lambda=\lambda^*$ and $\theta=\theta^*$. On the other hand, $P_0^*$ is supposed, for simplicity, to belong to the union $\cup_{\alpha\in\mathcal{A}}\mathcal{M}_{\alpha}$. We may now write,
\[P_0^*\in \mathcal{N} \bigcap \cup_{\alpha\in\mathcal{A}}\mathcal{M}_{\alpha}.\]
If we suppose now that the intersection $\mathcal{N} \bigcap \cup_{\alpha\in\mathcal{A}}\mathcal{M}_{\alpha}$ contains only one element (see paragraph \ref{subsec:UniquenessSolMom} for a discussion) which would be a fortiori $P_0^*$, then it is very reasonable to consider an estimation procedure by calculating some "distance" between the two sets $\mathcal{N}$ and $\cup_{\alpha\in\mathcal{A}}\mathcal{M}_{\alpha}$. Such distance can be measured using a $\varphi-$divergence by (see Definition \ref{def:phiDistance}):
\begin{equation}
D_{\varphi}(\mathcal{M},\mathcal{N}) = \inf_{Q\in\mathcal{N}}\inf_{P_0\in\mathcal{M}}D_{\varphi}(P_0,Q).
\label{eqn:DistanceTwoMeasureSets}
\end{equation}
We may reparametrize this distance using the definition of $\mathcal{N}$. Indeed,
\begin{eqnarray}
D_{\varphi}(\cup_{\alpha}\mathcal{M}_{\alpha},\mathcal{N}) & = & \inf_{Q\in\mathcal{N}}\inf_{P_0\in\cup_{\alpha}\mathcal{M}_{\alpha}}D_{\varphi}(P_0,Q) \nonumber\\
 & = & \inf_{\lambda,\theta}\inf_{\alpha, P_0\in\mathcal{M}_{\alpha}}D_{\varphi}\left(P_0,\frac{1}{1-\lambda} P_T - \frac{\lambda}{1-\lambda} P_1(.|\theta)\right).
\label{eqn:DistanceModelConstraint}
\end{eqnarray}
If we still have $P_0^*$ as the only signed measure which belongs to both $\mathcal{N}$ and $\cup_{\alpha}\mathcal{M}_{\alpha}$, then, the argument of the infimum in (\ref{eqn:DistanceModelConstraint}) is none other than $(\lambda^*,\theta^*,\alpha^*)$, i.e.
\begin{equation}
(\lambda^*,\theta^*,\alpha^*) = \arginf_{\lambda,\theta,\alpha}\inf_{P_0\in\mathcal{M}_{\alpha}}D_{\varphi}\left(P_0,\frac{1}{1-\lambda} P(.|\phi^*) - \frac{\lambda}{1-\lambda} P_1(.|\theta)\right).
\label{eqn:MomentTrueEstimProc}
\end{equation}
It is important to notice that if $P_0^*\notin\cup\mathcal{M}_{\alpha}$, then the procedure still makes sense. Indeed, we are searching for the best measure of the form $\frac{1}{1-\lambda}P_T - \frac{\lambda}{1-\lambda}P_1(.|\theta)$ which verifies the constraints.
%%%%%%%%%%%%%%%%%%%%%%%%%%%%%%%%%%%%%%%%%%%%%%%%%%%%%%%%%%%%%%%%%%%%%%%%%%%%%%%%%%%%%%
%%%%%%%%%%%%%%%%%%%%%%%%%%%%%%%%%%%%%%%%%%%%%%%%%%%%%%%%%%%%%%%%%%%%%%%%%%%%%%%%%%%%%%
\subsection{The algorithm in practice : Estimation using the duality technique and plug-in estimate}
The Fenchel-Legendre duality permits to transform the problem of minimizing under linear constraints in a possibly infinite dimensional space into an unconstrained optimization problem in the space of Lagrangian parameters over $\mathbb{R}^{\ell+1}$, where $\ell+1$ is the number of constraints. We will apply the duality result presenter earlier in paragraph \ref{subsec:DualTechRes} on the inner optimization in equation (\ref{eqn:DistanceModelConstraint}). Redefine the function $m$ as $m(\alpha)=(m_0(\alpha),m_1(\alpha),\cdots,m_{\ell}(\alpha))$ where $m_0(\alpha)=1$. We have:
\begin{eqnarray*}
\inf_{Q\in\mathcal{M}_{\alpha}} D_{\varphi}\left(Q,\frac{1}{\lambda-1} P_T - \frac{\lambda}{1-\lambda} P_1(.|\theta)\right) & = & \sup_{\xi\in\mathbb{R}^{l+1}} \xi^t m(\alpha)  - \frac{1}{1-\lambda}\int{\psi\left(\xi^t g(x)\right) \left(dP_T(x) - \lambda dP_1(x|\theta)\right)}\\
 & = & \sup_{\xi\in\mathbb{R}^{l+1}} \xi^t m(\alpha)  - \frac{1}{1-\lambda}\int{\psi\left(\xi^t g(x)\right) dP_T(x)} \\
	&  & + \frac{\lambda}{1-\lambda} \int{\psi\left(\xi^t g(x)\right) dP_1(x|\theta)}.
\end{eqnarray*}
Inserting this result in (\ref{eqn:MomentTrueEstimProc}) gives that:
\begin{eqnarray*}
(\lambda^*,\theta^*,\alpha^*)  & = &  \arginf_{\phi}\inf_{Q\in\mathcal{M}_{\alpha}} D_{\varphi}\left(Q,\frac{1}{\lambda-1} P_T - \frac{\lambda}{1-\lambda} P_1(.|\theta)\right) \\
& = & \arginf_{\phi}\sup_{\xi\in\mathbb{R}^{l+1}} \xi^t m(\alpha)  - \frac{1}{1-\lambda}\int{\psi\left(\xi^t g(x)\right) dP_T(x)} \\
	&  & + \frac{\lambda}{1-\lambda} \int{\psi\left(\xi^t g(x)\right) dP_1(x|\theta)}.
\end{eqnarray*}
The right hand side can be estimated on the basis of an $n-$sample drawn from $P_T$, say $X_1,\cdots,X_n$, by a simple plug-in of the empirical measure $P_n$. The resulting procedure can now be written as:
\begin{eqnarray}
(\hat{\lambda}, \hat{\theta},\hat{\alpha}) & = & \arginf_{\lambda, \theta,\alpha} \sup_{\xi\in\mathbb{R}^{l+1}} \xi^t m(\alpha)  - \frac{1}{1-\lambda}\frac{1}{n}\sum_{i=1}^n{\psi\left(\xi^t g(X_i)\right)} \nonumber\\
 &  & \qquad \qquad \qquad + \frac{\lambda}{1-\lambda} \int{\psi\left(\xi^t g(x)\right) dP_1(x|\theta)}.
\label{eqn:MomentEstimProc}
\end{eqnarray}
This is a feasible procedure in the sens that we only need the data, the set of constraints and the model of the parametric component.
%%%%%%%%%% Example
\begin{example}[Chi square] 
\label{ex:Chi2LinConstr}
Let's take the case of the $\chi^2$ divergence for which $\varphi(t)=(t-1)^2/2$. The Convex conjugate of $\varphi$ is given by $\psi(t)=t^2/2+t$. For $(\lambda,\theta,\alpha)\in\Phi$, we have:
\begin{multline*}
\inf_{Q\in\mathcal{M}_{\alpha}} D_{\varphi}\left(Q,\frac{1}{\lambda-1} P_T - \frac{\lambda}{1-\lambda} P_1(.|\theta)\right)
 =  \sup_{\xi\in\mathbb{R}^{l+1}} \xi^t m(\alpha)  - \frac{1}{1-\lambda}\int{\left[\frac{1}{2}\left(\xi^t g(x)\right)^2+\xi^t g(x)\right] dP_T(x)} \\
	+ \frac{\lambda}{1-\lambda} \int{\left[\frac{1}{2}\left(\xi^t g(x)\right)^2+\xi^t g(x)\right] dP_1(x|\theta)}.
\end{multline*}
It is interesting to note that the supremum over $\xi$ can be calculated explicitly. Clearly, the optimized function is a polynomial of $\xi$ and thus infinitely differentiable. The Hessian matrix is equal to $-\Omega$ where:
\begin{equation}
\Omega = \int{g(x)g(x)^t\left(\frac{1}{1-\lambda}dP(x)-\frac{\lambda}{1-\lambda}dP_1(x|\theta)\right)}. \\
\label{eqn:OmegaMat}
\end{equation}
If the measure $\frac{1}{1-\lambda}dP-\frac{\lambda}{1-\lambda}dP_1(.|\theta)$ is positive, then $\Omega$ is symmetric definite positive (s.d.p) and the Hessian matrix is symmetric definite negative. Consequently, the supremum over $\xi$ is ensured to exist. If it is a signed measure, then the supremum might be infinity. We may now write:
\[
\xi(\phi) = \Omega^{-1}\left(m(\alpha) - \int{g(x)\left(\frac{1}{1-\lambda}dP(x)-\frac{\lambda}{1-\lambda}dP_1(x|\theta)\right)}\right), \quad \text{if } \Omega \text{ is s.d.p}
\]
For the empirical criterion, we define similarly $\Omega_n$ by:
\begin{equation}
\Omega_n = \frac{1}{n}\frac{1}{1-\lambda}\sum_{i=1}^n{g(X_i)g(X_i)^t}-\frac{\lambda}{1-\lambda}\int{g(x)g(x)^tdP_1(x|\theta)}. \\
\label{eqn:OmeganMat}
\end{equation}
The solution to the corresponding supremum over $\xi$ is given by:
\[
\xi_n(\phi) = \Omega_n^{-1}\left(m(\alpha) - \frac{1}{n}\frac{1}{1-\lambda}\sum_{i=1}^n{g(X_i)}+\frac{\lambda}{1-\lambda}\int{g(x)dP_1(x|\theta)}\right), \quad \text{if } \Omega_n \text{ is s.d.p}\]
\end{example}
%%%%%%%%%%%%%%%%%%%%%%%%%%%%%%%%%%%%%%%%%%%%%%%%%%%%%%%%%%%%%%%%%%%
%%%%%%%%%%%%%%%%%%%%%%%%%%%%%%%%%%%%%%%%%%%%%%%%%%%%%%%%%%%%%%%%%%%
\subsection{Uniqueness of the solution "under the model"}\label{subsec:UniquenessSolMom}
By a unique solution we mean that only one measure, which can be written in the form of $\frac{1}{1-\lambda}P_T-\frac{\lambda}{1-\lambda}P_1(.|\theta)$, verifies the constraints with a unique triplet $(\lambda^*,\theta^*,\alpha^*)$. The existence of a unique solution is essential in order to ensure that the procedure (\ref{eqn:MomentTrueEstimProc}) is a reasonable estimation method. We provide next a result ensuring the uniqueness of the solution. The idea is based on the identification of the intersection of the set $\mathcal{N}\cap\mathcal{M}$. The proof is differed to Appendix \ref{AppendSemiPara:Prop1}.
\begin{proposition}
\label{prop:identifiability}
Assume that $P_0^*\in\mathcal{M}=\cup_{\alpha}\mathcal{M}_{\alpha}$. Suppose also that:
\begin{enumerate}
\item the system of equations:
\begin{equation}
\int{g_i(x)\left(dP(x|\phi^*) - \lambda dP_1(x|\theta)\right)} = (1-\lambda)m_i(\alpha), \qquad  i=1,\cdots,\ell
\label{eqn:NlnSys}
\end{equation}
has a unique solution $(\lambda^*,\theta^*,\alpha^*)$;
\item the function $\alpha\mapsto m(\alpha)$ is one-to-one;
\item for any $\theta\in\Theta$ we have :
\[\lim_{\|x\|\rightarrow \infty} \frac{dP_1(x|\theta)}{dP_T(x)} = c,\quad \text{ with } c\in [0,\infty)\setminus\{1\};\]
\item the parametric component is identifiable, i.e. if $P_1(.|\theta) = P_1(.|\theta')\;\; dP_T-$a.e. then $\theta=\theta'$,
\end{enumerate}
then, the intersection $\mathcal{N}\cap\mathcal{M}$ contains a unique measure $P_0^*$, and there exists a unique vector $(\lambda^*,\theta^*,\alpha^*)$ such that $P_T = \lambda^*P_1(.|\theta^*)+(1-\lambda)P_0^*$ where $P_0^*$ is given by (\ref{eqn:TrueP0model}) and belongs to $\mathcal{M}_{\alpha^*}$. Moreover, provided assumptions 2-4, the conclusion holds if and only if assumption 1 is fulfilled.
\end{proposition}
There is no general result for a non linear system of equations to have a unique solution; still, it is necessary to ensure that $\ell\geq d+s+1$, otherwise there would be an infinite number of signed measures in the intersection $\mathcal{N} \bigcap \cup_{\alpha\in\mathcal{A}}\mathcal{M}_{\alpha}$.
\begin{remark}
Assumptions 3 and 4 of Proposition \ref{prop:identifiability} are used to prove the identifiability of the "model" $\left(\frac{1}{1-\lambda}P_T - \frac{\lambda}{1-\lambda}P_1(.|\theta)\right)_{\lambda,\theta}$. Thus, according to the considered situation we may find simpler ones for particular cases (or even for the general case). Our assumptions remain sufficient but not necessary for the proof.
\end{remark}
%%%%%%%%%%%%
% Example
\begin{example}
One of the most popular models in clustering is the Gaussian multivariate mixture (GMM). Suppose that we have two classes. Linear discriminant analysis (LDA) is based on the hypothesis that the covariance matrix of the two classes is the same. Let $X$ be a random variable which takes its values in $\mathbb{R}^2$ and is drawn from a mixture model of two components. In the context of LDA, the model has the form:
\[f(x,y|\lambda,\mu_1,\mu_2,\Sigma) = \lambda f_1(x,y|\mu_1,\Sigma) + (1-\lambda)f_1(x,y|\mu_2,\Sigma),\]
with:
\[f_1(x,y|\mu_1,\Sigma)=\frac{1}{2\pi\sqrt{|\text{det}(\Sigma)|}}\exp\left[-\frac{1}{2}((x,y)^t-\mu_1)^t\Sigma((x,y)^t-\mu_1)\right], \quad \Sigma = \left(\begin{array}{cc}\sigma^2 & \rho \\ \rho & \sigma^2\end{array}\right).\]
We would like to relax the assumption over the second component by keeping the fact that the covariance matrix is the same as the one of the first component. We will start by imposing the very natural constraints on the second component. 
\begin{eqnarray*}
\int{xf_0(x,y)dxdy} & = & \mu_{2,1}, \\
\int{yf_0(x,y)dxdy} & = & \mu_{2,2},\\ 
\int{x^2f_0(x,y)dxdy} & = & \sigma^2, \\
\int{y^2f_0(x,y)dxdy} & = & \sigma^2, \\
\int{xyf_0(x,y)dxdy} & = & \rho + \mu_{2,1}\mu_{2,2} - \mu_{2,1}^2 - \mu_{2,2}^2.
\end{eqnarray*}
These constraints concern only the fact that the covariance matrix $\Sigma$ is the same as the one of the Gaussian component (the parameteric one). In order to see whether this set of constraints is sufficient for the existence of a unique measure in the intersection $\mathcal{N}\cap\mathcal{M}$, we need to write the set of equations corresponding to (\ref{eqn:NlnSys}) in Proposition \ref{prop:identifiability}.
\begin{eqnarray*}
\int{x\left[\frac{1}{1-\lambda}f(x,y) - \frac{\lambda}{1-\lambda}f_1(x,y|\mu_1,\sigma,\rho)\right]dxdy} & = & \mu_{2,1}, \\
\int{y\left[\frac{1}{1-\lambda}f(x,y) - \frac{\lambda}{1-\lambda}f_1(x,y|\mu_1,\sigma,\rho)\right]dxdy} & = & \mu_{2,2}, \\ 
\int{x^2\left[\frac{1}{1-\lambda}f(x,y) - \frac{\lambda}{1-\lambda}f_1(x,y|\mu_1,\sigma,\rho)\right]dxdy} & = & \sigma^2, \\
\int{y^2\left[\frac{1}{1-\lambda}f(x,y) - \frac{\lambda}{1-\lambda}f_1(x,y|\mu_1,\sigma,\rho)\right]dxdy} & = & \sigma^2, \\
\int{xy\left[\frac{1}{1-\lambda}f(x,y) - \frac{\lambda}{1-\lambda}f_1(x,y|\mu_1,\sigma,\rho)\right]dxdy} & = & \rho + \mu_{2,1}\mu_{2,2} - \mu_{2,1}^2 - \mu_{2,2}^2,
\end{eqnarray*}
The number of parameters is 7, and we only have 5 equations. In order for the problem to have a unique solution, it is necessary to either add two other constraints or to consider for example $\mu_1 = (\mu_{1,1},\mu_{1,2})$ to be known\footnote{or estimated by another procedure such as $k-$means.}. Other solutions exist, but depend on the prior information. We may imagine an assumption of the form $\mu_{1,1}=a\mu_{1,2}$ and $\mu_{2,1}=b\mu_{2,2}$ for given constants $a$ and $b$.\\
The gain from relaxing the normality assumption on the second component is that we are building a model which is not constrained to a Gaussian form for the second component, but rather to a form which suits the data. The price we pay is the number of relevant constraints which must be at least equal to the number of unknown parameters.
\end{example}

%%%%%%%%%%%%%%%%%%%%%%%%%%%%%%%%%%%%%%%%%%%%%%%%%%%%%%%%%%%%%%%%%%%%%%%%%%%%%%%%%
%
%==============================================================
%%%%%%%%%%%%%%%%%%%%%%%%%%%%%%%%%%%%%%%%%%%%%%%%%%%%%%%%%%%%%%%%%%%%%%%%%%%%
%==============================================================
%
%%%%%%%%%%%%%%%%%%%%%%%%%%%%%%%%%%%%%%%%%%%%%%%%%%%%%%%%%%%%%%%%%%%%%%%%%%%%%%%%%

\section{Asymptotic properties of the new estimator}\label{sec:AsymptotResults}
\subsection{Consistency}
The double optimization procedure defining the estimator $\hat{\phi}$ defined by (\ref{eqn:MomentEstimProc}) does not permit us to use M-estimates methods to prove consistency. In \cite{KeziouThesis} Proposition 3.7 and in \cite{BroniaKeziou09} Proposition 3.4, the authors propose a method which can simply be generalized to any double optimization procedure since the idea of the proof slightly depends on the form of the optimized function. In order to restate this result here and give an exhaustive and a general proof, suppose that our estimator $\hat{\phi}$ is defined through the following double optimization procedure. Let $H$ and $H_n$ be two generic functions such that $H_n(\phi,\xi) \rightarrow H(\phi,\xi)$ in probability for any couple $(\phi,\xi)$. Define $\hat{\phi}$ and $\phi^*$ as follows:
\begin{eqnarray*}
\hat{\phi} & = & \arginf_{\phi} \sup_{\xi} H_n(\phi,\xi);\\
\phi^* & = & \arginf_{\phi} \sup_{\xi} H(\phi,\xi).
\end{eqnarray*}
We adapt the following notation:
\[\xi(\phi) = \argsup_{t} H(\phi,t), \qquad \xi_n(\phi) = \argsup_{t} H_n(\phi,t)\]
The following theorem provides sufficient conditions for consistency of $\hat{\phi}$ towards $\phi^*$. This result will then be applied to the case of our estimator.\\
Assumptions:
\begin{itemize}
\item[A1.] the estimate $\hat{\phi}$ exists (even if it is not unique);
\item[A2.] $\sup_{\xi,\phi} \left|H_n(\phi,\xi) - H(\phi,\xi)\right|$ tends to 0 in probability;
\item[A3.] for any $\phi$, the supremum of $H$ over $\xi$ is unique and isolated, i.e. $\forall \varepsilon>0, \forall \tilde{\xi}$ such that $\|\tilde{\xi}-\xi(\phi)\|>\varepsilon$, then there exists $\eta>0$ such that $H(\phi,\xi(\phi)) - H(\phi,\tilde{\xi})>\eta$;
\item[A4.] the infimum of $\phi\mapsto H(\phi,\xi(\phi))$ is unique and isolated, i.e. $\forall \varepsilon>0, \forall \phi$ such that $\|\phi-\phi^*\|>\varepsilon$, there exists $\eta>0$ such that $H(\phi,\xi(\phi))-H(\phi^*,\xi(\phi^*))>\eta$;
\item[A5.] for any $\phi$ in $\Phi$, function $\xi\mapsto H(\phi,\xi)$ is continuous.
\end{itemize}
In assumption A4, we suppose the existence and uniqueness of $\phi^*$. It does not, however, imply the uniqueness of $\hat{\phi}$. This is not a problem for our consistency result. The vector $\hat{\phi}$ may be any point which verifies the minimum of function $\phi\mapsto\sup_{\xi} H_n(\phi,\xi)$. Our consistency result shows that all vectors verifying the minimum of $\phi\mapsto\sup_{\xi} H_n(\phi,\xi)$ converge to the unique vector $\phi^*$. We also prove an asymptotic normality result which shows that even if $\hat{\phi}$ is not unique, all possible values should be in a neighborhood of radius $\mathcal{O}(n^{-1/2})$ centered at $\phi^*$.\\
The following lemma establishes a uniform convergence result for the argument of the supremum over $\xi$ of function $H_n(\phi,\xi)$ towards the one of function $H(\phi,\xi)$. It constitutes a first step towards the proof of convergence of $\hat{\phi}$ towards $\phi^*$. The proof is differed to Appendix \ref{AppendSemiPara:Lem1}.
%%%%%%%%%%%%%%%%%%%%%%%%%%%
\begin{lemma}
\label{lem:SupXiPhiDiff}
Assume A2 and A3 are verified, then 
\[\sup_{\phi}\|\xi_n(\phi) - \xi(\phi)\| \rightarrow 0 ,\qquad \textit{in probability.}\] 
\end{lemma}
\noindent We proceed now to announce our consistency theorem. The proof is differed to Appendix \ref{AppendSemiPara:Theo1}.
%%%%%%%%%%%%%%%%%%%%%%%%%
\begin{theorem}
\label{theo:MainTheorem}
Let $\xi(\phi)$ be the argument of the supremum of $\xi\mapsto H(\phi,\xi)$ for a fixed $\phi$. Assume that A1-A5 are verified, then $\hat{\phi}$ tends to $\phi^*$ in probability.
\end{theorem}
%%%%%%%%%%%%%%%%%%%%%%%%%%%%%
\noindent Let's now go back to our optimization problem (\ref{eqn:MomentEstimProc}) in order to simplify the previous assumptions. First of all, we need to specify functions $H$ and $H_n$. Define function $h$ as follows. Let $\phi = (\lambda,\theta,\alpha)$,
\[h(\phi,\xi, z) = \xi^t m(\alpha)  - \frac{1}{1-\lambda}\psi\left(\xi^t g(z)\right) + \frac{\lambda}{1-\lambda} \int{\psi\left(\xi^t g(x)\right) dP_1(x|\theta)}.\]
Functions $H$ and $H_n$ can now be defined through $h$ by:
\[H(\phi,\xi) = P_T h(\phi,\xi,.),\qquad H_n(\phi,\xi) = P_n h(\phi,\xi,.).\]
In example \ref{ex:Chi2LinConstr}, we considered the the case of the Pearson's $\chi^2$. The supremum is infinity whenever the matrix $\Omega$ defined by (\ref{eqn:OmegaMat}) is s.d.p. It is thus interesting to define the \emph{effective} set of parameters. Define the set $\Phi^+$ by
\[\Phi^+ = \{\phi\in\Phi \text{ s.t. } \xi\mapsto H(\phi,\xi) \text{ is strictly concave}\}\]
Outside the set $\Phi^+$, function $\xi\mapsto H(\phi,\xi)$ is not upper bounded. 
\begin{theorem}
\label{theo:MainTheoremMomConstr}
Assume that A1, A4 and A5 are verified for $\Phi$ replaced by $\Phi^+$. Suppose also that 
\begin{equation}
\sup_{\xi\in\mathbb{R}^{\ell+1}}\left|\int{\psi\left(\xi^t g(x)\right)dP_T(x)} - \frac{1}{n}\sum_{i=1}^n{\psi\left(\xi^t g(X_i)\right)}\right| \xrightarrow[\mathbb{P}]{n\rightarrow\infty} 0,
\label{eqn:AssumptionB2}
\end{equation} 
then the estimator defined by (\ref{eqn:MomentEstimProc}) is consistent.
\end{theorem}

\noindent The proof is differed to Appendix \ref{AppendSemiPara:Theo1}. Assumption A5 could be handled using Lebesgue's continuity theorem if one finds a $P_T-$integrable function $\tilde{h}$ such that $|\psi\left(\xi^t g(z)\right)|\leq \tilde{h}(z)$. This is, however, not possible in general unless we restrain $\xi$ to a compact set. Otherwise, we need to verify this assumption according the situation we have in hand, see example \ref{ex:Chi2Consistency} below for more details. The uniform limit (\ref{eqn:AssumptionB2}) can be treated according to the divergence and the constraints which we would like to impose. A general method is to prove that the class of functions $\{x\mapsto \psi\left(\xi^t g(x)\right),\xi\in\mathbb{R}^{\ell+1}\}$ is a Glivenko-Cantelli class of functions, see \cite{Vaart} Chap. 19 Section 2 and the examples therein for some possibilities.
%existence of $\hat{\phi}$ and continuity of $\xi(\phi)$.
\begin{remark}
\label{rem:PhiPlusJacobMat}
Under suitable differentiability assumptions, the set $\Phi^+$ defined earlier can be rewritten as:
\[\Phi^+ = \Phi\cap \left\{\phi: \quad J_{H(\phi,.)} \text{ is definite negative}\right\},\]
where $J_{H(\phi,.)}$ is the Hessian matrix of function $\xi\mapsto H(\phi,\xi)$ and is given by:
\begin{equation}
J_{H(\phi,.)} = -\int{g(x)g(x)^t\psi''(\xi^tg(x)) \left(\frac{1}{1-\lambda}dP_T - \frac{\lambda}{1-\lambda}dP_1\right)(x)}.
\label{eqn:JacobHMomentConstr}
\end{equation}
\end{remark}
\noindent The problem with using the set $\Phi^+$ is that if we take a point $\phi$ in the interior of $\Phi$, there is no guarantee that it would be an interior point of $\Phi^+$. This will impose more difficulties in the proof of the asymptotic normality. We prove in the next proposition that this is however true for $\phi^*$. Besides, the set $\Phi^+$ is open as soon as $\int{\Phi}$ is not void. The proof is differed to Appendix \ref{AppendSemiPara:Prop2}.
%%Proposition
\begin{proposition}
\label{prop:ContinDiffxiMom}
Assume that function $\xi\mapsto H(\phi,\xi)$ is of class $\mathcal{C}^2$ for any $\phi\in\Phi^+$. Suppose that $\phi^*$ is an interior point of $\Phi$, then there exists a neighborhood $\mathcal{V}$ of $\phi^*$ such that for any $\phi\in\mathcal{V}$, $J_{H(\phi,.)}$ is definite negative and thus $\xi(\phi)$ exists and is finite. Moreover, function $\phi\mapsto\xi(\phi)$ is continuously differentiable on $\mathcal{V}$.
\end{proposition}
\begin{corollary}
\label{cor:GlivenkoCantelliClass}
Assume that function $\xi\mapsto H(\phi,\xi)$ is of class $\mathcal{C}^2$ for any $\phi\in\Phi^+$. If $\Phi$ is bounded, then there exists a compact neighborhood $\bar{\mathcal{V}}$ of $\phi^*$ such that $\xi(\bar{\mathcal{V}})$ is bounded and $\{x\mapsto \psi\left(\xi^t g(x)\right),\xi\in\xi(\bar{\mathcal{V}})\}$ is a Glivenko-Cantelli class of functions.
\end{corollary}
\begin{proof}
The first part of the corollary is an immediate result of Proposition \ref{prop:ContinDiffxiMom} and the continuity of function $\phi\mapsto\xi(\phi)$ over $\text{int}(\Phi^+)$. The implicit functions theorem permits to conclude that $\xi(\phi)$ is continuously differentiable over $\text{int}(\Phi^+)$. The second part is an immediate result of Example 19.7 page 271 from \cite{Vaart}.
\end{proof}
\noindent This corollary suggests that in order to prove the consistency of $\hat{\phi}$, it suffices to restrict the values of $\phi$ on $\Phi^+$ and the values of $\xi$ on $\xi(\text{int}(\Phi^+))$ in the definition of $\hat{\phi}$ (\ref{eqn:MomentEstimProc}). Besides, since $\{x\mapsto \psi\left(\xi^t g(x)\right),\xi\in\xi(\text{int}(\Phi^+))\}$ is a Glivenko-Cantelli class of functions, the uniform limit (\ref{eqn:AssumptionB2}) is verified by the Glivenko-Cantelli theorem.

\begin{remark}
There is a great difference between the set $\Phi^+$ where $\Omega$ is s.d.p. ($J_H$ is s.d.n.) and the set where only $\frac{1}{1-\lambda}dP_T - \frac{\lambda}{1-\lambda}dP_1$ is a probability measure. Indeed, there is a strict inclusion in the sense that if $\frac{1}{1-\lambda}dP_T - \frac{\lambda}{1-\lambda}dP_1$ is a probability measure, then $\Omega$ is s.d.p., but the inverse is not right. Figure (\ref{fig:PosSetVsPosDefSet}) shows this difference. Furthermore, it is clearly simpler to check for a vector $\phi$ if the matrix $\Omega$ is s.d.p. It suffices to calculate the integral\footnote{If function $g$ is a polynomial, i.e. moment constraints, then the integral is a mere subtractions between the moments of $P_T$ and the ones of $P_1$.} (even numerically) and then use some rule such as Sylvester's rule to check if it is definite negative, see the example below. However, in order to check if the measure $\frac{1}{1-\lambda}dP_T - \frac{\lambda}{1-\lambda}dP_1$ is positive, we need to verify it on all $\mathbb{R}^r$.
\begin{figure}[h]
\centering
\includegraphics[scale=0.4]{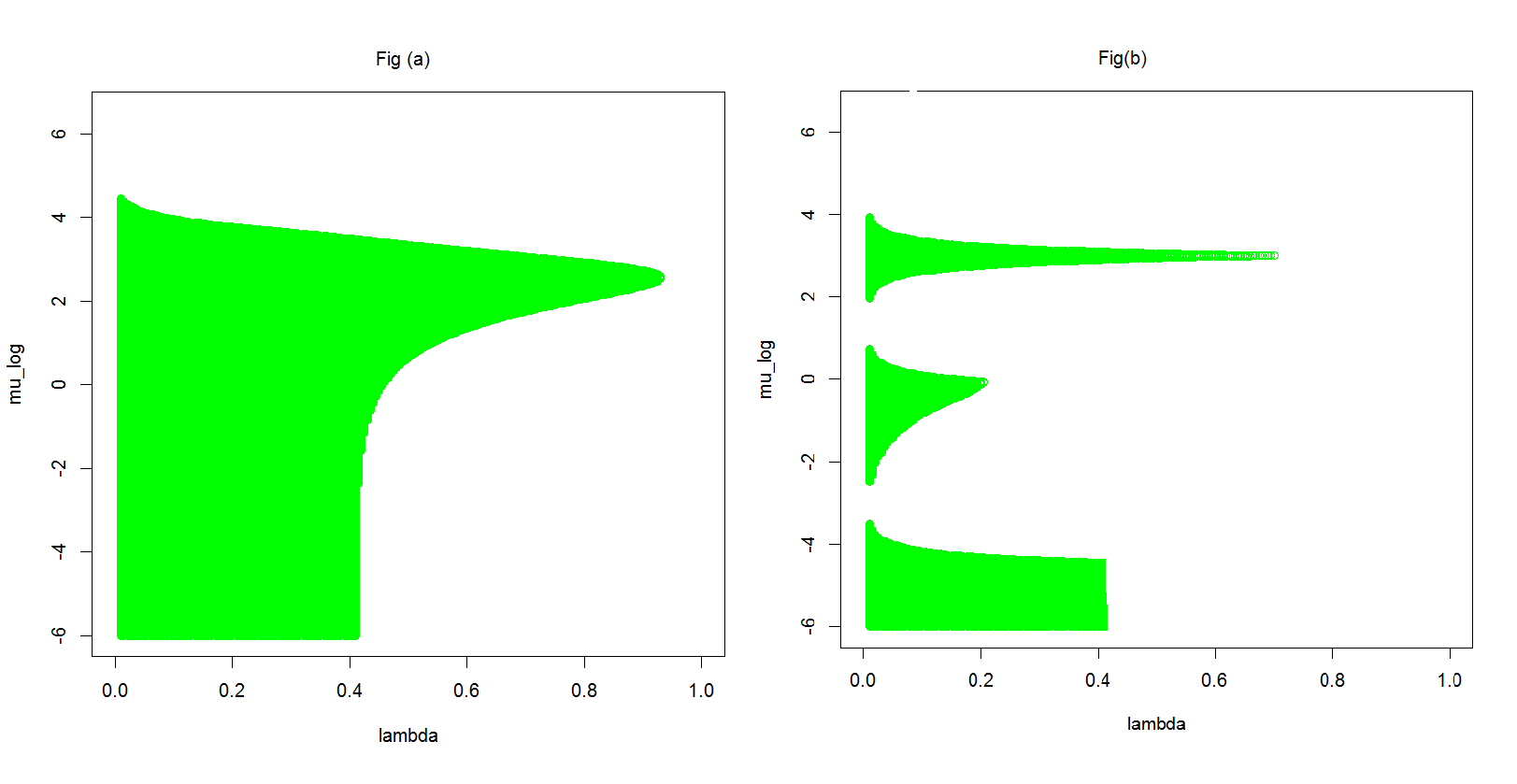}
\caption{Differences between the set where $\frac{1}{1-\lambda}dP_T - \frac{\lambda}{1-\lambda}dP_1$ is positive (Fig (b)) and the set $\Phi^+$ (Fig (a)) in a Weibull--Lognormal mixture.}
\label{fig:PosSetVsPosDefSet}
\end{figure}
\end{remark}
\begin{remark}
The previous remark shows the interest of adapting a methodology based on signed measures and not only positive ones. We have a larger and better space to search inside for the triplet $(\lambda^*,\theta^*,\alpha^*)$. For example, in figure (\ref{fig:PosSetVsPosDefSet}), the optimization algorithm which tries to solve (\ref{eqn:MomentEstimProc}) gets stuck if we only search in the set of parameters for which $\frac{1}{1-\lambda}dP_T - \frac{\lambda}{1-\lambda}dP_1$ is a probability measure. This does not happen if we search in the set $\Phi^+$. Moreover, even if the algorithm returns a triplet $(\hat{\lambda},\hat{\theta},\hat{\alpha})$ for which the semiparametric component $P_0=\frac{1}{1-\lambda}dP_T - \frac{\lambda}{1-\lambda}dP_1$ is not a probability measure, it should not mean that the procedure failed. This is because we are looking for the parameters and not to estimate $P_0$. Besides, it is still possible to threshold the negative values from the density and then regularize in order to integrate to one. 
\end{remark}
\begin{example}[$\chi^2$ case]
\label{ex:Chi2Consistency}
Consider the case of a two-component semiparapetric mixture model where $P_0$ is defined through its first three moments. In other words, the set of constraints $\mathcal{M}_{\alpha}$ is given by:
\[\mathcal{M}_{\alpha} = \left\{Q: \int{dQ(x)}=1,\; \int{xdQ(x)}=m_1(\alpha),\; \int{x^2dQ(x)}=m_2(\alpha),\; \int{x^3dQ(x)}=m_3(\alpha)\right\}.\]
We have already seen in example \ref{ex:Chi2LinConstr} that if $\psi(t)=t^2/2+t$, the Pearson's $\chi^2$ convex conjugate, then the optimization over $\xi$ can be solved and the solution is given by:
\[\xi(\phi) = \Omega^{-1}\left(m(\alpha) - \int{g(x)\left(\frac{1}{1-\lambda}dP(x)-\frac{\lambda}{1-\lambda}dP_1(x|\theta)\right)}\right), \text{ for } \phi\in\Phi^+.\]
Let $M_i$ denotes the moment of order $i$ of $P_T$. Denote also $M_i^{(1)}(\theta)$ the moment of order $i$ of the parametric component $P_1(.|\theta)$.
\[M_i = \mathbb{E}_{P_T}[X^i],\qquad M_i^{(1)}(\theta)=\mathbb{E}_{P_1(.|\theta)}[X^i].\]
A simple calculus shows that:
\begin{eqnarray*}
\Omega & = & \int{g(x)g(x)^t\left(\frac{1}{1-\lambda}dP(x)-\frac{\lambda}{1-\lambda}dP_1(x|\theta)\right)} \\
  & = & \left[\frac{1}{1-\lambda}M_{i+j-2} - \frac{\lambda}{1-\lambda}M_{i+j-2}^{(1)}(\theta)\right]_{i,j\in\{1,\cdots,4\}}.
\end{eqnarray*}
The solution holds for any $\phi\in\text{int}(\Phi^+)$. Continuity assumption A5 over $\xi\mapsto H(\phi,\xi)$ is simplified here because function $H$ is a polynomial of degree 2. We have:
\begin{multline*}
H(\phi,\xi) = \xi^tm(\alpha) - \left[\frac{1}{2}\xi_1^2+\xi_1+(\xi_1\xi_2+\xi_2)\left(\frac{1}{1-\lambda}M_1-\frac{\lambda}{1-\lambda}M_1^{(1)}(\theta)\right)\right. \\ +(\xi_2^2/2+\xi_1\xi_2+\xi_3)\left(\frac{1}{1-\lambda}M_2-\frac{\lambda}{1-\lambda}M_2^{(1)}(\theta)\right) + (\xi_1\xi_4+\xi_2\xi_3+\xi_4)\left(\frac{1}{1-\lambda}M_3-\frac{\lambda}{1-\lambda}M_3^{(1)}(\theta)\right) \\
+ (\xi_3^2/2+\xi_2\xi_4)\left(\frac{1}{1-\lambda}M_4-\frac{\lambda}{1-\lambda}M_4^{(1)}(\theta)\right) + \xi_3\xi_4\left(\frac{1}{1-\lambda}M_5-\frac{\lambda}{1-\lambda}M_5^{(1)}(\theta)\right) \\
\left. + \xi_4^2/2 \left(\frac{1}{1-\lambda}M_6-\frac{\lambda}{1-\lambda}M_6^{(1)}(\theta)\right)\right].
\end{multline*}
Regularity of function $\phi\mapsto\xi(\phi)$ is directly tied by the regularity of the moments of $P_1(.|\theta)$ with respect to $\theta$. If $M_i^{(1)}$ is continuous with respect to $\theta$ and $m(\alpha)$ is continuous with respect to $\alpha$, then the existence of $\phi^*$ becomes immediate as soon as the set $\Phi$ is compact.\\
If $\phi^*$ is an interior point of $\Phi$, then Proposition \ref{prop:ContinDiffxiMom} and Corollary \ref{cor:GlivenkoCantelliClass} apply. Thus int$(\Phi^+)$ is non void and the class $\{x\mapsto \psi\left(\xi^t g(x)\right),\xi\in\xi(\text{int}(\Phi^+))\}$ is a Glivenko-Cantelli class of functions. Assumption A4 remains specific to the model we consider. \\
The previous calculus shows that our procedure for estimating $\hat{\phi}$ can be done efficiently and the complexity of the calculus does not depend on the dimension of the data. Besides, no numerical integration is needed.
\end{example}
%%%%%%%%%%%%%%%%%%%%%%%%%%%%%%%%%%%%%%%%%%%%%%%%%%%%%%%
% =============
%%%%%%%%%%%%%%%%%%%%%%%%%%%%%%%%%%%%%%%%%%%%%%%%%%%%%%%
\subsection{Asymptotic normality}
We will suppose that the model $p_{\phi}$ is $\mathcal{C}^2(\text{int}(\Phi^+))$ and that $\psi$ is $\mathcal{C}^2(\mathbb{R})$. In order to simplify the formula below, we suppose that $\psi'(0)=1$ and $\psi''(0)=1$. These are not restrictive assumptions and can be relaxed. Recall that they are both verified in the class of Cressie-Read functions (\ref{eqn:CressieReadPhi}).\\
Define the following matrices:
\begin{eqnarray}
J_{\phi^*,\xi^*} & = & \resizebox{0.75\textwidth}{!}{ $\left( \frac{1}{(1-\lambda^*)^2}\left[-\mathbb{E}_{P_T}\left[g(X)\right] + \mathbb{E}_{P_1(.|\theta^*)}\left[g(X)\right]\right], \frac{\lambda^*}{1-\lambda^*}\int{g(x)\nabla_{\theta}p_1(x|\theta^*)dx}, \nabla m(\alpha^*) \right)$}; \label{eqn:NormalAsymMomJ1} \\
J_{\xi^*,\xi^*} & = & \mathbb{E}_{P_0^*}\left[g(X)g(X)^t\right]; \label{eqn:NormalAsymMomJ2}\\
\Sigma & = & \left(J_{\phi^*,\xi^*}^t J_{\xi^*,\xi^*} J_{\phi^*,\xi^*}\right)^{-1}; \label{eqn:NormalAsymMomSigma}\\
H & = & \Sigma J_{\phi^*,\xi^*}^t J_{\xi^*,\xi^*}^{-1}; \label{eqn:NormalAsymMomH}\\  
W & = & J_{\xi^*,\xi^*}^{-1} - J_{\xi^*,\xi^*}^{-1} J_{\phi^*,\xi^*} \Sigma J_{\phi^*,\xi^*}^t J_{\xi^*,\xi^*}^{-1}. \label{eqn:NormalAsymMomP}
\end{eqnarray}
Recall the definition of $\Phi^+$ and define similarly the set $\Phi_n^+$
\begin{eqnarray}
\Phi^+ & = & \left\{\phi: \quad \int{g(x)g(x)^t\left(\frac{1}{1-\lambda}dP_T - \frac{\lambda}{1-\lambda}dP_1\right)(x|\theta)}\text{ is s.p.d.}\right\}; \label{eqn:PhiPlus}\\
\Phi_n^+ & = & \left\{\phi: \quad \frac{1}{n}\frac{1}{1-\lambda}\sum_{i=1}^n{g(X_i)g(X_i)^t} - \frac{\lambda}{1-\lambda}\int{g(x)g(x)^tdP_1(x|\theta)} \text{ is s.p.d.} \right\} \label{eqn:PhiPlusn}.
\end{eqnarray}
These two sets are the feasible sets of parameters for the optimization problems (\ref{eqn:MomentTrueEstimProc}) and (\ref{eqn:MomentEstimProc}) respectively. In other words, outside of the set $\Phi^+$, we have $H(\phi,\xi(\phi))=\infty$. Similarly, outside of the set $\Phi_n^+$, we have $H_n(\phi,\xi_n(\phi))=\infty$.
\begin{theorem}
\label{theo:AsymptotNormalMomConstr}
Suppose that:
\begin{enumerate}
\item $\hat{\phi}$ is consistent and $\phi^*\in\text{int}\left(\Phi\right)$;
\item the function $\alpha\mapsto m(\alpha)$ is $\mathcal{C}^2$;
\item $\forall \phi\in B(\phi^*,\tilde{r})$ and any $\xi\in\xi(B(\phi^*,\tilde{r}))$, there exist functions $h_{1,1},h_{1,2}\in L^1(p_1(.|\theta))$ such that $\left\|\psi'\left(\xi^t g(x)\right)g(x)\right\|\leq h_{1,1}(x)$ and $\left\|\psi''\left(\xi^t g(x)\right)g(x)g(x)^t\right\|\leq h_{1,2}(x)$;
\item $\forall \xi\in\xi(B(\phi^*,\tilde{r}))$, there exist functions $h_{2,1},h_{2,2}\in L^1(dx)$ such that $\left\|\psi\left(\xi^t g(x)\right)\nabla_{\theta}p_1(x|\theta)\right\|\leq h_2(x)$ and $\left\|\psi\left(\xi^t g(x)\right)J_{p_1(.|\theta)}\right\|\leq h_2(x)$;
\item for any couple $(\phi,\xi)\in B(\phi^*,\tilde{r})\times\xi(B(\phi^*,\tilde{r}))$, there exists a function $h_3\in L^1(dx)$ such that $\left\|\psi'\left(\xi^t g(x)\right)g(x)\nabla_{\theta}p_1(x|\theta)^t\right\|\leq h_3(x)$;
\item finite second order moment of $g$ under $P_T$, i.e. $\mathbb{E}_{P_T}\left[g_i(X)g_j(X)\right]<\infty$ for $i,j\leq\ell$;
\item matrices $J_{\xi^*,\xi^*}$ and $J_{\phi^*,\xi^*}^t J_{\xi^*,\xi^*} J_{\phi^*,\xi^*}$ are invertible,
\end{enumerate}
then
\[\left(\begin{array}{c}  \sqrt{n}\left(\hat{\phi}-\phi^*\right) \\ \sqrt{n}\xi_n(\hat{\phi})\end{array}\right) \xrightarrow[\mathcal{L}]{} \mathcal{N}\left(0,\frac{1}{(1-\lambda^*)^2}\left(\begin{array}{c}H \\ W\end{array}\right) \text{Var}_{P_T}(g(X)) \left(H^t\quad W^t\right)\right),\]
where $H$ and $P$ are given by formulas (\ref{eqn:NormalAsymMomH}) and (\ref{eqn:NormalAsymMomP}).
\end{theorem}
The proof is differed to Appendix \ref{AppendSemiPara:Theo3}. Assumption 3 entails the differentiability of function $H_n(\xi,\phi)$ up to second order with respect to $\xi$ whatever the value of $\phi$ in a neighborhood of $\phi^*$. Assumption 4 entails the differentiability of function $H_n(\xi,\phi)$ up to second order with respect to $\theta$ in a neighborhood of $\theta^*$ inside $\xi(B(\phi^*,\tilde{r}))$. Finally, assumption 5 implies the cross-differentiability of function $H_n(\xi,\phi)$ with respect to $\xi$ and $\theta$.\\ 
Differentiability assumptions in Theorem \ref{theo:AsymptotNormalMomConstr} can be relaxed in the case of the Pearson's $\chi^2$ since all integrals in functions $H_n$ and $H$ can be calculated. Our result covers the general case and thus we need to ensure differentiability of the integrals using Lebesgue theorems which requires the existence of integrable functions which upperbound the integrands.
%\begin{remark}
%It is important to notice that the variance of the estimator becomes higher as the proportion of the parametric part becomes higher.
%\end{remark}

%%%%%%%%%%%%%%%%%%%%%%%%%%%%%%%%%%%%%%%%%%%%%%%%%%%%%%%%%%%%%%%%%%%%%%%%%%%%%%%%%
%
%==============================================================
%%%%%%%%%%%%%%%%%%%%%%%%%%%%%%%%%%%%%%%%%%%%%%%%%%%%%%%%%%%%%%%%%%%%%%%%%%%%
%==============================================================
%
%%%%%%%%%%%%%%%%%%%%%%%%%%%%%%%%%%%%%%%%%%%%%%%%%%%%%%%%%%%%%%%%%%%%%%%%%%%%%%%%%

\section{Simulation study}\label{sec:SemiParaSimulations}
We perform several simulations in univariate and multivariate situations and show how prior information about the moments of the distribution of the semiparametric component $P_0$ can help us better estimate the set of parameters $(\lambda^*,\theta^*,\alpha^*)$ in regular examples, i.e. the components of the mixture can be clearly distinguished when we plot the probability density function. We also show how our approach permits to estimate even in difficult situations when the proportion of the parametric component is very low; such cases could \emph{not} be estimated using existing methods.\\
Another important problem in existing methods is their quadratic complexity. For example, an EM-type method such as \cite{Robin}'s algorithm or its stochastic version introduced by \cite{BordesStochEM} performs
%needs at each iteration to calculate first a weighted kernel density estimator and calculate it at each observation of the sample. This calculus has a complexity of order $n^2$. We then need to calculate a vector of weights of length $n$, and estimate the proportion by averaging this vector of weights. Finally, we need to estimate the parameters of $P_1$ by maximum likelihood which can be done at the best cases by averaging $n$ terms. This means that such an algorithm needs to do at least 
$n^2+3n$ operations in order to complete a single iteration. An EM-type algorithm for semiparametric mixture models needs in average 100 iterations to converge and may attain 1000 iterations\footnote{This was the case of the Weibull mixture.} for each sample. To conclude, the estimation procedure performs at least $100(n^2+3n)$ operations. In a signal-noise situations where the signal has a very low proportion around $0.05$, we need a greater number of observations say $n=10^5$. Such experiences cannot be performed using an EM-type method such as \cite{Robin}'s algorithm or its stochastic version introduced by \cite{BordesStochEM} unless one has a "super computer". The method of \cite{Bordes10} shares similar complexity\footnote{we need more than 24 hours to estimate the parameters of one sample with $10^5$ observations.} $\mathcal{O}(n^2)$.
% because one needs to calculate a cumulative density estimator on each observation, and thus a complexity $n^2$. There is then the optimization step which needs at least 100 iterations to converge\footnote{we need more than 24 hours to estimate the parameters of one sample with $10^5$ observations.}.
Last but not least, the EM-type method of \cite{Song} and their $\pi-$maximizing one have the advantage over other methods, because we need only to calculate a kernel density estimator once and for all, then use it at each iteration\footnote{We were able to perform simulations with $n=10^5$ observations but needed about 5 days on an i7 laptop clocked at 2.5 GHz with 8GB of RAM. For \cite{Robin}'s algorithm, a few iterations took about one day. One can imagine the time needed to estimate 100 samples with $10^5$ observations in each sample.}. Nevertheless, the method has still a complexity of order $n^2$.\\
Our approach, although has a double optimization procedure, it can be implemented when $g$ is polynomial and $\varphi$ corresponds to the Pearson's $\chi^2$ in a way that it has a linear complexity $\mathcal{O}(n)$. First of all, using the $\chi^2$ divergence, the optimization over $\xi$ in (\ref{eqn:MomentEstimProc}) can be calculated directly. On the other hand, all integrals are mere calculus of empirical moments and moments of the parametric part, see Example \ref{ex:Chi2Consistency}. Empirical moments can be calculated once and for all whereas moments of the parametric part can be calculated using direct formulas available for a large class of probability distributions. What remains is the optimization over $\phi$. In the simulations below, our method produced the estimates instantly even for a number of observations of order $10^7$ whereas other existing methods needed from several hours (algorithms of \cite{Song}) to several days (for other algorithms). It is however important to notice that if the number of constraints is large enough, say a function of $n$, then we no longer have a linear complexity.\\
Because of the very long execution time of existing methods, we restricted the comparison to simulations in regular situations with $n<10^4$. Experiments with greater number of observations were only treated using our method and the methods in \cite{Song}. In all tables presented hereafter, we performed 100 experiments and calculated the average of resulting estimators. We provided also the standard deviation of the 100 experiments in order to get a measure of preference in case the different estimation methods gave close results.\\ 
Our experiments cover the following models:
\begin{itemize}
\item[$\bullet$] A two-component Weibull mixture;
\item[$\bullet$] A two-component Weibull - Lognormal mixture;
\item[$\bullet$] A two-component Gaussian -- Two-sided Weibull mixture;
\item[$\bullet$] A two-component bivariate Gaussian mixture.
\end{itemize}
We apply the several estimation methods from Section \ref{sec:LiteratureSemiparaMix}. We have chosen a variety of values for the parameters especially the proportion. The second model may represent a problem from queue theory where the left component represents the impatient customers whereas the right component represents the regular customers. The third model stems from a signal-noise application where the signal is centered at zero whereas the noise is repartitioned at both sides. The fourth model appears in clustering and is only presented to show how our method performs in multivariate contexts.\\
In all our experiments, no numerical integration was used since they can be easily calculated as functions of the empirical moments of the data and the moments of the parametric component, see Example \ref{ex:Chi2Consistency}. Simulations were done using the \cite{Rtool}. Optimization was performed using the Nelder-Mead algorithm, see \cite{NelderMead}. For the $\pi-$maximizing algorithm of \cite{Song}, we used the Brent's method because the optimization was carried over one parameter.\\ 
For our procedure, we only used the $\chi^2$ divergence, because the optimization over $\xi$ can be calculated without numerical methods\footnote{We noticed no great difference when using a Hellinger divergence.}. Recall that the optimized function over $\xi$ is not always strictly concave and the Hessian matrix may be definite positive, see remark \ref{rem:PhiPlusJacobMat}. It is thus important to check for each vector $\phi=(\lambda,\theta,\alpha)$ if the Hessian matrix is still definite negative for example using Sylvester's criterion. If it is not, we set the objective function to a value such as $10^2$. Besides, since the resulting function $\phi\mapsto H_n(\phi,\xi_n(\phi))$ as a function of $\phi$ is not ensured to be strictly convex, we used 10 random initial feasible points inside the set $\Phi_n^+$ defined by (\ref{eqn:PhiPlusn}). We then ran the Nelder-Mead algorithm and chose the vector of parameters for which the objective function has the lowest value. We applied a similar procedure on the algorithm of \cite{Bordes10} in order to ensure a \emph{good and fair} optimization.\\
\begin{remark}
In the literature on the stochastic EM algorithm, it is advised that we iterate the algorithm for some time until it reaches a stable state, then continue iterating long enough and average the values obtained in the second part. The trajectories of the algorithm were very erratic especially for the estimation of the proportion. For us, we iterated for the stochastic EM-type algorithm of \cite{BordesStochEM} 5000 times and averaged the 4000 final iterations.
\end{remark}
\begin{remark}
Initialization of both the EM-type algorithm of \cite{Song} and the SEM-type algorithm of \cite{BordesStochEM} was not very important, and we got the same results when the vector of weights was initialized uniformly or in a "good" way. The method of \cite{Robin} was more influenced by such initialization and we used most of the time a good starting points.
\end{remark}
\begin{remark}
For the methods of \cite{Song}, we need to estimate mixture's distribution using a kernel density estimator. For the data generated from a Weibull mixture and the data generated from a Weibull Lognormal mixture, we used a reciprocal inverse Gaussian kernel density estimator with a window equal to 0.01 according to our simulations in Chapter 1.
\end{remark}
\begin{remark}
Matrix inversion was done manually using direct inversion methods, because the function \texttt{solve} in the statistical program R produced errors sometimes because the matrix was highly sensible at some point during the optimization. For matrices of dimension $4\times 4$ and $5\times 5$ we used block matrix inversion, see for example \cite{BlockMat}. The inverse of a $3\times 3$ was calculated using a direct formula.
\end{remark}

%\subsection{Gamma-Lognormal mixture}
%Gamma defined through moments. Lognormal is the parametric part. Many situations (48 exps) convergence towards $\lambda=0.1$ and thus rejected the initial value and replace with another reasonable local inf even with greater value of the obj fun. Keep parameters $\nu$ and $\mu$ inside a compact set. $\mu\geq 0$ and $\nu\in(0.1,10)$ to avoid explosions.
%%\begin{figure}[ht]
%%\centering
%%\includegraphics[scale=0.5]{PosSetGammaLogNormal.png}
%%\caption{The set of positive densities in a Gamma-Lognormal mixture. The proportion of the parametric part (Lognormal) is 0.7 with meanlog $\mu=3$. The Gamma has a shape $\nu=1.5$.}
%%\label{fig:PosSetGammaLogNormal}
%%\end{figure}
%
%\begin{table}[ht]
%\centering
%\begin{tabular}{|c|c|c|c|c|c|c|}
%\hline
%Estimation method & $\lambda$ & sd$(\lambda)$ & $\mu$ & sd($\mu$) & $\nu$ & sd($\nu$)\\
%\hline
%\hline
%\multicolumn{7}{|c|}{Mixture 1 : $\lambda^* = 0.7$, $\mu^*=3$, $\sigma_2^*=0.5$(fixed), $\nu^*=1.5$, $\sigma_1^*=1$(fixed) }\\
%\hline
%Pearson's $\chi^2$ & 0.278 & 0.222 & 2.604 & 0.178 & 1.343 & 0.632 \\
%Pearson's $\chi^2$ with integration to 1 & 0.372 & 0.192 & 2.180 & 0.974 & 2.925 & 3.082 \\
%\hline
%\end{tabular}
%\caption{The mean value with the standard deviation of estimates in a 100-run experiment on a two-component Gamma-log normal mixture.}
%\label{tab:3by3ResultsGammaLogNormMoment}
%\end{table}
%%%%%%%%%%%%%%%%%%%%%%%%%%%%%%%%%%%%%%%%%%%%%%%
\subsection{Data generated from a two-component Weibull mixture modeled by a semiparametric Weibull mixture}
We consider a mixture of two Weibull components with scales $\sigma_1 = 0.5,\sigma_2=1$ and shapes $\nu_1=2,\nu_2=1$ in order to generate the dataset. In the semiparametric mixture model, the parametric component will be "the one to the right", i.e. the component whose true set of parameters is $(\nu_1=2,\sigma_1=0.5)$. We illustrate several values of the proportion $\lambda\in\{0.7,0.3\}$, see figure (\ref{fig:WeibullMixtures}). This constitutes a difficult example for both our method and existing methods such as EM-type methods or the $\pi-$maximizing algorithm of \cite{Song}. We therefore, simulated 10000 samples and fixed both scales during estimation. We estimate the proportion and the shapes of both components. For our method, the variance of the estimator of $\nu_1$ was high and we needed to use 4 moments to reduce it to an acceptable range. Of course, as the number of observations increases, the variance reduces. We, however, avoided greater number of observations because methods such as \cite{Robin} need very long execution time for even one sample. The method of \cite{Bordes10} cannot be applied here since the support of the mixture is $\mathbb{R}_+$.\\
Moments of the Weibull distribution are given by:
\[\mathbb{E}[X^i] = \sigma^i\Gamma(1+i/\nu),\qquad\forall i\in\mathbb{N}.\]
For our method , function $g$ is a vector of polynomials that is $g(x)=(1,x,x^2,x^3)^t$ when we use only 3 moment constraints. It suffices to add $x^4$ when the number of constraints is 4. On the other hand, $m(\alpha)=(1,\sigma\Gamma(1+1/\nu),\sigma^2\Gamma(1+2/\nu),\sigma^3\Gamma(1+3/\nu))$. The results of our method are clearly better than existing methods which practically failed and could not see but one main component with shape in between the two shapes, see table (\ref{tab:3by3ResultsWeibullMoment}). Although our method presents an inconvenient greater variance for $\nu_1$, the Monte-Carlo mean of the hundred experiences is still unbiased. We believe that the use of other types of constraints would have resulted in better results without the need to add one more constraint. 
\begin{table}[ht]
\centering
\begin{tabular}{|c|c|c|c|c|c|c|}
\hline
Nb of observations & $\lambda$ & sd$(\lambda)$ & $\nu_1$ & sd($\nu_1$) & $\nu_2$ & sd($\nu_2$)\\
\hline
\hline
\multicolumn{7}{|c|}{Mixture 1 : $n=10^4$ $\lambda^* = 0.7$, $\nu_1^*=2$, $\sigma_1^*=0.5$(fixed), $\nu_2^*=1$, $\sigma_2^*=1$(fixed) }\\
\hline
Pearson's $\chi^2$ 3 moments & 0.700 & 0.010 & 2.006 & 0.217 & 1.005 & 0.024 \\
Pearson's $\chi^2$ 4 moments & 0.701 & 0.010 & 2.014 & 0.086 & 1.013 & 0.024 \\
Robin & 0.654 & 0.101 & 1.591 & 0.085 & --- & --- \\
Song EM-type & 0.907 & 0.004 & 1.675 & 0.020 & --- & --- \\
Song $\pi-$maximizing & 0.782 & 0.006 & 1.443 & 0.012 & --- & --- \\
\hline
\hline
\multicolumn{7}{|c|}{Mixture 1 : $n=10^4$ $\lambda^* = 0.3$, $\nu_1^*=2$, $\sigma_1^*=0.5$(fixed), $\nu_2^*=1$, $\sigma_2^*=1$(fixed) }\\
\hline
Pearson's $\chi^2$ 3 moments & 0.304 & 0.016 & 2.191 & 0.887 & 0.998 & 0.013 \\
Pearson's $\chi^2$ 4 moments & 0.303 & 0.016 & 2.120 & 0.285 & 1.001 & 0.013 \\
Robin & 0.604 & 0.029 & 1.256  & 0.037 & --- & --- \\
Song EM-type & 0.806 & 0.005 & 1.185 & 0.018 & --- & --- \\
Song $\pi-$maximizing & 0.624 & 0.007 & 1.312 & 0.013 & --- & --- \\
\hline
\end{tabular}
\caption{The mean value with the standard deviation of estimates in a 100-run experiment on a two-component Weibull mixture.}
\label{tab:3by3ResultsWeibullMoment}
\end{table}

\begin{figure}[ht]
\centering
\includegraphics[scale=0.3]{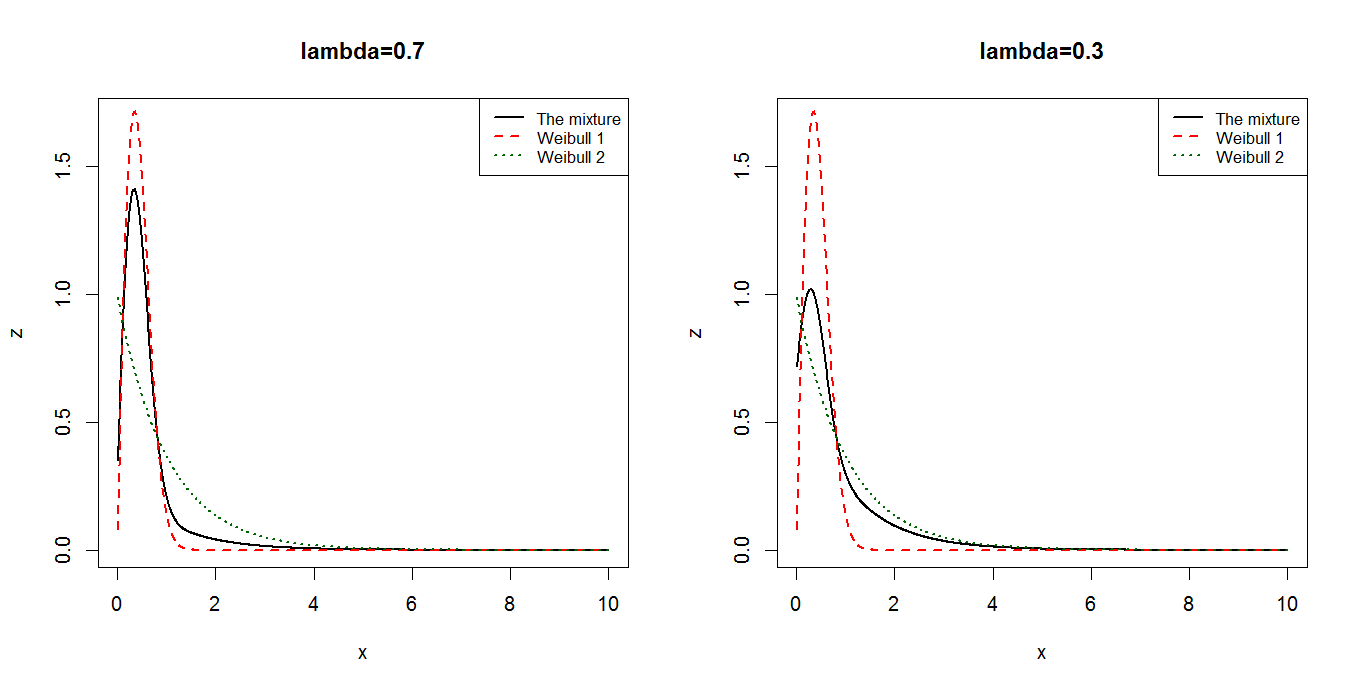}
\caption{The Weibull mixtures, see table (\ref{tab:3by3ResultsWeibullMoment})}
\label{fig:WeibullMixtures}
\end{figure}
\clearpage
%%%%%%%%%%%%%%%%%%%%%%%%%%%%%%%%%%%%%%%%%%%%%%%%
\subsection{Data generated from a two-component Weibull-LogNormal mixture modeled by a semiparametric Weibull-LogNormal mixture}\label{subsec:WeibLognormMom}
We consider a mixture of two components; a Weibull component with scale $=1$ and shape $=1.5$, and a log-normal component with meanlog $=3$ and scale $=0.5$, see figure (\ref{fig:WeibullLogNormalMixtures}) for the considered cases. In the results of table (\ref{tab:ConsistWeibullLognormMoment}), the parametric part is considered to be the Lognormal distribution. In the results of table (\ref{tab:3by3ResultsWeibullLognormMoment}), the parametric part of the semiparametric model is considered to be the Weibull distribution.\\
The number of observations for each mixture depends on its subtlety. As the proportion of the parametric component becomes lower, we needed more observations to produce reasonable estimates. We chose the number of observations in a way that the standard deviation of estimated parameters does not exceed 1.\\
We used the first three moments such that we can estimate three parameters; the proportion of the parameteric component, the shape of the Weibull component and the mean-parameter of the Lognormal component. Thus, the scales of both components are supposed to be known during the estimation procedure. The moments of these two distributions are given by:
\begin{eqnarray*}
\text{Weibull:}\qquad \mathbb{E}[X^i] & = &  \sigma^i\Gamma(1+i/\nu); \\
\text{Lognormal:}\qquad \mathbb{E}[X^i] & = & e^{i\mu+i^2\sigma^2/2}.
\end{eqnarray*}
For our method , function $g$ is a vector of polynomials that is $g(x)=(1,x,x^2,x^3)^t$. Function $m(\alpha)$ is given by $(1,\mathbb{E}[X],\mathbb{E}[X^2],\mathbb{E}[X^3])^t$ with the corresponding moments according to whether we consider the Weibull or the log-normal as the semiprametric component.\\
Our new method seems to produce high variance of the shape of the Weibull component. This should not be surprising, because the part which influences on the moments of the model is the Lognormal component. Its moments have an exponential form and small differences in the mean-parameter could compensate for a great differences in the shape of the Weibull component. The results are still satisfactory since we get to estimate an information of the semiparametric component at a great precision together with the proportion.
\begin{table}[ht]
\centering
\begin{tabular}{|c|c|c|c|c|c|c|}
\hline
Nb of observations & $\lambda$ & sd$(\lambda)$ & $\mu$ & sd($\mu$) & $\nu$ & sd($\nu$)\\
\hline
\hline
\multicolumn{7}{|c|}{Mixture 1 : $\lambda^* = 0.7$, $\mu^*=3$, $\sigma_2^*=0.5$(fixed), $\nu^*=1.5$, $\sigma_1^*=1$(fixed) }\\
\hline
$n=10^2$ & 0.384 & 0.117 & 2.654 & 0.153 & 0.488 & 0.018 \\
$n=10^3$ & 0.518 & 0.068 & 2.806 & 0.099 & 0.473 & 0.014 \\
$n=10^4$ & 0.605 & 0.044 & 2.903 & 0.069 & 0.531 & 0.326 \\
$n=10^5$ & 0.651 & 0.030 & 2.957 & 0.041 & 0.809 & 0.630 \\
$n=10^6$ & 0.682 & 0.018 & 2.979 & 0.022 & 1.638 & 0.813 \\
\hline
\end{tabular}
\caption{The mean value with the standard deviation of estimates produced by our procedure with three moments constraints in a 100-run experiment on a two-component Weibull-- Lognormal mixture. The parametric component is the Lognormal distribution.}
\label{tab:ConsistWeibullLognormMoment}
\end{table}

\begin{table}[ht]
\centering
\begin{tabular}{|c|c|c|c|c|c|c|}
\hline
Nb of observations & $\lambda$ & sd$(\lambda)$ & $\nu$ & sd($\nu$) & $\mu$ & sd($\mu$)\\
\hline
\hline
\multicolumn{7}{|c|}{Mixture 1 : $n=10^3$, $\lambda^* = 0.3$, $\nu^*=1.5$, $\sigma_1^*=1$(fixed), $\mu^*=3$, $\sigma_2^*=0.5$(fixed) }\\
\hline
%$n=10^2$ & 0.315 & 0.051 & 1.775 & 1.261 & 2.946 & 0.132 \\
Pearson's $\chi^2 $& 0.308 & 0.017 & 1.484 & 0.624 & 3.002 & 0.026 \\
Robin & 0.296 & 0.015 & 1.557 & 0.068 & --- & --- \\
Song EM-type & 0.291 & 0.015 & 1.614 & 0.087 & --- & --- \\
Song $\pi-$maximizing & 0.230 & 0.022 & 1.662 & 0.251 & --- & --- \\
SEM & 0.284 & 0.041 & 1.570 & 0.263 & --- & ---\\
\hline
\hline
\multicolumn{7}{|c|}{Mixture 2 : $n=10^4$, $\lambda^* = 0.1$, $\nu^*=1$, $\sigma_1^*=1$(fixed), $\mu^*=3$, $\sigma_2^*=0.5$(fixed) }\\
\hline
Pearson's $\chi^2$ & 0.103 & 0.006 & 1.284 & 0.677 & 3.001 & 0.007 \\
Robin & 0.095 & 0.003 & 1.049 & 0.031 & --- & --- \\
%Song EM-type (true init) & 0.102 & 0.004 & 0.941 & 0.037 & --- & --- \\
%Song EM-type (0.2,1.5 init) & 0.100 & 0.004 & 0.894 & 0.039 & --- & --- \\
Song EM-type & 0.100 & 0.004 & 0.894 & 0.039 & --- & --- \\
Song $\pi-$maximizing & 0.085 & 0.005 & 1.024 & 0.055 & --- & --- \\
SEM & 0.094 & 0.015 & 1.054 & 0.228 & --- & --- \\
\hline
\hline
\multicolumn{7}{|c|}{Mixture 3 : $n=10^4$, $\lambda^* = 0.05$, $\nu^*=1$, $\sigma_1^*=1$(fixed), $\mu^*=3$, $\sigma_2^*=0.5$(fixed) }\\
\hline
Pearson's $\chi^2$ & 0.052 & 0.004 & 1.312 & 0.703 & 3.001 & 0.006 \\
Song EM-type & 0.050 & 0.003 & 0.855 & 0.068 & --- & --- \\
Song $\pi-$maximizing & 0.042 & 0.003 & 1.013 & 0.052 & --- & --- \\
\hline
\hline
\multicolumn{7}{|c|}{Mixture 4 : $n=5\times 10^4$, $\lambda^* = 0.05$, $\nu^*=0.4$, $\sigma_1^*=1$(fixed), $\mu^*=3$, $\sigma_2^*=0.5$(fixed) }\\
\hline
Pearson's $\chi^2$ & 0.049 & 0.002 & 0.629 & 0.438 & 3.001 & 0.004 \\
Song EM-type & 0.064 & 0.001 & 0.345 & 0.004 & --- & --- \\
Song $\pi-$maximizing & 0.024 & 0.001 & 0.773 & 0.010 & --- & --- \\
\hline
\end{tabular}
\caption{The mean value with the standard deviation of estimates in a 100-run experiment on a two-component Weibull-log normal mixture. The parametric component is the Weibull distribution.}
\label{tab:3by3ResultsWeibullLognormMoment}
\end{table}

\begin{figure}[ht]
\centering
\includegraphics[scale=0.4]{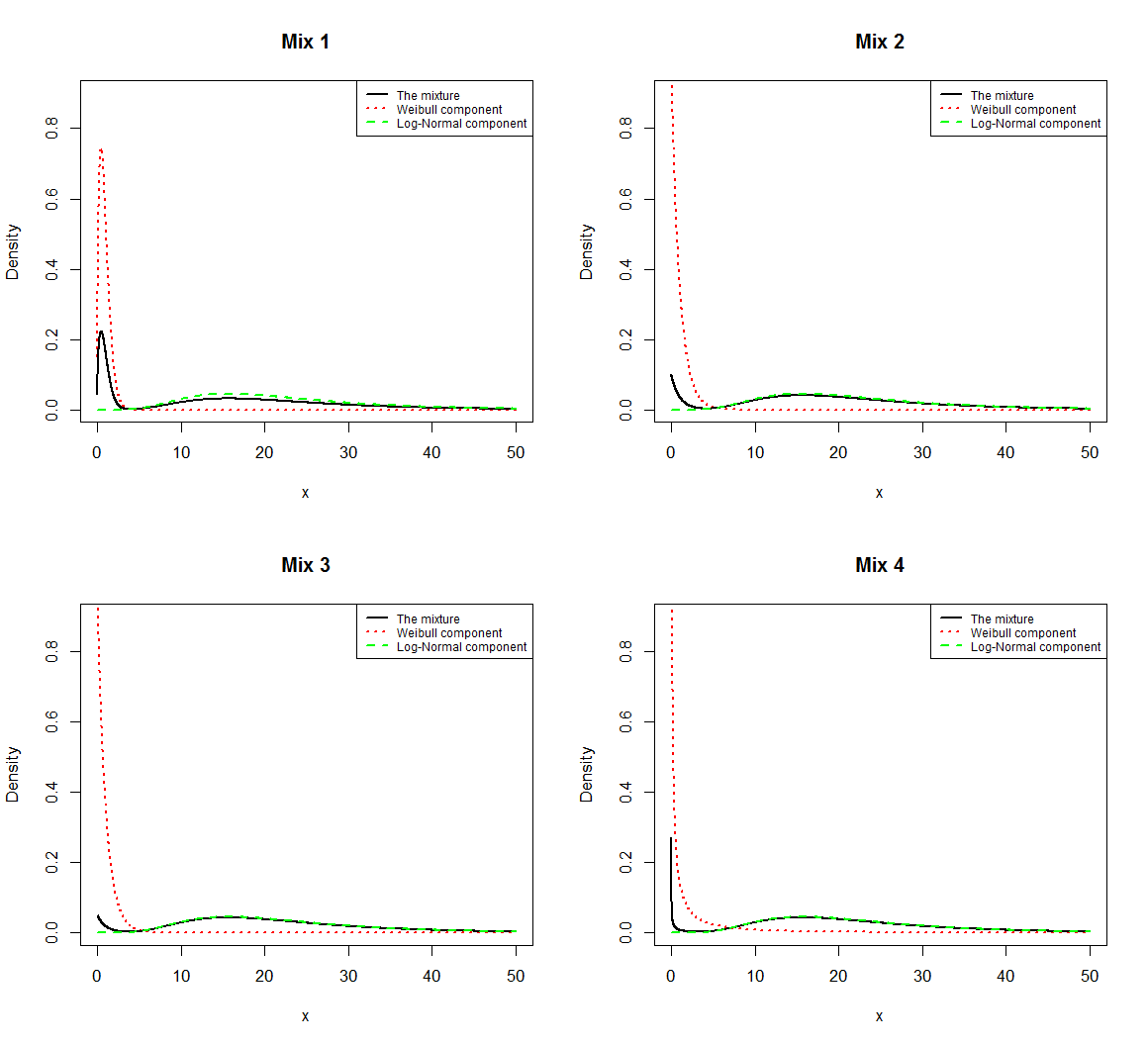}
\caption{The Weibull -- Lognormal mixtures, see tables (\ref{tab:ConsistWeibullLognormMoment},\ref{tab:3by3ResultsWeibullLognormMoment})}
\label{fig:WeibullLogNormalMixtures}
\end{figure}
\clearpage
%%%%%%%%%%%%%%%%%%%%%%%%%%%%%%%%%%%%%%%%%%%%%%
\subsection{Data generated from a two-sided Weibull Gaussian mixture modeled by a semiparametric two-sided Weibull Gaussian mixture}\label{subsec:TwoSidGaussMom}
The (symmetric) two-sided Weibull distribution can be considered as a generalization of the Laplace distribution and can be defined through either its density or its distribution function as follows:
\[f(x|\nu,\sigma) = \frac{1}{2}\frac{\sigma}{\nu}\left(\frac{|x|}{\sigma}\right)^{\nu-1}e^{-\left(\frac{|x|}{\sigma}\right)^{\nu}}, \qquad \mathbb{F}(x|\nu,\sigma) = \left\{ \begin{array}{cc} 1-\frac{1}{2}e^{-\left(\frac{x}{\sigma}\right)^{\nu}} & x\geq 0 \\ 
e^{-\left(\frac{-x}{\sigma}\right)^{\nu}} & x< 0\end{array} \right.\]
We can also define a skewed form of the two-sided Weibull distribution by attributing different scale and shape parameters to the positive and the negative parts, and then normalizing in a suitable way so that $f(x)$ integrates to one; see \cite{Chen2sideWeibull}. The moments of the symmetric two-sided Weibull distribution we consider here are given by:
\begin{eqnarray*}
\mathbb{E}[X^{2k}] & = & \sigma^{2k}\Gamma(1+2k/\nu) \\
\mathbb{E}[X^{2k+1}] & = & 0, \forall k\in\mathbb{N}.
\end{eqnarray*}
\begin{figure}[h]
\centering
\includegraphics[scale=0.4]{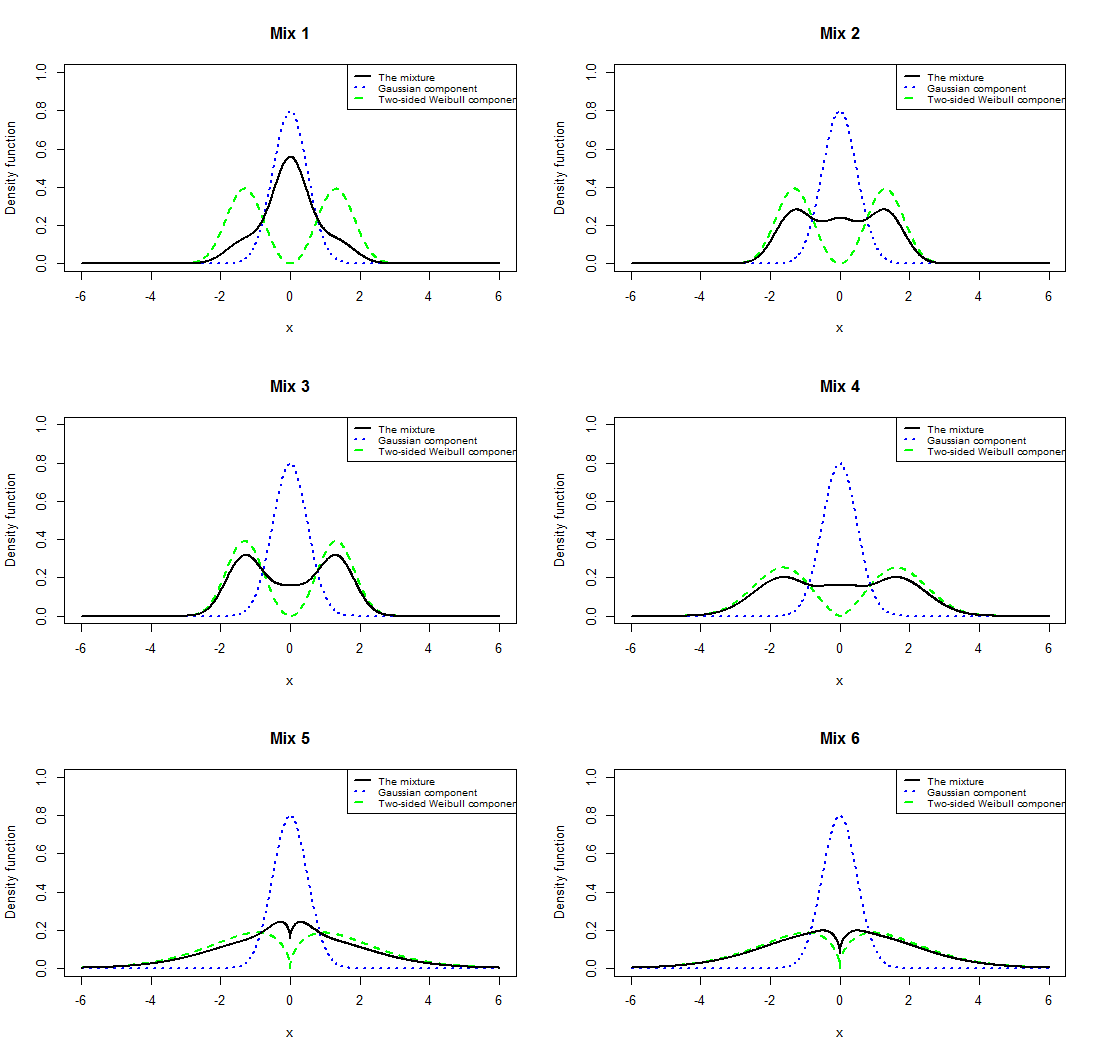}
\caption{Mixtures of two-sided Weibull -- Gaussian with low and high proportion of the parametric part. See table (\ref{tab:3by3ResultsTwoSideWeibullGaussMom})}
\label{fig:TwoSidedWeibullGaussMixure}
\end{figure}

\noindent We simulate different samples from a two-component mixture with a parametric component $f_1$ a Gaussian $\mathcal{N}(\mu=0,\sigma=\sigma_1=0.5)$ and a semiparametric component $f_0$ a (symmetric) two-sided Weibull distribution with parameters $\nu\in\{3,2.5,1.5\}$ and a scale $\sigma_0\in\{1.5,2\}$, see figure (\ref{fig:TwoSidedWeibullGaussMixure}) for different choices of the proportion. We perform different experiments to estimate the proportion and the mean of the parametric part (the Gaussian) and the shape of the semiparametric component. The values of the scale of the two components are considered to be known during estimation. We consider the following two sets of constraints:
\begin{eqnarray*}
\mathcal{M}_{1:3} & = & \left\{f_0:\int_{\mathbb{R}}{f_0(x)dx}=1,\mathbb{E}_{f_0}[X] = 0, \mathbb{E}_{f_0}[X^2]=\sigma_0^2\Gamma(1+2/\nu), \mathbb{E}_{f_0}[X^3]=0, \nu>0\right\}; \\
\mathcal{M}_{2:4} & = & \left\{f_0: \int_{\mathbb{R}}{f_0(x)dx}=1,\mathbb{E}_{f_0}[X^2]=\sigma_0^2\Gamma(1+2/\nu), \mathbb{E}_{f_0}[X^3]=0,\right. \\
&  & \text{\hspace{8cm}} \left.  \mathbb{E}_{f_0}[X^4]=\sigma_0^4\Gamma(1+4/\nu), \nu>0\right\}.
\end{eqnarray*}
The first set imposes that the semiparametric component is centered around zero whereas the second one does not impose it. \\
The first set of constraints is not really suitable for estimation especially when the number of observations is high enough. The reason is simple and is based on the original idea behind our procedure, see paragraph \ref{subsec:procIntrod}. The first and the third moment constraints are practically the same constraint. Indeed, the number of models of the form $\frac{1}{1-\lambda}f(.) - \frac{\lambda}{1-\lambda}f_1(x|\theta)$ verifying the constraints of $\mathcal{M}_{1:3}$ is infinite because the first and the third constraints give rise to the following equations:
\begin{eqnarray*}
\lambda\mu & = & 0 \\
\lambda\mu\left(\mu^2+3\sigma_1^2\right) & = & 0.
\end{eqnarray*}
The zero in the right hand side comes from the fact that the first and the third true moments of the whole mixture are zero. These are two equations in $\lambda$ and $\mu$ (since $\sigma_1$ is supposed to be known) with infinite number of solutions $(\mu,\lambda)\in\{0\}\times[0,1]$. This entails that theoretically, there is an infinite number of models of the form $\frac{1}{1-\lambda}f(.) - \frac{\lambda}{1-\lambda}f_1(x|\theta)$ in the intersection $\mathcal{N}\cap M_{1:3}$. Still, the empirical version of these equations is
\begin{eqnarray*}
\lambda\mu & = & \frac{1}{n}\sum_{i=1}^n{X_i}; \\
\lambda\mu\left(\mu^2+3\sigma_1^2\right) & = & \frac{1}{n}\sum_{i=1}^n{X_i^3}.
\end{eqnarray*}
As the number of observations is very small, the right hand side of both equations is biased enough from zero and it is highly possible that the number of solutions becomes not only finite but reduced to one. As the number of observations increases, the law of large numbers implies directly that the right hand side becomes arbitrarily close to zero and the set of solutions becomes infinite. This is exactly what happened in the simulation results in table (\ref{tab:3by3ResultsTwoSideWeibullGaussMom}) below. The algorithm favored the value zero for the estimate of the proportion as the true proportion of the parametric component became close to zero, whereas the estimates of the mean took values very dispersed centered around zero but with a high standard deviation. The set of constraints $\mathcal{M}_{2:4}$ gave clear better results even for very low proportions. On the other hand, our method outperforms other semiparametric algorithms without prior information especially when the proportion of the parameteric component is low. This shows once more the interest of incorporating a prior information in the estimation procedure.
\begin{table}[ht]
\centering
\resizebox{\textwidth}{!}{ 
\begin{tabular}{|c|c|c|c|c|c|c|}
\hline
Estimation method & $\lambda$ & sd$(\lambda)$ & $\mu$ & sd($\mu$) & $\nu$ & sd($\nu$)\\
\hline
\hline
\multicolumn{7}{|c|}{Mixture 1 :$n=100$ $\lambda^* = 0.7$, $\mu^*=0$, $\sigma_1^*=0.5$(fixed), $\nu^*=3$, $\sigma_0^*=1.5$(fixed) }\\
\hline
Pearson's $\chi^2$ under $\mathcal{M}_{1:3}$ & 0.713 & 0.064 & -0.0003 & 0.085 & 4.315  & 0.118 \\
Pearson's $\chi^2$ under $\mathcal{M}_{2:4}$ & 0.764 & 0.067 & -0.012 & 0.342 & 2.893  & 0.731 \\
Bordes symmetry Triangular Kernel & 0.309 & 0.226 & 0.240 & 0.609 & $\mu_2=-$0.220 & sd$(\mu_2)=$0.398 \\
Bordes symmetry Gaussian Kernel & 0.211 & 0.133 & 0.106 & 0.533 & $\mu_2=-$0.035 & sd$(\mu_2)=$0.203 \\
Robin et al. & 0.488 & 0.137 & -0.005 & 0.114 & --- & --- \\
EM-type Song et al. & 0.762 & 0.040 & -0.005 & 0.092 & --- & --- \\
$\pi-$maximizing Song et al. & 0.717 & 0.156 & -0.161 & 2.301 & --- & --- \\
Stochastic EM & 0.539 & 0.083 & -0.005 & 0.112 & --- & --- \\
\hline
\hline
\multicolumn{7}{|c|}{Mixture 2 :$n=100$ $\lambda^* = 0.3$, $\mu^*=0$, $\sigma_1^*=0.5$(fixed), $\nu^*=3$, $\sigma_0^*=1.5$(fixed) }\\
\hline
Pearson's $\chi^2$ under $\mathcal{M}_{1:3}$ & 0.333 & 0.079 & 0.001 & 0.316 & 4.243  & 0.442 \\
Pearson's $\chi^2$ under $\mathcal{M}_{2:4}$ & 0.407 & 0.077 & 0.012 & 0.575 & 2.925  & 0.454 \\
Bordes symmetry Triangular Kernel & 0.272 & 0.119 & 0.773 & 0.947 & $\mu_2=-$0.430 & sd$(\mu_2)=$0.393 \\
Bordes symmetry Gaussian Kernel & 0.206 & 0.104 & 0.855 & 0.911 & $\mu_2=-$0.308 & sd$(\mu_2)=$0.350 \\
Robin et al. & 0.203 & 0.078 & -0.109 & 0.947 & --- & --- \\
EM-type Song et al. & 0.494 & 0.035 & -0.132 & 0.806 & --- & --- \\
$\pi-$maximizing Song et al. & 0.384 & 0.129 & 0.014 & 1.321 & --- & --- \\
Stochastic EM & 0.263 & 0.040 & -0.062 & 0.646 & --- & --- \\
\hline
\hline
\multicolumn{7}{|c|}{Mixture 3 :$n=300$ $\lambda^* = 0.2$, $\mu^*=0$, $\sigma_1^*=0.5$(fixed), $\nu^*=3$, $\sigma_0^*=1.5$(fixed) }\\
\hline
Pearson's $\chi^2$ under $\mathcal{M}_{1:3}$ & 0.200 & 0.058 & 0.004 & 0.215 & 4.058  & 0.684 \\
Pearson's $\chi^2$ under $\mathcal{M}_{2:4}$ & 0.252 & 0.055 & 0.069 & 0.573 & 2.932  & 0.200 \\
Bordes symmetry Triangular Kernel & 0.439 & 0.108 & -0.972 & 0.328 & $\mu_2=$1.036 & sd$(\mu_2)=$0.496 \\
Bordes symmetry Gaussian Kernel & 0.414 & 0.096 & -0.928 & 0.289 & $\mu_2=-$1.125 & sd$(\mu_2)=$0.470 \\
Robin et al. & 0.278 & 0.068 & -0.062 & 1.253 & --- & --- \\
EM-type Song et al. & 0.461 & 0.023 & 0.162 & 1.128 & --- & --- \\
$\pi-$maximizing Song et al. & 0.362 & 0.020 & 0.025 & 1.224 & --- & --- \\
Stochastic EM & 0.292 & 0.057 & 0.118 & 1.027 & --- & --- \\
\hline
\hline
\multicolumn{7}{|c|}{Mixture 4 :$n=10^5$ $\lambda^* = 0.2$, $\mu^*=0$, $\sigma_1^*=0.5$(fixed), $\nu^*=2.5$, $\sigma_0^*=2$(fixed) }\\
\hline
Pearson's $\chi^2$ under $\mathcal{M}_{1:3}$ & 0.161 & 0.010 & -0.002 & 0.019 & 3.874  & 0.661 \\
Pearson's $\chi^2$ under $\mathcal{M}_{2:4}$ & 0.203 & 0.004 & -0.018 & 0.213 & 2.492  & 0.012 \\
EM-type Song et al. & 0.325 & 0.012 & -0.061 & 1.469 & --- & --- \\
$\pi-$maximizing Song et al. & 0.251 & 0.002 & -0.158 & 1.592 & --- & --- \\
\hline
\hline
\multicolumn{7}{|c|}{Mixture 5 :$n=10^5$ $\lambda^* = 0.2$, $\mu^*=0$, $\sigma_1^*=0.5$(fixed), $\nu^*=1.5$, $\sigma_0^*=2$(fixed) }\\
\hline
Pearson's $\chi^2$ under $\mathcal{M}_{1:3}$ & 0.015 & 0.030 & 0.203 & 2.381 & 2.150  & 0.138 \\
Pearson's $\chi^2$ under $\mathcal{M}_{2:4}$ & 0.213 & 0.013 & -0.004 & 0.436 & 1.494  & 0.009 \\
EM-type Song et al. & 0.397 & 0.002 & 0.001 & 0.021 & --- & --- \\
%\hline
%\hline
%\multicolumn{7}{|c|}{Mixture 6 :$n=10^5$ $\lambda^* = 0.1$, $\mu^*=0$, $\sigma_2^*=0.5$(fixed), $\nu^*=1.5$, $\sigma_1^*=2$(fixed) }\\
%\hline
%Pearson's $\chi^2$ under $\mathcal{M}_{1:3}$ & 0.0097 & 0.013 & 0.344 & 2.754 & 1.712  & 0.042 \\
%Pearson's $\chi^2$ under $\mathcal{M}_{2:4}$ & 0.116 & 0.014 & -0.004 & 0.666 & 1.492 & 0.008 \\
\hline
\hline
\multicolumn{7}{|c|}{Mixture 6 :$n=10^5$ $\lambda^* = 0.05$, $\mu^*=0$, $\sigma_1^*=0.5$(fixed), $\nu^*=1.5$, $\sigma_0^*=2$(fixed) }\\
\hline
Pearson's $\chi^2$ under $\mathcal{M}_{1:3}$ & 0.005 & 0.033 & -0.105 & 2.693 & 1.581  & 0.056 \\
Pearson's $\chi^2$ under $\mathcal{M}_{2:4}$ & 0.066 & 0.013 & -0.036 & 0.857 & 1.493 & 0.008 \\
EM-type Song et al. & 0.304 & 0.014 & -0.030 & 0.910 & --- & --- \\
$\pi-$maximizing Song et al. & 0.231 & 0.002 & 0.017 & 0.801 & --- & ---\\
\hline
\hline
\multicolumn{7}{|c|}{Mixture 7 :$n=10^7$ $\lambda^* = 0.05$, $\mu^*=0$, $\sigma_1^*=0.5$(fixed), $\nu^*=1.5$, $\sigma_0^*=2$(fixed) }\\
\hline
Pearson's $\chi^2$ under $\mathcal{M}_{1:3}$ & 0.006 & 0.010 & 0.024 & 0.197 & 1.500  & 0.019 \\
Pearson's $\chi^2$ under $\mathcal{M}_{2:4}$ & 0.051 & 0.001 & 0.002 & 0.259 & 1.500  & 0.001 \\
\hline
\hline
\multicolumn{7}{|c|}{Mixture 8 :$n=10^7$ $\lambda^* = 0.01$, $\mu^*=0$, $\sigma_1^*=0.5$(fixed), $\nu^*=1.5$, $\sigma_0^*=2$(fixed) }\\
\hline
Pearson's $\chi^2$ under $\mathcal{M}_{1:3}$ & 0.005 & 0.002 & -0.011 & 0.162 & 1.509  & 0.004 \\
Pearson's $\chi^2$ under $\mathcal{M}_{2:4}$ & 0.011 & 0.001 & -0.013 & 0.594 & 1.499  & 0.001 \\
\hline
\end{tabular}}
\caption{The mean value with the standard deviation of estimates in a 100-run experiment on a two-component two-sided Weibull--Gaussian mixture.}
\label{tab:3by3ResultsTwoSideWeibullGaussMom}
\end{table}
\clearpage
%%%%%%%%%%%%%%%%%%%%%%%%%%%%%%%%%%
\subsection{Data generated from a bivariate Gaussian mixture and modeled by a semiparametric bivariate Gaussian mixture}
We generate 1000 i.i.d. observations from a bivariate Gaussian mixture with proportion $\lambda=0.7$ for the parametric component. The parametric component is a bivariate Gaussian with mean $(0,-1)$ and covariance matrix $I_2$. The unknown component is a bivariate Gaussian with mean $(3,3)$ and covariance matrix:
\[\Sigma_2 = \left(\begin{array}{cc}{\sigma_2^*}^2 & \rho^* \\ \rho^* & {\sigma_2^*}^2\end{array}\right),\qquad {\sigma_2^*}^2 = 0.5,\quad \rho^* \in\{0, 0.25\}.\]
In a first experiment, we suppose that we know the whole parametric component, and that the unknown component belongs to the set $\mathcal{M}_1$
\begin{equation*}
\mathcal{M}_1 = \left\{\int_{\mathbb{R}^2}{f_0(x,y)dxdy}=1,\quad \int_{\mathbb{R}^2}{xf_0(x,y)dxdy}=\int_{\mathbb{R}^2}{yf_0(x,y)dydx}=\theta,\quad \theta\in\mathbb{R}\right\}.
\end{equation*}
We suppose that the only unknown parameters are the center of the unknown cluster $(\theta,\theta)$ and the proportion of the parametric component.\\
In a second experiment, we suppose that the center of the parametric component is unknown but given by $(\mu,\mu-1)$ for some unknown $\mu\in\mathbb{R}$. The set of constraints is now replaced with $\mathcal{M}_2$ given by
\begin{eqnarray*}
\mathcal{M}_2 & = &  \left\{\int_{\mathbb{R}^2}{f_0(x,y)dxdy}=1,\quad \int_{\mathbb{R}^2}{xf_0(x,y)dxdy}=\int_{\mathbb{R}^2}{yf_0(x,y)dydx}=\theta, \right.\\
 &  & \left.\qquad \qquad \qquad \int_{\mathbb{R}^2}{xyf_0(x,y)dxdy}=\theta^2+\rho^*,\theta\in\mathbb{R}\right\}.
\end{eqnarray*}
The covariance between the two coordinates $\rho^*$ in the unknown component is supposed to be known. We tested two values for $\rho^*= 0$ and $\rho^*=0.25$, see figure (\ref{fig:BivariateGaussMix}).\\
Although existing methods were only proposed for univariate cases, we see no problem in using them in multivariate cases without any changes. The only method which cannot be used directly is the method of \cite{Bordes10} because it is based on the symmetry of the density function, so it remained out of the competition.\\
For methods which use a kernel estimator, we used a kernel estimator for each coordinate of the random observations, i.e. $K_{w_x,w_y}(x,y) = K_{w_x}(x)K_{w_y}(y)$. The EM-type algorithm of \cite{Song} performs as good as our algorithm. The SEM algorithm of \cite{BordesStochEM} gives also good results. The algorithm of \cite{Robin} and the $\pi-$maximizing algorithm of \cite{Song} failed to give satisfactory results.
\begin{figure}[h]
\centering
\includegraphics[scale=0.35]{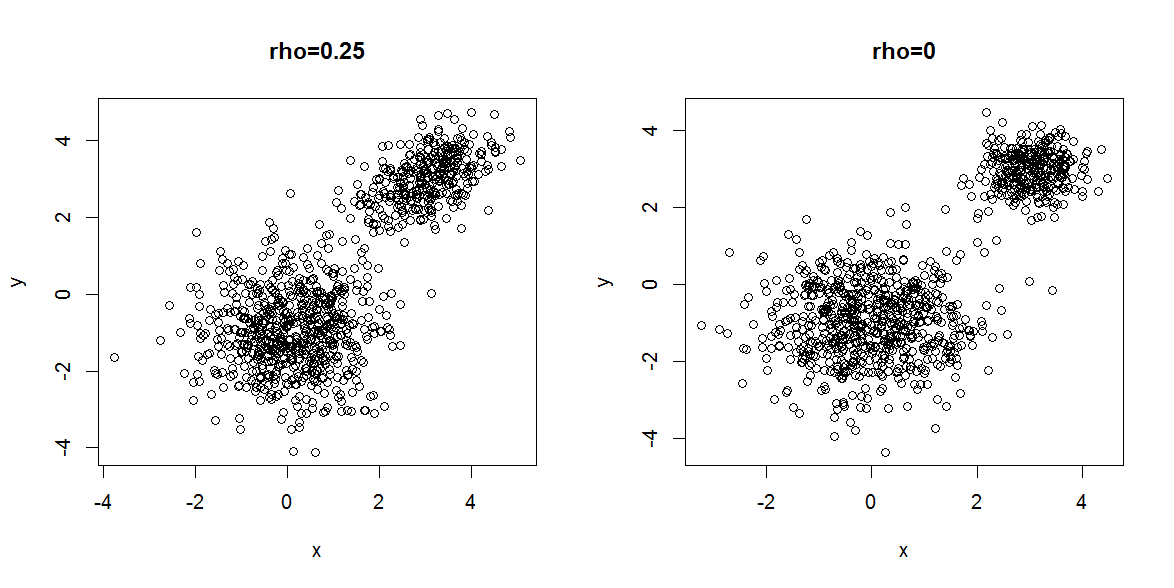}
\caption{The two bivariate Gaussian mixtures.}
\label{fig:BivariateGaussMix}
\end{figure}

\begin{table}[h]
\centering
\begin{tabular}{|c|c|c|c|c|c|c|}
\hline
Estimation method & $\lambda$ & sd$(\lambda)$ & $\mu$ & sd($\mu$) & $\theta$ & sd($\theta$)\\
\hline
\hline
\multicolumn{7}{|c|}{\bf{Mixture 2} : $\rho^*=0$ and $\mu_1 = (\mu,1-\mu)$ is unknown }\\
\hline
Pearson's $\chi^2$ under $\mathcal{M}_1$ & 0.680 & 0.027 & --- & --- & 2.854 & 0.233 \\
Pearson's $\chi^2$ under $\mathcal{M}_2$ & 0.694 & 0.019 & 0.016 & 0.035 & 3.034 & 0.045 \\
SEM  & 0.724 & 0.015 & 0.090 & 0.043 & $\mu_{1,2}=$ -0.880 & sd$(\mu_{1,2})=$0.053 \\
Robin & 0.954 & 0.064 & 0.779 & 0.212 & $\mu_{1,2}=$ -0.221 & sd$(\mu_{1,2})=$0.218 \\
Song EM  & 0.697 & 0.014 & 0.003 & 0.038 & $\mu_{1,2}=$ -0.996 & sd$(\mu_{1,2})=$0.039 \\
Song $\pi-$maximizing & 0.114 & 0.297 & 0.538 & 1.810 & $\mu_{1,2}=$ -0.463 & sd$(\mu_{1,2})=$1.810 \\
\hline
\hline
\multicolumn{7}{|c|}{\bf{Mixture 3} : $\rho^*=0.25$ and $\mu_1 = (\mu,1-\mu)$ is unknown }\\
\hline
Pearson's $\chi^2$ under $\mathcal{M}_2$ & 0.704 & 0.026 & 0.033 & 0.060 & 3.071 & 0.101 \\
SEM & 0.730 & 0.016 & 0.083 & 0.052 & $\mu_{1,2}=$-0.878 & sd$(\mu_{1,2})=$0.055 \\
Robin & 0.890 & 0.025 & 0.566 & 0.117 & $\mu_{1,2}=$ -0.434 & sd$(\mu_{1,2})=$0.117 \\
Song EM & 0.704 & 0.015 & 0.016 & 0.047 & $\mu_{1,2}=$ -0.973 & sd$(\mu_{1,2})=$0.040 \\
Song $\pi-$maximizing & 0.095 & 0.268 & 0.564 & 1.606 & $\mu_{1,2}=$ -0.436 & sd$(\mu_{1,2})=$1.606 \\
\hline

\end{tabular}
\caption{The mean value with the standard deviation of estimates in a 100-run experiment on a two-component bivariate normal mixture.}
\label{tab:3by3ResultsBivaraiteGaussMom}
\end{table}
\clearpage

% ==========================================================
% -------------------------------------------------
%%%%%%%%%%%%%%%%%%%%%%%%%%%%%%%%%%%%%%%%%%%%%%%%%%%%%%%%%%%%%%%%%%%%%%%
% -------------------------------------------------
% ==========================================================
\section{Conclusions}
In this chapter, we proposed a structure for a two-component semiparametric mixture models where one component is parameteric with unknown parameter, and a component defined by linear constraints on its distribution function. These constraints may be moments constraints for example. We proposed also an algorithm which estimates the parameters of this model and showed how we can implement it efficiently even in multivariate contexts. The algorithm has a linear complexity when we use the Pearson's $\chi^2$ divergence and the constraints are polynomials (thus moments constraints). We provided sufficient conditions in order to prove the consistency and the asymptotic normality of the resulting estimators.\\
Simulations show the gain we have by adding moments constraints in comparison to existing methods which do not consider any prior information. The method give clear good results even if the proportion of the parametric component is very low (equal to 0.01). In signal-noise applications, this can be interpreted otherwise. As long as we are able to estimate with relatively high precision the proportion of the signal (parametric component), we are proving the existence of the signal in a very heavy noise ($99\%$ of the data) even if the position of the signal is not accurately estimated. We showed in a simple example that our model can be applied in multivariate contexts. The new model shows encouraging properties and results, and should be tested further on real datasets.

% ==========================================================
% -------------------------------------------------
%%%%%%%%%%%%%%%%%%%%%%%%%%%%%%%%%%%%%%%%%%%%%%%%%%%%%%%%%%%%%%%%%%%%%%%
% -------------------------------------------------
% ==========================================================
\section{Appendix: Proofs}
\subsection{Proof of Proposition \ref{prop:identifiabilityMixture}}\label{AppendSemiPara:PropIdenitifiability}
\begin{proof}
Based on equation (\ref{eqn:IdenitifiabilityDefEq}), we may write the corresponding constraints equations, which are a fortiori equal:
\[\lambda \int{g(x)dP_1(x|\theta)} + (1-\lambda)m(\alpha) = \tilde{\lambda} \int{g(x)dP_1(x|\tilde{\theta})} + (1-\tilde{\lambda})m(\tilde{\alpha}).\]
Define the following function:
\[G:\mathbb{R}^d\rightarrow\mathbb{R}^{\ell}: (\lambda,\theta,\alpha)\mapsto \lambda \int{g(x)dP_1(x|\theta)} + (1-\lambda)m(\alpha).\]
The solution to the previous system of equations is now equivalent to the fact that function $G$ is one-to-one. This means that for a fixed $m^*$, we need that the nonlinear system of equations:
\begin{equation}
\frac{1}{1-\lambda} m^* - \frac{\lambda}{1-\lambda}m_1(\theta) = m_0(\alpha)
\label{eqn:IdentSys}
\end{equation}
has a unique solution $(\lambda,\theta,\alpha)$. The value of $m^*$ is given by $\int{g(x)dP_T}$ where $P_T$ is the mixture we are considering. To conclude, suppose that the system (\ref{eqn:IdentSys}) has a unique solution $(\lambda^*,\theta^*,\alpha^*)$ for each given $m^*$, then function $G$ is one-to-one and the constraints equations imply that $\lambda = \tilde{\lambda},\theta = \tilde{\theta}$ and $\alpha=\tilde{\alpha}$. Finally, using (\ref{eqn:IdenitifiabilityDefEq}), we may deduce that $P_0 = \tilde{P}_0$. Thus, the semiparametric mixture model is identifiable as soon as the nonlinear system of equations (\ref{eqn:IdentSys}) has a unique solution $(\lambda^*,\theta^*,\alpha^*)$.
\end{proof}

%%%%%%%%%%%%%%%%%%%%%%%%%%%%%%%%%%%%%%%%%%%%%%%%%%%%%%%
%%%%%%%%%%%%%%%%%%%%%%%%%%%%%%%%%%%%%%%%%%%%%%%%%%%%%%%

\subsection{Proof of Proposition \ref{prop:identifiability}}\label{AppendSemiPara:Prop1}
\begin{proof}
Let $P_0$ be some signed measure which belongs to the intersection $\mathcal{N} \cap \mathcal{M}$. Since $P_0$ belongs to $\mathcal{N}$, there exists a couple $(\lambda,\theta)$ such that:
\begin{equation}
P_0 = \frac{1}{1-\lambda} P_T - \frac{\lambda}{1-\lambda} P_1(.|\theta).
\label{eqn:SetNelement}
\end{equation}
This couple is unique by virtue of assumptions 3 and 4. Indeed, let $(\lambda,\theta)$	and $(\tilde{\lambda},\tilde{\theta})$ be two couples such that:
\begin{equation}
\frac{1}{1-\lambda} P_T - \frac{\lambda}{1-\lambda} P_1(.|\theta) = \frac{1}{1-\tilde{\lambda}} P_T - \frac{\tilde{\lambda}}{1-\tilde{\lambda}} P_1(.|\tilde{\theta})\quad dP_T-a.e.
\label{eqn:identifEquality}
\end{equation}
This entails that:
\[\frac{1}{1-\lambda} - \frac{\lambda}{1-\lambda} \frac{dP_1(x|\theta)}{dP_T(x)} = \frac{1}{1-\tilde{\lambda}} - \frac{\tilde{\lambda}}{1-\tilde{\lambda}} \frac{dP_1(x|\tilde{\theta})}{dP_T(x)}.\]
Taking the limit as $\|x\|$ tends to $\infty$ results in:
\[\frac{1-c\lambda}{1-\lambda}  = \frac{1-\tilde{c}\lambda}{1-\tilde{\lambda}}.\]
Note that function $z\mapsto (1-cz)/(1-z)$ is strictly monotone as long as $c\neq 1$. Hence, it is a one-to-one map. Thus $\lambda=\tilde{\lambda}$. Inserting this result in equation (\ref{eqn:identifEquality}) entails that:
\[P_1(.|\theta) = P_1(.|\tilde{\theta})\qquad dP_T-a.e.\]
Using the identifiability of $P_1$ (assumption 4), we get $\theta=\tilde{\theta}$ which proves the existence of a unique couple $(\lambda,\theta)$ in (\ref{eqn:SetNelement}).\\
On the other hand, since $P_0$ belongs to $\mathcal{M}$, there exists a unique $\alpha$ such that $P_0\in\mathcal{M}_{\alpha}$. Uniqueness comes from the fact that function $\alpha\mapsto m(\alpha)$ is one-to-one (assumption 2). Thus, $P_0$ verifies the constraints
\[\int{dP_0(x)} = 1,\qquad \int{g_i(x)dP_0(x)} = m_i(\alpha),\quad \forall i=1,\cdots,\ell.\]
Combining this with (\ref{eqn:SetNelement}), we get:
\begin{equation}
\int{\left(\frac{1}{1-\lambda} dP_T - \frac{\lambda}{1-\lambda} dP_1(x|\theta)\right)} = 1,\; \int{g_i(x)\left(\frac{1}{1-\lambda} dP_T - \frac{\lambda}{1-\lambda} dP_1(x|\theta)\right)} = m_i(\alpha),
\label{eqn:NlnSysMalpha}
\end{equation}
for all $i=1,\cdots,\ell$. This is a non linear system of equations with $\ell+1$ equations. The first one is verified for any couple $(\lambda,\theta)$ since both $P(.|\phi^*)$ and $P_1$ are probability measures. This reduces the system to $\ell$ nonlinear equations.\\
Now, let $P_0$ and $\tilde{P}_0$ be two elements in $\mathcal{N}\cap\mathcal{M}$, then there exist two couples $(\lambda,\theta)$ and $(\tilde{\lambda},\tilde{\theta})$ with $\lambda\neq\tilde{\lambda}$ or $\theta\neq\tilde{\theta}$. Since $P_0\in\mathcal{M}$, there exists $\alpha$ such that $P_0\in\mathcal{M}_{\alpha}$. Similarly, there exists $\tilde{\alpha}$ possibly different from $\alpha$. Now, $(\lambda,\theta,\alpha)$ and $(\tilde{\lambda},\tilde{\theta},\tilde{\alpha})$ are two solutions to the system of equations (\ref{eqn:NlnSysMalpha}) which contradicts with assumption 1 of the present proposition.\\
We may now conclude that, if a signed measure $P_0$ belongs to the intersection $\mathcal{N} \cap \mathcal{M}$, then it has the representation (\ref{eqn:SetNelement}) for a unique couple $(\lambda,\theta)$ and there exists a unique $\alpha$ such that the triplet $(\lambda,\theta,\alpha)$ is a solution to the non linear system (\ref{eqn:NlnSysMalpha}). Conversely, if there exists a triplet $(\lambda,\theta,\alpha)$ which solves the non linear system (\ref{eqn:NlnSysMalpha}), then the signed measure $P_0$ defined by $P_0 = \frac{1}{1-\lambda} P(.|\phi^*) - \frac{\lambda}{1-\lambda} P_1(.|\theta)$ belongs to the intersection $\mathcal{N} \cap \mathcal{M}$. This is because on the one hand, it clearly belongs to $\mathcal{N}$ by its definition and on the other hand, it belongs to $\mathcal{M}_{\alpha}$ since it verifies the constraints and thus belongs to $\mathcal{M}$.\\
It is now reasonable to conclude that under assumptions 2-4, the intersection $\mathcal{N} \cap \mathcal{M}$ includes a \emph{unique} signed measure $P_0$ if and only if the set of $\ell$ non linear equations (\ref{eqn:NlnSysMalpha}) has a unique solution $(\lambda,\theta,\alpha)$.
\end{proof}

% -------------------------------------------------
%%%%%%%%%%%%%%%%%%%%%%%%%%%%%%%%%%%%%%%%%%%%%%%%%%%%%%%%%%%%%%%%%%%%%%%
% -------------------------------------------------
\subsection{Proof of Lemma \ref{lem:SupXiPhiDiff}}\label{AppendSemiPara:Lem1}
\begin{proof}
The proof is based partially on the proof of Proposition 3.7 part (ii) in \cite{KeziouThesis}.\\
We proceed by contradiction. Let $\varepsilon>0$ be such that $\sup_{\phi}\|\xi_n(\phi) - \xi(\phi)\|>\varepsilon$. Then, there exists a sequence $a_k\in\Phi$ such that $\|\xi_n(a_k) - \xi(a_k)\|>\varepsilon$. By assumption A3, there exists $\eta>0$ such that:
\[H(a_k,\xi(a_k)) - H(a_k,\xi_n(a_k))>\eta.\]
Thus,
\begin{equation}
\mathbb{P}\left(\sup_{\phi}\|\xi_n(\phi) - \xi(\phi)\|>\varepsilon\right) \leq \mathbb{P}\left(H(a_k,\xi(a_k)) - H(a_k,\xi_n(a_k))>\eta\right).
\label{eqn:ProofPart1InteriorConsis}
\end{equation}
Let's prove that the right hand side tends to zero as $n$ goes to infinity which is sufficient to accomplish our claim.\\
By definition of $\xi_n(a_k)$ and assumption A2, we can write:
\begin{eqnarray*}
H_n(a_k,\xi_n(a_k)) & \geq & H_n(a_k,\xi(a_k)) \\
	& \geq & H(a_k,\xi(a_k)) - o_P(1)
\end{eqnarray*}
where $o_P(1)$ does not depend upon $a_k$ by virtue of A2. Now we have:
\begin{eqnarray*}
H(a_k,,\xi(a_k)) - H(a_k,\xi_n(a_k)) & \leq & H_n(a_k,\xi_n(a_k)) - H(a_k,,\xi_n(a_k)) + o_P(1) \\
 & \leq & \sup_{\xi,\phi} \left|H_n(\phi,\xi) - H(\phi,\xi)\right| + o_P(1).
\end{eqnarray*}
Last but not least, assumption A2 permits to conclude that the right hand side tends to zero in probability. Since the left hand side is already nonnegative by definition of $\xi(a_k)$, then by the previous result we conclude that $H(a_k,,\xi(a_k)) - H(a_k,\xi_n(a_k))$ tends to zero in probability. Employing this final result in inequality (\ref{eqn:ProofPart1InteriorConsis}), we get that $\sup_{\phi}\|\xi_n(\phi) - \xi(\phi)\|$ tends to zero in probability.
\end{proof}

% -------------------------------------------------
%%%%%%%%%%%%%%%%%%%%%%%%%%%%%%%%%%%%%%%%%%%%%%%%%%%%%%%%%%%%%%%%%%%%%%%
% -------------------------------------------------
\subsection{Proof of Theorem \ref{theo:MainTheorem}}\label{AppendSemiPara:Theo1}
\begin{proof}
We proceed by contradiction in a similar way to the proof of Lemma \ref{lem:SupXiPhiDiff}. Let $\kappa>0$ be such that $\|\phi^*-\hat{\phi}\|>\kappa$, then by assumption A4, there exists $\eta>0$ such that :
\[H(\hat{\phi},\xi(\hat{\phi})) - H(\phi^*,\xi(\phi^*)) > \eta.\]
This can be rewritten as:
\begin{equation}
\mathbb{P}\left(\|\phi^*-\hat{\phi}\|>\kappa\right) \leq \mathbb{P}\left(H(\hat{\phi},\xi(\hat{\phi})) - H(\phi^*,\xi(\phi^*)) > \eta\right).
\label{eqn:ProofPart2ExtConsis}
\end{equation}
We now demonstrate that the right hand side tends to zero as $n$ goes to infinity. Let $\varepsilon>0$ be such that for $n$ sufficiently large, we have $\sup_{\xi,\phi} \left|H(\phi,\xi)-H_n(\phi,\xi)\right|<\varepsilon$. This is possible by virtue of assumption A2. The definition of $\hat{\phi}$ together with assumption A2 will now imply:
\begin{eqnarray}
H_n(\hat{\phi},\xi_n(\hat{\phi})) & \leq & H_n(\phi^*,\xi_n(\phi^*)) \nonumber\\
  & \leq & H(\phi^*,\xi_n(\phi^*)) + \sup_{\xi,\phi} \left|H(\phi,\xi)-H_n(\phi,\xi)\right| \nonumber\\
	& \leq & H(\phi^*,\xi_n(\phi^*)) + \varepsilon.
	\label{eqn:ProofPart2ExtConsisIneq}
\end{eqnarray}
We use now the continuity assumption A5 of function $\xi\mapsto H(\phi^*,\xi)$ at $\xi(\phi^*)$. For the $\varepsilon$ chosen earlier, there exists $\delta(\phi^*,\varepsilon)$ such that if $\|\xi(\phi^*)-\xi_n(\phi^*)\|<\delta(\phi^*,\varepsilon)$, then:
\[|H(\phi^*,\xi_n(\phi^*)) - H(\phi^*,\xi(\phi^*))|<\varepsilon.\]
This is possible for sufficiently large $n$ since $\sup_{\phi}\|\xi(\phi^*)-\xi_n(\phi^*)\|$ tends to zero in probability by Lemma \ref{lem:SupXiPhiDiff}. Inserting this result in (\ref{eqn:ProofPart2ExtConsisIneq}) gives:
\[H_n(\hat{\phi},\xi_n(\hat{\phi})) \leq H(\phi^*,\xi(\phi^*)) + 2\varepsilon.\]
We now have:
\begin{eqnarray*}
H(\hat{\phi},\xi(\hat{\phi})) - H(\phi^*,\xi(\phi^*)) & \leq & H(\hat{\phi},\xi(\hat{\phi})) - H_n(\hat{\phi},\xi_n(\hat{\phi})) + 2\varepsilon \\
 & \leq & H(\hat{\phi},\xi(\hat{\phi})) - H(\hat{\phi},\xi_n(\hat{\phi})) + H(\hat{\phi},\xi_n(\hat{\phi})) - H_n(\hat{\phi},\xi_n(\hat{\phi})) + 2\varepsilon.
\end{eqnarray*}
Continuity assumption of $H$ implies that for $\varepsilon>0$, there exists $\delta(\hat{\phi},\varepsilon)>0$ such that if $\|\xi(\hat{\phi}) - \xi_n(\hat{\phi})\|<\delta(\hat{\phi},\varepsilon)$, then:
\[\left|H(\hat{\phi},\xi(\hat{\phi})) - H(\hat{\phi},\xi_n(\hat{\phi}))\right| \leq \varepsilon.\]
This is again possible for sufficiently large $n$ since $\sup_{\phi}\|\xi(\phi^*)-\xi_n(\phi^*)\|$ tends to zero in probability by Lemma \ref{lem:SupXiPhiDiff}. This entails that:
\begin{eqnarray*}
H(\hat{\phi},\xi(\hat{\phi})) - H(\phi^*,\xi(\phi^*)) & \leq & H(\hat{\phi},\xi_n(\hat{\phi})) - H_n(\hat{\phi},\xi_n(\hat{\phi})) + 3\varepsilon \\
  & \leq & \sup_{\xi,\phi} |H(\phi,\xi) - H_n(\phi,\xi)| + 3\varepsilon \\
	& \leq & 4\varepsilon
\end{eqnarray*}
We conclude that the right hand side in (\ref{eqn:ProofPart2ExtConsis}) goes to zero and the proof is completed.
\end{proof}

% -------------------------------------------------
%%%%%%%%%%%%%%%%%%%%%%%%%%%%%%%%%%%%%%%%%%%%%%%%%%%%%%%%%%%%%%%%%%%%%%%
% -------------------------------------------------
\subsection{Proof of Theorem \ref{theo:MainTheoremMomConstr}}\label{AppendSemiPara:Theo2}
\begin{proof}
We will use Theorem \ref{theo:MainTheorem}. We need to verify assumptions A2 and A3. Since the class of functions $\{(\phi,\xi)\mapsto h(\phi,\xi,.)\}$ is a Glivenko-Cantelli class of functions, then assumption A2 is fulfilled by the Glivenko-Cantelli theorem. Finally, assumption A3 can be checked by strict concavity of function $\xi\mapsto H(\phi,\xi)$. Indeed, for any $\eta\in(0,1)$ and any $\xi_1,\xi_2$, we have by strict convexity of $\psi$ :
\[\psi\left(\eta\xi_1^tg(x)+(1-\eta)\xi_2^tg(x)\right)<\eta\psi\left(\xi_1^tg(x)\right)+(1-\eta)\psi\left(\xi_2^tg(x)\right).\]
If the measure $dP/(1-\lambda) - \lambda dP_1(.|\theta)/(1-\lambda)$ is positive\footnote{This measure can never be zero since it integrates to one, thus we do not need to suppose that it is nonnegative.}, we may write:
\begin{multline*}
\int{\psi\left(\eta\xi_1^tg(x)+(1-\eta)\xi_2^tg(x)\right)\left(\frac{1}{1-\lambda}dP(x)-\frac{\lambda}{1-\lambda}dP_1(x|\theta)\right)}< \\ \eta\int{\psi\left(\xi_1^tg(x)\right)\left(\frac{1}{1-\lambda}dP(x)-\frac{\lambda}{1-\lambda}dP_1(x|\theta)\right)} + (1-\eta)\int{\psi\left(\xi_2^tg(x)\right)\left(\frac{1}{1-\lambda}dP(x)-\frac{\lambda}{1-\lambda}dP_1(x|\theta)\right)},
\end{multline*}
which entails that
\[H(\phi,\eta\xi_1+(1-\eta)\xi_2)> \eta H(\phi,\xi_1)+(1-\eta)H(\phi,\xi_2),\]
and function $\xi\mapsto H(\phi,\xi)$ becomes strictly concave. However, the measure $dP/(1-\lambda) - \lambda dP_1(.|\theta)/(1-\lambda)$ is in general a signed measure and the previous implication does not hold. This is not dramatic because function $\xi\mapsto H(\phi,\xi)$ has only two choices; it is either strictly convex or strictly concave. In case function $\xi\mapsto H(\phi,\xi)$ is strictly convex, then its supremum is infinity and the corresponding vector $\phi$ does not count in the calculus of the infimum after all. This means that the only vectors $\phi\in\Phi$ which interest us are those for which function $\xi\mapsto H(\phi,\xi)$ is strictly concave. In other words, the infimum in (\ref{eqn:MomentEstimProc}) can be calculated over the set:
\[\Phi^+ = \Phi\cap \left\{\phi: \quad \xi\mapsto H(\phi,\xi) \text{ is strictly concave}\right\}\]
instead of over $\Phi$. All assumptions of Theorem \ref{theo:MainTheorem} are now fulfilled and $\hat{\phi}$ converges in probability to $\phi^*$.
\end{proof}

% -------------------------------------------------
%%%%%%%%%%%%%%%%%%%%%%%%%%%%%%%%%%%%%%%%%%%%%%%%%%%%%%%%%%%%%%%%%%%%%%%
% -------------------------------------------------
\subsection{Proof of Proposition \ref{prop:ContinDiffxiMom}}\label{AppendSemiPara:Prop2}
\begin{proof}
We already have:
\[\frac{1}{1-\lambda^*}P_T - \frac{\lambda^*}{1-\lambda^*}P_1(.|\theta^*)=P_0^*,\]
and since $P_0^*$ is supposed to be a probability measure, the matrix $J_{H(\phi^*,.)}$ is definite negative. Thus $\phi^*\in\Phi^+$. Since the set of negative definite matrices is an open set (see for example page 36 in \cite{OptimKenneth}), there exists a ball $\mathcal{U}$ of negative definite matrices centered at $J_{H(\phi^*,.)}$. Continuity of $\phi\mapsto J_{H(\phi,.)}$ permits\footnote{To see this, consider Sylvester's rule which is based on a test using the determinant of the sub-matrices of $J_{H}$. Each determinant needs to be negative. The continuity of the determinant function together with the continuity of $\phi\mapsto J_{H(\phi,.)}$ will imply that we may move around $J_(H(\phi^*,.))$ in a small neighborhood in a way that the determinants of the sub-matrices stay negative.} to find a ball $B(\phi^*,\tilde{r})$ such that the subset $\{J_{H(\phi,.)}: \phi\in B(\phi^*,\tilde{r})\}$ is inside $\mathcal{U}$. Now the neighborhood we are looking at is the ball $B(\phi^*,\tilde{r})$.\\
For the second part of the proposition, the existence and finiteness of $\xi(\phi)$ for $\phi\in\mathcal{V}=B(\phi^*,\tilde{r})$ is immediate since function $\xi\mapsto H(\phi,\xi)$ is strictly concave. Besides the the differentiability of the function $\phi\mapsto\xi(\phi)$ is a direct result of the implicit function theorem applied on the equation $\xi\mapsto \nabla H(\phi,.)$. Notice that the Hessian matrix of $H(\phi,.)$ is invertible since it is symmetric definite negative.
\end{proof}

% -------------------------------------------------
%%%%%%%%%%%%%%%%%%%%%%%%%%%%%%%%%%%%%%%%%%%%%%%%%%%%%%%%%%%%%%%%%%%%%%%
% -------------------------------------------------
\subsection{Proof of Theorem \ref{theo:AsymptotNormalMomConstr}}\label{AppendSemiPara:Theo3}
\begin{proof}
We follow the steps of Theorem 3.2 in \cite{NeweySmith}. The idea behind the proof is a mean value expansion with Lagrange remainder of the estimating equations. \\
We need at first to verify if $\hat{\phi}$ belongs to the interior of $\Phi^+$ in order to be able to differentiate $\phi\mapsto H_n(\phi,\xi)$. This can be done similarly to Proposition \ref{prop:ContinDiffxiMom}. We also can prove (by replacing $H$ by $H_n$ and $\xi(\phi)$ by $\xi_n(\phi)$) that $\phi\mapsto\xi_n(\phi)$ is continuously differentiable in a neighborhood of $\phi^*$.\\
We may now proceed to the mean value expansion. By the very definition of $\xi_n(\phi)$, we have:
\[\frac{\partial H_n}{\partial \xi}(\phi,\xi_n(\phi)) = 0\qquad \forall \phi\in\text{int}(\Phi^+),\]
which also holds for $\phi=\hat{\phi}$, i.e. 
\[\frac{\partial H_n}{\partial \xi}(\hat{\phi},\xi_n(\hat{\phi})) = 0.\]
On the other hand, the definition of $\hat{\phi}$ implies that:
\[\left.\frac{\partial}{\partial \phi}H_n(\phi,\xi_n(\phi))\right|_{\phi=\hat{\phi}} = 0.\]
Since function $\phi\mapsto\xi_n(\phi)$ is continuously differentiable. A simple chain rule implies
\begin{eqnarray*}
\left.\frac{\partial}{\partial \phi}\left(H_n(\phi,\xi_n(\phi))\right)\right|_{\phi=\hat{\phi}} & = &  \frac{\partial}{\partial \phi}H_n(\hat{\phi},\xi_n(\hat{\phi})) + \frac{\partial}{\partial \xi} H_n(\hat{\phi},\xi_n(\hat{\phi})) \frac{\partial \xi_n}{\partial \phi}(\hat{\phi}) \\
 & = & \frac{\partial}{\partial \phi}H_n(\hat{\phi},\xi_n(\hat{\phi})).
\end{eqnarray*}
The second line comes from the definition of $\xi_n(\phi)$ as the argument of the supremum of function $\xi\mapsto H_n(\phi,\xi)$. Now, the estimating equations are given simply by:
\begin{eqnarray*}
\frac{\partial H_n}{\partial \xi}(\hat{\phi},\xi_n(\hat{\phi})) & = & 0; \\
\frac{\partial H_n}{\partial \phi}(\hat{\phi},\xi_n(\hat{\phi})) & = & 0.
\end{eqnarray*}
We need to calculate these partial derivatives. We start by the derivative with respect to $\xi$:
\begin{equation}
\frac{\partial H_n}{\partial \xi}(\phi,\xi) = m(\alpha) - \frac{1}{1-\lambda}\frac{1}{n}\sum_{i=1}^n{\psi'\left(\xi^tg(x_i)\right)g(x_i)} + \frac{\lambda}{1-\lambda}\int{\psi'\left(\xi^tg(x)\right)g(x)p_1(x|\theta)dx}
\label{eqn:DerivWRTxiMom}
\end{equation}
We calculate the partial derivatives with respect to $\alpha,\lambda$ and $\theta$:
\begin{eqnarray}
\frac{\partial H_n}{\partial \alpha} & = &  \xi^t \nabla m(\alpha) \label{eqn:DerivWRTalphaMom}\\
\frac{\partial H_n}{\partial \lambda} & = &  -\frac{1}{(1-\lambda)^2} \frac{1}{n}\sum_{i=1}^n{\psi\left(\xi^tg(x_i)\right)} + \frac{1}{(1-\lambda)^2}\int{\psi\left(\xi^tg(x)\right)p_1(x|\theta)dx} \label{eqn:DerivWRTlambdaMom}\\
\frac{\partial H_n}{\partial \theta} & = &  \frac{\lambda}{1-\lambda}\int{\psi\left(\xi^tg(x)\right)\nabla_{\theta} p_1(x|\theta)dx} \label{eqn:DerivWRTthetaMom}
\end{eqnarray}
Notice that by Lemma \ref{lem:SupXiPhiDiff}, the continuity of $\phi\mapsto\xi(\phi)$ and the consistency of $\hat{\phi}$ towards $\phi^*$, we have $\xi_n(\hat{\phi})\rightarrow \xi(\phi^*)=0$ in probability. A mean value expansion of the estimating equation between $(\hat{\phi},\xi_n(\hat{\phi}))$ and $(\phi^*,0)$ implies that there exists $(\bar{\phi},\bar{\xi})$ on the line between these two points such that:
\begin{equation}
\left(\begin{array}{c} \frac{\partial H_n}{\partial \phi}(\hat{\phi},\xi_n(\hat{\phi})) \\ \frac{\partial H_n}{\partial \xi}(\hat{\phi},\xi_n(\hat{\phi})) \end{array}\right) = \left(\begin{array}{c}  \frac{\partial H_n}{\partial \phi}(\phi^*,0) \\\frac{\partial H_n}{\partial \xi}(\phi^*,0) \end{array}\right)  + J_{H_n}(\bar{\phi},\bar{\xi}) \left(\begin{array}{c} \hat{\phi}-\phi^* \\ \xi_n(\hat{\phi})\end{array}\right),
\label{eqn:StochExpansion}
\end{equation}
where $J_{H_n}(\bar{\phi},\bar{\xi})$ is the matrix of second derivatives of $H_n$ calculated at the mid point $(\bar{\phi},\bar{\xi})$. The left hand side is zero, so we need to calculate the first vector in the right hand side. We have by simple substitution in formula (\ref{eqn:DerivWRTxiMom}):
\[\frac{\partial H_n}{\partial \xi}(\phi^*,0) = m(\alpha^*) - \frac{1}{1-\lambda^*}\frac{1}{n}\sum_{i=1}^n{g(x_i)} + \frac{\lambda^*}{1-\lambda^*}\int{g(x)p_1(x|\theta^*)dx}.\]
Using the assumption that the model (\ref{eqn:TrueP0model}) verify the set of constraints defining $\mathcal{M}_{\alpha}$ together with the CLT, we write:
\begin{equation}
\sqrt{n}\frac{\partial H_n}{\partial \xi}(\phi^*,0) \xrightarrow[\mathcal{L}]{} \mathcal{N}\left(0,\frac{1}{(1-\lambda^*)^2}\text{Var}_{P_T}(g(X))\right).
\label{eqn:LimitLawPartialDerivHn}
\end{equation}
Using formulas (\ref{eqn:DerivWRTalphaMom}), (\ref{eqn:DerivWRTlambdaMom}) and (\ref{eqn:DerivWRTthetaMom}), we may write:
\begin{eqnarray*}
\frac{\partial H_n}{\partial \alpha}(\phi^*,0) & = & 0; \\
\frac{\partial H_n}{\partial \lambda}(\phi^*,0) & = & -\frac{1}{(1-\lambda^*)^2} + \frac{1}{(1-\lambda^*)^2} = 0;\\
\frac{\partial H_n}{\partial \theta}(\phi^*,0) & = & \frac{\lambda^*}{1-\lambda^*}\int{\nabla_{\theta} p_1(x|\theta^*)dx} = \frac{\lambda^*}{1-\lambda^*}\nabla_{\theta} \int{p_1(x|\theta^*)dx} = 0.
\end{eqnarray*}
The final line holds since by Lebesgue's differentiability theorem using assumption 5 for $\xi=0$, we can change between the sign of integration and derivation. Combine this with the fact that $p_1(x|\theta^*)$ is a probability density function which integrates to 1, gives the result in the last line.\\
We need now to write an explicit form for the matrix $J_{H_n}(\bar{\phi},\bar{\xi})$ and study its limit in probability. It contains the second order partial derivatives of function $H_n$ with respect to its parameters. We start by the double derivatives. Using formulas (\ref{eqn:DerivWRTxiMom}), (\ref{eqn:DerivWRTalphaMom}), (\ref{eqn:DerivWRTlambdaMom}) and (\ref{eqn:DerivWRTthetaMom}), we write:
\begin{eqnarray*}
\frac{\partial^2H_n}{\partial \xi^2} & = & - \frac{1}{1-\lambda}\frac{1}{n}\sum_{i=1}^n{\psi''\left(\xi^tg(x_i)\right)g(x_i)g(x_i)^t} + \frac{\lambda}{1-\lambda}\int{\psi''\left(\xi^tg(x)\right)g(x)g(x)^tp_1(x|\theta)dx};\\
\frac{\partial^2H_n}{\partial \alpha^2} & = & \xi^t J_m(\alpha);\\
\frac{\partial^2H_n}{\partial \lambda^2} & = & -\frac{2}{(1-\lambda)^3} \frac{1}{n}\sum_{i=1}^n{\psi\left(\xi^tg(x_i)\right)} + \frac{2}{(1-\lambda)^3}\int{\psi\left(\xi^tg(x)\right)p_1(x|\theta)dx};\\
\frac{\partial^2H_n}{\partial \theta^2} & = & \frac{\lambda}{1-\lambda}\int{\psi\left(\xi^tg(x)\right)J_{p_1(x|\theta)}dx}; \\
\frac{\partial^2H_n}{\partial \xi\partial\alpha} & = & \nabla m(\alpha); \\
\frac{\partial^2H_n}{\partial \xi\partial\lambda} & = &  - \frac{1}{(1-\lambda)^2}\frac{1}{n}\sum_{i=1}^n{\psi'\left(\xi^tg(x_i)\right)g(x_i)} + \frac{1}{(1-\lambda)^2}\int{\psi'\left(\xi^tg(x)\right)g(x)p_1(x|\theta)dx};\\
\frac{\partial^2H_n}{\partial \xi\partial\theta} & = & \frac{\lambda}{1-\lambda}\int{\psi'\left(\xi^tg(x)\right)g(x)\nabla_{\theta}p_1(x|\theta)^tdx}; \\
\frac{\partial^2H_n}{\partial \alpha\partial\lambda} & = & 0; \\
\frac{\partial^2H_n}{\partial \alpha\partial\theta} & = & 0; \\
\frac{\partial^2H_n}{\partial \lambda\partial\theta} & = &  \frac{1}{(1-\lambda)^2}\int{\psi\left(\xi^tg(x)\right)\nabla_{\theta}p_1(x|\theta)dx}.
\end{eqnarray*}
As $n$ goes to infinity, we have $\bar{\xi}\rightarrow 0$ and $\bar{\phi}\rightarrow \phi^*$. Then, under regularity assumptions of the present theorem, we can calculate the limit in probability of the matrix $J_{H_n}(\bar{\phi},\bar{\xi})$. The blocks limits are given by:
\[
\frac{\partial^2H_n}{\partial \xi^2} \stackrel{\mathbb{P}}{\rightarrow} - \mathbb{E}_{P_0^*}\left[g(X)g(X)^t\right],\qquad \frac{\partial^2H_n}{\partial \alpha^2} \stackrel{\mathbb{P}}{\rightarrow} 0, \qquad \frac{\partial^2H_n}{\partial \lambda^2} \stackrel{\mathbb{P}}{\rightarrow} 0, \qquad \frac{\partial^2H_n}{\partial \theta^2} \stackrel{\mathbb{P}}{\rightarrow} 0,\qquad \frac{\partial^2H_n}{\partial \xi\partial\alpha} \stackrel{\mathbb{P}}{\rightarrow} \nabla m(\alpha^*) \]

\[\frac{\partial^2H_n}{\partial \xi\partial\lambda} \stackrel{\mathbb{P}}{\rightarrow}  - \frac{1}{(1-\lambda^*)^2}\mathbb{E}_{P_T}\left[g(X)\right] + \frac{1}{(1-\lambda^*)^2}\int{g(x)p_1(x|\theta^*)dx}\]

\[\frac{\partial^2H_n}{\partial \xi\partial\theta} \stackrel{\mathbb{P}}{\rightarrow} \frac{\lambda^*}{1-\lambda^*}\int{g(x)\nabla_{\theta}p_1(x|\theta^*)dx},\qquad \frac{\partial^2H_n}{\partial \alpha\partial\lambda} \stackrel{\mathbb{P}}{\rightarrow} 0, \qquad
\frac{\partial^2H_n}{\partial \alpha\partial\theta}  \stackrel{\mathbb{P}}{\rightarrow} 0, \qquad \frac{\partial^2H_n}{\partial \lambda\partial\theta} \stackrel{\mathbb{P}}{\rightarrow}  0, \]
taking into account that $\psi(0)=0,\psi'(0)=1$ and $\psi''(0)=1$. The limit in probability of the matrix $J_{H_n}(\bar{\phi},\bar{\xi})$ can be written in the form:
\[ J_H = \left[\begin{array}{cc}
0 & J_{\phi^*,\xi^*}^t \\
J_{\phi^*,\xi^*} & J_{\xi^*,\xi^*}
\end{array}\right],\]
where $J_{\phi^*,\xi^*}$ and $J_{\xi^*,\xi^*}$ are given by (\ref{eqn:NormalAsymMomJ1}) and (\ref{eqn:NormalAsymMomJ2}). The inverse of matrix $J_H$ has the form:
\[J_H^{-1} = \left(\begin{array}{cc} -\Sigma & H \\ H^t & W\end{array}\right),\]
where
\[
\Sigma = \left(J_{\phi^*,\xi^*}^t J_{\xi^*,\xi^*} J_{\phi^*,\xi^*}\right)^{-1},\quad  H = \Sigma J_{\phi^*,\xi^*}^t J_{\xi^*,\xi^*}^{-1},\quad  W = J_{\xi^*,\xi^*}^{-1} - J_{\xi^*,\xi^*}^{-1} J_{\phi^*,\xi^*} \Sigma J_{\phi^*,\xi^*}^t J_{\xi^*,\xi^*}^{-1}.
\]
Going back to (\ref{eqn:StochExpansion}), we have:
\begin{equation*}
\left(\begin{array}{c} 0 \\ 0 \end{array}\right) = \left(\begin{array}{c}  0 \\ \frac{\partial H_n}{\partial \xi}(\phi^*,0) \end{array}\right)  + J_{H_n}(\bar{\phi},\bar{\xi}) \left(\begin{array}{c}  \hat{\phi}-\phi^* \\ \xi_n(\hat{\phi}) \end{array}\right).
\end{equation*}
Solving this equation in $\phi$ and $\xi$ gives:
\begin{equation*}
\left(\begin{array}{c}  \sqrt{n}\left(\hat{\phi}-\phi^*\right) \\ \sqrt{n}\xi_n(\hat{\phi})\end{array}\right) = J_H^{-1}\left(\begin{array}{c}  0 \\ \sqrt{n}\frac{\partial H_n}{\partial \xi}(\phi^*,0) \end{array}\right) + o_P(1).
\end{equation*}
Finally, using (\ref{eqn:LimitLawPartialDerivHn}), we get that:
\[\left(\begin{array}{c}  \sqrt{n}\left(\hat{\phi}-\phi^*\right) \\ \sqrt{n}\xi_n(\hat{\phi})\end{array}\right) \xrightarrow[\mathcal{L}]{} \mathcal{N}\left(0,S\right)\]
where 
\[S = \frac{1}{(1-\lambda^*)^2}\left(\begin{array}{c}H \\ W\end{array}\right) \text{Var}_{P_T}(g(X)) \left(H^t\quad W^t\right).\]
This ends the proof.
\end{proof}

%%%%%%%%%%%%%%%%%%%%%%%%%%%%%%%%%%%%%%%%%%%%%%%%%%%%%%%%%%%%%%%%%%%%%%%%%%%%%%%%%
%
%==============================================================
%%%%%%%%%%%%%%%%%%%%%%%%%%%%%%%%%%%%%%%%%%%%%%%%%%%%%%%%%%%%%%%%%%%%%%%%%%%%
%==============================================================
%
%%%%%%%%%%%%%%%%%%%%%%%%%%%%%%%%%%%%%%%%%%%%%%%%%%%%%%%%%%%%%%%%%%%%%%%%%%%%%%%%%

\chapter{Semiparametric two-component mixture models where one component is defined through L-moments constraints}
Recall that a semiparametric two-component mixutre model is defined by:
\begin{equation}
f(x) = \lambda f_1(x|\theta) + (1-\lambda) f_0(x), \qquad \text{for } x\in\mathbb{R}
\end{equation}
for $\lambda\in(0,1)$ and $\theta\in\mathbb{R}^d$ to be estimated and the density $f_0$ is considered to be unknown. We have proposed in Chapter 3 a method which incorporates moment-type constraints on the unknown component. The method outperforms other semiparametric methods which do not use prior information encouraging the use of a suitable prior information. Moment-type constraints are not suitable for positive-support mixtures especially when the density does not decrease fast enough. The method needs more observations to be able to estimate such mixtures. \\
We thus propose here to use L-moments constraints. L-moments have become classical tools alternative to central moments for the description of dispersion, skewness and kurtosis of a univariate heavy-tailed distribution. Distributions such as the Lognormal, the Pareto and the Weibull distributions are standard examples of such distributions. The use of L-moments is evolving since their introduction by \cite{Hoskings}. One of the main interests of L-moments is that they can be defined as soon as the expectation of the random variable exists. \cite{AlexisGSI13} has proposed a structure and an estimation procedure for semiparametric models defined through L-moments conditions. The resulting estimators performe well under the model, and they outperforme existing methods in misspecification contexts.\\
Similarly to the case of moment constraints seen in the previous chapter, the incorporation of L-moments constraints cannot be done directly in existing methods for semiparametric mixtures (see paragraph \ref{sec:LiteratureSemiparaMix}) because the optimization will be carried over a (possibly) infinite dimensional space on the one hand, and on the other hand, existing methods use either the distribution function or the probability density function and cannot adapt an approach based on the quantile function. Our approach introduced in the previous chapter cannot be used either because L-moments are not linear functions of the distribution function as we will see in paragraph \ref{subsec:LmomDefProper}. We thus need a new tool. Convex analysis offer away using Fenchel-Legendre duality to transform an optimization problem over an infinite dimensional space to the space of Lagrangian parameters (finite dimensional one). $\varphi-$divergences, by their convexity properties, are suitable tools in order to use the duality result. Chap 1 in the PhD thesis of \cite{AlexisThesis} introduced a method based on $\varphi-$divergences to estimate a semiparametric model defined subject to L-moments constraints. We will exploit his methodology to build a new estimation procedure which takes into account L-moments constraints over the unknown component's distribution.\\
%%%%%%%%%%%%%%%%%%%%%%%%%%%%%%%%%%%%%%%%%%%%%%%%%%%%%%%%%%%%%%%%%%%%%%%%%%%%%%%%%%%%%%%%%%%%%%%%%%%%%%%
%
% ======================================================================
%%%%%%%%%%%%%%%%%%%%%%%%%%%%%%%%%%%%%%%%%%%%%%%%%%%%%%%%%%%%%%%%%%%%%%%%%%%%%%%%
% ======================================================================
%
%%%%%%%%%%%%%%%%%%%%%%%%%%%%%%%%%%%%%%%%%%%%%%%%%%%%%%%%%%%%%%%%%%%%%%%%%%%%%%%%%%%%%%%%%%%%%%%%%%%%%%%

\section{Semiparametric models defined through L-moments constraints}
In this section, we present a definition of a semiparametric model subject to L-moments constraints in a similar way to semiparametric models defined through moments constraints. An essential part to begin with is the definition of L-moments. We will keep this part brief and one can consult \cite{AlexisThesis} Chap. 1 or \cite{Hoskings} for more details.\\
We recall two important notions; the quantile function and the quantile measure. Let $X_{1}, \ldots  X_n$ be $n$ i.i.d. copies of a random variable $X$  taking values in $\mathbb{R}$ with unknown cumulative distribution function (cdf) $\mathbb{F}$. Denote by $\mathbb{F}^{-1}(u)$ for $u \in (0,1)$ the associated quantile function of the cdf $\mathbb{F}$ defined by
\begin{equation*}
 \mathbb{F}^{-1}(u) = \inf\left\lbrace x \in \mathbb{R},\;\; s.t. \;\; \mathbb{F}(x) \geq u \right\rbrace, \;\; u \in (0,1).
\end{equation*}
We can associate to $\mathbb{F}^{-1}$ a measure ${\bf{F}}^{-1}$ on $\mathcal{B}([0,1])$ given by
\[{\bf{F}}^{-1}(B)=\int_0^1{\ind{x\in B}d\mathbb{F}^{-1}(x)} \in\mathbb{R}\cup\{-\infty,+\infty\}.\]
The integral here is a Riemann-Stieltjes one. $\F^{-1}$ is a $\sigma-$finite measure since $\mathbb{F}^{-1}$ has bounded variations on every subinterval $[a,b]$ from $(0,1)$.\\
In this section, we suppose that $\mathbb{E}|X| < \infty$ and $\int{|x|dF(x)}<\infty$. We adapt the standard notation for the cumulative distribution function (cdf) and measures, i.e. a measure $P$ has a cdf $\mathbb{F}$, a density $p$ with respect to the Lebesgue measure and a quantile measure $\F^{-1}$, and a measure $Q$ has a cdf $\mathbb{Q}$, a density $q$ with respect to the Lebesgue measure and a quantile measure $\Q$.
%%%%%%%%%%%%%%%%%%%%%%%%%%%%%%%%%%%%%%%%%%%%%%%%%%%%%%%%%%%%%%%%%%%%%%%%
%%%%%%%%%%%%%%%%%%%%%%%%%%%%%%%%%%%%%%%%%%%%%%%%%%%%%%%%%%%%%%%%%%%%%%%%
\subsection{L-moments: Definition and first properties}\label{subsec:LmomDefProper}
Let $X_{1:n} < \ldots < X_{n:n} $ be the order statistics associated to the sample $X_1,\cdots,X_n$.
\begin{definition}
The L-moment of order $r$, denoted $\lambda_r$, $r=1,2,\ldots$ is defined as a linear combination of the expectation of order statistics:
\begin{equation*}
  \lambda_r = \frac{1}{r} \sum_{k=0}^{r-1}(-1)^k \binom{r-1}{k}\mathbb{E}\left( X_{r-k:r} \right).
\end{equation*}
\end{definition}
\n If $\mathbb{F}$ is continuous, then the expectation of the $j$-th order statistic is given by
\begin{equation}
  \mathbb{E}\left[ X_{j:r} \right]  = \dfrac{r!}{(j-1)!(r-j)!} \int_{\mathbb{R}} x\mathbb{F}(x)^{j-1}\left[1-\mathbb{F}(x)\right]^{r-j}d\mathbb{F}(x).
	\label{eqn:OrderStatLaw}
\end{equation}
\n In particular, the first three L-moments are
\begin{eqnarray*}
\lambda_{1} &=& \mathbb{E}[X];\\
\lambda_{2} &=& \left( \mathbb{E}\left[X_{2:2}\right] - \mathbb{E}\left[X_{1:2}\right] \right)/2; \\
\lambda_{3} &=& \left( \mathbb{E}\left[X_{3:3}\right] - 2\mathbb{E}\left[X_{2:3}\right] + \mathbb{E}\left[X_{1:3}\right] \right)/3.
\end{eqnarray*}
\n Using formula (\ref{eqn:OrderStatLaw}), L-moments can be expressed using the quantile function $\mathbb{F}^{-1}$ (see Proposition 1.1. from \cite{AlexisThesis}) as follows:
\[\lambda_r = \int_0^1{\mathbb{F}^{-1}(u)L_{r-1}(u)du}\qquad \forall r\geq 1,\]
where $L_r$ is the shifted Legendre polynomial of order $r$ and is given by:
\[L_r(u) = \sum_{k=0}^r{(-1)^{r-k}\binom{r}{k}\binom{r+k}{k}u^k}.\]
Moreover, for $r\geq 2$:
\begin{equation}
\lambda_r = -\int_{\mathbb{R}}{K_r(t)d\mathbb{F}^{-1}(t)},
\label{eqn:LmomRepIntShiftLeg}
\end{equation}
where
\begin{equation}
K_r(t) = \int_0^t{L_{r-1}(u)du} = \sum_{k=0}^{r-1}{\frac{(-1)^{r-k}}{k+1}\binom{r}{k}\binom{r+k}{k}t^{k+1}}
\label{eqn:IntShiftLegPoly}
\end{equation}
is the integrated shifted Legendre polynomial (see Proposition 1.2 in \cite{AlexisThesis}). Notice that L-moments are polynomials in the cdf and linear in the quantile measure.

%%%%%%%%%%%%%%%%%%%%%%%%%%%%%%%%%%%%%%%%%%%%%%%%%%%%%%%%%%%%%%%%%%%%%%%%%%%%%%%%%%%%%%%%%%%%%%%%%%%%%%%
%%%%%%%%%%%%%%%%%%%%%%%%%%%%%%%%%%%%%%%%%%%%%%%%%%%%%%%%%%%%%%%%%%%%%%%%%%%%%%%%%%%%%%%%%%%%%%%%%%%%%%%

\subsection{Semiparametric Linear Quantile Models (SPLQ)}
SPLQ models were introduced by \cite{AlexisGSI13}. The definition passes by the quantile measures instead of the distribution function. It is possible to define semiparametric models subject to L-moments constraints using the distribution function. However, their estimation would be very difficult because the constraints are not linear in the distribution function. They are instead linear in the quantile measure. This will become clearer as we go further in this subject. Denote $M^{-1}$ the set of all $\sigma-$finite measures. \\
\begin{definition} A semiparametric linear quantile model related to some quantile measure $\F_T^{-1}$ is a collection of quantile measures absolutely continuous with respect to $\F_T^{-1}$ sharing the same form of L-moments, i.e.
\[\mathcal{M} = \bigcup_{\alpha\in\mathcal{A}} \left\{\F^{-1}\ll \F_T^{-1},\; s.t.\; \int_{0}^1{K_r(t)\F^{-1}(dt)} = m(\alpha)\right\},\]
where $m(\alpha) = (-\lambda_2,\cdots,-\lambda_{\ell})$ and $\alpha\in\mathcal{A}\subset\mathbb{R}^s$.
\end{definition}
\begin{example}[\cite{AlexisThesis}]
\label{Example:Weibull}
Consider the model which is the family of all the distributions of a r.v. $X$ whose second, third and fourth L-moments satisfy:
\begin{eqnarray*}
\lambda_{2} &=& \sigma\left( 1-2^{-1/\nu} \right)\Gamma\left( 1 + \dfrac{1}{\nu} \right)\\
\lambda_{3} &=& \lambda_{2}\left[ 3 - 2\dfrac{1-3^{1/\nu}}{1-2^{-1/\nu}} \right]\\
\lambda_{4} &=& \lambda_{2}\left[ 6 + \dfrac{5(1-4^{-1/\nu})-10(1-3^{-1/\nu})}{1-2^{-1/\nu}} \right],
\end{eqnarray*}
for $\sigma > 0$, $\nu > 0$. These distributions share their first L-moments of order 2, 3 and 4 with those of a Weibull distribution with scale and shape parameter; $\sigma$, $\nu$.
\end{example}
%\textcolor{blue}{For constrained optimization problems with linear constraints, convex analysis offers a tool using the Fenchel duality to solve such optimization problems as long as the objective function is convex. In the previous chapter, we exploited this tool to transform the optimization of the $\varphi-$divergence under moment-type constraints in possibly infinite dimensional space into an optimization problem over the finite dimensional space $\mathbb{R}^{\ell-1}$. L-moments constraints, however, are \emph{not} linear functionals of the cdf. They are still linear functionals of the quantile measures. Based upon this idea, \cite{AlexisThesis} built a mathematical procedure to treat efficiently semiparametric models constrained to L-moments conditions. }
%
%
In SPLQ models, the objective is to estimate the value of $\alpha^*$ for which the true quantile measure $\F^{-1}_T$ of the data belongs to $\mathcal{M}_{\alpha^*}$ on the basis of a sample $X_1,\cdots,X_n$. The estimation procedure is generally done by either solving the set of equations defining the constraints or by minimizing a suitable distance-like function between the set $\mathcal{M}$ and some estimator of $\F^{-1}_T$ based on an observed sample. In other words, we search for the "projection" of $\F^{-1}_T$ on $\mathcal{M}$.\\
We have seen in the previous chapter that $\varphi-$divergences offer a way to calculate a "projection" of a finite signed measure on a set of finite signed measures, see definitions \ref{def:phiDistance} and \ref{def:phiProj}. $\varphi-$divergences can still be used to identify some distance between a $\sigma-$finite measure and a set of $\sigma-$finite measures. We may write:
\begin{eqnarray}
\alpha^* & = & \arginf_{\alpha\in\mathcal{A}} D_{\varphi}\left(\mathcal{M}_{\alpha},\F_T^{-1}\right) \nonumber \\
 & = & \arginf_{\alpha\in\mathcal{A}} \inf_{\F^{-1}\in\mathcal{M}_{\alpha}}D_{\varphi}\left(\F^{-1},\F_T^{-1}\right).
\label{eqn:EstimSPLQ}
\end{eqnarray}
Of course, if $\F_T^{-1}\in\mathcal{M}_{\alpha^*}$ for some $\alpha^*\in\mathcal{A}$, then $D_{\varphi}\left(\cup_{\alpha}\mathcal{M}_{\alpha},\F_T^{-1}\right)=0$. Otherwise, $\alpha^*$ corresponds to the parameter of the closest set $\mathcal{M}_{\alpha}$ from the $\varphi-$divergence point of view to the quantile measure $\F_T^{-1}$.

%%%%%%%%%%%%%%%%%%%%%%%%%%%%%%%%%%%%%%%%%%%%%%%%
\subsection{Estimation using the duality technique}\label{subsec:SPLQDuality}
The estimation procedure (\ref{eqn:EstimSPLQ}) is not feasible because it concerns the minimization over a subset of possibly infinite dimension. The duality technique presented in paragraph \ref{subsec:DualTechRes} can be applied here too in order to transform the calculus of the projection from an optimization problem over a possibly infinite dimensional space into an optimization problem over $\mathbb{R}^{\ell-1}$, where $\ell-1$ is the number of constraints defining the set $\mathcal{M}_{\alpha}$. We recall briefly this techniques by applying it directly in the context of quantile measures. Corollary 1.1 from \cite{AlexisThesis} states the following. If there exists some $\F^{-1}$ in $\mathcal{M}_{\alpha}$ such that $a_{\varphi}<d\F^{-1}/d\F_T^{-1}<b_{\varphi}$ $\F_T^{-1}$-a.s. where dom$\varphi = (a_{\varphi},b_{\varphi})$ then,
\begin{equation}
\inf_{\F^{-1}\in\mathcal{M}_{\alpha}}\int_{0}^1{\varphi\left(\frac{d\F^{-1}}{d\F_T^{-1}}\right)(u)\F_T^{-1}(du)} = \sup_{\xi\in\mathbb{R}^{\ell-1}} \xi^t m(\alpha) - \int_0^1{\psi\left(\xi^tK(u)\right)\F_T^{-1}(du)}.
\label{eqn:DualityLmom}
\end{equation}
This formula permits to build a plug-in estimate for $\alpha$ by considering a sample $X_1,\cdots,X_n$, see Remark 1.15 in \cite{AlexisThesis}.
\begin{equation}
\hat{\alpha} = \arginf_{\alpha} \sup_{\xi\in\mathbb{R}^{\ell-1}} \xi^t m(\alpha) - \sum_{i=1}^{n-1}{\psi\left(\xi^tK\left(\frac{i}{n}\right)\right)\left(X_{i+1:n} - X_{i:n}\right)}.
\label{eqn:DualityLmomEmpirical}
\end{equation}
This plug-in estimate is very interesting in its own, because it does not need any numerical integration. Besides, if we take $\varphi$ to be the $\chi^2$ generator, i.e. $\varphi(t)=(t-1)^2/2$ whose convex conjugate is $\psi(t)=t^2/2+t$, the optimization over $\xi$ can be solved directly in a similar way to Example \ref{ex:Chi2LinConstr}, see also Example 1.12 in \cite{AlexisThesis}. We will get back to this interesting case study later on.

%%%%%%%%%%%%%%%%%%%%%%%%%%%%%%%%%%%%%%%%%%%%%%%%%%%%%%%%%%%%%%%%%%%%%%%%%%%%%%%%%%%%%%%%%%%%%%%%%%%%%%%
%
% ======================================================================
%%%%%%%%%%%%%%%%%%%%%%%%%%%%%%%%%%%%%%%%%%%%%%%%%%%%%%%%%%%%%%%%%%%%%%%%%%%%%%%%
% ======================================================================
%
%%%%%%%%%%%%%%%%%%%%%%%%%%%%%%%%%%%%%%%%%%%%%%%%%%%%%%%%%%%%%%%%%%%%%%%%%%%%%%%%%%%%%%%%%%%%%%%%%%%%%%%

\section{Semiparametric two-component mixture models when one component is defined through L-moments constraints}

\subsection{Definition and identifiability}
\begin{definition}
\label{def:SemiParaModelLmom}
Let $X$ be a random variable taking values in $\mathbb{R}$ distributed from a probability measure $P$ whose cdf is $\mathbb{F}$. We say that $P(.|\phi)$ with $\phi=(\lambda,\theta,\alpha)$ is a two-component semiparametric mixture model subject to L-moments constraints if it can be written as follows:
\begin{eqnarray}
P(.| \phi) & = &  \lambda P_1(.|\theta) + (1-\lambda) P_0 \quad \text{s.t. } \nonumber\\
\F_0^{-1}\in\mathcal{M}_{\alpha} & = & \left\{\Q^{-1} \in M^{-1}, \Q^{-1}\ll\F_0^{-1} \text{ s.t. } \int_{0}^1{K(u)\Q^{-1}(du)}=m(\alpha)\right\}
\label{eqn:SetMalphaLmom}
\end{eqnarray}
for $\lambda\in(0,1)$ the proportion of the parametric component, $\theta\in\Theta\subset\mathbb{R}^{d}$ a set of parameters defining the parametric component, $\alpha\in\mathcal{A}\subset\mathbb{R}^{s}$ is the constraints parameter, $K=(K_2,...,K_{\ell})$ is defined through formula (\ref{eqn:IntShiftLegPoly}) and finally $m(\alpha)=(m_2(\alpha),\cdots,m_{\ell}(\alpha))$ is a vector-valued function determining the values of the L-moments.
\end{definition}

Notice that $m(\alpha)$ must contain the negative values of the L-moments by equation (\ref{eqn:LmomRepIntShiftLeg}), i.e $m_r(\alpha)=-\lambda_r$. In this definition, it may appear that we have mixed quantiles with probabilities. This is however necessary in order to show the structure of the mixture model which generates the data. This structure is uniquely defined through the distribution function and does not have a "proper" writing using the quantile measure. In general, there is no formula which gives the quantile of a mixture model, and in practice, statisticians use approximations to calculate the quantile of a mixture model. Thus, working with the quantiles will make us lose the linearity property relating the two components with the mixture's distribution. In the previous chapter, this linearity played an essential role in the estimation procedure and simplified the calculus of the estimator on several levels. We will get back to this idea later on, and a "partial" solution will be proposed in order to get back to the mixture distribution instead of its quantile.\\
It is important to recall that the use of quantiles in the definition of semiparametric models subject to L-moments constraints stems from the fact that the constraints are linear functionals in the quantiles. Thus, an estimation procedure which employs the quantiles instead of the distribution function can be solved using the Fenchel-Legendre duality in a similar way to paragraph \ref{subsec:SPLQDuality}.

%%%%%%%%%%%%%%%%%%%%%%%%%%%%%%%%%%%%%%%%%%%%%%%%%%%%%%%
%\subsection{Identifiability}
The identifiability of the model was not questioned in the context of SPLQ models because it suffices that the sets $\mathcal{M}_{\alpha}$ are disjoint (the function $m(\alpha)$ is one-to-one). However, in the context of this semiparametric mixture model, identifiability cannot be achieved only by supposing that the sets $\mathcal{M}_{\alpha}$ are disjoint. \\
\begin{definition}
We say that the two-component semiparametric mixture model subject to L-moments constraints is identifiable if it verifies the following assertion. If
\begin{equation}
\lambda P_1(.|\theta) + (1-\lambda)P_0 = \tilde{\lambda} P_1(.|\tilde{\theta}) + (1-\tilde{\lambda}) \tilde{P}_0,\quad \text{with } \F_0^{-1}\in\mathcal{M}_{\alpha}, \tilde{\F}_0^{-1}\in\mathcal{M}_{\tilde{\alpha}}, 
\label{eqn:IdenitifiabilityDefEqLmom}
\end{equation}
then $\lambda = \tilde{\lambda},\theta = \tilde{\theta}$ and $P_0=\tilde{P}_0$ (and hence $\alpha=\tilde{\alpha}$).
\end{definition}
This is the same identifiability concept considered in Definition \ref{def:identifiabilitySemipParaMom} (and by \cite{Bordes06b}) except that the unknown component's quantile belongs to the set $\mathcal{M}_{\alpha}$.\\
\begin{proposition}
\label{prop:identifiabilityMixtureLmom}
For a given mixture distribution $P_T = P(.| \phi^*)$ whose cdf is $\mathbb{F}_T$, suppose that the system of equations:
\begin{equation}
\int_{0}^1{K(u)\left(\frac{1}{1-\lambda}\F_T - \frac{\lambda}{1-\lambda}\F_1(.|\theta)\right)^{-1}(du)} = m(\alpha)
\label{eqn:SysLmomEq}
\end{equation}
has a unique solution $(\lambda^*,\theta^*,\alpha^*)$. Then, equation (\ref{eqn:IdenitifiabilityDefEqLmom}) has a unique solution, i.e. $\lambda = \tilde{\lambda},\theta = \tilde{\theta}$ and $P_0=\tilde{P}_0$, and the semiparametric mixture model $P_T = P(.| \phi^*)$ is identifiable.
\end{proposition}
The proof is differed to Appendix \ref{AppendSemiPara:PropIdenitifiabilityLmom}. Note that the proof although has a close idea to the proof of Proposition \ref{prop:identifiabilityMixture} is different with more technical difficulties.
%Example
\begin{example}[Two-component exponential mixture]
We propose to look at an exponential mixture defined by:
\[f(x|\lambda^*,a_1^*) = \lambda^* a_1^* e^{-a_1^* x} + (1-\lambda^*) a_0^* e^{-a_0^* x}\]
where $a_1^*=1.5, a_0^*=0.5$ and $\lambda^*\in\{0.3,0.5,0.7,0.85\}$. This is considered to be the distribution generating the observed data. Suppose that the second component $f_0^*(x)=a_0^* e^{-a_0^* x}$ is unknown during the estimation. Furthermore, suppose that we hold an information about $f_0^*$ that its quantile ${\F_0^*}^{-1}$ belongs to the following class of functions:
\[\mathcal{M} = \left\{\F^{-1}\ll{\F_0^*}^{-1},\quad \int_0^1{u(1-u)\F^{-1}(du)} = \frac{1}{2a_0^*}\right\}.\]
%\int_{\mathbb{R}_+}{(2\mathbb{F}_0(x)-1)f_0(x)dx}
This set contains all probability distributions whose second L-moment has the value $\frac{1}{2a_0^*}$. We would like to check the identifiability of the semiparametric mixture model subject to the second L-moment constraint of the exponential distribution $\mathcal{E}(a_0^*)$. The system of equations (\ref{eqn:SysLmomEq}) is given by:
\[\int_0^1{u(1-u)\left(\frac{1}{1-\lambda}\F_T - \frac{\lambda}{1-\lambda}\F_1(.|a_1)\right)^{-1}(du)} = \frac{1}{2a_0^*}.\]
In order to calculate the left hand side, we use the alternative definition of the second L-moment $\lambda_2 = \left(\mathbb{E}[X_{2:2}] - \mathbb{E}[X_{1:2}]\right)/2$ and exploit formula (\ref{eqn:OrderStatLaw}). We have
\begin{multline*} \int_0^1{u(1-u)\left(\frac{1}{1-\lambda}\F_T - \frac{\lambda}{1-\lambda}\F_1(.|a_1)\right)^{-1}(du)} = \\ 
\int_{\mathbb{R}_+}{x\left[2\left(\frac{1}{1-\lambda}\mathbb{F}_T(x) - \frac{\lambda}{1-\lambda}\mathbb{F}_1(x|a_1)\right)-1\right]\left(\frac{1}{1-\lambda}p_T(x) - \frac{\lambda}{1-\lambda}p_1(x|a_1)\right)dx} 
\end{multline*}
A direct calculus of the right hand side shows:
\begin{multline*}
\int_{\mathbb{R}_+}{x\left[2\left(\frac{1}{1-\lambda}\mathbb{F}_T(x) - \frac{\lambda}{1-\lambda}\mathbb{F}_1(x|a_1)\right)-1\right]\left(\frac{1}{1-\lambda}p_T(x) - \frac{\lambda}{1-\lambda}p_1(x|a_1)\right)dx}  = \\ \frac{2C_1-(\lambda+1)C_2}{(1-\lambda)^2} + \frac{\lambda^2-2\lambda}{2a_1(1-\lambda)^2} + \frac{2\lambda^*\lambda}{(1-\lambda)^2(a_1+a_1^*)} + \frac{2\lambda(1-\lambda^*)}{(1-\lambda)^2(a_1+a_0^*)}
\end{multline*}
where
\begin{eqnarray*}
C2 & = & \frac{\lambda^*}{a_1^*} + \frac{1-\lambda^*}{a_0^*} \\
C1 & = & \frac{\lambda^*}{a_1^*} + \frac{1-\lambda^*}{a_0^*} - \frac{{\lambda^*}^2}{4a_1^*} - \frac{(1-\lambda^*)^2}{4a_0^*} - \frac{\lambda^*(1-\lambda^*)}{a_1^*+a_0^*}.
\end{eqnarray*}
In figure (\ref{fig:2ndLmomConstr}), we show the set of solutions of the following equation:
\begin{equation}
\frac{2C_1-(\lambda+1)C_2}{(1-\lambda)^2} + \frac{\lambda^2-2\lambda}{2a_1(1-\lambda)^2} + \frac{2\lambda^*\lambda}{(1-\lambda)^2(a_1+a_1^*)} + \frac{2\lambda(1-\lambda^*)}{(1-\lambda)^2(a_1+a_0^*)} = \frac{1}{2a_0^*},
\label{eqn:LmomentExpoEquation}
\end{equation}
for several values of $\lambda^*$ in the figure to the left. The figure to the right shows the intersection between the set of solutions and the set $\Phi^+=\{(\lambda,a), \text{ s.t. }\frac{1}{1-\lambda}\mathbb{F}_T(x) - \frac{\lambda}{1-\lambda}\mathbb{F}_1(x|a_1) \text{ is a cdf}\}$. It is clear that the nonlinear system of equations (\ref{eqn:SysLmomEq}) has an infinite number of solutions. In order to reduce the number of solutions into one, we need to consider another L-moment constraint. We do not pursue this here because the calculus is already complicated even in this simple model.\\ 
Note that the set of solutions is shrinking as the proportion of the unknown component $f_0$ becomes smaller (the value of $\lambda^*$ increases). This gives rise to a difficult and an important question; what happens if we have a number of constraints inferior to the number of parameters. This question is not pursued here.
\begin{figure}[ht]
\centering
\includegraphics[scale=0.45]{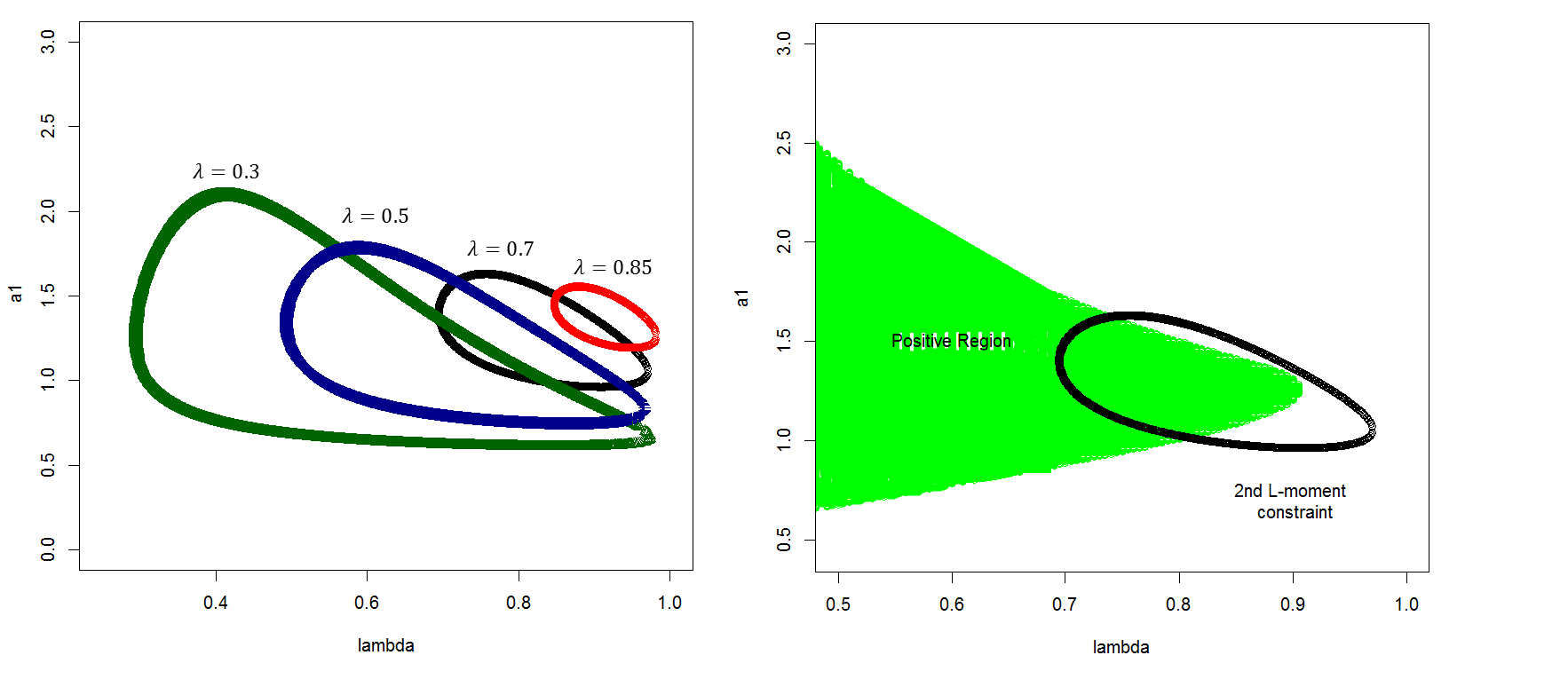}
\caption{The set of solutions under a constraint over the second L-moment. Each closed trajectory corresponds to a value of the proportion of the parametric part indicated above of it. The figure to the left represents the whole set of solutions of the equation (\ref{eqn:LmomentExpoEquation}) for different values of the true proportion $\lambda^*$. The figure to the right represents the intersection between the set of solutions of equation (\ref{eqn:LmomentExpoEquation}) for $\lambda^*=0.7$ with the set $\Phi^+$.}
\label{fig:2ndLmomConstr}
\end{figure}

\end{example}
%%%%%%%%%%%%%%%%%%%%%%%%%%%%%%%%%%%%%%%%%%%%%%%%%%%%%%%
\subsection{An algorithm for the estimation of the semiparametric mixture model}
In the context of our semiparametric mixture model, we want to estimate the parameters $(\lambda,\theta,\alpha)$ on the basis of two pieces of information; an i.i.d. sample $X_1,\cdots,X_n$ drawn from $P_T$ and the fact that ${\F_0^*}^{-1}$ belongs to the set $\mathcal{M}$. For SPLQ models, we have seen that using $\varphi-$divergences, we were able to construct an estimation procedure by minimizing some distance between the set of constraints and the distribution generating the data. The resulting estimation procedure is an optimizing problem over an infinite dimensional space. We exploited the linearity of the constraints and transformed the estimation procedure into a feasible optimization problem over $\mathbb{R}^{\ell-1}$ using the Fenchel-Legendre duality. \\
In order to use the Fenchel-Legendre duality, the constraints need to apply over the whole mixture. In our semiparametric mixture model, the constraints apply over the quantile of only one component; $\F_0^{-1}$. We thus propose to define another "model" based on $\F_0^{-1}$. We have:
\[{\mathbb{F}_0^*}^{-1} = \left(\frac{1}{1-\lambda^*}\mathbb{F}_T(.|\phi^*) - \frac{\lambda^*}{1-\lambda^*}\mathbb{F}_1(|\theta^*)\right)^{-1}.\]
Denote the associated quantile measure 
\[{\F_0^*}^{-1} = \left(\frac{1}{1-\lambda^*}\F_T(.|\phi^*) - \frac{\lambda^*}{1-\lambda^*}\F_1(|\theta^*)\right)^{-1}.\]
Define the set $\mathcal{N}^{-1}$ by:
\[\mathcal{N}^{-1} = \left\{\Q^{-1} \in M^{-1}\; : \; \exists (\lambda,\theta)\in(0,1)\times\Theta \text{ s.t. } \Q^{-1} = \left(\frac{1}{1-\lambda}\F_T - \frac{\lambda}{1-\lambda}\F_1(.|\theta)\right)^{-1}\right\}.\]
Notice that the set $\mathcal{N}^{-1}$ here is different from the set $\mathcal{N}$ defined in the previous chapter by (\ref{eqn:SetNnewModel}). Here, not all the couples $(\lambda,\theta)$ in $(0,1)\times\Theta$ are accepted, because function $\frac{1}{1-\lambda}\mathbb{F}_T - \frac{\lambda}{1-\lambda}\mathbb{F}_1(|\theta)$ may not be a cdf for for these couples. Define the set of effective parameters $\Phi^+$ by:
\begin{equation}
\Phi^+=\left\{(\lambda,\theta)\in(0,1)\times\Theta : \frac{1}{1-\lambda}\mathbb{F}_T - \frac{\lambda}{1-\lambda}\mathbb{F}_1(.|\theta) \text{ is a cdf}\right\}.
\end{equation}
Now, the set $\mathcal{N}^{-1}$ can be characterized using $\Phi^+$ by:
\[\mathcal{N}^{-1} = \left\{\left(\frac{1}{1-\lambda}\F_T - \frac{\lambda}{1-\lambda}\F_1(.|\theta)\right)^{-1}, \text{ for } (\lambda,\theta)\in\Phi^+\right\}.\]
%A simple check can still be done. It is simple to see that if the associated "density function" $\frac{1}{1-\lambda}p_T - \frac{\lambda}{1-\lambda}p_1(|\theta)$ is positive (hence, a pdf), then function $\frac{1}{1-\lambda}\mathbb{F}_T - \frac{\lambda}{1-\lambda}\mathbb{F}_1(|\theta)$ is a cdf. We have seen in the case of moment-type constraints, that the effective set of parameters $\Phi^+$ is governed by the matrix $J_H$ being definite negative which can be fulfilled if $\frac{1}{1-\lambda}p_T - \frac{\lambda}{1-\lambda}p_1(|\theta)$ is positive. This will not create any difficulty or loss of information in the final step of our procedure as we will explain later. It is the inverse. We will be able to use the whole set of parameters $\Phi$ instead of only the set where the matrix $J_H$ is definite negative.\\
%For the time being, we are obliged in order to give a proper introduction of our estimation procedure to work with couples $(\lambda,\alpha)$ which makes $\frac{1}{1-\lambda}\mathbb{F}_T - \frac{\lambda}{1-\lambda}\mathbb{F}_1(|\theta)$ a cdf.\\
The introduction of the set $\Phi^+$ is only temporary, and we will not need it at the end of this section in order to build our estimation procedure. Notice now that ${\F_0^*}^{-1}$ is a member of $\mathcal{N}^{-1}$ for $(\lambda,\theta)=(\lambda^*,\theta^*)$. On the other hand, and by definition of the semiparametric mixture model, ${\F_0^*}^{-1} \in\mathcal{M}_{\alpha^*}$. We may write:
\begin{equation}
{\F_0^*}^{-1}\in \mathcal{N}^{-1} \bigcap \cup_{\alpha}\mathcal{M}_{\alpha}.
\label{eqn:F0InIntersect}
\end{equation}
If we suppose that the intersection (which is not void) contains only one element which will be ${\F_0^*}^{-1}$ (see paragraph \ref{subsec:UniquenessSolLmom} for a discussion on the uniqueness), then it becomes reasonable to consider an estimation procedure by calculating a "distance" between the two sets $\cup_{\alpha}\mathcal{M}_{\alpha}$ and $\mathcal{N}^{-1}$. Using definition \ref{def:phiDistance}, we may write:
\begin{eqnarray}
D_{\varphi}\left(\cup_{\alpha}\mathcal{M}_{\alpha}, \mathcal{N}^{-1}\right) & = & \inf_{\Q^{-1}\in\mathcal{N}^{-1}}\inf_{\F_0^{-1}\in\cup_{\alpha}\mathcal{M}_{\alpha}} D_{\varphi}\left(\F_0^{-1},\Q^{-1}\right) \nonumber\\
 & = & \inf_{(\lambda,\theta)\in\Phi^+,\alpha\in\mathcal{A}}\inf_{\F_0^{-1}\in\mathcal{M}_{\alpha}} D_{\varphi}\left(\F_0^{-1},\left(\frac{1}{1-\lambda}\F_T - \frac{\lambda}{1-\lambda}\F_1(.|\theta)\right)^{-1}\right).\nonumber\\
\label{eqn:EstimProcLmomNotDual}
\end{eqnarray}
%where $\left(\frac{1}{1-\lambda}\F_T - \frac{\lambda}{1-\lambda}\F_1(.|\theta)\right)^{-1}$ purely denotes\footnote{We chose this notation to emphasize the link with mixture structure and its parameters.} the quantile measure associated to the cdf $\frac{1}{1-\lambda}\mathbb{F}_T - \frac{\lambda}{1-\lambda}\mathbb{F}_1(.|\theta)$. 
Now by virture of (\ref{eqn:F0InIntersect}), it holds that
\begin{equation}
(\lambda^*,\theta^*,\alpha^*) \in \arginf_{(\lambda,\theta,\alpha)\in\Phi^+}\inf_{\F_0^{-1}\in\mathcal{M}_{\alpha}} D_{\varphi}\left(\F_0^{-1},\left(\frac{1}{1-\lambda}\F_T - \frac{\lambda}{1-\lambda}\F_1(.|\theta)\right)^{-1}\right).
\label{eqn:EstimProcQuantileFuns}
\end{equation}
Next, we will treat this estimation procedure using the Fenchel duality in order to write a feasible optimization procedure, and then proceed to build upon a plug-in estimator based on an observed dataset $X_1,\cdots,X_n$.
%%%%%%%%%%%%%%%%%%%%%%%%%%%%%%%%%%%%%%%%%%%%%%%%%%%%%%%%%%%%%%%%%%%%%%%
\subsection{Estimation using the duality technique}
Applying the duality result (\ref{eqn:DualityLmom}) on the estimation procedure (\ref{eqn:EstimProcLmomNotDual}) gives:
\begin{equation}
D_{\varphi}\left(\cup_{\alpha}\mathcal{M}_{\alpha}, \mathcal{N}^{-1}\right) =  \inf_{(\lambda,\theta,\alpha)\in\Phi^+}\sup_{\xi\in\mathbb{R}^{\ell-1}} \xi^t m(\alpha) - \int_0^1{\psi\left(\xi^tK(u)\right)\left(\frac{1}{1-\lambda}\F_T - \frac{\lambda}{1-\lambda}\F_1(.|\theta)\right)^{-1}(du)}.
\label{eqn:DualityApplyLmom}
\end{equation}
\noindent In order to keep formulas clearer, we adapt the following notation:
\[
\mathbb{F}_0(y|\phi) = \frac{1}{1-\lambda} \mathbb{F}_T(y) - \frac{\lambda}{1-\lambda} \mathbb{F}_1(y|\theta)
\]
Note that we must ensure the integrability condition 
\[\int{\|K\left(\mathbb{F}_0(y|\phi)\right)\|} dx <\infty,\]
in order to be able to use the duality technique. This is ensured by the definition of the polynomial vector $K$. Indeed, there exists a constant $c$ such that:
\[\left\|K\left(\mathbb{F}_0(y|\phi)\right)\right\| \leq c\left(\mathbb{F}_0(y|\phi)\right)\left(1-\mathbb{F}_0(y|\phi)\right).\]
Since $\mathbb{F}_0(y|\phi) = \left(\frac{1}{1-\lambda}\mathbb{F}_T(y) - \frac{\lambda}{1-\lambda}\mathbb{F}_1(y|\theta)\right)$ is supposed here to be a cdf because $(\lambda,\theta,\alpha)\in\Phi^+$, it suffices then that $\mathbb{F}_0(y|\phi)$ has a finite expectation so that the previous integral becomes finite.\\
This formulation is only useful when one has a sample of i.i.d. observations of the distribution $\frac{1}{1-\lambda}P_T - \frac{\lambda}{1-\lambda}P_1(.|\theta)$ for every $\lambda$ and $\theta$, because the integral can be approximated directly using the order statistics as in formula (\ref{eqn:DualityLmomEmpirical}). We need, however, a formula which shows directly the cdf because it would permit to approximate directly the objective function and avoid the calculus of the inverse of $\frac{1}{1-\lambda}\mathbb{F}_T - \frac{\lambda}{1-\lambda}\mathbb{F}_1(.|\theta)$. Besides, the replacement of the true cdf by the empirical one does not guarantee that the difference $\frac{1}{1-\lambda}\mathbb{F}_T - \frac{\lambda}{1-\lambda}\mathbb{F}_1(.|\theta)$ remains a cdf and more complications would appear in the proof of the consistency.\\ 
Using Lemma 1.2 from \cite{AlexisThesis} we may make the change of variable desired.
\begin{equation}
\int_0^1{\psi\left(\xi^tK(u)\right)\left(\frac{1}{1-\lambda}\F_T - \frac{\lambda}{1-\lambda}\F_1(.|\theta)\right)^{-1}(du)} = \int_{\mathbb{R}}{\psi\left(\xi^tK\left(\frac{1}{1-\lambda}\mathbb{F}_T(x) - \frac{\lambda}{1-\lambda}\mathbb{F}_1(x|\theta)\right)\right)dx}.
\label{eqn:ChangeVarLmom}
\end{equation}
Employing (\ref{eqn:ChangeVarLmom}) and (\ref{eqn:DualityApplyLmom}) in (\ref{eqn:EstimProcQuantileFuns}), we may write:
\begin{equation}
(\lambda^*,\theta^*,\alpha^*) \in \arginf_{(\lambda,\theta,\alpha)\in\Phi^+}\sup_{\xi\in\mathbb{R}^{\ell-1}} \xi^tm(\alpha) - \int_{\mathbb{R}}{\psi\left(\xi^tK\left(\frac{1}{1-\lambda}\mathbb{F}_T(x) - \frac{\lambda}{1-\lambda}\mathbb{F}_1(x|\theta)\right)\right)dx}.
\label{eqn:EstimProcPhiPLus}
\end{equation}
We may now construct an estimator of $\phi^*$ by replacing $\mathbb{F}_T$ by the empirical cdf calculated on the basis of an i.i.d. sample $X_1,\cdots,X_n$. The resulting estimation procedure is still very complicated. This is mainly because we need to characterize the set $\Phi^+$. It is possible but is very expensive. For example, we may think about checking if the derivative with respect to $x$ $\frac{1}{1-\lambda}p_T(x) - \frac{\lambda}{1-\lambda}p_1(x|\theta)$ is non negative at a large randomly selected set of points. On the other hand, the set $\Phi^+$ can take \emph{fearful} forms for some mixtures. In Figure (\ref{fig:DiffFormPhiPlus}), we have two examples of $\Phi^+$. In the exponential mixture (the figure to the left), $\Phi^+$ has a "good" form in the sense that it is convex and contains $(\lambda^*,\theta^*)=(0.7,1.5)$ with a sufficiently large neighborhood around it. Thus, optimization procedures should not face any problem finding the optimum. However, in the Weibull-Lognormal mixture (the figure to the right) with $(\lambda^*,\mu^*)=(0.7,3)$, the set $\Phi^*$ is not even connected. Besides, there is not a sufficient neighborhood around $(\lambda^*,\mu^*)$ which permits an optimization algorithm to move around. During my simulations on data generated from a Weibull-Lognormal mixture distribution, the optimization algorithms could not reach such optimum and were always stuck at the initial point. A solution will be proposed in the next paragraph where we introduce the final step in the sequel of this estimation procedure.
\begin{figure}[ht]
\centering
\includegraphics[scale=0.45]{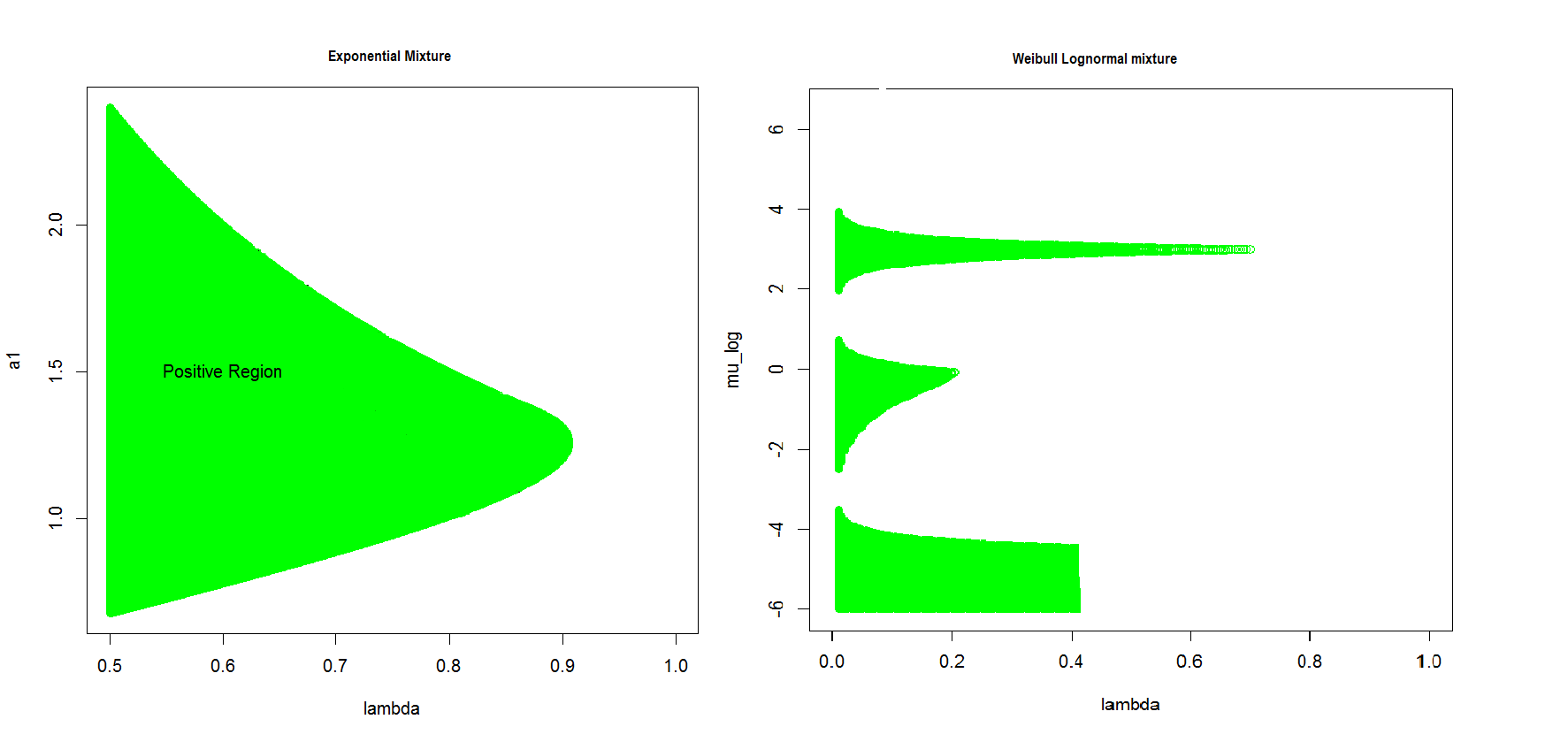}
\caption{Different forms of the set $\Phi^+$. For the Weibull-Lognormal mixture, the Weibull is the semiparametric component.}
\label{fig:DiffFormPhiPlus}
\end{figure}
%%%%%%%%%%%%%%%%%%%%%%%%%%%%%%%%%%%%%%%%%%%%%%%%%%%%%%%%%%%%%%%%%%%%%%
\subsection{The algorithm in practice and a plug-in estimate}
The problem with the estimation procedure (\ref{eqn:EstimProcPhiPLus}) is that the optimization is over the set $\Phi^+$ which may take "non-practical forms" as explained in the previous paragraph. The problem can be reread otherwise. The difficulty comes mainly from the fact that function $\frac{1}{1-\lambda}\mathbb{F}_T - \frac{\lambda}{1-\lambda}\mathbb{F}_1(.|\theta)$ may not be a cdf and the quantile would not exist. Thus, the estimation procedure in formula (\ref{eqn:EstimProcQuantileFuns}) cannot be used. We have, however, made disappear the quantiles in formula (\ref{eqn:EstimProcPhiPLus}) using a change of variable. Besides, there is no problem in calculating the optimized function in formula (\ref{eqn:EstimProcPhiPLus}) for any triplet $(\lambda,\theta,\alpha)\in\Phi$ even if the parameters do not define a proper cdf for function $\frac{1}{1-\lambda}\mathbb{F}_T - \frac{\lambda}{1-\lambda}\mathbb{F}_1(.|\theta)$. Besides, and more importantly, $\phi^*$ is a global infimum of the objective function $H(\phi,\xi(\phi))$ over the whole set $\Phi$ (and not only overy $\Phi^+$) where:
\[H(\phi,\xi)  =  \xi^t m(\alpha) - \int{\psi\left[\xi^tK\left(\frac{1}{1-\lambda} \mathbb{F}_T(y) - \frac{\lambda}{1-\lambda} \mathbb{F}_1(y|\theta)\right)\right]dy}\]
and $\xi(\phi) = \argsup_{\xi\in\mathbb{R}^{\ell-1}} H(\phi,\xi)$. Indeed, for any $\phi\in\Phi$, we have:
\[H(\phi,\xi(\phi)) \geq H(\phi,0) = 0.\]
Besides, using the duality attainment at $\phi=\phi^*$, we may write
\begin{eqnarray*}
H(\phi^*,\xi(\phi^*)) & = & \inf_{\F_0^{-1}\in\mathcal{M}_{\alpha^*}}D_{\varphi}\left(\F_0^{-1},\left(\frac{1}{1-\lambda^*}\F_T - \frac{\lambda^*}{1-\lambda^*}\F_1(.|\theta^*)\right)^{-1}\right) \\
 & = & \inf_{\F_0^{-1}\in\mathcal{M}_{\alpha^*}}D_{\varphi}\left(\F_0^{-1},{\F_0^*}^{-1}\right)\\
 & = & 0.
\end{eqnarray*}
Thus, if function $H(\phi,\xi(\phi))$ does not have several global infima inside $\Phi$, $(\lambda^*,\theta^*,\alpha^*)$ will hold as the only global minimum of it. In other words
\begin{equation}
\phi^* = \arginf_{(\alpha,\theta,\lambda)\in\Phi}\sup_{\xi\in\mathbb{R}^{\ell-1}} \xi^t m(\alpha) - \int{\psi\left[\xi^tK\left(\frac{1}{1-\lambda} \mathbb{F}_T(x) - \frac{\lambda}{1-\lambda} \mathbb{F}_1(x|\theta)\right)\right]dx}.
\label{eqn:EstimProcLmomDualCDFVersionTrue}
\end{equation}
Provided an i.i.d. sample $X_1,\cdots,X_n$ distributed from $P_T$, the cdf $\mathbb{F}_T$ can be approximated by its empirical version $\frac{1}{n}\sum{\ind{X_i\leq x}}$. Hence, $\phi^*$ can be estimated by:
\begin{equation}
\hat{\phi} = \arginf_{(\alpha,\theta,\lambda)\in\Phi}\sup_{\xi\in\mathbb{R}^{\ell-1}} \xi^t m(\alpha) - \int{\psi\left[\xi^tK\left(\frac{1}{1-\lambda} \mathbb{F}_n(x) - \frac{\lambda}{1-\lambda} \mathbb{F}_1(x|\theta)\right)\right]dx}.
\label{eqn:EstimProcLmomDualCDFVersion}
\end{equation}
\begin{remark}
Notice that the dual attainment no longer holds on the complementary set $\Phi\setminus\Phi^+$ since we are working with "signed cumulative functions". Our idea is to offer the optimization algorithm a larger neighborhood around the optimum in order to be able to find it. The important fact in the extended procedure is that $\phi^*$ is a \emph{global} infimum of the objective function. Our simulation study shows that the extension to $\Phi$ does not affect the results in several examples, and the estimator $\hat{\phi}$ is not biased and has an acceptable variance, see Section \ref{sec:SimulationLmom} for more details.
\end{remark}
%%%%%%%%%%%%%%%%%%%%%%%%%%%%%%%%%%%%%%%%%%%%%%%%%%%%%%%%%%%%%%%%%%%%%%%%%%%%%%%%
%%%%%%%%%%%%%%%%%%%%%%%%%%%%%%%%%%%%%%%%%%%%%%%%%%%%%%%%%%%%%%%%%%%%%%%%%%%%%%%%
\subsection{Uniqueness of the solution "under the model"}\label{subsec:UniquenessSolLmom}
By a unique solution we mean that only one quantile measure, which can be written in the form of $\left(\frac{1}{1-\lambda}\F_T-\frac{\lambda}{1-\lambda}\F_1(.|\theta)\right)^{-1}$ for $(\lambda,\theta)\in\Phi^+$, verifies the L-moments constraints with a unique triplet $(\lambda^*,\theta^*,\alpha^*)$. The existence of a unique solution is essential in order to ensure that the procedure (\ref{eqn:EstimProcPhiPLus}) is a reasonable estimation method. We provide next a result ensuring the uniqueness of the solution. The proof is differed to Appendix \ref{AppendSemiPara:Prop1Lmom}. The proof does not provide sufficient conditions for the existence of a unique solution over $\Phi$ because in the proof we only study the intersection $\mathcal{N}^{-1}\cap\mathcal{M}$ and characterize it without using the Fenchel duality.
\begin{proposition}
\label{prop:UniqueSolLmom}
Assume that ${\F_0^*}^{-1}\in\mathcal{M}=\cup_{\alpha}\mathcal{M}_{\alpha}$. Suppose also that:
\begin{enumerate}
\item the system of equations:
\begin{equation}
\int_{0}^1{K(u)\left(\frac{1}{1-\lambda}\F_T - \frac{\lambda}{1-\lambda}\F_1(.|\theta)\right)^{-1}(du)} = m(\alpha)
\label{eqn:NlnSysLmom}
\end{equation}
has a unique solution $(\lambda^*,\theta^*,\alpha^*)$;
\item the function $\alpha\mapsto m(\alpha)$ is one-to-one;
\item for any $\theta\in\Theta$ we have :
\[\lim_{x\rightarrow \infty} \frac{p_1(x|\theta)}{p_T(x)} = c,\quad \text{ with } c\in [0,\infty)\setminus\{1\};\]
\item the parametric component is identifiable, i.e. if $p_1(.|\theta) = p_1(.|\theta')\;\; dP_T-$a.e. then $\theta=\theta'$,
\end{enumerate}
then, the intersection $\mathcal{N}^{-1}\cap\mathcal{M}$ contains a unique measure ${\F_0^*}^{-1}$, and there exists a unique vector $(\lambda^*,\theta^*,\alpha^*)$ such that $P_T = \lambda^*P_1(.|\theta^*)+(1-\lambda)P_0^*$ where $P_0^*$ is given by (\ref{eqn:TrueP0model}) and belongs to $\mathcal{M}_{\alpha^*}$. Moreover, provided assumptions 2-4, the conclusion holds if and only if assumption 1 is fulfilled.
\end{proposition}
There is no general result for a non linear system of equations to have a unique solution; still, it is necessary to ensure that we impose a number of constraints at least equal to the number of unknown variables, otherwise there would be an infinite number of $\sigma-$finite measures in the intersection $\mathcal{N}^{-1} \bigcap \cup_{\alpha\in\mathcal{A}}\mathcal{M}_{\alpha}$.
\begin{remark}
Assumptions 3 and 4 of Proposition \ref{prop:identifiability} are used to prove the identifiability of the "model" $\left(\frac{1}{1-\lambda}P_T - \frac{\lambda}{1-\lambda}P_1(.|\theta)\right)_{\lambda,\theta}$. These conditions may be rewritten using the cdf. Furthermore, according to the considered situation we may find simpler ones for particular cases (or even for the general case). Our assumptions remain sufficient but not necessary for the proof. Note also that similar assumption to 3 can be found in the literature on semiparametric mixture models, see Proposition 3 in \cite{Bordes06b}.
\end{remark}

%%%%%%%%%%%%%%%%%%%%%%%%%%%%%%%%%%%%%%%%%%%%%%%%%%%%%%%%%%%%%%%%%%%%%%%%%%%%%%%%%%%%%%%%%%%%%%%%%%%%%%%
%
% ======================================================================
%%%%%%%%%%%%%%%%%%%%%%%%%%%%%%%%%%%%%%%%%%%%%%%%%%%%%%%%%%%%%%%%%%%%%%%%%%%%%%%%
% ======================================================================
%
%%%%%%%%%%%%%%%%%%%%%%%%%%%%%%%%%%%%%%%%%%%%%%%%%%%%%%%%%%%%%%%%%%%%%%%%%%%%%%%%%%%%%%%%%%%%%%%%%%%%%%%

\section{Asymptotic properties}
We study the asymptotic properties of the estimator $\hat{\phi}$ defined by (\ref{eqn:EstimProcLmomDualCDFVersion}). For the consistency, we will assume that function $H(\phi,\xi(\phi))$ has a unique infimum on $\Phi$. This infimum is a fortiori $\phi^*$.
On the other hand, the limiting law would not change if the infimum is truly $\phi^*$ or any other point. $\hat{\phi}$ will be centered at the infimum with a multivariate Gaussian limit law. It would not, however, be interesting unless it is centered around $\phi^*$.\\
\subsection{Consistency}
We will use Theorem \ref{theo:MainTheorem} since we are in the same context of double optimization. This time, function $H_n$ does not have the form of $P_nh$. Let's start by precising the functions $H$ and $H_n$.
\begin{eqnarray*} 
H(\phi,\xi) & = & \xi^t m(\alpha) - \int{\psi\left[\xi^tK\left(\frac{1}{1-\lambda} \mathbb{F}_T(y) - \frac{\lambda}{1-\lambda} \mathbb{F}_1(y|\theta)\right)\right]dy}; \\
H_n(\phi,\xi) & = & \xi^t m(\alpha) - \int{\psi\left[\xi^tK\left(\frac{1}{1-\lambda} \mathbb{F}_n(y) - \frac{\lambda}{1-\lambda} \mathbb{F}_1(y|\theta)\right)\right]dy},
\end{eqnarray*}
and recall the notations:
\begin{eqnarray*}
\mathbb{F}_0(y|\phi) & = & \frac{1}{1-\lambda} \mathbb{F}_T(y) - \frac{\lambda}{1-\lambda} \mathbb{F}_1(y|\theta); \\
\hat{\mathbb{F}}_0(y|\phi) & = & \frac{1}{1-\lambda} \mathbb{F}_n(y) - \frac{\lambda}{1-\lambda} \mathbb{F}_1(y|\theta);\\
\xi(\phi) & = & \argsup_{\xi\in\mathbb{R}^{\ell-1}} H(\phi,\xi); \\
\xi_n(\phi) & = & \argsup_{\xi\in\mathbb{R}^{\ell-1}} H_n(\phi,\xi).
\end{eqnarray*}
We start by calculating the difference $H(\phi,\psi) - H_n(\phi,\psi)$.
\begin{equation}
H(\phi,\xi) - H_n(\phi,\xi) = \int{\psi\left[\xi^tK\left(\hat{\mathbb{F}}_0(y|\phi)\right)\right] - \psi\left[\xi^tK\left(\mathbb{F}_0(y|\phi)\right)\right] dy}.
\label{eqn:DiffFunConsistency}
\end{equation}
The following lemma is essential for the proof of the consistency. We need to transform the optimization over $\xi$ onto a compact set. Thus, \emph{important} values of $\xi$ which are necessary for the calculus of the supremum are bounded. The proof is differed to Appendix \ref{Append:LemmaCompactXi}.
%%%%%%%%%%%%%%%%%%%%%%%%%%%%%%%%%%%%%%%%%%%%
\begin{lemma}
\label{lem:LmomCompactXi}
Suppose that function $\xi\mapsto H(\phi,\xi)$ is of class $\mathcal{C}^2(\mathbb{R}^{\ell-1})$. Then, functions $\phi\mapsto\xi(\phi)$ and $\phi\mapsto\xi_n(\phi)$ are well defined and $\mathcal{C}^1$ on the interior of the whole set $\Phi$. Moreover, if $\Phi$ is compact, then $\hat{\phi}$ and $\phi^*$ exist and the sets Im$(\xi(.))$ and Im$(\xi_n(.))$ are compact.
\end{lemma}

\noindent Differentiability of function $H$ with respect to $\xi$ can be checked in general using Lebesgue theorems, but it would not have been wise to impose an assumption over the integrand since $\psi'$ is increasing and $\xi$ is a priori in $\mathbb{R}^{\ell-1}$. For the class of functions of Cressie-Read (\ref{eqn:CressieReadPhi}), we have $\psi(t)=\frac{1}{\gamma}(\gamma t - t +1)^{\gamma/(\gamma-1)}-\frac{1}{\gamma}$. Thus, for $\gamma>1$, $\psi(t)\rightarrow\infty$ as $t\rightarrow\infty$. Therefore, it is important to study each special case alone. For example, $\psi(y)=y^2/2+y$ is the dual of the Chi square generator $\varphi(t)=(t-1)^2/2$, then $H(\xi,\phi)$ is a polynomial of degree 2 in $\xi$ and hence differentiable up to second order, see Example \ref{Example:Chi2Lmom} below for more details.\\
We state the consistency of the estimator $\hat{\phi}$ defined by (\ref{eqn:EstimProcLmomDualCDFVersion}). The proof is based on Theorem \ref{theo:MainTheorem} from the previous chapter and is differed to Appendix \ref{Append:TheoConsistLmom}.
%%%%%%%%%%%%%%%%%%%%%%%%
\begin{theorem}
\label{theo:ConsistencyLmom}
Suppose that
\begin{itemize}
\item[C1.] $\Phi$ is a compact subset of $\mathbb{R}^d$;
\item[C2.] function $\psi$ is continuously differentiable;
\item[C3.] the infimum of $\phi\mapsto H(\phi,\xi(\phi))$ is unique and isolated, i.e. $\forall \varepsilon>0, \forall \phi$ such that $\|\phi-\phi^*\|>\varepsilon$, there exists $\eta>0$ such that $H(\phi,\xi(\phi))-H(\phi^*,\xi(\phi^*))>\eta$;
\item[C4.] function $\alpha\mapsto m(\alpha)$ is continuous;
\item[C5.] function $\xi\mapsto H(\phi,\xi)$ is of class $\mathcal{C}^2(\mathbb{R}^{\ell-1})$;
\item[C6.] the integral $\int{\sqrt{\mathbb{F}_T(y)(1-\mathbb{F}_T(y))}dy}$ is finite,
\end{itemize}
then the estimator $\hat{\phi}$ defined by (\ref{eqn:EstimProcLmomDualCDFVersion}) converges in probability to $\phi^*$.
\end{theorem}
%%%%%%%%%%%%%%%%%%%%%%%
\begin{remark}
If we use $\tilde{\phi}$ defined by (\ref{eqn:EstimProcPhiPLus}), only assumption C3 should be changed. We need to suppose that the infimum exists and is unique inside $\Phi^+$ instead of the whole parameter space $\Phi$. This is less restrictive than assumption C3 since we are working inside a subset of $\Phi$.
\end{remark}
\begin{remark}
Assumption C5 is used (together with assumption C1) in order to apply Lemma \ref{lem:LmomCompactXi}. As discussed earlier after Lemma \ref{lem:LmomCompactXi}, differentiability of function $\xi\mapsto H(\phi,\xi)$ may be very difficult to check using Lebesgue theorems. When $\psi(t)=t^2/2+t$, function $H(\phi,\xi)$ is twice differentiably continuous as a function of $\xi$, because it is a polynomial of order 2 in $\xi$. Assumption C6 will be needed again in the proof of the asymptotic normality. Sufficient conditions are discussed in Remark \ref{remark:RegVar} hereafter.
\end{remark}
\begin{example}[$\chi^2$ case]
\label{Example:Chi2Lmom}
The case of the $\chi^2$ divergence is very interesting similarly to the case of moment-type constraints (see Example \ref{ex:Chi2Consistency}) simply because the optimization over $\xi$ can be calculated. Write function $H(\phi,\xi)$ for $\psi(t)=t^2/2+t$.
\[H(\phi,\xi) = \xi^tm(\alpha) - \int{\frac{1}{2}\left(\xi^tK\left(\mathbb{F}_0(y|\phi)\right)\right)^2 + \xi^tK\left(\mathbb{F}_0(y|\phi)\right) dy}.\]
This is a polynomial of order 2 in $\xi$ and thus $H(\phi,\xi)$ is of class $\mathcal{C}^2(\mathbb{R}^{\ell-1})$ as soon as the integrals exist. Indeed, for any $r\leq\ell$, there exists $c_r$ such that:
\begin{eqnarray}
\left|K_r\left(\mathbb{F}_0(y,|\phi)\right) \right| & \leq & c_r \left|\mathbb{F}_0(y,|\phi)\left(1-\mathbb{F}_0(y,|\phi)\right)\right|\label{eqn:IntegrabilityKF0Part1} \\
 & \leq & \frac{c_r}{(1-\lambda)^2}\left[\mathbb{F}_T(y)\left(1-\mathbb{F}_T(y)\right) + \lambda\mathbb{F}_T(y)(1-\mathbb{F}_1(y))+\lambda\mathbb{F}_1(1-\mathbb{F}_T(y))+ \right.\nonumber \\
 & & \left. \lambda^2\mathbb{F}_1(y)(1-\mathbb{F}_1(y))\right]. \label{eqn:IntegrabilityKF0Part2}
\end{eqnarray}
For example, if the distributions $\mathbb{F}_T$ and $\mathbb{F}_1$ are defined on $\mathbb{R}_+$, then the right hand side is integrable as soon as the expectations of $\mathbb{F}_T$ and $\mathbb{F}_1$ are finite.\\
A simple calculus of the derivative of function $\xi\mapsto H(\phi,\xi)$ gives
\begin{multline*}
\frac{\partial H}{\partial \xi} (\xi,\phi) = m(\alpha) - \int{K\left(\mathbb{F}_0(y,|\phi)\right)\xi^tK\left(\mathbb{F}_0(y,|\phi)\right)dy} + \int{K\left(\mathbb{F}_0(y,|\phi)\right)dy}.
\end{multline*}
The optimum is attained for:
\[\xi(\phi) = \Omega^{-1}\left(m(\alpha) - \int{K(\mathbb{F}_0(y|\phi))}dy\right),\]
where
\[\Omega = \int{K\left(\mathbb{F}_0(y|\phi)\right)K\left(\mathbb{F}_0(y|\phi)\right)^tdy}.\]
Furthermore, the Hessian matrix is equal to $-\Omega$, so it is symmetric definite negative whatever the value of the vector $\phi$. Thus, $\xi(\phi)$ is a global maximum of function $\xi\mapsto H(\xi,\phi)$ for any $\phi\in\Phi$. This was not the case for moment-type constraints since the Hessian matrix might be definite positive for some values of the vector $\phi$. The empirical version of this calculus is obtained similarly by replacing $\mathbb{F}_0(y|\phi)$ by $\hat{\mathbb{F}}_0(y|\phi)$.\\
Conditions of the consistency theorem can be verified. Assumption C1 is very natural in practice since in general, we have in mind a range of values for the parameters. Assumption C2 is fulfilled since $\psi(t)$ is polynomial of degree 2. Assumption 3 is not simple in general and depends on the model. Assumption C4 follows the problem we have. In Example \ref{Example:Weibull}, $m(\alpha) = (-\lambda_2,-\lambda_3,-\lambda_4)$ is continuous on $(0,\infty)\times(0,\infty)$, and assumption C4 becomes verified. We have verified assumption C5 at the beginning of the example. Assumption C6 is not restrictive. It is verified for example in an exponential mixture. The idea is to control the tail behavior of the distribution, see remark (\ref{remark:RegVar}) for a general approach.
\end{example}

%==========================================================
%%%%%%%%%%%%%%%%%%%%%%%%%%%%%%%%%%%%%%%%%%%%
\subsection{Asymptotic normality}
The convergence in law of the estimator $\hat{\phi}$ defined by (\ref{eqn:EstimProcLmomDualCDFVersion}) is not simply deduced in the same way we obtained it in the moment-constraints case. A Taylor expansion of the gradient of function $H$ would not show directly the empirical distribution which combined with the CLT gives the asymptotic normality. The expansion results in the term $\int{K(\hat{\mathbb{F}}_0(x))dx}$ which is a functional of the empirical distribution, that is
\begin{equation*}
\left(\begin{array}{c}  \sqrt{n}\left(\hat{\phi}-\phi^*\right) \\ \sqrt{n}\xi_n(\hat{\phi})\end{array}\right) = J_H^{-1}\left(\begin{array}{c}  0 \\ \sqrt{n}\left[m(\alpha^*) - \int{K(\hat{\mathbb{F}}_0(y|\phi^*))dy}\right] \end{array}\right) + o_P(1).
\end{equation*}
In the case of simply one component (no parametric component) defined through L-moment constraints, \cite{AlexisThesis} used a result based on Theorem 6 from \cite{Stigler} in order to establish the limit law of $\sqrt{n}\left[m(\alpha^*) - \int{K(\hat{\mathbb{F}}_0(y|\phi^*))dy}\right]$. This result is based on sums of order statistics which \emph{cannot} be adapted to our context since $\hat{\mathbb{F}}_0$ is an estimator of $\mathbb{F}_0$ different from the corresponding empirical distribution. We present a new result adapted to our context where the proof still shares a part of the idea of the proof of the result of \cite{Stigler}. The proof is differed to Appendix \ref{Append:PropAsyptotNormIntegral}.\\
\begin{proposition}
\label{prop:LimitLawLmomConstrPart}
Suppose that $\mathbb{E}|X_i|<\infty$. Suppose also that
\begin{eqnarray}
\int{\sqrt{\mathbb{F}_T(y)\left(1-\mathbb{F}_T(y)\right)}dy} & < & \infty, \label{eqn:AsymptotNormCond1}\\
\int{\int{\mathbb{F}_T(\min(x,y)) - \mathbb{F}_T(x)\mathbb{F}_T(y)dx}dy} & < & \infty. \label{eqn:AsymptotNormCond2}
\end{eqnarray}
For any vector $\phi=(\lambda,\theta,\alpha)\in\Phi$, we then have
\begin{multline*}
\sqrt{n}\left[\int{K\left(\frac{1}{1-\lambda}\mathbb{F}_n(y) - \frac{\lambda}{1-\lambda}\mathbb{F}_1(y|\theta)\right)dy} - \int{K\left(\frac{1}{1-\lambda}\mathbb{F}_T(y) - \frac{\lambda}{1-\lambda}\mathbb{F}_1(y|\theta)\right)dy}\right] \\ \xrightarrow{\quad \mathcal{D} \quad}{}\mathcal{N}\left(0,\Sigma\right), 
\end{multline*}
where the covariance matrix $\Sigma$ is given by
\begin{equation}
\Sigma_{r_1,r_2} = \int{\int{\left(\mathbb{F}_T\left(\min(x,y)\right) - \mathbb{F}_T(x)\mathbb{F}_T(y)\right)\sum_{k=0}^{r_1-1}{c_{r_1,k}\mathbb{F}_0(x|\phi)^k}\sum_{k=0}^{r_2-1}{c_{r_2,k}\mathbb{F}_0(y|\phi)^k}dy}dx},
\label{eqn:VarCovMatConstrPart}
\end{equation}
and $c_{r,k}=(-1)^{r-k-1}\binom{r-1}{k}\binom{r+k-1}{k}$ for $r,r_1,r_2\in\{2,\cdots,\ell\}$.
\end{proposition}
\begin{remark}
It was not possible to use a functional delta method (see \cite{Vaart} Chap. 20, Theorem 20.8) in a similar way to Theorem 3.2 in \cite{Bordes10} in order to prove the limiting law here because the functional $G\mapsto\int{K(G)}$ is not Hadamard differentiable.
\end{remark}
\begin{remark}
\label{remark:RegVar}
Integrability conditions (\ref{eqn:AsymptotNormCond1}) and (\ref{eqn:AsymptotNormCond2}) over the distribution function can be reformulated by imposing directly conditions over the distribution function using the notion of regular variations and the Lemma page 280 in \cite{Feller}. Regular variations transform the problem into conditions over the tails of the distribution functions. Suppose that there exists a constant $\rho_+<-2$ and a function $L_+(x)$ such that:
\begin{equation}
1-\mathbb{F}_T(x) = x^{\rho_+}L_+(x), \text{  with }\qquad \frac{L_+(tx)}{L_+(t)}\xrightarrow[\quad t\rightarrow\infty\quad]{} 1,\forall x>0.
\label{eqn:AsymptotNormCond1Alter1}
\end{equation}
Then, the integral $\int_y^{\infty}{\sqrt{1-\mathbb{F}_T(x)}dx}$ converges and there exists a function $M_+(y)$ such that $M_+(ty)/M_+(t)\rightarrow 1,\forall y$ and $\int_y^{\infty}{\left[1-\mathbb{F}_T(x)\right]dx} = y^{\rho_+ +1}M_+(y)$. For the neighborhood of $-\infty$, we make similar assumptions over $\mathbb{F}_T(x)$. Suppose that there exists a constant $\rho_-<-2$ and a function $L_-(x)$ such that:
\begin{equation}
\mathbb{F}_T(x) = x^{\rho_-}L_-(x), \text{  with } \qquad \frac{L_-(-tx)}{L_-(t)}\xrightarrow[\quad t\rightarrow-\infty\quad]{} 1,\forall x<0.
\label{eqn:AsymptotNormCond1Alter2}
\end{equation}
Then, the integral $\int_{-\infty}^y{\sqrt{\mathbb{F}_T(x)}dx}$ converges and there exists a function $M_-(y)$ such that $M_-(ty)/M_-(t)\rightarrow 1,\forall y$ and $\int_y^{\infty}{\mathbb{F}_T(x)dx} = y^{\rho_- +1}M_-(y)$.\\
These two assertions permit to conclude that condition (\ref{eqn:AsymptotNormCond1}) is verified since $\sqrt{\mathbb{F}_T(x)(1-\mathbb{F}_T(x))}\leq \sqrt{\mathbb{F}_T(x)}\ind{x\in(-\infty,0)} + \sqrt{1-\mathbb{F}_T(x)}\ind{x\in(0,\infty)}$. Moreover, condition (\ref{eqn:AsymptotNormCond2}) can also be check. Let's discuss what happens when $y$ is at a neighborhood of either $+\infty$ or $-\infty$. For any $y>0$, one may write:
\begin{eqnarray*}
\int_y^{+\infty}{[\mathbb{F}_T(\min(x,y))-\mathbb{F}_T(y)\mathbb{F}_T(x)]dx} & = & \mathbb{F}_T(y)\int_y^{+\infty}{[1-\mathbb{F}_T(x)]dx} \\
 & = & \mathbb{F}_T(y) y^{\rho_+ +1}M_+(y)
\end{eqnarray*}
which is integrable in a neighborhood of $+\infty$ with respect to $y$ by (\ref{eqn:AsymptotNormCond1Alter1}). On the other hand, for any $y<0$, one may write
\begin{eqnarray*}
\int_{-\infty}^y{[\mathbb{F}_T(\min(x,y))-\mathbb{F}_T(y)\mathbb{F}_T(x)]dx} & = & \left[1-\mathbb{F}_T(y)\right]\int_{-\infty}^y{\mathbb{F}_T(x)dx} \\
 & = & \left[1-\mathbb{F}_T(y)\right] y^{\rho_- +1}M_-(y)
\end{eqnarray*}
which is integrable in a neighborhood of $-\infty$ with respect to $y$ by (\ref{eqn:AsymptotNormCond1Alter2}). Thus, condition (\ref{eqn:AsymptotNormCond2}) is ensured under assumptions (\ref{eqn:AsymptotNormCond1Alter1},\ref{eqn:AsymptotNormCond1Alter2}).\\
\end{remark}
%%%%%%%%%%%%%%%%%%%%%%%%%%%%
\noindent We move on now to show the asymptotic normality of the estimator $\hat{\phi}$. Define the following matrices:
\begin{eqnarray}
J_{\phi^*,\xi^*} & = & \left(\begin{array}{c}-\int{\left[\frac{1}{(1-\lambda^*)^2}\mathbb{F}_T(y) - \frac{1}{(1-\lambda^*)^2}\mathbb{F}_1(y|\theta^*)\right]K'(\mathbb{F}_0(y|\phi^*))dy} \\ \frac{\lambda^*}{1-\lambda^*} \int{\nabla_{\theta}\mathbb{F}_1(y|\theta^*)K'(\mathbb{F}_0(y|\phi^*))^tdy} \\ \nabla m(\alpha^*) \end{array}\right)^t;\label{eqn:NormalAsymLMomJ1} \\
J_{\xi^*,\xi^*} & = & \int{K(\mathbb{F}_0(y|\phi^*))K(\mathbb{F}_0(y|\phi^*))^tdy}; \label{eqn:NormalAsymLMomJ2} \\
\tilde{\Sigma} & = & \left(J_{\phi^*,\xi^*}^t J_{\xi^*,\xi^*} J_{\phi^*,\xi^*}\right)^{-1};\nonumber \\  
H & = & \tilde{\Sigma} J_{\phi^*,\xi^*}^t J_{\xi^*,\xi^*}^{-1}; \label{eqn:NormalAsymLmomH}\\
P & = & J_{\xi^*,\xi^*}^{-1} - J_{\xi^*,\xi^*}^{-1} J_{\phi^*,\xi^*} \tilde{\Sigma} J_{\phi^*,\xi^*}^t J_{\xi^*,\xi^*}^{-1}. \label{eqn:NormalAsymLmomP}
\end{eqnarray}
We use the same notations considered at the beginning of this section for $\mathbb{F}_0(x|\phi),\hat{\mathbb{F}}_0(x|\phi),\xi(\phi)$ and $\xi_n(\phi)$.
\begin{theorem}
\label{theo:AsymptotNormLmom}
Suppose that assumptions of Proposition \ref{prop:LimitLawLmomConstrPart} are fulfilled. Suppose also that
\begin{enumerate}
\item $(\hat{\phi},\xi_n(\hat{\phi}))$ tends to $(\phi^*,0)$ in probability;
\item $\phi^*\in$ int$(\Phi)$;
\item $\alpha\mapsto m(\alpha)$ is of class $\mathcal{C}^2$;
\item there exists an integrable function $B_1$ such that $\|\nabla_{\theta}\mathbb{F}_1(y|\theta)\|\leq B_1(y)$ for $\theta$ in a neighborhood of $\theta^*$;
\item there exist integrable functions $B_{2,1}$ and $B_{2,2}$ such that $\|\nabla_{\theta}\mathbb{F}_1(y|\theta)\nabla_{\theta}\mathbb{F}_1(y|\theta)^t\|\leq B_{2,2}(y)$ and $\|J_{\mathbb{F}_1(y|\theta)}\|\leq B_{2,1}(y)$ for $\theta$ in a neighborhood of $\theta^*$;
\item the integral $\int{[\mathbb{F}_T(y)-\mathbb{F}_1(y)]dy}$ exists and is finite;
\item the matrices $J_{\xi^*,\xi^*}$ and $J_{\phi^*,\xi^*}^t J_{\xi^*,\xi^*} J_{\phi^*,\xi^*}$ are invertible.
\end{enumerate}
Then,
\[\left(\begin{array}{c}  \sqrt{n}\left(\hat{\phi}-\phi^*\right) \\ \sqrt{n}\xi_n(\hat{\phi})\end{array}\right) \xrightarrow[\mathcal{L}]{} \mathcal{N}\left(0,\left(\begin{array}{c}H \\ P\end{array}\right) \Sigma \left(H^t\quad P^t\right)\right),\]
where $H,P$ and $\Sigma$ are given respectively by formulas (\ref{eqn:NormalAsymLmomH}), (\ref{eqn:NormalAsymLmomP}) and (\ref{eqn:VarCovMatConstrPart}).
\end{theorem}
\noindent The proof of this theorem is differed to Appendix \ref{Append:TheoNormalAsymptotLmom}. In assumption 1, we could only demand the consistency of $\hat{\phi}$, since the consistency of $\xi_n(\hat{\phi})$ can be deduced from it using the continuity of $\phi\mapsto\xi(\phi)$ and the uniform convergence of $\xi_n(.)$ towards $\xi(.)$, see Lemma \ref{lem:SupXiPhiDiff}. Assumptions 4-6 are used in the proof to ensure the differentiability up to second order with respect to $\xi$ and $\phi$ of $H_n(\phi,\xi)$ for any $n$.
%%%%%%%%%%%%%%%%%%%%%%%%%%%%%%%%%%%%%%%%%%%%%%%%%%%%%%%%%%%%%%%%%%%%%
% ===============================================
%%%%%%%%%%%%%%%%%%%%%%%%%%%%%%%%%%%%%%%%%%%%%%%%%%%%%%%%%%%%%%%%%%%%
\section{Simulation study}\label{sec:SimulationLmom}
We perform several simulations and show how a prior information about the distribution of the semiparametric component $P_0$ can help us better estimate the set of parameters $(\lambda^*,\theta^*,\alpha^*)$ in regular examples, i.e. the components of the mixture can be clearly distinguished when we plot the probability density function. We also show how our approach permits to estimate even in difficult situations when the proportion of the parametric component is very low; such cases could \emph{not} be estimated using existing methods. We show also the advantage of using L-moments constraints over moment constraints using the approach developed in the previous chapter. \\
In our experiments, the datasets were generated by the following mixtures:
\begin{itemize}
\item[$\bullet$] A two-component Weibull mixture;
\item[$\bullet$] A two-component Weibull -- Lognormal mixture;
\item[$\bullet$] A two-component Gaussian -- Two-sided Weibull mixture;
\end{itemize}
We have chosen a variety of values for the parameters especially the proportion. Programming tools are the same as in the case of the moment-type constraints. We only used the $\chi^2$ divergence, because the optimization over $\xi$ can be calculated without numerical methods, see Examples \ref{Example:Chi2Lmom} and \ref{ex:Chi2Consistency}. Since the objective function $\phi\mapsto H_n(\phi,\xi_n(\phi))$ as a function of $\phi$ is not ensured to be strictly convex, we used 6 fixed initial points which we specify for each example separately. We then ran the Nelder-Mead algorithm and chose the vector of parameters for which the objective function has the lowest value. We applied a similar procedure on the algorithm of \cite{Bordes10} in order to ensure a \emph{fair} comparison.\\
All numerical integrations were calculated using function \texttt{integral} of package \texttt{pracma}. It was the only function that converged on all the calculus, see Section \ref{sec:Simulations} for more details about other numerical integration functions.\\
We did not use any function error criterion here because the compared methods do not provide the same set of parameters. For example, the method of \cite{Bordes10} estimates a mean value for the unknown component whereas our approach estimates a shape parameter. Other existing methods do not estimate any information about the parameters of the unknown component.
%%%%%%%%%%%%%%%%%%%%%%%%%%%%%%%%%%%%
\subsection{Data generated from a two-component Weibull mixture modeled by a semiparametric Weibull mixture}
We consider a mixture of two Weibull components with scales $\sigma_1 = 0.5,\sigma_2=1$ and shapes $\nu_1=2,\nu_2=1$ in order to generate the dataset. In the semiparametric mixture model, the parametric component will be "the one to the right", i.e. the component whose true set of parameters is $(\nu_1=2,\sigma_1=0.5)$.\\
We impose on the unknown component three L-moments constraints; the second, the third and the fourth Weibull L-moments. They are given in Example \ref{Example:Weibull}. We thus have 
\[m(\alpha=\nu) = \left(\begin{array}{c}
-\lambda_2=-\sigma\left(1-2^{-1/\nu}\right)\Gamma(1+1/\nu)\\
-\lambda_3=-\lambda_2\times\left(3-2\frac{1-3^{-1/\nu}}{1-2^{-1/\nu}}\right)\\
-\lambda_4=-\lambda_2\times \left(6+\frac{5(1-4^{-1/\nu})-10*(1-3^{-1/\nu})}{1-2^{-1/\nu}}\right)
\end{array}
\right) \]
and $K(t) = (t(t-1),t(t-1)(2t-1),t(t-1)(1+5(t-1)+5(t-1)^2))^t$. This mixture was not easily estimated by either our estimation procedure or the semiparametric methods from the literature. Our estimator, although has a higher variance, is still not biased in the same way estimates of other methods are. The L-moment constraints gave an estimator with less variance than the estimator based on moments constraints, but with slightly higher bias on the proportion.
\begin{table}[ht]
\centering
\begin{tabular}{|c|c|c|c|c|c|c|}
\hline
Nb of observations & $\lambda$ & sd$(\lambda)$ & $\nu_1$ & sd($\nu_1$) & $\nu_2$ & sd($\nu_2$)\\
\hline
\hline
\multicolumn{7}{|c|}{Mixture 1 : $n=10^4$ $\lambda^* = 0.3$, $\nu_1^*=2$, $\sigma_1^*=0.5$(fixed), $\nu_2^*=1$, $\sigma_2^*=1$(fixed) }\\
\hline
Pearson's $\chi^2$ 3 moments & 0.304 & 0.016 & 2.191 & 0.887 & 0.998 & 0.013 \\
Pearson's $\chi^2$ 3 L-moments & 0.348 & 0.062 & 1.828 & 0.648 & 0.984 & 0.021 \\
Robin & 0.604 & 0.029 & 1.256  & 0.037 & --- & --- \\
Song EM-type & 0.806 & 0.005 & 1.185 & 0.018 & --- & --- \\
Song $\pi-$maximizing & 0.624 & 0.007 & 1.312 & 0.013 & --- & --- \\
\hline
\end{tabular}
\caption{The mean value with the standard deviation of estimates in a 100-run experiment on a two-component Weibull mixture.}
\label{tab:3by3ResultsWeibullLMoment}
\end{table}

%%%%%%%%%%%%%%%%%%%%%%%%%%%%%%%%%%%%%%%%%

\subsection{Data generated from a two-component Weibull-LogNormal mixture modeled by a semiparametric Weibull-LogNormal mixture}
We consider a dataset generated from a mixture of a Weibull and a Lognormal distributions. The Weibull component has a scale $\sigma_1^*=1$ and a shape $\nu_1^*\in\{1.5,1,0.4\}$ in order to illustrate several scenarios; a distribution whose pdf explodes to infinity at zero, a distribution whose pdf has finite value at zero and a distribution whose pdf goes back to zero at zero. The Lognormal component has a scale $\sigma_2^*=0.5$ and a mean parameter $\mu^*=3$. The Lognormal distribution has a heavy tail which is inherited in the mixture distribution.\\
In a first part, we perform a comparison of convergence speed between the method under moments constraints and the method under L-moments constraints as we increase the number of observations $n$.  Details about the simulations under moments constraints can be found in paragraph \ref{subsec:WeibLognormMom}. The Weibull component is considered as the unknown component during estimation, and impose three L-moments constraints. The first 4 L-moments of the Weibull distribution are given in Example \ref{Example:Weibull}.\\
In a second part, we perform an estimation of a semiparametric mixture model where the Lognormal component is considered unknown and defined through 3 L-moments conditions; the second, the third and the fourth L-moment.  The L-moments of the Lognormal distribution do not have a close formula and are calculated numerically using function \texttt{lmrln3} of package \texttt{lmom} written by Hosking.\\
Results in table (\ref{tab:3by3ResultsWeibullLognormLmom}) show that L-moments are more informative and we need less data in order to get good estimates in comparison to moments constraints. In order to calculate the estimate $\hat{\phi}$, we considered 6 initial points; namely the set 
\[\phi^{(0)} \in \left\{(0.8,2,1),(0.5,2,1),(0.8,1,1),(0.7,3,1.5),(0.7,2,2),(0.5,4,2),(0.5,1.5,2)\right\}.\] 
The vector $\hat{\phi}$ was taken as the one which corresponds to the lowest value among the infima produced by the optimization algorithm.\\
In table (\ref{tab:3by3ResultsLognormWeibullLmoment}) the Lognormal component is the unknown component during estimation. Initialization of the optimization algorithm, for example in mixture 2, was taken from the set $\{(0.1,0.5,1),(0.15,0.5,0.7),(0.05,1.5,2.5),(0.1,1,3)\}$.\\
It is clear that the moments constraints gave better results than L-moments constraints in mixture 1 for the estimation of the scale of the Weibull component. For the second mixture, both types of constraints give similar results. The two methods have the same bias in the estimation of the scale of Weibull component; the moments constraints produced a positive bias whereas the L-moments constraints produced a negative bias. The L-moments produced a smaller variance. In the third mixture, the L-moments constraints gave clear better results. The last mixture is the most difficult one in the sense that the proportion of the parametric component is very low. 

\begin{table}[ht]
\centering
\begin{tabular}{|c|c|c|c|c|c|c|c|}
\hline
nb of observations & Estimation method & $\lambda$ & sd$(\lambda)$ & $\mu$ & sd($\mu$) & $\nu$ & sd($\nu$)\\
\hline
\hline
\multicolumn{8}{|c|}{True Parameters : $\lambda^* = 0.7$, $\mu^*=3$, $\sigma_2^*=0.5$(fixed), $\nu^*=1.5$, $\sigma_1^*=1$(fixed) }\\
\hline
\multirow{2}{2.5cm}{$n = 10^2$} & Pearson's $\chi^2$ L-moments & 0.685 & 0.069 & 2.798 & 0.413 & 0.436 & 0.074 \\
 &  Pearson's $\chi^2$ Moments & 0.384 & 0.117 & 2.654 & 0.153 & 0.488 & 0.018 \\
\hline 
\multirow{2}{2.5cm}{$n=10^3$} &  Pearson's $\chi^2$ L-moments & 0.677 & 0.017 & 3.014 & 0.028 & 0.726 & 0.272 \\
 &  Pearson's $\chi^2$ Moments & 0.518 & 0.068 & 2.806 & 0.099 & 0.473 & 0.014 \\
\hline
\multirow{2}{2.5cm}{$n=10^4$} &  Pearson's $\chi^2$ L-moments & 0.697 & 0.009 & 3.003 & 0.010 & 1.343 & 0.185 \\
 &  Pearson's $\chi^2$ Moments & 0.605 & 0.044 & 2.903 & 0.069 & 0.531 & 0.326 \\
\hline
\end{tabular}
\caption{The mean value with the standard deviation of estimates in a 100-run experiment on a two-component Weibull-log normal mixture. The parametric component is the log-normal with unknown mean parameter $\mu$. The semiparametric component is the Weibull component which is defined by its first three L-moments (moments resp.) with unknown shape $\nu$.}
\label{tab:3by3ResultsWeibullLognormLmom}
\end{table}

\begin{table}[ht]
\centering
\begin{tabular}{|c|c|c|c|c|c|c|}
\hline
Nb of observations & $\lambda$ & sd$(\lambda)$ & $\nu$ & sd($\nu$) & $\mu$ & sd($\mu$)\\
\hline
\hline
\multicolumn{7}{|c|}{Mixture 1 : $n=10^3$, $\lambda^* = 0.3$, $\nu^*=1.5$, $\sigma_1^*=1$(fixed), $\mu^*=3$, $\sigma_2^*=0.5$(fixed) }\\
\hline
Pearson's $\chi^2$ L-moments& 0.313 & 0.019 & 1.027 & 0.541 & 2.992 & 0.050 \\
Pearson's $\chi^2$ Moments& 0.308 & 0.017 & 1.484 & 0.624 & 3.002 & 0.026 \\
\hline
\hline
\multicolumn{7}{|c|}{Mixture 2 : $n=10^4$, $\lambda^* = 0.1$, $\nu^*=1$, $\sigma_1^*=1$(fixed), $\mu^*=3$, $\sigma_2^*=0.5$(fixed) }\\
\hline
Pearson's $\chi^2$ L-moments & 0.104 & 0.006 & 0.795 & 0.379 & 2.994 & 0.015 \\
Pearson's $\chi^2$ Moments & 0.103 & 0.006 & 1.284 & 0.677 & 3.001 & 0.007 \\
%\hline
%\hline
%\multicolumn{7}{|c|}{Mixture 3 : $n=10^4$, $\lambda^* = 0.05$, $\nu^*=1$, $\sigma_1^*=1$(fixed), $\mu^*=3$, $\sigma_2^*=0.5$(fixed) }\\
%\hline
%Pearson's $\chi^2$ & 0.052 & 0.004 & 1.312 & 0.703 & 3.001 & 0.006 \\
\hline
\hline
\multicolumn{7}{|c|}{Mixture 3 : $n=5\times 10^4$, $\lambda^* = 0.05$, $\nu^*=0.4$, $\sigma_1^*=1$(fixed), $\mu^*=3$, $\sigma_2^*=0.5$(fixed) }\\
\hline
Pearson's $\chi^2$ L-Moments & 0.049 & 0.002 & 0.448 & 0.129 & 3.000 & 0.006 \\
Pearson's $\chi^2$ Moments & 0.049 & 0.002 & 0.629 & 0.438 & 3.001 & 0.004 \\
\hline
\end{tabular}
\caption{The mean value with the standard deviation of estimates in a 100-run experiment on a two-component Weibull-log normal mixture. The parametric component is the Weibull with unknown shape $\nu$. The semiparametric component is the lognormal component which is defined by its first three L-moments (moments resp.) with unknown mean parameter $\mu$.}
\label{tab:3by3ResultsLognormWeibullLmoment}
\end{table}

%%%%%%%%%%%%%%%%%%%%%%%%%%%%%%%%%%%%%%%%%

\subsection{Data generated from a two-sided Weibull Gaussian mixture modeled by a semiparametric two-sided Weibull Gaussian mixture}
We have already presented this model in paragraph \ref{subsec:TwoSidGaussMom}. The 2nd, 3rd and 4th L-moments of the two-sided Weibull distribution are given by:
\begin{eqnarray*}
\lambda_2 & = & \left[1-\frac{1}{2^{1+1/\nu}}\right]\sigma_2\Gamma\left(1+\frac{1}{\nu}\right);\\
\lambda_3 & = & 0 ;\\
\lambda_4 & = & \left[1-\frac{6}{2^{1+1/\nu}}+\frac{15}{2\times 3^{1+1/\nu}}-\frac{5}{2\times 4^{1+1/\nu}}\right] \sigma_2\Gamma\left(1+\frac{1}{\nu}\right).
\end{eqnarray*}
Results are presented in table (\ref{tab:3by3ResultsTwoSideWeibullGaussLMom}). The L-moments constraints produce clear better results than the moments constraints in all the mixtures. The estimation based on L-moments constraints produced clear lower variance. Besides, and once again, the L-moments constraints seem to be more informative and we need less number of observations than moments constraints in order to produce good estimates.\\
In this example we presented a challenge to our estimation method by simulating mixtures with very low proportion of the parametric part; mixture 3 with $\lambda^*=0.05$ and mixture 4 with $\lambda^*=0.01$. Using signal-noise terms, in mixture 4, only one percent of the data comes from the signal whereas $99\%$ of the data is pure noise. The location of the signal is then estimated around zero with standard deviation of $0.3$ with the L-moments constraints. It is not well localized however using moments constraints with $10^5$ observations, and we need at least $10^8$ observations to reach a similar precision to the result obtained with L-moments constraints. It is still important to notice that using moments or L-moments constraints, we were able to confirm the existence of a signal component (the parametric component).\\
In what concerns the initialization of the algorithm under L-moments constraints, we used:
\begin{eqnarray*}
\text{Mix 1} & : & \left\{(0.8,1,1),(0.5,-1,2.5),(0.8,0.5,2),(0.7,0,3),(0.7,1,4),(0.5,2,3.5)\right\} \\
\text{Mix 2} & : & \left\{(0.2,1,1),(0.5,-1,2.5),(0.2,0.5,2),(0.3,0,3),(0.3,1,4)\right\} \\
\text{Mix 3} & : & \left\{(0.1,1,1),(0.05,-1,2.5),(0.03,0.5,2),(0.01,0,1.5),(0.005,1,0.7)\right\} \\
 \text{Mix 4} & : & \left\{(0.1,1,1),(0.005,1,0.7)\right\}
\end{eqnarray*}
For the last mixture, we have found no changes in using more initial points than the two given points. Besides, execution time was very long (about 5 samples per day), so we preferred to use only two starting points.
\begin{table}[ht]
\centering
\begin{tabular}{|c|c|c|c|c|c|c|}
\hline
Estimation method & $\lambda$ & sd$(\lambda)$ & $\mu$ & sd($\mu$) & $\nu$ & sd($\nu$)\\
\hline
\hline
\multicolumn{7}{|c|}{Mixture 1 : $n=100$, $\lambda^* = 0.7$, $\mu^*=0$, $\sigma_2^*=0.5$(fixed), $\nu^*=3$, $\sigma_1^*=1.5$(fixed) }\\
\hline
Pearson's $\chi^2$ -- L-Moments & 0.758 & 0.067 & -2.28$\times 10^{-3}$ & 0.098 & 3.040  & 0.639 \\
Pearson's $\chi^2$ under $\mathcal{M}_{2:4}$ & 0.764 & 0.067 & -0.012 & 0.342 & 2.893  & 0.731 \\
%Bordes symmetry Triangular Kernel & 0.309 & 0.226 & 0.240 & 0.609 & $\mu_2=-$0.220 & sd$(\mu_2)$0.398 \\
%Bordes symmetry Gaussian Kernel & 0.211 & 0.133 & 0.106 & 0.533 & $\mu_2=-$0.035 & sd$(\mu_2)$0.203 \\
\hline
\hline
\multicolumn{7}{|c|}{Mixture 2 : $n=100$, $\lambda^* = 0.3$, $\mu^*=0$, $\sigma_2^*=0.5$(fixed), $\nu^*=3$, $\sigma_1^*=1.5$(fixed) }\\
\hline
Pearson's $\chi^2$ -- L-Moments & 0.364 & 0.082 & -0.016 & 0.246 & 3.058  & 0.418 \\
Pearson's $\chi^2$ under $\mathcal{M}_{2:4}$ & 0.407 & 0.077 & 0.012 & 0.575 & 2.925  & 0.454 \\
%Bordes symmetry Triangular Kernel & 0.272 & 0.119 & 0.773 & 0.947 & $\mu_2=-$0.430 & sd$(\mu_2)$0.393 \\
%Bordes symmetry Gaussian Kernel & 0.206 & 0.104 & 0.855 & 0.911 & $\mu_2=-$0.308 & sd$(\mu_2)$0.350 \\
\hline
\hline
\multicolumn{7}{|c|}{Mixture 3 : $n=5000$, $\lambda^* = 0.05$, $\mu^*=0$, $\sigma_2^*=0.5$(fixed), $\nu^*=1.5$, $\sigma_1^*=2$(fixed) }\\
\hline
Pearson's $\chi^2$ -- L-Moments & 0.050 & 0.013 & 0.026 & 0.365 & 1.496  & 0.020 \\
Pearson's $\chi^2$ under $\mathcal{M}_{2:4}$ 0.066 & 0.013 & -0.036 & 0.857 & 1.493 & 0.008\\
\hline
\hline
\multicolumn{7}{|c|}{Mixture 4 : $n=10^5$, $\lambda^* = 0.01$, $\mu^*=0$, $\sigma_2^*=0.5$(fixed), $\nu^*=1.5$, $\sigma_1^*=2$(fixed) }\\
\hline
Pearson's $\chi^2$ -- L-Moments & 0.011 & 0.003 & 0.023 & 0.377 & 1.500  & 0.005 \\
Pearson's $\chi^2$ under $\mathcal{M}_{2:4}$ & 0.025 & 0.010&  - 0.047 & 1.356 & 1.495 & 0.006\\
\hline
\end{tabular}
\caption{The mean value with the standard deviation of estimates in a 100-run experiment on a two-component two-sided Weibull--Gaussian mixture under L-moment constraints.}
\label{tab:3by3ResultsTwoSideWeibullGaussLMom}
\end{table}
\clearpage
\subsection{Conclusions}
In this chapter, we introduced another structure for semiparametric mixture models with unknown component by imposing L-moments constraints on it. The method was proved to be consistent and asymptotic normal under standard assumptions. The estimation method under L-moments constraints presented several advantages in comparison to the estimation method under moments constraints. We were able to estimate over the whole parameter space and no need to check if the optimized function $\xi\mapsto H(\phi,\xi)$ is strictly concave for every $\phi$. Although the estimation method under L-moments constraints need numerical integrations (which is not the case of moments-type constraints procedure), the resulting estimator seems to have lower variance. Moreover, L-moments are demonstrated through simulations to be more informative than moments constraints, and we need less number of observations in order to obtain good estimates.

%%%%%%%%%%%%%%%%%%%%%%%%%%%%%%%%%%%%%%%%%%%%%%%%%%%%%%%%%%%%%%%%%%%%%%%%%%%%%%%%%%%%%%%%%%%%%%%%%%%%%%%
%
% ======================================================================
%%%%%%%%%%%%%%%%%%%%%%%%%%%%%%%%%%%%%%%%%%%%%%%%%%%%%%%%%%%%%%%%%%%%%%%%%%%%%%%%
% ======================================================================
%
%%%%%%%%%%%%%%%%%%%%%%%%%%%%%%%%%%%%%%%%%%%%%%%%%%%%%%%%%%%%%%%%%%%%%%%%%%%%%%%%%%%%%%%%%%%%%%%%%%%%%%%

\section{Appendix: Proofs}
\subsection{Proof of Proposition \ref{prop:identifiabilityMixtureLmom}}\label{AppendSemiPara:PropIdenitifiabilityLmom}
\begin{proof}
Denote $M^{1}$ the set of all probability measures. Based on equation (\ref{eqn:IdenitifiabilityDefEqLmom}), we have:
\begin{eqnarray*}
P_0 & = & \frac{1}{1-\lambda} P_T - \frac{\lambda}{1-\lambda}P_1(.|\theta) \\
\tilde{P}_0 & = & \frac{1}{1-\tilde{\lambda}} P_T - \frac{\tilde{\lambda}}{1-\tilde{\lambda}}P_1(.|\tilde{\theta})
\end{eqnarray*}
Define the following function:
\[G:\mathbb{R}^{d-s}\times M^+\rightarrow \text{Im}(G)\subset M^1: (\lambda,\theta,P_0)\mapsto \lambda P_1(.|\theta) + (1-\lambda)P_0.\]
where 
\[M^+ = \{P_0 \in M^1 \text{  s.t. } \F_0^{-1}\in\mathcal{M}\}.\]
Identifiability is now equivalent to the fact that function $G$ is one-to-one. This means that for a given mixture distribution $P_T\in$Im$(G)$, we need that there exists a unique triplet $(\lambda,\theta,P_0)$ such that
\[P_T = \lambda P_1(.|\theta) + (1-\lambda)P_0\]
In other words:
\[P_0 = \frac{1}{1-\lambda}P_T - \frac{\lambda}{1-\lambda}P_1(.|\theta)\]
The equality of measures imply the equality of the quantiles. Thus, we may write:
\begin{equation}
\int_0^1{K(u)\F_0^{-1}(du)} = m(\alpha) = \int_0^1{K(u)\left(\frac{1}{1-\lambda}\F_T - \frac{\lambda}{1-\lambda}\F_1(.|\theta)\right)^{-1}(du)}.
\label{eqn:SysLmom}
\end{equation}
The assumption of the present proposition imposes the existence of unique solution $(\lambda^*,\theta^*,\alpha^*)$ to the previous nonlinear system of equations. Let's go back to function $G$. For a given mixture distribution $P_T\in$ Im$(G)$, take $\lambda=\lambda^*,\theta=\theta^*$ to be the solution to the nonlinear system (\ref{eqn:SysLmom}),  and define $P_0^*$ by:
\[P_0^* = \frac{1}{1-\lambda^*}P_T - \frac{\lambda^*}{1-\lambda^*}P_1(.|\theta^*).\]
Notice that $P_0^*\in\mathcal{M}_{\alpha^*}$. Suppose that $P_T$ can be written in two manners. In other words, suppose that there exists another triplet $(\tilde{\lambda},\tilde{\theta},\tilde{P}_0)$ with $\tilde{P}_0\in\mathcal{M}_{\tilde{\alpha}}$ such that:
\[P_T = \tilde{\lambda} P_1(.|\tilde{\theta}) + (1-\tilde{\lambda})\tilde{P}_0.\]
We then have:
\[\tilde{P}_0 = \frac{1}{1-\tilde{\lambda}}P_T - \frac{\tilde{\lambda}}{1-\tilde{\lambda}}P_1(.|\tilde{\theta}),\]
and consequently,
\[m(\tilde{\alpha}) = \int_0^1{K(u)\left(\frac{1}{1-\tilde{\lambda}}\F_T - \frac{\tilde{\lambda}}{1-\tilde{\lambda}}\F_1(.|\tilde{\theta})\right)^{-1}(du)}.\]
Thus, $(\tilde{\lambda},\tilde{\theta},\tilde{\alpha})$ is a second solution to the system (\ref{eqn:SysLmom}). Nevertheless, the system of equations (\ref{eqn:SysLmom}) has a unique solution by assumption of the present proposition. Hence, a contradiction is reached and the triplet $(\lambda^*,\theta^*,P_0^*)$ is unique. We conclude that function $G$ is one-to-one and the semiparametric mixture model subject to L-moments constraints is identifiable.
\end{proof}

%%%%%%%%%%%%%%%%%%%%%%%%%%%%%%%%%%%%%%%%%%%%%%%%%%%%%%%
%%%%%%%%%%%%%%%%%%%%%%%%%%%%%%%%%%%%%%%%%%%%%%%%%%%%%%%

\subsection{Proof of Proposition \ref{prop:UniqueSolLmom}}\label{AppendSemiPara:Prop1Lmom}
\begin{proof}
Let $\F_0^{-1}$ be some quantile measure which belongs to the intersection $\mathcal{N}^{-1} \cap \mathcal{M}$. Since $\F_0^{-1}$ belongs to $\mathcal{N}^{-1}$, there exists a couple $(\lambda,\theta)\in\Phi^+$ such that:
\begin{equation}
\F_0^{-1} = \left(\frac{1}{1-\lambda} \F_T - \frac{\lambda}{1-\lambda} \F_1(.|\theta)\right)^{-1}.
\label{eqn:SetNelementQuantile}
\end{equation}
This couple is unique by virtue of assumptions 3 and 4. Indeed, let $(\lambda,\theta)$	and $(\tilde{\lambda},\tilde{\theta})$ be two couples such that:
\begin{equation*}
\left(\frac{1}{1-\lambda} \F_T - \frac{\lambda}{1-\lambda} \F_1(.|\theta)\right)^{-1} = \left(\frac{1}{1-\tilde{\lambda}} \F_T - \frac{\tilde{\lambda}}{1-\tilde{\lambda}} \F_1(.|\tilde{\theta})\right)^{-1}
\end{equation*}
This entails that:
\begin{equation}
\frac{1}{1-\lambda} \mathbb{F}_T(x) - \frac{\lambda}{1-\lambda} \mathbb{F}_1(x|\theta) = \frac{1}{1-\tilde{\lambda}} \mathbb{F}_T(x) - \frac{\tilde{\lambda}}{1-\tilde{\lambda}} \mathbb{F}_1(x|\tilde{\theta}).
\label{eqn:identifEqualityQuantile}
\end{equation}
By derivation of both sides, we get an identity in the densities:
\[\frac{1}{1-\lambda} - \frac{\lambda}{1-\lambda} \frac{p_1(x|\theta)}{p_T(x)} = \frac{1}{1-\tilde{\lambda}} - \frac{\tilde{\lambda}}{1-\tilde{\lambda}} \frac{p_1(x|\tilde{\theta})}{p_T(x)}.\]
Taking the limit as $x$ tends to $\infty$ results in:
\[\frac{1-c\lambda}{1-\lambda}  = \frac{1-c\tilde{\lambda}}{1-\tilde{\lambda}}.\]
Note that function $z\mapsto (1-cz)/(1-z)$ is strictly monotone as long as $c\neq 1$. Hence, it is a one-to-one map. Thus $\lambda=\tilde{\lambda}$. Inserting this result in equation (\ref{eqn:identifEqualityQuantile}) entails that:
\[\mathbb{F}_1(.|\theta) = \mathbb{F}_1(.|\tilde{\theta}).\]
Using the identifiability of $P_1$ (assumption 4), we get $\theta=\tilde{\theta}$ which proves the existence of a unique couple $(\lambda,\theta)$ in (\ref{eqn:SetNelementQuantile}).\\
On the other hand, since $\F_0^{-1}$ belongs to $\mathcal{M}$, there exists a unique $\alpha$ such that $\F_0^{-1}\in\mathcal{M}_{\alpha}$. Uniqueness comes from the fact that the function $\alpha\mapsto m(\alpha)$ is one-to-one (assumption 2). Thus, $\F_0^{-1}$ verifies the constraints
\[\int_0^1{K(u)\F_0^{-1}(du)} = m(\alpha).\]
Combining this with (\ref{eqn:SetNelementQuantile}), we get:
\begin{equation}
\int_0^1{K(u)\left(\frac{1}{1-\lambda} \F_T - \frac{\lambda}{1-\lambda} \F_1(.|\theta)\right)^{-1}(du)} = m(\alpha).
\label{eqn:NlnSysMalphaQuantile}
\end{equation}
This is a non linear system of equations with $\ell$ equations. Now, let $\F_0^{-1}$ and $\tilde{\F}_0^{-1}$ be two elements in $\mathcal{N}^{-1}\cap\mathcal{M}$, then there exist two couples $(\lambda,\theta)$ and $(\tilde{\lambda},\tilde{\theta})$ with $\lambda\neq\tilde{\lambda}$ or $\theta\neq\tilde{\theta}$ such that $\F_0^{-1}$ and $\tilde{\F}_0^{-1}$ can be written in the form of (\ref{eqn:SetNelementQuantile}) with respectively $(\lambda,\theta)$ and $(\tilde{\lambda},\tilde{\theta})$. Since $\F_0^{-1}\in\mathcal{M}$, there exists $\alpha$ such that $\F_0^{-1}\in\mathcal{M}_{\alpha}$. Similarly, there exists $\tilde{\alpha}$ possibly different from $\alpha$ such that $\tilde{\F}_0^{-1}\in\mathcal{M}_{\tilde{\alpha}}$. Now, $(\lambda,\theta,\alpha)$ and $(\tilde{\lambda},\tilde{\theta},\tilde{\alpha})$ are two solutions to the system of equations (\ref{eqn:NlnSysMalphaQuantile}) which contradicts with assumption 1 of the present proposition.\\
We may now conclude that, if a quantile measure $\F_0^{-1}$ belongs to the intersection $\mathcal{N}^{-1} \cap \mathcal{M}$, then it has the representation (\ref{eqn:SetNelementQuantile}) for a unique couple $(\lambda,\theta)$ and there exists a unique $\alpha$ such that the triplet $(\lambda,\theta,\alpha)$ is a solution to the non linear system (\ref{eqn:NlnSysMalphaQuantile}). Conversely, if there exists a triplet $(\lambda,\theta,\alpha)$ which solves the non linear system (\ref{eqn:NlnSysMalphaQuantile}), then the quantile measure $\F_0^{-1}$ defined by $\F_0^{-1} = \left(\frac{1}{1-\lambda} \F_T - \frac{\lambda}{1-\lambda} \F_1(.|\theta)\right)^{-1}$ belongs to the intersection $\mathcal{N}^{-1} \cap \mathcal{M}$. This is because on the one hand, it clearly belongs to $\mathcal{N}^{-1}$ by its definition and on the other hand, it belongs to $\mathcal{M}_{\alpha}$ since it verifies the constraints and thus belongs to $\mathcal{M}$.\\
It is now reasonable to conclude that under assumptions 2-4, the intersection $\mathcal{N}^{-1} \cap \mathcal{M}$ includes a \emph{unique} quantile measure $\F_0^{-1}$ if and only if the set of $\ell$ non linear equations (\ref{eqn:NlnSysMalpha}) has a unique solution $(\lambda,\theta,\alpha)$.
\end{proof}

%%%%%%%%%%%%%%%%%%%%%%%%%%%%%%%%%%%%%%%%%%%%%%%%%%%%%%%%%%%%%%%%%%%%%%%%%%%%%%%%%%%%%%%%%%%%%%%%%%%%%%%
%%%%%%%%%%%%%%%%%%%%%%%%%%%%%%%%%%%%%%%%%%%%%%%%%%%%%%%%%%%%%%%%%%%%%%%%%%%%%%%%%%%%%%%%%%%%%%%%%%%%%%%

\subsection{Proof of Lemma \ref{lem:LmomCompactXi}}\label{Append:LemmaCompactXi}
\begin{proof}
The same arguments hold for both functions $\xi(\phi)$ and $\xi_n(\phi)$. We therefore, proceed with $\xi(\phi)$. Function $\xi\mapsto H(\phi,\xi)$ is strictly concave since\footnote{One can prove the strict concavity simply by calculating $H(\phi,u\xi_1+(1-u)\xi_2)$.} it is $\mathcal{C}^2$ and have the following Hessian matrix:
\[J_{H(\phi,.)} = -\int{K\left(\mathbb{F}_0(y,|\phi)\right)K\left(\mathbb{F}_0(y,|\phi)\right)^t} \psi''\left(\xi^tK\left(\mathbb{F}_0(y,|\phi)\right)\right)dy.\]
Since $\psi$ is strictly convex, then $\psi''(z)>0$ for any $z$. Thus the matrix $J_{H(\phi,.)}$ is definite negative and $\xi\mapsto H(\phi,\xi)$ is strictly concave. By the implicit function theorem, function $\phi\mapsto\xi(\phi)$ is uniquely defined and $\mathcal{C}^1$ over int$(\Phi)$. Notice here that even if $\frac{1}{1-\lambda} \mathbb{F}_T(y) - \frac{\lambda}{1-\lambda} \mathbb{F}_1(y|\theta)$ is negative, the matrix $J_{H(\phi,.)}$ can still be definite negative unlike the case of moment constraints.\\
The second part of the proposition is a direct consequence from the continuity of function $\phi\mapsto\xi(\phi)$.
\end{proof}

%%%%%%%%%%%%%%%%%%%%%%%%%%%%%%%%%%%%%%%%%%%%%%%%%%%%%%%%%%%%%%%%%%%%%%%%%%%%%%%%%%%%%%%%%%%%%%%%%%%%%%%
%%%%%%%%%%%%%%%%%%%%%%%%%%%%%%%%%%%%%%%%%%%%%%%%%%%%%%%%%%%%%%%%%%%%%%%%%%%%%%%%%%%%%%%%%%%%%%%%%%%%%%%
\subsection{Proof of Theorem \ref{theo:ConsistencyLmom}}\label{Append:TheoConsistLmom}
\begin{proof}
We will use Theorem \ref{theo:MainTheorem}. We start with assumption A2. We prove, first, that the supremum over $\xi$ can only be calculated over a compact subset of $\mathbb{R}^l$. This is a direct result from Lemma \ref{lem:LmomCompactXi}. One can redefine the estimator by maximizing over $\xi$ on the subset $\Xi=$Im$(\xi(.))\subset\mathbb{R}^l$ independently of $\phi$. We thus have:
\begin{eqnarray*}
D_{\varphi}(\mathcal{M}_{\alpha},\mathbb{F}_0(.|\phi)) & = & \sup_{\xi\in\Xi} H(\phi,\xi) \\
\phi^* & = & \arginf_{\phi}\sup_{\xi\in\Xi} H(\phi,\xi).
\end{eqnarray*}
We redefine now the estimation procedure (\ref{eqn:EstimProcLmomDualCDFVersion}) as follows:
\[\hat{\phi} = \arginf_{\alpha,\theta,\lambda}\sup_{\xi\in\Xi} \xi^t m(\alpha) - \int{\psi\left[\xi^tK\left(\frac{1}{1-\lambda} \mathbb{F}_T(y) - \frac{\lambda}{1-\lambda} \mathbb{F}_1(y|\theta)\right)\right]dy}\]
Using the mean value theorem, there exists $\eta(y)\in(0,1)$ such that\footnote{In the case of the Chi square, $\lambda(y)=\frac{1}{2}$}:
\begin{multline}
\psi\left(\xi^tK\left(\mathbb{F}_0(y|\phi)\right)\right) - \psi\left(\xi^tK\left(\hat{\mathbb{F}}_0(y|\phi)\right)\right) = \xi^t\left(K\left(\mathbb{F}_0(y|\phi)\right) - K\left(\hat{\mathbb{F}}_0(y|\phi)\right)\right) \\ 
\times \psi'\left(\eta(y)\xi^tK\left(\mathbb{F}_0(y|\phi)\right) + (1-\eta(y))\xi^tK\left(\hat{\mathbb{F}}_0(y|\phi)\right)\right)
\label{eqn:MeanValResConsist}
\end{multline}
An exact formula of function $\eta(y)$ will not be needed. We will only use the fact that its image is included in $(0,1)$. By the central limit theorem, one can write:
\[\sqrt{n}\frac{\mathbb{F}_n(y) - \mathbb{F}_T(y)}{\sqrt{\mathbb{F}_T(y)(1-\mathbb{F}_T(y))}} \rightarrow \mathcal{N}\left(0,1\right).\]
Since $\hat{\mathbb{F}}_0(y|\phi) - \mathbb{F}_0(y|\phi) = \mathbb{F}_n(y) - \mathbb{F}_T(y)$, we write
\[\sqrt{n}\frac{\hat{\mathbb{F}}_0(y|\phi) - \mathbb{F}_0(y|\phi)}{\sqrt{\mathbb{F}_T(y)(1-\mathbb{F}_T(y))}} \rightarrow \mathcal{N}\left(0,1\right),\]
which entails by the delta method that:
\begin{equation}
\sqrt{n}\frac{K\left(\hat{\mathbb{F}}_0(y|\phi)\right) - K\left(\mathbb{F}_0(y|\phi)\right)}{\sqrt{\mathbb{F}_T(y)(1-\mathbb{F}_T(y))}} \rightarrow \mathcal{N}\left(0,\nabla K\left(\mathbb{F}_0(y|\phi)\right) \nabla K\left(\mathbb{F}_0(y|\phi)\right)^t\right).
\label{eqn:LimiLawConsist}
\end{equation}
Since function $K$ is a vector of polynomials, its gradient is a matrix of polynomials. Besides, the distribution function $\mathbb{F}_0(y|\phi)$ takes its values in $[0,1]$, thus the variance of the limiting law in (\ref{eqn:LimiLawConsist}) is of order $\frac{1}{n}$ independently of $y$ and $\phi$. We may now write:
\begin{equation}
\frac{K\left(\hat{\mathbb{F}}_0(y|\phi)\right) - K\left(\mathbb{F}_0(y|\phi)\right)}{\sqrt{\mathbb{F}_T(y)(1-\mathbb{F}_T(y))}} = o_P(1)
\label{eqn:KdiffConsist}
\end{equation}
Going back to equation (\ref{eqn:DiffFunConsistency}), we use equations (\ref{eqn:MeanValResConsist}) and (\ref{eqn:KdiffConsist}) to write:
\begin{eqnarray*}
H(\phi,\xi) - H_n(\phi,\xi) & = & \int{\xi^t\left(K\left(\mathbb{F}_0(y|\phi)\right) - K\left(\hat{\mathbb{F}}_0(y|\phi)\right)\right)\psi'\left[\eta(y)\xi^tK\left(\mathbb{F}_0(y|\phi)\right) \right.}\\
				& & \text{\vspace{2cm}} \left.+ (1-\eta(y))\xi^tK(\hat{\mathbb{F}}_0(y|\phi))\right]dy \\
 & = & \int{\sqrt{\mathbb{F}_T(y)(1-\mathbb{F}_T(y))}\xi^t\frac{\left(K\left(\mathbb{F}_0(y|\phi)\right) - K\left(\hat{\mathbb{F}}_0(y|\phi)\right)\right)}{\sqrt{\mathbb{F}_T(y)(1-\mathbb{F}_T(y))}} }\\
& & \times \psi'\left(\eta(y)\xi^tK\left(\mathbb{F}_0(y|\phi)\right) + (1-\eta(y))\xi^tK\left(\hat{\mathbb{F}}_0(y|\phi)\right)\right)dy \\
& = & \xi^to_p(1)\int{\sqrt{\mathbb{F}_T(y)(1-\mathbb{F}_T(y))}\psi'\left[\eta(y)\xi^tK\left(\mathbb{F}_0(y|\phi)\right)\right.}\\
& & \left. + (1-\eta(y))\xi^tK(\hat{\mathbb{F}}_0(y|\phi))\right]dy.
\end{eqnarray*}
The finale line can also be justified by the Chebyshev's inequality, see Remark \ref{rem:ConsistKdiffChebyshev}, or even using the calculus in the proof of Proposition \ref{prop:LimitLawLmomConstrPart} below.\\
It suffices now to prove that the integral in the previous display is finite. Here, $\xi$ (resp. $\phi$) is inside the compact set $\Xi$ (resp. $\Phi$), and functions $\eta(y), \mathbb{F}_0(y|\phi)$ and $\hat{\mathbb{F}}_0(y|\phi)$ all take values inside the compact interval $[0,1]$. Thus, continuity of $\psi'$ suffices to conclude that there exists a constant $M$ independent of $y$, $\phi$ and $\xi$ such that:
\begin{equation}
\left|\psi'\left(\eta(y)\xi^tK\left(\mathbb{F}_0(y|\phi)\right) + (1-\eta(y))\xi^tK\left(\hat{\mathbb{F}}_0(y|\phi)\right)\right)\right| \leq M.
\label{eqn:BoundednessPsiConsist}
\end{equation}
This entails using assumption C6 that:
\begin{eqnarray*}
\int{\sqrt{\mathbb{F}_T(y)(1-\mathbb{F}_T(y))}\left|\psi'\left(\eta(y)\xi^tK\left(\mathbb{F}_0(y|\phi)\right) + (1-\eta(y))\xi^tK\left(\hat{\mathbb{F}}_0(y|\phi)\right)\right)\right|dy} & \leq & \\
 M\int{\sqrt{\mathbb{F}_T(y)(1-\mathbb{F}_T(y))}dy} & & \\
 & < & +\infty.
\end{eqnarray*}
Finally, the integral is finite and the compactness of $\Xi$ implies that $\|\xi\|$ is bounded. Therefore, we have:
\[H(\phi,\xi) - H_n(\phi,\xi) = o_P(1),\]
independently of $\xi$ and $\phi$. We may deduce now that:
\[\sup_{\phi,\xi}\left|H(\phi,\xi) - H_n(\phi,\xi)\right| \xrightarrow[n\rightarrow\infty]{\quad \mathbb{P} \quad} 0.\]
This proves assumption A2.\\
Assumption A3 is immediately verified since function $\xi\mapsto H(\phi,\xi)$ is strictly concave. Assumption A4 is what we have assumed in assumption C3. Finally, continuity assumption A5 is a direct result from assumptions C4 and C5 using Lebesgue's continuity theorem. All assumptions of Theorem \ref{theo:MainTheorem} are fulfilled and the consistency of $\hat{\phi}$ follows as a consequence.
\end{proof}
% remark
\begin{remark}
\label{rem:ConsistKdiffChebyshev}
We can prove assumption A2 in the previous proof without the use of the "small o" notation. We first have:
\[\frac{K\left(\hat{\mathbb{F}}_0(y|\phi)\right) - K\left(\mathbb{F}_0(y|\phi)\right)}{\sqrt{\mathbb{F}_T(y)(1-\mathbb{F}_T(y))}} \stackrel{\mathbb{P}}{\rightarrow} 0.\]
This is translated into the following limit:
\[\forall \varepsilon>0, \quad \mathbb{P}\left(\left|\frac{K\left(\hat{\mathbb{F}}_0(y|\phi)\right) - K\left(\mathbb{F}_0(y|\phi)\right)}{\sqrt{\mathbb{F}_T(y)(1-\mathbb{F}_T(y))}}\right|<\varepsilon\right)  \xrightarrow[n\rightarrow\infty]{} 1\]
Thus, there exists a sequence of positive numbers $(a_n)_n$ independent of\footnote{This is possible using Chebyshev's inequality and using the fact that $K\left(\mathbb{F}_0(y|\phi)\right)$ can be bounded independently of $y$ and $\phi$.} $y$ which goes to zero at infinity such that:
\[\mathbb{P}\left(\left|\frac{K\left(\hat{\mathbb{F}}_0(y|\phi)\right) - K\left(\mathbb{F}_0(y|\phi)\right)}{\sqrt{\mathbb{F}_T(y)(1-\mathbb{F}_T(y))}}\right|<\frac{\varepsilon}{\tilde{M}}\right)\geq 1-a_n\]
where $\tilde{M}=M\sup_{\Xi}\|\xi\|\int{\mathbb{F}_T(y)(1-\mathbb{F}_T(y))dy}$ and $M$ is defined through inequality (\ref{eqn:BoundednessPsiConsist}). On the other hand, the event:
\[\left\|\frac{K\left(\hat{\mathbb{F}}_0(y|\phi)\right) - K\left(\mathbb{F}_0(y|\phi)\right)}{\sqrt{\mathbb{F}_T(y)(1-\mathbb{F}_T(y))}}\right\|<\frac{\varepsilon}{\tilde{M}}\]
implies the event:
\begin{eqnarray*}
& & \int{\sqrt{\mathbb{F}_T(y)(1-\mathbb{F}_T(y))}\|\xi\|\left\|\frac{\left(K\left(\mathbb{F}_0(y|\phi)\right) - K\left(\hat{\mathbb{F}}_0(y|\phi)\right)\right)}{\sqrt{\mathbb{F}_T(y)(1-\mathbb{F}_T(y))}}\right\| \psi'\left(\eta(y)\xi^tK\left(\mathbb{F}_0(y|\phi)\right) \right.}\\
&  &  \left.+ (1-\eta(y))\xi^tK\left(\hat{\mathbb{F}}_0(y|\phi)\right)\right)dy \\
&  & <\frac{\varepsilon}{\tilde{M}} \int{\sqrt{\mathbb{F}_T(y)(1-\mathbb{F}_T(y))}\|\xi\| \psi'\left(\eta(y)\xi^tK\left(\mathbb{F}_0(y|\phi)\right) + (1-\eta(y))\xi^tK\left(\hat{\mathbb{F}}_0(y|\phi)\right)\right)dy}\\
& & <  \varepsilon.
\end{eqnarray*}
This entails that 
\[\left|H(\phi,\xi) - H_n(\phi,\xi)\right|<\varepsilon.\]
%\begin{multline*}
%\mathbb{P}\left(\left\|\frac{K\left(\hat{\mathbb{F}}_0(y|\phi)\right) - K\left(\mathbb{F}_0(y|\phi)\right)}{\sqrt{\mathbb{F}_T(y)(1-\mathbb{F}_T(y))}}\right\|<\frac{\varepsilon}{M}\right) \leq \\
%\mathbb{P}\left(\int{\sqrt{\mathbb{F}_T(y)(1-\mathbb{F}_T(y))}\|\xi\|\left\|\frac{\left(K\left(\mathbb{F}_0(y|\phi)\right) - K\left(\hat{\mathbb{F}}_0(y|\phi)\right)\right)}{\sqrt{\mathbb{F}_T(y)(1-\mathbb{F}_T(y))}}\right\| \psi'\left(\eta(y)\xi^tK\left(\mathbb{F}_0(y|\phi)\right) + (1-\eta(y))\xi^tK\left(\hat{\mathbb{F}}_0(y|\phi)\right)\right)dy}\right. < \\
%\left. \frac{\varepsilon}{M} \int{\sqrt{\mathbb{F}_T(y)(1-\mathbb{F}_T(y))}\|\xi\| \psi'\left(\eta(y)\xi^tK\left(\mathbb{F}_0(y|\phi)\right) + (1-\eta(y))\xi^tK\left(\hat{\mathbb{F}}_0(y|\phi)\right)\right)dy}\right) \leq \\
%\mathbb{P}\left(\int{\sqrt{\mathbb{F}_T(y)(1-\mathbb{F}_T(y))}\|\xi\|\left\|\frac{\left(K\left(\mathbb{F}_0(y|\phi)\right) - K\left(\hat{\mathbb{F}}_0(y|\phi)\right)\right)}{\sqrt{\mathbb{F}_T(y)(1-\mathbb{F}_T(y))}}\right\| \psi'\left(\eta(y)\xi^tK\left(\mathbb{F}_0(y|\phi)\right) + (1-\eta(y))\xi^tK\left(\hat{\mathbb{F}}_0(y|\phi)\right)\right)dy}\varepsilon\right) \leq \\
%\mathbb{P}\left(\left|H(\phi,\xi) - H_n(\phi,\xi)\right|<\varepsilon\right).
%\end{multline*}
The final line does not depend on $(\phi,\xi)$, and we may deduce that:
\[\mathbb{P}\left(\sup_{\phi,\xi}\left|H(\phi,\xi) - H_n(\phi,\xi)\right|<\varepsilon\right) \geq 1-a_n.\]
\end{remark}

%%%%%%%%%%%%%%%%%%%%%%%%%%%%%%%%%%%%%%%%%%%%%%%%%%%%%%%%%%%%%%%%%%%%%%%%%%%%%%%%%%%%%%%%%%%%%%%%%%%%%%%
%%%%%%%%%%%%%%%%%%%%%%%%%%%%%%%%%%%%%%%%%%%%%%%%%%%%%%%%%%%%%%%%%%%%%%%%%%%%%%%%%%%%%%%%%%%%%%%%%%%%%%%
\subsection{Proof of Proposition \ref{prop:LimitLawLmomConstrPart}}\label{Append:PropAsyptotNormIntegral}
\begin{proof}
We would like to calculate the difference $\int{K(\hat{\mathbb{F}}_0(y|\phi))dy}-\int{K(\mathbb{F}_0(y|\phi))dy}$ as a functional of the difference $\hat{\mathbb{F}}_0(y|\phi)-\mathbb{F}_0(y|\phi)$. For two reals $a$ and $b$, we have:
\[K_r(a)-K_r(b) = \sum_{k=0}^{r-1}{\frac{c_{r,k}}{k+1}\left(a^{k+1}-b^{k+1}\right)},\]
where $c_{r,k}=(-1)^{r-k-1}\binom{r-1}{k}\binom{r+k-1}{k}$. Using the identity $a^{k+1}-b^{k+1} = (a-b)\sum_{j=0}^{k}{a^jb^{k-j}}$, we can write:
\[K_r(a)-K_r(b) = (a-b)\sum_{k=0}^{r-1}\sum_{j=0}^{k}{\frac{c_{r,k}}{k+1}a^jb^{k-j}}.\]
Applying this formula on $a=\hat{\mathbb{F}}_0(y|\phi)$ and $b=\mathbb{F}_0(y|\phi)$ yields
\begin{eqnarray}
K_r\left(\hat{\mathbb{F}}_0(y|\phi)\right)-K_r\left(\mathbb{F}_0(y|\phi)\right) & = & \left(\hat{\mathbb{F}}_0(y|\phi)-\mathbb{F}_0(y|\phi)\right)\sum_{k=0}^{r-1}\sum_{j=0}^{k}{\frac{c_{r,k}}{k+1}\hat{\mathbb{F}}_0(y|\phi)^j \mathbb{F}_0(y|\phi)^{k-j}} \nonumber \\
& = & \frac{1}{1-\lambda}\left(\mathbb{F}_n(y)-\mathbb{F}_T(y)\right)\sum_{k=0}^{r-1}\sum_{j=0}^{k}{\frac{c_{r,k}}{k+1}\hat{\mathbb{F}}_0(y|\phi)^j \mathbb{F}_0(y|\phi)^{k-j}}. \nonumber \\
\label{eqn:Kdifference}
\end{eqnarray}
We will show that the sum term can be rewritten using only $\mathbb{F}_0(y|\phi)$. By the Kolmogorov-Smirnov theorem, we have:
\[\sup_{y}\left|\hat{\mathbb{F}}_0(y|\phi)-\mathbb{F}_0(y|\phi)\right| = \sup_{y}\left|\mathbb{F}_n(y)-\mathbb{F}_T(y)\right|  = O_P\left(\frac{1}{\sqrt{n}}\right).\]
This permits us to simply write that
\[\hat{\mathbb{F}}_0(y|\phi)=\mathbb{F}_0(y|\phi) + O_P\left(\frac{1}{\sqrt{n}}\right),\]
with $O_P\left(\frac{1}{\sqrt{n}}\right)$ tends to zero in probability as $n$ goes to infinity independently of $y$. Thus formula (\ref{eqn:Kdifference}) can be rewritten as:
\begin{eqnarray*}
K_r\left(\hat{\mathbb{F}}_0(y|\phi)\right)-K_r\left(\mathbb{F}_0(y|\phi)\right) & = & \frac{1}{1-\lambda}\left(\mathbb{F}_n(y)-\mathbb{F}_T(y)\right)\sum_{k=0}^{r-1}\sum_{j=0}^{k}{\frac{c_{r,k}}{k+1}\left(\mathbb{F}_0(y|\phi)^j + O_P\left(\frac{1}{\sqrt{n}}\right)\right) }\\ 
 & & \text{\vspace{3cm}} \times \mathbb{F}_0(y|\phi)^{k-j}\\
 & = & \frac{1}{1-\lambda}\left(\mathbb{F}_n(y)-\mathbb{F}_T(y)\right)\sum_{k=0}^{r-1}{c_{r,k}\mathbb{F}_0(y|\phi)^k} \\
 & &   + O_P\left(\frac{1}{\sqrt{n}}\right) \frac{1}{1-\lambda}\left(\mathbb{F}_n(y)-\mathbb{F}_T(y)\right)\sum_{k=0}^{r-1}\sum_{j=0}^{k}{\frac{c_{r,k}}{k+1} \mathbb{F}_0(y|\phi)^{k-j}}.
\end{eqnarray*}
Integrating the two sides of the previous equation and multiplying by $\sqrt{n}$ gives:
\begin{multline}
\sqrt{n}\int{\left[K_r\left(\hat{\mathbb{F}}_0(y|\phi)\right)-K_r\left(\mathbb{F}_0(y|\phi)\right)\right]dy} = \frac{1}{1-\lambda}\int{\sqrt{n}\left(\mathbb{F}_n(y)-\mathbb{F}_T(y)\right)\sum_{k=0}^{r-1}{c_{r,k}\mathbb{F}_0(y|\phi)^k}dy}\\ + O_P\left(\frac{1}{\sqrt{n}}\right)\frac{1}{1-\lambda}\int{\sqrt{n}\left(\mathbb{F}_n(y)-\mathbb{F}_T(y)\right)\sum_{k=0}^{r-1}\sum_{j=0}^{k}{\frac{c_{r,k}}{k+1} \mathbb{F}_0(y|\phi)^{k-j}} dy}.
\label{eqn:IntegDiffK}
\end{multline}
The first integral in the right hand side is the part which will produce the Gaussian distribution of the limit law using the CLT. It remains to prove that the second integral in the right hand side tends to zero in probability. Using the law of iterated logarithm, we can write:
\begin{equation}
\limsup_{n\rightarrow\infty}\sqrt{\frac{n}{\log\log n}}\frac{\mathbb{F}_n(y)-\mathbb{F}_T(y)}{\sqrt{\mathbb{F}_T(y)\left(1-\mathbb{F}_T(y)\right)}} = \sqrt{2}.
\label{eqn:IterLogLaw}
\end{equation}
We now may write the integral in the second term as follows:
\begin{multline*}
O_P\left(\frac{1}{\sqrt{n}}\right)\int{\sqrt{n}\left(\mathbb{F}_n(y)-\mathbb{F}_T(y)\right)\sum_{k=0}^{r-1}\sum_{j=0}^{k}{\frac{c_{r,k}}{k+1} \mathbb{F}_0(y|\phi)^{k-j}} dy} = \\ 
O_P\left(\sqrt{\frac{\log\log n}{n}}\right)\int{\sqrt{\frac{n}{\log\log n}}\frac{\mathbb{F}_n(y)-\mathbb{F}_T(y)}{\sqrt{\mathbb{F}_T(y)\left(1-\mathbb{F}_T(y)\right)}} \sqrt{\mathbb{F}_T(y)\left(1-\mathbb{F}_T(y)\right)}  \sum_{k=0}^{r-1}\sum_{j=0}^{k}{\frac{c_{r,k}}{k+1} \mathbb{F}_0(y|\phi)^{k-j}} dy}.
\end{multline*}
The sum term inside the integral is bounded uniformly on $y$. Combine this with the limit in (\ref{eqn:IterLogLaw}), we may deduce that for $n$ sufficiently large, there exists a constant $M$ such that:
\begin{eqnarray*}
\int{\sqrt{\frac{n}{\log\log n}}\frac{\left|\mathbb{F}_n(y)-\mathbb{F}_T(y)\right|}{\sqrt{\mathbb{F}_T(y)\left(1-\mathbb{F}_T(y)\right)}} \sqrt{\mathbb{F}_T(y)\left(1-\mathbb{F}_T(y)\right)}  \sum_{k=0}^{r-1}\sum_{j=0}^{k}{\frac{|c_{r,k}|}{k+1} \mathbb{F}_0(y|\phi)^{k-j}} dy} & \leq & \\
 M\int{\sqrt{\mathbb{F}_T(y)\left(1-\mathbb{F}_T(y)\right)} dy} & &  \\
& < & \infty
\end{eqnarray*}
Thus, the integral exists and is finite for sufficiently large $n$. This entails that:
\begin{equation}
O_P\left(\frac{1}{\sqrt{n}}\right)\int{\sqrt{n}\left(\mathbb{F}_n(y)-\mathbb{F}_T(y)\right)\sum_{k=0}^{r-1}\sum_{j=0}^{k}{\frac{c_{r,k}}{k+1} \mathbb{F}_0(y|\phi)^{k-j}} dy} \xrightarrow[\mathbb{P}]{\quad n\rightarrow\infty\quad} 0,\quad \text{in probability.}
\label{eqn:LimitLawPart2}
\end{equation}
Going back to equation (\ref{eqn:IntegDiffK}), the second term in the right hand side tends to zero in probability. We need now to treat the first term.
\begin{equation}
\int{\sqrt{n}\left(\mathbb{F}_n(y)-\mathbb{F}_T(y)\right)\sum_{k=0}^{r-1}{c_{r,k}\mathbb{F}_0(y|\phi)^k}} = \frac{1}{\sqrt{n}}\sum_{i=1}^n{\int{\left(\ind{X_i\leq y} - \mathbb{F}_T(y)\right)}\sum_{k=0}^{r-1}{c_{r,k}\mathbb{F}_0(y|\phi)^k}dy}.
\label{eqn:AsymptotNormGaussSumCLT}
\end{equation}
This is a sum of i.i.d. random variables. Before proceeding any further, it is necessary to prove that such random variables are well defined (the integrals exist) and have a finite variance. First of all, we have:
\begin{multline}
\int_{-\infty}^{\infty}{\left|\ind{X_i\leq y} - \mathbb{F}_T(y)\right|\sum_{k=0}^{r-1}{c_{r,k}\mathbb{F}_0(y|\phi)^k}dy} =  \int_{-\infty}^{X_i}{\mathbb{F}_T(y)}\sum_{k=0}^{r-1}{c_{r,k}\mathbb{F}_0(y|\phi)^kdy} \\ + \int_{X_i}^{\infty}{\left(1 - \mathbb{F}_T(y)\right)\sum_{k=0}^{r-1}{c_{r,k}\mathbb{F}_0(y|\phi)^k}dy}
\label{eqn:AsymptoNormLmomIntegrability}
\end{multline}
On the other hand, since the $|X_i|$'s have finite expectation, then $X_i$ is finite almost surely and we have:
\begin{eqnarray*}
\mathbb{E}|X_i| & = & \int_{t=0}^{\infty}{\mathbb{P}\left(|X_i|>t\right)dt} \\
 & = & \int_{0}^{\infty}{\left(1-\mathbb{F}_T(t)\right)dt} + \int_{-\infty}^{0}{\mathbb{F}_T(t)dt}.
\end{eqnarray*}
Thus, $\mathbb{F}_T(t)$ is integrable in the neighborhood of $-\infty$, and $1-\mathbb{F}_T(t)$ is integrable in the neighborhood of $+\infty$. This proves that the integral in equation (\ref{eqn:AsymptoNormLmomIntegrability}) exists and is finite. Now the random variables in (\ref{eqn:AsymptotNormGaussSumCLT}) are well defined. The expectation is zero using the Fubini's theorem:
\[\mathbb{E}\left[\int{\left(\ind{X_i\leq y} - \mathbb{F}_T(y)\right)\sum_{k=0}^{r-1}{c_{r,k}\mathbb{F}_0(y|\phi)^k}dy}\right] = \int{\mathbb{E}\left(\ind{X_i\leq y} - \mathbb{F}_T(y)\right)\sum_{k=0}^{r-1}{c_{r,k}\mathbb{F}_0(y|\phi)^k}dy}=0.\]
The final part of the proof is to calculate the covariance matrix. Let $r_1$ and $r_2$ be two positive natural numbers such that $r_1\leq\ell$ and $r_2\leq\ell$. The Fubini's theorem yields:
\begin{multline*}
\mathbb{E}\int{\left(\ind{X_i\leq y} - \mathbb{F}_T(y)\right)\sum_{k=0}^{r_1-1}{c_{r_1,k}\mathbb{F}_0(y|\phi)^k}dy}\int{\left(\ind{X_i\leq x} - \mathbb{F}_T(x)\right)\sum_{k=0}^{r_2-1}{c_{r_2,k}\mathbb{F}_0(x|\phi)^k}dx} = \\
\int{\int{\mathbb{E}\left(\ind{X_i\leq x} - \mathbb{F}_T(x)\right)\left(\ind{X_i\leq y} - \mathbb{F}_T(y)\right)\sum_{k=0}^{r_1-1}{c_{r_1,k}\mathbb{F}_0(x|\phi)^k}\sum_{k=0}^{r_2-1}{c_{r_2,k}\mathbb{F}_0(y|\phi)^k}dy}dx}
\end{multline*}
Denoting $\Sigma$ the covariance matrix, we may write:
\[\Sigma_{r_1,r_2} = \int{\int{\left(\mathbb{F}_T\left(\min(x,y)\right) - \mathbb{F}_T(x)\mathbb{F}_T(y)\right)\sum_{k=0}^{r_1-1}{c_{r_1,k}\mathbb{F}_0(x|\phi)^k}\sum_{k=0}^{r_2-1}{c_{r_2,k}\mathbb{F}_0(y|\phi)^k}dy}dx}.\]
The sum of i.i.d. variables in (\ref{eqn:AsymptotNormGaussSumCLT}) are now well defined and the CLT applies and gives:
\[\frac{1}{\sqrt{n}}\sum_{i=1}^n{\int{\left(\ind{X_i\leq y} - \mathbb{F}_T(y)\right)}\sum_{k=0}^{r-1}{c_{r,k}\mathbb{F}_0(y|\phi)^k}dy} \xrightarrow[]{\quad \mathcal{D}\quad} \mathcal{N}(0,\Sigma).\]
This result together with (\ref{eqn:LimitLawPart2}) and (\ref{eqn:IntegDiffK}) complete the proof.
\end{proof}

%%%%%%%%%%%%%%%%%%%%%%%%%%%%%%%%%%%%%%%%%%%%%%%%%%%%%%%%%%%%%%%%%%%%%%%%%%%%%%%%%%%%%%%%%%%%%%%%%%%%%%%
%%%%%%%%%%%%%%%%%%%%%%%%%%%%%%%%%%%%%%%%%%%%%%%%%%%%%%%%%%%%%%%%%%%%%%%%%%%%%%%%%%%%%%%%%%%%%%%%%%%%%%%
\subsection{Proof of Theorem \ref{theo:AsymptotNormLmom}}\label{Append:TheoNormalAsymptotLmom}
\begin{proof}
The proof is based on a mean value expansion between $(\hat{\phi},\xi_n(\hat{\phi}))$ and $(\phi^*,0)$ similarly to the case of moment constraints Theorem \ref{theo:AsymptotNormalMomConstr}. We therefore, need to calculate the first and second order derivatives.\\ 
First order derivatives are given by:
\begin{eqnarray*}
\frac{\partial H_n}{\partial \xi}(\phi,\xi) & = & m(\alpha) - \int{K(\hat{\mathbb{F}}_0(y|\phi)) \psi'\left(\xi^tK(\hat{\mathbb{F}}_0(y|\phi))\right)dy} \\
\frac{\partial H_n}{\partial \alpha}(\phi,\xi) & = & \xi^t\nabla m(\alpha) \\
\frac{\partial H_n}{\partial \lambda}(\phi,\xi) & = & -\int{\left[\frac{1}{(1-\lambda)^2}\mathbb{F}_n(y) - \frac{1}{(1-\lambda)^2}\mathbb{F}_1(y|\theta)\right]\xi^tK'(\hat{\mathbb{F}}_0(y|\phi))\psi'\left(\xi^t K(\hat{\mathbb{F}}_0(y|\phi))\right)dy} \\
\frac{\partial H_n}{\partial \theta}(\phi,\xi) & = & \frac{\lambda}{1-\lambda} \int{\nabla_{\theta}\mathbb{F}_1(y|\theta)\xi^t K'(\hat{\mathbb{F}}_0(y|\phi)) \psi'\left(\xi^t K(\hat{\mathbb{F}}_0(y|\phi))\right)dy}.
\end{eqnarray*}
Second order derivatives are given by:
\begin{eqnarray*}
\frac{\partial^2 H_n}{\partial \xi^2}(\phi,\xi) & = & \int{K(\hat{\mathbb{F}}_0(y|\phi))K(\hat{\mathbb{F}}_0(y|\phi))^t \psi''\left(\xi^tK(\hat{\mathbb{F}}_0(y|\phi))\right)dy} \\
\frac{\partial^2 H_n}{\partial \alpha^2}(\phi,\xi) & = & \xi^t J_{m(\alpha)} \\
\frac{\partial^2 H_n}{\partial^2 \lambda}(\phi,\xi) & = & -\int{\left[\frac{2}{(1-\lambda)^3}\mathbb{F}_n(y) - \frac{2}{(1-\lambda)^3}\mathbb{F}_1(y|\theta)\right]\xi^tK'(\hat{\mathbb{F}}_0(y|\phi))\psi'\left(\xi^t K(\hat{\mathbb{F}}_0(y|\phi))\right)dy} \\  
					& & -\int{\left[\frac{1}{(1-\lambda)^2}\mathbb{F}_n(y) - \frac{1}{(1-\lambda)^2}\mathbb{F}_1(y|\theta)\right]^2\xi^tK''(\hat{\mathbb{F}}_0(y|\phi))\psi'\left(\xi^t K(\hat{\mathbb{F}}_0(y|\phi))\right)dy} \\
					&  &	-\int{\left[\frac{1}{(1-\lambda)^2}\mathbb{F}_n(y) - \frac{1}{(1-\lambda)^2}\mathbb{F}_1(y|\theta)\right]^2\left[(\xi^tK'(\hat{\mathbb{F}}_0(y|\phi))\right]^2\psi''\left(\xi^t K(\hat{\mathbb{F}}_0(y|\phi))\right)dy} \\
\frac{\partial^2 H_n}{\partial \theta^2}(\phi,\xi) & = & \frac{\lambda}{1-\lambda} \int{J_{\mathbb{F}_1(.|\theta)}\xi^t K'(\hat{\mathbb{F}}_0(y|\phi)) \psi'\left(\xi^t K(\hat{\mathbb{F}}_0(y|\phi))\right)dy} \\
				&  & - \frac{\lambda^2}{(1-\lambda)^2} \int{\nabla_{\theta}\mathbb{F}_1(y|\theta)\nabla_{\theta}\mathbb{F}_1(y|\theta)^t \xi^t K''(\hat{\mathbb{F}}_0(y|\phi)) \psi'\left(\xi^t K(\hat{\mathbb{F}}_0(y|\phi))\right)dy} \\
				 &  & -\frac{\lambda^2}{(1-\lambda)^2} \int{\nabla_{\theta}\mathbb{F}_1(y|\theta)\nabla_{\theta}\mathbb{F}_1(y|\theta)^t \left[\xi^t K'(\hat{\mathbb{F}}_0(y|\phi))\right]^2 \psi''\left(\xi^t K(\hat{\mathbb{F}}_0(y|\phi))\right)dy}
\end{eqnarray*}
Crossed derivatives:
\begin{eqnarray*}
\frac{\partial^2 H_n}{\partial \xi \partial\alpha}(\phi,\xi) & = & \nabla m(\alpha) \\
\frac{\partial^2 H_n}{\partial \xi \partial\lambda}(\phi,\xi) & = &  -\int{\left[\frac{1}{(1-\lambda)^2}\mathbb{F}_n(y) - \frac{1}{(1-\lambda)^2}\mathbb{F}_1(y|\theta)\right]K'(\hat{\mathbb{F}}_0(y|\phi))\psi'\left(\xi^t K(\hat{\mathbb{F}}_0(y|\phi))\right)dy} - \\
 &  & \int{K(\hat{\mathbb{F}}_0(y|\phi))\left[\frac{1}{(1-\lambda)^2}\mathbb{F}_n(y) - \frac{1}{(1-\lambda)^2}\mathbb{F}_1(y|\theta)\right]\xi^tK'(\hat{\mathbb{F}}_0(y|\phi))\psi'\left(\xi^t K(\hat{\mathbb{F}}_0(y|\phi))\right)dy} \\
\frac{\partial^2 H_n}{\partial \xi \partial\theta}(\phi,\xi) & = & \frac{\lambda}{1-\lambda} \int{\nabla_{\theta}\mathbb{F}_1(y|\theta)K'(\hat{\mathbb{F}}_0(y|\phi))^t \psi'\left(\xi^t K(\hat{\mathbb{F}}_0(y|\phi))\right)dy} \\
 &  & + \frac{\lambda}{1-\lambda} \int{K(\hat{\mathbb{F}}_0(y|\phi))\nabla_{\theta}\mathbb{F}_1(y|\theta)^t\xi^t K'(\hat{\mathbb{F}}_0(y|\phi)) \psi'\left(\xi^t K(\hat{\mathbb{F}}_0(y|\phi))\right)dy}\\
\frac{\partial^2 H_n}{\partial \alpha \partial\lambda}(\phi,\xi) & = & 0 \\
\frac{\partial^2 H_n}{\partial \alpha \partial\theta}(\phi,\xi) & = & 0 \\
\frac{\partial^2 H_n}{\lambda \partial\theta}(\phi,\xi) & = & \frac{1}{(1-\lambda)^2}\int{\nabla\mathbb{F}_1(y|\theta)\xi^tK'(\hat{\mathbb{F}}_0(y|\phi))\psi'\left(\xi^t K(\hat{\mathbb{F}}_0(y|\phi))\right)dy} \\
	&  &  + \frac{\lambda}{1-\lambda} \int{\nabla_{\theta}\mathbb{F}_1(y|\theta) \left[\frac{1}{(1-\lambda)^2}\mathbb{F}_n(y) - \frac{1}{(1-\lambda)^2}\mathbb{F}_1(y|\theta)\right]\xi^t K''(\hat{\mathbb{F}}_0(y|\phi))} \\
	& & \qquad \times \psi'\left(\xi^t K(\hat{\mathbb{F}}_0(y|\phi))\right)dy \\  
	&  & + \frac{\lambda}{1-\lambda} \int{\nabla_{\theta}\mathbb{F}_1(y|\theta) \left[\frac{1}{(1-\lambda)^2}\mathbb{F}_n(y) - \frac{1}{(1-\lambda)^2}\mathbb{F}_1(y|\theta)\right] \left[\xi^t K'(\hat{\mathbb{F}}_0(y|\phi))\right]^2} \\
	 & & \qquad \times \psi''\left(\xi^t K(\hat{\mathbb{F}}_0(y|\phi))\right)dy
\end{eqnarray*}
Notice that by assumption 1, interesting values of $\xi$ are only in a neighborhood of the vector 0 which can be taken to be the ball $B(0,\varepsilon)$ for some $\varepsilon>0$. Besides, the derivatives given here above are well defined using Lebesgue theorems. Indeed, all integrands are controlled by either $K(\hat{\mathbb{F}}(y|\phi))$ or $\mathbb{F}_n(y)-\mathbb{F}_1(.|\theta)$ which are both integrable independently of $\phi$ as soon as $\mathbb{F}_1(.|\theta)$ has a finite expectation. Similar discussion for the former was given in Example \ref{Example:Chi2Lmom}, and for the later in the proof of Proposition \ref{prop:LimitLawLmomConstrPart} but for $\mathbb{F}_n(y)-\mathbb{F}_T(.|\theta)$ instead. Other derivatives are controlled by assumtions 4-6 of the present theorem.\\
A mean value expansion of the gradient of $H_n$ between $(\hat{\phi},\xi_n(\hat{\phi}))$ with Lagrange remainder gives that there exists $(\bar{\phi},\bar{\xi})$ on the line between these two points such that:
\begin{equation}
\left(\begin{array}{c} \frac{\partial H_n}{\partial \phi}(\hat{\phi},\xi(\hat{\phi})) \\ \frac{\partial H_n}{\partial \xi}(\hat{\phi},\xi_n(\hat{\phi})) \end{array}\right) = \left(\begin{array}{c}  \frac{\partial H_n}{\partial \phi}(\phi^*,0) \\\frac{\partial H_n}{\partial \xi}(\phi^*,0) \end{array}\right)  + J_{H_n}(\bar{\phi},\bar{\xi}) \left(\begin{array}{c} \hat{\phi}-\phi^* \\ \xi_n(\hat{\phi})\end{array}\right),
\label{eqn:StochExpansionLmom}
\end{equation}
where $J_{H_n}(\bar{\phi},\bar{\xi})$ is the matrix of second derivatives of $H_n$ calculated at the mid point $(\bar{\phi},\bar{\xi})$. First order optimality condition at $(\hat{\phi},\xi_n(\hat{\phi}))$ is translated by:
\begin{eqnarray*}
\frac{\partial}{\partial \xi} H_n(\hat{\phi},\xi_n(\hat{\phi})) & = & 0 \\
\left.\frac{\partial}{\partial \phi}\left(H_n(\phi,\xi_n(\phi))\right)\right|_{\phi=\hat{\phi}} & = & 0. 
\end{eqnarray*}
The chain rule permits us to calculate the second line simply as a derivative with respect to $\phi$ calculated at the optimal point $(\hat{\phi},\xi_n(\hat{\phi}))$, i.e.
\begin{eqnarray*}
\left.\frac{\partial}{\partial \phi}\left(H_n(\phi,\xi_n(\phi))\right)\right|_{\phi=\hat{\phi}} & = &  \frac{\partial}{\partial \phi}H_n(\hat{\phi},\xi_n(\hat{\phi})) + \frac{\partial}{\partial \xi} H_n(\hat{\phi},\xi_n(\hat{\phi})) \frac{\partial \xi_n}{\partial \phi}(\hat{\phi}) \\
 & = & \frac{\partial}{\partial \phi}H_n(\hat{\phi},\xi_n(\hat{\phi})).
\end{eqnarray*}
Thus, optimality conditions at $(\hat{\phi},\xi_n(\hat{\phi}))$ are given by:
\[\frac{\partial H_n}{\partial \xi}(\hat{\phi},\xi_n(\hat{\phi})) = 0, \quad 
\frac{\partial H_n}{\partial \alpha}(\hat{\phi},\xi_n(\hat{\phi})) = 0, \quad
\frac{\partial H_n}{\partial \lambda}(\hat{\phi},\xi_n(\hat{\phi})) = 0, \quad
\frac{\partial H_n}{\partial \theta}(\hat{\phi},\xi_n(\hat{\phi})) = 0.
\]
On the other hand, we have at $(\phi^*,0)$:
\begin{eqnarray*}
\frac{\partial H_n}{\partial \xi}(\phi^*,0) = m(\alpha^*) - \int{K(\hat{\mathbb{F}}_0(y|\phi^*))dy},
\frac{\partial H_n}{\partial \alpha}(\phi^*,0) = 0,
\frac{\partial H_n}{\partial \lambda}(\phi^*,0) = 0,
\frac{\partial H_n}{\partial \theta}(\phi^*,0) =  0.
\end{eqnarray*}
By proposition \ref{prop:LimitLawLmomConstrPart}, since $m(\alpha^*) = \int{K(\mathbb{F}_0(y|\phi^*))}$,
\begin{equation}
\sqrt{n}\left[m(\alpha^*) - \int{K(\hat{\mathbb{F}}_0(y|\phi^*))dy}\right] \xrightarrow{\mathcal{L}}{} \mathcal{N}(0,\Sigma)
\label{eqn:LimitLawPartialDerivHnLmom}
\end{equation}
with $\Sigma$ is the matrix of covariance defined by formula (\ref{eqn:VarCovMatConstrPart}). It remains now to calculate the limit in probability of the matrix $J_{H_n}(\bar{\phi},\bar{\xi})$. Recall first that as $n$ goes to infinity $\bar{\phi}\rightarrow\phi^*$ and $\bar{\xi}\rightarrow 0$. Moreover, by the Slutsky theorem and the law of large numbers, we have:
\[\hat{\mathbb{F}}_0(y,|\bar{\phi}) = \frac{1}{1-\bar{\lambda}} \mathbb{F}_n(y) - \frac{\bar{\lambda}}{1-\bar{\lambda}} \mathbb{F}_1(y|\bar{\theta})\xrightarrow{n\rightarrow\infty}{} \frac{1}{1-\lambda^*} \mathbb{F}_T(y) - \frac{\lambda^*}{1-\lambda^*} \mathbb{F}_1(y|\theta^*) = \mathbb{F}_0(y|\phi^*).\]
We may now give the limit of the blocs of the matrix $J_{H_n}(\bar{\phi},\bar{\xi})$:
\begin{eqnarray*}
\frac{\partial^2 H_n}{\partial \xi^2}(\phi^*,0) =  \int{K(\mathbb{F}_0(y|\phi^*))K(\mathbb{F}_0(y|\phi^*))^tdy}, \qquad
\frac{\partial^2 H_n}{\partial \alpha^2}(\phi^*,0) = 0, \\
\frac{\partial^2 H_n}{\partial^2 \lambda}(\phi^*,0)  =  0, \qquad 
\frac{\partial^2 H_n}{\partial \theta^2}(\phi^*,0) = 0.
\end{eqnarray*}
Crossed derivatives:
\[\frac{\partial^2 H_n}{\partial \xi \partial\alpha}(\phi^*,0) = \nabla m(\alpha^*),\quad \frac{\partial^2 H_n}{\partial \alpha \partial\lambda}(\phi^*,0) = 0,\quad \frac{\partial^2 H_n}{\partial \alpha \partial\theta}(\phi^*,0) = 0  ,\quad 
\frac{\partial^2 H_n}{\lambda \partial\theta}(\phi^*,0) = 0, \]
\[\frac{\partial^2 H_n}{\partial \xi \partial\lambda}(\phi^*,0) = -\int{\left[\frac{1}{(1-\lambda^*)^2}\mathbb{F}_T(y) - \frac{1}{(1-\lambda^*)^2}\mathbb{F}_1(y|\theta^*)\right]K'(\mathbb{F}_0(y|\phi^*))dy}\]
\[\frac{\partial^2 H_n}{\partial \xi \partial\theta}(\phi^*,0) = \frac{\lambda^*}{1-\lambda^*} \int{\nabla_{\theta}\mathbb{F}_1(y|\theta^*)K'(\mathbb{F}_0(y|\phi^*))^tdy}.\]
The limit in probability of the matrix $J_{H_n}(\bar{\phi},\bar{\xi})$ can be written in the form:
\[ J_H = \left[\begin{array}{cc}
0 & J_{\phi^*,\xi^*}^t \\
J_{\phi^*,\xi^*} & J_{\xi^*,\xi^*}
\end{array}\right]\]
where $J_{\phi^*,\xi^*}$ and $J_{\xi^*,\xi^*}$ are given by (\ref{eqn:NormalAsymLMomJ1}) and (\ref{eqn:NormalAsymLMomJ2}). The inverse of matrix $J_H$ has the form:
\[J_H^{-1} = \left(\begin{array}{cc} -\tilde{\Sigma} & H \\ H^t & P\end{array}\right),\]
where
\[
\tilde{\Sigma} = \left(J_{\phi^*,\xi^*}^t J_{\xi^*,\xi^*} J_{\phi^*,\xi^*}\right)^{-1},\quad  H = \tilde{\Sigma} J_{\phi^*,\xi^*}^t J_{\xi^*,\xi^*}^{-1},\quad  P = J_{\xi^*,\xi^*}^{-1} - J_{\xi^*,\xi^*}^{-1} J_{\phi^*,\xi^*} \tilde{\Sigma} J_{\phi^*,\xi^*}^t J_{\xi^*,\xi^*}^{-1}
\]
Going back to (\ref{eqn:StochExpansionLmom}), we have:
\begin{equation*}
\left(\begin{array}{c} 0 \\ 0 \end{array}\right) = \left(\begin{array}{c}  0 \\ \frac{\partial H_n}{\partial \xi}(\phi^*,0) \end{array}\right)  + J_{H_n}(\bar{\phi},\bar{\xi}) \left(\begin{array}{c}  \hat{\phi}-\phi^* \\ \xi_n(\hat{\phi}) \end{array}\right).
\end{equation*}
Solving this equation in $\phi$ and $\xi$ gives:
\begin{equation*}
\left(\begin{array}{c}  \sqrt{n}\left(\hat{\phi}-\phi^*\right) \\ \sqrt{n}\xi_n(\hat{\phi})\end{array}\right) = J_H^{-1}\left(\begin{array}{c}  0 \\ \sqrt{n}\frac{\partial H_n}{\partial \xi}(\phi^*,0) \end{array}\right) + o_P(1).
\end{equation*}
Finally, using (\ref{eqn:LimitLawPartialDerivHnLmom}), we get that:
\[\left(\begin{array}{c}  \sqrt{n}\left(\hat{\phi}-\phi^*\right) \\ \sqrt{n}\xi_n(\hat{\phi})\end{array}\right) \xrightarrow[\mathcal{L}]{} \mathcal{N}\left(0,S\right)\]
where 
\[S=\left(\begin{array}{c}H \\ P\end{array}\right) \Sigma \left(H^t\quad P^t\right).\]
This ends the proof.
\end{proof}

\chapter*{Conclusions and Perspectives}
\addcontentsline{toc}{chapter}{Conclusions and Perspectives}
We summarize some of the most important contributions achieved in this work, and give some future perspectives and research directions concerning the different subjects presented in this manuscript.\\
\begin{itemize}
\item We studied in the first chapter the dual formula of $\varphi-$divergences and showed the limitations of the estimators built using it. We emphasized on the lack of robustness of the so-called MD$\varphi$DE and explained the reason behind this problem. This permitted us to introduce a new robust estimator which we called the kernel-based MD$\varphi$DE. The new estimator is proved to be consistent, asymptotically Gaussian and robust under standard conditions. \\
\item The detailed simulation study presented at the end of Chapter 1 opened several questions. When we work with symmetric kernels, the choice of the window influences on the estimation result. Automatic methods did not give the best results and a suitably-chosen fixed value of the kernel gave always a better result. This gives rise to the question about the best window choice with respect to the estimation procedure instead of the density estimator.\\
\item The use of asymmetric kernels is very useful and must be considered in estimation procedures which uses kernels as soon as we are working with distributions defined on a subset of $\mathbb{R}$. The choice of the kernel was not of a great importance, but the choice of the window was essential. Indeed, existing methods for the choice of the window for asymmetric kernels do not give satisfactory results and a fixed choice of the window gave clear good results almost all the time.\\
\item We presented in the second chapter a proximal-point algorithm for the calculus of divergence-based estimators. We studied the convergence properties of this algorithm and relaxed the identifiability assumption over the proximal term. \\
\item Our simulations show that the proximal algorithm give the same results as a direct optimization algorithm. The question is: Can the proximal algorithm give \emph{clear} better results than direct optimization methods? We could not explore this question in the present work. We could consider a well-known model where direct optimization algorithms fail and converge to "bad" local optima and test whether our proximal algorithm succeed to give a better result.\\
\item The role of the proximal term could also be studied. Indeed, we have noticed that the use of a proximal term of the form $\|\phi-\phi^k\|$ is not suitable for our simulations. The use of a Hellinger-type proximal term gives better results. \\
\item We presented in the third chapter a new structure for semiparametric two-component mixture models where a component is defined through linear constraints such as moments constraints. The new structure permits the addition of a relatively general prior information about the unknown component. The new structure puts the model in between a (restrictive) fully parametric model and a complex semiparametric setup. The new structure permits to estimate the parameters of the parametric component keeping the unknown component in a neighborhood of some family of distributions. The resulting estimator is proved under standard conditions to be consistent and asymptotically normal. \\
\item The estimation procedure presented in Chap. 3 has a linear complexity when we use the $\chi^2$ divergence and when the constraints are polynomials in the distribution function (moments constraints). Besides, no numerical integration or smoothing are needed which permits to calculate the estimates instantly which is a clear advantage over existing methods. The later requires from several hours to several days in order to estimate the parameters of only one sample when the sample size becomes high enough (of order $10^5$). Besides, it permitted in several simulated examples the identification of a parametric component even when the proportion of it is very low (of order 0.01).\\
\item It is necessary and intriguing that we apply our new model on real data and see if we can get satisfactory results. Moreover, we should test the performance of our method on data where the true unknown component $P_0^*$ does not verify the constraints. \\
\item In chapter 4, we presented another structure for semiparametric two-component mixture models when one component is defined by L-moments constraints. The resulting estimator was proved to be consistent and asymptotically normal under standard conditions. In comparison to the structure introduced in Chap. 3 using moments constraints, the use of L-moments constraints shows a clear improvement of performance. In several simulations, we were able to obtain better results with L-moments constraints than with moments constraints and with a smaller number of observations. \\
\item In the literature on L-moments, there exist some propositions and attempts to define multivariate L-moments. Our approach in Chap. 4 only treat univariate L-moments. This may be explorer in a future work.\\
\item An important and difficult question common for both chapters 3 and 4 is: can we use less number of constraints than the number of parameters ? and even if we have an infinite number of solutions, can we ensure that these solutions are in a small neighborhood of the true set of parameters ?
\end{itemize}
\bibliographystyle{plainnat}
\bibliography{PhDbibliography}

\end{document}